%% file: Edgeworth-Uniform--arXiv-v4.tex
\documentclass[bj]{imsart}

\RequirePackage{amsthm,amsmath,amsfonts,amssymb}
\RequirePackage[numbers]{natbib}

\usepackage{appendix, bbm, bm, color, enumitem, forloop, framed, graphicx, hyperref, lscape, mathrsfs, pdfpages, rotating, subcaption, threeparttable, tocloft, multirow, cases}
\usepackage{tikz}
\pdfoutput=1

\startlocaldefs
\newtheorem{theorem}{Theorem}
\newtheorem{corollary}{Corollary}

\newtheorem{lemma}{Lemma}
\newtheorem{assumption}{Assumption}
\theoremstyle{remark}
\newtheorem{remark}{Remark}
	\newcommand{\sumi}{\sum_{i=1}^n}
	\newcommand{\sumj}{\sum_{j=1}^n}
	\newcommand{\sumk}{\sum_{k=1}^n}
	\newcommand{\sumjNoti}{\sum_{j \neq i}}

	\DeclareMathOperator*{\argmin}{arg\,min}
	
	\DeclareMathOperator*{\trace}{trace}
	\DeclareMathOperator*{\sign}{sgn}
	\DeclareMathOperator*{\supp}{supp}
	\DeclareMathOperator*{\diag}{diag}
	\DeclareMathOperator*{\vech}{vech}

	\newcommand{\E}{\mathbb{E}}
	
	\newcommand{\V}{\mathbb{V}}
	\renewcommand{\P}{\mathbb{P}}

	\newcommand{\R}{\mathbb{R}}
	
	\newcommand{\Eqref}[1]{Eqn.\ \eqref{#1}}

	\newcommand{\defsym}{:=}	

	\newcommand{\One}{\mathbbm{1}}

	\newcommand{\that}{\hat{\theta}}
	\newcommand{\shat}{\hat{\sigma}}

	\newcommand{\mhat}{\hat{\mu}}

	\newcommand{\e}{\varepsilon} 

	\newcommand{\vertiii}[1]{{\left\vert\kern-0.25ex\left\vert\kern-0.25ex\left\vert #1  \right\vert\kern-0.25ex\right\vert\kern-0.25ex\right\vert}}
	
	\newcommand{\oeq}{\stackrel{o}{=}}

	\newcommand{\F}{\mathscr{F}}
	\newcommand{\f}{F}
	
		\renewcommand{\S}{S}  %Number of derivatives
		\newcommand{\s}{s}
	
	\newcommand{\I}{\mathscr{I}}
	\renewcommand{\i}{I}

	\newcommand{\tf}{\theta_\f}

    \newcommand{\x}{\mathsf{x}}

	\renewcommand{\v}{\nu}   % other options are "v" and "\nu", but both of these look like the conditional variance function v(x)
	\newcommand{\pp}{p} %generic version, includes undersmoothing as a special case
	\newcommand{\US}{\mathtt{us}}
	
	\newcommand{\RBC}{\mathtt{rbc}}
	\newcommand{\MSE}{\mathtt{mse}}
	\newcommand{\TO}{\mathtt{to}}   %trading off coverage and length
	
	\newcommand{\thatp}{\that_\pp}
	
	\newcommand{\thatrbc}{\that_\RBC}

	\newcommand{\ti}{T}  %generic version
	\newcommand{\tp}{\ti_\pp} %generic version, includes undersmoothing as a special case

	\newcommand{\trbc}{\ti_\RBC}

	\newcommand{\ip}{\i_\pp} %generic version, includes undersmoothing as a special case
	\newcommand{\ius}{\i_\US}
	
	\newcommand{\irbc}{\i_\RBC}

		\newcommand{\biasletter}{\Psi}

		\newcommand{\bi}{\biasletter_{\i,\f}}  %generic version, interval

		\newcommand{\birbc}{\biasletter_{\irbc,\f}}
	
		\newcommand{\bti}{\biasletter_{\ti,\f}}  %generic version, t-stat
		\newcommand{\btp}{\biasletter_{\tp,\f}}

		\newcommand{\btrbc}{\biasletter_{\trbc,\f}}
	
		\newcommand{\biasConstant}{\psi}
		
		\newcommand{\bci}{\biasConstant_{\i,\f}}  %generic version, interval

		\newcommand{\bcirbc}{\biasConstant_{\irbc,\f}}
		
		\newcommand{\bct}{\biasConstant_{\ti,\f}}  %generic version, t stat
		\newcommand{\bctp}{\biasConstant_{\tp,\f}}

		\newcommand{\bctrbc}{\biasConstant_{\trbc,\f}}
	
		\newcommand{\bctrbcHat}{\hat{\biasConstant}_{\RBC,\f}}
		
	    \newcommand{\preTaylor}{\bm{B}}
	
		\newcommand{\az}{\zeta}  		%% just made it 'az' to make find/replace easier. 

	\renewcommand{\sp}{\sigma_\pp} %generic version, includes undersmoothing as a special case
	
	\newcommand{\srbc}{\sigma_\RBC}
	\newcommand{\shatp}{\hat{\sigma}_\pp}
	
	\newcommand{\shatrbc}{\hat{\sigma}_\RBC}	

	\newcommand{\st}{\tilde{\sigma}}
	\newcommand{\sbr}{\breve{\sigma}}
	\newcommand{\sbp}{\sbr_\pp} %generic version, includes undersmoothing as a special case

		\newcommand{\vhat}{\hat{\vartheta}}
		\newcommand{\vhatrbc}{\vhat_\RBC}

	\newcommand{\h}{h}
	\renewcommand{\H}{\bm{H}}   %diagonal matrix for h^k, k=0, 1, 2, etc. NOT the constant part of h
	\renewcommand{\b}{b}
	
	\newcommand{\tO}{s_n}  % = \sqrt{n h}
		\newcommand{\bI}{\bm{I}}

		\newcommand{\bA}{\bm{A}}

		\newcommand{\bZ}{\bm{Z}}
		\newcommand{\bB}{\bm{B}}
		
		\newcommand{\bSn}{\bm{S}_n}

		\newcommand{\bV}{\bm{V}}

		\newcommand{\bb}{\bm{b}}

	\newcommand{\bbeta}{\bm{\beta}}

	\newcommand{\bhat}{\bm{\hat{\beta}}}

	\newcommand{\be}{\bm{e}}

	\newcommand{\br}{\bm{r}}
	\newcommand{\bt}{\bm{t}}
	\newcommand{\bR}{\bm{R}}
	\newcommand{\bRc}{\bm{\check{R}}}
	\newcommand{\bW}{\bm{W}}

	\newcommand{\bY}{\bm{Y}}
	\newcommand{\bM}{\bm{M}}  %vector of conditional expectations

		\newcommand{\G}{\bm{\Gamma}}
		\renewcommand{\L}{\bm{\Lambda}}
		\renewcommand{\O}{\bm{\Omega}}
	
		\newcommand{\Gp}{\bm{\G}}
		\newcommand{\Lp}{\bm{\L}}
		\newcommand{\Op}{\bm{\O}}
	
		\newcommand{\Gq}{\bm{\bar{\G}}}
		\newcommand{\Lq}{\bm{\bar{\L}}}
		\newcommand{\Oq}{\bm{\bar{\O}}}
		
		\newcommand{\Orbc}{\bm{\O}_{\RBC}}
	
		\newcommand{\Gt}{\bm{\tilde{\G}}}
		\newcommand{\Lt}{\bm{\tilde{\L}}}

		\newcommand{\Gpt}{\bm{\Gt}}
		\newcommand{\Lpt}{\bm{\Lt}}

		\newcommand{\Gqt}{\bm{\tilde{\Gq}}}
		\newcommand{\Lqt}{\bm{\tilde{\Lq}}}

		\newcommand{\bS}{\bm{\Sigma}}
		\newcommand{\Shat}{\bm{\hat{\Sigma}}}

		\newcommand{\Shatp}{\bm{\hat{\Sigma}}_\pp}
		
		\newcommand{\Shatrbc}{\bm{\hat{\Sigma}}_\RBC}

	\renewcommand{\l}{\ell}
	\newcommand{\w}{\omega}

	\newcommand{\Xhi}{X_{\h,i}}
	\newcommand{\Xhj}{X_{\h,j}}
	\newcommand{\Xhk}{X_{\h,k}}
	\newcommand{\Xbi}{X_{\b,i}}

\newcommand{\blackline}{\raisebox{2pt}{\tikz{\draw[-,black!40!black,solid,line width = 1.5pt](0,0) -- (5mm,0);}}}

\newcommand{\redline}{\raisebox{2pt}{\tikz{\draw[-,red,dash dot,line width = 1.5pt](0,0) -- (5mm,0);}}}

\newcommand{\blueline}{\raisebox{2pt}{\tikz{\draw[-,blue, dotted,line width = 0.9pt](0,0) -- (5mm,0);}}}

\newcommand{\greenline}{\raisebox{2pt}{\tikz{\draw[-,green, dashed,line width = 0.9pt](0,0) -- (5mm,0);}}}

\endlocaldefs

\begin{document}

%%%%%%%%%%%%%%%%%%%%%%%%%%%%%%%%%%%
%
%	Making the TOC correct for the supplement, command 1 of 3
%
% This must come right after begin{document}, and it ensures
% all the sections go into the PDFs TOC

	\addtocontents{toc}{\protect\setcounter{tocdepth}{2}}
%
%
%%%%%%%%%%%%%%%%%%%%%%%%%%%%%%%%%%%

\begin{frontmatter}
%\title{}
\title{Coverage Error Optimal Confidence Intervals for Local Polynomial Regression}
\runtitle{Coverage Error Optimality}
%\thankstext{T1}{}

\begin{aug}
\author[A]{\fnms{Sebastian} \snm{Calonico}\ead[label=e1]{sebastian.calonico@columbia.edu}},
\author[B]{\fnms{Matias D.} \snm{Cattaneo}\ead[label=e2]{cattaneo@princeton.edu}}
\and
\author[C]{\fnms{Max H.} \snm{Farrell}\ead[label=e3]{max.farrell@chicagobooth.edu}}
%%%%%%%%%%%%%%%%%%%%%%%%%%%%%%%%%%%%%%%%%%%%%%
%% Addresses                                %%
%%%%%%%%%%%%%%%%%%%%%%%%%%%%%%%%%%%%%%%%%%%%%%
\address[A]{Department of Health Policy and Management, Columbia University, New York, New York, U.S.A., \printead{e1}}

\address[B]{Department of Operations Research and Financial Engineering, Princeton University, Princeton, New Jersey, U.S.A., \printead{e2}}

\address[C]{Booth School of Business, University of Chicago, Chicago, Illinois, U.S.A., \printead{e3}}
\end{aug}

\begin{abstract}
This paper studies higher-order inference properties of nonparametric local polynomial regression methods under random sampling. We prove Edgeworth expansions for $t$ statistics and coverage error expansions for interval estimators that (i) hold uniformly in the data generating process, (ii) allow for the uniform kernel, and (iii) cover estimation of derivatives of the regression function. The terms of the higher-order expansions, and their associated rates as a function of the sample size and bandwidth sequence, depend on the smoothness of the population regression function, the smoothness exploited by the inference procedure, and on whether the evaluation point is in the interior or on the boundary of the support. We prove that robust bias corrected confidence intervals have the fastest coverage error decay rates in all cases, and we use our results to deliver novel, inference-optimal bandwidth selectors. The main methodological results are implemented in companion \textsf{R} and \textsf{Stata} software packages.
\end{abstract}


\end{frontmatter}

%%%%%%%%%%%%%%%%%%%%%%%%%%%%%%%%%%%%%%%%%%%%%%
%%%% Main text entry area:

%%%%%%%%%%%%%%%%%%%%%%%%%%%%%%%%%%%
%
%	Making the TOC correct for the supplement, command 2 of 3
%
% This must come before the first section of the main paper,
% and removes the main paper entries from the TOC of the supplement

	\addtocontents{toc}{\protect\setcounter{tocdepth}{0}}
%
%
%%%%%%%%%%%%%%%%%%%%%%%%%%%%%%%%%%%

%%%%%%%%%%%%%%%%%%%%%%%%%%%%%%%%%%%%%%%%%%%%%%%%%%%%%%%%%%%%%%%%%%%%%%%%%%%%%%%%%%%
%%%%%%%%%%%%%%%%%%%%%%%%%%%%%%%%%%%%%%%%%%%%%%%%%%%%%%%%%%%%%%%%%%%%%%%%%%%%%%%%%%%
\section{Introduction}
	\label{sec:intro}

%%%%%%%%%%%%%%%%%%%%%%%%%%%%%%%%%%%%%%%%%%%%%%%%%%%%%%%%%%%%%%%%%%%%%%%%%%%%%%%%%%%
%%%% ADD DISCUSSION OF: \cite{Calonico-Cattaneo-Farrell2020_EJ}
%%%%%%%%%%%%%%%%%%%%%%%%%%%%%%%%%%%%%%%%%%%%%%%%%%%%%%%%%%%%%%%%%%%%%%%%%%%%%%%%%%%

We study local polynomial inference in the general heteroskedastic nonparametric regression model:
\begin{equation}
	\label{eqn:model}
	Y = \mu(X) + \e,   \qquad   \E[\e \vert X ] = 0,   \qquad   \E[\e^2 \vert X] = v(X),
\end{equation}
where $(Y,X)$ is a pair of random variables with distribution $\f$. The parameter of interest is the level or derivative of the regression function at $X = \x$:
\begin{equation}
	\label{eqn:theta lp}
	\mu^{(\v)} = \mu^{(\v)}(\x) := \left.  \frac{\text{d}^\v}{\text{d} x^\v} \E \left[ Y \vert X \!=\! x \right] \right\vert_{x=\x},
	\qquad \v\in\mathbb{Z}_{+},
\end{equation}
where the evaluation point $\x$ may be in the interior or on the boundary of the support of $X$. We drop the evaluation point from the notation when possible, and employ the usual convention $\mu=\mu^{(0)}$. Derivatives at boundary points are defined as one-sided derivatives from the interior. Given a random sample $(Y_1,X_1),\dots,(Y_n,X_n)$ of size $n$ from $\f$, we investigate the quality of statistical inference for $\mu^{(\v)}$ when using kernel-based local polynomial regression methods \cite{Fan-Gijbels1996_book,Fan-Yao2005_book}, focusing in particular on higher-order distributional properties of $t$ statistics as well as on coverage error and length of Wald-type confidence interval estimators. We also employ our results to compare and optimize inference procedures for empirical practice and to shed light on the sometimes underappreciated gap between point estimation and inference.

Our main technical contributions are novel Edgeworth expansions for local polynomial based Wald-type $t$ statistics of the form
\begin{equation}
	\label{eqn:t stat}
    \ti = \frac{\that - \mu^{(\v)}}{\vhat},
\end{equation}
for different choices of point estimator $\that$ and standard error estimator $\vhat$ detailed in Section \ref{sec:set up}. We study the accuracy of the Gaussian approximation to the distribution of such $\ti$ and the error in coverage probability of their dual confidence interval estimators. Our expansions capture the dependence on implementation choices including the polynomial order, the kernel function, and the bandwidth sequence.

Edgeworth expansions are a long-standing tool for more detailed (higher-order) analyses of asymptotic distributional approximations, named after the author of a series of papers on the idea, beginning with \cite{Edgeworth1883} and treated more extensively in \cite{Edgeworth1906_JRSS}. See \cite{Hall1992_book} for a textbook review. Informally, an Edgeworth expansion characterizes the leading terms of the difference between the distribution of $\ti$ and the Gaussian distribution, denoted $\Phi(z)$, for $z \in \R$. That is, an Edgeworth expansion gives the leading terms $E_{\ti,\f}(z)$ and the rate $r_{\ti,\f}$ (which depend on the distribution generating the data, the specific $t$ statistic at issue, along with $n$, $\x$, and other particulars) such that
\begin{equation}
	\label{eqn:EE-pointwise}
	\lim_{n\to\infty} \; r_{\ti,\f}^{-1} \; \sup_{z \in \R} \; \Big| \P_\f[ \ti < z]  - \Phi(z) - E_{\ti,\f}(z) \Big| = 0,
\end{equation}
where $\P_\f$ is the probability law when $\f$ is the true data generating process.

We improve on prior work on valid Edgeworth expansions for nonparametric kernel-based regression in three ways: (i) the expansions hold uniformly over a class of data-generating processes (instead of only for one $\f$), (ii) the uniform kernel is allowed (instead of only for kernel functions with sufficient variation), and (iii) the expansions hold for any derivative $\v \geq 0$ (instead of only for the level $\v = 0$). As discussed below, these improvements offer new theoretical and practical conclusions.

Edgeworth expansions are almost always established pointwise in the underlying distribution, that is, for a single, fixed $\f$, as in \eqref{eqn:EE-pointwise}. Indeed, standard references on the subject \cite{Bhattacharya-Rao1976_book,Hall1992_book} do not even mention uniformity. However, $\f$ is unknown, and a researcher would like some assurances that their inference is equally accurate regardless of the specific underlying data generating process. This motivates expansions that are valid uniformly over a class of plausible distributions for the data, denoted $\F_\S$, encoding the researcher's statistical model, accompanying assumptions, and the empirical regularities of the application of interest. Thus, instead of \eqref{eqn:EE-pointwise}, in Section \ref{sec:ee}, we prove
\begin{equation}
	\label{eqn:EE-uniform}
	\lim_{n\to\infty} \; \sup_{\f \in \F_\S} \; r_{\ti,\f}^{-1} \; \sup_{z \in \R} \; \Big| \P_\f[ \ti < z]  - \Phi(z) - E_{\ti,\f}(z) \Big| = 0.
\end{equation}
We also characterize the worst-case rate $r_{\ti} = \inf_{\f \in \F_\S} r_{\ti,\f}$ of distributional approximation over $\F_\S$. The specific class $\F_\S$ we consider, defined precisely in Section \ref{sec:set up}, matches standard empirical settings employing kernel-based nonparametric inference for $\mu^{(\v)}$, and therefore our theoretical and methodological results speak directly to common practice. Uniformly valid expansions have some precedence in the literature when studying notions of optimality, perhaps originating with \cite{Beran1982_AoS}, but these results are rare and confined to parametric models. Our corresponding uniform results for nonparametric kernel-smoothing do not appear to have a direct antecedent in the literature. 

Second, the uniform kernel is ruled out in all prior work on Edgeworth expansions for kernel-based nonparametrics, both for density estimation \cite{Hall1991_Statistics,Hall1992_AoS_density,Hall1992_book} and regression \cite{Chen-Qin2000_Bmka,Chen-Qin2002_SJS,Calonico-Cattaneo-Farrell2018_JASA}, due to a technical limitation in the proofs that we overcome. Other work on nonparametric regression has assumed away the issue by studying non-random designs \cite{Hall1992_AoS_regression,Neumann1997_Statistics}. In fact, \cite[p. 218]{Hall1991_Statistics} conjectured that valid Edgeworth expansions would require techniques for lattice-valued random variables if the uniform kernel was used. On the contrary, we show that such techniques are not needed. Allowing for the uniform kernel is important for empirical work because it is the optimal kernel shape in terms of minimizing interval length (as discussed in Section \ref{sec:length}) and because unweighted local least squares regression is a popular choice in some applications.

Finally, inference on derivatives of the regression function, again ignored in prior work, is a common task in empirical work and therefore it is valuable to have valid Edgeworth expansions and implementation guidance specifically for this case, including inference-optimal bandwidth selection. Moreover, considering derivatives yields several interesting theoretical conclusions, highlighting the difference between point estimation and inference: we find not only that the rate of the inference-optimal bandwidth does not depend on the specific derivative order $\v$ being considered, analogous to the well-known result for mean squared error (MSE) optimal bandwidth, but also that the rate for inference itself does not depend on $\v$, in sharp contrast to the MSE of the point estimator. 

The main result, a generic Edgeworth expansion encompassing all three of these contributions, is Theorem \ref{thm:ee lp} in Section \ref{sec:ee}. We then use this general result to examine the error in coverage probability of confidence interval estimators dual to each $t$ statistic, and the roles of smoothing bias and Studentization in both the distributional approximation and the coverage error of the confidence intervals.

The role of bias is concretized in Section \ref{sec:bias}. Given a level of smoothness of the unknown function $\mu$ and polynomial order of the local polynomial procedure $\that$, the nonparametric bias must be removed for valid inference. The method of robust bias correction (RBC) addresses this issue by incorporating explicit bias estimation into the centering $\that$ and then also adjusting the scale $\vhat$ to account for the additional variability introduced by the bias estimation \cite{Calonico-Cattaneo-Titiunik_2014_ECMA,Calonico-Cattaneo-Farrell2018_JASA}. An alternative principled inference method relies on removing the bias by shrinking the bandwidth used when conducting inference, often called undersmoothing. Other ad-hoc inference approaches rely on either upper bounding the bias, inflating the scale of the $t$ statistic, or simply ignoring the bias altogether. Using our higher-order expansions, we show that RBC leads to demonstrable higher-order superior inference for $\mu^{(\v)}$ relative to the other approaches in the literature.

Our results also show that the choice of Studentization $\vhat$ is crucial for good higher-order properties. This is in contrast to first order approximations, where only consistency of the standard errors is required. An important finding here is that using asymptotic approximations to the variance of $\that$ will increase the leading remainder terms $E_{\ti,\f}(z)$ and hence also coverage error. Using fixed-$n$ Studentization, where $\vhat$ directly estimates the variability of $\that$, completely removes these errors. This result was first proved in \cite{Calonico-Cattaneo-Farrell2018_JASA}, but only pointwise in $\f$ and excluding the uniform kernel and derivatives of $\mu$. Section \ref{sec:variance} shows this in full generality. Failure to account for the effect of using asymptotic variance approximations has lead to some confusion in the prior literature: for example, \cite{Chen-Qin2002_SJS} found inflated coverage error at boundary points and \cite{Hall1992_AoS_density} found that undersmoothing provides more accurate coverage than bias correction, but both conclusions are due to improper Studentization.

A key practical consequence of the foregoing is that RBC with fixed-$n$ Studentization has leading remainder terms $E_{\ti,\f}(z)$ and rate $r_{\ti,\f}$, of the expansion \eqref{eqn:EE-uniform}, that vanish at least as fast as, and often strictly faster than, undersmoothing-based approaches, both at interior and boundary evaluation points and for any derivative $\v$. Intuitively, this holds because RBC exploits all available smoothness to remove bias, but is not punished (in rates) if no additional smoothness is available to remove bias. Section \ref{sec:rbc} discusses novel implementation of RBC intervals, giving inference-optimal bandwidth and kernel choices that further improve the coverage properties and length of RBC intervals.

More broadly, our results speak to the sometimes neglected gap between point estimation and inference. Implementations focused on optimizing point estimation may not deliver optimal, or even valid, inference. In particular, they need not proceed at the same rate, and perhaps more surprisingly, the inference rate can be \emph{faster}: the rate $r_{\ti,\f}$ at which the distribution of $\that$ collapses to its asymptotic value (namely $\Phi(\cdot)$) can be \emph{faster} than the rate at which $\that$ itself collapses to its asymptotic value ($\mu^{(\v)}$). Indeed, there are cases where a bandwidth choice yields the fastest possible inference rate $r_{\ti,\f}$ but yields invalid point estimation. This is the reverse of the better-known fact that using the estimation-optimal bandwidth (minimizing mean squared error) yields invalid inference. Rate optimality is not as well studied for inference as it is for estimation, but Section \ref{sec:minimax} follows \cite{Hall-Jing1995_AoS} to develop minimax optimal rates in the sense of achieving the fastest (minimal) rate at which the worst-case (maximal) coverage error vanishes and finds that RBC attains this rate. 

The paper closes with simulation evidence supporting our theoretical and methodological work reported in Section \ref{sec:numerical}, and a brief conclusion in Section \ref{sec:conclusion}. An appendix contains formulas omitted to improve the exposition, while an online supplement gives all proofs, detailed simulation results, and other methodological results. Software implementing our main results is provided in \textsf{R} and \textsf{Stata} \cite{Calonico-Cattaneo-Farrell2019_JSS}. Last but not least, some of the ideas in this paper have been applied to causal inference and treatment effect estimation in the context of regression discontinuity designs in \cite{Calonico-Cattaneo-Farrell2020_EJ}.

%%%%%%%%%%%%%%%%%%%%%%%%%%%%%%%%%%%%%%%%%%%%%%%%%%%%%%%%%%%%%%%%%%%%%%%%%%%%%%%%%%%
%%%%%%%%%%%%%%%%%%%%%%%%%%%%%%%%%%%%%%%%%%%%%%%%%%%%%%%%%%%%%%%%%%%%%%%%%%%%%%%%%%%
\section{Model Assumptions and Estimators}
\label{sec:set up}

We define the class $\F_\S$ of distributions for the pair $(Y,X)$ and make precise the local polynomial point estimator $\that$ and scale estimator $\vhat$ of the $t$ statistic \eqref{eqn:t stat}. The class $\F_\S$ is determined through the following assumption. (Recall that derivatives at the boundary of the support of $X$ correspond to one-sided derivatives from the interior of the support.)

\begin{assumption}
	\label{asmpt:dgp lp}
	Let $\F_\S$ be the set of distributions $\f$ for the pair $(Y,X)$ which obey model \eqref{eqn:model} and for which there exist constants $\S \geq \v$, $\s \in(0,1]$, $0<c<C<\infty$, and a neighborhood of $\x$ on the support of $X$, none of which depend on $\f$, such that for all $x, x'$ in the neighborhood the following hold. 
	\begin{enumerate}
		
		\item The Lebesgue density of $(Y,X)$, $f_{yx}(\cdot)$, the Lebesgue density of $X$, $f(\cdot)$, and $v(x) := \V[Y | X=x]$, are each continuous and lie inside $[c,C]$, and $\E[|Y|^{8+c} \vert X = x] \leq C$.
				
		\item $\mu(\cdot)$ is $\S$-times continuously differentiable and $|\mu^{(\S)}(x) - \mu^{(\S)}(x') |\leq C |x - x'|^{\s}$.
		
	\end{enumerate}	
	Throughout,	$\{(Y_1, X_1), \ldots, (Y_n, X_n)\}$ is a random sample from $(Y,X)$.
\end{assumption}

These conditions are not materially stronger than usual in kernel-based nonparametric settings. The restrictions on densities and moments are imposed to achieve uniform validity of Edgeworth expansions. The smoothness condition on $\mu$ plays a key role: the assumed smoothness, captured by $\S$ and $\s$, and its relationship to the smoothness utilized in estimation, will be important for coverage error.

We consider several options for the elements of the $t$ statistic $\ti = (\that - \mu^{(\v)})/\vhat$ given in \eqref{eqn:t stat}. The starting point is the standard local polynomial regression point estimate of $\mu^{(\v)}$. See \cite{Fan-Gijbels1996_book} for an introduction. We index the classical local polynomial estimate by $p\in\mathbb{Z}_+$, the order of the polynomial used, assumed to be at least $\v$. Suppressing the dependence on $\x$ to simplify notation, we therefore set
\begin{equation}
	\label{eqn:lp}
	\mhat_p^{(\v)} = \v! \be_\v' \bhat_p = \frac{1}{n \h^\v} \v! \be_\v'\Gp^{-1} \Op  \bY, 		 \quad\quad 		\bhat_p = \argmin_{\bb \in \mathbb{R}^{p+1}} \sumi ( Y_i - \br_p(\h\Xhi)'\bb)^2  K \left( \Xhi \right),
\end{equation}
where $K$ is a kernel or weighting function, $\h =\h(n) \to 0$ is a bandwidth sequence, $\Xhi = (X_i - \x)/\h$, $\br_p(u) = (1, u, u^2, \ldots, u^p)'$, 
\[\Gp = \frac{1}{n\h}\sumi K(\Xhi)\br_p(\Xhi)\br_p( X_{\h,i} )', \quad
\Op = \frac{1}{\h} \left[ K( X_{\h,1})\br_p(  X_{\h,1} ), \ldots, K(  X_{\h,n} )\br_p( X_{\h,n} )\right],
\]
$\be_\v$ is the $(p+1)$-vector with a one in the $(\v+1)^{\text{th}}$ position and zeros in the rest, and $\bY = (Y_1, \ldots, Y_n)'$. 

The point estimator $\that$ in $\ti$ is then finalized depending on how the smoothing bias is to be accounted for. The traditional approach takes $\that=\mhat_p^{(\v)}$, and then for inference to be valid undersmoothing is required. Explicit bias correction incorporates into $\that$ an estimate of the leading bias term of $\mhat_p^{(\v)}$. Both approaches are motivated by the fact that the conditional bias of $\mhat_p^{(\v)}$ is of order $\h^{p+1 - \v}$ and given by
\begin{equation}
	\label{eqn:bias lp}
	\E\left[\mhat_p^{(\v)} \big| X_1, \ldots, X_n \right] - \mu^{(\v)} =   \h^{p+1 - \v}  \v! \be_\v'\Gp^{-1} \Lp \frac{\mu^{(p+1)}}{(p+1)!}   +   o_\P( \h^{p+1 - \v}),
\end{equation}
with $\Lp = \Op [ X_{h,1}^{p+1}, \cdots, X_{h,n}^{p+1}]'/n$, provided $p-\v$ is odd and $p\leq\S-1$, the standard setting in the literature. Section \ref{sec:bias} details other cases for $p$ and $\S$. Throughout, asymptotic orders and their in-probability versions always hold uniformly in $\F_\S$, as required by our framework: for example, $A_n = o_\P(a_n)$ means $\sup_{\f \in \F_\S} \P_\f [\vert A_n/a_n \vert > \epsilon] \to 0$ for every $\epsilon > 0$. Limits are taken as $n\to\infty$ unless stated otherwise.

Undersmoothing leaves the center of the interval at $\that=\mhat_p^{(\v)}$ unchanged and assumes that the bandwidth $\h$ vanishes rapidly enough to render the leading term of \eqref{eqn:bias lp} negligible relative to the standard error of the point estimator. The term \emph{under}smoothing refers to using less nonparametric smoothing than would be optimal from a mean squared error (MSE) point estimation point of view \citep[][Section 4]{Fan-Gijbels1996_book}. The MSE-optimal bandwidth choice is the most common by far, and indeed, the default in most software. With $p\leq\S-1$, the MSE-optimal bandwidth for $\mhat_p^{(\v)}$ is well-defined whenever $\mu^{(p+1)}(\x) \neq 0$. However, the MSE-optimal bandwidth is too ``large'' for standard Gaussian inference: the bias remains first-order important when scaled by the standard deviation of the point estimator, and so valid inference requires a bandwidth that vanishes faster.

Explicit bias correction, on the other hand, subtracts an estimate of the leading term of \eqref{eqn:bias lp}, of which only $\mu^{(p+1)}$ is unknown. Thus we have:
\begin{equation}
	\label{eqn:thatrbc}
	\thatrbc  :=  \mhat_p^{(\v)} - \h^{p+1 - \v}  \v! \be_\v'\Gp^{-1} \Lp \be_{p+1}'\bhat_{p+1} = \frac{1}{n \h^\v} \v! \be_\v'\Gp^{-1} \Orbc  \bY,
\end{equation}
where $\Orbc = \Op - \rho^{p+1}  \Lp \be_{p+1}' \Gq^{-1} \Oq$ and $\bhat_{p+1}$, $\Gq$, and $\Oq$ are defined akin to $\bhat_p$, $\Gp$, and $\Op$ of \eqref{eqn:lp}, but with $p+1$ in place of $p$ and a bandwidth $\b := \rho^{-1}\h$ instead of $\h$. The parameter $\rho$ will play a key role in the Edgeworth and coverage error expansions and we will derive optimal choices below.

With the point estimator $\that$ defined, we now define the choice of standard errors $\vhat$. We will focus primarily on ``fixed-$n$'' Studentization, also called ``preasymptotic'' by \cite{Fan-Yao2005_book}, which means choosing the Studentization to directly estimate $\V[ \that | X_1, \ldots, X_n ]$, a population quantity but not an asymptotic one. Such choices have superior coverage, as shown below, particularly compared to employing an estimator of an asymptotic representation of $\V[ \that | X_1, \ldots, X_n ]$. Importantly, when $\that = \thatrbc$, a fixed-$n$ approach makes bias correction robust, because the Studentization accounts for the variability of bias estimation.

These fixed-$n$ variances are easy to compute based on standard least squares logic. Referring to \eqref{eqn:lp}, for $\that=\mhat_p^{(\v)}$,
\begin{equation}
	\label{eqn:variance}
	n\h^{1+2\v} \;\V[ \mhat_p^{(\v)} | X_1, \ldots, X_n ] =  \v!^2 \be_\v'\Gp^{-1} (\h \Op \bS \Op' /n) \Gp^{-1} \be_\v,
\end{equation}
where $\bS$ is the $n$-diagonal matrix of conditional variances $v(X_i)$. This formula applies to $\thatrbc$ as well, upon replacing $\Op$ with $\Orbc$, because the two estimators share the same structure, as shown by comparing the second form in \eqref{eqn:thatrbc} to \eqref{eqn:lp}. The fixed-$n$ Studentization is obtained by replacing $\bS$ with an appropriate plug-in estimator, and we then obtain the final $\vhat$ as follows:
\begin{eqnarray}
	\begin{split}
		\label{eqn:se lp}
		\vhat^2 & = \frac{\shatp^2}{n\h^{1+2\v}}, \qquad\qquad && \shatp^2  :=  \v!^2 \be_\v'\Gp^{-1} (\h \Op \Shatp \Op' /n) \Gp^{-1} \be_\v,  \qquad\quad \text{and}   	\\
		\vhat^2 & = \vhatrbc^2 := \frac{\shatrbc^2}{n\h^{1+2\v}}, \qquad\quad  && \shatrbc^2 :=  \v!^2 \be_\v'\Gp^{-1} (\h \Orbc \Shatrbc \Orbc' /n) \Gp^{-1} \be_\v,
	\end{split}
\end{eqnarray}
where $\Shatp$ and $\Shatrbc$ are the $n$-diagonal matrices of the squared residuals $\hat{v}(X_i) = (Y_i - \br_p(X_i)'\bhat_p)^2$ and $\hat{v}(X_i) = (Y_i - \br_{p+1}(X_i)'\bhat_{p+1})^2$, respectively.
The above variance estimators separate explicitly the ``constant'' portions, denoted $\shat_p^2$ and $\shatrbc^2$, which will be used in Section \ref{sec:length} for interval length optimization. More precisely, $\shat_p^2$ and $\shatrbc^2$ will both be bounded and bounded away from zero in probability under our assumptions.

To complete the set of $t$ statistics under consideration, we impose the following standard conditions on the kernel function. This assumption allows for standard choices such as not only the triangular and Epanechnikov kernels, but also the uniform kernel.
\begin{assumption}
	\label{asmpt:ci lp}
	The kernel $K$ is supported on $[-1,1]$, positive, bounded, and even. Further, $K(u)$ is either constant (the uniform kernel) or $(1, K(u) \br_{3(k+1)}(u))'$ is linearly independent on $[-1,0]$ and $[0,1]$, where $k = p$ if $\ti$ is based on $\mhat_p^{(\v)}$ and $\shatp$, and $k = p+1$ if $\ti$ uses $\thatrbc$ or $\shatrbc$. The order $p$ is at least $\v$.
\end{assumption}

%%%%%%%%%%%%%%%%%%%%%%%%%%%%%%%%%%%%%%%%%%%%%%%%%%%%%%%%%%%%%%%%%%%%%%%%%%%%%%%%%%%
%%%%%%%%%%%%%%%%%%%%%%%%%%%%%%%%%%%%%%%%%%%%%%%%%%%%%%%%%%%%%%%%%%%%%%%%%%%%%%%%%%%
\section{Uniformly Valid Edgeworth and Coverage Error Expansions}
\label{sec:ee}

We now give the main technical result of this paper: a uniformly valid, generic Edgeworth expansion as in \eqref{eqn:EE-uniform}, for the $t$-statistic $\ti$ in \eqref{eqn:t stat} when using local polynomial regression methods as described in the previous section. To state the result we need some notation. Here we give only what is needed conceptually, leaving cumbersome formulas to the appendix. The terms of the Edgeworth expansion are defined as
\begin{align}
	\begin{split}
		\label{eqn:ee terms}
		E_{\ti,\f}(z)  &  = \frac{1}{\sqrt{n \h} } \w_{1,\ti,\f}(z)     +    \bti \w_{2,\ti,\f}(z)     +    \lambda_{\ti,\f} \w_{3,\ti,\f}(z)   		\\
		               &    \quad   +    \frac{1}{n \h}  \w_{4,\ti,\f}(z)       +    \bti^2 \w_{5,\ti,\f}(z)   +    \frac{1}{\sqrt{n \h} } \bti  \w_{6,\ti,\f}(z) ,
	\end{split}
\end{align}
where $z$ is the point of evaluation of the distribution, $\bti$ denotes the generic non-random (fixed-$n$) bias of the $\sqrt{n\h^{1+2\v}}$-scaled numerator of $\ti$, $\lambda_{\ti,\f}$ denotes the mismatch between the variance of the numerator of the $t$-statistic and the population standardization used, and the six terms $\w_{k,\ti,\f}(z)$, $k=1,2, \ldots, 6$, are non-random functions bounded uniformly in $\F_\S$, and bounded away from zero for at least one $\f\in\F_\S$. Section \ref{sec:bias} provides further details on $\bti$ and Section \ref{sec:variance} discusses $\lambda_{\ti,\f}$. The quantities $\w_{k,\ti,\f}(z)$, $k=1,2, \ldots, 6$ are relatively less important, beyond their parity, because they cannot be altered by implementation choices. 

We then have the following result (Theorem \ref{thm:ee lp}), establishing \eqref{eqn:EE-uniform}. This result is general, covering interior and boundary points, $p-\v$ even and odd, any derivative $\v \geq 0$, and any combination of $p$ and $\S$. Different cases for each of these primarily affect the expansion, and the final rates, through the bias $\bti$, as explored in the next section. The conditions imposed are strengthened relative to typical pointwise first-order analyses only by $\log(n\h)$ factors on the bandwidth(s) and the other uniformity requirements of Assumption \ref{asmpt:dgp lp}.  (Recall that asymptotic orders and their in-probability versions are always required to hold uniformly in $\F_\S$ throughout.)
\begin{theorem}
	\label{thm:ee lp}
	Let Assumptions \ref{asmpt:dgp lp} and \ref{asmpt:ci lp} hold, and assume that
	\begin{equation*}
		\label{eqn:bandwidth requirements}
		\log(n\h)^{2+\gamma} / n \h =o(1),     \quad    \bti \log(n\h)^{1+\gamma} =o(1),     \quad \lambda_{\ti,\f} = o(1), \quad  \rho = O(1),
	\end{equation*}
	for any $\gamma$ bounded away from zero uniformly in $\F_\S$. Then,
	\[
		\lim_{n\to\infty} \; \sup_{\f \in \F_\S} \; r_{\ti,\f}^{-1} \; \sup_{z \in \R} \; \Big| \P_\f[ \ti < z]  - \Phi(z) - E_{\ti,\f}(z) \Big| = 0
	\]
	holds with $E_{\ti,\f}(z)$ of \eqref{eqn:ee terms} and $r_{\ti,\f} = \max\{(n\h)^{-1}, \bti^2,  (n\h)^{-1/2}\bti, \lambda_{\ti,\f} \}$.
\end{theorem}

A crucial piece in the proof of Theorem \ref{thm:ee lp} is establishing that the appropriate Cram\'er's condition holds under Assumption \ref{asmpt:ci lp}, and in particular the linear independence condition. Such linear independence fails when $K$ is uniform and $u$ runs over the support of $K(u)$, and this failure has prevented the uniform kernel from being covered by past work. Our key insight is that previous approaches ignored the region \emph{outside} the support of $K(\cdot)$ but \emph{inside} the neighborhood of Assumption \ref{asmpt:dgp lp}. Loosely speaking, $(1, K(u), uK(u), \ldots)'$ may be linearly \emph{dependent} on $u \in [-1,1]$ (when $K$ is uniform), but $(1, K(\frac{x - \x}{\h}), (\frac{x - \x}{\h}) K(\frac{x - \x}{\h}), \ldots)'$ is linearly \emph{independent} on $x$ in a fixed neighborhood of $\x$. This allows us to verify Cram\'er's condition. See the supplement for details.

In practice, the error in coverage probability of two-sided interval estimators may be more directly relevant than the distributional approximation of the Edgeworth expansion. We therefore turn to interval estimators dual to each $t$ statistic, given by
\begin{equation}
	\label{eqn:ci lp}
	\i  = \left[ \that - z_u \;\vhat \ , \  \that - z_l \;\vhat \right],
\end{equation}
where $z_l$ and $z_u$ denote chosen quantiles. Our starting point is a generic coverage error expansion for confidence intervals $\i$, dual to a given $\ti$, which follows immediately from Theorem \ref{thm:ee lp} by evaluating the Edgeworth expansion at the interval quantiles (see the supplement).
\begin{corollary}
	\label{cor:ce} 
	Let the conditions of Theorem \ref{thm:ee lp} hold, assume that $\Phi(z_u) - \Phi(z_l) = 1-\alpha$, and define $C_{\i,\f}(z_l,z_u) = E_{\ti,\f}(z_u) - E_{\ti,\f}(z_l)= O(r_\i)$ for some sequence $r_\i$. Then,
	\[
	\lim_{n\to\infty} \; r_\i^{-1} \; \sup_{\f \in \F_\S} \; \Big|  \P_\f \big[ \mu^{(\v)}(\x) \in \i \big] - (1-\alpha) - C_{\i,\f}(z_l,z_u) \Big| = 0.
	\]
\end{corollary}

This result is as general as Theorem \ref{thm:ee lp}. The uniform-in-$\F_\S$ rate $r_\i$ is the slowest vanishing of the rates of each term in the Edgeworth expansion \eqref{eqn:ee terms}, which without specifying any elements further, can only be known to vanish at least as fast as $r_{\ti} = \sup_{\f \in \F_\S} r_{\ti,\f}$ from Theorem \ref{thm:ee lp}. However, even at this level of generality, several conclusions are already evident due to the parity of the functions $\w_k$ making up $E_{\ti,\f}(z)$ and hence $C_{\i,\f}(z_l,z_u)$. First, regarding the choice of quantiles, we recover the classical finding that symmetric intervals, where $z_l = -z_u$, have superior coverage properties, because $\w_1$ and $\w_2$ are even functions of $z$. Asymmetric choices that still have $\Phi(z_u) - \Phi(z_l) = 1-\alpha$ can yield correct coverage, but the error will vanish more slowly, whereas other choices will not yield uniformly correct coverage. Bootstrap-based quantiles will, in general, not improve coverage error rates in nonparametric contexts beyond the symmetric case \cite{Hall1992_AoS_density}, and can in fact be detrimental for coverage error \cite{Hall-Kang2001_AoS}. Second, the remaining $w_k$ functions are odd, and therefore to obtain better coverage properties we should focus on intervals with small (rapidly vanishing) $\bti$ and $\lambda_{\ti,\f}$. The upcoming subsections discuss each of these pieces in turn.

Our expansions highlight the conceptual gap between point estimation and inference. The rate at which the distribution of $\that$ collapses to its asymptotic value ($\Phi(\cdot)$) can be \emph{faster} than the rate at which the point estimator $\that$ itself collapses to its asymptotic value ($\mu^{(\v)}$). Moreover, it is possible that coverage error may vanish even if mean squared error does not, and vice versa. One direction of this phenomenon captures the well-known result that the coverage error of a confidence interval centered at the MSE-optimal point estimator does not vanish. That is, $\hat{\mu}_p^{(\v)}$ in \eqref{eqn:lp} using the MSE-optimal bandwidth $\h_\MSE = H_\MSE n^{-1/(2p+3)}$, for some constant $H_\MSE$, is optimal for point estimation given a fixed $p$, but
\[ \sup_{\f \in \F_\S} \left| \P_\f \left[ \mu^{(\v)} \in \left\{\mhat_p^{(\v)} \pm z_{\alpha/2} \shatp H_\MSE^{-1/2} n^{-1/2  + (1 + 2\v)/(4p+6)} \right\} \right] - (1-\alpha)  \right| \; \asymp \; 1,\]
where $a \asymp b$ denotes that $a \leq C_1 b$ and $b \leq C_1 a$ for some constants $C_1$ and $C_2$.

The other direction may be more surprising and novel: we find that the variance of $\that$ can be too large for mean-square consistency, but nonetheless be captured well enough by $\vhat$ for valid inference. For example, consider inference on $\mu^{(1)}(\x)$ using $\ip$ with local linear regression ($p=1$). Choosing $\h \asymp n^{-1/3}$ yields $r_{\ip} \asymp n^{-2/3}$, which is the fastest attainable rate for $\ip$ in this case, but also gives $\V[\mhat_p^{(\v)} | X_1, \ldots, X_n ] \asymp_\P (n\h^{1 + 2v})^{-1} \asymp 1$, and thus $\mhat_p^{(1)}$ is not consistent in mean square. Therefore, we found a confidence interval that is \emph{optimal for coverage} of $\mu^{(1)}$, but implicitly relies on a point estimator that is \emph{not even consistent} in mean square.

%%%%%%%%%%%%%%%%%%%%%%%%%%%%%%%%%%%%%%%%%%%%%%%%%%%%%%%%%%%%%%%%%%%%%%%%%%%%%%%%%%%
\subsection{Bias Details}
\label{sec:bias}

We now give details for the bias term, $\bti$, highlighting three main points. First, the rate at which $\bti$ vanishes does not depend on the derivative $\v$. Second, we establish that performing bias correction never slows the rate at which $\bti$ vanishes. The third goal is then practical: we spell out several cases of the rates and constants for the bias of $\thatrbc$ so that we may use these for bandwidth and kernel selection later.

To describe $\btp$, the bias term for $\tp$, let $\bbeta_{p}$ be the $p+1$ vector with $(j+1)$ element equal to $\mu^{(j)}(\x)/j!$ for $j = 0, 1, \ldots, p$ as long as $j \leq \S$, and zero otherwise, and $\preTaylor_{p}$ as the $n$-vector with $i^{\text{th}}$ entry $[\mu(X_i) - \br_{p}(X_i - \x)'\bbeta_{p}]$. Then,
\begin{equation}
	\label{eqn:us bias}
	\btp  = \sqrt{n\h} \;  \v! \be_\v'\E[\Gp]^{-1}  \E[\Op \preTaylor_p].
\end{equation}
Turning to bias correction, define $\bbeta_{p+1}$ and $\preTaylor_{p+1}$ as above, but with $p+1$ in place of $p$ in all cases. Then, using the definition of $\Orbc$ in \eqref{eqn:thatrbc},
\begin{equation}
	\label{eqn:rbc bias}
	\btrbc  =  \sqrt{n\h} \;  \v! \be_\v'\E[\Gp]^{-1} \left( \E[\Op \preTaylor_{p+1}] - \rho^{p+1}  \E[\Lp] \be_{p+1}' \E[\Gq]^{-1} \E[\Oq \preTaylor_{p+1}] \right).
\end{equation}
These bias terms are non-random but otherwise non-asymptotic: all expectations are fixed-$n$ and we have not done the typical Taylor expansion. The derivative $\v$ only appears in the constant term $\v! \be_\v$, and therefore the rate at which $\btp$ vanishes does not depend on the derivative being estimated. Intuitively, this can be seen from the second form for $\mhat_p^{(\v)}$ in \eqref{eqn:lp}, $n^{-1} \h^{-\v} \v! \be_\v'\Gp^{-1} \Op  \bY$, coupled with rate $\sqrt{n\h^{1+2\v}}$ of the Studentizations of \eqref{eqn:se lp}: together, these account for the derivative, and distributional properties of $\Gp^{-1} \Op  \bY$ are left independent of $\v$; the first conclusion of this subsection.

The rate of convergence (to zero) of $\btp$ or $\btrbc$ can be deduced by first expanding $\mu(X_i)$ entering $\preTaylor_p$ and $\preTaylor_{p+1}$ around $\x$, and then specializing to a given $p$ and $\S$. For any $p$, we have
\[
	\mu(X_i) - \br_p(X_i - \x)'\bbeta_p = \sum_{k= \S \wedge p + 1}^S \frac{1}{k!}  (X_i - \x)^k \mu^{(k)}(\x)   +   \frac{1}{\S!}  (X_i - \x)^\S \left( \mu^{(\S)}(\bar{x}) - \mu^{(\S)}(\x) \right),
\]
where the summation is taken to be zero if $p\geq\S$. To obtain the final rate, this expansion is substituted into $\preTaylor_p$ and the leading terms are identified by stabilizing the expectation of the terms involving $(X_i - \x)^k$ by writing $\h^k (\Xhi)^k$, thus isolating the rate. The rate will depend on the smoothness, location of $\x$, parity of $p-\v$, and the bandwidth $\h$. For $\preTaylor_{p+1}$, replace $p$ with $p+1$ everywhere and use $\b$ in place of $\h$ in the second term. The supplement gives complete details. 

Our second point is that $\btrbc = O(\btp)$, which follows from the expansion above and taking $\rho$ bounded and bounded away from zero. First, observe from the Taylor expansion applied to \eqref{eqn:rbc bias} that $\btrbc$ depends on higher order derivatives than $\btp$, which follows from applying the Taylor expansion to \eqref{eqn:us bias}, and therefore stabilizing leads to higher powers of $\h$. Intuitively, the bias of $\mhat_p^{(\v)}$ is the product of the rate $\h^{p+1}$ and the constant targeted by bias correction. Therefore, the bias of $\thatrbc$ is at most $\h^{p+1}$ times the bias of the bias correction plus the higher order term of \eqref{eqn:bias lp}. For a fixed sequence $\h$, neither of these can be greater than $\h^{p+1}$. Second, the rate for $\btrbc$ cannot be improved by letting $\rho = \h/\b$ vanish or diverge: $\rho$ vanishing decreases the second term, but the first term is unchanged, while letting $\rho$ diverge can only inflate the second term. Further, diverging $\rho$ renders the effective sample size $n\b$, which is smaller than $n\h$, which would only inflate the Edgeworth expansion terms without reducing bias (hence the restriction in Theorem \ref{thm:ee lp} to bounded $\rho$).

Therefore, in optimizing inference later on, we will focus on $\thatrbc$ and take $\rho$ bounded and bounded away from zero. We need the leading bias constants for this case, which follow from carrying on the Taylor expansion completely in \eqref{eqn:rbc bias}. The bias is always of the form 
\[\btrbc = O(\sqrt{n\h} \h^\az)\]
for an exponent $\az$ that depends on the location of $\x$, the parity of $p-\v$, and the smoothness $\S$. A complete list of $\az$ is shown in Table \ref{table:rbc bias list}. From there, we see that if $p$ is large enough relative to $\S$ (how large depends on the specific case), then $\az = \S + \s$, implying $\btrbc = O(\sqrt{n\h} \h^{\S + \s})$. 

The more empirically relevant case is to treat $p$ as fixed and smaller than $\S$, specifically $p \leq \S - 3$ for interior $\x$ with $p-\v$ odd and $p \leq \S-2$ otherwise (i.e. for boundary points or if $\x$ is an interior point with $p-\v$ even). In these cases, we can use the Taylor expansion above to characterize the leading term, and write
\[\btrbc = \sqrt{n\h} \h^{\az} \bctrbc [1 + o(1)],\]
where $\az = p+3$ for interior $\x$ with $p-\v$ odd and $p+2$ otherwise. The term $\bctrbc$ will be referred to as the constant term for simplicity, though technically it is a non-random sequence with known form, uniformly bounded in $\F_\S$, and nonzero for some $\f \in \F_\S$. Referring to Table \ref{table:rbc bias list} for the different cases, $\bctrbc$ can be
\begin{subnumcases}{}
	\frac{ \mu^{(p+2)} } { (p+2)! } \v! \be_\v'\E[\Gp]^{-1} \Big\{ \E[\Lp_2]     -    \rho^{-1} \E[\Lp_1] \be_{p+1}'\E[\Gq]^{-1} \E[\Lq_1]  \Big\},
	\label{table eqn bnd} \\
	\frac{ \mu^{(p+2)} } { (p+2)! }  \v! \be_\v'\E[\Gp]^{-1} \E[\Lp_2], \quad \qquad \text{or}
	\label{table eqn int even} \\
	\begin{split}
		\v! \be_\v'\E[\Gp]^{-1}  \bigg\{  \frac{ \mu^{(p+2)} } { (p+2)! } \Big[ \h^{-1} \E[\Lp_2]     -    \rho^{-2} \b^{-1} \E[\Lp_1] \be_{p+1}'\E[\Gq]^{-1} \E[\Lq_1]  \Big] 
		\\
		+   \frac{ \mu^{(p+3)} } { (p+3)! }  \Big[ \E[\Lp_3]     -    \rho^{-2} \E[\Lp_1] \be_{p+1}'\E[\Gq]^{-1} \E[\Lq_2]  \Big] \bigg\},
	\end{split}
	\label{table eqn int odd}
\end{subnumcases}
where $\Lp_k = \Op [ X_{h,1}^{p+k}, \cdots, X_{h,n}^{p+k}]'/n$ and $\Lq_k = \Oq [ X_{\b,1}^{p+1+k}, \ldots, X_{\b,n}^{p+1+k}]'/n$, and hence in particular $\Lp_1 \equiv \Lp$ as defined in Section \ref{sec:set up}.

\begin{table}
	\centering
	\renewcommand{\arraystretch}{1.75}
	\begin{tabular}{| l | l | l | l | l |}
		\hline
		Location of $\x$  & Parity of $p\!-\!\v$ & Smoothness &  $\az$ &  $\bctrbc$  \\
		\hline 
		\multirow{2}{*}{Boundary} 
		& \multirow{2}{*}{Odd or Even} & $p + 2 \leq \S$    & $p+2$ & Equation \eqref{table eqn bnd} \\
		&  & $p + 2 > \S$ &  $\S+\s$ &  N/A \\ 
		\hline
		\multirow{4}{*}{Interior} 
		& \multirow{2}{*}{Even} & $p + 2 \leq  \S$    & $p+2$ & Equation \eqref{table eqn int even}  \\
		&  & $p + 2 > \S$   & $\S+\s$ &  N/A \\ 
		\cline{2-5} 
		& \multirow{2}{*}{Odd} & $p + 3 \leq \S$ &   $p+3$ &  Equation \eqref{table eqn int odd} \\
		&  & $p + 2 \geq \S$  & $\S+\s$ &  N/A \\
		\hline
	\end{tabular}
	\caption{Bias Terms For Bias-Corrected Centering $\thatrbc$. With $\rho$ bounded and bounded away from zero, $\btrbc = O(\sqrt{n\h}\h^\az)$ and further, if $p$ is small relative to $\S$, $\btrbc = \sqrt{n\h}\h^\az \bctrbc[1+o(1)]$.}
	\label{table:rbc bias list}
\end{table}

%%%%%%%%%%%%%%%%%%%%%%%%%%%%%%%%%%%%%%%%%%%%%%%%%%%%%%%%%%%%%%%%%%%%%%%%%%%%%%%%%%%
\subsection{Variance Details}
\label{sec:variance}

In contrast to first order distributional analysis, where only consistency is required, the choice of scaling, or Studentization, is crucial for higher order properties. Our detailed expansions show that, in general, there are two types of higher-order terms that arise due to Studentization. One is the unavoidable estimation error incurred when replacing any population quantity with a feasible counterpart. The second error arises from the difference between the population variability of the centering $\that$ and the population standardization chosen as the target. This second type of error is what is captured by $\lambda_{\ti,\f}$, and the most important conclusion is that the fixed-$n$ standard errors in \eqref{eqn:se lp} achieve $\lambda_{\ti,\f} \equiv 0$, and are therefore demonstrably superior choices for inference. That is, there should not be a ``mismatch'' between the population variability of the $t$ statistic numerator and the population standardization.

Using an asymptotic approximation to $\V[ \that | X_1, \ldots, X_n ]$ may yield nonzero $\lambda_{\ti,\f}$, and thus the distributional approximation (and coverage) will suffer. There are too many options to treat comprehensively, but several points warrant discussion. In general, if the chosen standard errors are consistent, $\lambda_{\ti,\f}$ has the form $\lambda_{\ti,\f} = l_n L$, for a rate $l_n \to 0$ and a sequence $L$ that is bounded and bounded away from zero, a ``constant'', capturing the difference between the variance of the numerator of the $t$-statistic and the population standardization chosen.

At boundary points the use of asymptotic approximations can be particularly deleterious for coverage, and this has lead to some confusion in the literature. A headline finding of \cite{Chen-Qin2000_Bmka} is that an empirical likelihood confidence interval estimator has coverage error of the same order at interior and boundary points, which is claimed (in the abstract) to be a ``significant improvement over confidence intervals based directly on the asymptotic normal distribution''. This claim is based on work by the same authors \cite{Chen-Qin2002_SJS} who study, in our notation, the interval with centering $\that = \mhat_1^{(0)}$ and scaling $\vhat = (n\h)^{-1/2} \hat{v}(\x) \hat{f}(\x)^{-1} \mathcal{V}$, for $\hat{v}(\x)$, $\hat{f}(\x)$, and $\mathcal{V}$ given therein, where $v(\x) f(\x)^{-1} \mathcal{V}$ is the probability limit of $\V[ (n\h)^{1/2} \mhat_1^{(0)} \mid X_1, \ldots, X_n ]$. They find that $\lambda_{\ti,\f} = l_n L$ holds with $l_n = \h$ at boundary points, meaning greatly increased coverage error. Concerned that this result is due to estimation error, they confirm that $l_n = \h$ holds with the infeasible standardization $\vartheta = (n\h)^{-1/2} v(\x) f(\x)^{-1} \mathcal{V}$. This neglects the fact that $\lambda_{\ti,\f}$ captures only the estimation error, not the ``mismatch'' error, and their result is entirely due to using an asymptotic standardization as opposed to a fixed-$n$ one, and thus empirical likelihood, in particular, does \emph{not} offer higher-order improvements over normality-based intervals.

Explicit bias correction was claimed by \cite{Hall1992_AoS_density} to be inferior to undersmoothing for inference; a finding also based entirely on using an asymptotic standardization. In this case, nonrobust bias correction was studied, which pairs $\thatrbc$ with $\shatp$. This is valid to first order if $\rho = o(1)$, because then $\V[ \thatrbc | X_1, \ldots, X_n ] = \V[ \mhat_p^{(\v)} \mid X_1, \ldots, X_n ] = o_\P(n^{-1}\h^{-1-2\v})$. However, higher order expansions find $\lambda_{\ti,\f} = \rho ^{p+2}(L_1 + \rho^{p+2} L_2)$, where $L_1$ captures the (scaled) covariance between $\mhat^{(\v)}$ and $\mhat^{(p+1)}$ and $L_2$ the variance of $\mhat^{(p+1)}$. These terms lead \cite{Hall1992_AoS_density} to conclude that bias correction is inferior to undersmoothing, which \cite{Calonico-Cattaneo-Farrell2018_JASA} later showed is not true for robust bias correction. Our results extend this conclusion to hold for derivatives, boundary points, all smoothness cases, and uniformly in $\F_\S$, while also allowing for the uniform kernel.

%%%%%%%%%%%%%%%%%%%%%%%%%%%%%%%%%%%%%%%%%%%%%%%%%%%%%%%%%%%%%%%%%%%%%%%%%%%%%%%%%%%
%%%%%%%%%%%%%%%%%%%%%%%%%%%%%%%%%%%%%%%%%%%%%%%%%%%%%%%%%%%%%%%%%%%%%%%%%%%%%%%%%%%
\section{Optimizing Interval Estimation in Practice}
	\label{sec:rbc}

We turn to optimizing inference in practice, using the conclusions from the previous sections. Collectively, the previous sections imply that the best coverage will be from using symmetric RBC intervals, i.e. those with $z_l = -z_u = z_{\alpha/2}=\Phi^{-1}(\alpha/2)$, $\thatrbc$ as in \eqref{eqn:thatrbc}, $\vhatrbc$ as in \eqref{eqn:se lp}, and $\rho$ bounded and bounded away from zero (implying $\h = \rho \b$). With an eye toward empirical work, we assume in this section that $p$ is fixed and small compared to $\S$. The other cases detailed in Section \ref{sec:bias} are of relatively little practical value. In practice researchers first choose $p$ and then conduct inference based on that choice (witness the ubiquity of local linear regression and cubic splines).

Letting
\[\irbc(\h) = \Big[ \thatrbc + z_{\alpha/2} \; \vhatrbc \ , \  \thatrbc - z_{\alpha/2} \; \vhatrbc \Big] \]
denote the recommended RBC confidence interval, now with its dependence on the bandwidth $h$ explicit to enhance the exposition, we readily deduce from Corollary \ref{cor:ce} that
\begin{equation}
	\label{eqn:something rbc}
	C_{\irbc(\h),\f}(z_{\alpha/2}, -z_{\alpha/2})  = \frac{1}{n \h} 2 \w_{4,\RBC,\f} + 2 n \h^{1+2\az} \bctrbc \w_{5,\RBC,\f} + \h^{\az} 2 \bctrbc  \w_{6,\RBC,\f},
\end{equation}
where the coverage error rate is $r_{\RBC} = \max\{(n\h)^{-1}, n\h^{1+2\az}, \h^\az  \}$, with $\az = p+3$ if $p-\v$ is odd and $\x$ is a boundary point, or $\az =p+2$ otherwise. Furthermore, its interval length is
\begin{equation}
	\label{eqn:il length}
	\left|\irbc(\h) \right| = 2 z_{\alpha/2}  \vhatrbc = 2 z_{\alpha/2}  \frac{\shatrbc}{ \sqrt{n\h^{1 + 2\v} } }.
\end{equation}
Notice that the rate of contraction of length does depend on $\v$, while the coverage error rate does not.

In the next two subsection we use the above two displays, \eqref{eqn:something rbc} and \eqref{eqn:il length}, to choose the bandwidth parameters $\h$ and $\rho = \h/\b$, and the kernel shape. Before any choices can be made, the researcher must decide on the usual size versus power trade off. In our context, this translates to the relative value they place on coverage error, the discrepancy from nominal level, versus interval length. Because we give the first characterizations of coverage error in many cases, and the first uniformly valid ones, this issue can now be studied in detail: our theoretical ideas can inform this trade off, providing new insights to consider, as well as guiding implementation given a preference for coverage error and length.

At one extreme is the approach that requires only that the interval is not anti-conservative, and then minimizes (expected) length. In this case, a shorter interval that uniformly over-covers is preferred to an interval that is longer but has correct coverage asymptotically. Our results lead one to consider the other extreme: minimize the coverage error directly, and only after optimize length. That is, seek for the confidence interval $\i$ such that, in the notation of Corollary \ref{cor:ce}, $r_\i$ vanishes as fast as possible. In applications, an interval with a faster decaying coverage error may approximate its nominal level more closely in finite samples. Such approach focuses on the accuracy of the Gaussian approximation for coverage error, and thus for inference. However, both of these extremes may be unappealing in practice because neither may be optimal from a coverage-length (or, perhaps, size-power for the dual hypothesis test) perspective. Therefore, we will also consider compromises, trading off between coverage error and interval length. One option is to minimize length among consistent interval estimators: seek the shortest interval such that $r_\i = o(1)$. In the context of kernel-based nonparametrics, interval estimators with good control of worst-case coverage are able to use larger bandwidths in general, and are thus shorter in large samples; an analogue to the adage that ``similar tests have higher power''. In general, we will let the user determine a trade off between the two and thus find a bandwidth choice to implement their preference.

%%%%%%%%%%%%%%%%%%%%%%%%%%%%%%%%%%%%%%%%%%%%%%%%%%%%%%%%%%%%%%%%%%%%%%%%%%%%%%%%%%%
\subsection{Optimizing Interval Estimation: Bandwidth Selection}
\label{sec:bandwidth}

We now focus on choosing the bandwidth $\h$ optimally, leaving $\rho$ and $K$ to the next section. With pragmatism in mind, we restrict attention to bandwidth sequences that are polynomial in $n$, that is, of the form $\h = H n^{-\eta}$ for some constants $H>0$ and $\eta >0$. For implementation purposes, we optimize $C_{\irbc(\h),\f}(z_{\alpha/2}, -z_{\alpha/2})$ pointwise in $\f$. The optimal bandwidths will be functions of $\f$ and their implementations are functions of the data, which are draws from $\f$; neither depend explicitly upon $\F_\S$. The resulting coverage error rates still hold uniformly, because the bandwidths are of the form $\h = H n^{-\eta}$, where $\eta$ does not depend on $\f$ and $H$ is well-behaved uniformly in $\F_\S$. We will focus on cases where coverage is consistent, leveraging our new higher-order results in this paper.

An obvious candidate for $\h$ in applications is the classical MSE-optimal choice, denoted $\h_\MSE$, for the point estimator $\hat{\mu}_p^{(v)}(x)$ used as part of the centering of the confidence interval $\irbc(h)$. This bandwidth choice is popular and readily available in most statistical software. Although designed to optimize point estimation, our theoretical results show that it yields valid robust bias corrected inference, that is, $\sup_{\f \in \F_\S} \; |  \P_\f [ \mu^{(\v)}(\x) \in \irbc(\h_\MSE) ] - (1-\alpha) | \to 0$, in contrast to the traditional interval $\i_p(\h_\MSE)$, which undercovers. This gives a principled endorsement for using $\h_\MSE$ coupled with robust bias correction in applications, if a researcher wishes to optimize point estimation instead of inference when choosing the bandwidth $\h$. To be more precise, our results give formal justification (and demonstrate higher-order coverage improvements) for reporting $\mhat_p^{(\v)}(\x)$ along with $\irbc(\h_\MSE)$, both implemented using the same bandwidth $\h_\MSE$, that is, pairing an MSE-optimal point estimator with a valid measure of uncertainty that uses the same samples. In fact, an interesting consequence of our results is that for interior points and local linear regression ($p=1$), $\irbc(\h_\MSE)$ has coverage error that vanishes as fast as possible: for this special case, both the mean squared error and coverage error are optimal in rates upon setting $\h = H n^{-1/(2p+3)}$ for a constant $H > 0$. In other cases, coverage of confidence intervals implemented using $\h_\MSE$ remains consistent but the coverage rate is suboptimal.

To see this, we now turn to inference-optimal bandwidths. We start with the point of view that minimizing coverage error alone is the goal and therefore we choose $\h$ by minimizing the terms of \eqref{eqn:something rbc}. This means setting $h_\RBC = H n^{-\eta_\RBC}$ for $\eta_\RBC = 1/(p+4)$ for interior $\x$ with $p-\v$ odd and $\eta_\RBC = 1/(p+3)$ otherwise: Corollary \ref{cor:ce} holds for $\irbc(\h_\RBC)$ with rates $r_\RBC = n^{-(p+3)/(p+4)}$ and $r_\RBC = n^{-(p+2)/(p+3)}$, respectively. In terms of rates, $\h_\RBC$ balances the variance and bias of the point estimator, instead of the squared bias as in MSE optimality. 

A natural way of choosing the constant $H$ in practice is to minimize the constant portion of the coverage error of \eqref{eqn:something rbc}. Plugging in $h_\RBC = H n^{-\eta_\RBC}$ and factoring out the rate we get
\[ H_\RBC = \argmin_{H > 0} \left\vert H^{-1}  \big\{ 2 \w_{4,\RBC,\f} \big\}       +    H^{1+2\az}  \big\{ 2 \bctrbc^2 \w_{5,\RBC,\f} \big\}   +    H^\az \big\{ 2 \bctrbc  \w_{6,\RBC,\f} \big\} \right\vert.\]
It is straightforward to give a data-driven version of $H_\RBC$, and therefore of $\h_\RBC$, because all quantities involved can be estimated. We defer the details to the supplement to conserve space. In a nutshell, plug-in estimators can be constructed, denoted by $\hat{\w}_{4,\RBC,\f}$, $\hat{\w}_{5,\RBC,\f}$ and $\hat{\w}_{6,\RBC,\f}$, as well as an estimate of the bias constant, $\bctrbcHat$. We then numerically solve
\[ \hat{H}_\RBC = \argmin_{H > 0} \left\vert H^{-1}  \big\{ 2 \hat{\w}_{4,\RBC,\f} \big\}       +    H^{1+2\az}  \big\{ 2 \bctrbcHat^2 \hat{\w}_{5,\RBC,\f} \big\}   +    H^\az \big\{ 2 \bctrbcHat  \hat{\w}_{6,\RBC,\f} \big\} \right\vert.\]
Because this bandwidth depends on the specific data-generating process $\f$, we view it as a rule-of-thumb implementation.

As discussed above, we can also seek for a shorter interval (more power) by sacrificing coverage error (size control). Interval length \eqref{eqn:il length} is reduced for larger bandwidths, meaning smaller exponents $\eta$. Corollary \ref{cor:ce}, or Equation \eqref{eqn:something rbc} specifically, shows that the smallest $\eta$ (i.e., the slowest vanishing bandwidth) such that the coverage of $\irbc(n^{-\eta})$ to be (uniformly) asymptotically correct is $\eta > (1/(1 + 2 \az)$, where recall that $\az = p+3$ for interior points with $p-\v$ odd and $\az = p+2$ otherwise. Therefore, taking $\h = H n^{-\eta}$ for any $\eta > (1/(1 + 2 \az)$ and $H>0$ results in the ideal interval given these preferences over coverage error and length. 

This same idea can be extended to accomplish a \emph{trade-off} between coverage error and length. Researchers may want to have an interval that is closer to nominal level, and therefore may be concerned that in finite samples an interval with coverage error only known to obey $r_\RBC = o(1)$ will not be satisfactory. We can therefore take $\h_\TO = H_\TO n^{-\eta_\TO}$ for some $\eta_\TO \in (1/(1 + 2 \az) , \eta_\RBC]$. Note that if $\eta > \eta_\RBC$ (i.e.\ $\h = o(\h_\RBC)$), both the rate of coverage error decay and interval length contraction can be improved. There is no well-defined optimal choice in this range of asymptotically valid options, as the choice must reflect each researcher's preference for length vs.\ coverage error. This range does not depend on $\v$, even though the resulting length will, see \eqref{eqn:il length}. This may affect how the researcher wishes to trade off the two quantities. The endpoints of the range for $\eta_\TO$ represent preferences for only optimizing coverage error or only length. 

To select the constant for this trade off, $H_\TO$, note first that for $\eta < \eta_\RBC$ the middle term of the coverage error \eqref{eqn:something rbc} is dominant. This term, $n ^{ 1 - \eta_\TO(1+2\az)}  \big\{ 2 \bct^2 \w_{5,\ti,\f} \big\}$, shares the rate of the scaled, squared bias. Therefore, it is natural to balance this against the square of interval length, to match the trade off that $\h_\TO$ represents. The feasible choice of this constant, $\hat{H}_\TO$, will also be a direct plug-in rule that uses the estimators above and a pilot version of $\shatrbc^2$, as well a researcher's choice of weight $\mathcal{H} \in (0,1)$ capturing their trade off between the two. Put altogether, we can then set
\begin{align*}
	\hat{H}_\TO & = \argmin_{H > 0} \Big\{ \mathcal{H} \cdot H^{1+2\az}  \big( 2 \bctrbcHat^2 \hat{\w}_{5,\RBC,\f} \big) + (1 - \mathcal{H}) \cdot 4 z_{\alpha/2}^2  \frac{\shatrbc^2}{ H^{1 + 2\v} } \Big\} 		\\
	& = \left( \frac{(1 - \mathcal{H}) (1 + 2 \v)  4 z_{\alpha/2}^2  \shatrbc^2  }{\mathcal{H} ( 1 + 2 \az) 2 \bctrbcHat^2 \hat{\w}_{5,\RBC,\f}  } \right).
\end{align*}
The resulting data-driven bandwidth choice is $\hat{\h}_\TO = \hat{H}_\TO n^{-\eta_\TO}$, for a choice $\eta_\TO \in(1/(1 + 2 \az), \eta_\RBC]$, and weight $\mathcal{H} \in (0,1)$. The supplement contains details and some additional results.

%%%%%%%%%%%%%%%%%%%%%%%%%%%%%%%%%%%%%%%%%%%%%%%%%%%%%%%%%%%%%%%%%%%%%%%%%%%%%%%%%%%
\subsection{Interval Length Optimality: Choosing $\rho$ and $K(\cdot)$}
\label{sec:length}

To complete the implementation of $\irbc(\h)$ we need to select the bias-correction bandwidth $\b$, which we do in the form of $\rho = \h/\b$, and the kernel function $K(\cdot)$. We choose these to optimize the length \eqref{eqn:il length}. With $\rho$ bounded and bounded away from zero, this choice affects only the constant portions of the coverage error expansion of $\irbc(\h_\RBC)$, in particular changing the shape of the \emph{equivalent kernel} of $\thatrbc$. For more details on equivalent kernels, see \cite[Sect.\ 3.2.2]{Fan-Gijbels1996_book}. To find this equivalent kernel, begin by writing $\thatrbc =  \v! \be_\v'\Gp^{-1} \Orbc  \bY / n \h^\v$ as a weighted average of the $Y_i$. Recall that $\Xhi = (X_i - \x)/\h$ and similarly for $\Xbi$. Then,
\begin{align*}
	\thatrbc & =   \frac{1}{n \h^\v} \v! \be_\v'\Gp^{-1} \left( \Op - \rho^{p+1}  \Lp \be_{p+1}' \Gq^{-1} \Oq \right)  \bY   		\\
	& = \frac{1}{n \h^{1+\v}} \sumi \Big\{  \v! \be_\v'\Gp^{-1} \left(  K(\Xhi)\br_p(\Xhi) - \rho^{p+1} \frac{\h}{\b}  \Lp \be_{p+1}' \Gq^{-1}  K(\Xbi)\br_{p+1}( \Xbi) \right)   \Big\}  \; Y_i.  
\end{align*}
The weights here depend on the sample, as $\Gp$, $\Lp$, and $\Gq$ are sample quantities. The equivalent kernel replaces these with their limiting versions (not, as elsewhere, their fixed-$n$ expectations), which we shall denote $\bm{\mathsf{G}} = f(\x) \int K(u)\br_p(u)\br_p(u)'du$, $\bm{\mathsf{L}} = f(\x) \int K(u)\br_p(u) u^{p+1} du$, and $\bm{\bar{\mathsf{G}}} = f(\x) \int K(u)\br_{p+1}(u)\br_{p+1}(u)'du$, respectively. The integrals are over $[-1,1]$ if $\x$ is an interior point and appropriately truncated when $\x$ is a boundary point. Under our assumptions, convergence to these limits is fast enough that, for the equivalent kernel $\mathcal{K}_\RBC(u; K, \rho, \v)$ defined as
\[ \mathcal{K}_\RBC(u; K, \rho, \v) = \v! \be_\v'\bm{\mathsf{G}}^{-1} \left(  K(u)\br_p( u) - \rho^{p+2}  \bm{\mathsf{L}} \be_{p+1}' \bm{\bar{\mathsf{G}}}^{-1}  K(u \rho)\br_{p+1}( u \rho) \right), \]
and we have the representation 
\begin{align*}
	\thatrbc = \frac{1}{n \h^{1+\v}} \sumi \mathcal{K}_\RBC \big(\Xhi; K, \rho, \v\big)  Y_i \; \{1 + o_\P(1)\}.  
\end{align*}
It follows that the (constant portion of the) asymptotic length of $\irbc(\h)$ depends on $K(\cdot)$ and $\rho$ only through the specific functional $\int\big(\mathcal{K}_\RBC(u; K, \rho, \v)\big)^2du$, which corresponds to the asymptotic variance.

The asymptotic variance of a local polynomial point estimator at a boundary or interior point is minimized by employing the uniform kernel \cite{Cheng-Fan-Marron1997_AoS}. Therefore, to minimize the constant term of interval length we choose $\rho$, depending on $K$, to make $\mathcal{K}_\RBC(u; K, \rho, \v)$ as close as possible to the optimal equivalent kernel, i.e. the $\mathcal{K}^*_{p}(u)$ induced by the uniform kernel for a given $p$. If the uniform kernel is used initially, then $\rho^*=1$ is optimal: that is, $\mathcal{K}_\RBC(\cdot; \One\{|u|<1\}/2, 1, \v) \equiv \mathcal{K}^*_{p+1}(\cdot)$. This highlights the importance of being able to accommodate the uniform kernel in our higher-order expansions. If a kernel other than uniform is used, we look for the optimal choice of $\rho$ by minimizing the $L_2$ distance between the induced equivalent kernel and the optimal variance-minimizing equivalent kernel, solving
\[\rho^* = \argmin_{\rho>0} \int \left| \mathcal{K}_\RBC\big(u; K, \rho, \v \big)  -  \mathcal{K}^*_{p+1}(u) \right| ^2 du.\]
This is not a sample-dependent problem, only computational. For $p-\v$ odd, the standard case in practice, Table \ref{table:rho} shows the optimal $\rho^*$, for boundary and interior points, respectively, the triangular kernel ($K(u)=(1-|u|)\One(|u|\leq 1)$) and the Epanechnikov kernel ($K(u)=0.75(1-u^2)\One(|u|\leq 1)$). These two are popular choices and are MSE-optimal at boundary and interior points, respectively. The shapes of the resulting equivalent kernel, $\mathcal{K}_\RBC(u; K, \rho^*, \v)$, are shown in Figure \ref{fig:equivK} for $\v=\{0,1\}$. Note that although $\rho^*$ itself does not vary with $\v$, the equivalent kernel shape does. Additional choices of $p$ are illustrated in the supplement.

%%%%%%%%%%%%%%%%%%%%%%%%%%%%%%%%%%%%%%%%%%%%%%%%%%%%%%%%%%%%%%%%%%%%%%%%%%%%%%%%%%%
%%%%%%%%%%%%%%%%%%%%%%%%%%%%%%%%%%%%%%%%%%%%%%%%%%%%%%%%%%%%%%%%%%%%%%%%%%%%%%%%%%%
\section{Minimax Coverage Error Decay Rates}
\label{sec:minimax}

In this section we build on \cite{Hall-Jing1995_AoS} and look for a minimax result: characterizing the fastest (minimal) rate at which the worst-case (maximal) coverage error vanishes. The ``optimal'' interval estimator is one for which this maximal error is minimized. At an intuitive level, this corresponds to the desire for similarity in testing: the confidence interval should have ``similar'' coverage over the set of plausible distributions. \cite{Hall-Jing1995_AoS} proposed this inference-specific notion of minimax optimality and studied it in the case of one-sided confidence intervals in the i.i.d.\ parametric location model. This problem is different from the more typical minimaxity considered for point estimation, though the latter is established for robust bias correction by \cite{Tuvaandorj2020_JoE} and is discussed more broadly for local polynomials by \cite{Cheng-Fan-Marron1997_AoS} and \cite{Fan-etal1997_AISM}.

To state the problem more formally, let $\I_p$ denote a class of confidence interval estimators. We then define the \emph{minimax coverage error} as
\begin{equation*}
	\mathrm{MCE}_n :=  \inf_{\i \in \I_p} \sup_{\f \in \F_\S} \; \Big|  \P_\f \big[ \mu^{(\v)}(\x) \in \i \big] - (1-\alpha)  \Big|,
\end{equation*}
where the dependence on the fixed quantities, such as the classes $\I_p$ and $\F_\S$ or the level $\alpha$, are suppressed. Our goal is to characterize the \emph{minimax optimal coverage error decay rate bound}, which is the fastest vanishing sequence $r_\star=r_\star(n)$, $n\in\mathbb{N}$, such that for constants $c_1$ and $c_2$, 
\begin{equation}
	\label{eqn:minimax}
	0 < c_1 \leq \liminf_{n \to \infty}  r_\star^{-1} \mathrm{MCE}_n \leq \limsup_{n \to \infty}  r_\star^{-1} \mathrm{MCE}_n \leq c_2 < \infty.
\end{equation}

We have already characterized the worst-case coverage error in Corollary \ref{cor:ce} for the class of distributions defined in Section \ref{sec:set up}. The key point here is that if we take $\I_p$ to be the class of intervals for which we studied worst-case coverage error in Corollary \ref{cor:ce}, then we can characterize the minimax rate $r_\star$ as well as intervals which attain it. Specifically, we take $\I_p$ to be the Wald-type intervals of the form \eqref{eqn:ci lp}, based on a local polynomial of degree $p$, with any choice of centering, scaling, bandwidth(s), kernel shape, and quantiles, discussed in Section \ref{sec:set up}. This includes all those intervals dual to $t$ statistics covered by Theorem \ref{thm:ee lp}, but also includes other choices which are not asymptotically level $1-\alpha$. Examples include trivial cases such improper choices of quantiles or inconsistent variance estimators, but also choices such as $\ip(\h_\MSE)$, i.e., using the MSE-optimal bandwidth sequence for with centering $\mhat_p^{(\v)}$ and scaling $\shatp^2/(n\h^{1+2\v})$. We could also include other procedures, such as bootstrap based quantiles, empirically chosen bandwidths, or empirical likelihood methods, as these will not improve on the worst-case coverage error \cite{Hall1992_AoS_density,Hall-Kang2001_AoS,Chen-Qin2000_Bmka}.

Crucial to proving that such an interval is minimax optimal is that the bias vanishes at the best possible rate, given the smoothness assumed ($\S$) and utilized ($p$), and this in turn depends on whether $\x$ is an interior or boundary point. Collecting all the smoothness cases studied in Section \ref{sec:bias}, we immediately obtain the following result (see the supplement for omitted details).
\begin{corollary}\label{thm:fixed p lp}
	Let Assumptions \ref{asmpt:dgp lp} and \ref{asmpt:ci lp} hold and let $\I_p$ be the class of Wald-type confidence intervals described in the foregoing paragraph.\\
	(i) Let $\x$ be an interior point in the support of $X$. If $p-\v$ is odd, then \eqref{eqn:minimax} holds with $r_\star = n^{-(p+3)/(p+4)}$ if $p\leq S-3$ and $r_\star = n^{-(\S + \s)/(\S + \s + 1)}$ if $p \geq S-2$. If $p-\v$ is even, then $r_\star = n^{-(p+2)/(p+3)}$ if $p\leq S-2$ and $r_\star = n^{-(\S + \s)/(\S + \s + 1)}$ if $p \geq S-1$.\\
	(ii) Let $\x$ be a boundary point of the support of $X$. Then, \eqref{eqn:minimax} holds with $r_\star = n^{-(p+2)/(p+3)}$ if $p\leq S-2$ and $r_\star = n^{-(\S + \s)/(\S + \s + 1)}$ if $p \geq S-1$.

\end{corollary}

For the classes $\F_\S$ and $\I_p$ considered herein, this result establishes the minimax rate bounds. The interplay between the two classes is crucial: they should be neither too ``large'' nor too ``small'' in order to obtain useful and interesting results. The larger is $\F_\S$, the more plausible a given data set is generated by some $\f \in \F_\S$, but well known results dating back at least to \cite{Bahadur-Savage1956_AoMS} show that if $\F_\S$ is too large it is impossible to construct an ``effective confidence interval'' that controls the worst-case coverage. Our particular $\F_\S$ captures common restrictions in the setting of nonparametric regression, and therefore matches empirical practice. The class $\I_p$ is restricted to contain Wald-type interval estimators commonly employed in practice using nonparametric kernel-based regression methods (but can be trivially extended to cover alternatives mentioned above). Recall that our goal is to identify if RBC confidence intervals improve over other options in a uniform sense, and this result is tailored to that goal. 

The main message of Corollary \ref{thm:fixed p lp} is that $\irbc(\h_\RBC)$ is minimax optimal in all cases. This strengthens the pointwise improvement offered by robust bias correction to optimality within the class $\I_p$ considered here. Intuitively, this is because robust bias correction successfully exploits additional smoothness if it exists, but is not punished (in rates) if there is no such smoothness due to the change in Studentization. This can be compared to $\i_p$, the classical interval that requires undersmoothing. This interval is optimal only in the case when $\S$ is known so that $p$ can be chosen large enough; for a fixed $p$ that is small relative to $\S$ this interval is dominated in the minimax sense.

%%%%%%%%%%%%%%%%%%%%%%%%%%%%%%%%%%%%%%%%%%%%%%%%%%%%%%%%%%%%%%%%%%%%%%%%%%%%%%%%%%%
%%%%%%%%%%%%%%%%%%%%%%%%%%%%%%%%%%%%%%%%%%%%%%%%%%%%%%%%%%%%%%%%%%%%%%%%%%%%%%%%%%%
\section{Simulation Study}
\label{sec:numerical}

This section presents results from a simulation study to examine the finite-sample performance of our methods. Additional results and implementation details can be found in the supplement. We focus on the performance of confidence intervals for $\mu(\x)$ and $\mu^{(1)}(\x)$ based on robust bias correction and traditional undersmoothing. Data is generated from model \eqref{eqn:model}, with $X_i$ uniformly distributed on $[-1,1]$, $\e$ standard normal, and
\[
\mu(x) = \frac{ \sin(3\pi x/2 ) }{ 1+18x^2(\sign(x) +1)   },
\]
where $\sign(x)=-1$, $0$, or $-1$ according to $x>0$, $x=0$ or $x<0$, respectively. 
This function, which was also analyzed in \cite{Calonico-Cattaneo-Farrell2018_JASA}, is displayed in Figure \ref{fig:dgp} together with $\mu^{(1)}(x)$. By looking at different evaluation points, we will be able to capture the performance of the methods under different levels of complexity.

We show results for sample sizes $n\in\{100,250,500,750,1000,2000\}$, always with $5,000$ replications. We study inference at three evaluation points: $\x=-1$ (boundary point), $\x=-0.6$ (low curvature), and $\x=-0.2$ (high curvature). The supplement shows results for $\x\in\{0.2,0.6,1\}$. For implementation, we use $p=1$ (for $\v=0$) and $p=2$ (for $\v=1$) with the Epanechnikov kernel (the supplement gives results for the uniform kernel). Finally, we evaluate the performance of the confidence intervals using several bandwidth choices. First, following the results from Section \ref{sec:rbc}, we use $\hat{h}_\RBC$, a data-driven version of the inference-optimal bandwidth $h_\RBC$. We also consider the analogous version for undersmoothed confidence intervals, denoted $\hat{h}_\US$ (detailed in the supplement), and the standard choice in practice, $\hat{h}_\MSE$.
Robust bias correction is implemented using $\rho=\rho^*$ according to Table \ref{table:rho}. All implementation details are available for \textsf{R} and \textsf{Stata} \cite{Calonico-Cattaneo-Farrell2019_JSS}.

Figures \ref{fig:sim_ec0} and \ref{fig:sim_ec1} present empirical coverage probabilities for $\v=0$ and $\v=1$, respectively, for each evaluation point and choice of bandwidth, as a function of the sample size. 
Overall, we can see that robust bias correction yields close to accurate coverage, improving over undersmoothing in almost every case. Performance is highly superior at points where the functions present high curvature and also at the boundary. Performance is never worse even when the function is quite linear and optimal bandwidths are (close to) ill-defined.

We also compare confidence interval performance in terms of length in Figure \ref{fig:sim_il}. We take coverage into account by looking at RBC and US confidence intervals implemented using their corresponding coverage error optimal bandwidth choices ($\hat{h}_\RBC$ and $\hat{h}_\US$, respectively), which is when they perform best in terms of coverage. We also include other valid, but non optimal choices $\irbc(\hat{h}_{\MSE})$,  $\irbc(\hat{h}_{\US})$. We find that RBC confidence intervals are, on average, not larger than US, and sometimes even shorter. Lastly, Figure \ref{fig:h_epa} shows the average estimated bandwidths at each point for each sample size, which behave as expected following our theory.

%%%%%%%%%%%%%%%%%%%%%%%%%%%%%%%%%%%
%%%%%%%%%%%%%%%%%%%%%%%%%%%%%%%%%%%
%%%%%%%%%%%%%%%%%%%%%%%%%%%%%%%%%%%
\section{Conclusion}
\label{sec:conclusion}

This paper derived higher order expansions for inference in nonparametric local polynomial regression. We provided new Edgeworth expansions and associated error in coverage probability expansions for standard and robust bias corrected methods, showing that the latter have superior coverage properties. Our results hold uniformly in the data generating process, cover derivative estimation, and allow for the uniform kernel. Using our results we developed novel bandwidth selections that target inference directly, achieving lower coverage error and/or shorter length. 

Our main results measured coverage error symmetrically, but it is worth mentioning that the absolute loss function may be replaced by the ``check'' loss function, and thus studying the maximal coverage error $\sup_{\f \in \F_\S} \mathcal{L}( \P_\f [ \tf \in \i ] - (1-\alpha) )$, with $\mathcal{L}(e) =\mathcal{L}_\tau(e)=e(\tau - \One\{e<0\})$, and where $\tau \in (0,1)$ encodes the researcher's weight for over- and under-coverage. Setting $\tau = 1/2$ recovers the above, symmetric measure of coverage error. Guarding more against undercoverage (a preference for conservative intervals) requires choosing a $\tau < 1/2$. For example, setting $\tau = 1/3$ encodes the belief that undercoverage is twice as bad as the same amount of overcoverage. All our results can be established for this loss function. 

Finally, this paper studied the properties of confidence intervals at a fixed evaluation point $\x$, but it would be of theoretical and practical interest to extent our results to the case of confidence band construction. Robust bias correction has recently been used to construct valid confidence bands for local polynomial estimation \cite{Cheng-Chen2019_EJS} and linear sieve estimation \cite{Cattaneo-Farrell-Feng2020_AoS}. Because the underlying distributional approximations for confidence band constructions are substantially more complex, obtaining results similar to those presented herein will require substantial extension of our technical work.

%%%%%%%%%%%%%%%%%%%%%%%%%%%%%%%%%%%%%%%%%%%%%%
%% Single Appendix:                         %%
%%%%%%%%%%%%%%%%%%%%%%%%%%%%%%%%%%%%%%%%%%%%%%
\begin{appendix}
	\section*{Appendix: Terms of the Edgeworth Expansion}%% if no title is needed, leave empty \section*{}.
		
	We give the definition of $\w_k$, $k = 1, 2, \ldots, 6$. First, define the following objects, all calculated in a fixed-$n$ sense, bounded uniformly in $\F_\S$, and nonzero for some $\f\in\F_\S$. As shorthand, let a tilde accent denote a fixed-$n$ expectation, so that $\Gpt = \E[\Gp]$, $\Lpt_1 = \E[\Lp_1]$, and so forth. Let
	\begin{align*}
		\l^0_{\tp}(X_i) & = \v! \be_\v' \Gpt^{-1} (K \br_p)(\Xhi) ; 		\\
		\l^0_{\trbc}(X_i) & = \l^0_{\tp}(X_i)  -   \rho^{p+1} \v! \be_\v' \Gpt^{-1} \Lpt_1 \be_{p+1}' \Gqt^{-1}  (K \br_{p+1})(\Xbi) ; 		\\
		\l^1_{\tp}(X_i, X_j) & = \v! \be_\v' \Gpt^{-1} \left( \E[(K \br_p \br_p')(\Xhj)] - (K \br_p \br_p')(\Xhj) \right) \Gpt^{-1} (K \br_p)(\Xhi) ; 		\\
		\l^1_{\trbc}(X_i, X_j) & = \l^1_{\tp}(X_i, X_j)     -    \rho^{p+1}   \v! \be_\v' \Gpt^{-1} \Bigl\{ \left( \E[(K \br_p \br_p')(\Xhj)] - (K \br_p \br_p')(\Xhj) \right)  \Gpt^{-1}  \Lpt_1 \be_{p+1}'   		\\
		& \qquad \qquad       +  \left( (K \br_p)(\Xhj) \Xhi^{p+1}   -  \E[(K \br_p)(\Xhj) \Xhi^{p+1}]  \right) \be_{p+1}'       		\\
		& \qquad \qquad       +  \Lpt_1 \be_{p+1}' \Gqt^{-1}  \left( \E[(K \br_{p+1} \br_{p+1}')(X_{\b,j})] - (K \br_{p+1} \br_{p+1}')(X_{\b,j}) \right)   \Bigr\} \Gqt^{-1}  (K \br_{p+1})(\Xbi).
	\end{align*}
	Then define $\st_{\ti}^2  = \E[ \h^{-1} \l^0_{\ti}(X)^2 v(X) ]$ and denote the standard Normal density as $\phi(z)$. Then we define
	\begin{align*}
		\w_{1,\ti,\f}(z) &  =   \phi(z) \st_{\ti}^{-3} \E \left[ \h^{-1} \l^0_{\ti}(X_i)^3 \e_i^3 \right] \left\{  (2z^2 - 1)/6 \right\}	  , 		 \\
		\w_{2,\ti,\f}(z) &  =   -  \phi(z)  \st_{\ti}^{-1}	  , 		 \\
		\w_{3,\ti,\f}(z) &  =   -  \phi(z) \left\{ z / 2 \right\}	  , 		 \\
		\w_{5,\ti,\f}(z) &  =   -  \phi(z)  \st_{\ti}^{-2} \left\{ z / 2 \right\}	  , 		 \\
		\w_{6,\ti,\f}(z) &  =   \phi(z)  \st_{\ti}^{-4} \E [ \h^{-1} \l^0_{\ti}(X_i)^3 \e_i^3 ]  \left\{ z^3 / 3 \right\}	 .
	\end{align*}
	For $\w_4$, it is not quite as simple to state a generic version. Let $\bm{\tilde{G}}$ stand in for $\Gpt$ or $\Gqt$, $\tilde{p}$ stand in for $p$ or $p+1$, and $d_n$ stand in for $\h$ or $\b$, all depending on if $\ti = \tp$ or $\trbc$. Note however, that $\h$ is still used in many places, in particular for stabilizing fixed-$n$ expectations, for $\trbc$. Indexes $i$, $j$, and $k$ are always distinct (i.e.\ $X_{\h,i} \neq X_{\h,j} \neq X_{\h,k}$). 
	\begin{align*}
		\w_{4,\ti,\f}(z) & =  \phi(z)  \st_{\ti}^{-6} \E \left[ \h^{-1} \l^0_{\ti}(X_i)^3 \e_i^3 \right]^2 \left\{  z^3/3 + 7 z /4 + \st_{\ti}^2 z (z^2-3)/4 \right\}   		\\
		& \quad +  \phi(z)  \st_{\ti}^{-2} \E \left[ \h^{-1} \l^0_{\ti}(X_i) \l^1_{\ti}(X_i, X_i) \e_i^2 \right] \left\{ - z (z^2 - 3) /2 \right\}   		\\
		& \quad +  \phi(z)  \st_{\ti}^{-4} \E \left[ \h^{-1} \l^0_{\ti}(X_i)^4 (\e_i^4 - v(X_i)^2) \right] \left\{ z(z^2-3)/8 \right\}   		\\
		& \quad -  \phi(z)  \st_{\ti}^{-2} \E \left[ \h^{-1}\l^0_{\ti}(X_i)^2 \br_{\tilde{p}}(X_{d_n,i})'\bm{\tilde{G}}^{-1} (K \br_{\tilde{p}})(X_{d_n,i}) \e_i^2 \right] \left\{ z(z^2 - 1)/2 \right\}   		\\
		& \quad -  \phi(z)  \st_{\ti}^{-4} \E \left[ \h^{-1} \l^0_{\ti} (X_i)^3  \br_{\tilde{p}}(X_{d_n,i})'\bm{\tilde{G}}^{-1} \e_i^2 \right] \E \left[ \h^{-1} (K \br_{\tilde{p}})(X_{d_n,i}) \l^0_{\ti} (X_i) \e_i^2 \right] \left\{ z(z^2 - 1)  \right\}   		\\
		& \quad +  \phi(z)  \st_{\ti}^{-2} \E \left[ \h^{-2} \l^0_{\ti}(X_i)^2 (\br_{\tilde{p}}(X_{d_n,i})'\bm{\tilde{G}}^{-1}(K \br_{\tilde{p}})(X_{d_n,j}) )^2 \e_j^2 \right] \left\{ z(z^2 - 1)/4 \right\}   		\\
		%	& \quad +   \phi(z)   \st_{\ti}^{-4} \E \left[ \h^{-1} \l^0_{\ti} (X_j)^2 \left( \E \left[ \h^{-1} \br_{\tilde{p}}(X_{d_n,j})'\bm{\tilde{G}}^{-1}(K \br_{\tilde{p}})(X_{d_n,i}) \l^0_{\ti}(X_i)  \e_i^2 \vert X_j \right] \right)^2 \right] \left\{  z(z^2 - 1) /2 \right\}   		\\ %%this is the old version of the line right below, with the conditional expectation squared
		& \quad +   \phi(z)   \st_{\ti}^{-4} \E \left[ \h^{-3} \l^0_{\ti} (X_j)^2 \br_{\tilde{p}}(X_{d_n,j})'\bm{\tilde{G}}^{-1}(K \br_{\tilde{p}})(X_{d_n,i}) \l^0_{\ti}(X_i) \br_{\tilde{p}}(X_{d_n,j})'\bm{\tilde{G}}^{-1}(K \br_{\tilde{p}})(X_{d_n,k}) \l^0_{\ti}(X_k)  \e_i^2 \e_k^2 \right]   		\\
		& \quad \qquad \qquad \qquad \qquad \qquad   \times \;   \left\{  z(z^2 - 1) /2 \right\}   		\\
		& \quad +  \phi(z)  \st_{\ti}^{-4} \E \left[ \h^{-1} \l^0_{\ti}(X_i)^4 \e_i^4 \right] \left\{ - z (z^2 - 3)/24 \right\}   		\\
		& \quad +  \phi(z)  \st_{\ti}^{-4} \E \left[ \h^{-1} \left( \l^0_{\ti}(X_i)^2 v(X_i) - \E[\l^0_{\ti}(X_i)^2 v(X_i)] \right) \l^0_{\ti}(X_i)^2 \e_i^2 \right] \left\{ z(z^2 - 1)/4 \right\}   		\\
		& \quad +  \phi(z)  \st_{\ti}^{-4} \E \left[ \h^{-2} \l^1_{\ti}(X_i, X_j) \l^0_{\ti}(X_i)\l^0_{\ti}(X_j)^2 \e_j^2 v(X_i) \right] \left\{  z (z^2 - 3) \right\}   		\\
		& \quad +  \phi(z)  \st_{\ti}^{-4} \E \left[ \h^{-2} \l^1_{\ti}(X_i, X_j)  \l^0_{\ti}(X_i) \left( \l^0_{\ti}(X_j)^2 v(X_j) - \E[\l^0_{\ti}(X_j)^2 v(X_j)] \right) \e_i^2 \right] \left\{ - z \right\}   		\\
		& \quad +  \phi(z)  \st_{\ti}^{-4} \E \left[ \h^{-1}  \left( \l^0_{\ti}(X_i)^2 v(X_i) - \E[\l^0_{\ti}(X_i)^2 v(X_i)] \right)^2 \right] \left\{ - z(z^2 + 1) /8 \right\}   .
	\end{align*}

\end{appendix}
%%%%%%%%%%%%%%%%%%%%%%%%%%%%%%%%%%%%%%%%%%%%%%
%% Multiple Appendixes:                     %%
%%%%%%%%%%%%%%%%%%%%%%%%%%%%%%%%%%%%%%%%%%%%%%
%\begin{appendix}
%\section{???}
%
%\section{???}
%
%\end{appendix}

%%%%%%%%%%%%%%%%%%%%%%%%%%%%%%%%%%%%%%%%%%%%%%
%% Support information (funding), if any,   %%
%% should be provided in the                %%
%% Acknowledgments section.                %%
%%%%%%%%%%%%%%%%%%%%%%%%%%%%%%%%%%%%%%%%%%%%%%
\section*{Acknowledgements}

We especially thank an Associate Editor, and the reviewers, for insightful comments that improve our manuscript. We also thank Chris Hansen, Michael Jansson, Adam McCloskey, Rocio Titiunik, and participants at various seminars and conferences for comments. The second author gratefully acknowledges financial support from the National Science Foundation (SES 1357561, SES 1459931 and SES-1947805). The third author gratefully acknowledges financial support from the Richard N. Rosett and John E. Jeuck Fellowships.

%%%%%%%%%%%%%%%%%%%%%%%%%%%%%%%%%%%%%%%%%%%%%%
%% Supplementary Material, if any, should   %%
%% be provided in {supplement} environment  %%
%% with title and short description.        %%
%%%%%%%%%%%%%%%%%%%%%%%%%%%%%%%%%%%%%%%%%%%%%%
\begin{supplement}
	\stitle{Online Appendix for ``Coverage Error Optimal Confidence Intervals for Local Polynomial Regression''}
	\sdescription{This supplement contains proofs of all results, other technical details, and complete simulation results.}
\end{supplement}

%%  imsart-number.bst  will be used to                     %%
%%  MR numbers will be added by VTeX.                      %%
%%                                                         %%
%%  Use \cite{...} to cite references in text.             %%
%%                                                         %%
%%%%%%%%%%%%%%%%%%%%%%%%%%%%%%%%%%%%%%%%%%%%%%%%%%%%%%%%%%%%%

%% if your bibliography is in bibtex format, uncomment commands:
\bibliographystyle{imsart-number} % Style BST file
\bibliography{Edgeworth-Uniform--Bibliography}       % Bibliography file (usually '*.bib')

%% or include bibliography directly:
% \begin{thebibliography}{}
% \bibitem{b1}
% \end{thebibliography}

%%%%%%%%%%%%%%%%%%%%%%%%%%%%%%%%%%%
%%%%%%%%%%%%%%%%%%%%%%%%%%%%%%%%%%%
%%%%%%%%%%%%%%%%%%%%%%%%%%%%%%%%%%%
\clearpage

\begin{table}
	\caption{$L_2$-Optimal Variance-Minimizing $\rho$\label{table:rho}}
	\begin{subtable}{.5\linewidth}
		\centering
		%latex.default(table1, file = paste("simuls/table_rho_bnd_L",     L.m, ".txt", sep = ""), n.cgroup = c(1, 3), cgroup = c("$p$",     "Kernel"), landscape = FALSE, center = "none", col.just = rep("c",     4), title = "", table.env = FALSE)%
		\begin{tabular}{ccccc}
			\multicolumn{1}{c}{\bfseries $p$}&\multicolumn{1}{c}{\bfseries }&\multicolumn{3}{c}{\bfseries Kernel}\tabularnewline
			\multicolumn{1}{c}{}&\multicolumn{1}{c}{}&\multicolumn{1}{c}{Triangular}&\multicolumn{1}{c}{Epanechnikov}&\multicolumn{1}{c}{Uniform}\tabularnewline
			$0$&&$0.778$&$0.846$&$1.000$\tabularnewline
			$1$&&$0.850$&$0.898$&$1.000$\tabularnewline
			$2$&&$0.887$&$0.924$&$1.000$\tabularnewline
			$3$&&$0.909$&$0.940$&$1.000$\tabularnewline
			$4$&&$0.924$&$0.950$&$1.000$\tabularnewline
		\end{tabular}
		
		\smallskip
		\caption{Boundary point}
	\end{subtable}%
	\begin{subtable}{.5\linewidth}
		\centering
		%latex.default(table2, file = paste("simuls/table_rho_int_L",     L.m, ".txt", sep = ""), n.cgroup = c(1, 3), cgroup = c("$p$",     "Kernel"), landscape = FALSE, center = "none", col.just = rep("c",     4), title = "", table.env = FALSE)%
		\begin{tabular}{ccccc}
			\multicolumn{1}{c}{\bfseries $p$}&\multicolumn{1}{c}{\bfseries }&\multicolumn{3}{c}{\bfseries Kernel}\tabularnewline
			\multicolumn{1}{c}{}&\multicolumn{1}{c}{}&\multicolumn{1}{c}{Triangular}&\multicolumn{1}{c}{Epanechnikov}&\multicolumn{1}{c}{Uniform}\tabularnewline
			$1$&&$0.798$&$0.865$&$1.000$\tabularnewline
			$3$&&$0.867$&$0.915$&$1.000$\tabularnewline
			$5$&&$0.900$&$0.938$&$1.000$\tabularnewline
			$7$&&$0.919$&$0.951$&$1.000$\tabularnewline
		\end{tabular}
		\smallskip
		\caption{Interior point}
	\end{subtable} 
	\vspace{.1in}
	\footnotesize\textbf{Note}: Optimal $\rho$ computed by minimizing the $L_2$ distance between the RBC induced equivalent kernel and the variance-minimizing equivalent kernel (Uniform Kernel).
\end{table}

\clearpage
\begin{figure}[!htb]
	\captionsetup[subfigure]{labelformat=empty}
	\centering
	\caption{$\mathcal{K}^*_{p+1}(u)$ vs. $\mathcal{K}_\RBC(u; K, \rho^*, \v)$}\label{fig:equivK}
	\subcaption{(a) $\v=0$ }
	\begin{subfigure}[b]{0.5\textwidth}
		\includegraphics[height=0.2\textheight,width=0.75\textwidth]{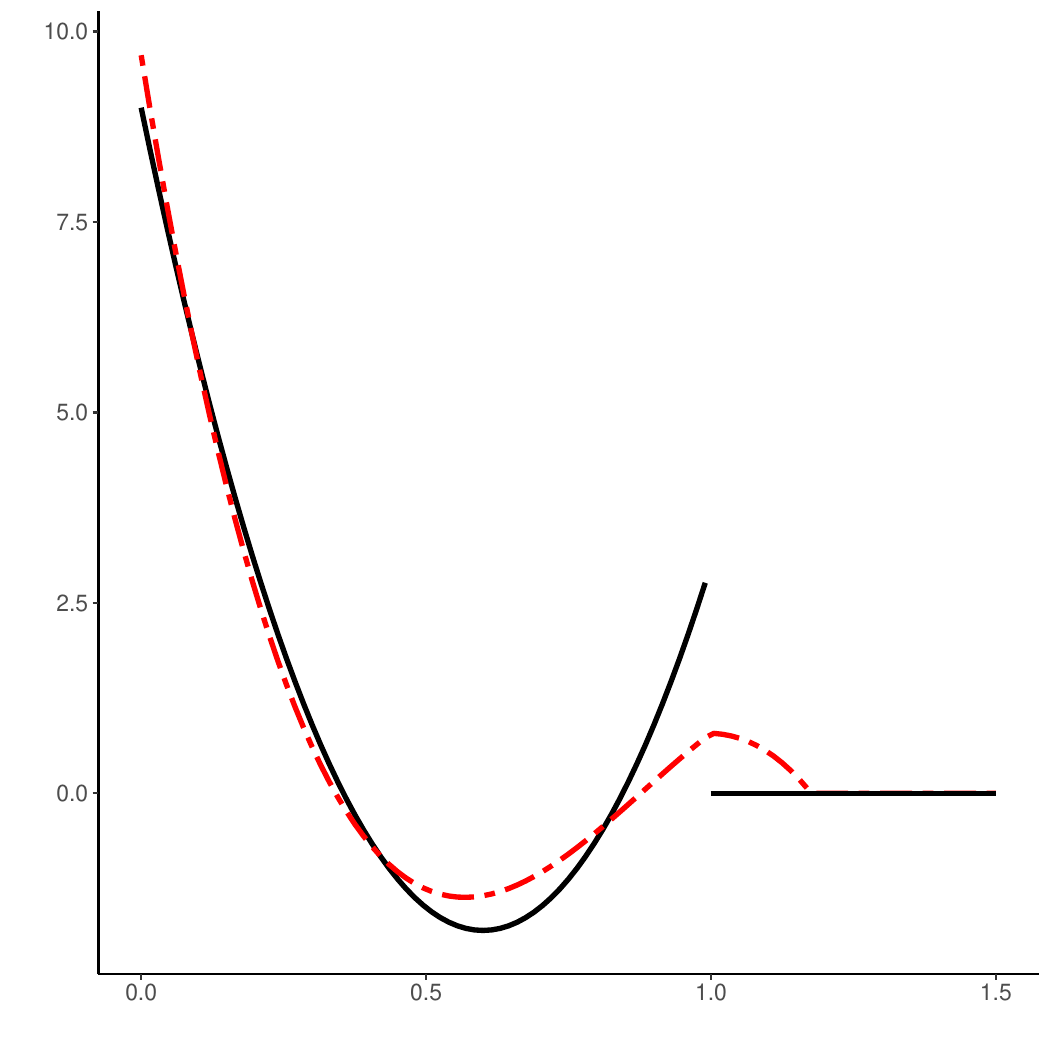}
		\caption{($i$) Triangular Kernel, Boundary Point, $p=1$}
	\end{subfigure}%
	\begin{subfigure}[b]{0.5\textwidth}
		\includegraphics[height=0.2\textheight,width=0.75\textwidth]{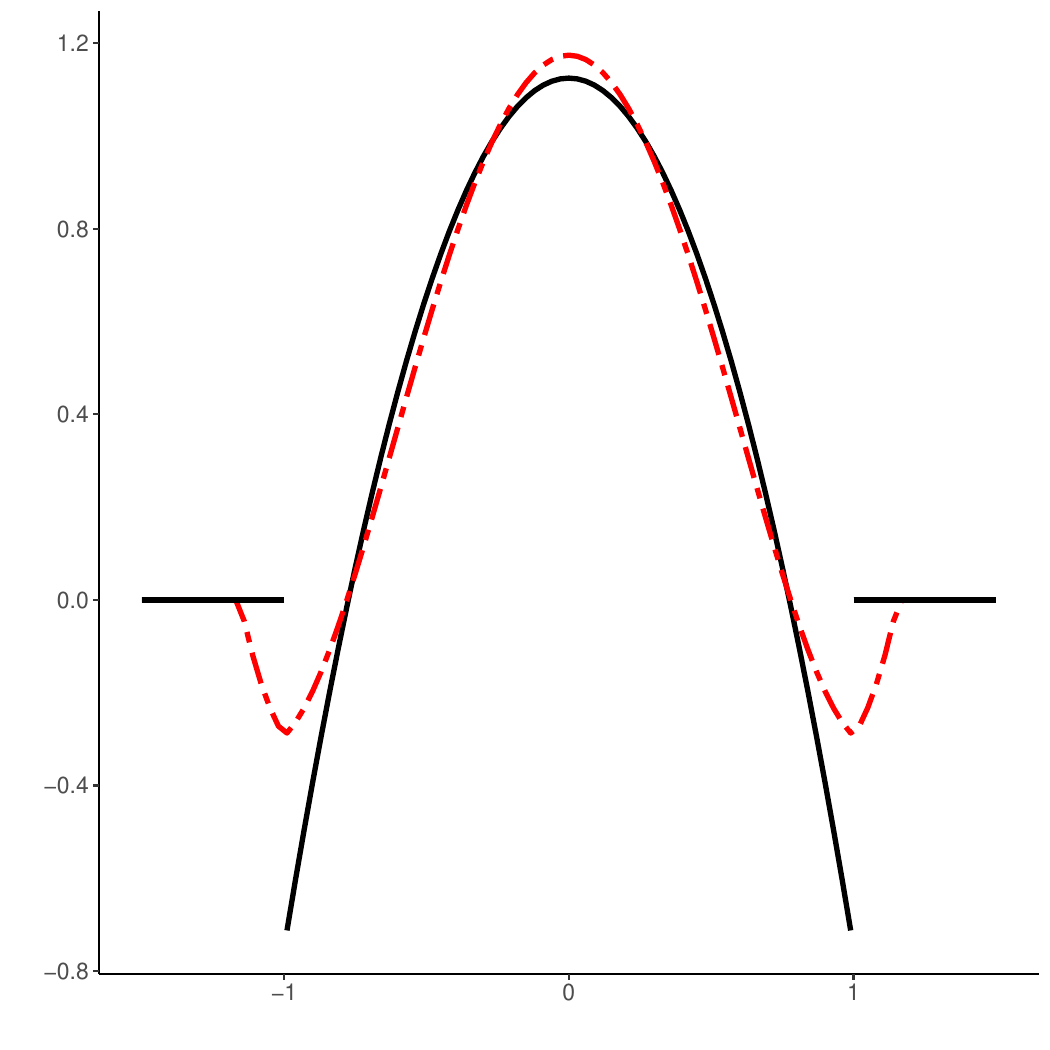}
		\caption{($ii$) Epanechnikov Kernel, Interior Point, $p=1$}
	\end{subfigure}%
	
	\subcaption{(b) $\v=1$}
	\begin{subfigure}[b]{0.5\textwidth}
		\includegraphics[height=0.2\textheight,width=0.75\textwidth]{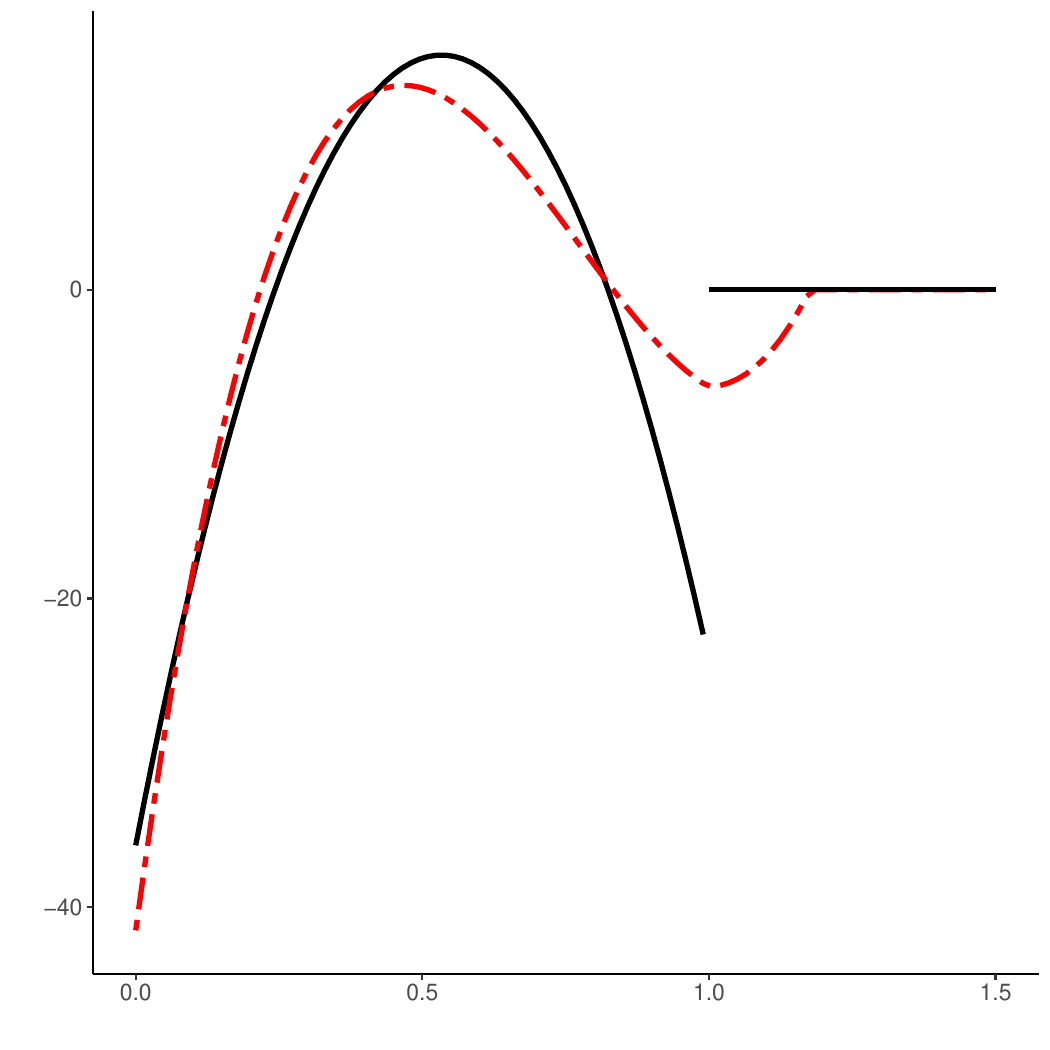}
		\caption{($iii$) Triangular Kernel, Boundary Point, $p=1$}
	\end{subfigure}%
	\begin{subfigure}[b]{0.5\textwidth}
		\includegraphics[height=0.2\textheight,width=0.75\textwidth]{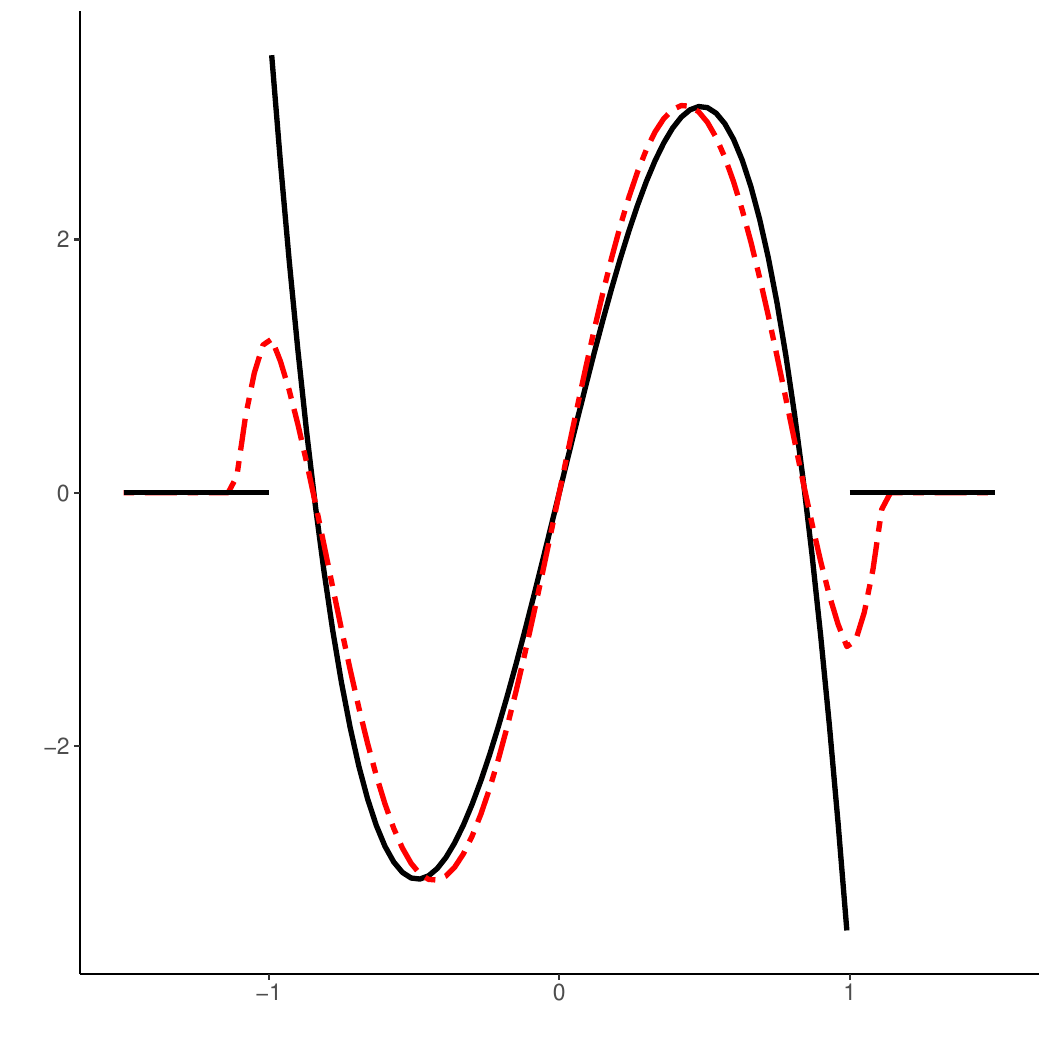}
		\caption{($iv$) Epanechnikov Kernel, Interior Point, $p=2$}
	\end{subfigure}%
	\begin{flushleft}\footnotesize Notes: \blackline $\mathcal{K}^*_{p+1}(u)$, \redline $\mathcal{K}_\RBC(u; K, \rho^*, \v)$\end{flushleft}
\end{figure}

%%%%%%%%%%%%%%%%%%%%%%%%%%%%%%%%%%%%%%%%%%%%%%%%%%%%%%%%%%%%%%%%%
\begin{figure}[!htb]
	\centering
	\caption{Conditional mean function and first derivative, $\mu^{(\v)}(x)$}\label{fig:dgp}	\begin{subfigure}[b]{0.5\textwidth}
		\includegraphics[height=0.25\textheight,width=0.85\textwidth]{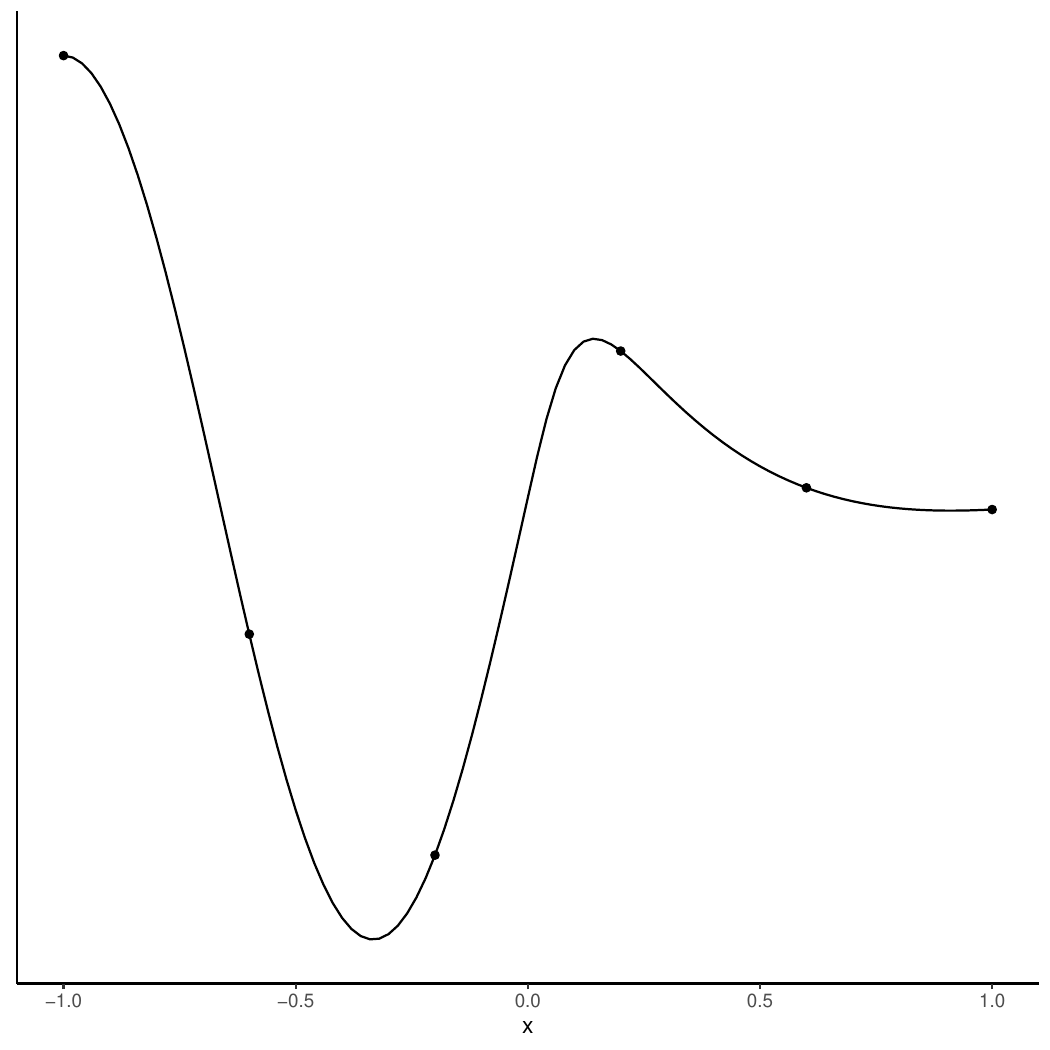}
		\subcaption{$\v=0$}
	\end{subfigure}%
	\begin{subfigure}[b]{0.5\textwidth}
		\includegraphics[height=0.25\textheight,width=0.85\textwidth]{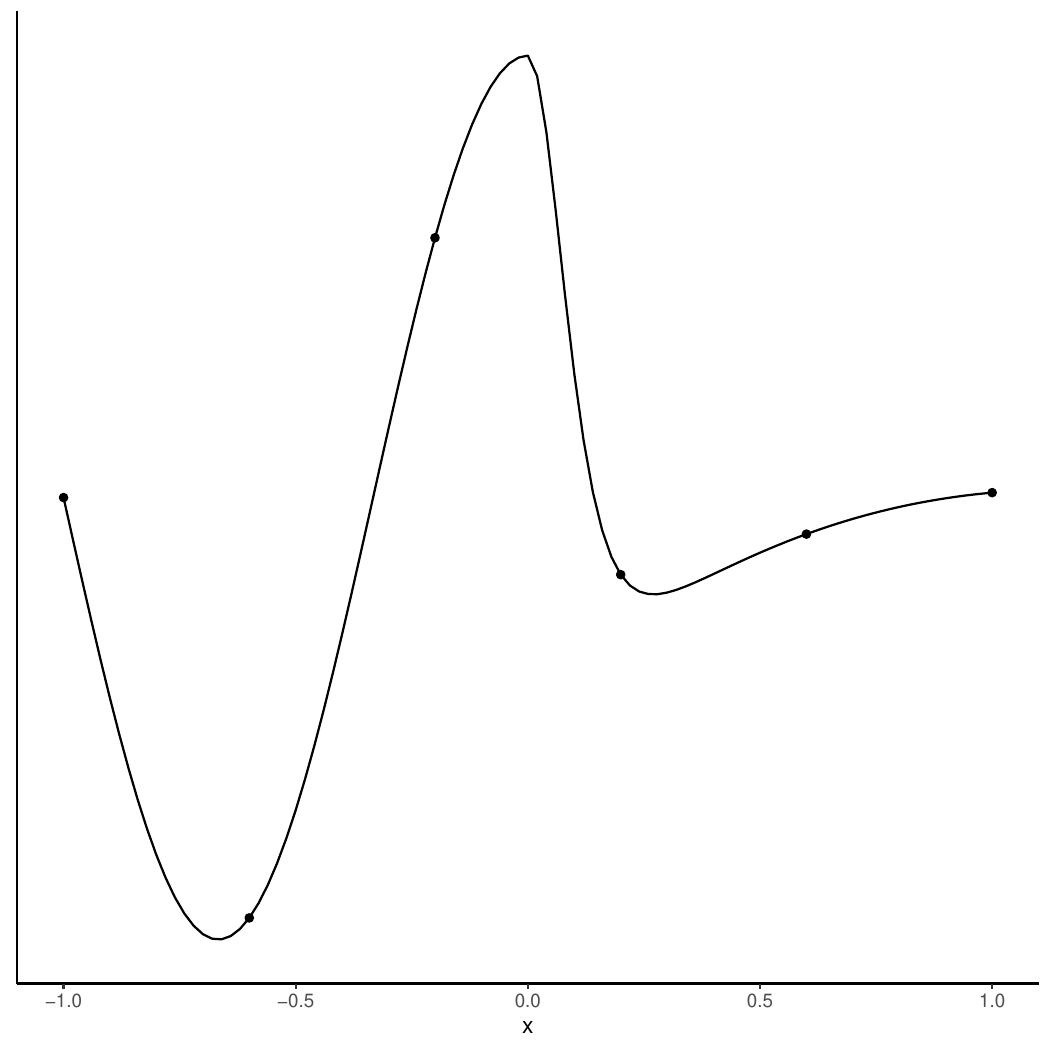}
		\subcaption{$\v=1$}
	\end{subfigure}%
\end{figure}

%%%%%%%%%%%%%%%%%%%%%%%%%%%%%%%%%%%%%%%%%%%%%%%%%%%%%%%%%%%%%%%%%

\clearpage

\begin{figure}[!htb]
	\captionsetup[subfigure]{labelformat=empty}
	\centering
	\caption{Empirical Coverage for 95\% Confidence Intervals, $\v=0$}
	\label{fig:sim_ec0}
	\subcaption{(a) $\hat{h}_{\RBC}$}
	\begin{subfigure}[b]{0.3\textwidth}
		\includegraphics[width=\linewidth]{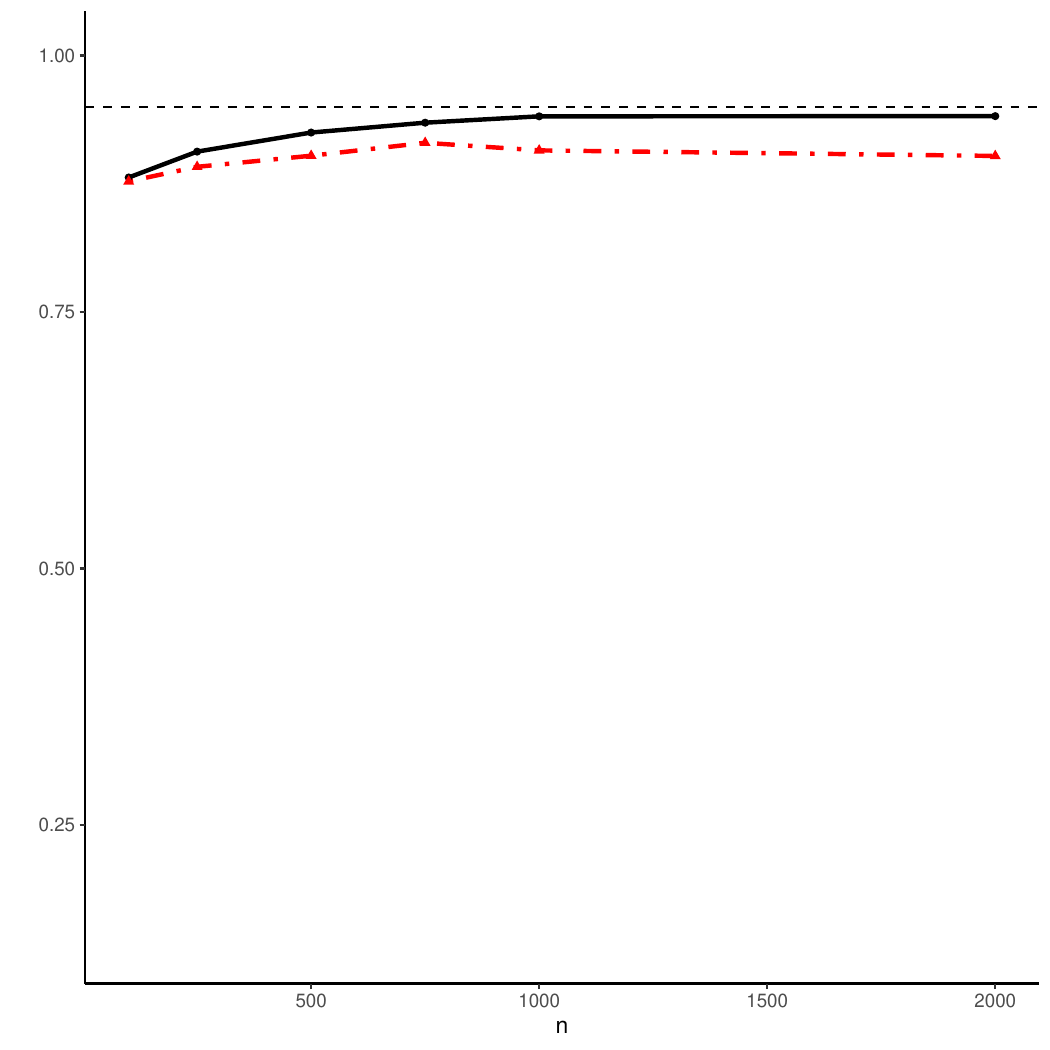}
		\subcaption{$(i)$ $\x=-1$}
	\end{subfigure}%\hfill
	\begin{subfigure}[b]{0.3\textwidth}
		\includegraphics[width=\linewidth]{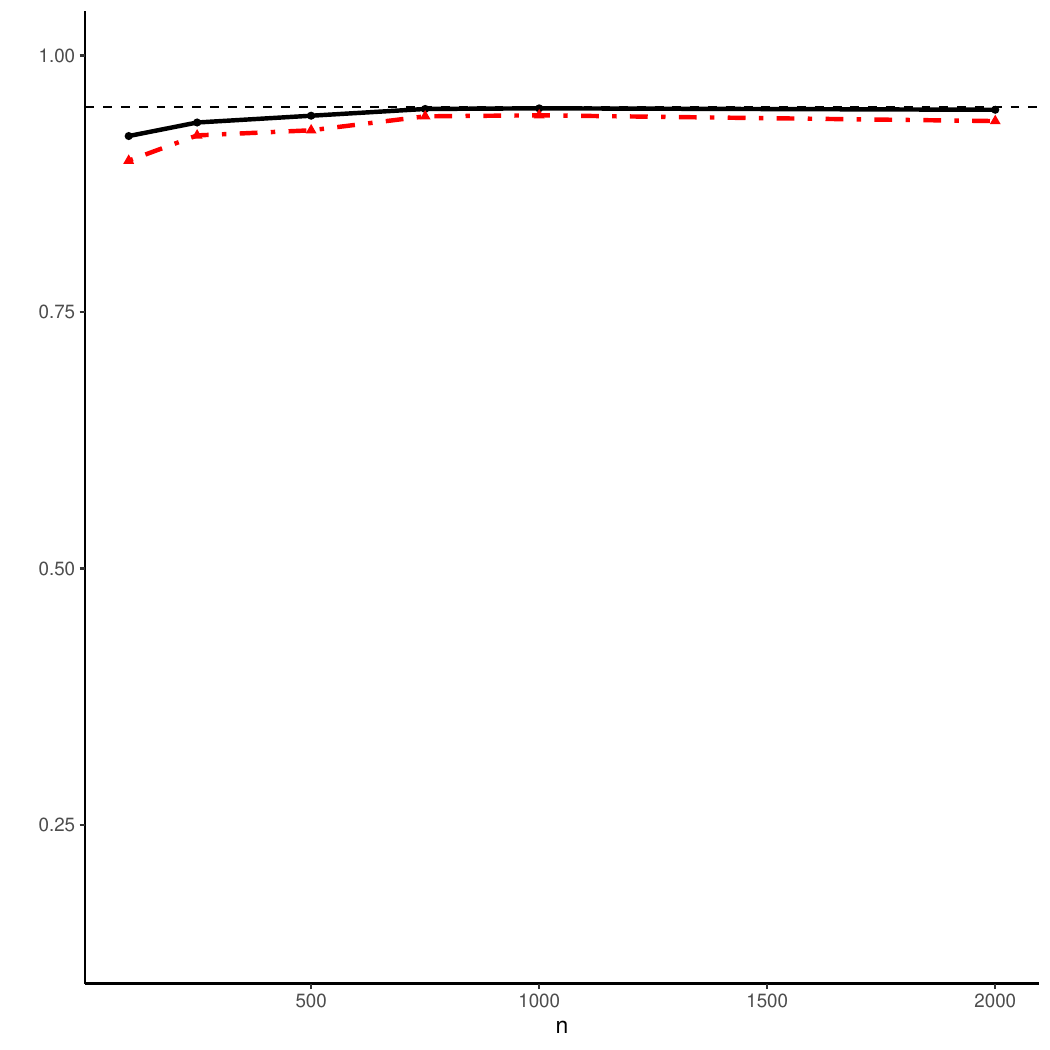}
		\subcaption{$(ii)$ $\x=-0.6$}
	\end{subfigure}%\hfill
	\begin{subfigure}[b]{0.3\textwidth}
		\includegraphics[width=\linewidth]{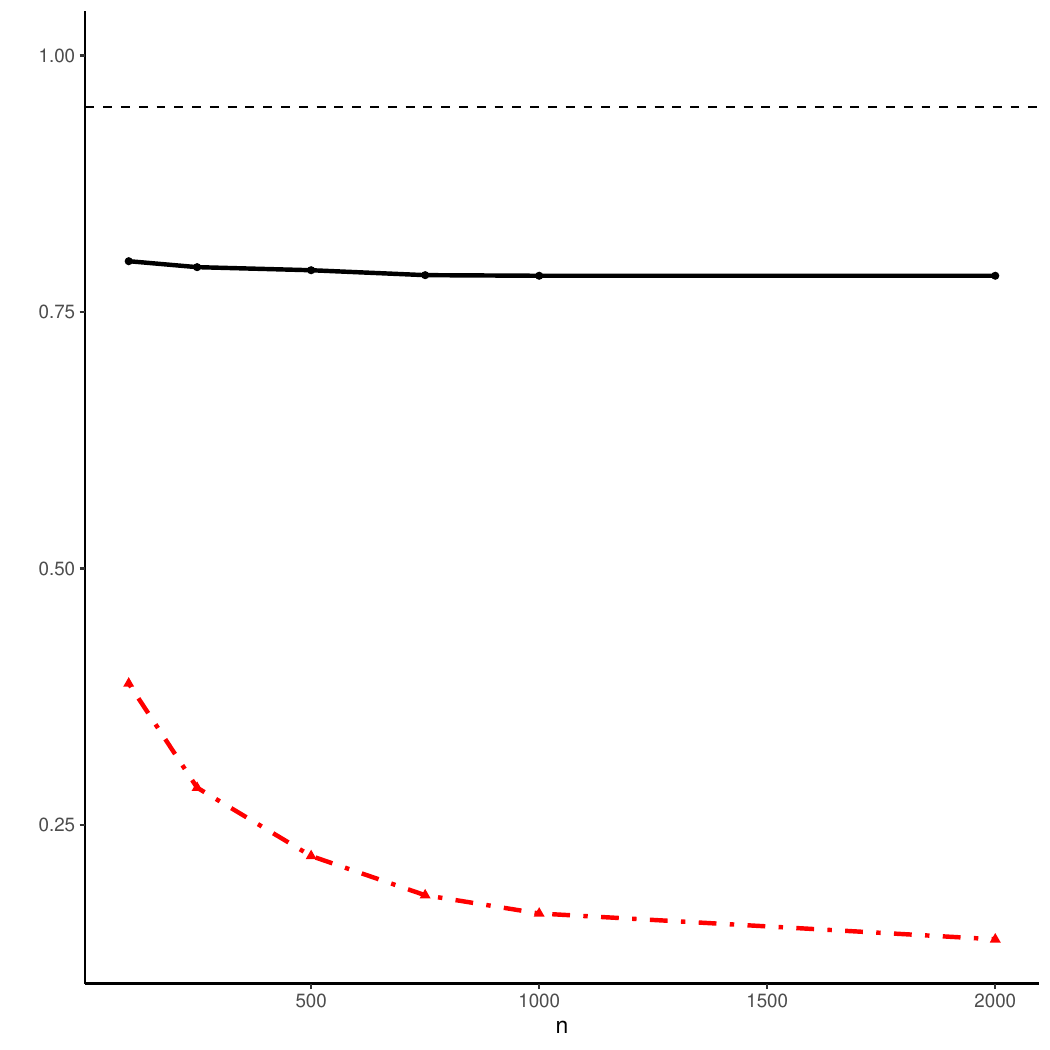}
		\subcaption{$(iii)$ $\x=-0.2$}
	\end{subfigure}
	
	\subcaption{(b) $\hat{h}_{\US}$}
	\begin{subfigure}[b]{0.3\textwidth}
		\includegraphics[width=\linewidth]{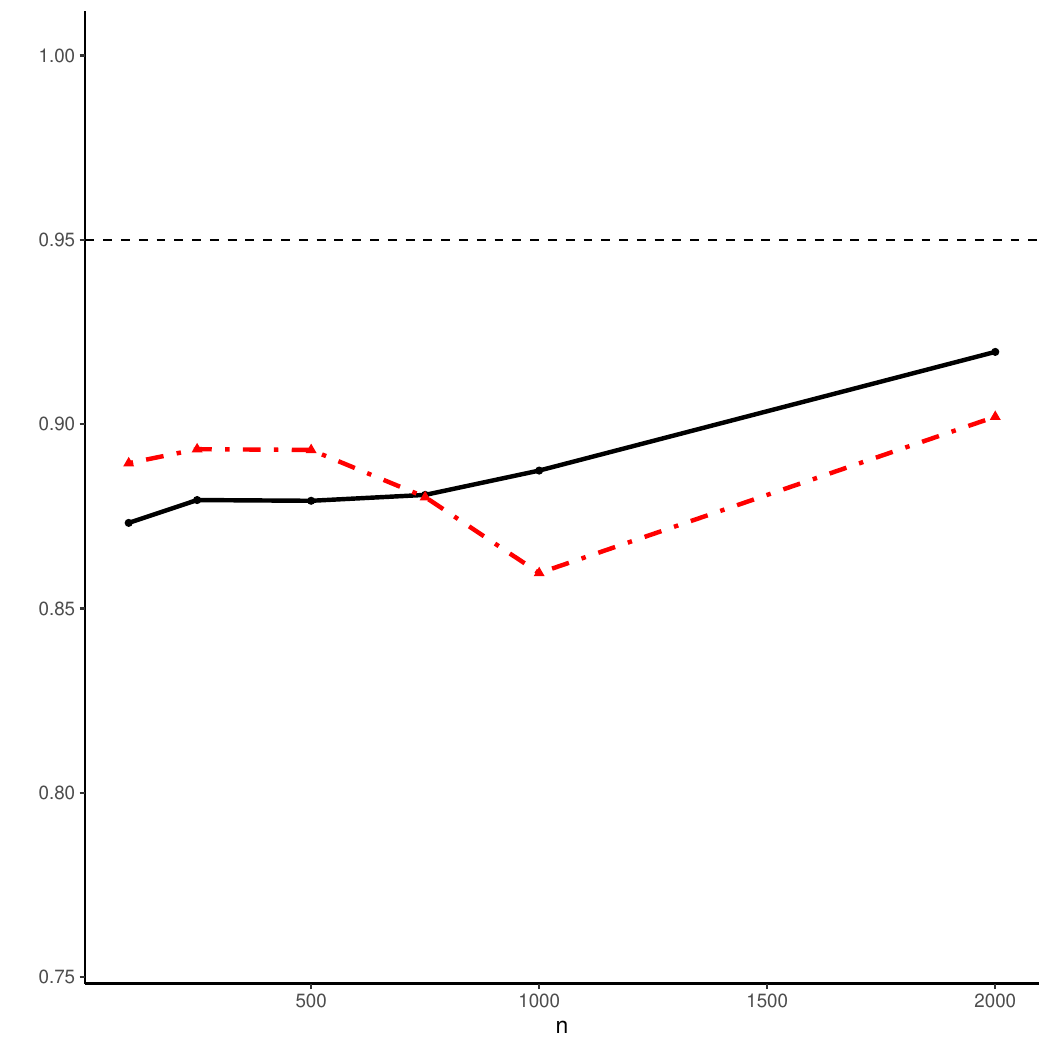}
		\subcaption{$(iv)$ $\x=-1$}
	\end{subfigure}%\hfill
	\begin{subfigure}[b]{0.3\textwidth}
		\includegraphics[width=\linewidth]{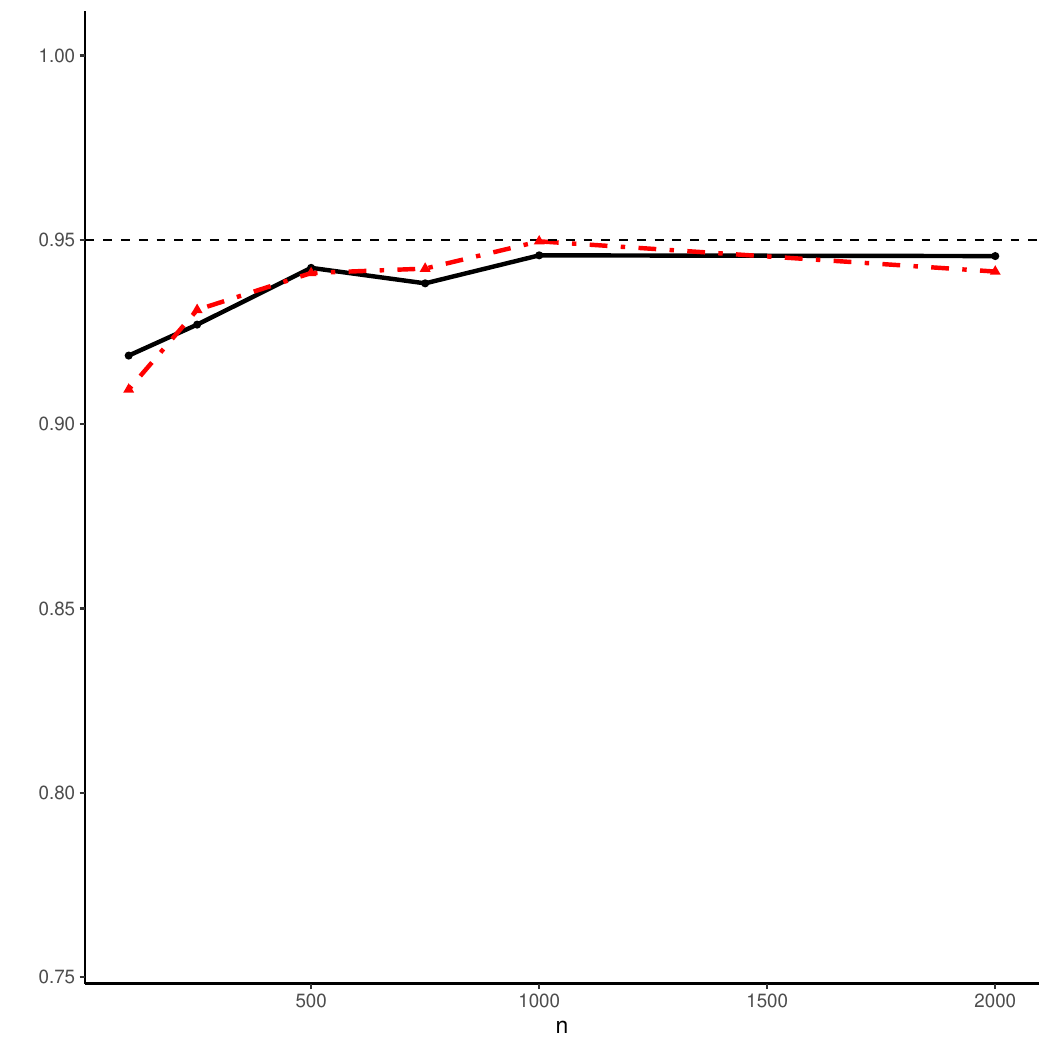}
		\subcaption{$(v)$ $\x=-0.6$}
	\end{subfigure}%\hfill
	\begin{subfigure}[b]{0.3\textwidth}
		\includegraphics[width=\linewidth]{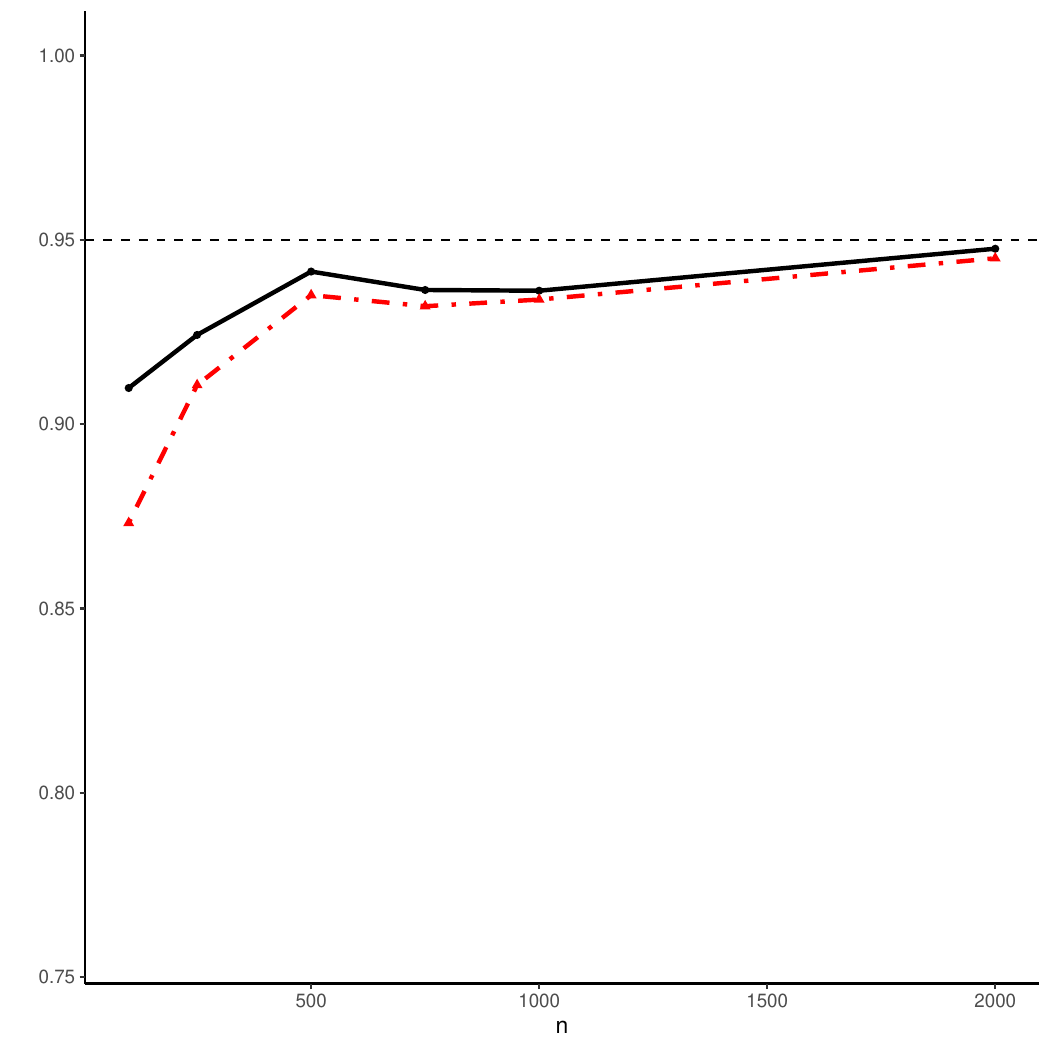}
		\subcaption{$(vi)$ $\x=-0.2$}
	\end{subfigure}
	
	\subcaption{(c) $\hat{h}_{\MSE}$}
	\begin{subfigure}[b]{0.3\textwidth}
		\includegraphics[width=\linewidth]{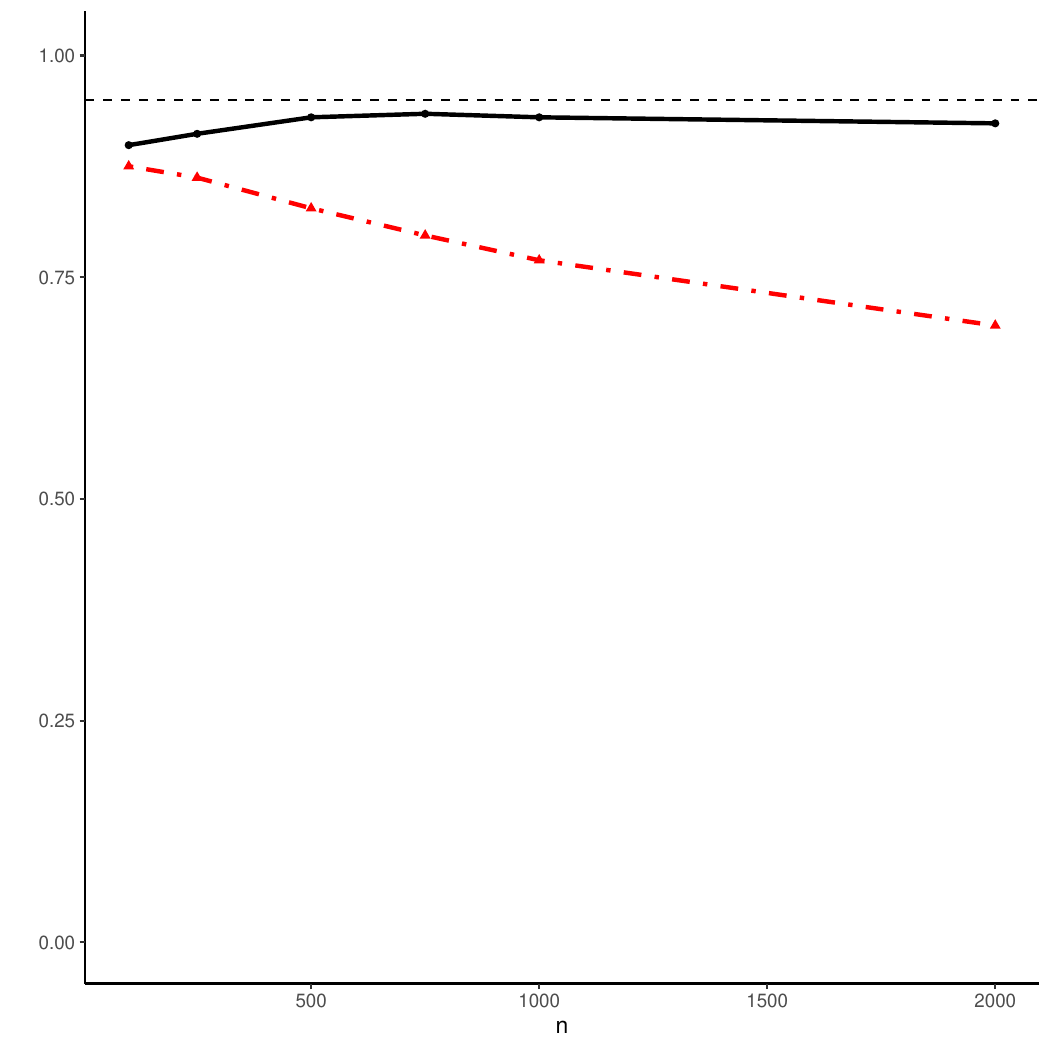}
		\subcaption{$(vii)$ $\x=-1$}
	\end{subfigure}%\hfill
	\begin{subfigure}[b]{0.3\textwidth}
		\includegraphics[width=\linewidth]{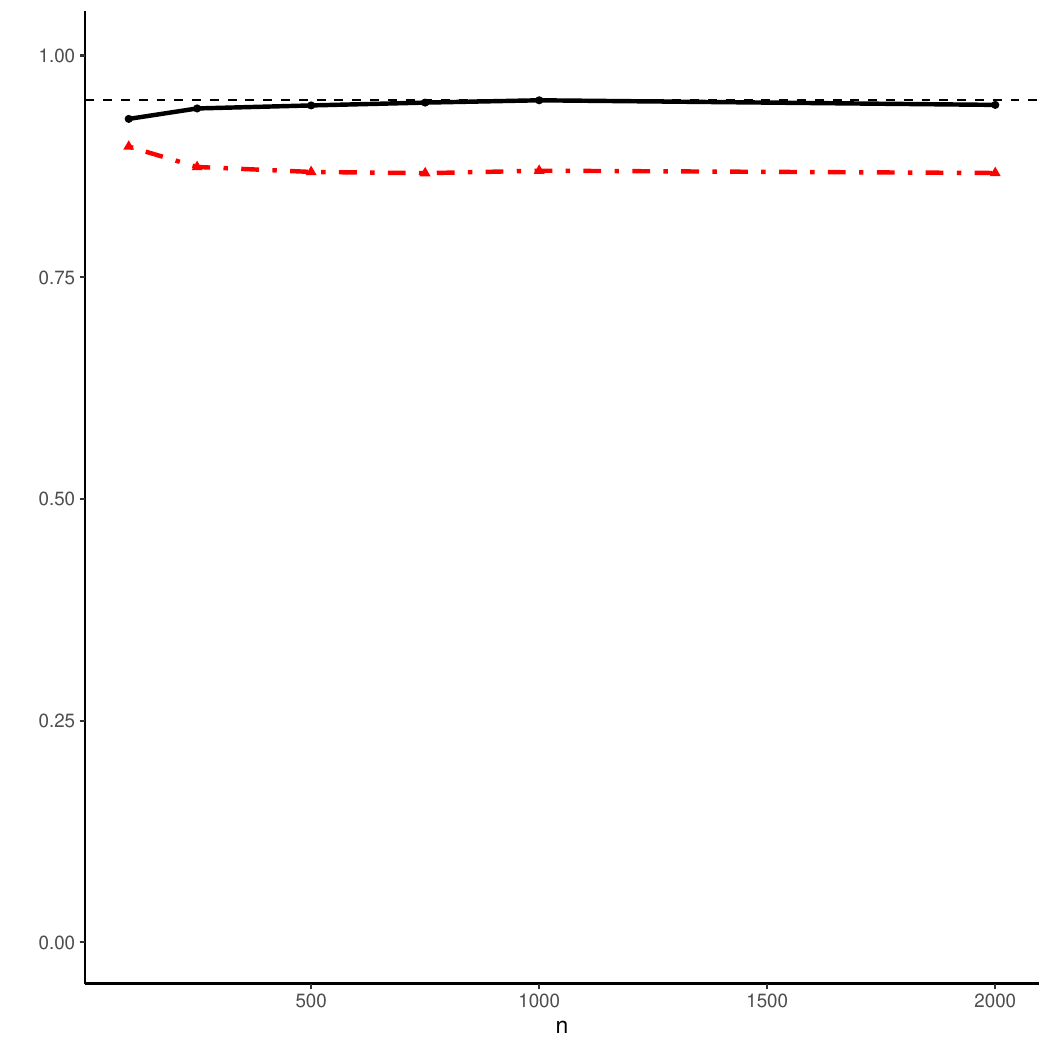}
		\subcaption{$(viii)$ $\x=-0.6$}
	\end{subfigure}%\hfill
	\begin{subfigure}[b]{0.3\textwidth}
		\includegraphics[width=\linewidth]{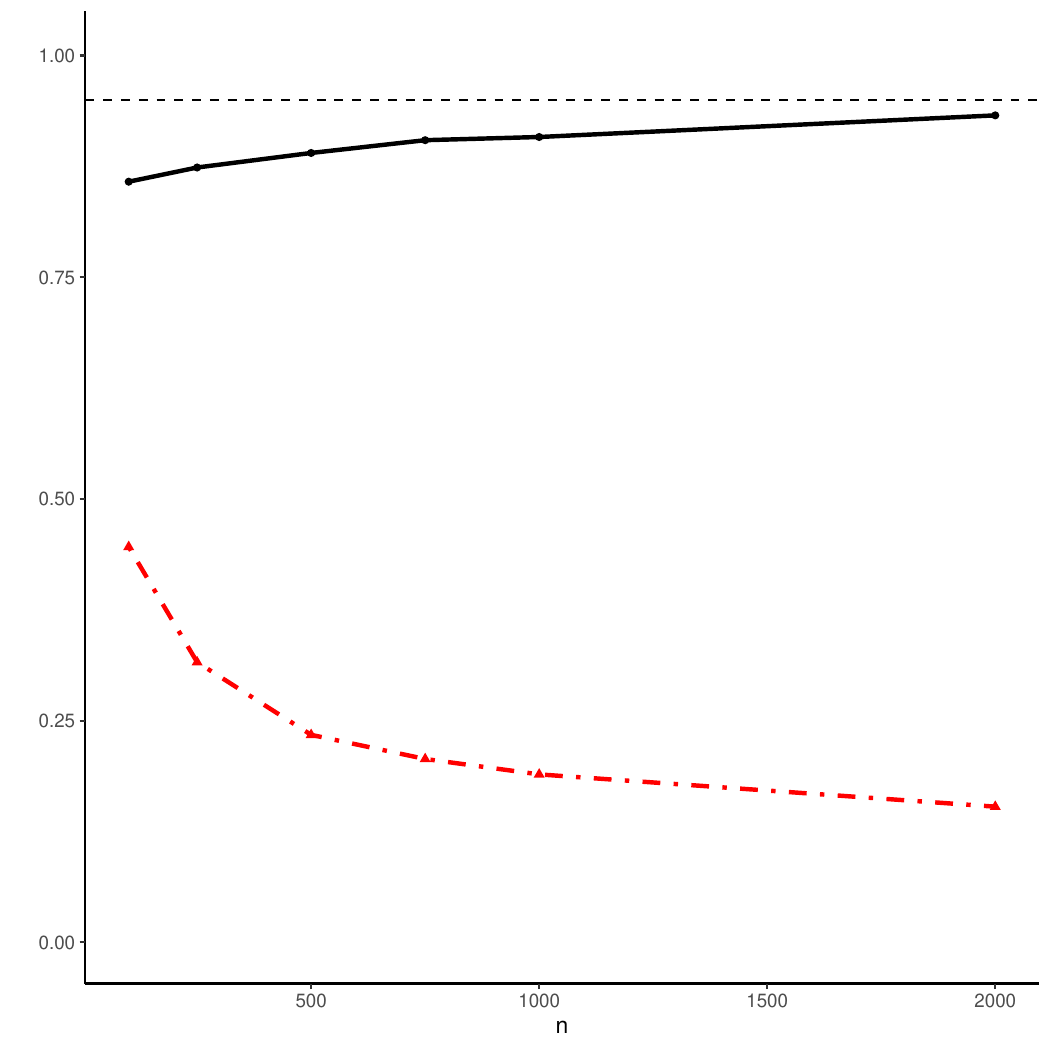}
		\subcaption{$(ix)$ $\x=-0.2$}
	\end{subfigure}%
	\begin{flushleft}\footnotesize Notes: \blackline Robust Bias Correction, \redline Undersmoothing; Epanechnikov Kernel
	\end{flushleft}
\end{figure}

\clearpage

\begin{figure}[!htb]
	\captionsetup[subfigure]{labelformat=empty}
	\centering
	\caption{Empirical Coverage for 95\% Confidence Intervals, $\v=1$}
	\label{fig:sim_ec1}
	\subcaption{(a) $\hat{h}_{\RBC}$}
	\begin{subfigure}[b]{0.3\textwidth}
		\includegraphics[width=\linewidth]{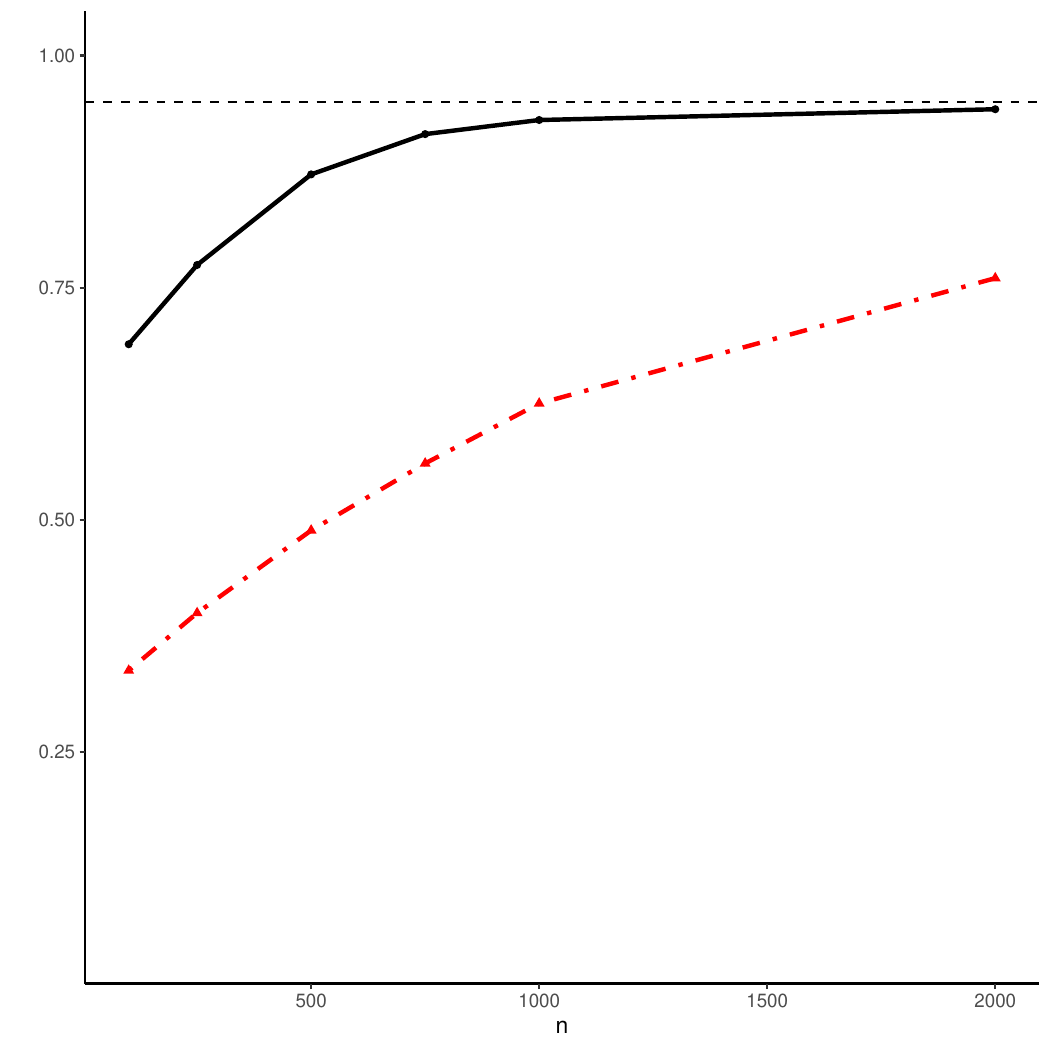}
		\subcaption{$(i)$ $\x=-1$}
	\end{subfigure}%\hfill
	\begin{subfigure}[b]{0.3\textwidth}
		\includegraphics[width=\linewidth]{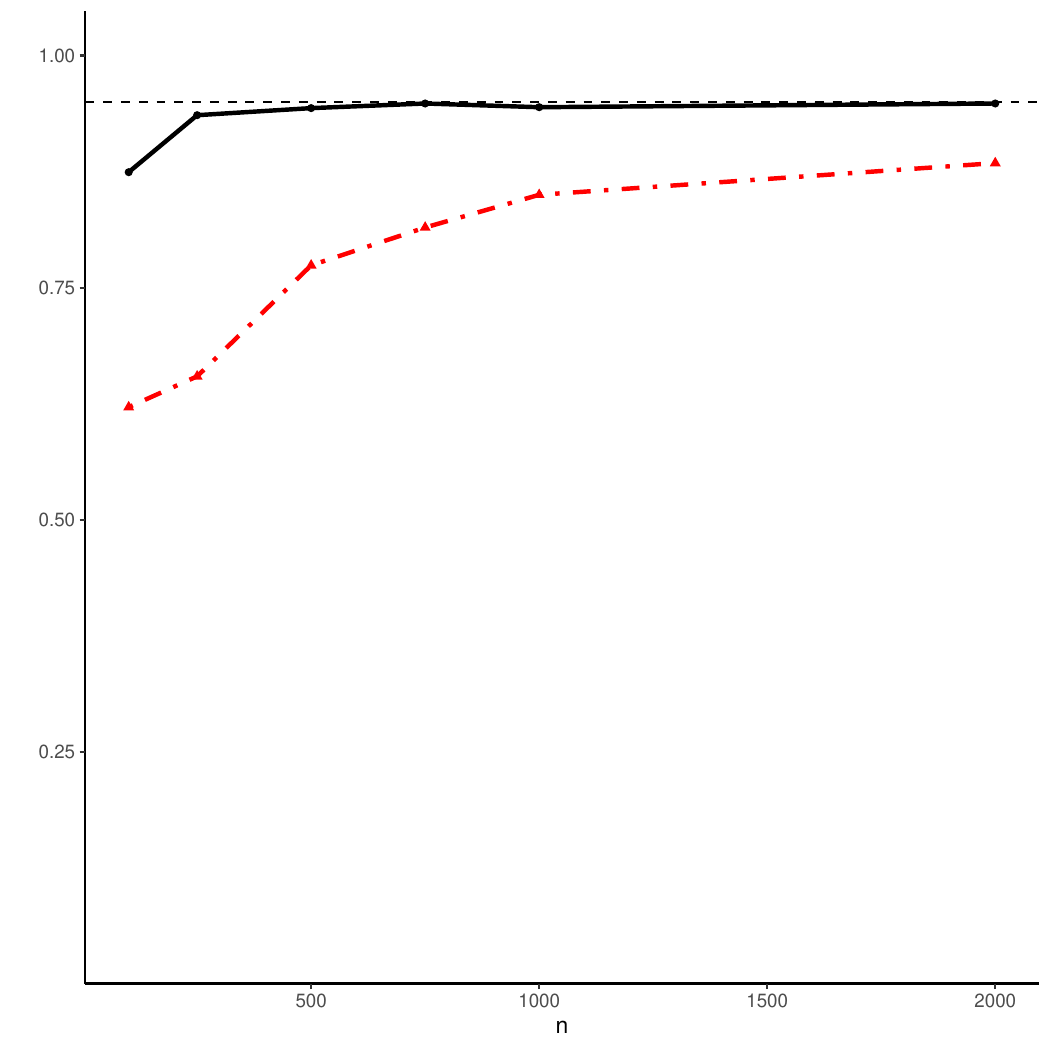}
		\subcaption{$(ii)$ $\x=-0.6$}
	\end{subfigure}%\hfill
	\begin{subfigure}[b]{0.3\textwidth}
		\includegraphics[width=\linewidth]{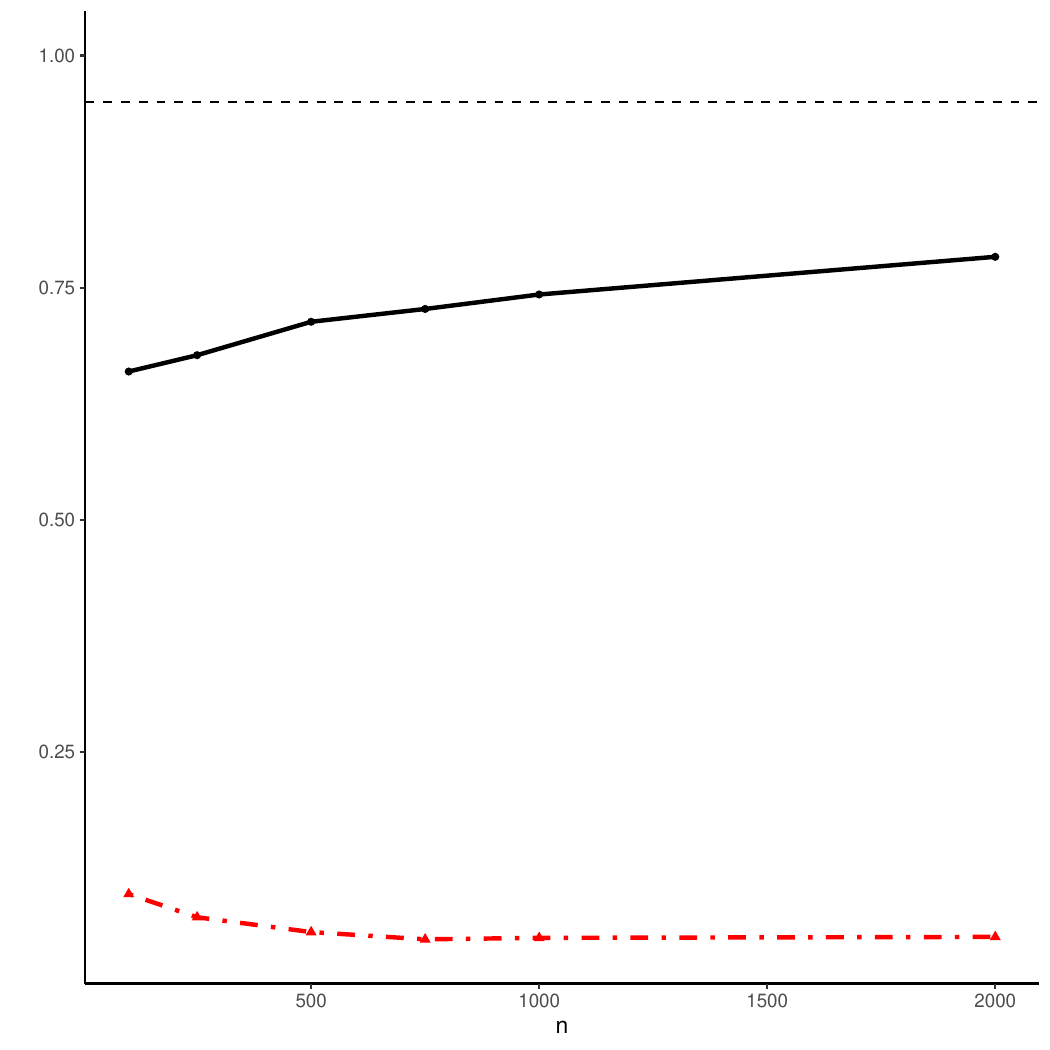}
		\subcaption{$(iii)$ $\x=-0.2$}
	\end{subfigure}
	
	\subcaption{(b) $\hat{h}_{\US}$}
	\begin{subfigure}[b]{0.3\textwidth}
		\includegraphics[width=\linewidth]{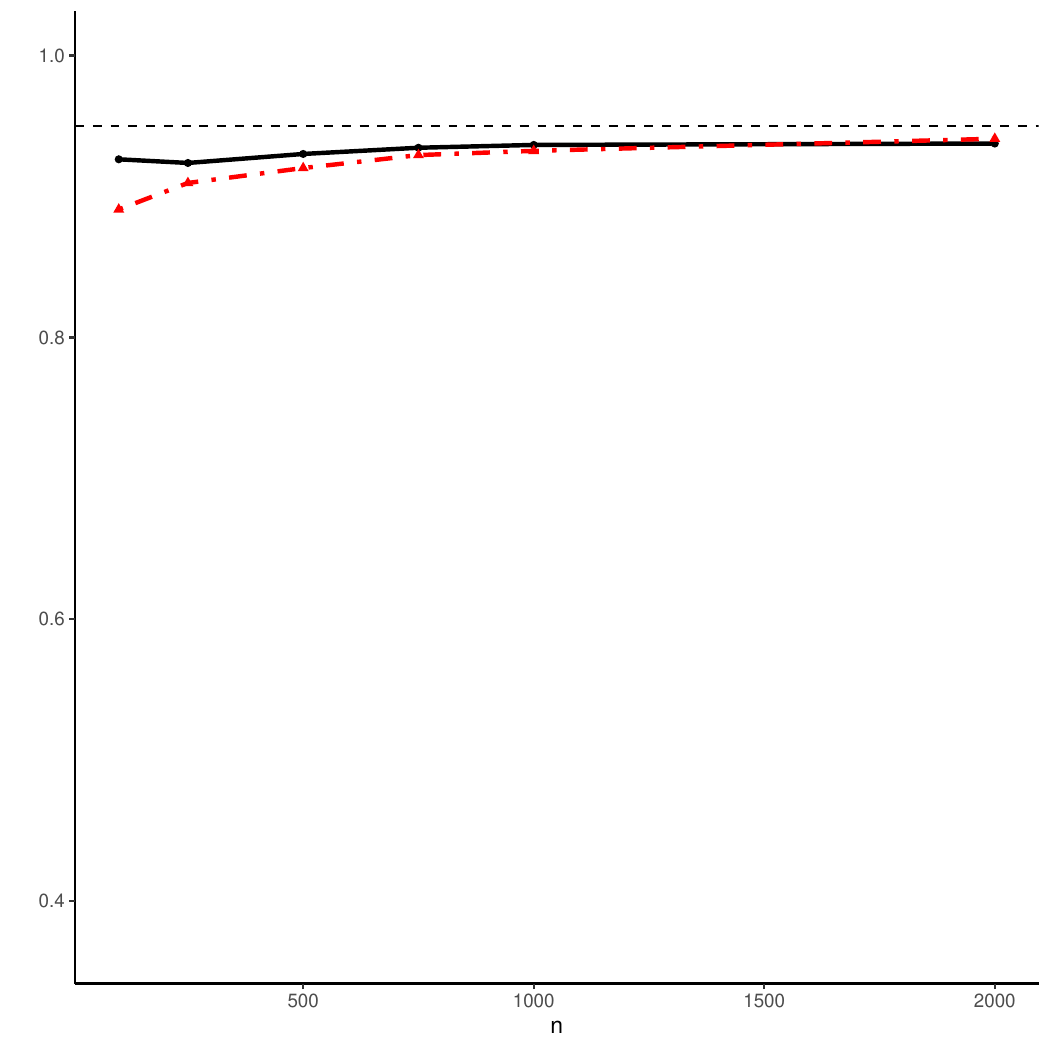}
		\subcaption{$(iv)$ $\x=-1$}
	\end{subfigure}%\hfill
	\begin{subfigure}[b]{0.3\textwidth}
		\includegraphics[width=\linewidth]{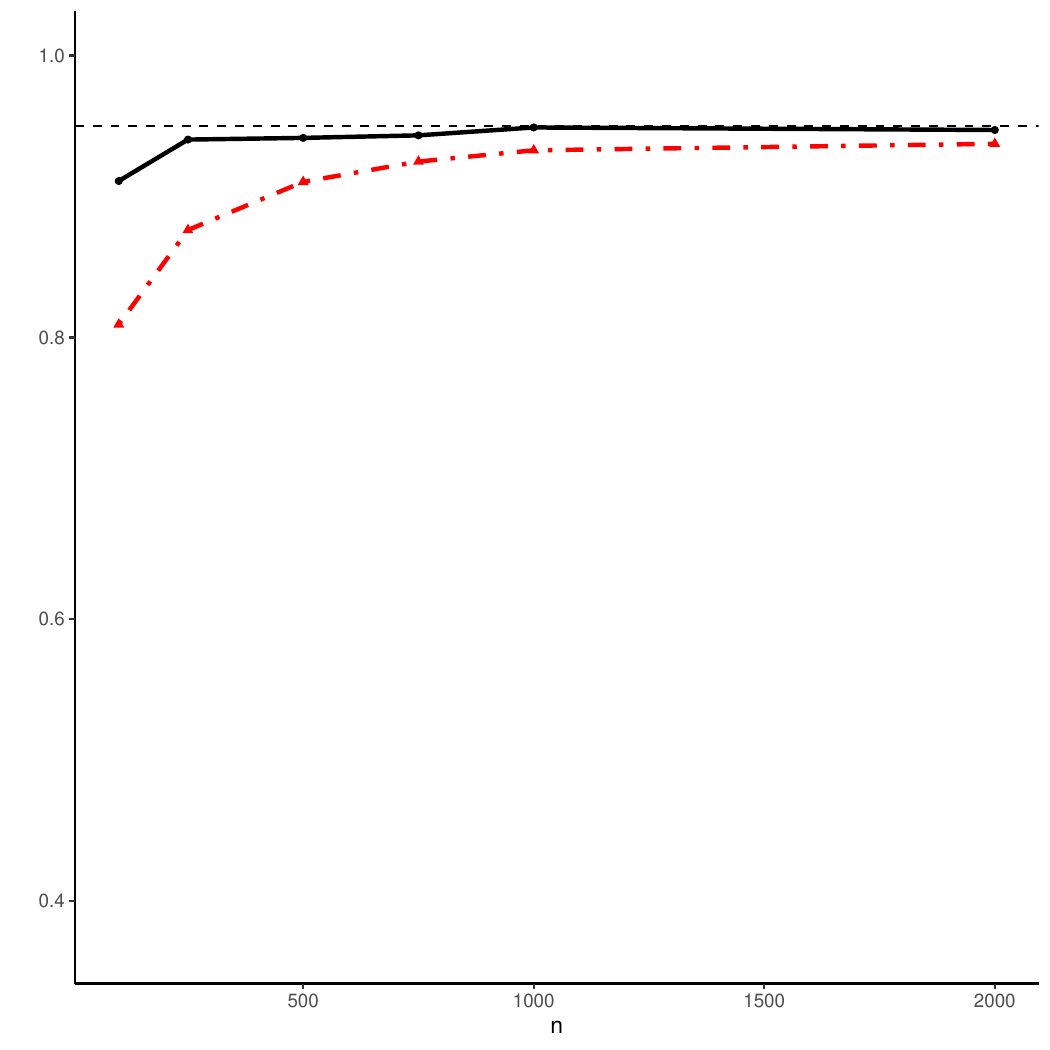}
		\subcaption{$(v)$ $\x=-0.6$}
	\end{subfigure}%\hfill
	\begin{subfigure}[b]{0.3\textwidth}
		\includegraphics[width=\linewidth]{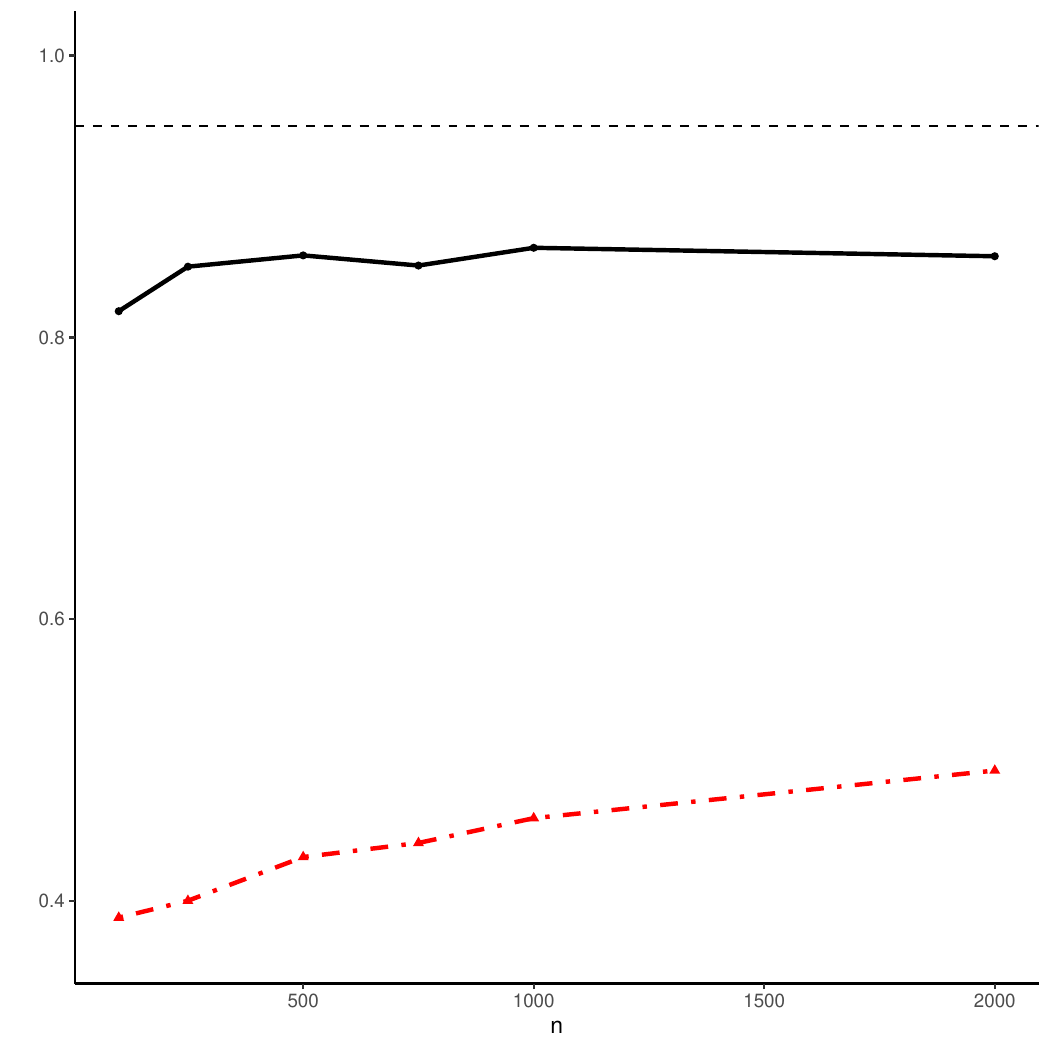}
		\subcaption{$(vi)$ $\x=-0.2$}
	\end{subfigure}
	
	\subcaption{(c) $\hat{h}_{\MSE}$}
	\begin{subfigure}[b]{0.3\textwidth}
		\includegraphics[width=\linewidth]{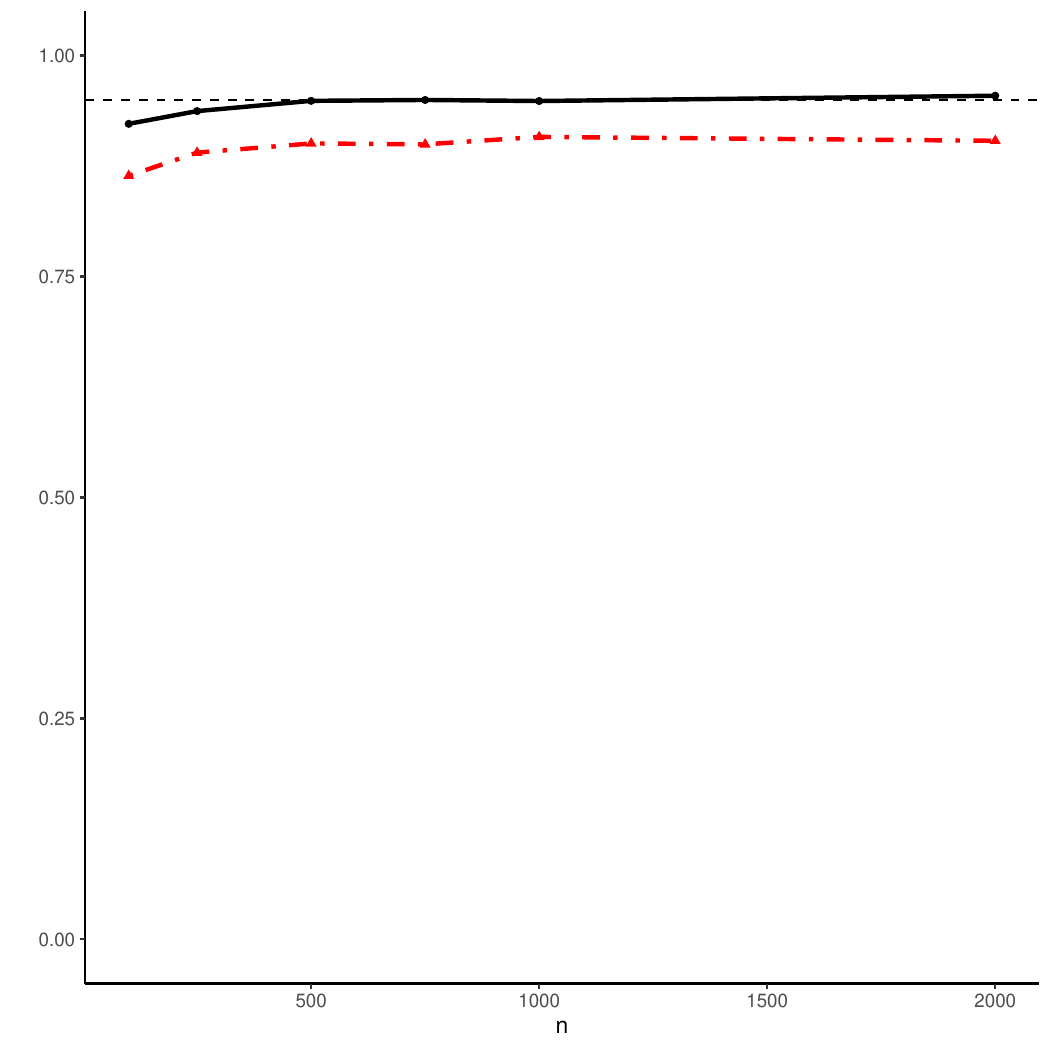}
		\subcaption{$(vii)$ $\x=-1$}
	\end{subfigure}%\hfill
	\begin{subfigure}[b]{0.3\textwidth}
		\includegraphics[width=\linewidth]{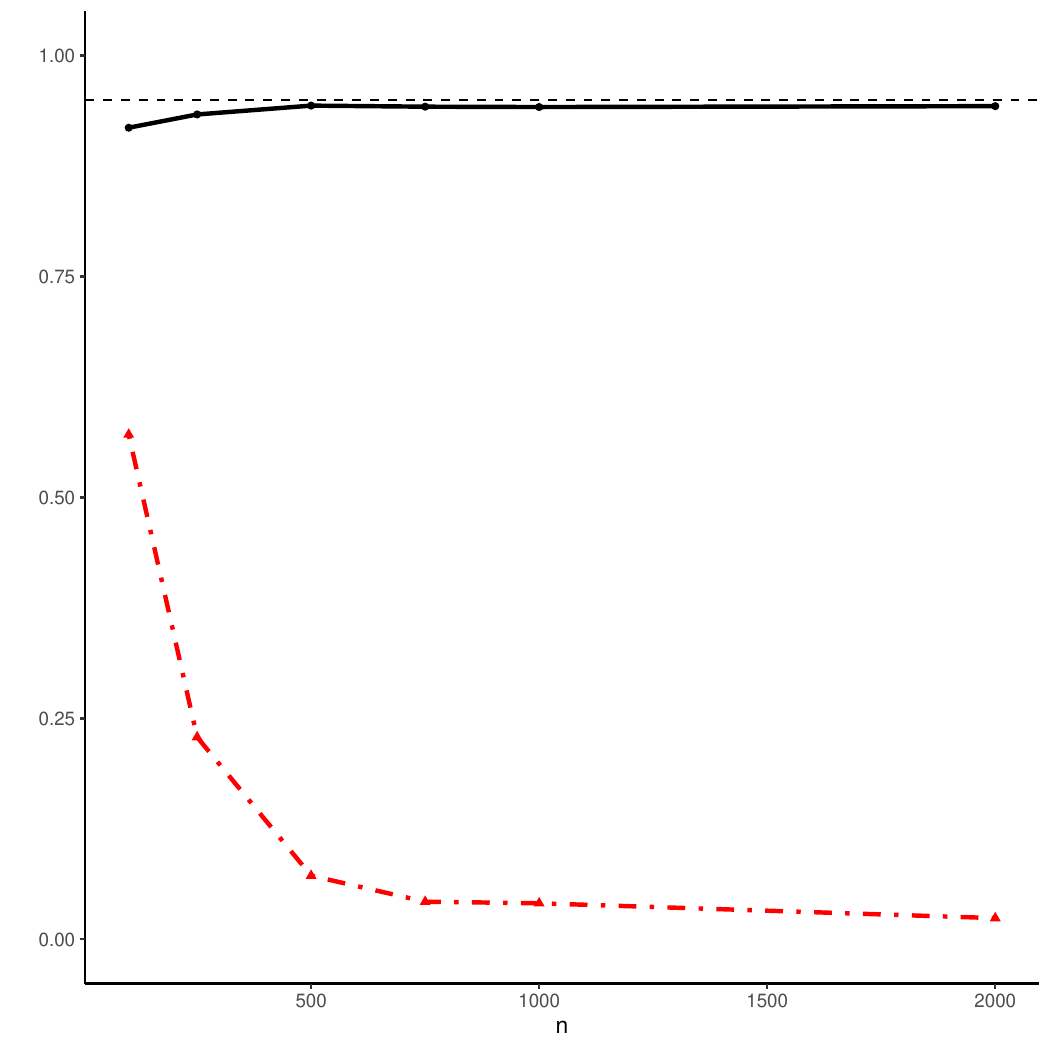}
		\subcaption{$(viii)$ $\x=-0.6$}
	\end{subfigure}%\hfill
	\begin{subfigure}[b]{0.3\textwidth}
		\includegraphics[width=\linewidth]{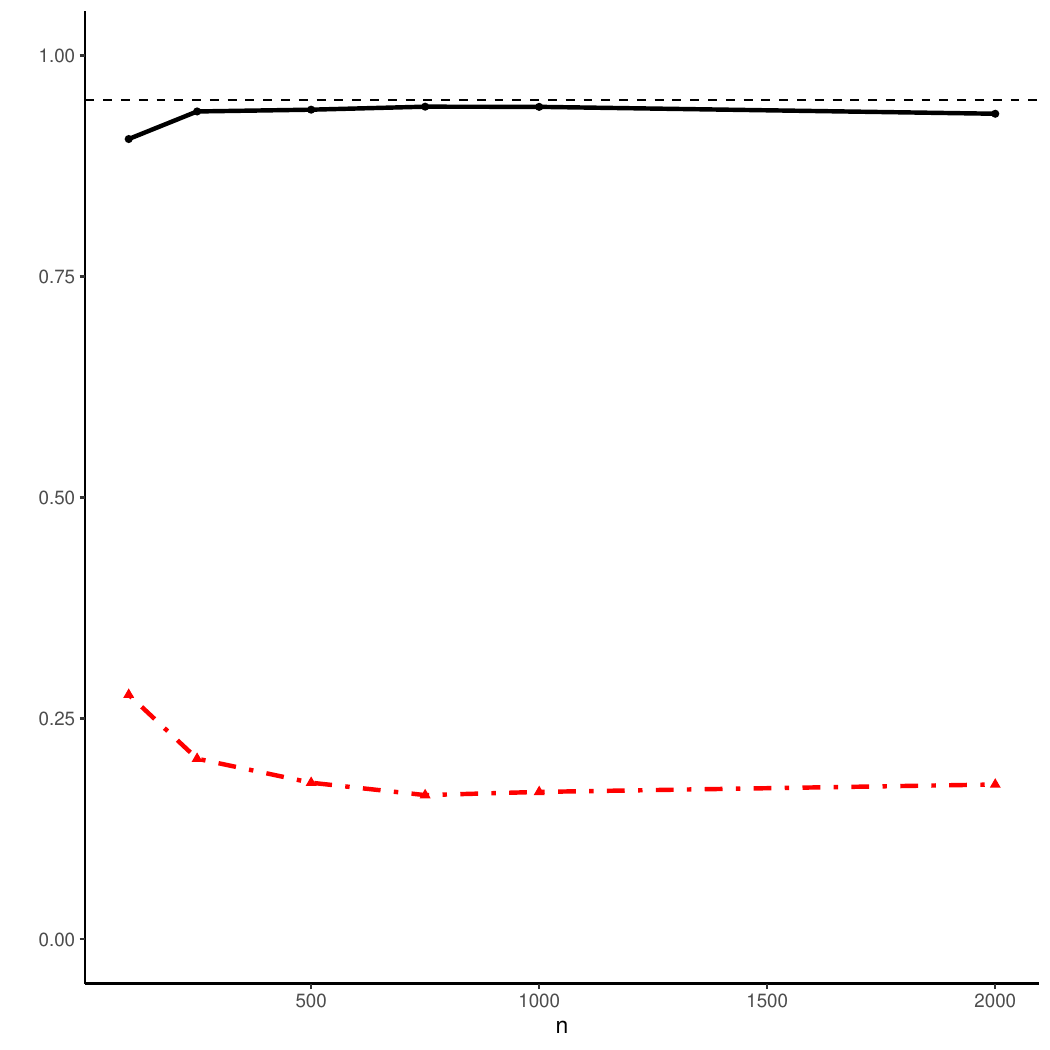}
		\subcaption{$(ix)$ $\x=-0.2$}
	\end{subfigure}%
	\begin{flushleft}\footnotesize Notes: \blackline Robust Bias Correction, \redline Undersmoothing; Epanechnikov Kernel
	\end{flushleft}
\end{figure}

\clearpage

\begin{figure}[!htb]
	\captionsetup[subfigure]{labelformat=empty}
	\centering
	\caption{Average Interval Length for 95\% Confidence Intervals}
	\label{fig:sim_il}
	\subcaption{(a) $\v=0$}
	\begin{subfigure}[b]{0.3\textwidth}
		\includegraphics[width=\linewidth]{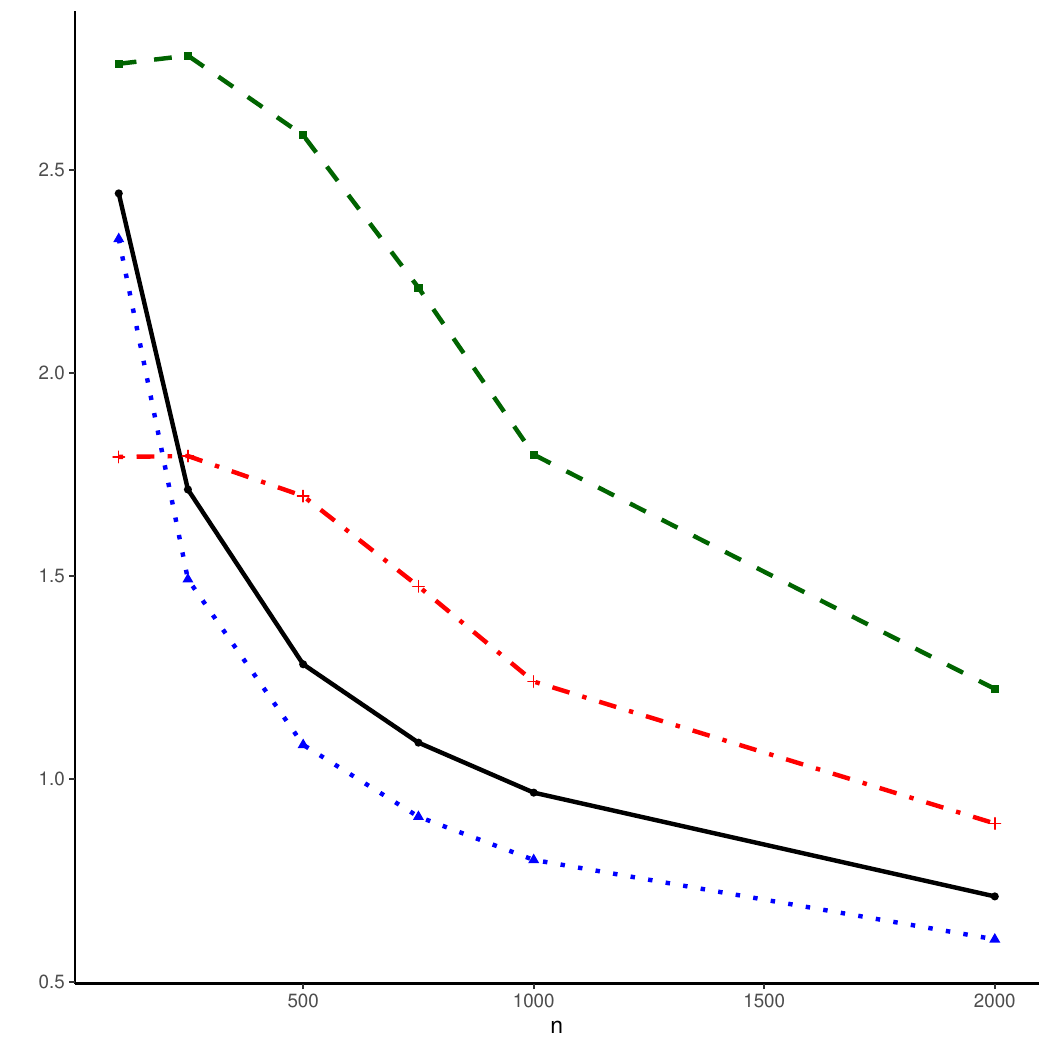}
		\subcaption{$(i)$ $\x=-1$}
	\end{subfigure}%
	\begin{subfigure}[b]{0.3\textwidth}
		\includegraphics[width=\linewidth]{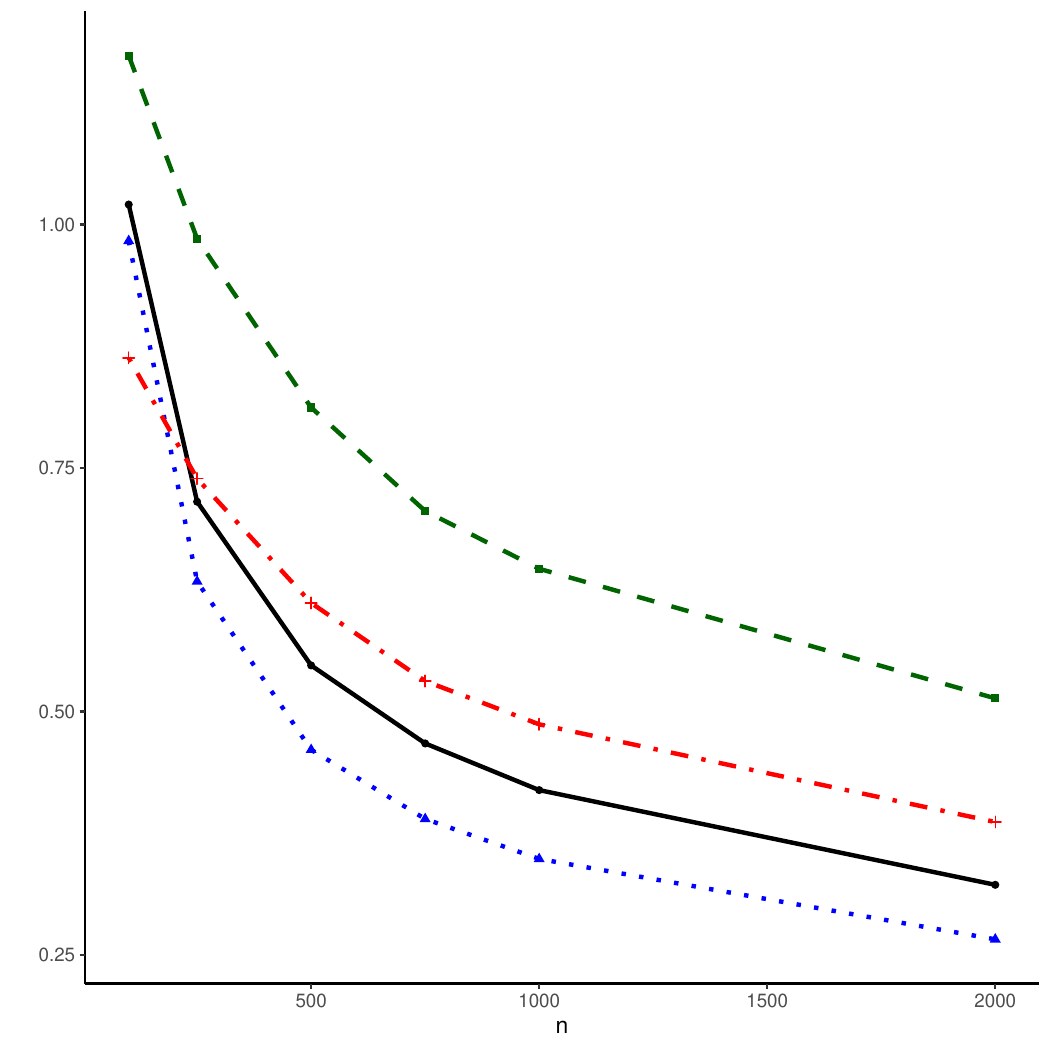}
		\subcaption{$(ii)$ $\x=-0.6$}
	\end{subfigure}%
	\begin{subfigure}[b]{0.3\textwidth}
		\includegraphics[width=\linewidth]{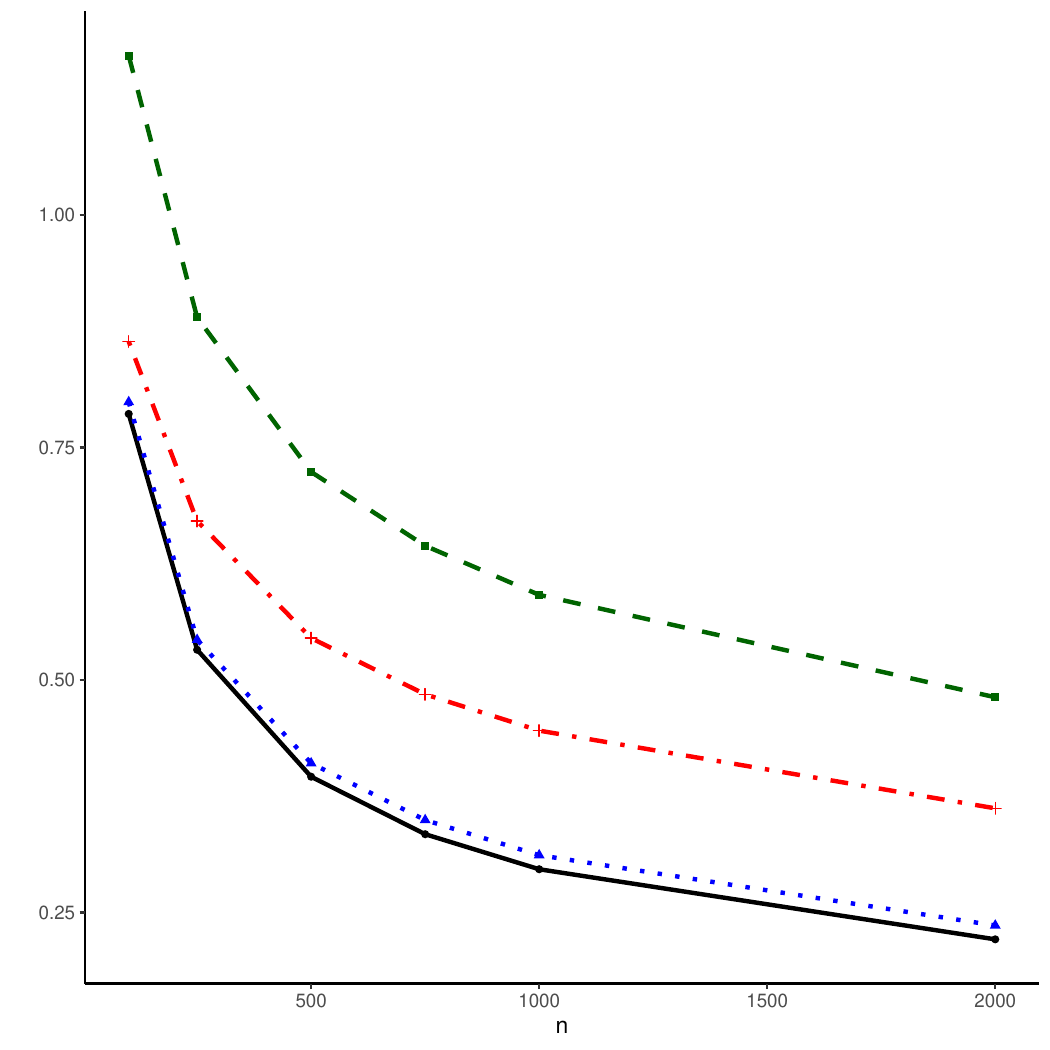}
		\subcaption{$(iii)$ $\x=-0.2$}
	\end{subfigure}
	
	\subcaption{(b) $\v=1$}
	\begin{subfigure}[b]{0.3\textwidth}
		\includegraphics[width=\linewidth]{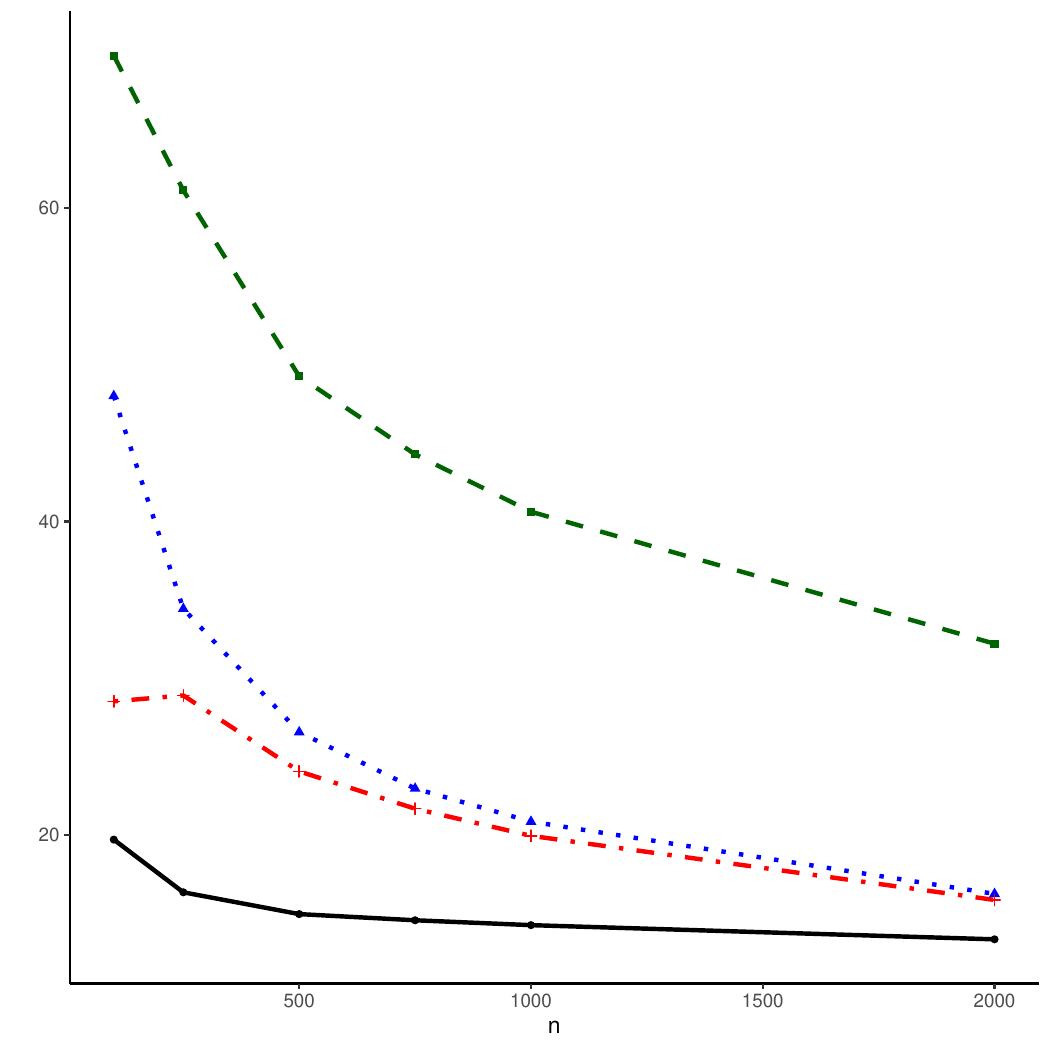}
		\subcaption{$(iv)$ $\x=-1$}
	\end{subfigure}%
	\begin{subfigure}[b]{0.3\textwidth}
		\includegraphics[width=\linewidth]{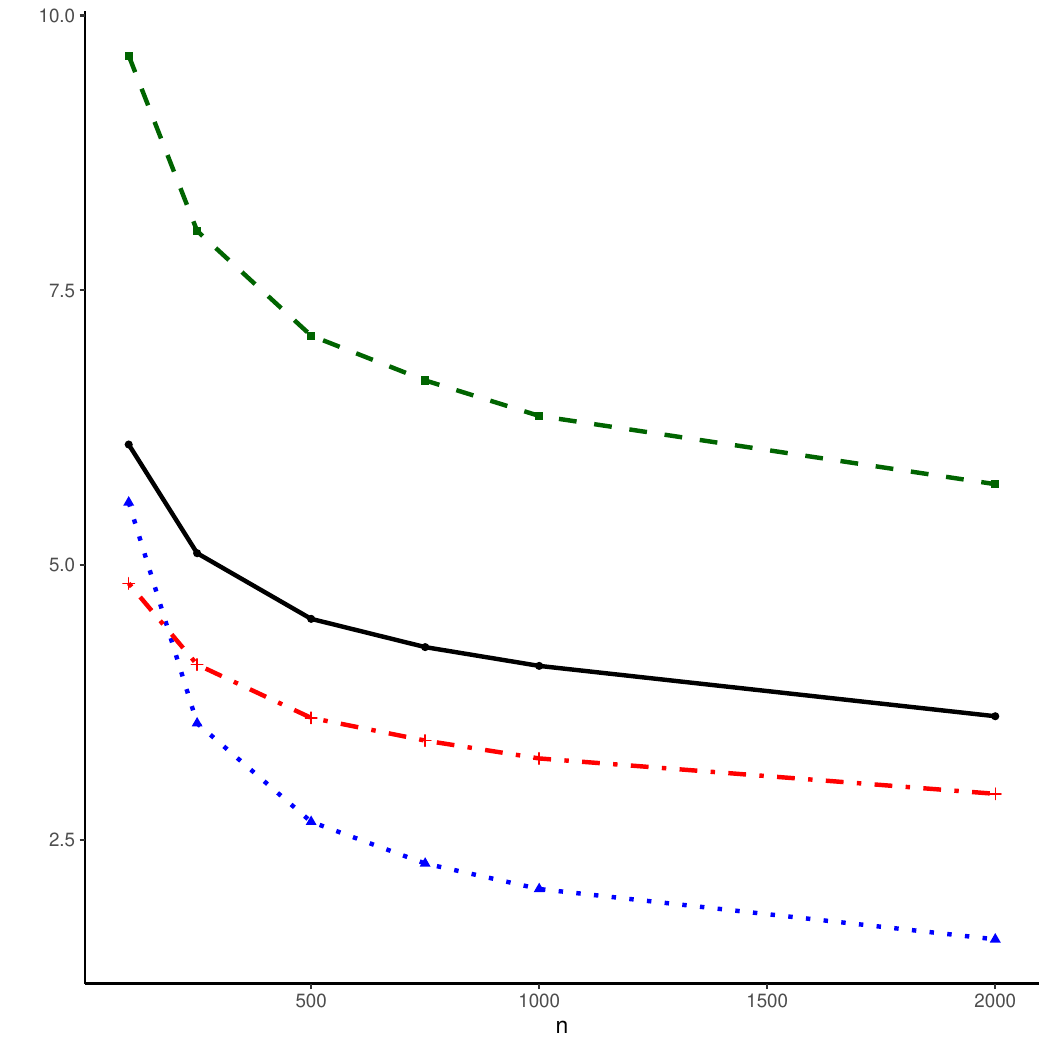}
		\subcaption{$(v)$ $\x=-0.6$}
	\end{subfigure}%
	\begin{subfigure}[b]{0.3\textwidth}
		\includegraphics[width=\linewidth]{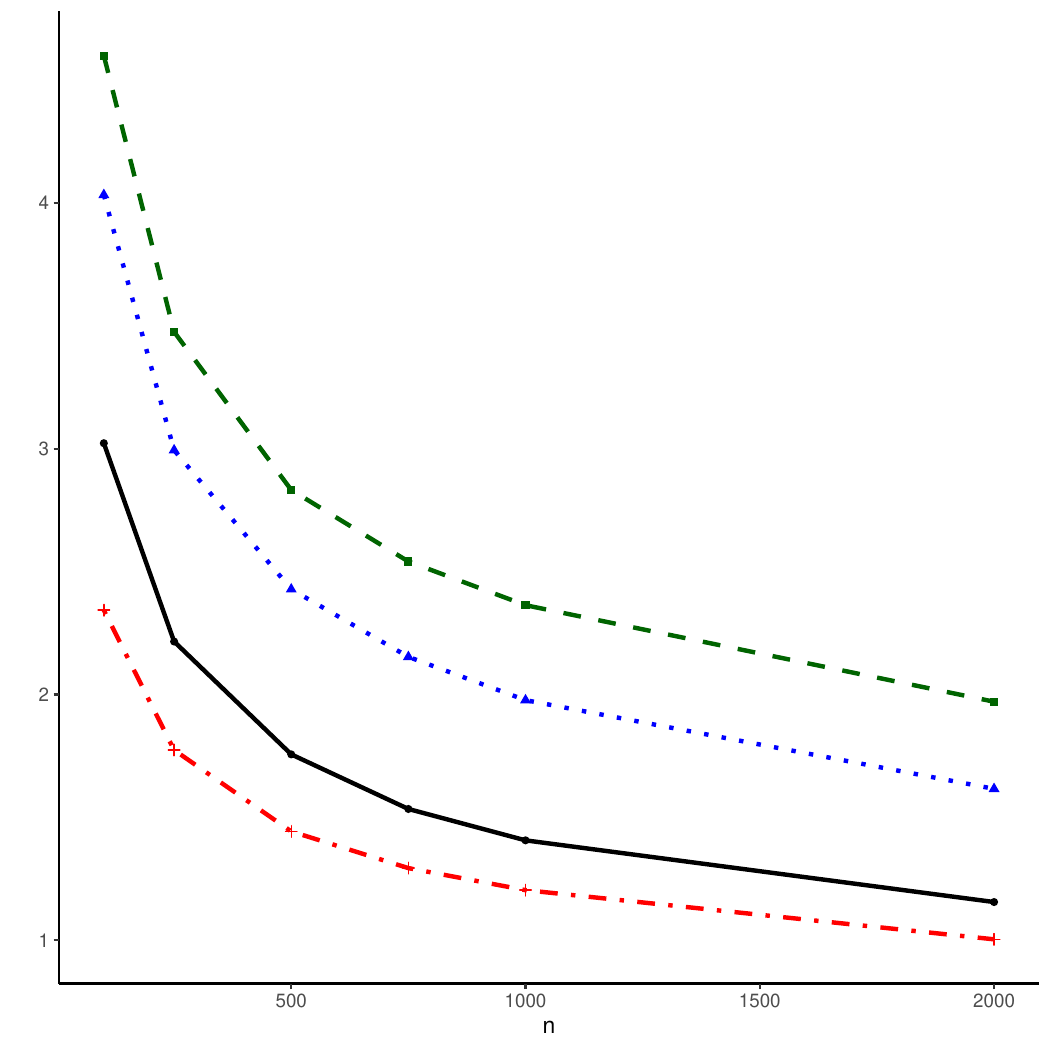}
		\subcaption{$(vi)$ $\x=-0.2$}
	\end{subfigure}
	\begin{flushleft}
		\footnotesize Notes: \blackline $\irbc(\hat{h}_{\RBC})$, \blueline $\irbc(\hat{h}_{\MSE})$, \greenline   $\irbc(\hat{h}_{\US})$,  \redline$\ius(\hat{h}_{\US})$; Epanechnikov Kernel
	\end{flushleft}
\end{figure}

\clearpage

\clearpage
\begin{figure}[!htb]
	\captionsetup[subfigure]{labelformat=empty}
	\centering
	\caption{Average Estimated Bandwidth}
	\label{fig:h_epa}
	
	\subcaption{(a) $\v=0$}
	\begin{subfigure}[b]{0.3\textwidth}
		\includegraphics[width=\linewidth]{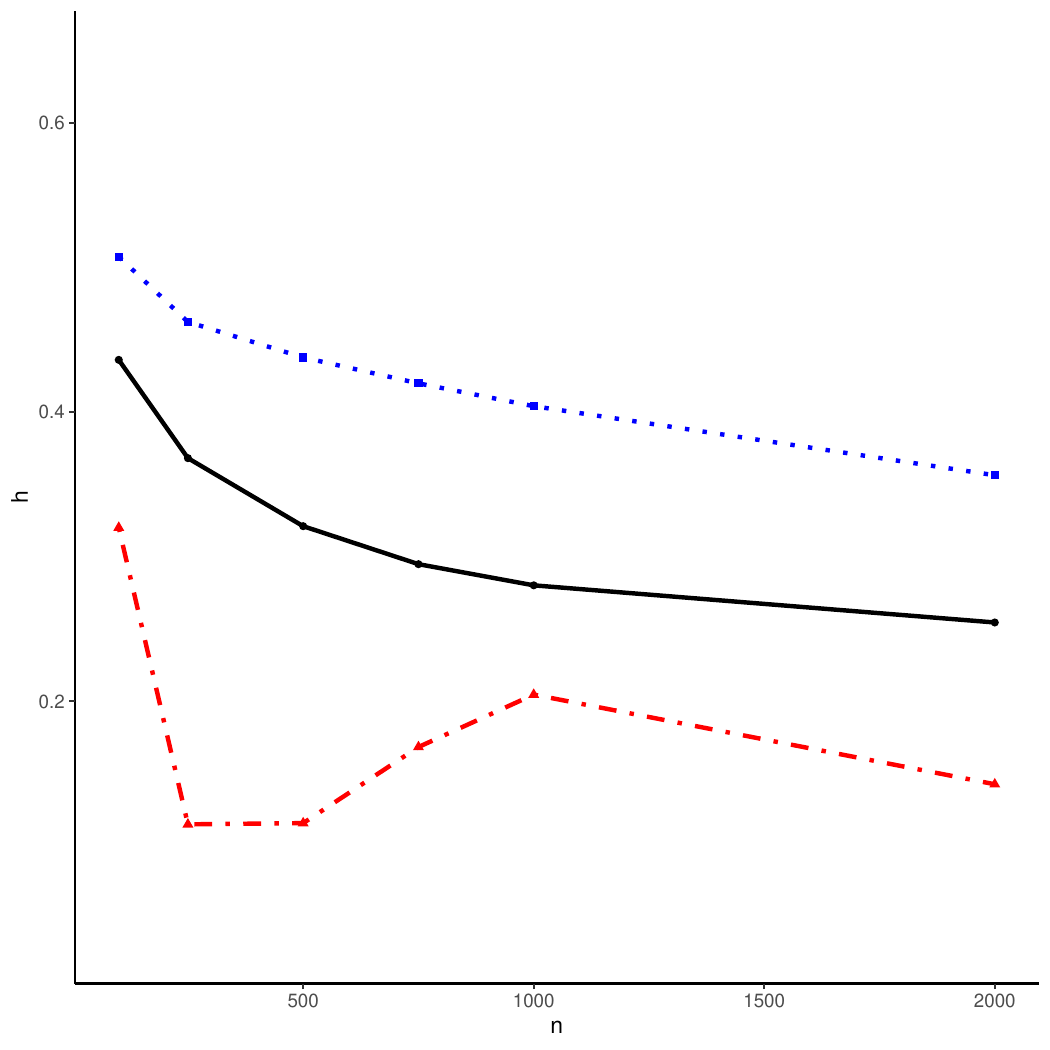}
		\subcaption{$(i)$ $\x=-1$}
	\end{subfigure}%
	\begin{subfigure}[b]{0.3\textwidth}
		\includegraphics[width=\linewidth]{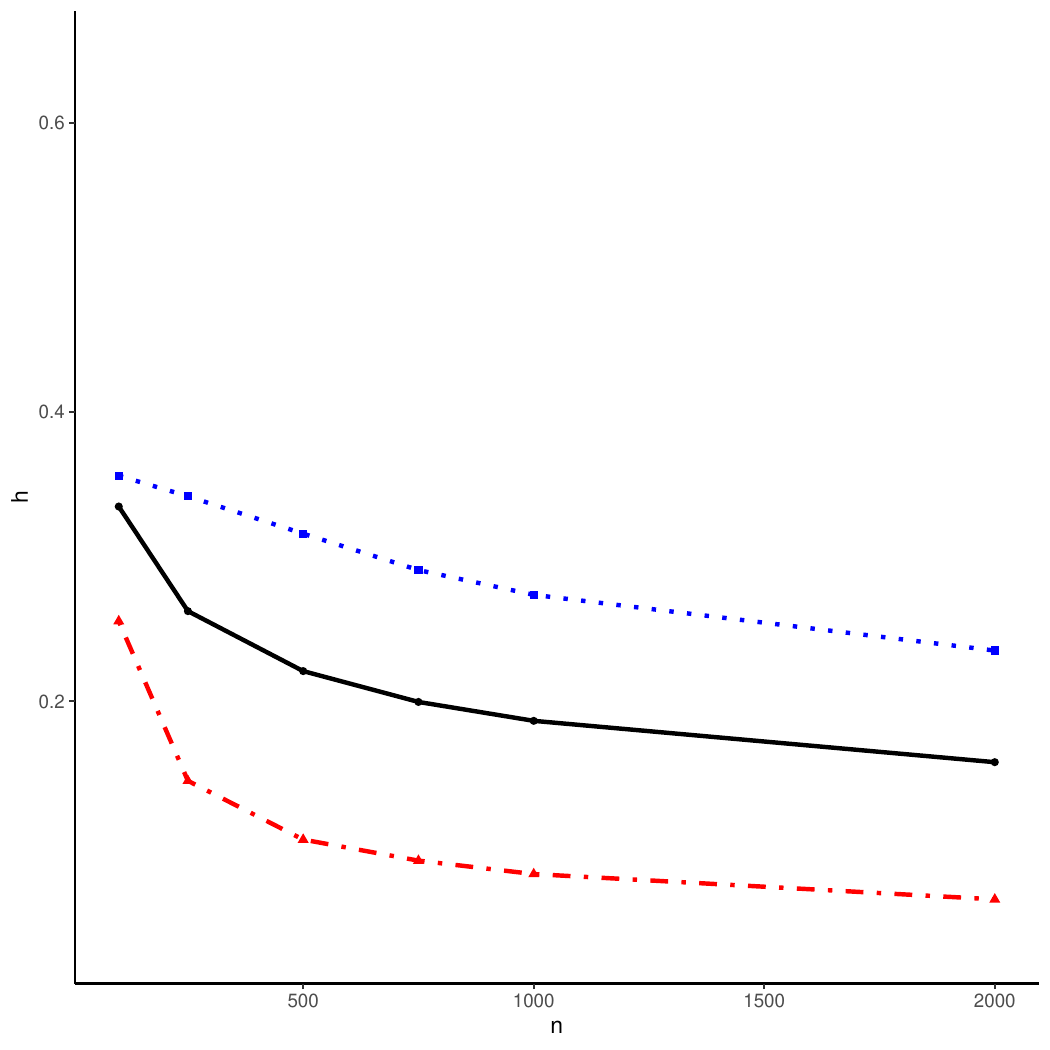}
		\subcaption{$(ii)$ $\x=-0.6$}
	\end{subfigure}%	
	\begin{subfigure}[b]{0.3\textwidth}
		\includegraphics[width=\linewidth]{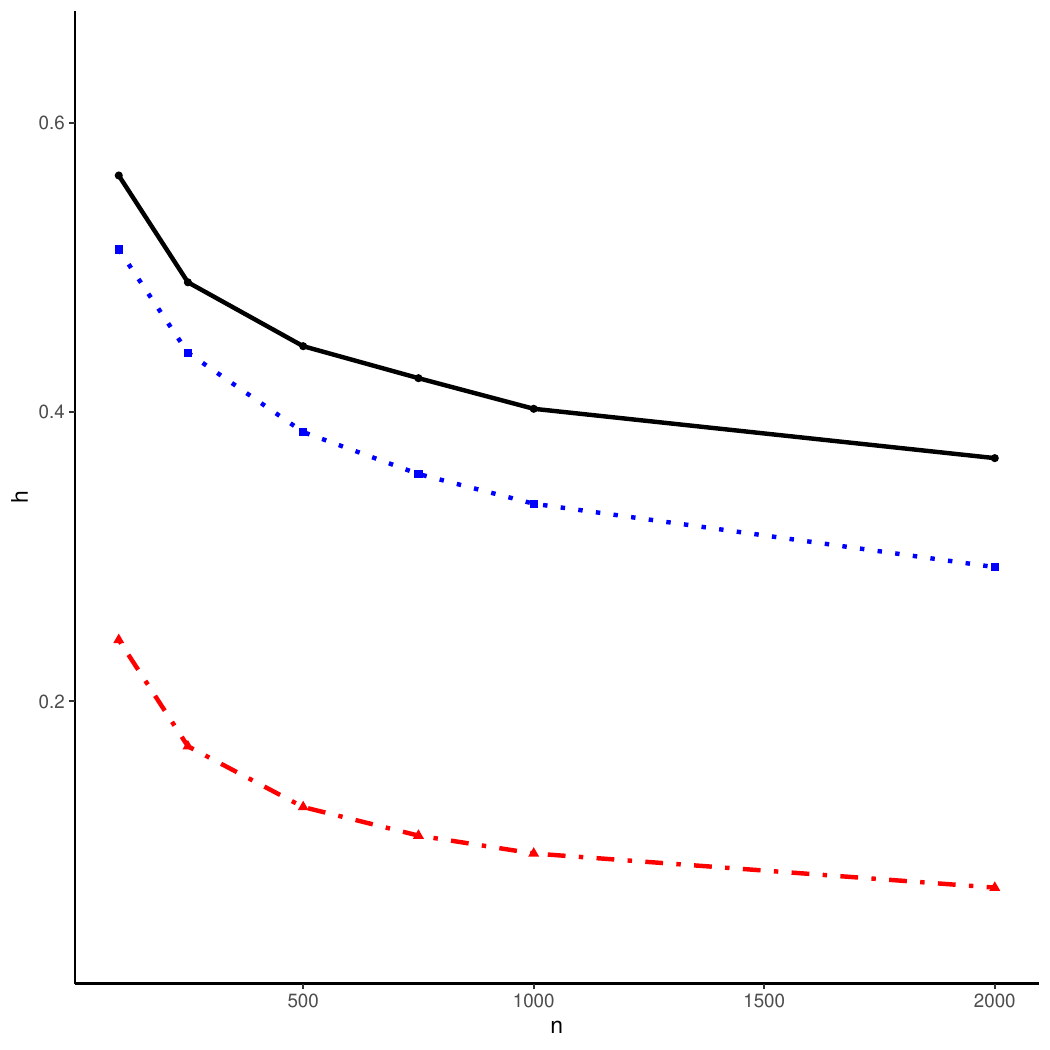}
		\subcaption{$(iii)$ $\x=-0.2$}
	\end{subfigure}
	
	\subcaption{(b) $\v=1$}
	\begin{subfigure}[b]{0.3\textwidth}
		\includegraphics[width=\linewidth]{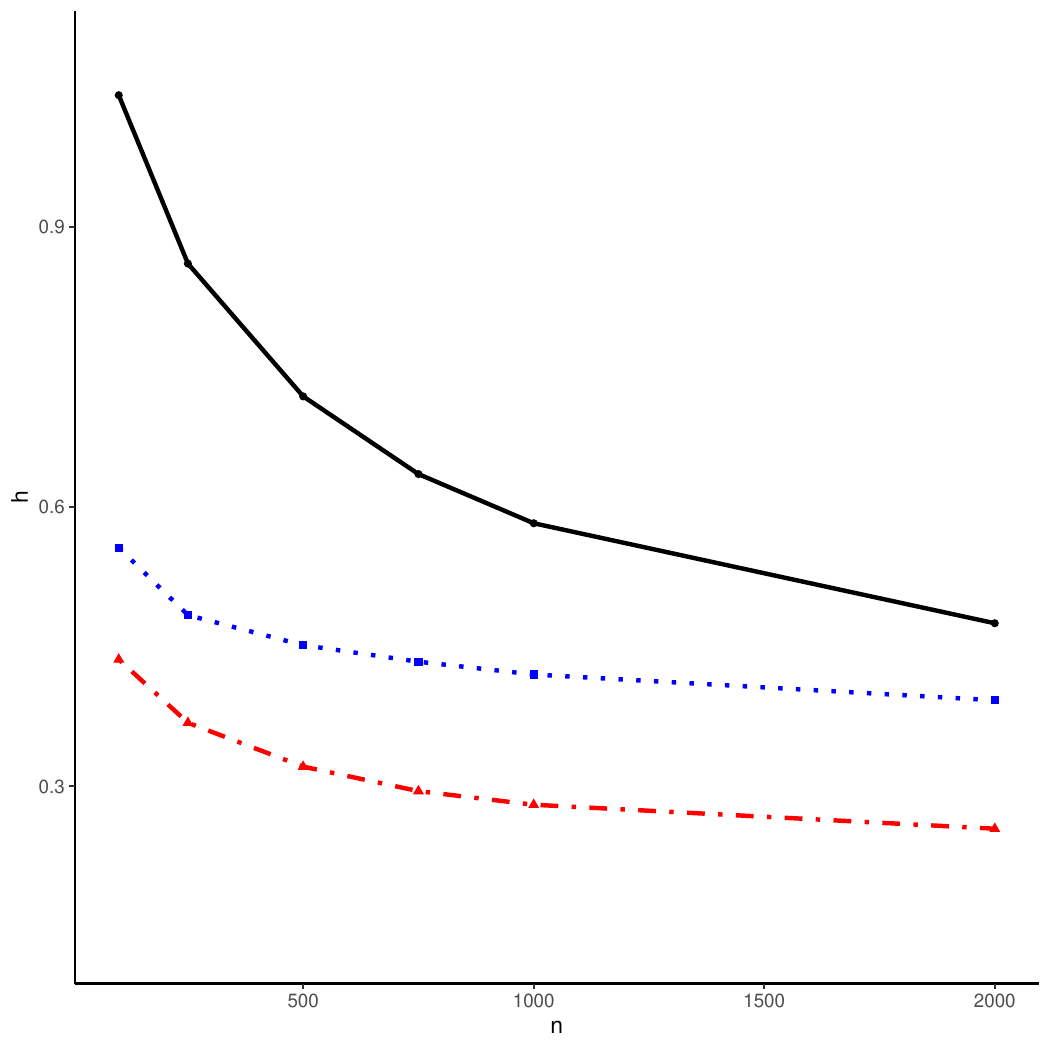}
		\subcaption{$(iv)$ $\x=-1$}
	\end{subfigure}%
	\begin{subfigure}[b]{0.3\textwidth}
		\includegraphics[width=\linewidth]{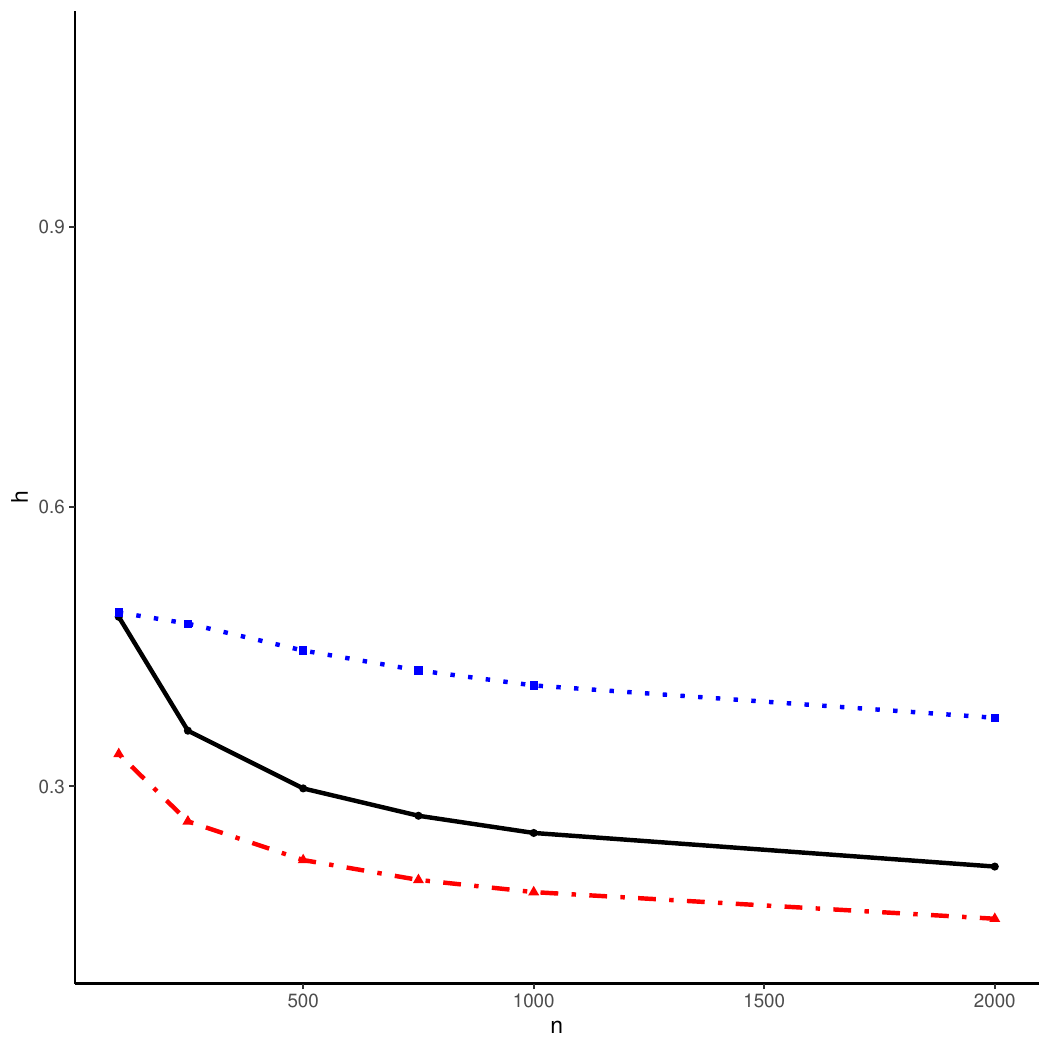}
		\subcaption{$(v)$ $\x=-0.6$}
	\end{subfigure}%	
	\begin{subfigure}[b]{0.3\textwidth}
		\includegraphics[width=\linewidth]{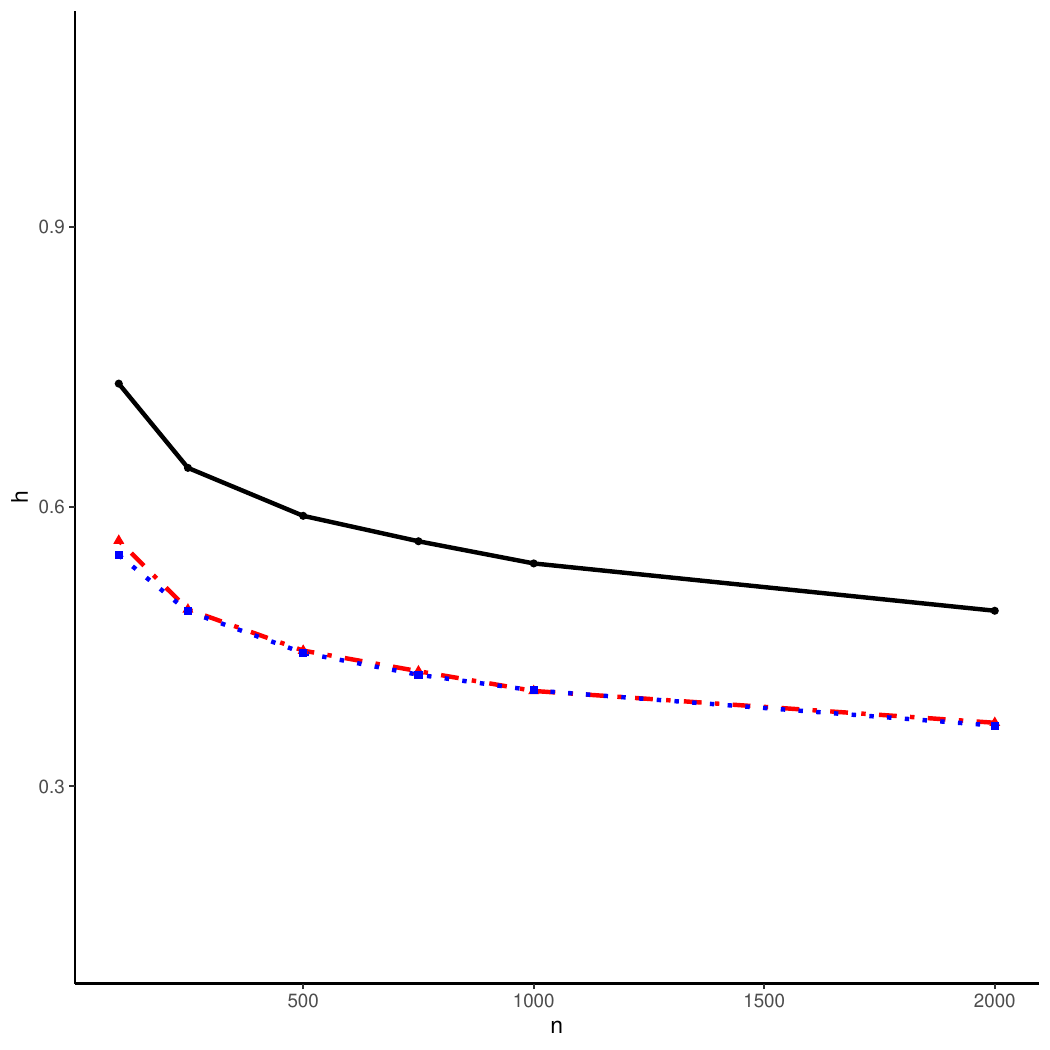}
		\subcaption{$(vi)$ $\x=-0.2$}
	\end{subfigure}%
	\begin{flushleft}\footnotesize Notes: \blackline $\hat{h}_{\RBC}$, \redline $\hat{h}_{\US}$, \blueline $\hat{h}_{\MSE}$; Epanechnikov Kernel
	\end{flushleft}
\end{figure}

\clearpage

%%%%%%%%%%%%%%%%%%%%%%%%%%%%%%%%%%%
%%%%%%%%%%%%%%%%%%%%%%%%%%%%%%%%%%%
%%%%%%%%%%%%%%%%%%%%%%%%%%%%%%%%%%%

%%      Supplement     

%%%%%%%%%%%%%%%%%%%%%%%%%%%%%%%%%%%
%%%%%%%%%%%%%%%%%%%%%%%%%%%%%%%%%%%
%%%%%%%%%%%%%%%%%%%%%%%%%%%%%%%%%%%

	\renewcommand\thepart{S\Roman{part}} 
	\renewcommand\thesection{\thepart.\arabic{section}}	%section numbering has the right format
	\numberwithin{section}{part}	%without this, section numbering won't restart for each part

	\setcounter{equation}{0}
	\renewcommand\theequation{S.\arabic{equation}}

	\setcounter{lemma}{0}
	\renewcommand\thelemma{S.\arabic{lemma}}
	
	\setcounter{remark}{0}
	\renewcommand\theremark{S.\arabic{remark}}	
		
	\setcounter{theorem}{0}
	\renewcommand\thetheorem{S.\arabic{theorem}}

	\setcounter{corollary}{0}
	\renewcommand\thecorollary{S.\arabic{corollary}}
	
	\setcounter{assumption}{0}
	\renewcommand\theassumption{S.\arabic{assumption}}
	
	\setcounter{prop}{0}
	\renewcommand\theprop{S.\arabic{prop}}

	\setcounter{figure}{0}
	\setcounter{table}{0}
	\renewcommand\thetable{S.\arabic{table}}
	\renewcommand\thefigure{S.\arabic{figure}}

%Table of contents
	\setlength{\cftbeforesecskip}{2pt}
	\setlength{\cftbeforepartskip}{2em}
	\setlength{\cftsecnumwidth}{3.5em}
	\setlength{\cftsubsecnumwidth}{3.75em}
	\setlength{\cftsubsubsecnumwidth}{4.75em}
	\setlength{\cftsecindent}{2em}
	\setlength{\cftsubsecindent}{3em}
	\setlength{\cftsubsubsecindent}{4em}

\setenumerate[1]{label=\bf(\alph*)}
\setenumerate[2]{label=\bf(\roman*)}

%%%%%%%%%%%%%%%%%%%%%%%%%%%%%%%%%%%
%
%	Making the TOC correct for the supplement, command 3 of 3
%
% This must come before the supplement starts, to mark the following 
% sections for appearance in the TOC

	\addtocontents{toc}{\protect\setcounter{tocdepth}{2}}
%
%
%%%%%%%%%%%%%%%%%%%%%%%%%%%%%%%%%%%

\clearpage
\numberwithin{section}{part}
\part*{}
\setcounter{section}{0}

\setcounter{page}{0}
\begin{center}
{\LARGE
\vspace{-0.75in} Supplement to ``Coverage Error Optimal Confidence Intervals for Local Polynomial Regression''
}
\end{center}

\bigskip

\small 

This supplement contains proofs of all results, other technical details, and complete simulation results. Notation is kept mostly consistent with the main text, but this document is self-contained as all notation is redefined and all necessary constructions, assumptions, and so forth, are restated. Throughout, clarity is prized over brevity, and repetition is not avoided. The outline is as follows. Section \ref{supp:setup} gives a complete formalization of the set up, inference procedures, and assumptions, exactly as in exactly as given in Section 2 of the main paper. Section \ref{supp:theorems lp} gives the proofs for Theorem 1 and Corollaries 1 and 2 of the main paper. Theorem 1 of the paper is restated identically as Theorem \ref{suppthm:EE lp} here, for referencing. The proof of Theorem \ref{suppthm:EE lp} (Theorem 1 in the paper) is long and occupies several subsections. Section \ref{supp:bias lp} gives all details and derivations relating to bias, including formulas omitted from the main text, for all estimators, points of evaluation, and smoothness cases. Section \ref{supp:other standard errors} discusses standard errors. Section \ref{supp:check function} gives notes on the check function loss for asymmetric measurement of coverage error. Section \ref{supp:numerical} presents complete simulations results and computations. For reference a complete list of notation is given in Section \ref{supp:notation lp}.

\setcounter{tocdepth}{2}
\tableofcontents

\newpage
\normalsize

%%%%%%%%%%%%%%%%%%%%%%%%%%%%%%%%%%%%%%%%%%%%%%%%%%%%%%%%%%%%%%%%%%%%%%
%%%%%%%%%%%%%%%%%%%%%%%%%%%%%%%%%%%%%%%%%%%%%%%%%%%%%%%%%%%%%%%%%%%%%%
\section{Setup}
	\label{supp:setup}

We observe a random sample $\{(Y_1, X_1), \ldots, (Y_n, X_n)\}$ from the pair $(Y, X)$, which are distributed according to $\f$, the data-generating process. $\f$ is assumed to belong to a class $\F_\S$, as defined by Assumption \ref{suppasmpt:dgp lp} below, and in particular the pair $(Y, X)$ obeys the heteroskedastic nonparametric regression model
\begin{equation}
	\label{suppeqn:model}
	Y = \mu_\f(X) + \e,   \qquad   \E[\e \vert X ] = 0,   \qquad   \E[\e^2 \vert X = x] = v(x).
\end{equation}
The parameter of interest is a derivative of the regression function, defined as
\begin{equation}
	\label{suppeqn:theta lp}
	\tf = \mu^{(\v)}_\f(\x) := \left.  \frac{\partial^\v}{\partial x^\v} \E_\f \left[ Y \mid X \!=\! x \right] \right\vert_{x=\x},
\end{equation}
for a point $\x$ in the support of $X$ and an nonnegative integer $\v \leq \S$, the latter defined in Assumption \ref{suppasmpt:dgp lp}, and indexing the class $\F_\S$. As usual, we use the notation $\mu_\f(\x) = \mu^{(0)}_\f(\x) = \E_\f [ Y \mid X \!=\! \x ] $. 

Expectations and probability statements, as well as parameters and functions, are always understood to depend on $\f$, though for simplicity this will often be omitted when doing so causes no confusion. Similarly, unless it is explicitly required, we will omit the point of evaluation $\x$ as an argument. For example,
\[
	\mu^{(\v)}_\f(\x) = \mu^{(\v)}(\x) = \mu^{(\v)}.
\]

Our main technical contributions are novel Edgeworth expansions for local polynomial based Wald-type $t$ statistics of the form
\begin{equation}
	\label{suppeqn:t stat}
    \ti = \frac{\that - \mu^{(\v)}}{\vhat},
\end{equation}
for a centering estimator $\that$ and scale estimator $\vhat$. We establish this expansion uniformly in a class of distributions that generated the data, that is, we characterize the leading terms $E_{\ti,\f}(z)$ and rate $r_{\ti,\f}$, both specific to a $t$ statistic and distribution, and prove that
\begin{equation}
	\label{suppeqn:EE-uniform}
	\lim_{n\to\infty} \; \sup_{\f \in \F_\S} \; r_{\ti,\f}^{-1} \; \sup_{z \in \R} \; \Big| \P_\f[ \ti < z]  - \Phi(z) - E_{\ti,\f}(z) \Big| = 0.
\end{equation}
This Edgeworth expansion is Theorem 1 of the paper and Theorem \ref{suppthm:EE lp} herein. We also study the coverage error of commonly-used Wald-type confidence interval estimators given generically by
\begin{equation}
	\label{suppeqn:ci lp}
	\i  = \left[ \that - z_u \; \vhat \ , \  \that - z_l \; \vhat \right],
\end{equation}
for a pair of quantiles $z_l$ and $z_u$. See Corollary 2 of the main paper. 

Throughout, asymptotic orders and their in-probability versions always hold uniformly in $\F_\S$, as required by our framework: for example, $A_n = o_\P(a_n)$ means $\sup_{\f \in \F_\S} \P_\f [\vert A_n/a_n \vert > \epsilon] = o(1)$ for every $\epsilon > 0$. Limits are taken as $n\to\infty$ unless stated otherwise.

%%%%%%%%%%%%%%%%%%%%%%%%%%%%%%%%%%%%%%%%%%%%%%%%%%%%%%%%%%%%%%%%%%%%%%
\subsection{Centering Estimators}

We now define the centering estimators $\that$. These are based on local polynomial regressions. The standard local polynomial (of degree $p$) point estimator is defined via the local regression
\begin{equation}
	\label{suppeqn:lp}
	\mhat_p^{(\v)} = \v! \be_\v' \bhat_p = \frac{1}{n \h^\v} \v! \be_\v'\Gp^{-1} \Op  \bY, 		 \quad\quad 		\bhat_p = \argmin_{\bbeta \in \mathbb{R}^{p+1}} \sumi ( Y_i - \br_p(X_i - \x)'\bbeta)^2  K \left( \Xhi \right),
\end{equation}
where
\begin{itemize}

	\item $\be_k$ is a conformable zero vector with a one in the $(k+1)$ position, for example $\be_\v$ is the $(p+1)$-vector with a one in the $\v^{\text{th}}$ position and zeros in the rest, 
	
	\item $\h$ is a positive bandwidth sequence that vanishes as $n$ diverges,
	
	\item $p$ is an integer greater at least $\v$, sometimes restricted such that $p - \v$ odd,
	
	\item $\br_p(u) = (1, u, u^2, \ldots, u^p)'$,
	
	\item $\Xhi = (X_i - \x)/\h$, for a bandwidth $\h$ and point of interest $\x$,

	\item to save space, products of functions will often be written together, with only one argument, for example,
\[ (K \br_p \br_p')(\Xhi) \defsym K(\Xhi) r_p(\Xhi) r_p(\Xhi)' = K\left(\frac{X_i - \x}{\h}\right) \br_p \left(\frac{X_i - \x}{\h}\right) \br_p \left(\frac{X_i - \x}{\h}\right)' ,  \]

	\item $\bW = \diag\left(\h^{-1} K(\Xhi): i = 1, \ldots, n\right)$,
	
	\item $\H = \diag\left(1, \h, \h^2, \ldots, \h^p \right)$, where
	
	\item $\diag(a_i:i = 1, \ldots, k)$ denote the $k \times k$ diagonal matrix constructed using the elements $a_1, a_2, \cdots, a_k$,

	\item $\bR = \left[ \br_p( X_1 - \x), \cdots, \br_p( X_n - \x ) \right]'$,

	\item $\bRc = \bR \H^{-1} = \left[ \br_p( X_{\h,1}), \cdots, \br_p( X_{\h,n} ) \right]'$,
	
	\item $\Gp = \frac{1}{n\h} \sumi (K \br_p \br_p')(\Xhi) = (\bRc' \bW \bRc)/n$, 
	
	\item $\Op = \h^{-1}[ (K \br_p)(X_{\h,1}),  (K \br_p)(X_{\h,2}), \ldots,  (K \br_p)(X_{\h,n})] = \bRc' \bW$, and
	
	\item $\bY = (Y_1, \ldots, Y_n)'$.

\end{itemize}
We will also use, for bias correction,
\begin{itemize}

	\item $\bhat_{p+1}$ which is defined exactly as in Equation \eqref{suppeqn:lp} but with $p+1$ in place of $p$ and $\b$ in place of $\h$ in all instances. 

\end{itemize}
For more details on local polynomial methods and related theoretical results, see \cite{Fan-Gijbels1996_book}.

For computing the rate of convergence, and clarifying the appearance of $(n \h^\v)^{-1}$ in Equation \eqref{suppeqn:lp}, it is useful to spell out the form of $\bhat_p$, the solution to the minimization in Equation \eqref{suppeqn:lp}. Standard least squares algebra yields
\begin{align}
	\bhat_p & = \left( \bR' \bW \bR \right)^{-1} \bR' \bW \bY    		\nonumber \\
	& = \left( \left[  \bR\H^{-1}\H \right]' \bW \left[  \bR\H^{-1}\H \right] \right)^{-1} \left[  \bR\H^{-1}\H \right]' \bW \bY     		\nonumber \\
	& = \H^{-1} \left( \bRc' \bW \bRc \right)^{-1} \H^{-1} \H \bRc' \bW \bY    		\nonumber \\
	& = \H^{-1} \left( \bRc' \bW \bRc \right)^{-1}\bRc' \bW \bY,    		\nonumber \\
	& = \H^{-1} \Gp^{-1} \Op \bY / n,   		\label{suppeqn:bhat1}
\end{align}
and therefore, because $\be_\v' \H^{-1}  = \be_\v' \h^{-\v}$,
\begin{equation}
	\label{suppeqn:bhat2}
	\v! \be_\v' \bhat_p = \frac{1}{n \h^\v} \v! \be_\v'\Gp^{-1} \Op  \bY.
\end{equation}
The same applies to $\bhat_{p+1}$ with the necessary changes to the bandwidth and dimensions.

To conduct valid inference on $\tf$ the bias of the nonparametric estimator must be removed. Assuming that the true $\mu^{(\v)}(\cdot)$ is smooth enough at $\x$ (formally, $p + 1 \leq \S$, such as is required for computing the mean square error optimal bandwidth), we find that the (conditional) bias of $\mhat_p^{(\v)}$ is
\begin{equation}
	\label{suppeqn:bias lp}
	\E\left[\mhat_p^{(\v)} \big| X_1, \ldots, X_n \right] - \mu^{(\v)}  =   \h^{p+1 - \v}  \v! \be_\v'\Gp^{-1} \Lp_1 \frac{\mu^{(p+1)}}{(p+1)!}   +   o_\P( \h^{p+1 - \v}),
\end{equation}
where 
\begin{itemize}
	\item $\Lp_k = \Op \left[ X_{\h,1}^{p+k}, \ldots, X_{\h,n}^{p+k} \right]'/n$, where, in particular $\Lp_1$ was denoted $\Lp$ in the main text.
\end{itemize}
Throughout, asymptotic orders and their in-probability versions hold uniformly in $\F_\S$, as required by our framework; e.g., $A_n = o_\P(a_n)$ means $\sup_{\f \in \F_\S} \P_\f [\vert A_n/a_n \vert > \epsilon] = o(1)$ for every $\epsilon > 0$. This expression is valid for $p-\v$ odd or even, though in the latter case the leading term of will be zero due to symmetry for interior points, i.e.\ $\be_\v'\Gp^{-1} \Lp_1 = O(\h)$, and thus the rate will actually be faster \citep[see][]{Fan-Gijbels1996_book}. (Recall that asymptotic orders and their in-probability versions are always required to hold uniformly in $\F_\S$ throughout.) %This cancellation is not needed or used for any of our results below.

Sufficient smoothness for the validity of this calculation need not be available for many of the results herein to apply, and the amount of smoothness assumed to exist is a key factor in determining coverage error rates and optimality. See Section \ref{supp:bias lp} below for details and derivations in all cases, in addition to the discussion in the main paper. For the present, Equation \eqref{suppeqn:bias lp} serves to motivate explicit bias correction by subtracting from $\mhat_p^{(\v)}$ an estimate of the leading bias term. This estimate is formed as 
\[\h^{p+1 - \v}  \v! \be_\v'\Gp^{-1} \Lp_1 \be_{p+1}' \bhat_{p+1},   		  \qquad \text{ with } \qquad  		  \bhat_{p+1}  = \frac{1}{n \b^{p+1}}  \Gq^{-1} \Oq  \bY,\]
where $\bhat_{p+1}$ is exactly as in Equation \eqref{suppeqn:lp}, but with $p+1$ and $\b$ in place of $p$ and $\h$, respectively. \cite{Calonico-Cattaneo-Farrell2018_JASA,Calonico-Cattaneo-Farrell2018_JASAsupp} discuss more general methods of bias correction. It is sometimes convenient to use the form above, but we will also use the more explicit notation for what this approach does: estimating the unknown derivative $\mu^{(p+1)}$ and plugging it in directly
\[\h^{p+1 - \v}  \v! \be_\v'\Gp^{-1} \Lp_1 \frac{\mhat_{p+1}^{(p+1)}}{(p+1)!},   		  \qquad\qquad  		\mhat_{p+1}^{(p+1)} = (p+1)! \be_{p+1}'\bhat_{p+1}   = \frac{1}{n \b^{p+1}} (p+1)! \be_{p+1}'\Gq^{-1} \Oq  \bY,  \]
again matching \eqref{suppeqn:lp}, but with $p+1$ in place of $p$ and $\v$ and $\b$ in place of $\h$. In particular, we have defined the exact analogues for this new local regression:
\begin{itemize}
	
	\item $\Xbi = (X_i - \x)/\b$, for a bandwidth $\b$ and point of interest $\x$, exactly like $\Xhi$ but with $\b$ in place of $\h$,
	
	\item $\Oq = \b^{-1}[ (K \br_{p+1})(X_{\b,1}),  (K \br_{p+1})(X_{\b,2}), \ldots,  (K \br_{p+1})(X_{\b,n})]$, exactly like $\Op$ but with $\b$ in place of $\h$ and $p+1$ in place of $p$, 
	
	\item $\Gq = \frac{1}{n\b} \sumi (K \br_{p+1} \br_{p+1}')(\Xbi)$, exactly like $\Gp$ but with $\b$ in place of $\h$ and $p+1$ in place of $p$, and 

	\item $\Lq_k = \Oq \left[ X_{\b,1}^{p+1+k}, \ldots, X_{\b,n}^{p+1+k} \right]'/n$, exactly like $\Lp_k$  but with $\b$ in place of $\h$ and $p+1$ in place of $p$ (implying $\Oq$ in place of $\Op$).
	
\end{itemize}

We thus consider two types of centering estimators. Conventional nonparametric local polynomial inference sets $\that = \mhat_p^{(\v)}$, which typically requires undersmoothing for valid inference, and robust bias corrected centering, which incorporates the explicit bias correction. In sum, $\that$ of \eqref{suppeqn:ci lp} is one of 
\begin{align}
	\begin{split}
		\label{suppeqn:that lp}
		\mhat_p^{(\v)}  &  = \frac{1}{n \h^\v} \v! \be_\v'\Gp^{-1} \Op  \bY  ;   		\\
		\thatrbc & = \mhat_p^{(\v)}  - \h^{p+1 - \v}  \v! \be_\v'\Gp^{-1} \Lp_1 \frac{\mhat_{p+1}^{(p+1)}}{(p+1)!}  = \frac{1}{n \h^\v} \v! \be_\v'\Gp^{-1} \Orbc  \bY .
	\end{split}
\end{align}
where in the latter form of $\thatrbc$, which is useful for defining the scale estimators below, we define
\begin{itemize}
	\item $ \Orbc = \Op - \rho^{p+1}  \Lp_1 \be_{p+1}' \Gq^{-1} \Oq $ and 
	
	\item $\rho = \h / \b$, the ratio of the two bandwidth sequences.
\end{itemize}
Comparing the two we see that only the matrix $\O$ premultiplying $\bY$ changes.

%%%%%%%%%%%%%%%%%%%%%%%%%%%%%%%%%%%%%%%%%%%%%%%%%%%%%%%%%%%%%%%%%%%%%%
\subsection{Scale Estimators}

The next piece we define are the scaling estimators. As discussed in the paper, it is crucial for coverage error to use fixed-$n$ variance calculations, conditional in this case, to develop the Studentization, and we will focus most of our attention on these. Discussion of other options can be found in Section \ref{supp:other standard errors}, with some mention in Section \ref{supp:theorems lp}. The fixed-$n$ variance of the centering is defined as 
\[
	\vartheta^2 = \V\left[\that \big| X_1, \ldots, X_n \right]  = \frac{1}{n\h^{1 + 2\v}}  \v!^2 \be_\v'\Gp^{-1} (\h \O_\bullet \bS \O_\bullet' /n) \Gp^{-1} \be_\v,\]
where either $\O_\bullet = \Op$ or $\Orbc$ depending on the centering and
\begin{itemize}
	\item $\Sigma = \diag(v(X_i): i = 1,\ldots, n)$, with $v(x) = \V[Y \vert X = x]$. 
\end{itemize}
The rateless portions of the variance is defined by $\sigma^2 \defsym (n \h^{1 + 2\v})  \V\left[\that \big| X_1, \ldots, X_n \right] =  (n \h^{1 + 2\v}) \vartheta^2$, with, in particular
\begin{align}
	\begin{split}
		\label{suppeqn:variance lp}
		\sp^2 & =  \v!^2 \be_\v'\Gp^{-1} (\h \Op \bS \Op' /n) \Gp^{-1} \be_\v, 	\qquad\text{and}  		\\
		\srbc^2 & =  \v!^2 \be_\v'\Gp^{-1} (\h \O_\RBC \bS \O_\RBC' /n) \Gp^{-1} \be_\v ,
	\end{split}
\end{align}

The only unknown piece of these is the conditional variance matrix $\bS$, which we estimate using either
\begin{itemize}
	\item $\Shatp = \diag(\hat{v}(X_i): i = 1,\ldots, n)$, with $\hat{v}(X_i) = ( Y_i - \br_p(X_i - \x)'\bhat_p )^2$ for $\bhat_p$ defined in  Equation \eqref{suppeqn:lp}, or 
	\item $\Shatrbc = \diag(\hat{v}(X_i): i = 1,\ldots, n)$, with $\hat{v}(X_i) = ( Y_i - \br_{p+1}(X_i - \x)'\bhat_{p+1} )^2$ for $\bhat_{p+1}$ defined exactly as in Equation \eqref{suppeqn:lp} but with $p+1$ in place of $p$ and $\b$ in place of $\h$.
\end{itemize}
The estimators $\hat{v}(X_i)$, using either $p$ or $p+1$, are not estimators of the function $v(\cdot)$ of \eqref{suppeqn:model} per se, but rather are a convenient notation for predicted residuals.

The scale estimator $\vhat$ of $\i$ of \eqref{suppeqn:ci lp} is thus one of 
\begin{eqnarray}
	\begin{split}
		\label{suppeqn:se lp}
		\vhat^2 & = \frac{\shatp^2}{n\h^{1+2\v}}, \qquad\qquad && \shatp^2  :=  \v!^2 \be_\v'\Gp^{-1} (\h \Op \Shatp \Op' /n) \Gp^{-1} \be_\v,  \qquad \text{or}   	\\
		\vhat^2 & = \vhatrbc^2 := \frac{\shatrbc^2}{n\h^{1+2\v}}, \qquad\quad  && \shatrbc^2 :=  \v!^2 \be_\v'\Gp^{-1} (\h \Orbc \Shatrbc \Orbc' /n) \Gp^{-1} \be_\v,
	\end{split}
\end{eqnarray}

\begin{remark}
For notational, and more importantly, practical/computational simplicity, the standard errors use the same local polynomial regressions (same kernel, bandwidth, and order) as the point estimates. Changing this results in changes to the constants and potentially (depending on the choices of $\h$, $\b$, and $p$) the rates for the coverage error expansions. Further, the procedure as defined here is simple to implement because the bases $\br_p(X_i - \x)$ and $\br_{p+1}(X_i - \x)$ and vectors $\bhat_p$ and $\bhat_{p+1}$ are already available. Other standard errors are discussed in Section \ref{supp:other standard errors} and, for asymptotic versions, briefly in the main paper.
\end{remark}

%
%
%\begin{remark}[Other Inference Method: Empirical Likelihood]
%	\cite*{Chen-Qin2000_Bmka} propose empirical likelihood based linear polynomial inference that is not of the form \eqref{suppeqn:ci lp}, but still fundamentally uses undersmoothing. Under slightly different assumptions they obtain a coverage expansion that matches ours for $\ip$ with symmetric quantiles (in the main text), and therefore the empirical likelihood procedure suffers the same drawbacks as given here for undersmoothing. In particular, it will not attain the minimax rate that robust bias correction does. Indeed, \cite*{Chen-Qin2000_Bmka} give an optimal undersmoothing bandwidth, with the same rate of decay, $n^{-1/3}$, as those discussed here and as far back as \cite*{Hall1992_AoS_density}, and notes that this yields coverage error of $O(n^{-2/3})$. They do not consider bias correction.
%\end{remark}
%
%

%%%%%%%%%%%%%%%%%%%%%%%%%%%%%%%%%%%%%%%%%%%%%%%%%%%%%%%%%%%%%%%%%%%%%%
%%%%%%%%%%%%%%%%%%%%%%%%%%%%%%%%%%%%%%%%%%%%%%%%%%%%%%%%%%%%%%%%%%%%%%
\subsection{Assumptions}
	\label{supp:assumptions lp}

The two following assumptions are sufficient for our results, both directly copied from the main text. See discuss there. The first defines the class of distributions of the data, denoted $\F_\S$.
\begin{assumption}
	\label{suppasmpt:dgp lp}
	Let $\F_\S$ be the set of distributions $\f$ for the pair $(Y,X)$ which obey model \eqref{suppeqn:model} and for which there exist constants $\S \geq \v$, $\s \in(0,1]$, $0<c<C<\infty$, and a neighborhood of $\x$ on the support of $X$, none of which depend on $\f$, such that for all $x, x'$ in the neighborhood the following hold. 
	\begin{enumerate}
		
		\item The Lebesgue density of $(Y,X)$, $f_{yx}(\cdot)$, the Lebesgue density of $X$, $f(\cdot)$, and $v(x) := \V[Y | X=x]$, are each continuous and lie inside $[c,C]$, and $\E[|Y|^{8+c} \vert X = x] \leq C$.
				
		\item $\mu(\cdot)$ is $\S$-times continuously differentiable and $|\mu^{(\S)}(x) - \mu^{(\S)}(x') |\leq C |x - x'|^{\s}$.
		
	\end{enumerate}	
	Throughout,	$\{(Y_1, X_1), \ldots, (Y_n, X_n)\}$ is a random sample from $(Y,X)$.
\end{assumption}

Second, the class of confidence intervals is governed by the following condition on the kernel function $K(\cdot)$ and polynomial degree $p$. We impose the following throughout.
\begin{assumption}
	\label{suppasmpt:ci lp}
	The kernel $K$ is supported on $[-1,1]$, positive, bounded, and even. Further, $K(u)$ is either constant (the uniform kernel) or $(1, K(u) \br_{3(k+1)}(u))'$ is linearly independent on $[-1,0]$ and $[0,1]$, where $k = p$ if $\ti$ is based on $\mhat_p^{(\v)}$ and $\shatp$, and $k = p+1$ if $\ti$ uses $\thatrbc$ or $\shatrbc$. The order $p$ is at least $\v$.
\end{assumption}

%%%%%%%%%%%%%%%%%%%%%%%%%%%%%%%%%%%%%%%%%%%%%%%%%%%%%%%%%%%%%%%%%%%%%%
%%%%%%%%%%%%%%%%%%%%%%%%%%%%%%%%%%%%%%%%%%%%%%%%%%%%%%%%%%%%%%%%%%%%%%
\section{Main Theoretical Results}
	\label{supp:theorems lp}

%%%%%%%%%%%%%%%%%%%%%%%%%%%%%%%%%%%%%%%%%%%%%%%%%%%%%%%%%%%%%%%%%%%%%%
\subsection{Main Result (Theorem 1 in the Paper)}

We now give the main technical result of the paper, a uniformly (in $\f \in \F_\S$) valid Edgeworth expansion of the distribution function of a generic local polynomial based $t$-statistic, from which coverage error follows for any $\i$. This result is the same as Theorem 1 in the main text. 

The terms of the Edgeworth expansion are defined as
\begin{align}
	\begin{split}
	 	\label{suppeqn:ee terms}
		E_{\ti,\f}(z)  &  = \frac{1}{\sqrt{n \h} } \w_{1,\ti,\f}(z)     +    \bti \w_{2,\ti,\f}(z)     +    \lambda_{\ti,\f} \w_{3,\ti,\f}(z)   		\\
		               &    \quad   +    \frac{1}{n \h}  \w_{4,\ti,\f}(z)       +    \bti^2 \w_{5,\ti,\f}(z)   +    \frac{1}{\sqrt{n \h} } \bti  \w_{6,\ti,\f}(z) ,
	\end{split}
\end{align}
where:
\begin{itemize}
	\item $z$ is the point of evaluation of the distribution, 
	\item $\bti$ denotes the generic non-random (fixed-$n$) bias of the $\sqrt{n\h^{1+2\v}}$-scaled numerator of $\ti$, detailed in all cases in Section \ref{supp:bias lp},
	\item $\lambda_{\ti,\f}$ denotes the mismatch between the variance of the numerator of the $t$-statistic and the population standardization used, discussed in Section \ref{supp:other standard errors}, and
	\item the six terms $\w_{k,\ti,\f}(z)$, $k=1,2, \ldots, 6$, are non-random functions bounded uniformly in $\F_\S$, and bounded away from zero for at least one $\f\in\F_\S$, whose exact forms are computed in Section \ref{supp:terms lp}. 
\end{itemize}

The main result is now the following, which is identical to Theorem 1 in the main paper. Let $\Phi(z)$ is the standard Normal distribution function. (Recall that asymptotic orders and their in-probability versions are always required to hold uniformly in $\F_\S$ throughout.)
\begin{theorem}
	\label{suppthm:EE lp}
	Let Assumptions \ref{suppasmpt:dgp lp} and \ref{suppasmpt:ci lp} hold, and assume that
	\begin{equation*}
		\label{suppeqn:bandwidth requirements}
		\log(n\h)^{2+\gamma} / n \h =o(1),     \quad    \bti \log(n\h)^{1+\gamma} =o(1),     \quad \lambda_{\ti,\f} = o(1), \quad  \rho = O(1),
	\end{equation*}
	for any $\gamma$ bounded away from zero uniformly in $\F_\S$. Then,
	\[
		\lim_{n\to\infty} \; \sup_{\f \in \F_\S} \; r_{\ti,\f}^{-1} \; \sup_{z \in \R} \; \Big| \P_\f[ \ti < z]  - \Phi(z) - E_{\ti,\f}(z) \Big| = 0
	\]
	holds with $E_{\ti,\f}(z)$ of \eqref{suppeqn:ee terms} and $r_{\ti,\f} = \max\{(n\h)^{-1}, \bti^2,  (n\h)^{-1/2}\bti, \lambda_{\ti,\f} \}$.
\end{theorem}

%%%%%%%%%%%%%%%%%%%%%%%%%%%%%%%%%%%%%%%%%%%%%%%%%%%%%%%%%%%%%%%%%%%%%%
\subsection{Proofs for Corollaries 1 and 2 in the Main Paper}

\begin{proof}[Proof of Corollary 1 in the Main Paper]
Define $C_{\i,\f}(z_l,z_u) = E_{\ti,\f}(z_u) - E_{\ti,\f}(z_l)$ and let $r_\i$ be such that $C_{\i,\f}(z_l,z_u) = O(r_\i)$. One can always take $r_{\i} = \sup_{\f \in \F_\S} r_{\ti,\f}$ for $r_{\ti,\f}$ given in Theorem \ref{suppthm:EE lp}. Then, for any $\i$ dual to $\ti$,
\begin{align*}
	r_\i^{-1} & \; \sup_{\f \in \F_\S} \; \Big|  \P_\f \big[ \mu^{(\v)}(\x) \in \i \big] - (1-\alpha) - C_{\i,\f}(z_l,z_u) \Big|  		\\
	& = r_\i^{-1} \; \sup_{\f \in \F_\S} \; \Big|  \P_\f[ \ti < z_u] - \P_\f[ \ti < z_l] - (1-\alpha) - C_{\i,\f}(z_l,z_u) \Big|  		\\
	& \leq r_\i^{-1} \; \sup_{\f \in \F_\S} \; \Big|  \Phi(z_u) + E_{\ti,\f}(z_l) - \Phi(z_l) - E_{\ti,\f}(z_l) - (1-\alpha) - C_{\i,\f}(z_l,z_u) \Big|  		\\
	& \quad + r_\i^{-1} \; \sup_{\f \in \F_\S} \; \Big| \P_\f[ \ti < z]  - \Phi(z) - E_{\ti,\f}(z) \Big| + r_\i^{-1} \; \sup_{\f \in \F_\S} \; \Big| \P_\f[ \ti < z]  - \Phi(z) - E_{\ti,\f}(z) \Big|.
\end{align*}
The first line is zero by definition. Taking the limit as $n \to \infty$ of the second and applying Theorem \ref{suppthm:EE lp} yields the result.
\end{proof}

\begin{proof}[Proof of Corollary 2 in the Main Paper]
Recall that $C_{\i,\f}(z_l,z_u) = E_{\ti,\f}(z_u) - E_{\ti,\f}(z_l)$. The functions $\w_1$ and $\w_2$ are even functions of $z$ while the remainder are odd. Therefore, the coverage error of $i \in \I_p$ with $z_l = -z_u$ and $\lambda_{\ti,\f} \equiv 0$ vanishes faster than those without these properties. Identifying the minimum possible worst-case coverage error requires minimizing the $w_4$, $w_5$, and $w_6$ terms of Equation \eqref{suppeqn:ee terms}. For a fixed bandwidth sequence $\h$, this amounts to comparing the rate at which the bias $\bti = o(1)$. In every smoothness case, this rate can be found in Section \ref{supp:bias lp}: specifically, Tables \ref{supptable:rbc bias list} and \ref{supptable:rbc bias list} show the fastest attainable rate in every case. The result follows by plugging the case-specific rate into
\[
	\frac{1}{n \h} C_1         +    \bti^2 C_2    +    \frac{1}{\sqrt{n \h} } \bti  C_3,
\]
and minimizing with respect to $\h$. The constants $C_1$, $C_2$, and $C_3$, collecting the other portions of the terms, are immaterial, as this calculation only requires rates.
\end{proof}

	%% NOTES for interior point: 
	%%	p = S
	%%	US has already used all smoothness, RBC just adds noise
	%%	p = (S-1)
	%%	The US bias is order (p+1)=S, there's no "next" bias term. RBC attacks the only characterizable bias term and knocks it down.
	%%	This is the case where RBC eats up all the remaining smoothness.
	%%	p=(S-2)
	%%	This is where interior matters. RBC knocks down the p+1 term. 
	%%	The bias of the bias estimator would be (q+1) and the non-targeted bias would be (p+1)
	%%	but these are both zero by symmetry, and so are h^{p+3} = {q+2} = S+1 > S+s. 
	%%	p\leq S-3
	%%	Nonbinding smoothness. Matches the loc poly treatment in JASA. 
	%%	This is the case where we can characterize everything.

	%% NOTES for boundary point: 
	%%	p = S
	%%	US has already used all smoothness, RBC just adds noise
	%%	p = (S-1)
	%%	The US bias is order (p+1)=S, there's no "next" bias term. RBC attacks the only characterizable bias term and knocks it down.
	%%	This is the case where RBC eats up all the remaining smoothness.
	%%	p\leq S-2
	%%	Nonbinding smoothness.
	%%	This is the case where we can characterize everything.

%%%%%%%%%%%%%%%%%%%%%%%%%%%%%%%%%%%%%%%%%%%%%%%%%%%%%%%%%%%%%%%%%%%%%%
%%%%%%%%%%%%%%%%%%%%%%%%%%%%%%%%%%%%%%%%%%%%%%%%%%%%%%%%%%%%%%%%%%%%%%
\subsection{Proof of Theorem \ref{suppthm:EE lp} (Theorem 1 in the paper) without Bias Correction}
	\label{supp:proof lp us}

The goal of this section is to prove that the Edgeworth expansion of Theorem \ref{suppthm:EE lp} is valid for $\tp = T(\mhat_{p+1}^{(\v)},\shatp^2/(n\h^{1+2\v}))$. The proof for $\trbc$ is essentially the same from a conceptual and technical point of view, just with more notation and a repetition of the same steps, and so only a sketch is provided. See Section \ref{supp:proof lp rbc}. We also restrict to the fixed-$n$, HC0 standard errors of \eqref{suppeqn:se lp}, which, in particular, render $\lambda_{\ti,\f} \equiv 0$. Other possibilities are discussed in Section \ref{supp:other standard errors}. The terms of the expansion are computed, in a formal manner, in Section \ref{supp:terms lp}.

For notational ease, we sometimes drop subscripts, along with the point of evaluation and/or dependence on $\f$. Also define
\begin{itemize}
	\item $\tO = \sqrt{n\h}$
\end{itemize}
Recall that asymptotic orders and their in-probability versions are always required to hold uniformly in $\F_\S$ throughout.

The proof consists of three main steps, which are tackled in the subsections below.
\begin{description}
	\itemsep=1em

	\item [Step {\bf (I)}] -- Section \ref{supp:proof lp us step 1}
	
		Show that
		\begin{equation}
			\label{suppeqn:step 1}
			\P_\f \left[ \tp < z \right] = \P_\f \left[ \breve{T} < z \right]  + o\left( (n\h)^{-1} + (n\h)^{-1/2} \btp  +  \btp^2 \right),
		\end{equation}
		for a smooth function $\breve{T} := \breve{T}(\tO^{-1} \sumi \bZ_i)$, where $\bZ_i$ a random vector consisting of functions of $(Y_i, X_i, \e_i)$ that, among other requirements, obeys Cram\'er's condition under our assumptions.
		
	\item [Step {\bf (II)}] -- Section \ref{supp:proof lp us step 2}
	
		Prove that $\sumi \V[\bZ_i]^{-1/2}(\bZ_i - \E[\bZ_i])/\sqrt{n}$ obeys an Edgeworth expansion.
		
	\item [Step {\bf (III)}] -- Section \ref{supp:proof lp us step 3}
	
		Prove that the expansion for $\tp$ holds and that it holds uniformly over $\f \in \F_\S$. 

\end{description}
Numerous intermediate results relied upon in the proof are collected as lemmas that are stated and proved in Section \ref{supp:lemmas lp}.

Unless it is important to emphasize the dependence on $\f$, this will be suppressed to save notation; for example $\P = \P_\f$. Throughout proofs $C$ shall be a generic conformable constant that may take different values in different places. If more than one constant is needed, $C_1$, $C_2$, \ldots, will be used. Also define 
\begin{itemize}

	\item $r_{\tp,\f} = \max\{\tO^{-2}, \btp^2, \tO^{-1} \btp\}$, i.e.\ the slowest vanishing of the rates, and

	\item $r_n$ as a generic sequence that obeys $r_n = o(r_{\tp,\f})$.
\end{itemize}
We will frequently use the elementary probability bounds that for random $A$ and $B$ and positive fixed scalars $a$ and $b$, $\P[|A + B| > a] \leq \P[|A| > a/2] + \P[| B| > a/2]$ and $\P[|AB| > a] \leq \P[|A| > b] + \P[| B| > a/b]$, also relying on the elementary bound $|AB| \leq |A||B|$ for conformable vectors or matrixes $A$ and $B$.

%%%%%%%%%%%%%%%%%%%%%%%%%%%%%%%%%%%%%%%%%%%%%%%%%%%%%%%%%%%%%%%%%%%%%%
\subsubsection{Step (I)}
	\label{supp:proof lp us step 1}

We now prove Equation \eqref{suppeqn:step 1} holds for suitable choices of $\breve{T}$ and $\bZ_i$. Notice that the ``numerator'' portion, $\Gp^{-1} \Op \left( \bY  - \bR \bbeta_p \right)/n$ is already a smooth function of well-behaved random variables, and will thus be incorporated into $\breve{T}$. Our difficulty lies with the Studentization, and in particular, the estimated residuals. We will start by expanding $\shatp^2$ (see Equation \eqref{suppeqn:shat terms}). Substituting this expansion into $\tp$, we will identify the leading terms, collected as appropriate into $\breve{T}$ (Equation \eqref{suppeqn:breve T}) and $\bZ_i$ (Equation \eqref{suppeqn:bZ}), and the remainder terms, collected in $U_n := \tp - \breve{T}$ (Equation \eqref{suppeqn:U}). {\bf Step (I)} is complete upon showing that $U_n$ can be ignored in the expansion; this occupies the latter half of the present subsection.

To begin, recall that $\shatp^2 = \v!^2 \be_\v'\Gp^{-1}  (\h \Op \Shatp \Op' /n) \Gp^{-1} \be_\v$. The matrix $\Gp^{-1}$, present in the numerator as well, enters smoothly and is itself smooth in elements of $\tO^{-1} \sumi \bZ_i$. Thus our focus is on the center matrix, $(\h \Op \Shatp \Op' /n)$, which contains the estimated residuals. Using $\bRc\H = \bR$ (and for each observation, $\br_p(X_i - \x) \H^{-1} = \br_p(\Xhi)$) and $\Gp  = \Op \bRc/n$ we have
\[ \br_p(X_i - \x)'\bhat_p = \br_p(X_i - \x)' \H^{-1} \Gp^{-1} \Op \bY/n =  \br_p(\Xhi)' \Gp^{-1} \Op \bY/n\]
and
\[ \br_p(X_i - \x)'\bbeta_p = \br_p(X_i - \x)' \H^{-1} \Gp^{-1} (\Op \bRc/n) \H \bbeta_p =  \br_p(\Xhi)' \Gp^{-1} \Op \bR \bbeta_p/n .\]
We use these forms to expand as follows:
\begin{align*}
	\frac{\h}{n} \Op \Shatp \Op' & = \frac{1}{n \h} \sumi (K^2 \br_p \br_p')(\Xhi) \hat{v}(X_i) 		\\
	& = \frac{1}{n \h} \sumi (K^2 \br_p \br_p')(\Xhi) \left(  Y_i - \br_p(X_i - \x)'\bhat_p \right)^2   		 \\
	& = \frac{1}{n \h} \sumi (K^2 \br_p \br_p')(\Xhi)  \left( \e_i  + \left[\mu(X_i) -  \br_p(X_i - \x)'\bbeta_p\right]   +  \br_p(X_i - \x)'\left[ \bbeta_p - \bhat_p\right] \right)^2 	 \\
	& = \frac{1}{n \h} \sumi (K^2 \br_p \br_p')(\Xhi)  \left( \e_i  + \left[\mu(X_i) -  \br_p(X_i - \x)'\bbeta_p\right]   -  \br_p(\Xhi)' \Gp^{-1} \Op \left[\bY - \bR \bbeta_p\right] /n \right)^2. 	
\end{align*}
The expansion of $\shatp^2$ is then
\begin{equation}
	\label{suppeqn:shat terms}
	\shatp^2 = \v!^2 \be_\v'\Gp^{-1}  \Big( \bV_1  +  2 \bV_4  - 2 \bV_2  + \bV_3  - 2 \bV_5  + \bV_6 \Big) \Gp^{-1} \be_\v
\end{equation}
where
\begin{align*}
	\bV_1 & = \frac{1}{n \h} \sumi (K^2 \br_p \br_p')(\Xhi) \e_i^2,   		\\  %A_{1,1} in JASA, A^2 on whiteboard
	\bV_2 & = \frac{1}{n \h} \sumi (K^2 \br_p \br_p' \br_p')(\Xhi) \e_i \Gp^{-1} \Op \left[\bY - \bR \bbeta_p\right] /n,  		\\   %A_{1,3} in JASA, 2AC on whiteboard
	\bV_3 & = \frac{1}{n \h} \sumi (K^2 \br_p \br_p')(\Xhi)  \left[\mu(X_i) -  \br_p(X_i - \x)'\bbeta_p\right]^2   ,			\\    %A_{1,6} in JASA, B^2 on whiteboard
	\bV_4 & = \frac{1}{n \h} \sumi (K^2 \br_p \br_p')(\Xhi) \left\{  \e_i \left[\mu(X_i) -  \br_p(X_i - \x)'\bbeta_p\right] \right\},  		\\  %A_{1,2} in JASA, 2AB on whiteboard
	\bV_5 & = \frac{1}{n \h} \sumi (K^2 \br_p \br_p' \br_p')(\Xhi) \left[\mu(X_i) -  \br_p(X_i - \x)'\bbeta_p\right] \Gp^{-1} \Op \left[\bY - \bR \bbeta_p\right]/n ,  		\\  %A_{1,4}+A_{1,7} in JASA, 2BC on whiteboard
	\bV_6 & = \frac{1}{n \h} \sumi (K^2 \br_p \br_p')(\Xhi)  \big\{ \br_p(\Xhi)'\Gp^{-1} \Op \left[\bY - \bR \bbeta_p\right] /n \big\}^2.      		 %A_{1,5}+A_{1,8} in JASA, C^2 on whiteboard
\end{align*}
With these terms in hand, define
\begin{itemize}
	
	\item $\tO = \sqrt{n \h}$
	
	\item $\sbp^2 = \v!^2 \be_\v'\Gp^{-1}  \Big( \bV_1 -  2 \bV_2  +  2 \bV_4  -  2 \breve{\bV}_5  + \breve{\bV}_6 \Big) \Gp^{-1} \be_\v$, where, with $\left[\Gp^{-1}\right]_{l_i, l_j}$ the $\{l_i+1, l_j+1\}$ element of $\Gp^{-1}$, we define
		\begin{align*}
%			\breve{V}_4  &  =   \frac{1}{n \h} \sumi (K \br_p \br_p')(\Xhi) \left\{  \e_i \E\left[K(\Xhi)\left(\mu(X_i) -  \br_p(X_i - \x)'\bbeta_p\right)\right] \right\},   		\\
			\breve{\bV}_5  &  =   \sum_{l_i=0}^p \sum_{l_j=0}^p  \left[\Gp^{-1}\right]_{l_i, l_j} \E\left[ (K^2 \br_p \br_p') (\Xhi) (\Xhi)^{l_i} \left(\mu(X_i) -  \br_p(X_i - \x)'\bbeta_p\right)  \right]   			 \\
			& \qquad \qquad \qquad\qquad    \times   \frac{1}{n \h} \sumj \bigg\{  K(\Xhj)  (\Xhj)^{l_j}  \left(Y_j - \br_p(X_j - \x)'\bbeta_p\right)  \bigg\},    			\\
			\breve{\bV}_6 & = \sum_{l_{i_1}=0}^p \sum_{l_{i_2}=0}^p \sum_{l_{j_1}=0}^p \sum_{l_{j_2}=0}^p   \left[\Gp^{-1}\right]_{l_{i_1}, l_{j_1}}  \left[\Gp^{-1}\right]_{l_{i_2}, l_{j_2}}  \E\left[ \h^{-1} (K^2 \br_p \br_p')(\Xhi)   (\Xhi)^{l_{i_1} + l_{i_2}}   \right]		  			 \\
			& \qquad  \times  
			\frac{1}{(n \h)^2}   \sumj \sumk    K(\Xhj) (\Xhj)^{l_{j_1}}  \left(Y_j - \br_p(X_j - \x)'\bbeta_p\right)  K(\Xhk)   (\Xhk)^{l_{j_2}}  \left(Y_k - \br_p(X_k - \x)'\bbeta_p\right) .
		\end{align*}
		
\end{itemize}
Next, using Equation \eqref{suppeqn:bias 1} to rewrite $\mu^{(\v)}$, canceling $\h^\v$, and adding and subtracting $\sbp^{-1}$, write $\tp$ as
\begin{align*}
	\tp & = \shatp^{-1} \sqrt{n \h^{1 + 2\v}} (\thatp - \tf)   		\\
	& = \shatp^{-1} \sqrt{n \h^{1 + 2\v}} \v! \be_\v' \Gp^{-1} \Op \left( \bY  - \bR \bbeta_p \right)/ (n\h^\v)  		\\
	& = \shatp^{-1} \tO \v! \be_\v' \Gp^{-1} \Op \left( \bY  - \bR \bbeta_p \right)/ n  		\\
	& =  \sbp^{-1} \tO \v! \be_\v'  \Gp^{-1} \Op \left( \bY  - \bR \bbeta_p \right)/ n   +    \left( \shatp^{-1} -  \sbp^{-1}\right) \tO \v! \be_\v' \Gp^{-1} \Op \left( \bY  - \bR \bbeta_p \right)/ n  		\\
	& =: \breve{T}  +  U_n.
\end{align*}
Then, referring back to Equation \eqref{suppeqn:step 1}, we have 
\[
	\P \left[ \tp < z \right] = \P \left[ \breve{T}  +  U_n < z \right],
\]
with 
\begin{equation}
	\label{suppeqn:U}
	U_n = \left( \shatp^{-1} -  \sbp^{-1}\right) \tO \v! \be_\v' \Gp^{-1} \Op \left( \bY  - \bR \bbeta_p \right)/ n
\end{equation}
and
\begin{align}
	\label{suppeqn:breve T}
	\breve{T} & =  \sbp^{-1} \tO \v! \be_\v'  \Gp^{-1} \Op \left( \bY  - \bR \bbeta_p \right)/ n.
\end{align}
As required, $\breve{T}:= \breve{T}(\tO^{-1} \sumi \bZ_i)$ is a smooth function of the sample average of $\bZ_i$, which is given by
\begin{align}
	\begin{split}
		\label{suppeqn:bZ}
		\bZ_i  = \Bigg(
		& \Big\{  (K \br_p)(\Xhi) (Y_i - \br_p(X_i - \x)'\bbeta_p)  \Big\}'  ,    		\\ 		  %from  numerator
		& \vech\Big\{  (K \br_p \br_p')(\Xhi)  \Big\}' ,    		\\ 		  %from  numerator, Gamma matrix
		& \vech\Big\{  (K^2 \br_p \br_p')(\Xhi) \e_i^2   \Big\}'  ,    		\\ 		  %from  \bV_1 
		& \vech\Big\{  (K^2 \br_p \br_p' )(\Xhi) (\Xhi)^0 \e_i   \Big\}', \vech\Big\{  (K^2 \br_p \br_p' )(\Xhi) (\Xhi)^1 \e_i   \Big\}',   		\\
		& \qquad \vech\Big\{  (K^2 \br_p \br_p' )(\Xhi) (\Xhi)^2 \e_i   \Big\}' ,\ldots, \vech\Big\{  (K^2 \br_p \br_p' )(\Xhi) (\Xhi)^p \e_i   \Big\}' ,	\\   %from   \bV_2
		& \vech\Big\{  (K^2 \br_p \br_p')(\Xhi) \big\{  \e_i \big[\mu(X_i) -  \br_p(X_i - \x)'\bbeta_p\big] \big\}   \Big\}'   		  %from  \bV_4
		\Bigg)'.
	\end{split}
\end{align}
In order of their listing above, these pieces come from (i) the ``score'' portion of the numerator, (ii) the ``Gram'' matrix $\Gp$, (iii) $\bV_1$, (iv) $\bV_2$, and (v) $\bV_4$. Notice that $\breve{\bV}_5$ and $\breve{\bV}_6$ do not add any additional elements to $\bZ_i$.

Equation \eqref{suppeqn:step 1} now follows from Lemma \ref{supplem:delta}(a), which completes {\bf Step (I)}, if we can show that 
\begin{equation}
	\label{suppeqn:delta}
	r_{\tp,\f}^{-1} \P [|U_n| > r_n] = o(1),
\end{equation}
where $r_{\tp,\f} = \max\{\tO^{-2}, \btp^2, \tO^{-1} \btp \}$ and $r_n = o(r_{\tp,\f})$. 

We now establish that Equation \eqref{suppeqn:delta} holds. First 
\begin{equation*}
	\frac{1}{\shatp} = \frac{1}{\sbp} \left( \frac{\shatp^2}{\sbp^2} \right)^{-1/2} = \frac{1}{\sbp} \left( 1 + \frac{ \shatp^2 - \sbp^2}{\sbp^2} \right)^{-1/2},
\end{equation*}
and hence a Taylor expansion gives
\footnote{
	It is not necessary to retain higher order terms in the Taylor series, for example via
	\[\frac{1}{\shatp} = \frac{1}{\sbp} \left[ 1 - \frac{1}{2} \frac{ \shatp^2 - \sbp^2}{\sbp^2}  + \frac{1}{2!} \frac{3}{4} \left( \frac{ \shatp^2 - \sbp^2}{\sbp^2} \right)^2   - \frac{1}{3!} \frac{15}{8} \left( \frac{ \shatp^2 - \sbp^2}{\sbp^2} \right)^3 \frac{\sbp^7}{\bar{\sigma}^7}     \right],\]
	because $\sbp^2$ is constructed exactly to retain all the important terms from $\shatp^2$. Put differently, because $(\shatp^2 - \sbp^2) \tO \v! \be_\v' \Gp^{-1} \Op \left( \bY  - \bR \bbeta_p \right)/ n$ will be shown to be ignorable in the process of verifying Equation \eqref{suppeqn:delta}, it is immediate that terms from $(\shatp^2 - \sbp^2)^2$ can also be ignored, as they are higher order. A longer Taylor expansion can be useful when computing the terms of the Edgeworth expansion.
}
	\[\frac{1}{\shatp} = \frac{1}{\sbp} \left[ 1 - \frac{1}{2} \frac{ \shatp^2 - \sbp^2}{\sbp^2}  + \frac{1}{2!} \frac{3}{4} \left( \frac{ \shatp^2 - \sbp^2}{\sbp^2} \right)^2 \frac{\sbp^5}{\bar{\sigma}^5}     \right],\]
for a point $\bar{\sigma}^2 \in [\sbp^2, \shatp^2]$, and so
\begin{equation}
	\label{suppeqn:variance taylor}
	\shatp^{-1} - \sbp^{-1} =  - \frac{1}{2} \frac{ \shatp^2 - \sbp^2}{\sbp^3}  + \frac{3}{8} \frac{  \left( \shatp^2 - \sbp^2\right)^2}{\bar{\sigma}^5} .
\end{equation}
Plugging this into the definition of $U_n$ gives
\begin{equation*}
	U_n = \left(  - \frac{1}{2 \sbp^3}  + \frac{3}{8} \frac{ \shatp^2 - \sbp^2}{\bar{\sigma}^5} \right) \left( \shatp^2 - \sbp^2\right) \tO \v! \be_\v' \Gp^{-1} \Op \left( \bY  - \bR \bbeta_p \right)/ n.
\end{equation*}
Therefore, if $\left| \shatp^2 - \sbp^2 \right| = o_\P(1)$, the result in \eqref{suppeqn:delta} will hold, and {\bf Step (I)} will be complete, once we have shown that
\begin{align}
	r_{\tp,\f}^{-1} & \P \left[\left| \left( \shatp^2 - \sbp^2\right) \tO \v! \be_\v' \Gp^{-1} \Op \left( \bY  - \bR \bbeta_p \right)/ n \right|  > r_n  \right]    		\nonumber\\
	& = r_{\tp,\f}^{-1} \P \left[\left|  \left(  \v!^2 \be_\v'\Gp^{-1}  \left( \bV_3   -  2 [\bV_5 - \breve{\bV}_5]  +   [\bV_6 - \breve{\bV}_6]  \right) \Gp^{-1} \be_\v \right) \tO \v! \be_\v' \Gp^{-1} \Op \left( \bY  - \bR \bbeta_p \right)/ n \right|  > r_n  \right]   		\nonumber \\
	& = o(1).  		\label{suppeqn:delta 1}
\end{align}
Recall that $r_{\tp,\f} = \max\{\tO^{-2}, \btp^2, \tO^{-1} \btp \}$ and $r_n = o(r_{\tp,\f})$. This is what we now verify one term at a time.

First, for the $\bV_3$ term, we claim that
\begin{align}
	r_{\tp,\f}^{-1} \P & \left[\Big| \v!^2 \be_\v'\Gp^{-1}   \bV_3  \Gp^{-1} \be_\v \tO \v! \be_\v' \Gp^{-1} \Op \left( \bY  - \bR \bbeta_p \right)/ n \Big|  > r_n  \right]   		\nonumber \\
	& \leq r_{\tp,\f}^{-1} \P \left[\Big| \v!^2 \be_\v'\Gp^{-1}   \left( \bV_3 - \E[\bV_3] \right)  \Gp^{-1} \be_\v \tO \v! \be_\v' \Gp^{-1} \Op \left( \bY  - \bM \right)/ n \Big|  > r_n  \right]  		\nonumber \\
	& \quad + r_{\tp,\f}^{-1} \P \left[\Big| \v!^2 \be_\v'\Gp^{-1}   \E[\bV_3] \Gp^{-1} \be_\v \tO \v! \be_\v' \Gp^{-1} \Op \left( \bY  - \bM \right)/ n \Big|  > r_n  \right]  		\nonumber \\
	& \quad + r_{\tp,\f}^{-1} \P \left[\Big| \v!^2 \be_\v'\Gp^{-1}   \left( \bV_3 - \E[\bV_3] \right)  \Gp^{-1} \be_\v \tO \v! \be_\v' \Gp^{-1} \Op \left( \bM  - \bR \bbeta_p \right)/ n \Big|  > r_n  \right]  		\nonumber \\
	& \quad + r_{\tp,\f}^{-1} \P \left[\Big| \v!^2 \be_\v'\Gp^{-1}   \E[\bV_3] \Gp^{-1} \be_\v \tO \v! \be_\v' \Gp^{-1} \Op \left( \bM  - \bR \bbeta_p \right)/ n \Big|  > r_n  \right]   		\nonumber \\
	& = o(1). 		\label{suppeqn:V3}
\end{align}
For the first term, using the elementary bounds (note that $|e_q| = 1$), 
\begin{align*}
	& r_{\tp,\f}^{-1} \P \left[\Big| \v!^2 \be_\v'\Gp^{-1}   \left( \bV_3 - \E[\bV_3] \right)  \Gp^{-1} \be_\v \tO \v! \be_\v' \Gp^{-1} \Op \left( \bY  - \bM \right)/ n \Big|  > r_n  \right]   			\\
	&\quad  \leq  r_{\tp,\f}^{-1} 3 \P \left[ \left| \Gp^{-1} \right| > C_\Gamma \right]    		\\
	& \qquad +  r_{\tp,\f}^{-1} \P \left[ \tO \left| \Op \left( \bY  - \bM \right)/ n \right|  > \delta \log(\tO)^{1/2}  \right]   		\\
	& \qquad +  r_{\tp,\f}^{-1} \P \bigg[\bigg| \frac{1}{n \h} \sumi \Big\{ (K^2 \br_p \br_p')(\Xhi) \left[\mu(X_i) -  \br_p(X_i - \x)'\bbeta_p\right]^2      		\\
				& \qquad \qquad \qquad    - \E\left[ (K^2 \br_p \br_p')(\Xhi) \left[\mu(X_i) -  \br_p(X_i - \x)'\bbeta_p\right]^2 \right] \Big\}  \bigg|  > r_n \frac{1}{(|e_q| q!  C_\Gamma)^3\delta \log(\tO)^{1/2}}   \bigg]     			\\
	& = o(1),
\end{align*}
by Lemmas \ref{supplem:gamma}, \ref{supplem:var}, and \ref{supplem:bias2}. In applying the last, take the constant to be $ (|e_q| q!  C_\Gamma)^{-3} \delta^{-1}$ and note that $r_n = o(r_{\tp,\f})$ may be chosen such that $r_n \log(\tO)^{-1/2}$ vanishes slower than (i.e.\ is larger than) $\btp^2 \tO^{-2} \log(\tO)^\gamma$, making the probability in the penultimate line bounded by the one in the Lemma. For example, take $r_n = \btp \tO^{-1} \log(\tO)^{-1/2-\gamma}$ and note that
	\[\frac{r_n}{\log(\tO)^{1/2}}    =   \left(\frac{\btp}{\tO}\right)^2 \log(\tO)^\gamma  \left[ \left( \frac{\tO}{\btp}\right)^2 \frac{r_n}{\log(\tO)^{1/2 + \gamma}} \right] 
	    =   \left(\frac{\btp}{\tO}\right)^2 \log(\tO)^\gamma  \left[  \frac{\tO}{\btp} \right] ,   \]
where factor in square brackets diverges by assumption.

The second term required for result \eqref{suppeqn:V3} obeys
\begin{align*}
	& r_{\tp,\f}^{-1} \P \left[\Big| \v!^2 \be_\v'\Gp^{-1}  \E[\bV_3]  \Gp^{-1} \be_\v \tO \v! \be_\v' \Gp^{-1} \Op \left( \bY  - \bM \right)/ n \Big|  > r_n  \right]   			\\
	&\quad  \leq  r_{\tp,\f}^{-1} 3 \P \left[ \left| \Gp^{-1} \right| > C_\Gamma \right]    		\\
	& \qquad +  r_{\tp,\f}^{-1} \P \left[ \tO \left| \Op \left( \bY  - \bM \right)/ n \right|  > \log(\tO)^{1/2}  \left\{\frac{\tO^2}{\btp^2} r_n \frac{1}{(|e_q| q! C_\Gamma)^3\log(\tO)^{1/2}}\right\}  \right]   		\\
	& = o(1),
\end{align*}
using  Lemmas \ref{supplem:gamma} and \ref{supplem:var}, as the term in braces diverges (e.g.\ for $r_n = \btp^2 \log(\tO)^{-1/2}$) and $\E[\bV_3] = O (\btp^2 \tO^{-2})$ as follows:
\begin{align*}
	\E[\bV_3] & = \frac{1}{n \h} \sumi \E\left[(K^2 \br_p \br_p')(\Xhi)  \left[\mu(X_i) -  \br_p(X_i - \x)'\bbeta_p\right]^2 \right]  		\\
	& = \E\left[ \h^{-1} (K^2 \br_p \br_p')(\Xhi)  \left[\mu(X_i) -  \br_p(X_i - \x)'\bbeta_p\right]^2 \right]  		\\
	& = \frac{\btp^2}{\tO^2} \E\left[ \h^{-1} (K^2 \br_p \br_p')(\Xhi)  \left[\frac{\tO}{\btp}\left( \mu(X_i) -  \br_p(X_i - \x)'\bbeta_p\right)\right]^2 \right]  		\\
	& = O \left( \frac{\btp^2}{\tO^2} \right).
\end{align*}

The third term required for result \eqref{suppeqn:V3} obeys
\begin{align*}
	& r_{\tp,\f}^{-1} \P \left[\Big| \v!^2 \be_\v'\Gp^{-1}   \left( \bV_3 - \E[\bV_3] \right)  \Gp^{-1} \be_\v \tO \v! \be_\v' \Gp^{-1} \Op \left( \bM  - \bR \bbeta_p \right)/ n \Big|  > r_n  \right]   			\\
	&\quad  \leq  r_{\tp,\f}^{-1} 3 \P \left[ \left| \Gp^{-1} \right| > C_\Gamma \right]    		\\
	& \qquad +  r_{\tp,\f}^{-1} \P \left[ \left| \Op \left( \bM  - \bR \bbeta_p \right)/ n \right|  > \log(\tO)^{1/2}  \right]   		\\
	& \qquad +  r_{\tp,\f}^{-1} \P \bigg[\bigg|  \frac{1}{n \h} \sumi \Big\{ (K^2 \br_p \br_p')(\Xhi) \left[\mu(X_i) -  \br_p(X_i - \x)'\bbeta_p\right]^2      		\\
				& \qquad \qquad \qquad    - \E\left[ (K^2 \br_p \br_p')(\Xhi) \left[\mu(X_i) -  \br_p(X_i - \x)'\bbeta_p\right]^2 \right] \Big\}  \bigg|  > r_n \frac{1}{\tO (|e_q| q! C_\Gamma)^3\log(\tO)^{1/2}}   \bigg]     			\\
	& = o(1),
\end{align*}
by Lemmas \ref{supplem:gamma}, \ref{supplem:bias1}, and \ref{supplem:bias2}. In applying the last, take $\delta = (|e_q| q! C_\Gamma)^{-3}$ and note that $r_n = o(r_{\tp,\f})$ may be chosen such that $r_n \log(\tO)^{-1/2}$ vanishes slower than (i.e.\ is larger than) $\btp^2 \tO^{-2} \log(\tO)^\gamma$, making the probability in the penultimate line bounded by the one in the Lemma. For example, take $r_n = \btp \tO^{-1} \log(\tO)^{-\gamma}$ and note that
	\[\frac{r_n}{\tO \log(\tO)^{1/2}}    =   \left(\frac{\btp}{\tO}\right)^2 \log(\tO)^\gamma  \left[ \left( \frac{\tO}{\btp}\right)^2 \frac{r_n}{\tO \log(\tO)^{1/2 + \gamma}} \right] 
	    =   \left(\frac{\btp}{\tO}\right)^2 \log(\tO)^\gamma  \left[  \frac{1}{\btp \log(\tO)^{1/2 + 2 \gamma}} \right] ,   \]
where factor in square brackets diverges by assumption.

The fourth term follows the same pattern as the second, using Lemma \ref{supplem:bias1} in place of Lemma \ref{supplem:var}, the same way the third term followed the pattern of the first. This completes the proof of result \eqref{suppeqn:V3}.

Turning to the $\bV_5$ terms, first observe that, when all its components are considered, $\bV_5$ is a $(p+1) \times (p+1)$ matrix (from $(\br_p \br_p')(\Xhi)$) multiplied by a scalar. We write out
\begin{align*}
	\br_p' (\Xhi) \Gp^{-1} \Op \left[\bY - \bR \bbeta_p\right]/n  &  =  \frac{1}{n\h} \sumj  \left\{ \br_p' (\Xhi) \Gp^{-1} \br_p' (\Xhj) \right\} K(\Xhj) \left(Y_j - \br_p(X_j - \x)'\bbeta_p\right)    		\\
	  &  =  \frac{1}{n\h} \sumj  \left\{ \sum_{l_i=0}^p \sum_{l_j=0}^p \left[\Gp^{-1}\right]_{l_i, l_j} (\Xhi)^{l_i} (\Xhj)^{l_j} \right\} K(\Xhj) \left(Y_j - \br_p(X_j - \x)'\bbeta_p\right) .
\end{align*}
where $\left[\Gp^{-1}\right]_{l_i, l_j}$ is the $\{l_i+1, l_j+1\}$ element of $\Gp^{-1}$, which is well-behaved by Lemma \ref{supplem:gamma}. We make use of this in order to write
\begin{align}
	\v!^2 \be_\v'\Gp^{-1}  \Big[\bV_5 \Big] \Gp^{-1} \be_\v   & = \v!^2 \be_\v'\Gp^{-1} \frac{1}{n \h} \sumi (K^2 \br_p \br_p' \br_p')(\Xhi) \left[\mu(X_i) -  \br_p(X_i - \x)'\bbeta_p\right] \Gp^{-1} \Op \left[\bY - \bR \bbeta_p\right]/n\Gp^{-1} \be_\v    			\nonumber \\
	& =   \sum_{l_i=0}^p \sum_{l_j=0}^p  \v!^2 \be_\v'\Gp^{-1}  \left[\Gp^{-1}\right]_{l_i, l_j}  \frac{1}{(n \h)^2} \sumi \sumj \bigg\{   (K^2 \br_p \br_p') (\Xhi) K(\Xhj) (\Xhi)^{l_i} (\Xhj)^{l_j}   			 \nonumber \\
	& \qquad \qquad \qquad\qquad    \times \left[\mu(X_i) -  \br_p(X_i - \x)'\bbeta_p\right]    \left(Y_j - \br_p(X_j - \x)'\bbeta_p\right)  \bigg\}    \Gp^{-1} \be_\v    			\nonumber \\
	& =:  \sum_{l_i=0}^p \sum_{l_j=0}^p   \v!^2 \be_\v'\Gp^{-1} \Big\{   V_{5,1}(l_i, l_j) + V_{5,2}(l_i, l_j)  \Big\}  \Gp^{-1} \be_\v  , 		\label{suppeqn:V5 1}
\end{align}
where $V_{5,1}(l_i, l_j)$ and $V_{5,2}(l_i, l_j)$ are the ``own'' and ``cross'' summands
\begin{align*}
	V_{5,1}(l_i, l_j) & : =    \left[\Gp^{-1}\right]_{l_i, l_j}  \frac{1}{(n \h)^2} \sumi \bigg\{   (K^3 \br_p \br_p') (\Xhi) (\Xhi)^{l_i + l_j}   			 \\
	& \qquad \qquad \qquad\qquad    \times \left[\mu(X_i) -  \br_p(X_i - \x)'\bbeta_p\right]   \left(Y_i - \br_p(X_i - \x)'\bbeta_p\right)  \bigg\}        		\\
	V_{5,2}(l_i, l_j) & :=  \left[\Gp^{-1}\right]_{l_i, l_j}  \frac{1}{(n \h)^2} \sumi \sumjNoti \bigg\{   (K^2 \br_p \br_p') (\Xhi) K(\Xhj) (\Xhi)^{l_i} (\Xhj)^{l_j}   			 \\
	& \qquad \qquad \qquad\qquad    \times \left[\mu(X_i) -  \br_p(X_i - \x)'\bbeta_p\right]   \left(Y_j - \br_p(X_j - \x)'\bbeta_p\right)  \bigg\} .
\end{align*}

Recall that the goal is result \eqref{suppeqn:delta 1}. We will study one term of the double sum \eqref{suppeqn:V5 1}, i.e.\ $V_{5,1}(l_i, l_j)$ and $V_{5,2}(l_i, l_j)$ for a fixed pair $\{l_i,l_j\}$, as all terms are identically handled. If each term is ignorable in the expansion, then it follows that
\begin{align}
	r_{\tp,\f}^{-1} & \P  \left[\left|  \left(  \v!^2 \be_\v'\Gp^{-1}  \left( - 2 [\bV_5 - \breve{\bV}_5  \right) \Gp^{-1} \be_\v \right) \tO \v! \be_\v' \Gp^{-1} \Op \left( \bY  - \bR \bbeta_p \right)/ n \right|  > r_n  \right]   			\nonumber \\
	& \leq  C \max_{ 0 \leq l_i, l_j \leq p} r_{\tp,\f}^{-1} \P \bigg[\bigg|  \left(  \v!^2 \be_\v'\Gp^{-1}  \left( V_{5,1}(l_i, l_j) + V_{5,2}(l_i, l_j)   -  \breve{V}_{5,2}(l_i, l_j)  \right) \Gp^{-1} \be_\v \right)   		\nonumber \\
			& \qquad\qquad\qquad\qquad\qquad  \times \;  \tO \v! \be_\v' \Gp^{-1} \Op \left( \bY  - \bR \bbeta_p \right)/ n \bigg|  > r_n  \bigg]   		\nonumber \\
	& = o(1), 				\label{suppeqn:V5}
\end{align}
by Boole's inequality and $p$ fixed.

As hinted at in this display, $\breve{\bV}_5$ will be constructed from the pieces of $V_{5,2}(l_i, l_j)$ which contribute to the expansion. We first show that the $V_{5,1}(l_i, l_j)$ terms may be ignored. Begin by splitting $(Y_i - \br_p(X_i - \x)'\bbeta_p) = \e_i +  (\mu(X_i) - \br_p(X_ - \x)'\bbeta_p )$ everywhere, as the ``variance'' and ``bias'' type pieces have different rates, which must be accounted for: 
\begin{align*}
	r_{\tp,\f}^{-1} &  \P \left[\left|  \left(  \v!^2 \be_\v'\Gp^{-1}  \left( V_{5,1}(l_i, l_j)   \right) \Gp^{-1} \be_\v \right)   \;  \tO \v! \be_\v' \Gp^{-1} \Op \left( \bY  - \bR \bbeta_p \right)/ n \right|  > r_n  \right]    			\\
	& \leq r_{\tp,\f}^{-1} \P \left[\left|  \left(  \v!^2 \be_\v'\Gp^{-1}  \left( V_{5,1}(l_i, l_j)   \right) \Gp^{-1} \be_\v \right)   \;  \tO \v! \be_\v' \Gp^{-1} \Op \left( \bY  - \bM \right)/ n \right|  > r_n  \right]    			\\
	& \quad + r_{\tp,\f}^{-1} \P \left[\left|  \left(  \v!^2 \be_\v'\Gp^{-1}  \left( V_{5,1}(l_i, l_j)   \right) \Gp^{-1} \be_\v \right)   \;  \tO \v! \be_\v' \Gp^{-1} \Op \left( \bM  - \bR \bbeta_p \right)/ n \right|  > r_n  \right]    			\\
	& \leq r_{\tp,\f}^{-1} \P \Bigg[\Bigg|  \Bigg(  \v!^2 \be_\v'\Gp^{-1}  \Bigg(  \left[\Gp^{-1}\right]_{l_i, l_j}  \frac{1}{(n \h)^2} \sumi \bigg\{   (K^3 \br_p \br_p') (\Xhi) (\Xhi)^{l_i + l_j}   	 \\ & \qquad \qquad \qquad\qquad		\times\; \left[\mu(X_i) -  \br_p(X_i - \x)'\bbeta_p\right]^2  \bigg\}   \Bigg) \Gp^{-1} \be_\v \Bigg)   \;  \tO \v! \be_\v' \Gp^{-1} \Op \left( \bM  - \bR \bbeta_p \right)/ n \Bigg|  > r_n  \Bigg]    			\\ %bias, bias
	& \quad + r_{\tp,\f}^{-1} \P \Bigg[\Bigg|  \Bigg(  \v!^2 \be_\v'\Gp^{-1}  \Bigg(  \left[\Gp^{-1}\right]_{l_i, l_j}  \frac{1}{(n \h)^2} \sumi \bigg\{   (K^3 \br_p \br_p') (\Xhi) (\Xhi)^{l_i + l_j}   	 \\ & \qquad \qquad \qquad\qquad		\times\; \left[\mu(X_i) -  \br_p(X_i - \x)'\bbeta_p\right]^2  \bigg\}   \Bigg) \Gp^{-1} \be_\v \Bigg)   \;  \tO \v! \be_\v' \Gp^{-1} \Op \left( \bY - \bM \right)/ n \Bigg|  > r_n  \Bigg]    			\\ %bias, variance
	& \quad + r_{\tp,\f}^{-1} \P \Bigg[\Bigg|  \Bigg(  \v!^2 \be_\v'\Gp^{-1}  \Bigg(  \left[\Gp^{-1}\right]_{l_i, l_j}  \frac{1}{(n \h)^2} \sumi \bigg\{   (K^3 \br_p \br_p') (\Xhi) (\Xhi)^{l_i + l_j}   	 \\ & \qquad \qquad \qquad\qquad		\times\; \left[\mu(X_i) -  \br_p(X_i - \x)'\bbeta_p\right] \e_i  \bigg\}   \Bigg) \Gp^{-1} \be_\v \Bigg)   \;  \tO \v! \be_\v' \Gp^{-1} \Op \left( \bM   - \bR \bbeta_p  \right)/ n \Bigg|  > r_n  \Bigg]    			\\ %variance, bias
	& \quad + r_{\tp,\f}^{-1} \P \Bigg[\Bigg|  \Bigg(  \v!^2 \be_\v'\Gp^{-1}  \Bigg(  \left[\Gp^{-1}\right]_{l_i, l_j}  \frac{1}{(n \h)^2} \sumi \bigg\{   (K^3 \br_p \br_p') (\Xhi) (\Xhi)^{l_i + l_j}   	 \\ & \qquad \qquad \qquad\qquad		\times\; \left[\mu(X_i) -  \br_p(X_i - \x)'\bbeta_p\right] \e_i  \bigg\}   \Bigg) \Gp^{-1} \be_\v \Bigg)   \;  \tO \v! \be_\v' \Gp^{-1} \Op \left( \bY - \bM \right)/ n \Bigg|  > r_n  \Bigg]. %variance, variance
\end{align*}
For the first (i.e. the first term on the right hand side of the last inequality)
\begin{align*}
	r_{\tp,\f}^{-1} & \P \Bigg[\Bigg|  \Bigg(  \v!^2 \be_\v'\Gp^{-1}  \Bigg(  \left[\Gp^{-1}\right]_{l_i, l_j}  \frac{1}{(n \h)^2} \sumi \bigg\{   (K^3 \br_p \br_p') (\Xhi) (\Xhi)^{l_i + l_j}   	 \\ & \qquad \qquad \qquad\qquad		\times\; \left[\mu(X_i) -  \br_p(X_i - \x)'\bbeta_p\right]^2  \bigg\}   \Bigg) \Gp^{-1} \be_\v \Bigg)   \;  \tO \v! \be_\v' \Gp^{-1} \Op \left( \bM  - \bR \bbeta_p \right)/ n \Bigg|  > r_n  \Bigg]    			\\
	& \leq  r_{\tp,\f}^{-1} 4 \P \left[ \left| \Gp^{-1} \right| > C_\Gamma \right]    		\\
	& \quad +  r_{\tp,\f}^{-1} \P \left[ \left| \Op \left( \bM  - \bR \bbeta_p \right)/ n \right|  >  \log(\tO)^{1/2}  \right]   		\\
	& \quad +  r_{\tp,\f}^{-1} \P \bigg[\bigg| \frac{1}{n \h} \sumi \bigg\{   (K^3 \br_p \br_p') (\Xhi) (\Xhi)^{l_i + l_j}   	 \\ & \qquad \qquad \qquad\qquad\qquad\qquad		\times\; \left[\mu(X_i) -  \br_p(X_i - \x)'\bbeta_p\right]^2  \bigg\}  \bigg|  > r_n \frac{n \h}{\tO (|e_q| q! )^3C_\Gamma^4  \log(\tO)^{1/2}}   \bigg]     			\\
	& = o(1),
\end{align*}
by Lemmas \ref{supplem:gamma} and \ref{supplem:bias1}, the latter applied twice, and the fact that, for $r_n = \btp \tO \log(\tO)^{-\gamma}$, with any $\gamma > 0$
\[r_n \frac{n \h}{\tO (|e_q| q! )^3C_\Gamma^4  \log(\tO)^\gamma} \asymp \frac{\btp}{\tO} \log(\tO)^{1/2} \left[ \frac{\tO}{ \log(\tO)^{1/2 + 2 \gamma} } \right], \]
and the factor in square brackets diverges. The rest of the $V_{5,1}(l_i, l_j)$ terms are handled by exactly the same steps, but using Lemmas \ref{supplem:var}, \ref{supplem:bias1}, and \ref{supplem:bias3} as needed for the final convergence. This establishes the $V_{5,1}(l_i, l_j)$ part of Equation \eqref{suppeqn:V5}.

Turning to the $V_{5,2}(l_i, l_j)$ part of Equation \eqref{suppeqn:V5}, we again begin by splitting $(Y_i - \br_p(X_i - \x)'\bbeta_p) = \e_i +  (\mu(X_i) - \br_p(X_ - \x)'\bbeta_p )$ everywhere, just like above, 
\begin{align*}
	& r_{\tp,\f}^{-1}  \P \left[\left|  \left(  \v!^2 \be_\v'\Gp^{-1}  \left( V_{5,2}(l_i, l_j)   \right) \Gp^{-1} \be_\v \right)   \;  \tO \v! \be_\v' \Gp^{-1} \Op \left( \bY  - \bR \bbeta_p \right)/ n \right|  > r_n  \right]    			\\
	& \leq r_{\tp,\f}^{-1} \P \Bigg[\Bigg|  \Bigg(  \v!^2 \be_\v'\Gp^{-1}  \Bigg(  \left[\Gp^{-1}\right]_{l_i, l_j}  \frac{1}{(n \h)^2} \sumi \sumjNoti \bigg\{   (K^2 \br_p \br_p') (\Xhi) K(\Xhj) (\Xhi)^{l_i} (\Xhj)^{l_j}   			 \\
	& \qquad \qquad     \times \left[\mu(X_i) -  \br_p(X_i\! - \! \x)'\bbeta_p\right]   \left(\e_j \right)  \bigg\}    \Bigg) \Gp^{-1} \be_\v \Bigg)   \;  \tO \v! \be_\v' \Gp^{-1} \Op \left( \bY - \bM  \right)/ n \Bigg|  > r_n  \Bigg]    			\\  %variance, variance
	& \quad + r_{\tp,\f}^{-1} \P \Bigg[\Bigg|  \Bigg(  \v!^2 \be_\v'\Gp^{-1}  \Bigg(  \left[\Gp^{-1}\right]_{l_i, l_j}  \frac{1}{(n \h)^2} \sumi \sumjNoti \bigg\{   (K^2 \br_p \br_p') (\Xhi) K(\Xhj) (\Xhi)^{l_i} (\Xhj)^{l_j}   			 \\
	& \qquad \qquad     \times \left[\mu(X_i) -  \br_p(X_i\! - \! \x)'\bbeta_p\right]   \left(\mu(X_j) - \br_p(X_j\! - \! \x)'\bbeta_p\right)  \bigg\}    \Bigg) \Gp^{-1} \be_\v \Bigg)   \;  \tO \v! \be_\v' \Gp^{-1} \Op \left( \bY - \bM  \right)/ n \Bigg|  > r_n  \Bigg]    			\\  %bias, variance
	& \quad + r_{\tp,\f}^{-1} \P \Bigg[\Bigg|  \Bigg(  \v!^2 \be_\v'\Gp^{-1}  \Bigg(  \left[\Gp^{-1}\right]_{l_i, l_j}  \frac{1}{(n \h)^2} \sumi \sumjNoti \bigg\{   (K^2 \br_p \br_p') (\Xhi) K(\Xhj) (\Xhi)^{l_i} (\Xhj)^{l_j}   			 \\
	& \qquad \qquad     \times \left[\mu(X_i) -  \br_p(X_i\! - \! \x)'\bbeta_p\right]   \left(\e_j\right)  \bigg\}    \Bigg) \Gp^{-1} \be_\v \Bigg)   \;  \tO \v! \be_\v' \Gp^{-1} \Op \left(  \bM    - \bR \bbeta_p \right)/ n \Bigg|  > r_n  \Bigg]    			\\  %variance, bias
	& \quad + r_{\tp,\f}^{-1} \P \Bigg[\Bigg|  \Bigg(  \v!^2 \be_\v'\Gp^{-1}  \Bigg(  \left[\Gp^{-1}\right]_{l_i, l_j}  \frac{1}{(n \h)^2} \sumi \sumjNoti \bigg\{   (K^2 \br_p \br_p') (\Xhi) K(\Xhj) (\Xhi)^{l_i} (\Xhj)^{l_j}   			 \\
	& \qquad \qquad     \times \left[\mu(X_i) -  \br_p(X_i\! - \! \x)'\bbeta_p\right]   \left(\mu(X_j) - \br_p(X_j\! - \! \x)'\bbeta_p\right)  \bigg\}    \Bigg) \Gp^{-1} \be_\v \Bigg)   \;  \tO \v! \be_\v' \Gp^{-1} \Op \left(  \bM    - \bR \bbeta_p \right)/ n \Bigg|  > r_n  \Bigg]    		  %variance, bias
\end{align*}

For the first term, which has two ``variance'' terms and one bias-type term:
\begin{align*}
	& r_{\tp,\f}^{-1} \P \Bigg[\Bigg|  \Bigg(  \v!^2 \be_\v'\Gp^{-1}  \Bigg(  \left[\Gp^{-1}\right]_{l_i, l_j}  \frac{1}{(n \h)^2} \sumi \sumjNoti \bigg\{   (K^2 \br_p \br_p') (\Xhi) K(\Xhj) (\Xhi)^{l_i} (\Xhj)^{l_j}   			 \\
	& \qquad \qquad     \times \left[\mu(X_i) -  \br_p(X_i\! - \! \x)'\bbeta_p\right]   \left(\e_j \right)  \bigg\}    \Bigg) \Gp^{-1} \be_\v \Bigg)   \;  \tO \v! \be_\v' \Gp^{-1} \Op \left( \bY - \bM  \right)/ n \Bigg|  > r_n  \Bigg]    			\\  %variance, variance
	& \leq  r_{\tp,\f}^{-1} 4 \P \left[ \left| \Gp^{-1} \right| > C_\Gamma \right]    		\\
	& \quad +  r_{\tp,\f}^{-1} \P \left[ \left| \Op \left( \bY - \bM \right)/ n \right|  >  C_1 \tO^{-1} \log(\tO)^{1/2}  \right]   		\\
	& \quad +  r_{\tp,\f}^{-1} \P \bigg[\bigg| \frac{1}{n \h} \sumj \bigg\{   K(\Xhj) (\Xhi)^{l_j} \e_j  \bigg\}  \bigg|  > C_2 \tO^{-1} \log(\tO)^{1/2} \bigg]     			\\
	& \quad +  r_{\tp,\f}^{-1} \P \bigg[\bigg| \frac{1}{n \h} \sumi \left\{   (K^2 \br_p \br_p') (\Xhi) (\Xhi)^{l_i}   	\left[\mu(X_i) -  \br_p(X_i - \x)'\bbeta_p\right]  \right\}  \bigg|  > r_n \frac{\tO^2}{\tO (|e_q| q! )^3C_\Gamma^4 C_1 C_2  \log(\tO)}   \bigg]     			\\
	& = o(1),
\end{align*}
by Lemmas \ref{supplem:gamma}, \ref{supplem:var} applied twice, and \ref{supplem:bias1}. For the last, note that for $r_n = \btp \tO \log(\tO)^{-\gamma}$, with $\gamma >0$,
\[r_n \frac{\tO^2}{\tO (|e_q| q! )^3C_\Gamma^4 C_1 C_2  \log(\tO)} \asymp \frac{\btp}{\tO} \log(\tO)^\gamma \left[  \frac{\tO}{ \log(\tO)^{1+2\gamma} }  \right], \]
and the term in square brackets diverges by assumption.

Turning to the second $V_{5,2}$ term (the third and fourth will be similar), which has one ``variance'' terms and two bias-type terms:, observe that
\begin{multline*}
	r_{\tp,\f}^{-1} \P \Bigg[\Bigg|  \Bigg(  \v!^2 \be_\v'\Gp^{-1}  \Bigg(  \left[\Gp^{-1}\right]_{l_i, l_j}  \frac{1}{(n \h)^2} \sumi \sumjNoti \bigg\{   (K^2 \br_p \br_p') (\Xhi) K(\Xhj) (\Xhi)^{l_i} (\Xhj)^{l_j}   			 \\
	     \times \left[\mu(X_i) -  \br_p(X_i\! - \! \x)'\bbeta_p\right]   \left(\mu(X_j) - \br_p(X_j\! - \! \x)'\bbeta_p\right)  \bigg\}    \Bigg) \Gp^{-1} \be_\v \Bigg)   \;  \tO \v! \be_\v' \Gp^{-1} \Op \left( \bY - \bM  \right)/ n \Bigg|  > r_n  \Bigg]    \neq o(1),
\end{multline*}
because, compared to the above, Lemma \ref{supplem:var} is applied only once, while Lemma \ref{supplem:bias1} is needed twice, instead of vice versa. The slower rate in the latter implies that this term can not be ignored. Thus pieces of this will contribute to $\breve{\bV}_5$. To see which, we will first center some bias terms. Just for notational ease, define the shorthand
\[ V_{5,2,i}  =  (K^2 \br_p \br_p') (\Xhi)  (\Xhi)^{l_i} \left[\mu(X_i) -  \br_p(X_i\! - \! \x)'\bbeta_p\right] \]
and
\[ V_{5,2,j} = K(\Xhj)  (\Xhj)^{l_j}\left[\mu(X_j) - \br_p(X_j\! - \! \x)'\bbeta_p\right]. \] 
The term in question is then
\begin{align*}
	& \Bigg(  \v!^2 \be_\v'\Gp^{-1}  \Bigg(  \left[\Gp^{-1}\right]_{l_i, l_j}  \frac{1}{(n \h)^2} \sumi \sumjNoti  V_{5,2,i} V_{5,2,j}   \Bigg) \Gp^{-1} \be_\v \Bigg)   \;  \tO \v! \be_\v' \Gp^{-1} \Op \left( \bY - \bM  \right)/ n    			\\
	& = \Bigg(  \v!^2 \be_\v'\Gp^{-1}  \Bigg(  \left[\Gp^{-1}\right]_{l_i, l_j}  \E[\h^{-1} V_{5,2,i}] \frac{1}{n \h} \sumj V_{5,2,j}   \Bigg) \Gp^{-1} \be_\v \Bigg)   \;  \tO \v! \be_\v' \Gp^{-1} \Op \left( \bY - \bM  \right)/ n    			\\
	& \quad + \Bigg(  \v!^2 \be_\v'\Gp^{-1}  \Bigg(  \left[\Gp^{-1}\right]_{l_i, l_j}  \frac{1}{n \h} \sumi  \left( V_{5,2,i} - \E[V_{5,2,i}]\right) \E[\h^{-1} V_{5,2,j}]   \Bigg) \Gp^{-1} \be_\v \Bigg)   \;  \tO \v! \be_\v' \Gp^{-1} \Op \left( \bY - \bM  \right)/ n     			\\
	& \quad + \Bigg(  \v!^2 \be_\v'\Gp^{-1}  \Bigg(  \left[\Gp^{-1}\right]_{l_i, l_j}  \frac{1}{(n \h)^2} \sumi \sumjNoti  \left( V_{5,2,i} - \E[V_{5,2,i}]\right) \left( V_{5,2,j} - \E[V_{5,2,j}] \right)   \Bigg) \Gp^{-1} \be_\v \Bigg)   \;  \tO \v! \be_\v' \Gp^{-1} \Op \left( \bY - \bM  \right)/ n 
\end{align*}
The first term here will be incorporated into $\breve{\bV}_5$, and thus into $\breve{T}$. Note that it is a smooth function of the $\bZ_i$ from Equation \eqref{suppeqn:bZ}, which is why we choose the centering the way we do, that is, keeping the term with $\E[\h^{-1} V_{5,2,i}]$ instead of $\E[\h^{-1} V_{5,2,j}]$. Doing the reverse would force further variables into the vector $\bZ_i$, and require a stronger Cram\'er's condition, which we seek to avoid.\footnote{\cite{Calonico-Cattaneo-Farrell2018_JASA,Calonico-Cattaneo-Farrell2018_JASAsupp} use such an approach, requiring not only a strengthening of Cram\'er's condition, but also in the process, ruling out the uniform kernel.}

The next term obeys
\begin{align*}
	& r_{\tp,\f}^{-1} \P \Bigg[\Bigg|  \Bigg(  \v!^2 \be_\v'\Gp^{-1}  \Bigg(  \left[\Gp^{-1}\right]_{l_i, l_j}  \frac{1}{n \h} \sumi  \left( V_{5,2,i} - \E[V_{5,2,i}]\right) \E[\h^{-1} V_{5,2,j}]   \Bigg) \Gp^{-1} \be_\v \Bigg)   \;  \tO \v! \be_\v' \Gp^{-1} \Op \left( \bY - \bM  \right)/ n \Bigg|  > r_n  \Bigg]    			\\
	& \leq  r_{\tp,\f}^{-1} 4 \P \left[ \left| \Gp^{-1} \right| > C_\Gamma \right]    		\\
	& \quad +  r_{\tp,\f}^{-1} \P \left[ \left| \Op \left( \bY - \bM \right)/ n \right|  >  C_1 \tO^{-1} \log(\tO)^{1/2}  \right]   		\\
	& \quad +  r_{\tp,\f}^{-1} \P \bigg[\bigg| \frac{1}{n \h} \sumi  \left( V_{5,2,i} - \E[V_{5,2,i}]\right)  \bigg|  > r_n \frac{\tO}{C \btp \tO \log(\tO)^{1/2} }   \bigg]     			\\
	& = o(1),
\end{align*}
by Lemmas \ref{supplem:gamma}, \ref{supplem:var}, and \ref{supplem:bias2}, the fact that $\E[\h^{-1} V_{5,2,j}] \asymp \tO^{-1} \btp$ (see Section \ref{supp:bias lp} or the computation for $\E[\bV_3]$ above), and that for $r_n =  \btp \tO^{-1} \log(\tO)^{-\gamma}$, with any $\gamma >0$,
\[ r_n \frac{\tO}{C \btp \tO \log(\tO)^{1/2} } \asymp \frac{\btp}{\tO} \log(\tO)^\gamma \left[ \frac{1}{\btp \log(\tO)^{1/2 + 2\gamma} } \right] \]
the factor in square brackets diverges by assumption.

The final piece of the second $V_{5,2}$ term similarly obeys
\begin{align*}
	& r_{\tp,\f}^{-1} \P \Bigg[\Bigg|   \Bigg(  \v!^2 \be_\v'\Gp^{-1}  \Bigg(  \left[\Gp^{-1}\right]_{l_i, l_j}  \frac{1}{(n \h)^2} \sumi \sumjNoti  \left( V_{5,2,i} - \E[V_{5,2,i}]\right) \left( V_{5,2,j} - \E[V_{5,2,j}] \right)   \Bigg) \Gp^{-1} \be_\v \Bigg)     		\\  		& \qquad\qquad\qquad \qquad\qquad\qquad \qquad\qquad\qquad \times \;  \tO \v! \be_\v' \Gp^{-1} \Op \left( \bY - \bM  \right)/ n  \Bigg|  > r_n  \Bigg]    			\\
	& \leq  r_{\tp,\f}^{-1} 4 \P \left[ \left| \Gp^{-1} \right| > C_\Gamma \right]    		\\
	& \quad +  r_{\tp,\f}^{-1} \P \left[ \left| \Op \left( \bY - \bM \right)/ n \right|  >  C_1 \tO^{-1} \log(\tO)^{1/2}  \right]   		\\
	& \quad +  r_{\tp,\f}^{-1} \P \bigg[\bigg| \frac{1}{n \h} \sumj  \left( V_{5,2,j} - \E[V_{5,2,j}]\right)  \bigg|  > \frac{\btp}{\tO} \log(\tO)^\gamma   \bigg]    + o(1)  			\\
	& \quad +  r_{\tp,\f}^{-1} \P \bigg[\bigg| \frac{1}{n \h} \sumi  \left( V_{5,2,i} - \E[V_{5,2,i}]\right)  \bigg|  > r_n \frac{\tO}{C \btp \tO \log(\tO)^{1/2 + \gamma} }   \bigg]     			\\
	& = o(1),
\end{align*}
by Lemmas \ref{supplem:gamma}, \ref{supplem:var}, and \ref{supplem:bias2} applied twice, and that for $r_n =  \btp \tO^{-1} \log(\tO)^{-\gamma}$, with any $\gamma >0$,
\[ r_n \frac{\tO}{C \btp \tO \log(\tO)^{1/2} } \asymp \frac{\btp}{\tO} \log(\tO)^\gamma \left[ \frac{1}{\btp \log(\tO)^{1/2 + 3\gamma} } \right] \]
the factor in square brackets diverges by assumption. The $o(1)$ factor in the third to last line accounts for the missing term in the sum over the ``$j$'' index.

Comparing the first and second $V_{5,2}$ terms, we see the the first was ignorable because it had two ``variance'' type terms, while the second had only one. This generalizes to the third and fourth $V_{5,2}$ terms, the third being just like the second and the fourth having three bias-type terms. For these, the same centering must be done as was done here. The bounding is then nearly identical. Putting these pieces together, recall the definition of $V_{5,2}(l_i, l_j)$:
\begin{align*}
	V_{5,2}(l_i, l_j) & :=  \left[\Gp^{-1}\right]_{l_i, l_j}  \frac{1}{(n \h)^2} \sumi \sumjNoti \bigg\{   (K^2 \br_p \br_p') (\Xhi) K(\Xhj) (\Xhi)^{l_i} (\Xhj)^{l_j}   			 \\
	& \qquad \qquad \qquad\qquad    \times \left[\mu(X_i) -  \br_p(X_i - \x)'\bbeta_p\right]   \left(Y_j - \br_p(X_j - \x)'\bbeta_p\right)  \bigg\}.
\end{align*}
Following the logic above, always centering the ``$i$'' term first, we define 
\begin{align*}
	\breve{V}_{5,2}(l_i, l_j) & :=  \left[\Gp^{-1}\right]_{l_i, l_j} \E\left[ (K^2 \br_p \br_p') (\Xhi) (\Xhi)^{l_i} \left(\mu(X_i) -  \br_p(X_i - \x)'\bbeta_p\right)  \right]   			 \\
	& \qquad \qquad \qquad\qquad    \times   \frac{1}{n \h} \sumj \bigg\{  K(\Xhj)  (\Xhj)^{l_j}  \left(Y_j - \br_p(X_j - \x)'\bbeta_p\right)  \bigg\} 
\end{align*}
Returning to Equations \eqref{suppeqn:V5 1}, $\breve{\bV}_5$ is defined via
\begin{align*}
	\v!^2 \be_\v'\Gp^{-1}  \left[ \breve{\bV}_5 \right] \Gp^{-1} \be_\v   &  :=   \sum_{l_i=0}^p \sum_{l_j=0}^p  \v!^2 \be_\v'\Gp^{-1}   \left[\Gp^{-1}\right]_{l_i, l_j} \E\left[ (K^2 \br_p \br_p') (\Xhi) (\Xhi)^{l_i} \left(\mu(X_i) -  \br_p(X_i - \x)'\bbeta_p\right)  \right]   			 \\
	& \qquad \qquad \qquad\qquad    \times   \frac{1}{n \h} \sumj \bigg\{  K(\Xhj)  (\Xhj)^{l_j}  \left(Y_j - \br_p(X_j - \x)'\bbeta_p\right)  \bigg\}     \Gp^{-1} \be_\v.
\end{align*}
This completes the proof of Equation \eqref{suppeqn:V5}.

Lastly, we consider the $\bV_6 - \breve{\bV}_6$ term of \eqref{suppeqn:delta 1}. Proving this is ignorable will complete {\bf Step (I)}. Begin by expanding the inner product, just as was done for $\bV_5$:
\begin{align*}
	\bV_6 & = \frac{1}{n \h} \sumi (K^2 \br_p \br_p')(\Xhi)  \big\{ \br_p(\Xhi)'\Gp^{-1} \Op \left[\bY - \bR \bbeta_p\right] /n \big\}^2   			\\
	& = \frac{1}{n \h} \sumi (K^2 \br_p \br_p')(\Xhi)  \left\{ \frac{1}{n\h} \sumj \br_p(\Xhi)'\Gp^{-1}\br_p(\Xhj) K(\Xhj)  \left(Y_j - \br_p(X_j - \x)'\bbeta_p\right) \right\}^2   			\\
	& = \frac{1}{n \h} \sumi (K^2 \br_p \br_p')(\Xhi)  \left\{ \frac{1}{n\h} \sumj    \sum_{l_i=0}^p \sum_{l_j=0}^p (\Xhi)^{l_i}   \left[\Gp^{-1}\right]_{l_i, l_j} (\Xhj)^{l_j}
	 K(\Xhj)  \left(Y_j - \br_p(X_j - \x)'\bbeta_p\right) \right\}^2   			\\
	& = \sum_{l_{i_1}=0}^p \sum_{l_{i_2}=0}^p \sum_{l_{j_1}=0}^p \sum_{l_{j_2}=0}^p   \left[\Gp^{-1}\right]_{l_{i_1}, l_{j_1}}  \left[\Gp^{-1}\right]_{l_{i_2}, l_{j_2}}  \frac{1}{n \h} \sumi (K^2 \br_p \br_p')(\Xhi)   (\Xhi)^{l_{i_1} + l_{i_2}}   
		  			 \\
	& \qquad  \times  
	\frac{1}{(n \h)^2}   \sumj \sumk    K(\Xhj) (\Xhj)^{l_{j_1}}  \left(Y_j - \br_p(X_j - \x)'\bbeta_p\right)  K(\Xhk)   (\Xhk)^{l_{j_2}}  \left(Y_k - \br_p(X_k - \x)'\bbeta_p\right) .
\end{align*}
Define
\begin{align*}
	\breve{\bV}_6 & = \sum_{l_{i_1}=0}^p \sum_{l_{i_2}=0}^p \sum_{l_{j_1}=0}^p \sum_{l_{j_2}=0}^p   \left[\Gp^{-1}\right]_{l_{i_1}, l_{j_1}}  \left[\Gp^{-1}\right]_{l_{i_2}, l_{j_2}}  \E\left[ \h^{-1} (K^2 \br_p \br_p')(\Xhi)   (\Xhi)^{l_{i_1} + l_{i_2}}   \right]		  			 \\
	& \qquad  \times  
	\frac{1}{(n \h)^2}   \sumj \sumk    K(\Xhj) (\Xhj)^{l_{j_1}}  \left(Y_j - \br_p(X_j - \x)'\bbeta_p\right)  K(\Xhk)   (\Xhk)^{l_{j_2}}  \left(Y_k - \br_p(X_k - \x)'\bbeta_p\right) .
\end{align*}

Completely analogous steps to those above will show that
\begin{equation}
	\label{suppeqn:V6}
	r_{\tp,\f}^{-1} \P \left[\left|  \left(  \v!^2 \be_\v'\Gp^{-1}  \left( \bV_6 - \breve{\bV}_6  \right) \Gp^{-1} \be_\v \right) \tO \v! \be_\v' \Gp^{-1} \Op \left( \bY  - \bR \bbeta_p \right)/ n \right|  > r_n  \right]  = o(1).
\end{equation}
The starting point will again be splitting $(Y_i - \br_p(X_i - \x)'\bbeta_p) = \e_i +  (\mu(X_i) - \br_p(X_ - \x)'\bbeta_p )$ everywhere, which now occurs in three places, giving eight total terms. The most difficult of these will be when all three are bias terms. The rest of the terms will have at least one ``variance'' type term, and the faster rates of Lemma \ref{supplem:var} can be brought to bear. Thus, we shall only demonstrate the former. For a fixed set of the indexes $l_{i_1}, l_{i_2}, l_{j_1}, l_{j_2}$, let 
\begin{align*}
	V_{6,i} & = (K^2 \br_p \br_p')(\Xhi)   (\Xhi)^{l_{i_1} + l_{i_2}}  ,		  \\
	V_{6,j} & = K(\Xhj) (\Xhj)^{l_{j_1}}  \left(\mu(X_J) - \br_p(X_j - \x)'\bbeta_p\right)  , \quad\text{ and }		  \\
	V_{6,k} & = K(\Xhk)   (\Xhk)^{l_{j_2}}  \left(\mu(X_k) - \br_p(X_k - \x)'\bbeta_p\right).
\end{align*}
The term in question, with three ``bias'' type terms'', is:
\begin{align*}
	& \v!^2 \be_\v'\Gp^{-1}  \left( \bV_6 - \breve{\bV}_6  \right) \Gp^{-1} \be_\v \tO \v! \be_\v' \Gp^{-1} \Op \left( \bM  - \bR \bbeta_p \right)/ n   		\\
	& = \v!^2 \be_\v'\Gp^{-1}  \Bigg(   \left[\Gp^{-1}\right]_{l_{i_1}, l_{j_1}}  \left[\Gp^{-1}\right]_{l_{i_2}, l_{j_2}}  \frac{1}{n \h} \sumi \left( V_{6,i} - \E[V_{6,i}]\right)  \frac{1}{(n \h)^2}   \sumj \sumk     \E[V_{6,j}]   \E[V_{6,k}]     \Bigg)   			\\
  	& \qquad \qquad \qquad \times\;  \Gp^{-1} \be_\v \tO \v! \be_\v' \Gp^{-1} \Op \left( \bM  - \bR \bbeta_p \right)/ n    			\\
	& \quad + \v!^2 \be_\v'\Gp^{-1}  \Bigg(   \left[\Gp^{-1}\right]_{l_{i_1}, l_{j_1}}  \left[\Gp^{-1}\right]_{l_{i_2}, l_{j_2}}  \frac{1}{n \h} \sumi \left( V_{6,i} - \E[V_{6,i}]\right)  \frac{1}{(n \h)^2}   \sumj \sumk     \E[V_{6,j}]  \left( V_{6,k} -  \E[V_{6,k}] \right)    \Bigg)   			\\
  	& \qquad \qquad \qquad \times\;  \Gp^{-1} \be_\v \tO \v! \be_\v' \Gp^{-1} \Op \left( \bM  - \bR \bbeta_p \right)/ n    			\\
	& \quad + \v!^2 \be_\v'\Gp^{-1}  \Bigg(   \left[\Gp^{-1}\right]_{l_{i_1}, l_{j_1}}  \left[\Gp^{-1}\right]_{l_{i_2}, l_{j_2}}  \frac{1}{n \h} \sumi \left( V_{6,i} - \E[V_{6,i}]\right)  \frac{1}{(n \h)^2}   \sumj \sumk      \left( V_{6,j} - \E[V_{6,j}]\right)  \E[V_{6,k}]    \Bigg)   			\\
  	& \qquad \qquad \qquad \times\;  \Gp^{-1} \be_\v \tO \v! \be_\v' \Gp^{-1} \Op \left( \bM  - \bR \bbeta_p \right)/ n   			\\
	& \quad + \v!^2 \be_\v'\Gp^{-1}  \Bigg(   \left[\Gp^{-1}\right]_{l_{i_1}, l_{j_1}}  \left[\Gp^{-1}\right]_{l_{i_2}, l_{j_2}}  \frac{1}{n \h} \sumi \left( V_{6,i} - \E[V_{6,i}]\right)  \frac{1}{(n \h)^2}   \sumj \sumk      \left( V_{6,j} - \E[V_{6,j}]\right)  \left( V_{6,k} -  \E[V_{6,k}] \right)   \Bigg)   			\\
  	& \qquad \qquad \qquad \times\;  \Gp^{-1} \be_\v \tO \v! \be_\v' \Gp^{-1} \Op \left( \bM  - \bR \bbeta_p \right)/ n
\end{align*}
The first term is bounded as
\begin{align*}
	& \leq  r_{\tp,\f}^{-1} 5 \P \left[ \left| \Gp^{-1} \right| > C_\Gamma \right]    		\\
	& \quad +  r_{\tp,\f}^{-1} \P \left[ \left| \Op \left( \bM - \bR \bbeta_p  \right)/ n \right|  >  C_1 \log(\tO)^\gamma  \right]   		\\
	& \quad +  r_{\tp,\f}^{-1} \P \left[ \left| \frac{1}{n\h}\sumi  \left( V_{6,i} - \E[V_{6,i}]\right) \right|  >  r_n \frac{1}{C_1 C_\Gamma^5 \v!^3 |\be_\v|^3  \E[\h^{-1}V_{6,j}] \E[\h^{-1}V_{6,k}] \tO \log(\tO)^\gamma }  \right]     		\\
	& = o(1),
\end{align*}
by Lemmas \ref{supplem:gamma}, \ref{supplem:bias1}, and \ref{supplem:lambda}. In applying the last, we have used that $\E[\h^{-1}V_{6,j}] \asymp \E[\h^{-1}V_{6,k}] \asymp \tO^{-1} \btp$ (see Section \ref{supp:bias lp} or the computation for $\E[\bV_3]$ above) and $r_n = \tO^{-1} \btp \log(\tO)^{-\gamma}$ for $\gamma > 0$, leaving
\[  r_n \frac{1}{\E[\h^{-1}V_{6,j}]  \E[\h^{-1}V_{6,k}] \tO \log(\tO)^\gamma }  \asymp  \tO^{-1} \log(\tO)^{1/2} \left[  \frac{1}{\tO^{-1} \btp \log(\tO)^{1/2+2\gamma} }  \right].   \]
The factor in square brackets diverges by assumption. The second term is
\begin{align*}
	& \v!^2 \be_\v'\Gp^{-1}  \Bigg(   \left[\Gp^{-1}\right]_{l_{i_1}, l_{j_1}}  \left[\Gp^{-1}\right]_{l_{i_2}, l_{j_2}}  \frac{1}{n \h} \sumi \left( V_{6,i} - \E[V_{6,i}]\right)  \frac{1}{(n \h)^2}   \sumj \sumk     \E[V_{6,j}]  \left( V_{6,k} -  \E[V_{6,k}] \right)    \Bigg)   			\\
  	& \qquad \qquad \qquad \times\;  \Gp^{-1} \be_\v \tO \v! \be_\v' \Gp^{-1} \Op \left( \bM  - \bR \bbeta_p \right)/ n    			\\
	& \leq  r_{\tp,\f}^{-1} 5 \P \left[ \left| \Gp^{-1} \right| > C_\Gamma \right]    		\\
	& \quad +  r_{\tp,\f}^{-1} \P \left[ \left| \Op \left( \bM - \bR \bbeta_p  \right)/ n \right|  >  C_1 \log(\tO)^\gamma  \right]   		\\
	& \quad +  r_{\tp,\f}^{-1} \P \left[ \left| \frac{1}{n\h}\sumk \left( V_{6,k} - \E[V_{6,k}]\right) \right|  >  C_2  \frac{\btp}{\tO} \log(\tO)^\gamma \right]     		\\
	& \quad +  r_{\tp,\f}^{-1} \P \left[ \left| \frac{1}{n\h}\sumi  \left( V_{6,i} - \E[V_{6,i}]\right) \right|  >  r_n \frac{}{C_1 C_2 C_\Gamma^5 \v!^3 |\be_\v|^3  \E[\h^{-1}V_{6,j}] \btp \tO \tO^{-1} \log(\tO)^{2\gamma} }  \right]     		\\
	& = o(1),
\end{align*}
by nearly identical reasoning, additionally using Lemma \ref{supplem:bias2}. The third term is the identical to this one, and the fourth term is similar, requiring Lemma \ref{supplem:bias2} twice.

Referring back to the discussion following Equation \eqref{suppeqn:V6}, this completes the proof of that result for the case where the bias portion of $(Y_i - \br_p(X_i - \x)'\bbeta_p) = \e_i +  (\mu(X_i) - \br_p(X_ - \x)'\bbeta_p )$ is retained everywhere, which is the most difficult. All other pieces will follow by similar logic, applying Lemma \ref{supplem:var} when needed. Because this Lemma delivers a faster rate, these other terms will not require strong assumptions. Altogether, this establishes the convergence required by Equation \eqref{suppeqn:V6}.

Combining Equations \eqref{suppeqn:V3}, \eqref{suppeqn:V5}, and \eqref{suppeqn:V6} establishes that  $\left| \shatp^2 - \sbp^2 \right| = o_\P(1)$ and \eqref{suppeqn:delta 1} holds, proving \eqref{suppeqn:delta} and thus completing {\bf Step (I)}.

%%%%%%%%%%%%%%%%%%%%%%%%%%%%%%%%%%%%%%%%%%%%%%%%%%%%%%%%%%%%%%%%%%%%%%
\subsubsection{Step (II)}
	\label{supp:proof lp us step 2}

We now prove that 
\[ \bSn := \sumi \V[\bZ_i]^{-1/2}(\bZ_i - \E[\bZ_i])/\sqrt{n} \]
obeys an Edgeworth expansion by verifying the conditions of Theorem 3.4 of \cite{Skovgaard1981_SJS}. Repeating the definition of $\bZ_i$ from Equation \eqref{suppeqn:bZ}:
\begin{align*}
	\begin{split}
		\bZ_i  = \Bigg(
		& \Big\{  (K \br_p)(\Xhi) (Y_i - \br_p(X_i - \x)'\bbeta_p)  \Big\}'  ,    		\\ 		  %from  numerator
		& \vech\Big\{  (K \br_p \br_p')(\Xhi)  \Big\}' ,    		\\ 		  %from  numerator, Gamma matrix
		& \vech\Big\{  (K^2 \br_p \br_p')(\Xhi) \e_i^2   \Big\}'  ,    		\\ 		  %from  \bV_1 
		& \vech\Big\{  (K^2 \br_p \br_p' )(\Xhi) (\Xhi)^0 \e_i   \Big\}', \vech\Big\{  (K^2 \br_p \br_p' )(\Xhi) (\Xhi)^1 \e_i   \Big\}',   		\\
		& \qquad \vech\Big\{  (K^2 \br_p \br_p' )(\Xhi) (\Xhi)^2 \e_i   \Big\}' ,\ldots, \vech\Big\{  (K^2 \br_p \br_p' )(\Xhi) (\Xhi)^p \e_i   \Big\}' ,	\\   %from   \bV_2
		& \vech\Big\{  (K^2 \br_p \br_p')(\Xhi) \big\{  \e_i \big[\mu(X_i) -  \br_p(X_i - \x)'\bbeta_p\big] \big\}   \Big\}'   		  %from  \bV_4
		\Bigg)'.
	\end{split}
\end{align*}

First, define 
\[\bB := \h \V[\bZ_i], \]
which may be readily computed, but the constants are not needed here. All that matters at present is that, under our assumptions, $\bB$ is bounded and bounded away from zero. Write
\[ \bSn = \sumi \bB^{-1/2}(\bZ_i - \E[\bZ_i]) / \tO. \]
By construction, the mean of $\bSn$ is zero and the variance is the identity matrix. That is, for any $\bt \in \R^{\dim(\bZ_i)}$, $\E[\bt'\bSn] = 0$ and $\V[\bt'\bSn]=|\bt|^2$. 

To verify conditions (I) and (II) of \cite[Theorem 3.4]{Skovgaard1981_SJS} we first compute the third and fourth moments of $\bZ_i$, and use these to compute the required directional cumulants of $\bSn$. For a nonnegative integer $l$ and $k \in \{3,4\}$, by a change of variables we find that
\[ 
\E\left[ \left(K(\Xhi)(\Xhi)^l \right)^k \right] = \h \int K(u)^k u^{lk} f(\x - u\h)du = O( \h ),  		 \]
under the conditions on the kernel function and the marginal density of $X_i$, $f(\cdot)$. In exactly the same way, for the remaining pieces of $\bZ_i$, we find that:
\begin{gather*}
	\E\left[ \left(K(\Xhi)(\Xhi)^l (Y_i - \br_p(X_i - \x)'\bbeta_p) \right)^k \right]  = O( \h ),  			\\
	\E\left[ \left(K(\Xhi)(\Xhi)^l \e_i^2 \right)^k \right]  = O( \h ),  			\qquad \text{ and } \qquad
	\E\left[ \left(K(\Xhi)(\Xhi)^l \e_i \right)^k \right]  = O( \h ),  			\\
	\E\left[ \left(K(\Xhi)(\Xhi)^l \e_i (\mu(X_i) - \br_p(X_i - \x)'\bbeta_p) \right)^k \right]  = O( \h ),
\end{gather*}
using the assumed moment conditions on $\e_i$. Therefore, for a $\bt \in \R^{\dim(\bZ_i)}$ with $|\bt|=1$
\[\E \left[ \left(\bt'\bB^{-1/2}(\bZ_i - \E[\bZ_i]) \right)^3 \right] = O(\h).\]
and
\[\E \left[ \left(\bt'\bB^{-1/2}(\bZ_i - \E[\bZ_i]) \right)^4 \right] = O(\h).\]
Using these, and the fact that the $\bZ_i$ are i.i.d.\ and the summands of $\bSn$ are mean zero, we have, again for a $\bt \in \R^{\dim(\bZ_i)}$ with $|\bt|=1$,
\[  \E \left[ (\bt'\bSn)^3 \right]  =  \tO^{-3} \sumi \E \left[ \left(\bt'\bB^{-1/2}(\bZ_i - \E[\bZ_i]) \right)^3 \right]  =  O(\tO^{-3} n \h) =  O(\tO^{-1}). \]
The third moment agrees with the third cumulant of $\bSn$. The fourth cumulant is
\[  \E \left[ (\bt'\bSn)^4 \right]  - 3 \E \left[ (\bt'\bSn)^2 \right]^2 .  \]
The first term of these two is
\begin{align*}
	\E \left[ (\bt'\bSn)^4 \right] & =  \tO^{-4} {4 \choose 2}  \sumi\sumjNoti \E \left[ \left(\bt'\bB^{-1/2}(\bZ_i - \E[\bZ_i]) \right)^2 \right] \E \left[ \left(\bt'\bB^{-1/2}(\bZ_j - \E[\bZ_j]) \right)^2 \right]     		\\ 
	& \qquad   +   \tO^{-4} \sumi \E \left[ \left(\bt'\bB^{-1/2}(\bZ_i - \E[\bZ_i]) \right)^4 \right]   		\\ 
	&  =   3 \h^{-2} [1+o(1/n)] \E \left[ \left(\bt'\bB^{-1/2}(\bZ_i - \E[\bZ_i]) \right)^2 \right]^2   +  O( \tO^{-2}).
\end{align*}
By direct computation, the second piece of the fourth cumulant is
\begin{align*}
	\E \left[ (\bt'\bSn)^2 \right]^2 & =  \left( \tO^{-2} n \E \left[ \left(\bt'\bB^{-1/2}(\bZ_i - \E[\bZ_i]) \right)^2 \right] \right)^2.
\end{align*}
This cancels with the corresponding term of $\E \left[ (\bt'\bSn)^4 \right]$, and thus the fourth cumulant is $O(\tO^{-2})$. Thus, we find that, in the notation of \cite{Skovgaard1981_SJS}, $\rho_{s,n}(\bt) \asymp \tO^{-1}$, and so condition (II) of \cite{Skovgaard1981_SJS} is satisfied by setting $a_n(\bt) = C \tO$ for an appropriate constant $C$. Recall that $r_n = o(r_{\tp,\f})$, with $r_{\tp,\f} = \max\{\tO^{-2}, \btp^2, \tO^{-1} \btp \}$, i.e.\ the slowest vanishing of the rates. Thus our $r_n$ is $\varepsilon_n$ in the notation of \cite{Skovgaard1981_SJS}, and condition (I) therein is satisfied because $ a_n(\bt)^{-(s-1)}  = \tO^{-3} = o(\tO^{-2}) = O(r_n)$.

Next, we verify condition ($\text{III}''_\alpha$) of \cite[Theorem 3.4 and Remark 3.5]{Skovgaard1981_SJS}. Let $\xi_S(\bt)$ be the characteristic function of $\bSn$ and $\xi_Z(\bt)$ that of $\bZ_i$, where $\bt \in \R^{\dim(\bZ_i)}$. By the i.i.d.\ assumption, 
\begin{align*}
	\xi_S(\bt) = \E[\exp\{\text{i}\bt'\bSn\}] & = \prod_{i=1}^{n} \E[\exp\{\text{i}\bt'\bB^{-1/2}(\bZ_i - \E[\bZ_i]) / \tO\}]  		\\
	 & = \prod_{i=1}^{n} \E\left[\exp\left\{ \text{i} \left( \bt'\bB^{-1/2} / \tO \right) \bZ_i \right\} \right] \exp\{-\text{i}\bt'\bB^{-1/2}\E[\bZ_i]) / \tO\}.
\end{align*}
The second factor is bounded by one, leaving 
\begin{align*}
	\xi_S(\bt) \leq \left[\xi_Z\left( \bt'\bB^{-1/2} / \tO \right) \right]^n.
\end{align*}
Recall that, in the notation of \cite{Skovgaard1981_SJS}, $a_n(\bt) = C \tO^{-1}$, and so condition ($\text{III}''_\alpha$) of Theorem 3.4 (and Remark 3.5) is satisfied because
\begin{align*}
	\sup_{|\bt| > \delta C \tO^{-1}} |\xi_S(\bt)| & \leq \sup_{|\bt| > \delta C \tO^{-1}} \left|\xi_Z\left( \bt'\bB^{-1/2} / \tO \right) \right|^n  		\\
	& \leq \sup_{|\bt_1| > C_1 } \left| \xi_Z (\bt_1) \right|^n  		\\
	& = (1 - C_2 \h)^n = o(r_n^{-C_3}),
\end{align*}
for any $C_3 > 0$ by the assumption that $\log(n\h) / (n\h) = o(1)$. Thus condition ($\text{III}''_\alpha$) holds. The penultimate equality holds by Lemma \ref{supplem:Cramer lp}, which verifies that $\bZ_i$ obeys the $n$-varying version of Cram\'er's condition: for $\h$ sufficiently small, for all $C_1 > 0$ there is a $C_2 >0$ such that
\[ \sup_{|\bt| > C_1} |\xi_Z(\bt)| < (1 - C_2 \h). \]

Finally, we check condition (IV) of \cite[Theorem 3.4]{Skovgaard1981_SJS}. We aim to prove that
\begin{equation}
	\label{suppeqn:skov 4 1}
	\sup_{0<s<1} \frac{\left| \dfrac{d^5}{ds^5} \log \xi_S\left( s \dfrac{\delta a_n(\bt) \bt}{|\bt|}\right) \right|}{5! \left| \dfrac{\delta a_n(\bt) \bt}{|\bt|} \right|^5 } = O(a_n(\bt)^{-3}),
\end{equation}
for some $\delta > 0$, with $a_n(\bt) = C \tO$ defined by conditions (I) and (II). For the supremum, as $s$ ranges in $(0,1)$, the quantity $w = s \delta a_n(\bt)$ ranges in $(0,\delta a_n(\bt))$. Further, by the chain rule
\[\dfrac{d^5}{ds^5} \log \xi_S\left( s \dfrac{\delta a_n(\bt) \bt}{|\bt|}\right) =  \dfrac{d^5}{dw^5} \log \xi_S\left( \dfrac{w \bt}{|\bt|}\right) \left( \delta a_n(\bt) \right)^5.  \]
To see why, write $\log \xi_S\left( s \delta a_n(\bt) \bt / |\bt| \right)$ as $g(w(s))$, where $w(s) = s \delta a_n(\bt)$ and $g(w) = \log \xi_S\left( w \bt / |\bt| \right)$ and then the chain rule gives
\[ \frac{d^5}{ds^5} \log \xi_S\left( s \dfrac{\delta a_n(\bt) \bt}{|\bt|}\right) = \frac{d^5 g}{d w^5} \left(\frac{d w}{d s}\right)^5   \]
because all the other terms in the chain rule expansion involve higher derivatives of the linear function $w(s) = s \delta a_n(\bt)$ and hence are zero. Therefore
\[\sup_{0<s<1} \frac{\left| \dfrac{d^5}{ds^5} \log \xi_S\left( s \dfrac{\delta a_n(\bt) \bt}{|\bt|}\right) \right|}{5! \left| \dfrac{\delta a_n(\bt) \bt}{|\bt|} \right|^5 }  =  \sup_{0<w<\delta a_n(\bt)} \frac{\left| \dfrac{d^5}{dw^5} \log \xi_S\left(  \dfrac{w \bt}{|\bt|}\right)\left( \delta a_n(\bt) \right)^5 \right|}{5! \left| \dfrac{\delta a_n(\bt) \bt}{|\bt|} \right|^5 } =  \sup_{0<w<\delta a_n(\bt)} \frac{\left| \dfrac{d^5}{dw^5} \log \xi_S\left(  \dfrac{w \bt}{|\bt|}\right) \right|}{5!  }, \]
where we have canceled terms and used the fact that $| (\bt / |\bt| )| =1$. 

With $a_n(\bt) = C \tO$, proving Equation \eqref{suppeqn:skov 4 1} is equivalent to showing that 
\[\sup_{0<w<\delta a_n(\bt)} \left| \dfrac{d^5}{dw^5} \log \xi_S\left(  \dfrac{w \bt}{|\bt|}\right) \right| = O\left( \tO^{-3} \right).\]
Let $\xi_{\bar{Z}}(\bt)$ be the characteristic function of $(\bZ_i - \E[\bZ_i])$. (This is distinct from $\xi_Z(\bt)$, which is the characteristic function of $\bZ_i$ itself. The two are related via $\xi_{\bar{Z}}(\bt) = \xi_Z(\bt) \exp\{-\text{i}\bt'\E[\bZ_i]\}$.) By the i.i.d.\ assumption
\[\log \xi_S\left(  \dfrac{w \bt}{|\bt|}\right) =  n \log \xi_{\bar{Z}}\left(  \dfrac{w \bB^{-1/2} \bt}{|\bt| \tO} \right).\]

As $w$ varies in $(0 , \delta a_n(\bt))$, the quantity $u = w \bB^{-1/2} \tO^{-1}$ varies in $(0,C \delta \bB^{-1/2} )$, by the definition of $a_n(\bt)$.  Using the same chain rule logic as above,
\[ \dfrac{d^5}{dw^5} \log \xi_{\bar{Z}}\left(  \dfrac{w \bB^{-1/2} \bt}{|\bt| \tO}\right)  =  \left(  \dfrac{d^5}{du^5} \log \xi_{\bar{Z}}\left(  \dfrac{ u  \bt}{|\bt|}\right)  \right) \left(\frac{\bB^{-1/2}}{\tO}\right)^{5}.  \]
Therefore 
\begin{align*}
	\sup_{0<w<\delta a_n(\bt)} \left| \dfrac{d^5}{dw^5} \log \xi_S\left(  \dfrac{w \bt}{|\bt|}\right) \right|  & =  \sup_{0<w<\delta a_n(\bt)} \left|  \dfrac{d^5}{dw^5} n \log \xi_{\bar{Z}}\left(  \dfrac{w \bB^{-1/2} \bt}{|\bt| \tO} \right) \right|     		\\
	& = n  \left(\frac{\bB^{-1/2}}{\tO}\right)^{5} \sup_{0<u<C \delta \bB^{-1/2}} \left|  \dfrac{d^5}{du^5} \log \xi_{\bar{Z}}\left(  \dfrac{ u  \bt}{|\bt|}\right) \right|.
\end{align*}
We aim to show that the final quantity is $O\left( \tO^{-3} \right)$. As $\tO = \sqrt{n\h}$ and $\bB$ is bounded above and below, this will hold if
\begin{equation}
	\label{suppeqn:skov 4 2}
	\sup_{0<u<C \delta \bB^{-1/2}} \left|  \dfrac{d^5}{du^5} \log \xi_{\bar{Z}}\left(  \dfrac{ u \bt}{|\bt|}\right) \right| = O(\h).
\end{equation}
for some $\delta > 0$.

By Corollary 8.2 of \cite{Bhattacharya-Rao1976_book} for the first inequality and direct calculation for the second, 
\begin{equation}
	\label{suppeqn:skov 4 3}
	\left| \log \xi_{\bar{Z}}\left(  \frac{ u \bt}{|\bt|}\right) - 1 \right| \leq \frac{1}{2} \left|  \frac{ u \bt}{|\bt|} \right| \E\left[ \left|\bZ_i - \E[\bZ_i]\right|^2  \right] \leq C |u| \h.
\end{equation}
Therefore, for $\h$ small enough there is a $\delta > 0$ such that $C |u| \h < 1/2$ for all $u$ such that $0<u<C \delta \bB^{-1/2}$. This allows us to apply Lemma 9.4 of \cite{Bhattacharya-Rao1976_book}, yielding the bound
	\[	\sup_{0<u<C \delta \bB^{-1/2}} \left|  \dfrac{d^5}{du^5} \log \xi_{\bar{Z}}\left(  \dfrac{ u \bt}{|\bt|}\right) \right| \leq C \E\left[ \left|\bZ_i - \E[\bZ_i]\right|^5  \right].   \]
As the fifth moment of $\bZ_i$ is $O(\h)$, this establishes Equation \eqref{suppeqn:skov 4 2} and therefore Equation \eqref{suppeqn:skov 4 1}, verifying condition (IV) of \cite[Theorem 3.4]{Skovgaard1981_SJS}. All of the conditions of this Theorem are now verified, thus completing {\bf Step (II)}.

\begin{remark}
	\label{supprem:Cramer}
	For building intuition it is useful to compare the bound bound in Equation \eqref{suppeqn:skov 4 3} and the $n$-varying version of Cram\'er's condition established in Lemma \ref{supplem:Cramer lp}. Both reflect the fact that as $\h = o(1)$, $K(\Xhi) = o(1)$, and therefore in the limit $\bZ_i \equiv 0$ is a degenerate random variable. In this case of \eqref{suppeqn:skov 4 3}, the bound shows that as $\h = o(1)$, the characteristic function $\log \xi_{\bar{Z}}\left(  u \bt / |\bt| \right) \to 1$, uniformly. Lemma \ref{supplem:Cramer lp} shows the same thing, as it is proven therein that
	\[ \sup_{|\bt| > C_1} |\xi_Z(\bt)| < (1 - C_2 \h). \]
	Notice that in the limit as $\h = o(1)$, the conventional Cram\'er's condition fails. Equation \eqref{suppeqn:skov 4 3} and Lemma \ref{supplem:Cramer lp} are in qualitative agreement in this sense.
\end{remark}

%%%%%%%%%%%%%%%%%%%%%%%%%%%%%%%%%%%%%%%%%%%%%%%%%%%%%%%%%%%%%%%%%%%%%%
\subsubsection{Step (III)}
	\label{supp:proof lp us step 3}

We now prove that the expansion for $\tp$ holds and that it holds uniformly over $\f \in \F_\S$. First, by Equation \eqref{suppeqn:step 1} and Lemma \ref{supplem:delta1}, $\tp$ will obey the desired expansion (computed formally as in Section \ref{supp:terms lp}) if $\breve{T}$ obeys an Edgeworth expansion. Now, $\breve{T}$ is given by
\[\breve{T}\left(\tO^{-1} \sumi \bZ_i\right) = \breve{T}\left( \V[\bZ_i]^{1/2} \bSn + n \E[\bZ_i]/ \tO \right), \]
which is a smooth function of $\bSn := \sumi \V[\bZ_i]^{-1/2}(\bZ_i - \E[\bZ_i])/\tO$. {\bf Step (II)} proved that $\bSn$ obeys an Edgeworth expansion, and therefore by \cite{Skovgaard1986_ISR} we have that $\breve{T}$ does as well. Equation \eqref{suppeqn:step 1} and Lemma \ref{supplem:delta1} deliver the result pointwise for $\tp$. 

To prove that the expansion holds uniformly, first notice that all our results hold pointwise along a sequence $\f_n \in \F_\S$. That is, the results of \cite{Skovgaard1981_SJS} and \cite{Skovgaard1986_ISR} hold along this sequence. We thus proceed by arguing as in \cite{Romano2004_SJS}. Recall that $r_{\tp,\f} = \max\{\tO^{-2}, \btp^2, \tO^{-1} \btp \}$, i.e.\ the slowest vanishing of the rates. Suppose the result failed. Then we can extract a subsequence $\{\f_m \in \F_\S\}$ such that 
\[ r_{\tp,\f} \left| \P_{\f_m} \left[ \tp < z \right]  - \Phi(z) - E_{\tp,\f_m}(z) \right|  \not= o(1).	  \]
But this contradicts the result above, because $\tp$ obeys the expansion given on $\{\f_m \in \F_\S\}$.

%%%%%%%%%%%%%%%%%%%%%%%%%%%%%%%%%%%%%%%%%%%%%%%%%%%%%%%%%%%%%%%%%%%%%%
%%%%%%%%%%%%%%%%%%%%%%%%%%%%%%%%%%%%%%%%%%%%%%%%%%%%%%%%%%%%%%%%%%%%%%
\subsection{Proof of Theorem \ref{suppthm:EE lp} (Theorem 1 in the paper) with Bias Correction}
	\label{supp:proof lp rbc}

Proving Theorem \ref{suppthm:EE lp} for $\trbc$ follows the exact same steps as for $\tp$. The reason being that both are based such similar estimation procedures. To illustrate this point, recall that when $\rho = 1$, $\trbc$ is the same as $\tp$ but based on a higher degree polynomial. In this special case, there is nothing left to prove: simply apply Theorem \ref{suppthm:EE lp} with $p$ replaced with $p+1$. Or, alternatively, re-walk the entire proof replacing $p$ with $p+1$ everywhere. 

The more general case, that is, with generic $\rho$, is not conceptually more difficult, just more cumbersome. There are two chief changes. First, the bias rate changes due to the bias correction, but this is automatically accounted for by the terms of the expansion and the conditions of the theorem. For example, note that the rate $r_{\irbc}$ automatically includes the new bias rate, as it is defined in general in terms of $\bti$ Second, there are additional kernel-weighted averages that enter into $\trbc$ and these will enter into the construction of $\bZ_i$ and the bounding of remainder terms.

Recall the definitions of the point estimators, standard errors, and $t$-statistics from Section \ref{supp:setup}, specifically Equations \eqref{suppeqn:that lp}, \eqref{suppeqn:se lp}, and \eqref{suppeqn:t stat}: 
\[\mhat_p^{(\v)} = \frac{1}{n \h^\v} \v! \be_\v'\Gp^{-1} \Op  \bY,  		\qquad 		  \shatp^2  =  \v!^2 \be_\v'\Gp^{-1} (\h \Op \Shatp \Op' /n) \Gp^{-1} \be_\v,  		\qquad 		  \tp = \frac{\sqrt{n\h^{1 + 2\v}}( \mhat_p^{(\v)} - \tf)}{\shatp}\]
\[\thatrbc = \frac{1}{n \h^\v} \v! \be_\v'\Gp^{-1} \Orbc  \bY,  		\quad 		  \shatrbc^2  =  \v!^2 \be_\v'\Gp^{-1} (\h \O_\RBC \Shatrbc \O_\RBC' /n) \Gp^{-1} \be_\v,  		\quad 		  \trbc = \frac{\sqrt{n\h^{1 + 2\v}}( \thatrbc - \tf)}{\shatrbc}.\]
Comparing these, we see that the only differences in the change from $\Shatp$ and $\Op$ to $\Shatrbc$ and $\Orbc$, where (to repeat):
\begin{itemize}
	\item $\Shatrbc = \diag(\hat{v}(X_i): i = 1,\ldots, n)$, with $\hat{v}(X_i) = ( Y_i - \br_{p+1}(X_i - \x)'\bhat_{p+1} )^2$,
	\item $\Orbc = \Op - \rho^{p+1}  \Lp_1 \be_{p+1}' \Gq^{-1} \Oq $,
	\item $\rho = \h / \b$,
	\item $\Lp_k = \Op \left[ X_{\h,1}^{p+k}, \ldots, X_{\h,n}^{p+k} \right]'/n$,
	\item $\Xbi = (X_i - \x)/\b$,
	\item $\Gq = \frac{1}{n\b} \sumi (K \br_{p+1} \br_{p+1}')(\Xbi)$, and
	\item $\Oq = [ (K \br_{p+1})(X_{\b,1}),  (K \br_{p+1})(X_{\b,2}), \ldots,  (K \br_{p+1})(X_{\b,n})]$.
\end{itemize}
Notice that these are the same as their counterparts for $\tp$, but with $\b = \h \rho^{-1}$ in place of $\h$ and $p+1$ in place of $p$. With these comparisons in mind, we briefly discuss the three steps of Section \ref{supp:proof lp us}, highlighting key pieces.

For {\bf Step (I)}, first observe that the ``numerator'', or $\thatrbc$, portion of the $t$-statistic is once again already a smooth function of well-behaved random variables, albeit different ones that for $\tp$. Terms will be added to $\bZ_i$ to reflect this. In particular, $\Lp_1$, $\Gq$, and $\Oq$ are present. Importantly, Lemma \ref{supplem:gamma} applies to $\Gq$ with $\b = \h \rho^{-1}$ in place of $\h$ and $p+1$ in place of $p$.

Turning to the Studentization, Equation \eqref{suppeqn:shat terms} expands the quantity $(\h \Op \Shatp \Op' /n)$ and this needs to be adapted to account instead for $(\h \O_\RBC \Shatrbc \O_\RBC' /n)$, which requires two changes. The fundamental issue remains the estimated residuals and thus the terms represented by $\bV_1$ -- $\bV_6$ will remain conceptually the same. The first change, which is automatically accounted for by the rate assumptions of the Theorem and the terms of the expansion, are that the bias is now lower because the residuals are estimated with a $p+1$ degree fit. This matches the numerator bias, and thus the calculations are as above. Second, whereas the summands of each term of $\bV_1$ -- $\bV_6$ include $(K^2 \br_p \br_p')(\Xhi)$ stemming from the pre- and post-multiplying by $\Op$, now we multiply by $\Orbc$, which means the new versions of $\bV_1$ -- $\bV_6$ have
\[ \Big((K \br_p)(\Xhi) -  \rho^{p+1}  \Lp_1 \be_{p+1}' \Gq^{-1} (K \br_{p+1})(\Xbi)\Big)\Big((K \br_p)(\Xhi) -  \rho^{p+1}  \Lp_1 \be_{p+1}' \Gq^{-1} (K \br_{p+1})(\Xbi)\Big)'. \]
This is mostly a change in notation and increased complexity of all terms, which now will include many more factors that much be accounted for. This does not affect the rates or the identity of the important terms: in other words the expansion is not fundamentally changed. Notice that in estimating the residuals $\hat{v}(X_i) = ( Y_i - \br_{p+1}(X_i - \x)'\bhat_{p+1} )^2$ is used, and not, as might also be plausible, any further bias correction (such as  $\hat{v}(X_i) = ( Y_i - \br_{p+1}(X_i - \x)'\Gp^{-1} \Orbc  \bY/(n\h)  )^2$. This means no other terms appear.   

We illustrate with one example. Consider the first term bounded in Equation \eqref{suppeqn:V3}. For $\bV_3$ defined following Equation \eqref{suppeqn:shat terms} it was shown following Equation \eqref{suppeqn:V3} that 
	\[r_{\trbc,\f}^{-1} \P \left[\Big| \v!^2 \be_\v'\Gp^{-1}   \left( \bV_3 - \E[\bV_3] \right)  \Gp^{-1} \be_\v \tO \v! \be_\v' \Gp^{-1} \Op \left( \bY  - \bM \right)/ n \Big|  > r_n  \right] =o(1).\]
The corresponding bound required here is
\begin{equation}
	\label{suppeqn:V3 rbc}
	r_{\irbc,\f}^{-1} \P \left[\Big| \v!^2 \be_\v'\Gp^{-1}   \left( \bV_{3,\RBC} - \E[\bV_{3,\RBC}] \right)  \Gp^{-1} \be_\v \tO \v! \be_\v' \Gp^{-1} \O_\RBC \left( \bY  - \bM \right)/ n \Big|  > r_n  \right] =o(1).
\end{equation}
The analogue of $\bV_3$ is given by applying the two changes above: the bias term and replacing $(K^2 \br_p \br_p')(\Xhi)$ with the expression above, yielding what we will call $\bV_{3,\RBC}$:
\begin{align*}
	\bV_{3,\RBC} & = \frac{1}{n \h} \sumi (K^2 \br_p \br_p')(\Xhi)  \left[\mu(X_i) -  \br_{p+1}(X_i - \x)'\bbeta_{p+1}\right]^2  		\\
	& \quad + \rho^{2p+2} \Lp_1 \be_{p+1}' \Gq^{-1} \left\{ \frac{1}{n \h} \sumi (K^2 \br_{p+1} \br_{p+1}')(\Xbi)  \left[\mu(X_i) -  \br_{p+1}(X_i - \x)'\bbeta_{p+1}\right]^2 \right\}  \Gq^{-1} \be_{p+1} \Lp_1'		\\
	& \quad + \rho^{p+1} \Lp_1 \be_{p+1}' \Gq^{-1} \left\{ \frac{1}{n \h} \sumi (K \br_{p+1})(\Xbi) (K \br_p)(\Xhi)  \left[\mu(X_i) -  \br_{p+1}(X_i - \x)'\bbeta_{p+1}\right]^2 \right\} 		\\
	& \quad + \rho^{p+1} \left\{ \frac{1}{n \h} \sumi (K \br_p )(\Xhi) (K \br_{p+1}')(\Xbi)  \left[\mu(X_i) -  \br_{p+1}(X_i - \x)'\bbeta_{p+1}\right]^2 \right\}  \Gq^{-1} \be_{p+1} \Lp_1'.
\end{align*}
Verifying Equation \eqref{suppeqn:V3 rbc} now amounts to repeating the original logic (for the first term of Equation \eqref{suppeqn:V3}) four times, once for each line here. 

First, observe that all the conclusions of Lemma \ref{supplem:gamma} hold in exactly the same way for $\Gq$ (substituting $\b$ and $p+1$ for $\h$ and $p$ respectively, as needed), and thus the same type of bounds can be applied whenever necessary. Second, Lemma \ref{supplem:lambda} implies that we can bound and remove the $\Lp_1$ everywhere as well, just as was originally done with $\Gp^{-1}$. These two together imply that Lemma \ref{supplem:var} holds for $\Orbc$ in place of $\Op$ (again with $\b$ and $p+1$ where necessary).

For the first term listed of $\bV_{3,\RBC}$ the original logic now goes through almost as written, simply with additional bounds for $\Lp_1$ and $\Gq$. Lemma \ref{supplem:bias2} applies just the same, only $p$ is replaced by $p+1$ but this is accounted for automatically by the generic rates. 

For the remaining three terms listed of $\bV_{3,\RBC}$, the argument is much the same. The only additional complexity is the bandwidth $\b$ (or $\rho$). However, because $\b$ does not vanish faster than $\h$, this will not cause a problem. Firstly, pre-multiplication by $\rho$ to a positive power can only reduce the asymptotic order because $\rho \not\to \infty$. Secondly, for the factors enclosed in braces in each of the three terms, Lemma \ref{supplem:bias2} will still hold. Checking the proof of Lemma \ref{supplem:4}, which gives Lemma \ref{supplem:bias2}, we can see that we simply must substitute the appropriate bias calculations of Section \ref{supp:bias lp}. 

For the second term listed of $\bV_{3,\RBC}$ this is immediate, since the form is identical and we only need to substitute $\b$ and $p+1$ for $\h$ and $p$ respectively, after re-writing so the averaging is done according to $n\b$ instead of $n\h$. 
\[\rho^{2p+1} \Lp_1 \be_{p+1}' \Gq^{-1} \left\{ \frac{1}{n \b} \sumi (K^2 \br_{p+1} \br_{p+1}')(\Xbi)  \left[\mu(X_i) -  \br_{p+1}(X_i - \x)'\bbeta_{p+1}\right]^2 \right\}  \Gq^{-1} \be_{p+1} \Lp_1'.\]

For the third and fourth terms listed of $\bV_{3,\RBC}$, the only potential further complication is that the summand includes both $\Xhi$ and $\Xbi$. However, because $\Xbi = \rho \Xhi$, all applications of changing variables can proceed as usual, as typified by, for smooth functions $m_1$ and $m_2$ (c.f.\ Lemma \ref{supplem:generic})
\[ \h^{-1} \E[ (K m_1 )(\Xhi) (K m_2)(\Xbi)  ]  = \int_{-1}^{1} (K m_1 )(u) (K m_2)(\rho u) f(\x + u\h) du,  \]
which is just as well behaved as usual. 

Collecting all of these results establishes the convergence of Equation \eqref{suppeqn:V3 rbc}. This illustrates that although the notational complexity is increased and there are more terms to keep track of, there is nothing fundamentally different in {\bf Step (I)} for $\trbc$. We omit the rest of the details.

Moving to {\bf Step (II)}, the proof proceeds in almost exactly the same way as in Section \ref{supp:proof lp us step 2}, but now the quantity $\bZ_i$ is different. Collecting all the changes described above (the inclusion of $\Gq$, $Lp_1$, and $\Oq$, the change in estimated residuals to $\Shatrbc$, and the premultiplication by $\Orbc$), the new $\bZ_i$ is now the collection (deleting duplicate entries)
\begin{align}
	\label{suppeqn:bZ rbc}
	\begin{split}
		\bZ_{i,\RBC}  = \Bigg( 
		\bZ_{i,\RBC}^{\text{numer}}, \ & \bZ_{i,\RBC}^{\text{denom}}\Big[(K^2 \br_p \br_p')(\Xhi)\Big], \ \bZ_{i,\RBC}^{\text{denom}}\Big[(K^2 \br_{p+1} \br_{p+1}')(\Xbi)\Big],   			\\
		&  \bZ_{i,\RBC}^{\text{denom}}\Big[(K \br_p )(\Xhi) (K \br_{p+1}')(\Xbi)\Big]  
		\Bigg)', 
	\end{split}
\end{align}
where
\begin{align*}
	\bZ_{i,\RBC}^{\text{numer}}  = \Bigg(
	& \Big\{  (K \br_p)(\Xhi) (Y_i - \br_{p+1}(X_i - \x)'\bbeta_{p+1})  \Big\}'  ,    		\\ 		  %from  numerator, original estimate
	& \Big\{  (K \br_{p+1})(\Xbi) (Y_i - \br_{p+1}(X_i - \x)'\bbeta_{p+1})  \Big\}'  ,    		\\ 		  %from  numerator, bias correction
	& \vech\Big\{  (K \br_p \br_p')(\Xhi)  \Big\}' ,    		\\ 		  %from  numerator, Gamma matrix of original estimate
	& \vech\Big\{  (K \br_{p+1} \br_{p+1}')(\Xbi)  \Big\}' ,    		\\ 		  %from  numerator, Gamma matrix of bias correction
	& \vech\Big\{  (K \br_p)(\Xhi)(\Xhi)^{p+1}  \Big\}' ,    	 		  %from  Lambda_1 matrix
	\Bigg)
\end{align*}
and for a matrix depending on $(\Xhi, \Xbi)$, the function $\bZ_{i,\RBC}^{\text{denom}}\Big[\bm{\kappa}(\Xhi, \Xbi)\Big]$ is
\begin{align*}
	\bZ_{i,\RBC}^{\text{denom}}\Big[\bm{\kappa}(\Xhi, \Xbi)\Big]  = \Bigg(
	& \vech\Big\{  \bm{\kappa}(\Xhi, \Xbi) \e_i^2   \Big\}'  ,    		\\ 		  %from  \bV_1 
	& \vech\Big\{ \bm{\kappa}(\Xhi, \Xbi) (\Xbi)^0 \e_i   \Big\}', \vech\Big\{  (K^2 \br_p \br_p' )(\Xhi) (\Xbi)^1 \e_i   \Big\}',   		\\
	& \qquad \vech\Big\{  \bm{\kappa}(\Xhi, \Xbi) (\Xbi)^2 \e_i   \Big\}' ,\ldots, \vech\Big\{  (K^2 \br_p \br_p' )(\Xhi) (\Xbi)^{p+1} \e_i   \Big\}' ,	\\   %from   \bV_2
	& \vech\Big\{  \bm{\kappa}(\Xhi, \Xbi) \big\{  \e_i \big[\mu(X_i) -  \br_{p+1}(X_i - \x)'\bbeta_{p+1}\big] \big\}   \Big\}'   		  %from  \bV_4
	\Bigg).
\end{align*}
$\bZ_{i,\RBC}$ is notationally intimidating, but comparing this to the original $\bZ_i$ of Equation \eqref{suppeqn:bZ}, we see that nothing fundamentally different has been added: the additions are mostly just repetition to account for the higher degree local polynomial. Notice that if $\rho=1$, i.e.\ $\h=\b$, then many of the elements are duplicated (or contained in others) and can be removed: examples include the first, third, and fifth lines of $\bZ_{i,\RBC}^{\text{numer}}$ and all of $\bZ_{i,\RBC}^{\text{denom}}\Big[(K^2 \br_p \br_p')(\Xhi)\Big]$. (Note also that in estimating the residuals $\hat{v}(X_i) = ( Y_i - \br_{p+1}(X_i - \x)'\bhat_{p+1} )^2$ is used, and not, as might also be plausible, any further bias correction (such as  $\hat{v}(X_i) = ( Y_i - \br_{p+1}(X_i - \x)'\Gp^{-1} \Orbc  \bY/(n\h)  )^2$. This means no other terms appear.)   

Because, by assumption, $\rho \not\to \infty$, the asymptotic orders do not change. Therefore, verifying conditions (I), (II), and (IV) of Theorem 3.4 of \cite{Skovgaard1981_SJS} are nearly identical for this new $\bZ_{i,\RBC}$. For condition ($\text{III}''_\alpha$) of \cite[Theorem 3.4 and Remark 3.5]{Skovgaard1981_SJS} the crucial ingredient is Lemma \ref{supplem:Cramer lp}, which continues to hold in exactly the same way.

Finally, {\bf Step (III)} carries over essentially without change, completing the proof of Theorem \ref{suppthm:EE lp} with bias correction.

%%%%%%%%%%%%%%%%%%%%%%%%%%%%%%%%%%%%%%%%%%%%%%%%%%%%%%%%%%%%%%%%%%%%%%
\subsection{Lemmas}
	\label{supp:lemmas lp}

Our proof of Theorem \ref{suppthm:EE lp} relies on the following lemmas. Consistent with the above, we give mainly details for the $\tp$ case, i.e. the proof in Section \ref{supp:proof lp us}. The details for $\trbc$, Section \ref{supp:proof lp rbc}, are entirely analogous. Indeed, though all the results below are stated for a bandwidth sequence $\h$ and polynomial degree $p$, they generalize in the obvious way under the appropriate substitutions and appropriate assumptions.

%%%%%%%%%%%
%  Delta Method  %
%%%%%%%%%%%

The first lemma collects high level results regarding the Delta method for Edgeworth expansions, pertaining to {\bf Step (I)}, verifying Equation \eqref{suppeqn:step 1}.
\begin{lemma} \
	\label{supplem:delta}
	\begin{enumerate}[ref=\ref{supplem:delta}(\alph{*})]
	
		\item   \label{supplem:delta1}  
			Let $U_n := \tp - \breve{T}$. If $r_{\tp,\f}^{-1} \P [|U_n| > r_n] = o(1)$ for a sequence $r_n$ such that $r_n = o(r_{\tp,\f})$, then 			
			\[\P \left[ \tp < z \right] = \P \left[ \breve{T}  + U_n < z \right] = \P \left[ \breve{T} < z \right] + o(r_{\tp,\f}). \]
		
		\item    \label{supplem:delta2}
		  If $r_1 = O(r_1')$ and $r_2 = O(r_2')$, for sequences of positive numbers $r_1$, $r_1'$, $r_2$, and $r_2'$ and if a sequence of nonnegative random variables obeys $(r_1)^{-1} \P [ U_n > r_2 ] = o(1)0$ it also holds that $ (r_1')^{-1} \P [ U_n > r_2' ] = o(1)$. In particular, $r_1^{-1} \P [|U_n| > r_n] = o(1)$ implies $r_{\tp,\f}^{-1} \P [|U_n| > r_n] = o(1)$, for $r_1$ equal in order to any of $\tO^{-2}$, $\btp^2$, or $\tO^{-1} \btp$, because $r_{\tp,\f}$ is the largest of these, and any $r_n = o(r_{\tp,\f})$. Thus, for different pieces of $U_n$ defined above, we may make different choices for these two sequences, as convenient.
	\end{enumerate}
\end{lemma}
\begin{proof}
Part (a) is the Delta method for Edgeworth expansions, which essentially follows from the fact that the Edgeworth expansion itself is a smooth function. See \cite[Chapter 2.7]{Hall1992_book} or \cite[Lemma 2 and Remark following]{Maesono1997_JSPI}. Part (b) follows from elementary inequalities.
\end{proof}

%%%%%%%%%%%
%  Main Lemmas for Remainder Terms %
%%%%%%%%%%%

The next set of results, Lemmas \ref{supplem:gamma}--\ref{supplem:bias4}, give rate bounds on the probability of deviations for various kernel-weighted sample averages. These are used in establishing Equation \eqref{suppeqn:delta 1} in {\bf Step (I)}. The proofs for all these Lemmas are given in the subsubsection below.

\begin{lemma}
	\label{supplem:gamma}
	Let the conditions of Theorem \ref{suppthm:EE lp} hold. For some $\delta > 0$, a positive integer $k$, and $C_\Gamma < \infty$, we have 
	\begin{enumerate}
		\item $r_{\tp,\f}^{-1} \P[ |\Gp - \Gpt| > \delta \tO^{-1} \log(\tO)^{1/2} ] = o(1)$, 
		\item $r_{\tp,\f}^{-1} \P\Big[ \Big|\Gp^{-1} - \sum_{j=0}^k \big(\Gp^{-1}(\Gpt - \Gp)\big)^{j} \Gpt^{-1} \Big| > \delta \tO^{-(k+1)} \log(\tO)^{(k+1)/2} \Big] = o(1)$, and in particular (i.e.\ $k=0$) $r_{\tp,\f}^{-1} \P[ |\Gp^{-1} - \Gpt^{-1}| > \delta \tO^{-1} \log(\tO)^{1/2} ] = o(1)$, and 
		\item $r_{\tp,\f}^{-1}\P[ \Gp^{-1} > C_\Gamma ] = o(1)$.
	\end{enumerate}	 
	  
\end{lemma}

\begin{lemma}
	\label{supplem:lambda}
	Let the conditions of Theorem \ref{suppthm:EE lp} hold. Let $\bA$ be a fixed-dimension vector or matrix of continuous functions of $\Xhi$ that does not depend on $n$. For some $\delta > 0$, 
	\[r_{\tp,\f}^{-1} \P \left[\left|  \frac{1}{n \h} \sumi \left\{   (K \bA) (\Xhi)  - \E[(K \bA) (\Xhi)]  \right\} \right|  > \delta \tO^{-1} \log(\tO)^{1/2}  \right] = o(1).\]
	Further, there is some constant $C_{\bA} > 0$ such that $r_{\tp,\f}^{-1}\P[ \sumi (K \bA) (\Xhi)/(n\h) > C_{\bA} ] = o(1)$.
	In particular, $r_{\tp,\f}^{-1} \P[ |\Lp_1  - \Lpt_1 | > \delta \tO^{-1} \log(\tO)^{1/2} ] = o(1)$. Lemma \ref{supplem:gamma}(a) is also a special case.
\end{lemma}

\begin{lemma}
	\label{supplem:var}
	Let the conditions of Theorem \ref{suppthm:EE lp} hold. Let $\bA$ be a fixed-dimension vector or matrix of continuous functions of $\Xhi$ that does not depend on $n$. For some $\delta > 0$, 
	\[r_{\tp,\f}^{-1} \P \left[\left|  \frac{1}{n \h} \sumi \left\{   (K \bA) (\Xhi) \e_i  \right\} \right|  > \delta \tO^{-1} \log(\tO)^{1/2}  \right] = o(1).\]
	In particular, with $\bA = \br_p(\Xhi)$, $r_{\tp,\f}^{-1} \P \left[\left| \Op \left( \bY  - \bM \right)/ n \right|  > \delta \tO^{-1} \log(\tO)^{1/2}  \right]$.
%	The same holds with $\e_i^2 - v(X_i)$ in place of $\e_i$, since it is conditionally mean zero and has more than four moments.
\end{lemma}

\begin{lemma}
	\label{supplem:bias1} 
	Let the conditions of Theorem \ref{suppthm:EE lp} hold. Let $\bA$ be a fixed-dimension vector or matrix of continuous functions of $\Xhi$ that does not depend on $n$. For any $\delta > 0$, $\gamma > 0$, and positive integer $k$,
	\[r_{\tp,\f}^{-1} \P \bigg[\bigg| \frac{1}{n \h} \sumi \bigg\{   (K \bA) (\Xhi)   \left[\mu(X_i) -  \br_p(X_i - \x)'\bbeta_p\right]^k  \bigg\}  \bigg| >    \delta \frac{\btp^{k-1}}{\tO^{k-1}} \log(\tO)^\gamma \bigg] = o(1).\]
	In particular, with $k=1$ and $\bA = \br_p(\Xhi)$, $r_{\tp,\f}^{-1} \P \left[\left| \Op \left( \bM  - \bR \bbeta_p \right)/ n \right|  > \delta \log(\tO)^\gamma  \right] = o(1)$.
\end{lemma}

\begin{lemma}
	\label{supplem:bias2} 
	Let the conditions of Theorem \ref{suppthm:EE lp} hold. Let $\bA$ be a fixed-dimension vector or matrix of continuous functions of $\Xhi$ that does not depend on $n$. For any $\delta > 0$, $\gamma > 0$, and positive integer $k$,
	\begin{multline*}
		r_{\tp,\f}^{-1} \P \bigg[\bigg| \frac{1}{n \h} \sumi \Big\{ (K \bA)(\Xhi) \left[\mu(X_i) -  \br_p(X_i - \x)'\bbeta_p\right]^k      		\\
		    - \E\left[ (K \bA)(\Xhi) \left[\mu(X_i) -  \br_p(X_i - \x)'\bbeta_p\right]^k \right] \Big\}  \bigg|  > \delta_2 \frac{\btp^k}{\tO^k} \log(\tO)^\gamma   \bigg]   = o(1).
	\end{multline*}
\end{lemma}

\begin{lemma}
	\label{supplem:bias3} 
	Let the conditions of Theorem \ref{suppthm:EE lp} hold. Let $\bA$ be a fixed-dimension vector or matrix of continuous functions of $\Xhi$ that does not depend on $n$. For any $\delta > 0$ and $\gamma > 0$,
	\[r_{\tp,\f}^{-1} \P \bigg[\bigg| \frac{1}{n \h} \sumi \bigg\{   (K \bA) (\Xhi)   \left[\mu(X_i) -  \br_p(X_i - \x)'\bbeta_p\right] \e_i  \bigg\}  \bigg| >    \delta \frac{\btp}{\tO} \log(\tO)^\gamma \bigg] = o(1).\]
\end{lemma}

\begin{lemma}
	\label{supplem:bias4}
	Let the conditions of Theorem \ref{suppthm:EE lp} hold. For any $\delta > 0$ and $\gamma > 0$,
	\begin{multline*}
		r_{\tp,\f}^{-1} \P \bigg[\bigg| \frac{1}{n \h} \sumi  \Big\{ (K \br_p \br_p')(\Xhi) \Big( K(\Xhi) \left(\mu(X_i) -  \br_p(X_i - \x)'\bbeta_p\right)        		\\
		    -     \E\left[ K(\Xhi) \left(\mu(X_i) -  \br_p(X_i - \x)'\bbeta_p\right) \right]  \Big) \e_i  \Big\}   \bigg|  > \delta a_n \log(\tO)^\gamma   \bigg] = o(1).
	\end{multline*}
	where set $a_n = \tO^{-1} \btp$ if $r_{\tp,\f} = \tO^{-2}$; $a_n = \tO^{-2}$ if $r_{\tp,\f} = \btp^2$; or $a_n = \tO^{-3/2} \btp^{1/2}$ if $r_{\tp,\f} = \tO^{-1} \btp$.
\end{lemma}

Next, we show that the random variable $\bZ_i$, given in Equation \eqref{suppeqn:bZ}, obeys the appropriate $n$-varying version of Cram\'er's condition. This is used in {\bf Step (II)} to prove that the distribution of the (properly centered and scaled) sample average of $\bZ_i$ has an Edgeworth expansion. This type of Cram\'er's condition was first (to our knowledge) used by \cite{Hall1991_Statistics}.

\begin{lemma}
	\label{supplem:Cramer lp} 
	Let the conditions of Theorem \ref{suppthm:EE lp} hold. Let $\xi_Z(\bt)$ be the characteristic function of the random variable $\bZ_i$, given in Equation \eqref{suppeqn:bZ}. For $\h$ sufficiently small, for all $C_1 > 0$ there is a $C_2 >0$ such that
	\[ \sup_{|t| > C_1} |\xi_Z(\bt)| < (1 - C_2 \h). \]
\end{lemma}
\begin{proof}[Proof of Lemma \ref{supplem:Cramer lp}]
	\label{supplempf:Cramer lp}
	Recall the definition of $\bZ_i$ in Equation \eqref{suppeqn:bZ}. It is useful to consider $\bZ_i$ as a function of $(\Xhi, Y_i)$ rather than $(X_i, Y_i)$. We compute the characteristic function separately depending on whether $X_i$ is local to $\x$. Note that $\h$ is fixed. The characteristic function of $\bZ_i$ is 	
	\begin{equation}
		\label{suppeqn:Cramer1}
		\xi_Z(\bt) = \E[\exp\{\text{i}\bt'\bZ_i\}] = \E\left[\exp\{\text{i}\bt'\bZ_i\} \One\left\{ |\Xhi| >1 \right\} \right] +  \E\left[\exp\{\text{i}\bt'\bZ_i\} \One\left\{ |\Xhi| \leq 1 \right\} \right].
	\end{equation}
	We examine each piece in turn. For the first, begin by noticing that $|\Xhi| >1$ (i.e.\ $X_i \not\in \{\x \pm \h\}$), then $K(\Xhi) = 0$, in turn implying that $\bZ_i$ is the zero vector and $\exp\{\text{i}\bt'\bZ_i\} = 1$. Therefore
	\[\E\left[\exp\{\text{i}\bt'\bZ_i\} \One\left\{ |\Xhi| >1 \right\} \right] = \P[X_i \not\in \{\x \pm \h\}] . \]
	By assumption, the density of $X$ is bounded and bounded away from zero in a fixed neighborhood of $\x$. For now consider interior $\x$, we will return to the boundary case at the end. Assume that $\h$ is small enough that this neighborhood contains $\{\x \pm \h\}$. Then this probability is bounded as 
	\begin{equation}
		\label{suppeqn:Cramer2}
		\P[X_i \not\in \{\x \pm \h\}] = 1 - \int_{\x - \h}^{\x + \h} f(x)dx \leq 1 - \h 2 \left( \min_{x\in \{\x \pm \h\}} f(x) \right) := 1 - C_3 \h.
	\end{equation}

	Next, consider the event that $|\Xhi| \leq 1$. Let $f_{xy}(x,y)$ denote the joint density of $(X,Y)$ and explicitly write $\bZ_i = \bZ_i(\Xhi,Y_i)$. Using the change of variables $U = (X - \x)/\h$,
	\begin{align*}
		\E\left[\exp\{\text{i}\bt'\bZ_i(\Xhi,Y_i)\} \One\left\{ |\Xhi| \leq 1 \right\} \right] & = \int\int_{\x - \h}^{\x + \h} \exp\{\text{i}\bt'\bZ_i(x,y)\} f_{xy}(x,y) dxdy   		\\
		& = \h \int\int_{-1}^{1} \exp\{\text{i}\bt'\bZ_i(u,y)\}f_{xy}(\x + u\h,y)dudy.
	\end{align*}
	Suppose that $K$ is not the uniform kernel. The assumption that $(1, K \br_{3p})(u)'$ is linearly independent implies that $\bZ_i$ is a set of linearly independent and continuously differentiable functions of $(u,y)$ on $\{[-1,1]\} \cup \R$. Furthermore, by assumption, the density of $(U, Y)$, as random variables on $\{[-1,1]\} \cup \mathcal{Y}$, for some $\mathcal{Y} \subset \R$, is strictly positive. Therefore, by \cite[Lemma 1.4]{Bhattacharya1977_AoP}, $\bZ_i = \bZ_i(U,Y)$ obeys Cram\'er's condition (as a function of random variables on $\{[-1,1]\} \cup \R$), and so \citep[p.\ 207]{Bhattacharya-Rao1976_book} there is some $C > 0$ such that 
	\begin{equation}
		\label{suppeqn:Cramer3}
		\sup_{|t| > C } \left| \int\int_{-1}^{1} \exp\{\text{i}\bt'\bZ_i(u,y)\}f_{xy}(\x + u\h,y)dudy \right| < 1.
	\end{equation}
	Collecting Equations \eqref{suppeqn:Cramer1}, \eqref{suppeqn:Cramer2}, and \eqref{suppeqn:Cramer3} yields the result when the kernel is not uniform.

	If $K$ is the uniform kernel, Equation \eqref{suppeqn:Cramer3} will still hold, as follows. Note that one element of $\bZ_i(U,Y)$ is $K(U)$. For notational ease, let this be the first element, and further write $\bZ_i(U,Y)$ as $\bZ_i(U,Y) := 2(K(U), \tilde{\bZ}_i')'$ and $\bt \in \R^{\dim(\bZ)}$ as $\bt = (t_{(1)}, \tilde{\bt}')'$. Then, because $K(U) \equiv 1/2$ for $U \in [-1,1]$, 
	\begin{align*}
		\sup_{|t| > C } & \left| \int\int_{-1}^{1} \exp \left\{\text{i}\bt'\bZ_i(u,y) \right\} f_{xy}(\x + u\h,y)dudy \right|   		\\
		& = \sup_{|t| > C } \left| \int\int_{-1}^{1} \exp \left\{\text{i}\bt'\left[2(K(U), \tilde{\bZ}_i')'\right] \right\} f_{xy}(\x + u\h,y)dudy \right|  		\\
		& = \sup_{|t| > C } \left| \int\int_{-1}^{1} \exp \left\{\text{i}\bt'\left[(1, \tilde{\bZ}_i')'\right] \right\} f_{xy}(\x + u\h,y)dudy \right|  		\\
		& = \sup_{|t| > C } \left| e^{\text{i}t_{1}}\int\int_{-1}^{1} \exp \left\{\text{i}\tilde{\bt}'\tilde{\bZ}_i \right\} f_{xy}(\x + u\h,y)dudy \right| .
	\end{align*}
	Exactly as above, \cite[Lemma 1.4]{Bhattacharya1977_AoP} applies, but now to $\tilde{\bZ}_i$, and $|e^{\text{i}t_{1}}|$ is bounded by one, thus yielding Equation \eqref{suppeqn:Cramer3}.

	Finally, if $\x$ is a boundary point, then all that changes in the above proof are ranges of integration: replace $\x-\h$ with zero and remove the factor of 2 in the definition of $C_3$ in \eqref{suppeqn:Cramer2}, and then in the subsequent steps, integrate over $[0,1]$ instead of $[-1,1]$.	
\end{proof}

%%%%%%%%%%%%%%%%%%%%%%%%%%%%%%%%%%%
\subsubsection{Proofs of Lemmas \ref{supplem:gamma}--\ref{supplem:bias4}}

Before proving Lemmas \ref{supplem:gamma}--\ref{supplem:bias3} we first state some generic results that serve as building blocks for the main Lemmas above. Indeed, those results are often are almost immediate consequences of these generic results. The versions of these results for $\irbc$ are usually omitted, as they are entirely analogous (replacing $p$ and $\h$ by $p+1$ and $\b$, as well as other obvious modifications).

%%%%%%%%%%%
%  Generic Lemma  %
%%%%%%%%%%%

\begin{lemma}
	\label{supplem:generic}
	Let the conditions of Theorem \ref{suppthm:EE lp} hold. Let $g(\cdot)$ and $m(\cdot)$ be generic continuous scalar functions. For some $\delta_1 > 0$, any $\delta_2 > 0$, $\gamma >0$, and positive integer $k$, the following hold.
	\begin{enumerate}[ref=\ref{supplem:generic}(\alph{*})]

		\item \label{supplem:1}  $\displaystyle  \tO^2 \P\left[ \left\vert \tO^{-2} \sumi \left\{(Km)(\Xhi) g(X_i) - \E[(Km)(\Xhi) g(X_i)] \right\}\right\vert > \delta_1 \tO^{-1} \log(\tO)^{1/2} \right] = o(1).$

		\item \label{supplem:2}  $\displaystyle  \tO^2 \P\left[ \left\vert \tO^{-2} \sumi \left\{(Km)(\Xhi) g(X_i) \e_i \right\} \right\vert > \delta_1 \tO^{-1} \log(\tO)^{1/2} \right] = o(1).$

		\item \label{supplem:3}  $\displaystyle  \frac{\tO}{\btp} \P  \left[ \left\vert \tO^{-2} \sumi (Km)(\Xhi) g(X_i)  \left[\mu(X_i) - \br_p(X_i - \x)'\bbeta_p\right]^k  \right\vert > \delta_2 \frac{\btp^{k-1}}{\tO^{k-1}} \log(\tO)^\gamma \right] = o(1).$
		
		\item \label{supplem:4} $\displaystyle  \tO^2 \P\biggl[ \biggl\vert \tO^{-2} \sumi  \Bigl\{ (Km)(\Xhi) g(X_i) (\mu(X_i) - \br_p(X_i - x)'\bbeta_p)^k     		\\
	 \textcolor{white}{ spacing }  - \E\left[(Km)(\Xhi) g(X_i) (\mu(X_i) - \br_p(X_i - x)'\bbeta_p)^k\right] \Bigr\}\biggr\vert > \delta_2 \left(\frac{\btp}{\tO}\right)^k \log(\tO)^\gamma \biggr] = o(1).$

		\item \label{supplem:5}   $\displaystyle  \tO^2 \P\left[ \left\vert \tO^{-2} \sumi (Km)(\Xhi) g(X_i)  \e_i \left[\mu(X_i) - \br_p(X_i - x)'\bbeta_p\right] \right\vert > \delta_2 \frac{\btp}{\tO} \log(\tO)^\gamma \right] = o(1).$

		\item \label{supplem:6}   $\displaystyle r_{\tp,\f}^{-1} \P \bigg[\bigg| \frac{1}{n \h} \sumi  \Big\{ (K m) (\Xhi) \Big( K(\Xhi) \left(\mu(X_i) -  \br_p(X_i - \x)'\bbeta_p\right)        		\\
		 	\textcolor{white}{ for spacing }    -     \E\left[ K(\Xhi) \left(\mu(X_i) -  \br_p(X_i - \x)'\bbeta_p\right) \right]  \Big) \e_i  \Big\}   \bigg|  > \delta a_n \log(\tO)^\gamma   \bigg] = o(1)$, \\
		 	where set $a_n = \tO^{-1} \btp$ if $r_{\tp,\f} = \tO^{-2}$; $a_n = \tO^{-2}$ if $r_{\tp,\f} = \btp^2$; or $a_n = \tO^{-3/2} \btp^{1/2}$ if $r_{\tp,\f} = \tO^{-1} \btp$.
		 			%%% This one does not get used for anything. I thought it would help make \bV_4 go away, but it doesn't give fast enough rates

	\end{enumerate}
\end{lemma}

%%%%%%%%%%%%%%%%%%%%%%%%%%%%
%  proofs of generic results  %
%%%%%%%%%%%%%%%%%%%%%%%%%%%%

%%%%%%%
\begin{proof}[Proof of Lemma \ref{supplem:1}]
Because the kernel function has compact support and $g(\cdot)$ and $m(\cdot)$ are continuous, we have
	\[\left\vert (Km)(\Xhi) g(X_i) - \E[(Km)(\Xhi) g(X_i)] \right\vert < C_1.\]
Further, by a change of variables and using the assumptions on $f$, $g$ and $m$:
	\begin{align*}
		\V[(Km)(\Xhi) g(X_i)] \leq \E \left[ (Km)(\Xhi)^2 g(X_i)^2 \right]  &  = \int f(X_i) (Km)(\Xhi)^2 g(X_i)^2 dX_i 		\\
			&= \h \int f(\x + u \h) g(\x + u \h) (Km)(u)^2 du \leq C_2 \h.
	\end{align*}
Therefore, by Bernstein's inequality
\begin{align*}
	\tO^2 \P & \left[ \left\vert \frac{1}{\tO^2} \sumi \left\{(Km)(\Xhi) g(X_i) - \E[(Km)(\Xhi) g(X_i)] \right\} \right\vert > \delta_1 \tO^{-1} \log(\tO)^{1/2} \right]  		\\
	& \leq 2 \tO^2 \exp \left\{ - \frac{ (\tO^4) (\delta_1 \tO^{-1} \log(\tO)^{1/2})^2 / 2}{ C_2 \tO^2 + C_1 \tO^2 \delta_1 \tO^{-1} \log(\tO)^{1/2} / 3} \right\} 		\\
	& = 2 \exp\{2\log(\tO)\} \exp \left\{ - \frac{ \delta_1^2 \log(\tO) / 2}{ C_2 + C_1 \delta_1 \tO^{-1} \log(\tO)^{1/2} / 3} \right\} 		\\
	& = 2 \exp \left\{ \log(\tO) \left[ 2 - \frac{ \delta_1^2 / 2}{ C_2 + C_1 \delta_1 \tO^{-1} \log(\tO)^{1/2} / 3} \right] \right\},
\end{align*}
which vanishes for any $\delta_1$ large enough, as $\tO^{-1} \log(\tO)^{1/2} = o(1)$.
\end{proof}

%%%%%%%
\begin{proof}[Proof of Lemma \ref{supplem:2}]
For a sequence $a_n \to \infty$ to be given later, define
\[H_i = \tO^{-1} (Km)(\Xhi) g(X_i) \left(Y_i \mathbbm{1}\{Y_i \leq a_n\} - \E[Y_i \mathbbm{1}\{Y_i \leq a_n\} \mid X_i]\right)  \]
and
\[T_i = \tO^{-1}(Km)(\Xhi) g(X_i) \left(Y_i \mathbbm{1}\{Y_i > a_n\} - \E[Y_i \mathbbm{1}\{Y_i > a_n\} \mid X_i]\right).  \]
By the conditions on $g(\cdot)$ and $t(\cdot)$ and the kernel function, 
	\[\left\vert H_i \right\vert <  C_1 \tO^{-1} a_n\]
and
\begin{align*}
	\V[H_i] = \tO^{-2} \V[(Km)(\Xhi) g(X_i) Y_i \mathbbm{1}\{Y_i \leq a_n\}] & \leq \tO^{-2} \E\left[(Km)(\Xhi)^2 g(X_i)^2 Y_i^2 \mathbbm{1}\{Y_i \leq a_n\} \right]		\\
	& \leq \tO^{-2} \E\left[(Km)(\Xhi)^2 g(X_i)^2 Y_i^2 \right]  		\\
	& = \tO^{-2} \int (Km)(\Xhi)^2 g(X_i)^2 v(X_i) f(X_i) d X_i    		\\
	& = \tO^{-2} \h \int (Km)(u)^2 (gvf)(\x - u\h) d u   		\\
	& \leq C_2 /n.
\end{align*}
Therefore, by Bernstein's inequality
\begin{align*}
	\tO^2 \P  \left[ \left\vert  \sumi H_i \right\vert > \delta_1 \log(\tO)^{1/2} \right]  	& \leq 2 \tO^2 \exp \left\{ - \frac{ \delta_1^2 \log(\tO) / 2}{ C_2 + C_1 \tO^{-1} a_n \delta_1 \log(\tO)^{1/2} / 3} \right\} 		\\
	& \leq 2  \exp\{2\log(\tO)\}  \exp \left\{ - \frac{ \delta_1^2 \log(\tO) / 2}{ C_2 + C_1 \tO^{-1} a_n \delta_1 \log(\tO)^{1/2} / 3} \right\} 		\\
	& \leq 2   \exp \left\{\log(\tO) \left[ 2  - \frac{ \delta_1^2  / 2}{ C_2 +  C_1 \tO^{-1} a_n \delta_1 \log(\tO)^{1/2} / 3} \right] \right\},
\end{align*}
which vanishes for $\delta_1$ large enough as long as $\tO^{-1} a_n \log(\tO)^{1/2}$ does not diverge.

Next, let $\pi > 2$ be such that $\E[|Y|^{2 + \pi} \vert X = x]$ is finite in the neighborhood of $\x$, which is possible under Assumption \ref{suppasmpt:dgp lp}, and then, by Markov's inequality:
\begin{align*}
	\tO^2 \P \left[ \left\vert  \sumi T_i \right\vert > \delta \log(\tO)^{1/2} \right]  & \leq  \tO^2 \frac{1}{\delta^2 \log(\tO) } \E\left[ \left\vert  \sumi T_i \right\vert^2 \right]		\\
	& \leq  \tO^2 \frac{1}{\delta_1^2 \log(\tO)} n \E \left[ T_i^2 \right]		\\
	& \leq  \tO^2 \frac{1}{\delta_1^2 \log(\tO)} n \V \left[ \tO^{-1} (Km)(\Xhi) g(X_i) Y_i \mathbbm{1}\{Y_i > a_n\} \right]		\\
	& \leq  \tO^2 \frac{1}{\delta_1^2 \log(\tO)} n \tO^{-2} \E \left[  (Km)(\Xhi)^2 g(X_i)^2 Y_i^2 \mathbbm{1}\{Y_i > a_n\} \right]		\\
	& \leq  \tO^2 \frac{1}{\delta_1^2 \log(\tO)} n \tO^{-2} \E \left[   (Km)(\Xhi)^2 g(X_i)^2 \vert Y_i \vert^{2+\pi} a_n^{-\pi} \right]		\\
	& \leq  \tO^2 \frac{1}{\delta_1^2 \log(\tO)} n \tO^{-2} (C \h a_n^{-\pi}) 		\\
	& \leq  \frac{C}{\delta_1^2} \frac{\tO^2}{ \log(\tO) a_n^{\pi}},
\end{align*}
which vanishes if $\tO^2 \log(\tO)^{-1} a_n^{-\pi}= o(1)$. 

It thus remains to choose $a_n$ such that $\tO^{-1} a_n \log(\tO)^{1/2}$ does not diverge and $\tO^2 \log(\tO)^{-1} a_n^{-\pi}= o(1)$. This can be accomplished by setting $a_n = \tO^{A}$ for any $2/\pi \leq A <1$, which is possible as $\pi > 2$.
\end{proof}

%%%%%%%
\begin{proof}[Proof of Lemma \ref{supplem:3}]
By Markov's inequality
\begin{align*}
	\frac{\tO}{\btp} \P & \left[ \left\vert \tO^{-2} \sumi (Km)(\Xhi) g(X_i)  \left[\mu(X_i) - \br_p(X_i - \x)'\bbeta_p\right]^k  \right\vert > \delta_2 (\tO^{-1}\btp)^{k-1} \log(\tO)^\gamma \right]   		 \\
	& \leq \frac{\tO}{\btp} \left(\frac{\tO}{\btp}\right)^{k-1}  \frac{1}{\delta_2 \log(\tO)^\gamma} \E\left[ \h^{-1} (Km)(\Xhi) g(X_i)   \left[\mu(X_i) - \br_p(X_i - \x)'\bbeta_p\right]^k\right]   		 \\
	& \leq \frac{1}{\delta_2  \log(\tO)^\gamma} \E\left[ \h^{-1} (Km)(\Xhi) g(X_i)  \left[ \frac{\tO}{\btp} \left(\mu(X_i) - \br_p(X_i - \x)'\bbeta_p\right)\right]^k\right]   		 \\
	& = O(\log(\tO)^{-\gamma}) = o(1).
\end{align*}
This relies on the calculations in Section \ref{supp:bias lp}, and the compact support of the kernel and continuity of $m(\cdot)$ and $g(\cdot)$ to ensure that the expectation is otherwise bounded.
\end{proof}

%%%%%%%
\begin{proof}[Proof of Lemma \ref{supplem:4}]
Note that the summand is mean zero and apply Markov's inequality to find
\begin{align*}
	& \tO^2 \P\biggl[ \biggl\vert \tO^{-2} \sumi  \Bigl\{ (Km)(\Xhi) g(X_i) (\mu(X_i) - \br_p(X_i - x)'\bbeta_p)^k     		\\
	& \quad\qquad \qquad - \E\left[(Km)(\Xhi) g(X_i) (\mu(X_i) - \br_p(X_i - x)'\bbeta_p)^k\right] \Bigr\}\biggr\vert > \delta_2 \left(\frac{\btp}{\tO}\right)^k \log(\tO)^\gamma \biggr]  		\\
	& \leq  \tO^2  \left(\frac{\tO}{\btp}\right)^{2k} \frac{1}{ \delta_2^2 \log(\tO)^{2\gamma} } \tO^{-2} \E\left[ \h^{-1} (Km)(\Xhi) g(X_i) (\mu(X_i) - \br_p(X_i - x)'\bbeta_p)^{2k} \right]   		\\
	& =    \frac{1}{ \delta_2^2 \log(\tO)^{2\gamma} } \E\left[ \h^{-1} (Km)(\Xhi) g(X_i) \left[ \left(\frac{\tO}{\btp}\right) (\mu(X_i) - \br_p(X_i - x)'\bbeta_p) \right]^{2k} \right]  		\\
	& = o(1).
\end{align*}
The final line relies on the calculations in Section \ref{supp:bias lp}.
\end{proof}

%%%%%%%
\begin{proof}[Proof of Lemma \ref{supplem:5}]
By Markov's inequality, since $\e_i$ is conditionally mean zero, we have
\begin{align*}
	\tO^2 \P & \left[  \left\vert \tO^{-2} \sumi (Km)(\Xhi) g(X_i)  \e_i \left[\mu(X_i) - \br_p(X_i - \x)'\bbeta_p\right] \right\vert > \delta_2 (\tO^{-1}\btp) \log(\tO)^\gamma \right]   		\\
	& \leq \tO^2 \frac{1}{\delta_2^2 \tO^{-2}\btp^2 \log(\tO)^{2\gamma} } \frac{1}{\tO^2} \E \left[ \h^{-1} \left( (Km)(\Xhi) g(X_i)  \e_i \right)^2 \left[\mu(X_i) - \br_p(X_i - \x)'\bbeta_p\right]^2 \right]   		\\
	& \leq \frac{1}{\delta_2^2  \log(\tO)^{2\gamma} } \E \left[ \h^{-1} \left( (Km)(\Xhi) g(X_i)  \e_i \right)^2 \left[\frac{\tO}{\btp}\left(\mu(X_i) - \br_p(X_i - \x)'\bbeta_p\right)\right]^2 \right]  		\\
	& = O(\log(\tO)^{-2\gamma}) = o(1).
\end{align*}
This relies on the calculations in Section \ref{supp:bias lp}, and the compact support of the kernel and continuity of $m(\cdot)$ and $g(\cdot)$ to ensure that the expectation is otherwise bounded.
\end{proof}

%%%%%%%
\begin{proof}[Proof of Lemma \ref{supplem:6}]
By Markov's inequality, since $\e_i$ is conditionally mean zero, we have
\begin{align*}
	r_{\tp,\f}^{-1} & \P \bigg[\bigg| \frac{1}{n \h} \sumi  \Big\{ (K m) (\Xhi) \Big( K(\Xhi) \left(\mu(X_i) -  \br_p(X_i - \x)'\bbeta_p\right)        		\\
	& \qquad\qquad\qquad   -     \E\left[ K(\Xhi) \left(\mu(X_i) -  \br_p(X_i - \x)'\bbeta_p\right) \right]  \Big) \e_i  \Big\}   \bigg|  > \delta a_n \log(\tO)^\gamma   \bigg]  		\\
	& \leq \frac{r_{\tp,\f}^{-1} }{a_n^2 \log(\tO)^{2\gamma} } \frac{1}{n\h} \E\bigg[ \h^{-1} (K m)^2 (\Xhi) \Big( K(\Xhi) \left(\mu(X_i) -  \br_p(X_i - \x)'\bbeta_p\right)        		\\
	& \qquad\qquad\qquad\qquad\qquad   -     \E\left[ K(\Xhi) \left(\mu(X_i) -  \br_p(X_i - \x)'\bbeta_p\right) \right]  \Big)^2 v(X_i) \bigg]   		\\
	& = \frac{r_{\tp,\f}^{-1} }{a_n^2 \log(\tO)^{2\gamma} } \frac{1}{n\h} \Bigg\{ \E\bigg[ \h^{-1} (K m)^2 (\Xhi) K(\Xhi)^2 \left(\mu(X_i) -  \br_p(X_i - \x)'\bbeta_p\right)^2  v(X_i) \bigg]   		\\
	&  \qquad   -2 \E\bigg[ \h^{-1} (K m)^2 (\Xhi) K(\Xhi) \left(\mu(X_i) -  \br_p(X_i - \x)'\bbeta_p\right)  v(X_i) \bigg] \E\left[ K(\Xhi) \left(\mu(X_i) -  \br_p(X_i - \x)'\bbeta_p\right) \right]   			\\
	&  \qquad  + \E\bigg[ \h^{-1} (K m)^2 (\Xhi) v(X_i) \bigg] \E\left[ K(\Xhi) \left(\mu(X_i) -  \br_p(X_i - \x)'\bbeta_p\right) \right]^2   \Bigg\}
	\\
	& \asymp \frac{r_{\tp,\f}^{-1} }{a_n^2 \log(\tO)^{2\gamma} } \frac{1}{n\h} \left(  \frac{\btp}{\tO} \right)^2 \left\{1 + \h + \h^2\right\}  		\\
	& \asymp \frac{r_{\tp,\f}^{-1} }{a_n^2 \log(\tO)^{2\gamma} } \frac{1}{n\h} \left(  \frac{\btp}{\tO} \right)^2.
\end{align*}
If $r_{\tp,\f} = \tO^{-2}$, this vanishes for $a_n = \tO^{-1} \btp$. If $r_{\tp,\f} = \btp^2$, this vanishes for $a_n = \tO^{-2}$. If $r_{\tp,\f} = \tO^{-1} \btp$, this vanishes for $a_n = \tO^{-3/2} \btp^{1/2}$. 
This relies on the calculations in Section \ref{supp:bias lp}, and the compact support of the kernel and continuity of $m(\cdot)$ to ensure that the expectation is otherwise bounded.
\end{proof}

%%%%%%%%%%%%%%%%%%%%%%%%%%%%
%  proofs of main lemmas  %
%%%%%%%%%%%%%%%%%%%%%%%%%%%%

%%%%%%%
\begin{proof}[Proof of Lemma \ref{supplem:gamma}]
A typical element of $\Gp - \Gpt$ is, for some integer $k \in [0, 2p]$, 
	\[\frac{1}{n\h}\sumi \left\{  K(\Xhi) \Xhi^k - \E\left[ K(\Xhi) \Xhi^k \right] \right\},\]
which has the form treated in Lemma \ref{supplem:1}. Therefore, by Boole's inequality and $p$ fixed,
\begin{multline*}
	r_{\tp,\f}^{-1} \P[ |\Gp - \Gpt| > \delta \tO^{-1} \log(\tO)^{1/2} ]   		\\
	    \leq C r_{\tp,\f}^{-1} \max_{k \in [0, 2p]} \P\left[ \left|\frac{1}{n\h}\sumi \left\{  K(\Xhi) \Xhi^k - \E\left[ K(\Xhi) \Xhi^k \right] \right\} \right| > \delta \tO^{-1} \log(\tO)^{1/2} \right]  = o(1),
\end{multline*}
by Lemma \ref{supplem:delta2}. This establishes part (a).

To prove part (b), first note that for any fixed $\delta_1$, part (a) and the sub-multiplicativity of the Frobenius norm imply
\begin{equation}
	\label{suppeqn:gamma1}
	r_{\tp,\f}^{-1}  \P\left[ |\Gp^{-1}(\Gp - \Gpt)| \geq \delta_1 \right]    \ \leq \  r_{\tp,\f}^{-1}  \P\left[ |(\Gp - \Gpt)| \geq \delta_1 |\Gp^{-1}|^{-1}\right] = o(1),
\end{equation}
because under the maintained assumptions
\begin{align*}
	\Gpt = \E\left[ \h^{-1} (K \br_p \br_p')(\Xhi) \right] = \h^{-1} \int (K \br_p \br_p')(\Xhi) f(X_i) d X_i = \int (K \br_p \br_p')(u) f(\x + u \h) d u
\end{align*}
is bounded away from zero and infinity for $n$ large enough. 

Now, on the event $\mathcal{G}_n = \{ |\Gp^{-1}(\Gp - \Gpt)| < 1\}$, we use the identity $\Gp = \Gpt\left( \bI - \Gp^{-1}(\Gpt - \Gp)\right)$ to write $\Gp^{-1}$ as
\[\Gp^{-1} = \left( \bI - \Gp^{-1}(\Gpt - \Gp)\right)^{-1} \Gpt^{-1}   = \sum_{j=0}^\infty \left(\Gp^{-1}(\Gpt - \Gp)\right)^{j} \Gpt^{-1}.\]
Write $a_n = \tO^{-(k+1)} \log(\tO)^{(k+1)/2}$ Using results \eqref{suppeqn:gamma1} with $\delta_1 = 1$, we find that $r_{\tp,\f}^{-1} (1-\P[\mathcal{G}_n] ) = r_{\tp,\f}^{-1}  \P[ |\Gp^{-1}(\Gp - \Gpt)| \geq 1] = o(1)$. Therefore
\begin{align*}
	r_{\tp,\f}^{-1} \P &  \left [ \left|\Gp^{-1} - \sum_{j=0}^k \big(\Gp^{-1}(\Gpt - \Gp)\big)^{j} \Gpt^{-1} \right| > \delta a_n \right]    	  \\	
	& \leq r_{\tp,\f}^{-1} \P \left[ \left\{ \left|\Gp^{-1} - \sum_{j=0}^k \big(\Gp^{-1}(\Gpt - \Gp)\big)^{j} \Gpt^{-1}\right| > \delta  a_n  \right\} \cup \mathcal{G}_n \right]   +   r_{\tp,\f}^{-1} (1-\P[\mathcal{G}_n] )    		\\
	& \leq r_{\tp,\f}^{-1} \P \left[ \left| \sum_{j=0}^\infty \left(\Gp^{-1}(\Gpt - \Gp)\right)^{j} \Gpt^{-1} - \sum_{j=0}^k \big(\Gp^{-1}(\Gpt - \Gp)\big)^{j} \Gpt^{-1}\right| > \delta  a_n  \right]   +   o(1)    		\\
	& = r_{\tp,\f}^{-1} \P \left[  \left| \sum_{j=k+1}^\infty \left(\Gp^{-1}(\Gpt - \Gp)\right)^{j} \Gpt^{-1}\right| > \delta  a_n  \right]   +   o(1).
\end{align*}
Again using sub-multiplicativity and part (a), $|(\Gp^{-1}(\Gpt - \Gp))^{j}| \leq |\Gp^{-1}|^j |\Gpt - \Gp|^{j} = o(1)$, and so by dominated convergence and the partial sum formula, the above display is bounded as
\begin{align*}
	& \leq r_{\tp,\f}^{-1} \P \left[  \sum_{j=k+1}^\infty \left| \left(\Gp^{-1}(\Gpt - \Gp)\right)^{j}\right| \left| \Gpt^{-1}\right|  > \delta  a_n  \right]   +   o(1)    		\\
	& \leq r_{\tp,\f}^{-1} \P \left[  \frac{\left| \Gp^{-1}(\Gpt - \Gp) \right|^{k+1}}{1 - \left| \Gp^{-1}(\Gpt - \Gp) \right|} \left| \Gpt^{-1}\right|  > \delta  a_n  \right]   +   o(1)  .
\end{align*}
Finally, using result \eqref{suppeqn:gamma1} with some fixed $\delta_1 < 1$, this last display is bounded by
\[
	r_{\tp,\f}^{-1} \P \left[  \left|\Gpt - \Gp \right|^{k+1}   > \left| \Gpt^{-1}\right|^{-k-2} ( 1 - \delta_1) \delta  a_n  \right]    +     r_{\tp,\f}^{-1} \P\left[ \left| \Gp^{-1}(\Gpt - \Gp) \right| \geq \delta_1 \right] +   o(1)   = o(1),
\]
where the final convergence follows by part (a).

For part (c), let $C_\Gamma < \infty$ be such that $| \Gpt^{-1}| < C_\Gamma/2$. Then
\begin{align*}
	r_{\tp,\f}^{-1}\P[ \Gp^{-1} > C_\Gamma ]  & = r_{\tp,\f}^{-1}\P[ \left( \Gp^{-1} - \Gpt^{-1} \right)  +  \Gpt^{-1}  > C_\Gamma ]  		\\
	& \leq r_{\tp,\f}^{-1}\P\left[ \left| \Gp^{-1} - \Gpt^{-1} \right|  > \delta \tO^{-1} \log(\tO)^{1/2} \right]  +  r_{\tp,\f}^{-1}\P \left[ \left|\Gpt^{-1} \right|  > C_\Gamma - \delta \tO^{-1} \log(\tO)^{1/2} \right],
\end{align*}
which vanishes because the second term is zero for $n$ large enough such that $\delta \tO^{-1} \log(\tO)^{1/2} < C_\Gamma/2$ and the first is $o(1)$ by part (a).
\end{proof}

%%%%%%%
\begin{proof}[Proof of Lemma \ref{supplem:lambda}]
The result follows from identical steps to proving Lemma \ref{supplem:gamma}(a), because Lemma \ref{supplem:1} also applies. The second conclusion follows from the first exactly the same way Lemma \ref{supplem:gamma}(c) follows from Lemma \ref{supplem:gamma}(a).
\end{proof}

%%%%%%%
\begin{proof}[Proof of Lemma \ref{supplem:var}]
Let $\left[\bA\right]_{j,k}$ be the $\{j,k\}$ entry of $\bA$. By Boole's inequality, since the dimension of $\bA$ is fixed, and Lemma \ref{supplem:delta2},
\begin{align*}
	r_{\tp,\f}^{-1} & \P \left[\left|  \frac{1}{n \h} \sumi \left\{   (K \bA) (\Xhi) \e_i  \right\} \right|  > \delta \tO^{-1} \log(\tO)^{1/2}  \right]  		\\
	&  \leq C r_{\tp,\f}^{-1} \max_{j,k}  \P\left[ \left\vert \tO^{-2} \sumi \left\{\left(K\left[\bA\right]_{j,k}\right)(\Xhi) \e_i \right\} \right\vert > \delta \tO^{-1} \log(\tO)^\gamma \right]  		\\
	&  \leq C \tO^2 \max_{j,k}  \P\left[ \left\vert \tO^{-2} \sumi \left\{\left(K\left[\bA\right]_{j,k}\right)(\Xhi) \e_i \right\} \right\vert > \delta \tO^{-1} \log(\tO)^\gamma \right] ,
\end{align*}
which vanishes by Lemma \ref{supplem:2}.
\end{proof}

%%%%%%%
\begin{proof}[Proof of Lemma \ref{supplem:bias1}]
Exactly as above, but using Lemma \ref{supplem:3}.
\end{proof}

%%%%%%%
\begin{proof}[Proof of Lemma \ref{supplem:bias2}]
Exactly as above, but using Lemma \ref{supplem:4}.
\end{proof}

%%%%%%%
\begin{proof}[Proof of Lemma \ref{supplem:bias3}]
Exactly as above, but using Lemma \ref{supplem:5}.
\end{proof}

%%%%%%%
\begin{proof}[Proof of Lemma \ref{supplem:bias4}]
Exactly as above, but using Lemma \ref{supplem:6}.
\end{proof}

%%%%%%%%%%%%%%%%%%%%%%%%%%%%%%%%%%%%%%%%%%%%%%%%%%%%%%%%%%%%%%%%%%%%%%
%%%%%%%%%%%%%%%%%%%%%%%%%%%%%%%%%%%%%%%%%%%%%%%%%%%%%%%%%%%%%%%%%%%%%%
\subsection{Terms of the Expansion}
	\label{supp:terms lp}

We now give the precise forms of the terms in the Edgeworth expansion, $E_{\ti,\f}(z)$. We first define them and then show their computation in a subsection below. To list them amounts to defining the terms $\w_k$, $k = 1, 2, \ldots, 6$, $\bti$, and $\lambda_{\ti,\f}$. For all $\ti$ (or $\i$), $\bti$ is given in Section \ref{supp:bias lp} and explicitly given in Equation \eqref{suppeqn:eta lp}. For the expansion, the special cases are not needed. For the variance errors $\lambda_{\ti,\f}$, we mention a few examples. First, as already discussed, the fixed-$n$ standard errors of Equation \eqref{suppeqn:se lp} yield $\lambda_{\ti,\f} \equiv 0$. When it is nonzero, typically $\lambda_{\ti,\f}$ has the form $\lambda_{\ti,\f} = l_n L$, for a rate $l_n = o(1)$ and a constant (or at least, a sequence bounded and bounded away from zero) $L$. The term $L$ is exactly the difference between the variance of the numerator of the $t$-statistic and the population standardization chosen. This has nothing to do with estimation error. Loosely speaking,
	\[ L = \frac{\V\left[ \sqrt{n\h^{1 + 2\v}}( \that - \tf) \right]}{ \sigma^2} - 1, \]
where $\sigma^2$ is the limit of Studentization whatever $\shat^2$ has been chosen (c.f.\ Equation \eqref{suppeqn:t stat}). As an example, consider traditional explicit bias correction, where the point estimate (or numerator of $\ti$) is bias-corrected but it is assumed that $\sp$ provides valid standardization (this requires $\rho = o(1)$), we find that $\lambda_{\ti,\f} = \rho ^{p+2}(L_1 + \rho^{p+2} L_2)$, where $L_1$ captures the (scaled) covariance between $\mhat^{(\v)}$ and $\mhat^{(p+1)}$ and $L_2$ the variance of $\mhat^{(p+1)}$; see \cite{Calonico-Cattaneo-Farrell2018_JASA,Calonico-Cattaneo-Farrell2018_JASAsupp} for the exact expressions. For another example, for inference at the boundary when using the asymptotic variance for standardization (i.e.\ the probability limit of the conditional variance of the numerator), one finds $l_n = \h$ and $L$ capturing the difference between the conditional variance and its limit, based on the localization of the kernel; see \cite{Chen-Qin2002_SJS} for the exact expression.

It remains to define $\w_k$, $k = 1, 2, \ldots, 6$. More notation is required. As with the bias, all terms must be nonrandom. We will maintain, as far as possible, fixed-$n$ calculations. First, define the following functions, which depend on $\f$, $n$, $\h$, $\b$, $\v$, $p$, and $K$, though this is mostly suppressed notationally. These functions are all calculated in a fixed-$n$ sense and are all bounded and rateless. 
\begin{align*}
	\l^0_{\tp}(X_i) & = \v! \be_\v' \Gpt^{-1} (K \br_p)(\Xhi) ; 		\\
	\l^0_{\trbc}(X_i) & = \l^0_{\tp}(X_i)  -   \rho^{p+1} \v! \be_\v' \Gpt^{-1} \Lpt_1 \be_{p+1}' \Gqt^{-1}  (K \br_{p+1})(\Xbi) ; 		\\
	\l^1_{\tp}(X_i, X_j) & = \v! \be_\v' \Gpt^{-1} \left( \E[(K \br_p \br_p')(\Xhj)] - (K \br_p \br_p')(\Xhj) \right) \Gpt^{-1} (K \br_p)(\Xhi) ; 		\\
	\l^1_{\trbc}(X_i, X_j) & = \l^1_{\tp}(X_i, X_j)     -    \rho^{p+1}   \v! \be_\v' \Gpt^{-1} \Bigl\{ \left( \E[(K \br_p \br_p')(\Xhj)] - (K \br_p \br_p')(\Xhj) \right)  \Gpt^{-1}  \Lpt_1 \be_{p+1}'   		\\
		& \qquad \qquad       +  \left( (K \br_p)(\Xhj) \Xhi^{p+1}   -  \E[(K \br_p)(\Xhj) \Xhi^{p+1}]  \right) \be_{p+1}'       		\\
		& \qquad \qquad       +  \Lpt_1 \be_{p+1}' \Gqt^{-1}  \left( \E[(K \br_{p+1} \br_{p+1}')(X_{\b,j})] - (K \br_{p+1} \br_{p+1}')(X_{\b,j}) \right)   \Bigr\} \Gqt^{-1}  (K \br_{p+1})(\Xbi).
\end{align*}

With this notation, define
\begin{align*}
	\st_{\ti}^2  = \E[ \h^{-1} \l^0_{\ti}(X)^2 v(X) ].
\end{align*}
We can also rewrite the bias terms using this notation as 
\[ \btp = \sqrt{n\h} \E \left[ \h^{-1} \l^0_{\tp}(X_i) [\mu(X_i) - \br_p(X_i - x)' \bbeta_p] \right]  \]
and
\[ \btrbc  = \sqrt{n\h} \E \Bigl[ \h^{-1} \l^0_{\trbc}(X_i)   [\mu(X_i) - \br_{p+1}(X_i - x)' \bbeta_{p+1}]   \Bigr]. 	\]

Now we can define the Edgeworth expansion polynomials $\w_k$, $k = 1, 2, \ldots, 6$. The standard Normal density is $\phi(z)$. The term $\w_4$ is the most cumbersome. Beginning with the others:
\begin{align*}
	\w_{1,\ti,\f}(z) &  =   \phi(z) \st_{\ti}^{-3} \E \left[ \h^{-1} \l^0_{\ti}(X_i)^3 \e_i^3 \right] \left\{  (2z^2 - 1)/6 \right\}	  , 		 \\
	\w_{2,\ti,\f}(z) &  =   -  \phi(z)  \st_{\ti}^{-1}	  , 		 \\
	\w_{3,\ti,\f}(z) &  =   -  \phi(z) \left\{ z / 2 \right\}	  , 		 \\
	\w_{5,\ti,\f}(z) &  =   -  \phi(z)  \st_{\ti}^{-2} \left\{ z / 2 \right\}	  , 		 \\
	\w_{6,\ti,\f}(z) &  =   \phi(z)  \st_{\ti}^{-4} \E [ \h^{-1} \l^0_{\ti}(X_i)^3 \e_i^3 ]  \left\{ z^3 / 3 \right\}	 .
\end{align*}
For $\w_3$, it is not quite as simple to state a generic version. Let $\bm{\tilde{G}}$ stand in for $\Gpt$ or $\Gqt$, $\tilde{p}$ stand in for $p$ or $p+1$, and $d_n$ stand in for $\h$ or $\b$, all depending on if $\ti = \tp$ or $\trbc$. Note however, that $\h$ is still used in many places, in particular for stabilizing fixed-$n$ expectations, for $\trbc$. Indexes $i$, $j$, and $k$ are always distinct (i.e.\ $X_{\h,i} \neq X_{\h,j} \neq X_{\h,k}$). 
\begin{align*}
	\w_{4,\ti,\f}(z) & =  \phi(z)  \st_{\ti}^{-6} \E \left[ \h^{-1} \l^0_{\ti}(X_i)^3 \e_i^3 \right]^2 \left\{  z^3/3 + 7 z /4 + \st_{\ti}^2 z (z^2-3)/4 \right\}   		\\
	& \quad +  \phi(z)  \st_{\ti}^{-2} \E \left[ \h^{-1} \l^0_{\ti}(X_i) \l^1_{\ti}(X_i, X_i) \e_i^2 \right] \left\{ - z (z^2 - 3) /2 \right\}   		\\
	& \quad +  \phi(z)  \st_{\ti}^{-4} \E \left[ \h^{-1} \l^0_{\ti}(X_i)^4 (\e_i^4 - v(X_i)^2) \right] \left\{ z(z^2-3)/8 \right\}   		\\
	& \quad -  \phi(z)  \st_{\ti}^{-2} \E \left[ \h^{-1}\l^0_{\ti}(X_i)^2 \br_{\tilde{p}}(X_{d_n,i})'\bm{\tilde{G}}^{-1} (K \br_{\tilde{p}})(X_{d_n,i}) \e_i^2 \right] \left\{ z(z^2 - 1)/2 \right\}   		\\
	& \quad -  \phi(z)  \st_{\ti}^{-4} \E \left[ \h^{-1} \l^0_{\ti} (X_i)^3  \br_{\tilde{p}}(X_{d_n,i})'\bm{\tilde{G}}^{-1} \e_i^2 \right] \E \left[ \h^{-1} (K \br_{\tilde{p}})(X_{d_n,i}) \l^0_{\ti} (X_i) \e_i^2 \right] \left\{ z(z^2 - 1)  \right\}   		\\
	& \quad +  \phi(z)  \st_{\ti}^{-2} \E \left[ \h^{-2} \l^0_{\ti}(X_i)^2 (\br_{\tilde{p}}(X_{d_n,i})'\bm{\tilde{G}}^{-1}(K \br_{\tilde{p}})(X_{d_n,j}) )^2 \e_j^2 \right] \left\{ z(z^2 - 1)/4 \right\}   		\\
	%	& \quad +   \phi(z)   \st_{\ti}^{-4} \E \left[ \h^{-1} \l^0_{\ti} (X_j)^2 \left( \E \left[ \h^{-1} \br_{\tilde{p}}(X_{d_n,j})'\bm{\tilde{G}}^{-1}(K \br_{\tilde{p}})(X_{d_n,i}) \l^0_{\ti}(X_i)  \e_i^2 \vert X_j \right] \right)^2 \right] \left\{  z(z^2 - 1) /2 \right\}   		\\ %%this is the old version of the line right below, with the conditional expectation squared
	& \quad +   \phi(z)   \st_{\ti}^{-4} \E \left[ \h^{-3} \l^0_{\ti} (X_j)^2 \br_{\tilde{p}}(X_{d_n,j})'\bm{\tilde{G}}^{-1}(K \br_{\tilde{p}})(X_{d_n,i}) \l^0_{\ti}(X_i) \br_{\tilde{p}}(X_{d_n,j})'\bm{\tilde{G}}^{-1}(K \br_{\tilde{p}})(X_{d_n,k}) \l^0_{\ti}(X_k)  \e_i^2 \e_k^2 \right]   		\\
	& \quad \qquad \qquad \qquad \qquad \qquad   \times \;   \left\{  z(z^2 - 1) /2 \right\}   		\\
	& \quad +  \phi(z)  \st_{\ti}^{-4} \E \left[ \h^{-1} \l^0_{\ti}(X_i)^4 \e_i^4 \right] \left\{ - z (z^2 - 3)/24 \right\}   		\\
	& \quad +  \phi(z)  \st_{\ti}^{-4} \E \left[ \h^{-1} \left( \l^0_{\ti}(X_i)^2 v(X_i) - \E[\l^0_{\ti}(X_i)^2 v(X_i)] \right) \l^0_{\ti}(X_i)^2 \e_i^2 \right] \left\{ z(z^2 - 1)/4 \right\}   		\\
	& \quad +  \phi(z)  \st_{\ti}^{-4} \E \left[ \h^{-2} \l^1_{\ti}(X_i, X_j) \l^0_{\ti}(X_i)\l^0_{\ti}(X_j)^2 \e_j^2 v(X_i) \right] \left\{  z (z^2 - 3) \right\}   		\\
	& \quad +  \phi(z)  \st_{\ti}^{-4} \E \left[ \h^{-2} \l^1_{\ti}(X_i, X_j)  \l^0_{\ti}(X_i) \left( \l^0_{\ti}(X_j)^2 v(X_j) - \E[\l^0_{\ti}(X_j)^2 v(X_j)] \right) \e_i^2 \right] \left\{ - z \right\}   		\\
	& \quad +  \phi(z)  \st_{\ti}^{-4} \E \left[ \h^{-1}  \left( \l^0_{\ti}(X_i)^2 v(X_i) - \E[\l^0_{\ti}(X_i)^2 v(X_i)] \right)^2 \right] \left\{ - z(z^2 + 1) /8 \right\}   .
\end{align*}

For computation, note that the seventh term can be rewritten by factoring the expectation, after rearranging the terms using the fact that $\br_{\tilde{p}}(X_{d_n,j})'\bm{\tilde{G}}^{-1}\br_{\tilde{p}}(X_{d_n,i})$ is a scalar, as follows
\begin{align*}
	& \E \left[ \h^{-3} \l^0_{\ti} (X_j)^2 \br_{\tilde{p}}(X_{d_n,j})'\bm{\tilde{G}}^{-1}(K \br_{\tilde{p}})(X_{d_n,i}) \l^0_{\ti}(X_i) \br_{\tilde{p}}(X_{d_n,j})'\bm{\tilde{G}}^{-1}(K \br_{\tilde{p}})(X_{d_n,k}) \l^0_{\ti}(X_k)  \e_i^2 \e_k^2 \right]   		\\
	& \qquad = \E\left[\h^{-1} \l^0_{\ti}(X_i) \e_i^2 (K \br_{\tilde{p}}')(X_{d_n,i}) \bm{\tilde{G}}^{-1} \right] \; \E\left[\h^{-1} \br_{\tilde{p}}(X_{d_n,j}) \l^0_{\ti} (X_j)^2 \br_{\tilde{p}}(X_{d_n,j})'\bm{\tilde{G}}^{-1}  \right]    		\\
	& \qquad\qquad\qquad\qquad      \times \; \E\left[\h^{-1} (K \br_{\tilde{p}})(X_{d_n,k})   \l^0_{\ti}(X_k)  \e_k^2  \right].
\end{align*}
This will greatly ease implementation.

\subsubsection{Computing the Terms}

Computing the terms of the Edgeworth expansion of Theorem \ref{suppthm:EE lp}, listed above, is straightforward but tedious. We give a short summary here, following the essential steps of \cite[Chapter 2]{Hall1992_book}. In what follows, will always discard higher order terms (those that will not appear in the Theorem) and write $A \oeq B$ to denote $A = B + o((n\h)^{-1} + (n\h)^{-1/2}\bti + \bti^2)$. Let $\bm{\tilde{G}}$ stand in for $\Gpt$ or $\Gqt$, $\tilde{p}$ stand in for $p$ or $p+1$, and $d_n$ stand in for $\h$ or $\b$, all depending on if $\ti = \tp$ or $\trbc$. Note however, that $\h$ is still used in many places, in particular for stabilizing fixed-$n$ expectations, for $\trbc$. 

The steps to compute the expansion are as follows. First, we compute a Taylor expansion of $\ti$ around nonrandom denominators. Then we compute the first four moments of this expansion. These are then combined into cumulants, which determine the terms of the expansion. 

The Taylor expansion is
\begin{multline*}
	\ti \oeq \left\{ 1 - \frac{1}{2 \st_{\ti}^2}\left( W_{\ti,1}  +  W_{\ti,2}  + W_{\ti,3}\right) + \frac{3}{8 \st_{\ti}^4} \left( W_{\ti,1}  +  W_{\ti,2}  + W_{\ti,3}\right)^2  \right\}  		\\ 		\times \;  \st_{\ti}^{-1} \left\{   N_{\ti,1} + N_{\ti,2} + N_{\ti,3} + B_{\ti,1} \right\},
\end{multline*}
where
\begin{align*}
	W_{\ti,1} & =  \frac{1}{n\h} \sumi \left\{\l^0_{\ti}(X_i)^2  \left( \e_i^2 - v(X_i) \right) \right\}  -  2 \frac{1}{n^2 \h^2} \sumi \sumj \left\{ \l^0_{\ti}(X_i)^2 \br_{\tilde{p}}(X_{d_n,i})'  \bm{\tilde{G}}^{-1} (K \br_{\tilde{p}})(X_{d_n,i}) \e_i \e_j \right\}    		\\
	& \qquad \qquad \qquad \qquad +  \frac{1}{n^3 \h^3} \sumi \sumj \sumk \left\{ \l^0_{\ti}(X_i)^2 \br_{\tilde{p}}(X_{d_n,i})'  \bm{\tilde{G}}^{-1} (K \br_{\tilde{p}})(X_{d_n,i}) \e_j \e_k  \right\} ,   		\\
	W_{\ti,2} & =  \frac{1}{n\h}\sumi \left\{  \l^0_{\ti}(X_i)^2  v(X_i)^2  - \E[ \l^0_{\ti}(X_i)^2  v(X_i)^2] \right\}  +  2 \frac{1}{n^2 \h^2} \sumi\sumj \l_{\ti}^2(X_i, X_j) \l_{\ti}^0(X_i) v(X_i),   		\\
	W_{\ti,3} & = \frac{1}{n^3 \h^3} \sumi \sumj \sumk \l_{\ti}^1(X_i, X_j) \l_{\ti}^1(X_i,X_k) v(X_i)   +   2 \frac{1}{n^3 \h^3} \sumi \sumj \sumk \l_{\ti}^2(X_i, X_j, X_k) \l_{\ti}^0(X_i) v(X_i),   			\\
	B_{\ti,1} & = \tO \frac{1}{n\h} \sumi \l^0_{\ti}(X_i) [\mu(X_i) - \br_{\tilde{p}}(X_i - x)' \bbeta_{\tilde{p}}],  		\\
	N_{\ti,1} & = \tO  \frac{1}{n\h} \sumi \l^0_{\ti}(X_i) \e_i,  		\\
	N_{\ti,2} & = \tO  \frac{1}{(n\h)^2} \sumi \sumj \l^1_{\ti}(X_i, X_j) \e_i,  		\\
	N_{\ti,3} & = \tO  \frac{1}{(n\h)^3} \sumi \sumj \sumk \l^2_{\ti}(X_i, X_j, X_k) \e_i,
\end{align*}
with the final line defining $\l^2_{\ti}(X_i, X_j, X_k)$ in the obvious way following $\l^1_{\ti}$, i.e.\ taking account of the next set of remainders. Terms involving $\l^2_{\ti}(X_i, X_j, X_k)$ are higher-order, which is why it is not needed in the final terms of the expansion. To concretize the notation, note that $\bti = \E[B_{\ti,1}]$, and, for example for $\tp$ we are defining,
\begin{align*}
	N_{\tp,1} & = \tO \v! \be_\v' \Gpt^{-1}  \Op (\bY \bM)/n,   		\\
	N_{\tp,2} & = \tO \v! \be_\v' \Gpt^{-1}(\Gpt - \Gp) \Gpt^{-1} \Op (\bY \bM)/n, 		\\
	N_{\tp,3} & = \tO \v! \be_\v' \Gpt^{-1}(\Gpt - \Gp) \Gpt^{-1} (\Gpt - \Gp) \Gpt^{-1} \Op (\bY \bM)/n.
\end{align*}

Straightforward moment calculations yield, where ``$\E[\ti] \oeq$'' denotes moments of the Taylor expansion above,
\begin{align*}
\E[\ti] & \oeq \st_{\ti}^{-1} \E\left[B_{\ti,1}\right]  - \frac{1}{2 \st_{\ti}^2} \E\left[ W_{\ti,1} N_{\ti,1}\right],
\end{align*}
\begin{align*}
\E[\ti^2] & \oeq  \frac{1}{\st_{\ti}^2} \E\left[ N_{\ti,1}^2  +  N_{\ti,2}^2 + 2 N_{\ti,1} N_{\ti,2}  + 2 N_{\ti,1} N_{\ti,3}  \right]   		\\
& \quad - \frac{1}{\st_{\ti}^4} \E \left[ W_{\ti,1} N_{\ti,1}^2  +  W_{\ti,2} N_{\ti,1}^2  +   W_{\ti,3} N_{\ti,1}^2   + 2 W_{\ti,2} N_{\ti,1} N_{\ti,2}   \right]  		\\
& \quad + \frac{1}{\st_{\ti}^6} \E\left[ W_{\ti,1}^2 N_{\ti,1}^2  +  W_{\ti,2}^2 N_{\ti,1}^2 \right]  +  \frac{1}{\st_{\ti}^2} \E\left[ B_{\ti,1}^2 \right] - \frac{1}{\st_{\ti}^4} \E\left[ W_{\ti,1} N_{\ti,1} B_{\ti,1} \right],
\end{align*}
\begin{align*}
\E[\ti^3] & \oeq  \frac{1}{ \st_{\ti}^3} \E\left[ N_{\ti,1}^3 \right]    -   \frac{3}{2\st_{\ti}^5} \E \left[ W_{\ti,1} N_{\ti,1}^3  \right]  	+   \frac{3}{\st_{\ti}^3} \E\left[  N_{\ti,1}^2 B_{\ti,1} \right],
\end{align*}
and
\begin{align*}
\E[\ti^4] & \oeq  \frac{1}{\st_{\ti}^4} \E\left[ N_{\ti,1}^4  +  4 N_{\ti,1}^3 N_{\ti,2}  +  4 N_{\ti,1}^3 N_{\ti,3}  + 6 N_{\ti,1}^2 N_{\ti,3}^2  \right]   		\\
& \quad - \frac{2}{\st_{\ti}^6} \E \left[ W_{\ti,1} N_{\ti,1}^4  +  W_{\ti,2} N_{\ti,1}^4  +   4 W_{\ti,2} N_{\ti,1}^3 N_{\ti,2}  + W_{\ti,3} N_{\ti,1}  \right]  		\\
& \quad +  \frac{3}{\st_{\ti}^8} \E \left[ W_{\ti,1}^2 N_{\ti,1}^4  +  W_{\ti,2}^2 N_{\ti,1}^4   \right]  		\\
& \quad + \frac{4}{\st_{\ti}^4} \E\left[ N_{\ti,1}^3  B_{\ti,1} \right]   -   \frac{8}{\st_{\ti}^6} \E\left[ W_{\ti,1} N_{\ti,1}^3 B_{\ti,1} \right]  +  \frac{6}{\st_{\ti}^4} \E\left[ N_{\ti,1}^2 B_{\ti,1}^2 \right].
\end{align*}
Computing each factor, we get the following results. For these terms below, indexes $i$, $j$, and $k$ are always distinct (i.e.\ $X_{\h,i} \neq X_{\h,j} \neq X_{\h,k}$). 
\begin{align*}
\E\left[B_{\ti,1}\right] & = \bti,  		\\
\E\left[ W_{\ti,1} N_{\ti,1}\right] & \oeq \tO^{-1} \E\left[ \h^{-1} \l_{\ti}^0(X_i)^3  \e_i^3 \right],   		\\
\E\left[  N_{\ti,1}^2 \right] & \oeq \st_{\ti}^2,   		\\
\E\left[  N_{\ti,1}N_{\ti,2} \right] & \oeq \tO^{-2}  \E\left[ \h^{-1} \l^1_{\ti}(X_i,X_i) \l^0_{\ti}(X_i) \e_i^2 \right],   		\\ 
\E\left[ N_{\ti,2}^2 \right] & \oeq \tO^{-1} \E\left[ \h^{-2} \l^1_{\ti}(X_i, X_j)^2 \e_i^2 \right],   		\\
\E\left[  N_{\ti,2}N_{\ti,3} \right] & \oeq \tO^{-2}  \E\left[ \h^{-2} \l_v^2(X_i, X_j, X_j) \l^0_{\ti}(X_i) \e_i^2 \right],   		\\ 
\E\left[  W_{\ti,1}N_{\ti,1}^2 \right] & \oeq \tO^{-2} \Biggl\{  \E\left[ \h^{-1} \l^0_{\ti}(X_i)^4 \left( \e_i^4 - v(X_i)^2\right) \right]   		\\
& \quad  -  2 \st_{\ti}^2 \E\left[ \h^{-1} \l^0_{\ti}(X_i)^2 \br_{\tilde{p}}(X_{d_n,i})'  \bm{\tilde{G}}^{-1} (K \br_{\tilde{p}})(X_{d_n,i}) \e_i^2 \right]   		\\ 
& \quad  -  4  \E\left[ \h^{-1} \l^0_{\ti}(X_i)^4 \br_{\tilde{p}}(X_{d_n,i})'  \bm{\tilde{G}}^{-1} \e_i^2  \right] \E\left[ \h^{-1}  (K \br_{\tilde{p}})(X_{d_n,i}) \l^0_{\ti}(X_i) \e_i^2 \right]   		\\ 
& \quad  + \st_{\ti}^2 \E\left[ \h^{-2} \l^0_{\ti}(X_i)^2 \left(  \br_{\tilde{p}}(X_{d_n,i})'  \bm{\tilde{G}}^{-1} (K \br_{\tilde{p}})(X_{d_n,j}) \right)^2 \e_j^2 \right]   		\\ 
& \quad +  2 \E \left[ \h^{-1} \l^0_{\ti} (X_j)^2 \left( \E \left[ \h^{-1} \br_{\tilde{p}}(X_{d_n,j})'  \bm{\tilde{G}}^{-1} (K \br_{\tilde{p}})(X_{d_n,i}) \l^0_{\ti}(X_i)  \e_i^2 \vert X_j \right] \right)^2 \right]    \Biggr\},   		\\
\E\left[  W_{\ti,2}N_{\ti,1}^2 \right] & \oeq \tO^{-2} \Bigl\{ \E \left[ \h^{-1} \left( \l^0_{\ti}(X_i)^2 v(X_i) - \E[\l^0_{\ti}(X_i)^2 v(X_i)] \right) \l^0_{\ti}(X_i)^2 \e_i^2 \right]     		\\
& \quad + 2 \st_{\ti}^2 \E \left[ \h^{-1} \l^1_{\ti}(X_i, X_i) \l^0_{\ti}(X_i) v(X_i) \right]  \Bigr\},   		\\
\E\left[  W_{\ti,2}N_{\ti,1}N_{\ti,2} \right] & \oeq \tO^{-2} \Bigl\{ \E \left[ \h^{-2} \left( \l^0_{\ti}(X_j)^2 v(X_j) - \E[\l^0_{\ti}(X_j)^2 v(X_j)] \right)  \l^1_{\ti}(X_i, X_j) \l^0_{\ti}(X_i) \e_i^2 \right]     		\\
& \quad + 2  \E \left[ \h^{-3} \l^1_{\ti}(X_i, X_j)  \l^1_{\ti}(X_k, X_j) \l^0_{\ti}(X_i) \l^0_{\ti}(X_k) v(X_i) \e_k^2 \right]  \Bigr\},   		\\
\E\left[  W_{\ti,3}N_{\ti,1}^2 \right] & \oeq \tO^{-2} \Bigl\{ \st_{\ti}^2 \E \left[ \h^{-2} \left( \l^1_{\ti}(X_i, X_j)^2  + 2\l^2_{\ti}(X_i, X_j, X_j) \right) v(X_i)  \right]  \Bigr\},   		\\
\E\left[  W_{\ti,1}^2 N_{\ti,1}^2 \right] & \oeq \tO^{-2} \Bigl\{ \st_{\ti}^2 \E \left[ \h^{-1}  \l^0_{\ti}(X_i)^4 \left( \e_i^4 - v(X_i)^2 \right) \right]  + 2 \E\left[ \h^{-1} \l^0_{\ti}(X_i)^3   \e_i^3  \right]^2  \Bigr\},   		\\
\E\left[  W_{\ti,2}^2 N_{\ti,1}^2 \right] & \oeq \tO^{-2} \st_{\ti}^2 \Bigl\{ \E \left[ \h^{-1}  \left( \l^0_{\ti}(X_i)^2 v(X_i)  -  \E[\l^0_{\ti}(X_i)^2 v(X_i)] \right)^2 \right]  		\\
& \quad  + 4 \E\left[ \h^{-2}  \left( \l^0_{\ti}(X_i)^2 v(X_i)  -  \E[\l^0_{\ti}(X_i)^2 v(X_i)] \right) \l^1_{\ti}(X_j, X_i) \l^0_{\ti}(X_j) v(X_j) \right]   		\\
& \quad  + 4 \E\left[ \h^{-3} \l^1_{\ti}(X_i, X_j) \l^0_{\ti}(X_i) v(X_i) \l^1_{\ti}(X_k, X_j) \l^0_{\ti}(X_k) v(X_k) \right] \Bigr\},   		\\
\E\left[  W_{\ti,1}N_{\ti,1}B_{\ti,1}\right]  & \oeq  \E\left[  W_{\ti,1}N_{\ti,1} \right] \E\left[B_{\ti,1}\right],   		\\
\E\left[  N_{\ti,1}^3 \right] & \oeq \tO^{-1} \E\left[ \h^{-1} \l^0_{\ti}(X_i)^3 \e_i^3 \right],   		\\
\E\left[  W_{\ti,1} N_{\ti,1}^3 \right] & \oeq \E\left[  N_{\ti,1}^2 \right]  \E\left[  W_{\ti,1} N_{\ti,1} \right],   		\\
\E\left[  N_{\ti,1}^4 \right] & \oeq 3 \st_{\ti}^4  +  \tO^{-2} \E\left[ \h^{-1} \l^0_{\ti}(X_i)^4 \e_i^3 \right],   		\\
\E\left[  N_{\ti,1}^3 N_{\ti,2} \right] & \oeq  \tO^{-2} 6 \st_{\ti}^2 \E\left[ \h^{-1} \l^1_{\ti}(X_i, X_i) \l^0_{\ti}(X_i) \e_i^2 \right],   		\\
\E\left[  N_{\ti,1}^3 N_{\ti,3} \right] & \oeq  \tO^{-2} 3 \st_{\ti}^2 \E\left[ \h^{-2} \l^2_{\ti}(X_i, X_j, X_j) \l^0_{\ti}(X_i) \e_i^2 \right],   		\\
\E\left[  N_{\ti,1}^2 N_{\ti,2}^2 \right] & \oeq  \tO^{-2} \Bigl\{ \st_{\ti}^2 \E\left[ \h^{-2} \l^1_{\ti}(X_i, X_j)^2 \e_i^2 \right]   +   2  \E \left[ \h^{-3} \l^1_{\ti}(X_i, X_j)  \l^1_{\ti}(X_k, X_j) \l^0_{\ti}(X_i) \l^0_{\ti}(X_k) \e_i^2 \e_k^2 \right]  \Bigr\},   		\\
\E\left[  W_{\ti,1} N_{\ti,1}^4 \right] & \oeq \tO^{-2} \Bigl\{ \E\left[ \h^{-1} \l^0_{\ti}(X_i)^3 \e_i^3 \right] \E\left[ \h^{-1} \l^0_{\ti}(X_i)^3 \e_i^3 \right]  +  6 \E\left[  N_{\ti,1}^2 \right]  \E\left[  W_{\ti,1} N_{\ti,1}^2 \right]  \Bigr\},   		\\
\E\left[  W_{\ti,2} N_{\ti,1}^4 \right] & \oeq \tO^{-2}\st_{\ti}^2  6 \Bigl\{ \E \left[ \h^{-1} \left( \l^0_{\ti}(X_i)^2 v(X_i) - \E[\l^0_{\ti}(X_i)^2 v(X_i)] \right) \l^0_{\ti}(X_i)^2 \e_i^2 \right]     		\\
& \quad + 2  \E \left[ \h^{-2} \l^1_{\ti}(X_i, X_j) \l^0_{\ti}(X_i) \l^0_{\ti}(X_j)^2 \e_j^2 v(X_i) \right]   +  \E\left[ \h^{-1} \l^1_{\ti}(X_i, X_i) \l^0_{\ti}(X_i) v(X_i) \right] \Bigr\},   		\\
\E\left[  W_{\ti,2} N_{\ti,1}^3 N_{\ti,2} \right] & \oeq 3 \E\left[   N_{\ti,1}^2 \right] \E\left[  W_{\ti,2} N_{\ti,1} N_{\ti,2} \right] ,  		\\
\E\left[  W_{\ti,3} N_{\ti,1}^4 \right] & \oeq 3 \E\left[   N_{\ti,1}^2 \right] \E\left[  W_{\ti,3} N_{\ti,1}^2 \right] ,  		\\
\E\left[  W_{\ti,1}^2 N_{\ti,1}^4 \right] & \oeq 3 \E\left[   N_{\ti,1}^2 \right] \E\left[  W_{\ti,1}^2 N_{\ti,1}^2 \right] ,  		\\
\E\left[  W_{\ti,2}^2 N_{\ti,1}^4 \right] & \oeq 3 \E\left[   N_{\ti,1}^2 \right] \E\left[  W_{\ti,2}^2 N_{\ti,1}^2 \right] .
	\end{align*}

The so-called approximate cumulants of $\ti$, denoted here by $\kappa_{\ti,k}$ for the $k^{\text{th}}$ cumulant, can now be directly calculated from these approximate moments using standard formulas \citep[Equation (2.6)]{Hall1992_book}. It is useful to list these and collect their asymptotic orders. For the first two, we split them into two subterms each, by their different asymptotic order.
\begin{align*}
	\kappa_{\ti,1} & = \E[\ti] := \kappa_{\ti,1,1} + \kappa_{\ti,1,2} \oeq \tO^{-1} + \bti  , 			\\
	\kappa_{\ti,2} & = \E[\ti^2] - \E[\ti]^2 := 1 + \kappa_{\ti,2,1} + \kappa_{\ti,2,2} \oeq 1 + \tO^{-2} + \tO^{-1}\bti ,  			\\
	\kappa_{\ti,3} & = \E[\ti^3] - 3 \E[\ti^2] \E[\ti] + 2 \E[\ti]^3 \oeq \tO^{-1} ,  			\\
	\kappa_{\ti,4} & = \E[\ti^4] - 4 \E[\ti^3] \E[\ti] - 3 \E[\ti^2]^2 + 12 \E[\ti^2]\E[\ti]^2 - 6 \E[\ti]^4 \oeq \tO^{-2}  .
\end{align*}
	
Next, our equivalent of \cite[Equation (2.22)]{Hall1992_book} would be the exponential of
\begin{align*}
	& \kappa_{\ti,1}(it) + \frac{1}{2}(it)^2 (\kappa_{\ti,2} - 1) + \frac{1}{3!}(it)^3 \kappa_{\ti,3} + \frac{1}{4!}(it)^4 \kappa_{\ti,4}   			\\
	& \quad  + \frac{1}{2}(it)^2 \left(\kappa_{\ti,1,1}^2 + 2\kappa_{\ti,1,1}\kappa_{\ti,1,2} \kappa_{\ti,1,2}^2 \right)  +  \frac{1}{2} \frac{1}{3!^2}(it)^6 \kappa_{\ti,3}^2    			\\
	& \quad + \frac{1}{2}2\frac{1}{3!} (it)(it)^3 \left(\kappa_{\ti,1,1} \kappa_{\ti,3}  +  \kappa_{\ti,1,2}\kappa_{\ti,3} \right).
\end{align*}

Then, the final computation is done by following \cite[p.\ 44f, Equations (2.17)]{Hall1992_book}. We find that the Edgeworth expansion, with asymptotic order listed in parentheses at right, is given by
\begin{align}
	\Phi(z) - \phi(z) \Bigg\{ &  \bigg[\kappa_{\ti,1,1} + \frac{1}{3!} (z^2 - 1) \kappa_{\ti,3} \bigg]   			\tag{$\tO^{-1}$}   \\
	 & \bigg[ \kappa_{\ti,1,2} \bigg]   			\tag{$\bti$}   \\
	 \begin{split}
		 & \bigg[ \frac{1}{2} z \kappa_{\ti,1,1}^2  +  \frac{1}{2}\frac{1}{3!^2} z (z^4 - 10z^2 + 15) \kappa_{\ti,3}^2    			 \\
		 & \qquad  + \frac{1}{2}2\frac{1}{3!} z(z^2-3) \kappa_{\ti,1,1}\kappa_{\ti,3}  +  \frac{1}{2} z \kappa_{\ti,2,1}  +  \frac{1}{4!} z(z^2-3)\kappa_{\ti,4} \bigg]   			
	 \end{split}	   \tag{$\tO^{-2}$}  	    \\
	 & \bigg[ \frac{1}{2} z \kappa_{\ti,1,2}^2 \bigg]   			\tag{$\bti^2$}   \\
	 & \bigg[\frac{1}{2} z 2 \kappa_{\ti,1,1} \kappa_{\ti,1,2}  + \frac{1}{2}2\frac{1}{3!} z(z^2-3) \kappa_{\ti,1,2}\kappa_{\ti,3}  + \frac{1}{2}z \kappa_{\ti,2,2} \bigg]  			\tag{$\tO^{-1}\bti$}   
	 \Bigg\}.
\end{align}
This is exactly the result of Theorem \ref{suppthm:EE lp} and these terms, in the order displayed, are exactly the $\w_k(\ti,z), k=1,2,3,4,5$ above.

%%%%%%%%%%%%%%%%%%%%%%%%%%%%%%%%%%%%%%%%%%%%%%%%%%%%%%%%%%%%%%%%%%%%%%
%%%%%%%%%%%%%%%%%%%%%%%%%%%%%%%%%%%%%%%%%%%%%%%%%%%%%%%%%%%%%%%%%%%%%%
\section{Bias and the Role of Smoothness}
	\label{supp:bias lp}

In this section we derive (and list) all the necessary bias terms, both in generic form and for special cases. We will cover different centerings, different smoothness cases, as well as interior and boundary points. We first give a generic derivation, followed by discussion of the bias of $\that = \mhat_{p+1}^{(\v)}$ and then $\thatrbc$, and in the final subsection, a complete list of all results and formulae.

The conditional bias defined above in Equation \eqref{suppeqn:bias lp}, and the similarly computed $\E[\thatrbc \big| X_1, \ldots, X_n ]$, are useful for describing bias correction, first order asymptotics, and computing and implementing optimal bandwidths. However, these can not be present in the Edgeworth and coverage error expansions because they are random quantities. Further, the leading term isolated in Equation \eqref{suppeqn:bias lp} presumes sufficient smoothness, which we avoid for general results. (The analogous calculation for $\thatrbc$ is shown below.)

The bias terms in the expansions are generic and nonrandom. In Theorem \ref{suppthm:EE lp} we denote the bias contribution by $\bti$. This term, and its particular cases $\btp$ and $\btrbc = \biasletter_{\trbc,\f}$ in particular, capture the entire bias, that is both the rate and the constant. These terms are defined both (i) before a Taylor approximation is performed, and (ii) with $\Gp$, $\Gq$, and $\Lp_{1}$ replaced with their fixed-$n$ expectations, denoted $\Gpt$, $\Gqt$, and $\Lpt_1$. In both sense, these bias terms reflect the ``fixed-$n$'' approach. (A tilde always denotes a fixed-$n$ expectation, and all expectations are fixed-$n$ calculations unless explicitly denoted otherwise.)  

For notation, we maintain the dependence on $\f$ if it is useful to emphasize that for certain $\f \in \F_\S$ the bias may be lower or higher. For example, if it happens that $\mu_\f^{(p+1)}(\x) = 0$, the leading term of Equation \eqref{suppeqn:bias lp} will be zero even if $p-\v$ is odd. Further, at present we explicitly write these as functions of the $t$-statistic, as the expansions in Section \ref{supp:theorems lp} are for the $t$-statistics, but it would be equivalent to write them as functions of the corresponding interval: that is $\bi \equiv \bti$, in terms of $\i$ and $\f$. For example, $\btrbc = \biasletter_{\trbc,\f} = \birbc$.

%%%%%%%%%%%%%%%%%%%%%%%%%%%%%%%%%%%%%%%%%%%%%%%%%%%%%%%%%%%%%%%%%%%%%%
\subsection{Generic Bias Formulas}

Define
\begin{itemize}

	\item $\bbeta_k$ (usually $k=p$ or $k=p+1$) as the $k+1$ vector with $(j+1)$ element equal to $\mu^{(j)}(\x)/j!$ for $j = 0, 1, \ldots, k$ as long as $j \leq \S$, and zero otherwise,
	
	\item $\bM = [\mu(X_1), \ldots, \mu(X_n)]'$,
		
	\item $\preTaylor_k$ as the $n$-vector with $i^{\text{th}}$ entry $[\mu(X_i) - \br_{k}(X_i - \x)'\bbeta_{k}]$,
	
	\item $\rho = \h / \b$, the ratio of the two bandwidth sequences, and
	
	\item $\Gpt = \E[\Gp]$, $\Gqt = \E[\Gq]$, $\Lpt_1 = \E[\Lp_{1}]$, and so forth. A tilde always denotes a fixed-$n$ expectation, and all expectations are fixed-$n$ calculations unless explicitly denoted otherwise. The dependence on $\f$ and $\F_\S$ is suppressed. As a concrete example:
		\[\Lp_k = \Op \left[ X_{\h,1}^{p+k}, \ldots, X_{\h,n}^{p+k} \right]'/n = \frac{1}{n\h} \sumi  (K \br_p)(\Xhi) \Xhi^{p+k},\]
	and so
		\begin{align*}
			\Lpt_k = \E[\Lp_k] & = \h^{-1} \E\left[ (K \br_p)(\Xhi) \Xhi^{p+k} \right]  		\\
			& = \h^{-1} \int_{\supp\{X\}} K\left(\frac{X_i - \x}{\h}\right) \br_p \left(\frac{X_i - \x}{\h}\right) \left(\frac{X_i - \x}{\h}\right)^{p+k} f(X_i) dX_i  		\\
			& =  \int_{-1}^1 K(u) \br_p(u) u^{p+k} f(\x + u\h) du.
		\end{align*}
		
	The range of integration for integrals will generally be left implicit. The range will change when the point of interest is on a boundary, but the notation will remain the same and it is to be understood that moments and moments of the kernel be replaced by the appropriate truncated version. For example, if $\supp\{X\} = [0,\infty)$ and the point of interest is $\x = 0$, then by a change of variables
\[\Lpt_k = \h^{-1} \int_{\supp\{X\}}(K r_p)(\Xhi) \Xhi^{p+k} f(X_i) dX_i = \int_0^\infty (K r_p)(u) u^{p+k} f(u\h)du,\]
whereas if $\supp\{X\} = (-\infty,0]$ and $\x = 0$, then 
\[\Lpt_k = \int_{-\infty}^0 (K r_p)(u) u^{p+k} f(-u\h)du.\]
For the remainder of this section, the notation is left generic.

\end{itemize}

To compute the terms $\btp$ and $\btrbc$, begin with the conditional mean of $\mhat_p^{(\v)}$:
\begin{align*}
	\E\left[\mhat_p^{(\v)} \big| X_1, \ldots, X_n \right] & = \v! \be_\v'\E\left[ \bhat_p  \big| X_1, \ldots, X_n \right] = \frac{1}{n \h^\v} \v! \be_\v'\Gp^{-1} \Op \bM  		\\
	& = \frac{1}{n \h^\v} \v! \be_\v'\Gp^{-1} \Op ( \bM - \bR\bbeta_p)   +   \frac{1}{n \h^\v} \v! \be_\v'\Gp^{-1} \Op \bR\bbeta_p  			\\
	& = \frac{1}{n \h^\v} \v! \be_\v'\Gp^{-1} \Op \preTaylor_p   +   \frac{1}{n \h^\v} \v! \be_\v'\Gp^{-1} \Op \bR\bbeta_p.
\end{align*}
Because $\h^{-\v} \be_\v' = \be_\v'\H^{-1}$,  $\bRc = \bR \H^{-1}$,  $\Op = \bRc' \bW$, and $\Gp = \bRc' \bW \bRc/n = \Op \bRc/n$, (the same calculations used for \eqref{suppeqn:bhat1} and \eqref{suppeqn:bhat2}) the second term above is
\begin{equation}
	\label{suppeqn:bias 1}
	\v! \left( \be_\v'\H^{-1} \right) \Gp^{-1} \left(\Op \bRc/n \right)\H \bbeta_p = \v! \be_\v' \bbeta_p = \mu^{(\v)}(\x),
\end{equation}
using the definition of $\bbeta_p$ (the $\v+1$ element of the vector $\bbeta_p$ will not be zero, as $\v \leq \S$ holds by Assumption \ref{suppasmpt:dgp lp}). Therefore 
\begin{align}
	\E\left[\mhat_p^{(\v)} \big| X_1, \ldots, X_n \right]  -  \mu^{(\v)}  & = \frac{1}{n \h^\v} \v! \be_\v'\Gp^{-1} \Op \preTaylor_p   			\nonumber \\
	& = \h^{-\v} \v! \be_\v'\Gp^{-1} \frac{1}{n\h} \sumi (K \br_p)(\Xhi) \left( \mu(X_i) - \br_p(X_i - \x)'\bbeta_p\right).  \label{suppeqn:bias 2}
\end{align}

From here, a Taylor expansion of $\mu(X_i)$ around $X = \x$ immediately gives Equation \eqref{suppeqn:bias lp}, provided that $S \geq p+1$. Instead, the bias terms of the Edgeworth expansions use this form directly, replacing the sample averages with population averages. The biases, $\bti$ in general and $\btp$ and $\btrbc$ in particular, must explicitly account for the rate scaling of $\sqrt{n\h^{1+2\v}}$, because the Edgeworth expansions are proven directly for the $t$-statistics.

For $\that = \mhat^{(\v)}$, for $\tp$ or $\ip$, we apply the rate scaling to the above display and then define
\[\btp  = \sqrt{n\h^{1+2\v}} \h^{-\v} \v! \be_\v'\Gpt^{-1} \E \left[ \h^{-1} (K \br_p)(\Xhi) \left( \mu(X_i) - \br_p(X_i - \x)'\bbeta_p\right) \right].\]
Note that the $\h^{-\v}$ cancels, and thus the rate of decay of the scaled bias does not depend on the level of derivative of interest. Because of the fixed-$n$ nature of this calculation, the parity of $p-\v$ does not matter. If a Taylor series were performed \emph{and} the matrixes were allowed to converge to their limit, the well-known symmetry cancellation would occur for $p-\v$ even at interior $\x$ \citep{Fan-Gijbels1996_book}. The generic expansions are stated without being explicit on this, but for certain derivations and specific cases the symmetry will be exploited. It holds that $\btp = O(\sqrt{n\h} \h^{\az})$ uniformly in $\F_\S$ where $\az$ varies depending on smoothness, parity of $p-\v$, and location of $\x$. If $p$ is small relative to $\S$, depending again on parity and location, we can isolate the leading term $\bctp$ such that $\btp = \sqrt{n\h} \h^\az \bctp [1+o(1)]$ where $\bctp = O(1)$ uniformly in $\F_\S$ and is nonzero for some $\f \in \F_\S$. Results for every case are given in Section \ref{supp:bias cases US} and summarized in Table \ref{supptable:us bias list}.

For $\thatrbc$ (i.e. for $\trbc$ and $\irbc$),
\begin{multline*}
	\E\left[\thatrbc  \big| X_1, \ldots, X_n \right]  -  \tf = \left\{ \E\left[\mhat^{(\v)} \big| X_1, \ldots, X_n \right]  -   \mu^{(\v)}   \right\}  		 \\  	
					 -   \left\{ \h^{p+1 - \v}  \v! \be_\v'\Gp^{-1} \Lp_1 \frac{1}{(p+1)!} \E\left[ \mhat^{(p+1)}\big| X_1, \ldots, X_n \right]   \right\}.
\end{multline*}
The first term is given exactly in \eqref{suppeqn:bias 2}. For the second term, following exactly the same steps that we used to arrive at \eqref{suppeqn:bias 2}, but with $(p+1)$ in place of $v$ and $p$ and $\b$ in place of $\h$, we find that 
\begin{multline*}
	\E\left[ \mhat^{(p+1)}\big| X_1, \ldots, X_n \right] =  (p+1)! \be_{p+1}' \bbeta_{p+1}    		 \\  	
					 +     \b^{-p-1} (p+1)! \be_{p+1}'\Gq^{-1} \frac{1}{n\b} \sumi (K \br_{p+1})(\Xbi) \left( \mu(X_i) - \br_{p+1}(X_i - \x)'\bbeta_{p+1}\right) 
\end{multline*}
Inserting this result and \eqref{suppeqn:bias 2} into $\E\left[\thatrbc  \big| X_1, \ldots, X_n \right]  -  \tf$, we find that
\begin{align}
	\E & \left[\thatrbc  \big| X_1, \ldots, X_n \right]  -  \tf   		\nonumber \\
	& = \h^{-\v} \v! \be_\v'\Gp^{-1} \frac{1}{n\h} \sumi (K \br_p)(\Xhi) \left( \mu(X_i) - \br_p(X_i - \x)'\bbeta_p\right)  		
	   		-   \h^{p+1 - \v}  \v! \be_\v'\Gp^{-1} \Lp_1 \frac{1}{(p+1)!} (p+1)! \be_{p+1}' \bbeta_{p+1}   		\nonumber \\
	& \quad   -   \h^{p+1 - \v}  \v! \be_\v'\Gp^{-1} \Lp_1 \frac{1}{(p+1)!} \b^{-p-1} (p+1)! \be_{p+1}'\Gq^{-1}   		
		 \times  \frac{1}{n\b} \sumi (K \br_{p+1})(\Xbi) \left( \mu(X_i) - \br_{p+1}(X_i - \x)'\bbeta_{p+1}\right)       		\nonumber \\
	& = \h^{-\v} \v! \be_\v'\Gp^{-1} \frac{1}{n\h} \sumi (K \br_p)(\Xhi) \left( \mu(X_i) - \br_p(X_i - \x)'\bbeta_p\right)  		
		   -   \h^{p+1 - \v}  \v! \be_\v'\Gp^{-1} \Lp_1 \be_{p+1}' \bbeta_{p+1}   		\nonumber \\
	& \quad   -   \h^{-\v}  \rho^{p+1} \v! \be_\v'\Gp^{-1} \Lp_1 \be_{p+1}'\Gq^{-1}   		  
		\times  \frac{1}{n\b} \sumi (K \br_{p+1})(\Xbi) \left( \mu(X_i) - \br_{p+1}(X_i - \x)'\bbeta_{p+1}\right)       		\nonumber \\
	& = \h^{-\v} \v! \be_\v'\Gp^{-1} \frac{1}{n\h} \sumi (K \br_p)(\Xhi) \left( \mu(X_i) - \br_{p+1}(X_i - \x)'\bbeta_{p+1}\right)  		\nonumber \\
	& \quad   -   \h^{-\v}  \rho^{p+1} \v! \be_\v'\Gp^{-1} \Lp_1 \be_{p+1}'\Gq^{-1}   		  \times  \frac{1}{n\b} \sumi (K \br_{p+1})(\Xbi) \left( \mu(X_i) - \br_{p+1}(X_i - \x)'\bbeta_{p+1}\right)      . \label{suppeqn:rbc bias pre-taylor}
\end{align}
where the last equality combines the first two terms (in the penultimate line), by noticing that
\begin{align*}
	\h^{p+1 - \v}  \v! \be_\v'\Gp^{-1} \Lp_1 \be_{p+1}' \bbeta_{p+1}  &   =  \h^{p+1 - \v}  \v! \be_\v'\Gp^{-1} \frac{1}{n\h} \sumi (K \br_p)(\Xhi) (\Xhi)^{p+1}  \be_{p+1}' \bbeta_{p+1}  		\\
	& = \h^{p+1 - \v}  \v! \be_\v'\Gp^{-1} \frac{1}{n\h} \sumi (K \br_p)(\Xhi) \h^{-p-1} (X_i - \x)^{p+1}  \be_{p+1}' \bbeta_{p+1},
\end{align*}
and that $(X_i - \x)^{p+1}  \be_{p+1}' \bbeta_{p+1}$ is exactly the difference between $\br_p(X_i - \x)'\bbeta_p$ and $\br_{p+1}(X_i - \x)'\bbeta_{p+1}$.

As before, $\btrbc$ is now defined replacing sample averages with population averages and applying the scaling of $\sqrt{n\h^{1 + 2\v}}$ from the $t$-statistic. Again the $\h^{-\v}$ cancels, and thus the rate of decay of the scaled bias does not depend on the level of derivative of interest.

In sum, the generic formulas are
\begin{align}
	\begin{split}
		\label{suppeqn:eta lp}
		\btp & = \sqrt{n\h} \; \v! \be_\v'\Gpt^{-1} \E \left[ \h^{-1} (K \br_p)(\Xhi) \left( \mu(X_i) - \br_p(X_i - \x)'\bbeta_p\right) \right],  		\\
		\btrbc & =  \sqrt{n\h} \;  \v! \be_\v'\Gpt^{-1} \E\biggl[ \Bigl\{ \h^{-1} (K \br_p)(\Xhi)    -    \rho^{p+1} \Lpt_1 \be_{p+1}'\Gqt^{-1} \b^{-1} (K \br_{p+1})(\Xbi) \Bigr\}   		\\
					& \qquad \qquad \qquad\qquad\quad \times   \left( \mu(X_i) - \br_{p+1}(X_i - \x)'\bbeta_{p+1}\right)   \biggr]  	
	\end{split}
\end{align}
or using $\Op$ and $\Orbc$ as in \Eqref{suppeqn:that lp}, and $\preTaylor_k$,
\[
	\btp  = \sqrt{n\h} \;  \v! \be_\v'\Gpt^{-1}  \E[\Op \preTaylor_p]
\]
and
\[
	\btrbc  =  \sqrt{n\h} \;  \v! \be_\v'\Gpt^{-1} \left( \E[\Op \preTaylor_{p+1}] - \rho^{p+1}  \Lpt \be_{p+1}' \Gqt^{-1} \E[\Oq \preTaylor_{p+1}] \right).
\]

For the generic results of coverage error or the generic Edgeworth expansions of Theorem \ref{suppthm:EE lp} below, these definitions are suitable and the $\btp$ and $\btrbc$ may appear directly. For $\tp$, parity of $p-\v$ is not used, but can matter: the rate at which $\btp$ vanishes is faster by one factor of $\h$ at interior points \citep{Fan-Gijbels1996_book}. The validity of the Edgeworth expansions is not affected by this; the statements are seamless. 

However, it is also useful to separate the rate and leading constant term of these biases when possible. When it is possible we will isolate both the rate and the constant term of the bias. It holds that $\btp = O(\sqrt{n\h} \h^{\az})$ uniformly in $\F_\S$ and if $p$ is small relative to $\S$, depending again on parity and location, we can isolate the leading term $\bctp$ such that $\btp = \sqrt{n\h} \h^\az \bctp [1+o(1)]$ where $\bctp = O(1)$ uniformly in $\F_\S$ and is nonzero for some $\f \in \F_\S$. Similarly, it is always possible to show that $\btrbc = O(\sqrt{n\h}\; t(\h,\b))$ for a function $t(\cdot,\cdot)$ and further, if $\rho = \h/\b$ is bounded and bounded away from zero then $t(\cdot,\cdot)$ can be simplified to $\h^{\az}$. If $p$ is small relative to $\S$ we can isolate the leading terms via a Taylor expansion. If $p$ is small and $\rho$ is bounded and bounded away from zero, we can write $\btrbc = \sqrt{n\h} \h^\az \bctrbc [1+o(1)]$. 

For both $\btp$ and $\btrbc$, $\az$, $t(\h,\b)$, $\bctp$ and $\bctrbc$ depend on smoothness, parity of $p-\v$, and location of $\x$. Complete derivations for $\btp$ and $\btrbc$ are given in Sections \ref{supp:bias cases US} and \ref{supp:bias cases RBC} below and both are summarized in Tables \ref{supptable:us bias list} and \ref{supptable:rbc bias list} for lists of all cases.

The starting point of the derivations is a Taylor approximation. Recall the definitions of $\br_p(u)$ and $\bbeta_p$, where in particular elements of the latter beyond $S+1$ are zero. A Taylor approximation, for some $\bar{x}$, gives
\begin{align}
	\mu(X_i) - \br_p(X_i - \x)'\bbeta_p & = \sum_{k=0}^\S \frac{1}{k!}  (X_i - \x)^k \mu^{(k)}(\x) + \frac{1}{S\!}  (X_i - \x)^\S \left( \mu^{(\S)}(\bar{x}) - \mu^{(\S)}(\x) \right)  		\notag \\
		& \quad - \sum_{k=0}^{\S \wedge p} \frac{1}{k!}  (X_i - \x)^k \mu^{(k)}(\x)  		\notag \\
	& = \sum_{k= \S \wedge p + 1}^S \frac{1}{k!}  (X_i - \x)^k \mu^{(k)}(\x)   +   \frac{1}{\S!}  (X_i - \x)^\S \left( \mu^{(\S)}(\bar{x}) - \mu^{(\S)}(\x) \right)  		\notag \\
	& = \sum_{k= \S \wedge p + 1}^\S \frac{\h^k}{k!}  (\Xhi)^k \mu^{(k)}(\x)   +   O(\h^{\S+\s}),  		\label{suppeqn:Taylor lp}
\end{align}
where the first summation in the last two lines is taken to be zero if $p\geq \S$, and we have applied Assumption \ref{suppasmpt:dgp lp} and restricted to $X_i \in [\x \pm \h]$ (i.e.\ $K(\Xhi) > 0$). Note that by assumption the order of the remainder, $O(\h^{S+s})$, holds uniformly in $\F_\S$. We will use this expansion repeatedly below, or analogous results for other bandwidths and polynomial degrees.

%%%%%%%%%%%%%%%%%%%%%%%%%%%%%%%%%%%%%%%%%%%%%%%%%%%%%%%%%%%%%%%%%%%%%%
\subsection{No Bias Correction: Specific Cases and Leading Terms}
	\label{supp:bias cases US}

We now turn to specific cases for $\btp$. We will characterize the rate and leading constant terms in all cases, depending on depending on the relationship of $p$ and $\S$, the parity of $p - \v$, and whether $\x$ is an interior point or on the boundary. Note that here, unlike Equation \eqref{suppeqn:bias lp}, we are working with nonrandom quantities. The general case, from Equation \eqref{suppeqn:eta lp}, which appears in the Edgeworth expansion is
\[
	\btp  = \sqrt{n\h} \;  \v! \be_\v'\Gpt^{-1}  \E[\Op \preTaylor_p] = \sqrt{n\h} \; \v! \be_\v'\Gpt^{-1} \E \left[ \h^{-1} (K \br_p)(\Xhi) \left( \mu(X_i) - \br_p(X_i - \x)'\bbeta_p\right) \right].
\]
It is always true that the rate is captured by the exponent $\az$ in the form
\[
	\btp = O(\sqrt{n\h} \h^\az).
\]
If $p$ is small enough relative to $\S$, then we write
\[
	\btp = \sqrt{n\h} \h^\az \bctp [1+o(1)]
\]
and call $\bctp$ the leading constant. Recall that $\bctp$ is not truly constant, but rather a nonrandom sequence that is $O(1)$ uniformly in $\F_\S$ and is nonzero for some $\f \in \F_\S$. Table \ref{supptable:us bias list} is complete list of the results, including $\az$ and $\bctp$. These cases are derived in the rest of this section.

As an aside, it is technically possible to obtain the representation $\btp = \sqrt{n\h} \h^\az \bctp [1+o(1)]$ in general, that is for any $p$, by letting $\bctp$ to capture the final term in the Taylor expansion, $(X_i - \x)^\S [ \mu^{(\S)}(\bar{x}) - \mu^{(\S)}(\x)] /\S!$, see the penultimate step of Equation \eqref{suppeqn:Taylor lp}, and taking the $o(1)$ term to be exactly zero. However, we do not use $\bctp$ in this case because the representation is not useful for practice nor is it more concrete than simply using $\btp$, since in this case $\btp = \sqrt{n\h} \h^\az \bctp $ amounts to little more than a redefinition of notation.

\begin{table}
	\renewcommand{\arraystretch}{1.75}
	\centering
	\begin{tabular}{| l | l | l | l | l |}
		\hline
		Location of $\x$  & Parity of $p\!-\!\v$ & Smoothness & Rate Exponent $\az$ & $\bctp$  \\
		\hline 
		\multirow{2}{*}{Boundary} 
			& \multirow{2}{*}{odd or even} & $p < \S$ & $p+1$ & $\v! \be_\v' \Gpt^{-1} \Lpt_1 \frac{ \mu^{(p+1)} } { (p+1)! }$ \\
				&  & $p \geq \S$ & $\S+\s$ &  N/A \\ 
		\hline
		\multirow{4}{*}{Interior} 
			& \multirow{2}{*}{odd} & $p < \S$ & $p+1$ & $\v! \be_\v' \Gpt^{-1} \Lpt_1 \frac{ \mu^{(p+1)} } { (p+1)!}$ \\
				&  & $p \geq \S$ & $\S+\s$ &  N/A \\ 
			\cline{2-5}
			& \multirow{2}{*}{even} & $p + 2 \leq  \S$ & $p+2$ & $\v! \be_\v' \Gpt^{-1} \left( \h^{-1} \Lpt_1 \frac{ \mu^{(p+1)} } { (p+1)! }  + \Lpt_2 \frac{ \mu^{(p+2)} } { (p+2)! }  \right)$ \\
				&  & $p + 2 > \S$ & $\S+\s$ &  N/A \\ 
		\hline
	\end{tabular}
	\caption{Summary of Bias Terms in All Cases For Uncorrected Centering $\mhat_p^{(\v)}$. Rate exponent $\az$ is such that $\btp = O(\sqrt{n\h} \h^\az)$. When possible, $\bctp$ is such that $\btp = \sqrt{n\h} \h^\az \bctp [1+o(1)]$.}
	\label{supptable:us bias list}
\end{table}

%%%%%%%%%%%%%%%%%%%%%%%%%%%%%%%%%%%%%%%%%%%%%%%%%%%%%%%%%%%%%%%%%%%%%%
\subsubsection{Boundary Point}

Here parity plays no role. 

{\bf Case 1: $\bm{p < \S}$.} The leading bias term can be characterized, and we find (cf.\ Equation \eqref{suppeqn:bias lp})
\[\btp  = \sqrt{n\h^{1+2\v}} \h^{-\v} \h^{p+1} \frac{ \mu^{(p+1)} } { (p+1)! } \v! \be_\v' \Gpt^{-1} \Lpt_1 \left[1 + o(1) \right].\]
Note that this holds regardless of whether $\x$ is an interior or boundary point, with suitable changes to the ranges of integration in $\Gpt$ and $\Lpt_1$.

{\bf Case 2: $\bm{p \geq \S}$.} All that is left in Equation \eqref{suppeqn:Taylor lp} is this remainder term, and we therefore have
	\[\btp  = \sqrt{n\h^{1+2\v}} \h^{-\v} O(\h^{\S+\s})  =   O(\sqrt{n\h}\h^{\S+\s}),\]
and cannot say anything further regarding constants. This result applies any time $p\geq \S$, regardless of $\v$, parity of $p - \v$, and at interior and boundary points.

%%%%%%%%%%%%%%%%%%%%%%%%%%%%%%%%%%%%%%%%%%%%%%%%%%%%%%%%%%%%%%%%%%%%%%
\subsubsection{Interior Point: $p-\v$ odd}

The results for $p-\v$ odd are identical to the boundary point case. This automatic boundary carpentry is discussed briefly in the main text. It is one of the celebrated features of local polynomial regression, known for point estimation since their inception, see \cite{Fan-Gijbels1996_book} for review, and proven for inference for the first time in \cite{Calonico-Cattaneo-Farrell2018_JASA}.

\bigskip\noindent{\bf Case 1: $\bm{p < \S}$.} The leading bias term can be characterized, and we find (cf.\ Equation \eqref{suppeqn:bias lp})
\[\btp  = \sqrt{n\h^{1+2\v}} \h^{-\v} \h^{p+1} \frac{ \mu^{(p+1)} } { (p+1)! } \v! \be_\v' \Gpt^{-1} \Lpt_1 \left[1 + o(1) \right].\]
Note that this holds regardless of whether $\x$ is an interior or boundary point, with suitable changes to the ranges of integration in $\Gpt$ and $\Lpt_1$.

\bigskip\noindent{\bf Case 2: $\bm{p \geq \S}$.} All that is left in Equation \eqref{suppeqn:Taylor lp} is this remainder term, and we therefore have
	\[\btp  = \sqrt{n\h^{1+2\v}} \h^{-\v} O(\h^{\S+\s})  =   O(\sqrt{n\h}\h^{\S+\s}),\]
and cannot say anything further regarding constants. This result applies any time $p\geq \S$, regardless of $\v$, parity of $p - \v$, and at interior and boundary points.

%%%%%%%%%%%%%%%%%%%%%%%%%%%%%%%%%%%%%%%%%%%%%%%%%%%%%%%%%%%%%%%%%%%%%%
\subsubsection{Interior Point: $p-\v$ even}

Here the parity of $p$ will matter. It is worth spelling out three smoothness cases, though we will find the same result for the latter two.

\bigskip\noindent{\bf Case 1: $\bm{p + 2 \leq \S}$.} We begin by retaining \emph{two} terms of Equation \eqref{suppeqn:Taylor lp}:
\[\btp  = \sqrt{n\h^{1+2\v}} \h^{-\v} \h^{p+1} \v! \be_\v' \Gpt^{-1} \left( \Lpt_1 \frac{ \mu^{(p+1)} } { (p+1)! }  + \h \Lpt_2 \frac{ \mu^{(p+2)} } { (p+2)! }  \right) \left[1 + o(1) \right].\]
To find the leading term, we must appeal to the limits of (the fixed-$n$) expectations $\Gpt^{-1}$ and $\Lpt_k$ where it holds that
\begin{equation}
	\label{suppeqn:even}
	\be_\v' \Gpt^{-1}  \Lpt_k = A + \h B + o(\h),   \text{ with  $A = 0$  if $(p+k-\v)$ is odd and $\x$ is in the interior.}
\end{equation}
Note that at present we use this fact with $k=1$, and hence $(p+k-\v)$ is odd if $p-\v$ is even, the more common way of referring to this cancellation. Rather than derive the precise form of $A$ and $B$ in \eqref{suppeqn:even}, we maintain the fixed-$n$ approach by stabilizing $\be_\v' \Gpt^{-1}  \Lpt_k$ for interior points when needed. This has the dual the advantages of easy implementability (using the sample, non-tilde versions) and capturing all terms. We will thus write
\[\btp  = \sqrt{n\h^{1+2\v}} \h^{-\v} \h^{p+2} \v! \be_\v' \Gpt^{-1} \left( \h^{-1} \Lpt_1 \frac{ \mu^{(p+1)} } { (p+1)! }  + \Lpt_2 \frac{ \mu^{(p+2)} } { (p+2)! }  \right) \left[1 + o(1) \right].\]

\bigskip\noindent{\bf Case 2: $\bm{p + 1 = \S}$.} We can no longer retain the second term above, because $\mu^{(p+2)}$ does not exist. Instead we find that 
\[\btp  = \sqrt{n\h^{1+2\v}} \h^{-\v} \h^{p+1} \v! \be_\v' \Gpt^{-1} \left( \Lpt_1 \frac{ \mu^{(p+1)} } { (p+1)! }  + O(\h^\s)  \right) \left[1 + o(1) \right].\]
The same symmetry still applies to the first term however, and thus we have
\[\btp  = \sqrt{n\h^{1+2\v}} \h^{-\v} \h^{p+1+\s} \v! \be_\v' \Gpt^{-1} \left( \h^{1-\s} \h^{-1} \Lpt_1 \frac{ \mu^{(p+1)} } { (p+1)! }  +  O(1)  \right) \left[1 + o(1) \right],\]
but since $\s \leq 1$, the second term is (part of) the leading form, and we therefore write
\[\btp  = \sqrt{n\h^{1+2\v}} \h^{-\v} O(\h^{p+1+\s})   = \sqrt{n\h^{1+2\v}} \h^{-\v} O(\h^{\S+\s}),\]
with the final equality holding because, by assumption, $p + 1 = \S$ in this case.

\bigskip\noindent{\bf Case 3: $\bm{p \geq \S}$.} All that is left in Equation \eqref{suppeqn:Taylor lp} is this remainder term, and we therefore have
	\[\btp  = \sqrt{n\h^{1+2\v}} \h^{-\v} O(\h^{\S+\s})  =   O(\sqrt{n\h}\h^{\S+\s}),\]
and cannot say anything further regarding constants. This result applies any time $p\geq \S$, regardless of $\v$, parity of $p - \v$, and at interior and boundary points.

%%%%%%%%%%%%%%%%%%%%%%%%%%%%%%%%%%%%%%%%%%%%%%%%%%%%%%%%%%%%%%%%%%%%%%
\subsection{Post Bias Correction: Specific Cases and Leading Terms}
	\label{supp:bias cases RBC}

The general case, from Equation \eqref{suppeqn:eta lp}, which appears in the Edgeworth expansion is
\begin{align*}
	\btrbc  &  =  \sqrt{n\h} \;  \v! \be_\v'\Gpt^{-1} \bigg( \E\left[\Op \preTaylor_{p+1}\right] - \rho^{p+1}  \Lpt \be_{p+1}' \Gqt^{-1} \E\left[\Oq \preTaylor_{p+1}\right] \bigg)   			\\
	& =  \sqrt{n\h} \;  \v! \be_\v'\Gpt^{-1} \E\biggl[ \Bigl\{ \h^{-1} (K \br_p)(\Xhi)    -    \rho^{p+1} \Lpt_1 \be_{p+1}'\Gqt^{-1} \b^{-1} (K \br_{p+1})(\Xbi) \Bigr\}   		\\
		& \qquad \qquad \qquad\qquad\quad \times   \left( \mu(X_i) - \br_{p+1}(X_i - \x)'\bbeta_{p+1}\right)   \biggr] . 	
\end{align*}
It is always true that the rate is captured by a function $t(\cdot,\cdot)$ such that
\[
	\btrbc = O(\sqrt{n\h}\; t(\h,\b)),
\]
or if $\rho$ is bounded and bounded away from zero, the rate is captured by the exponent $\az$ such that
\[
	\btrbc = O(\sqrt{n\h} \h^\az).
\]
Additionally, if $p$ is small enough relative to $\S$, then we write
\[
	\btrbc = \sqrt{n\h} \h^\az \bctrbc [1+o(1)],
\]
and call $\bctrbc$ the leading constant. Recall that $\bctrbc$ is not truly constant, but rather a nonrandom sequence that is $O(1)$ uniformly in $\F_\S$ and is nonzero for some $\f \in \F_\S$. Table \ref{supptable:rbc bias list} is complete list of the results, including  $t(\h,\b))$, and where possible, $\az$ and $\bctp$. These cases are derived in the rest of this section.

\begin{table}
    \renewcommand{\arraystretch}{1.75}
    \centering    
    \begin{tabular}{| l | l | l | l | l | l |}
    	\hline
    	& & & & \multicolumn{2}{|c|}{$\rho$ bounded above 0, below $\infty$}   \\
    	\cline{5-6}
    	Location of $\x$  & Parity of $p\!-\!\v$ & Smoothness & Rate $t(\h,\b)$   &  $\az$ &  $\bctrbc$  \\
    	\hline 
    	\multirow{2}{*}{Boundary} 
    		& \multirow{2}{*}{odd or even} & $p + 2 \leq \S$ & $\h^{p+2}(1+\rho^{-1})$   & $p+2$ & \eqref{supptable:eqn bnd} \\
    			&  & $p + 2 > \S$ & $\h^{\S + \s}[1 + \rho^{p+1 - \S - \s}]$ &  $\S+\s$ &  N/A \\ 
    	\hline
    	\multirow{4}{*}{Interior} 
    		& \multirow{2}{*}{even} & $p + 2 \leq  \S$ & $\h^{p+2}$   & $p+2$ & \eqref{supptable:eqn int even} \\
    			&  & $p + 2 > \S$ & $ \h^{\S + \s}   \left[  1 +   \rho^{p+1 - \S - \s}  \right] $   & $\S+\s$ &  N/A \\ 
    			\cline{2-6} 
    		& \multirow{3}{*}{odd} & $p + 3 \leq \S$ & $\h^{p+3}(1+\rho^{-2})$ &  $p+3$ &  \eqref{supptable:eqn int odd} \\
    			&  & $p + 2 = \S$ & $ \h^{p + 2 + \s}[ 1 + \rho^{-1 - \s}]$ &  $p \!+\! 2 \!+\! \s = \S \!+\! \s$ &  N/A \\ 
    			&  & $p + 2 > \S$ & $\h^{\S + \s}   \left[  1 +   \rho^{p+1 - \S - \s}  \right]$  & $\S+\s$ &  N/A \\
    	\hline
    \end{tabular}
   	\caption{Summary of Bias Terms in All Cases For Bias-Corrected Centering $\thatrbc$. Rate function $t(\h,\b)$ is such that $\btrbc = O(\sqrt{n\h}\; t(\h,\b))$. If $\rho$ is bounded and bounded away from zero then we can take $t(\h,\b) = \h^\az$. When possible, $\bctp$ is such that $\btrbc = \sqrt{n\h} \h^\az \bctrbc [1+o(1)]$.}
	\label{supptable:rbc bias list}
\end{table}

\begin{subnumcases}{\hspace{-0.6in}\bctrbc \text{ in Table \ref{supptable:rbc bias list} can be}}
    \frac{ \mu^{(p+2)} } { (p+2)! } \v! \be_\v'\Gpt^{-1} \Big\{ \Lpt_2     -    \rho^{-1} \Lpt_1 \be_{p+1}'\Gqt^{-1} \Lqt_1  \Big\},
           \label{supptable:eqn bnd} \\
    \frac{ \mu^{(p+2)} } { (p+2)! }  \v! \be_\v'\Gpt^{-1} \Lpt_2, \quad \text{or}
           \label{supptable:eqn int even} \\
    \begin{split}
        \v! \be_\v'\Gpt^{-1}  \bigg\{  \frac{ \mu^{(p+2)} } { (p+2)! } \Big[ \h^{-1} \Lpt_2     -    \rho^{-2} \b^{-1} \Lpt_1 \be_{p+1}'\Gqt^{-1} \Lqt_1  \Big] 
        \\
         +    \frac{ \mu^{(p+3)} } { (p+3)! }  \Big[ \Lpt_3     -    \rho^{-2} \Lpt_1 \be_{p+1}'\Gqt^{-1} \Lqt_2  \Big] \bigg\}  
    \end{split}
        \label{supptable:eqn int odd},
\end{subnumcases}

The starting point of all the derivations is again a Taylor approximation. We use Equation \eqref{suppeqn:Taylor lp} with different choices for the bandwidth and polynomial degree. It will be useful at times to consider the two terms of $\bctrbc$ in Equation \eqref{suppeqn:eta lp} separately, as the bandwidths $\h$ and $\b$ may be different and even vanish at different rates. The two terms represent (i) the second bias term of $\mhat_p^{(\v)}$, not targeted by bias correction, and (ii) the bias of the bias estimator. For discussion in the context of kernel-based density estimation, see \cite{Hall1992_AoS_density} and \cite{Calonico-Cattaneo-Farrell2018_JASA,Calonico-Cattaneo-Farrell2018_JASAsupp}. See the latter also for bias correction using a generic polynomial of degree $q \geq p+1$; here we maintain degree $p+1$ for bias correction throughout.

The two terms of $\bctrbc$ in Equation \eqref{suppeqn:eta lp} are separated appropriately in Equation \eqref{suppeqn:rbc bias pre-taylor}. We will resume there and apply the Taylor expansion Equation \eqref{suppeqn:Taylor lp} with $p+1$ in place of $p$ and, for the second term of  \eqref{suppeqn:rbc bias pre-taylor}, also with $\b$ in place of $\h$. Doing this, assuming for the present sufficient smoothness, and applying the definitions of $\Lp_k$ and $\Lq_k$ and their respective fixed-$n$ expectations, we have,
\begin{align}
	\E & \left[\thatrbc  \big| X_1, \ldots, X_n \right]  -  \tf   		\nonumber \\
	& = \h^{-\v} \v! \be_\v'\Gp^{-1} \frac{1}{n\h} \sumi (K \br_p)(\Xhi) \left( \mu(X_i) - \br_{p+1}(X_i - \x)'\bbeta_{p+1}\right)  		\nonumber \\
	& \quad   -   \h^{-\v}  \rho^{p+1} \v! \be_\v'\Gp^{-1} \Lp_1 \be_{p+1}'\Gq^{-1}   		  \times  \frac{1}{n\b} \sumi (K \br_{p+1})(\Xbi) \left( \mu(X_i) - \br_{p+1}(X_i - \x)'\bbeta_{p+1}\right)   			\nonumber \\
	& = \h^{-\v} \v! \be_\v'\Gp^{-1} \left( \h^{p+2} \Lp_2 \frac{ \mu^{(p+2)} } { (p+2)! } + \h^{p+3} \Lp_3 \frac{ \mu^{(p+3)} } { (p+3)! } \right) [1 + o_\P(1)]
	  		\nonumber \\
	& \quad   -   \h^{-\v}  \rho^{p+1} \v! \be_\v'\Gp^{-1} \Lp_1 \be_{p+1}'\Gq^{-1}   	\left( \b^{p+2} \Lq_1 \frac{ \mu^{(p+2)} } { (p+2)! } + \b^{p+3} \Lq_2 \frac{ \mu^{(p+3)} } { (p+3)! }  \right) [1 + o_\P(1)].	    			\nonumber \\
\intertext{Collecting terms and replacing sample averages with expectations, we arrive at}
	\begin{split}
		\label{suppeqn:rbc bias post-taylor}
		& = \h^{p+2-\v} \v! \be_\v'\Gpt^{-1} 
		\Biggl\{ 
			\frac{ \mu^{(p+2)} } { (p+2)! } \left( \Lpt_2 -   \rho^{-1}  \Lpt_1 \be_{p+1}'\Gqt^{-1}  \Lqt_1  \right)
			\\
			& \qquad\qquad\qquad\qquad\qquad + \frac{ \mu^{(p+3)} } { (p+3)! } \left(\h \Lpt_3   
			-   \rho^{-1}  \b \Lpt_1 \be_{p+1}'\Gqt^{-1} \Lqt_2  \right) 	    			
		\Biggr\}	[1 + o_\P(1)]
	\end{split}
\end{align}
This final form will serve as the starting point for the special cases that follow.

%%%%%%%%%%%%%%%%%%%%%%%%%%%%%%%%%%%%%%%%%%%%%%%%%%%%%%%%%%%%%%%%%%%%%%
\subsubsection{Boundary Point}

Here parity does not matter. Therefore we need only the first term of \eqref{suppeqn:rbc bias post-taylor}, containing $\mu^{(p+2)}$. It matters only if there is sufficient smoothness.

\bigskip\noindent{\bf Case 1: $\bm{p + 2 \leq \S}$.} The first term of \eqref{suppeqn:rbc bias post-taylor} exists and dominates others if they exist, and so
\[
	\btrbc  =   \sqrt{n\h^{1+2\v}} \h^{-\v} \h^{p+2} \frac{ \mu^{(p+2)} } { (p+2)! } \v! \be_\v'\Gpt^{-1} \Big\{ \Lpt_2     -    \rho^{-1} \Lpt_1 \be_{p+1}'\Gqt^{-1} \Lqt_1  \Big\}  \left[1 + o(1) \right].
\]

\bigskip\noindent{\bf Case 2: $\bm{p + 2 > \S}$.} In this case $\mu^{(p+2)}$ does not exist, and therefore 
\begin{align*}
	\btrbc & =\sqrt{n\h^{1+2\v}} \h^{-\v} \left( O(\h^{\S+\s}) + \rho^{p+1} O(\b^{\S+\s}) \right)   		\\
	& = O\left( \sqrt{nh} \h^{\S + \s}[1 + \rho^{p+1 - \S - \s}]\right).
\end{align*}
The final rate depends on $p$ and $\rho$ in three cases: (i) if $\rho$ is bounded and bounded away from zero, then $\rho^{p+1 - \S - \s} \asymp 1$ and $\btrbc = O\left( \sqrt{nh} \h^{\S + \s}\right)$; (ii) the same rate is obtained if $\rho = o(1)$ and $p + 1 > \S$, because, since $p\geq\S$ and $1\geq\s$, the exponent on $\rho$ is positive and, with $\rho$ bounded, $\btrbc = O\left( \sqrt{nh} \h^{\S + \s}\right)$; (iii) if $\rho = o(1)$ and $p + 1 = \S$, then the second term is $\rho^{-\s} \to \infty$, thus $\btrbc = O\left( \sqrt{nh} \h^{\S + \s} \rho^{-\s}\right)$.

%%%%%%%%%%%%%%%%%%%%%%%%%%%%%%%%%%%%%%%%%%%%%%%%%%%%%%%%%%%%%%%%%%%%%%
\subsubsection{Interior Point: $p-\v$ odd}

Cancellations due to symmetry will occur here as well, even though the initial centering uses $p-\v$ odd, because bias correction involves $p+1-\v$, which is even. Again we will have three smoothness cases, though we will find the same result for the latter two.

The analogue of Equation \eqref{suppeqn:even} for the bias correction is
\begin{equation}
	\label{suppeqn:even rbc}
	\be_\v' \Gqt^{-1} \Lqt_k = \bar{A} + \b \bar{B} + o(\b),   \text{ with  $\bar{A} = 0$  if $(p+1+k-\v)$ is odd and $\x$ is in the interior.}
\end{equation}
We will use this along with \eqref{suppeqn:even}; both matter here because $\thatrbc$ involves both $\mhat_p^{(\v)}$ and $\mhat_{p+1}^{(p+1)}$.

\bigskip\noindent{\bf Case 1: $\bm{p +3 \leq \S}$.} Starting with the formula for $\btrbc$ at the boundary given above, Equations \eqref{suppeqn:even} and \eqref{suppeqn:even rbc} yield $\be_\v'\Gpt^{-1} \Lpt_2 = O(\h)$ and $\be_{p+1}'\Gqt^{-1} \Lqt_1 = O(\b)$. Therefore, these are the same order as the appropriate ``next'' term in the expansion \eqref{suppeqn:rbc bias post-taylor}, i.e.\ one further derivative must be retained. This is possible with $p + 3 \leq \S$.

Applying this to $\btrbc$, we find that 
\begin{align*}
	\btrbc & =  \sqrt{n\h}  \h^{p+3}   \;    \v! \be_\v'\Gpt^{-1}  \bigg\{ \frac{ \mu^{(p+2)} } { (p+2)! } \Big[ \h^{-1} \Lpt_2     -    \rho^{-2} \b^{-1} \Lpt_1 \be_{p+1}'\Gqt^{-1} \Lqt_1  \Big]      		\\
		& \qquad \qquad \qquad\qquad\qquad   +   \frac{ \mu^{(p+3)} } { (p+3)! }  \Big[ \Lpt_3     -    \rho^{-2} \Lpt_1 \be_{p+1}'\Gqt^{-1} \Lqt_2  \Big] \bigg\}  \;  \left[1 + o(1) \right].
\end{align*}
Notice that rather than spell out the limiting form of $\be_\v'\Gpt^{-1} \Lpt_2$ and $\be_{p+1}'\Gqt^{-1} \Lqt_1$, that is, the $\bm{C}_2$ and $\bm{\bar{C}}_2$ above, we keep with the fixed-$n$ spirit and write $\h^{-1} \be_\v'\Gpt^{-1} \Lpt_2$ and $\b^{-1} \be_{p+1}'\Gqt^{-1} \Lqt_1$, which dual the advantages of easy implementability (using the sample, non-tilde versions) and capturing all terms.

\bigskip\noindent{\bf Case 2: $\bm{p + 2 = \S}$.} The terms above involving $\mu^{(p+3)}$ must be replaced by the $O(\h^{\S + \s})$ (or $\b^{\S + \s}$) term of \eqref{suppeqn:Taylor lp}, which if $p + 2 = \S$, leaves the exponent as $p+2 + \s$. This gives
\begin{align*}
	\btrbc & =  \sqrt{n\h}  \h^{p+3}   \;    \v! \be_\v'\Gpt^{-1}  \bigg\{ \frac{ \mu^{(p+2)} } { (p+2)! } \Big[ \h^{-1} \Lpt_2     -    \rho^{-2} \b^{-1} \Lpt_1 \be_{p+1}'\Gqt^{-1} \Lqt_1  \Big] \bigg\}     		\\
		& \qquad   +   O\left( \sqrt{n\h}  \h^{p + 2 + \s}\right)   +   O\left( \sqrt{n\h}  \rho^{p+1} \b^{p + 2 + \s} \right)  		\\
		& =  \sqrt{n\h}  \h^{p+3}   \;    \v! \be_\v'\Gpt^{-1}  \bigg\{ \frac{ \mu^{(p+2)} } { (p+2)! } \Big[ \h^{-1} \Lpt_2     -    \rho^{-2} \b^{-1} \Lpt_1 \be_{p+1}'\Gqt^{-1} \Lqt_1  \Big] \bigg\}     		\\
		& \qquad   +   O\left( \sqrt{n\h}  \h^{p + 2 + \s}[ 1 + \rho^{-1 - \s}] \right).
\end{align*}
(Note that the order of second term is equivalently $\sqrt{n\h}  \h^{\S + \s}[ 1 + \rho^{-1 - \s}]$.) Recall that $\s \in (0,1]$. Therefore the first term above is higher order unless $\s = 1$ (which is not known) and $\rho \to \bar{\rho} \in (0,\infty)$, in which case the two are of the same order. Otherwise, the second term dominates, and further, if the $\rho^{-1 - \s}$ portion is the dominant rate if $\rho = \h/\b = o(1)$ regardless of $\s$. Therefore in this case it is more clear to suppress the constants of the higher order term and write
\[
	\btrbc =  O\left( \sqrt{n\h}  \h^{p + 2 + \s}[ 1 + \rho^{-1 - \s}] \right).
\]

\bigskip\noindent{\bf Case 3: $\bm{p + 2 > \S}$.} Now the symmetry does not apply (because only when the derivatives exist do the Taylor series terms collapse to $\Lp_k$ and $\Lq_k$) and so we find that $\btrbc = O\left( \sqrt{n\h}  \left[ \h^{\S + \s}   +   \rho^{p+1} \b^{\S + \s} \right] \right) = O\left( \sqrt{n\h}  \h^{\S + \s}   \left[  1 +   \rho^{p+1 - \S - \s}  \right] \right)$.

%%%%%%%%%%%%%%%%%%%%%%%%%%%%%%%%%%%%%%%%%%%%%%%%%%%%%%%%%%%%%%%%%%%%%%
\subsubsection{Interior Point: $p-\v$ even}

\bigskip\noindent{\bf Case 1: $\bm{p +3 \leq \S}$.} The conditions for $A=0$ and $\bar{A}=0$ in Equations \eqref{suppeqn:even} and \eqref{suppeqn:even rbc} reduce to whether or not $k$ is odd, because $p -\v$ is even for the former and the latter is always applied with $\v=p+1$. Using this to add the stabilization needed to Equation \eqref{suppeqn:rbc bias post-taylor} yields
\begin{align*}
	\E & \left[\thatrbc  \big| X_1, \ldots, X_n \right]  -  \tf   		\nonumber \\
	\begin{split}
		& = \h^{p+2-\v} \v! \be_\v'\Gpt^{-1} 
		\Biggl\{ 
			\frac{ \mu^{(p+2)} } { (p+2)! } \left( \Lpt_2 -   \rho^{-1}  \h \h^{-1} \Lpt_1 \b \b^{-1}\be_{p+1}'\Gqt^{-1}  \Lqt_1  \right)
			\\
			& \qquad\qquad\qquad\qquad\qquad + \frac{ \mu^{(p+3)} } { (p+3)! } \left(\h^2 \h^{-1} \Lpt_3   
			-   \rho^{-1}  \h \h^{-1} \b \Lpt_1 \be_{p+1}'\Gqt^{-1} \Lqt_2  \right) 	    			
		\Biggr\}	[1 + o_\P(1)]
	\end{split} \\
	\begin{split}
		& = \h^{p+2-\v} \v! \be_\v'\Gpt^{-1} 
		\Biggl\{ 
			\frac{ \mu^{(p+2)} } { (p+2)! } \left( \Lpt_2 -   \b^{2} \left[\h^{-1} \Lpt_1  \b^{-1}\be_{p+1}'\Gqt^{-1}  \Lqt_1 \right] \right)
			\\
			& \qquad\qquad\qquad\qquad\qquad + \frac{ \mu^{(p+3)} } { (p+3)! } \left(\h^2 \left[\h^{-1} \Lpt_3  \right] 
			-   \b^{2} \left[\h^{-1} \Lpt_1 \be_{p+1}'\Gqt^{-1} \Lqt_2 \right]  \right) 	    			
		\Biggr\}	[1 + o_\P(1)].
	\end{split}
\end{align*}
Therefore
\[
 \btrbc  =  \sqrt{n\h} \h^{p+2}  \frac{ \mu^{(p+2)} } { (p+2)! }  \v! \be_\v'\Gpt^{-1} \Lpt_2 	[1 + o(1)].
\]
To build intuition for why this result is correct, recall that if $\rho=1$ then $\thatrbc =  \mhat_p^{(\v)}  - \h^{p+1 - \v}  \v! \be_\v'\Gp^{-1} \Lp_1 \frac{\mhat_{p+1}^{(p+1)}}{(p+1)!} = \mhat_{p+1}^{(\v)}$, that is, $\thatrbc$ is equivalent to fitting a $p+1$ degree local polynomial rather than $p$. Here we are working under $p-\v$ even and so $p+1-\v$ is odd, and so naturally we recover the standard result for odd degree local polynomials.

\bigskip\noindent{\bf Case 2: $\bm{p + 2 = \S}$.} The terms above involving $\mu^{(p+3)}$ must be replaced by the $O(\h^{\S + \s})$ (or $\b^{\S + \s}$) term of \eqref{suppeqn:Taylor lp}, which if $p + 2 = \S$, leaves the exponent as $p+2 + \s$. Thus the leading term on the right of Equation \eqref{suppeqn:rbc bias post-taylor} becomes
\[\h^{p+2-\v} \v! \be_\v'\Gpt^{-1} 
		\Biggl\{ 
			\frac{ \mu^{(p+2)} } { (p+2)! } \left( \Lpt_2 -   \rho^{-1}  \Lpt_1 \be_{p+1}'\Gqt^{-1}  \Lqt_1  \right)
			+ O(\h^{\s} + \rho^{-1} \b^{\s})   			
		\Biggr\}.
\]
The same symmetry applies as in the previous case, and therefore we still have
\[
 \btrbc  =  \sqrt{n\h} \h^{p+2}  \frac{ \mu^{(p+2)} } { (p+2)! }  \v! \be_\v'\Gpt^{-1} \Lpt_2 	[1 + o(1)].
\]

\bigskip\noindent{\bf Case 3: $\bm{p + 2 > \S}$.} Now the symmetry does not apply (because only when the derivatives exist do the Taylor series terms collapse to $\Lp_k$ and $\Lq_k$) and so we find that $\btrbc = O\left( \sqrt{n\h}  \left[ \h^{\S + \s}   +   \rho^{p+1} \b^{\S + \s} \right] \right) = O\left( \sqrt{n\h}  \h^{\S + \s}   \left[  1 +   \rho^{p+1 - \S - \s}  \right] \right)$.

%%%%%%%%%%%%%%%%%%%%%%%%%%%%%%%%%%%%%%%%%%%%%%%%%%%%%%%%%%%%%%%%%%%%%%
%%%%%%%%%%%%%%%%%%%%%%%%%%%%%%%%%%%%%%%%%%%%%%%%%%%%%%%%%%%%%%%%%%%%%%
\section{Notes on Alternative Standard Errors}
	\label{supp:other standard errors}

The proofs above are based on specific standard errors. In particular, we use the fixed-$n$ form of the variance from Equation \eqref{suppeqn:variance lp}, namely
\[ \sp^2 =  \v!^2 \be_\v'\Gp^{-1} (\h \Op \bS \Op' /n) \Gp^{-1} \be_\v,\]
and estimate $\bS$ using regression residuals, $\Shatp = \diag(\hat{v}(X_i): i = 1,\ldots, n)$, with $\hat{v}(X_i) = ( Y_i - \br_p(X_i - \x)'\bhat_p )^2$ for $\bhat_p$ defined in Equation \eqref{suppeqn:lp}. This is the HC0 variance estimator. We discuss two types of alternatives here: (i) different estimators of essentially the same fixed-$n$ object and (ii) different population standardizations altogether. If other standard errors are used, the results may change. The type and severity of the change will depend on the choice of standard error. In particular, the coverage error rate can be slower, but not faster. This is because the Studentization and standardization do not affect the rate of any term besides the $\lambda_{\i,\f} \w_{3,\i,\f}$ term, and thus $\lambda_{\i,\f} \equiv 0$ is the most that can be accomplished through variance estimation.

Within the fixed-$n$ form, we consider two alternative estimators of (essentially) the conditional variances of Equation \eqref{suppeqn:variance lp}: the HC$k$ class estimators and nearest-neighbor based estimators.

First, motivated by the fact that the least-squares residuals are on average too small, we could implement one of the HC$k$ class of heteroskedasticity-consistent standard errors \citep{MacKinnon2013_BookChap} beyond HC0. In particular, HC0, HC1, HC2, and HC3 are allowed in the {\tt nprobust} package \citep{Calonico-Cattaneo-Farrell2019_JSS}. These are defined as follows. First, $\shatp^2$ (and $\shatrbc^2$) defined above and treated in the proofs is the HC0 estimator, employing the estimated residuals unweighted: $\hat{\e}_i^2  =  \hat{v}(X_i)  =  ( Y_i - \br_p(X_i - \x)'\bhat_p )^2$. Then, for $k=1, 2, 3$, the $\shatp^2$-HC$k$ estimator is obtained by dividing $\hat{\e}_i^2$ by, respectively, $(n-2\trace(\bm{Q}_p)+\trace(\bm{Q}_p'\bm{Q}_p))/n$, $(1-\bm{Q}_{p,ii})$, and $(1-\bm{Q}_{p,ii})^2$, where $\bm{Q}_{p,ii}$ is the $i$-th diagonal element of the projection matrix $\bm{Q}_p:= \bRc'(\bRc' \bW \bRc)^{-1}\bRc' \bW = \bRc'\Gp^{-1} \Op/n$. The corresponding estimators $\shatrbc^2$-HC$k$ are the same way, substituting the appropriate pieces. 

These estimators may perform better in small samples, a conjecture backed by simulation studies elsewhere. Adapting the proofs to allow for HC1, HC2, and HC3 would be notationally extremely cumbersome, but is conceptually straightforward. The building block of each is the matrix $\bm{Q}_p$, which is almost already a function of $\bZ_i$ from \eqref{suppeqn:bZ}; it is not difficult to see that Cram\'er's condition is plausible for this object. It is important to note that the rates in the expansion would not change, only the constants (through the terms of \eqref{suppeqn:shat terms}). 

A second option, still using the fixed-$n$ form and also designed to improve upon the least squares residuals, is to use a nearest-neighbor-based estimator with a fixed number of neighbors \citep{Muller-Stadtmuller1987_AoS}. This is also allowed in our software \citep{Calonico-Cattaneo-Farrell2019_JSS}. For a fixed, positive integer $J$, let $X_{j(i)}$ denote the $j$-th closest observation to $X_i$, $j=1, \ldots, J$. Set $\hat{v}(X_i) = \frac{J}{J+1} ( Y_i  -  \sum_{j=1}^J Y_{j(i)} / J )^2$. This estimate is unbiased for $v(X_i)$, and although $\hat{v}(\cdot)$ is inconsistent, the resulting $\shatp^2 =  \v!^2 \be_\v'\Gp^{-1} (\h \Op \Shat_{NN} \Op' /n) \Gp^{-1} \be_\v$ provides valid Studentization (as would the analogous $\shatrbc^2$). This approach, however, falls outside our proofs. Lemma \ref{supplem:Cramer lp} would not verify Cram\'er's condition for this estimator. A modified approach to verifying condition ($\text{III}''_\alpha$) of \cite{Skovgaard1981_SJS} would be required and Assumption \ref{suppasmpt:ci lp} would not be sufficient.

Finally, as discussed above, on may use a different form of standardization altogether. As argued in the main text and above, using variance forms other than \eqref{suppeqn:variance lp} can be detrimental to coverage by injecting terms with $\lambda_{\i,\f} \neq 0$. Examples were given in Section \ref{supp:terms lp} and discussed further in the main paper. The most common option would be to employ the asymptotic approximation to the conditional variance:
\[ \sigma^2 \to_\P \frac{v(\x)}{f(\x)} \mathcal{V},\]
where $f(\cdot)$ is the marginal density of $X$ and $\mathcal{V}$ is a known constant depending only on the equivalent kernel (and thus $\mathcal{V}_\pp$ and $\mathcal{V}_\RBC$ would be different); see \cite[Theorem 3.1]{Fan-Gijbels1996_book}. Estimating this quantity requires estimating the conditional variance function and the (inverse of the) density at a single point, the point of interest $\x$. If both of these are based on kernel methods using the same kernel and bandwidth $\h$, then Theorem \ref{suppthm:EE lp} allows for this choice. It is clear that the expansion of the Studentization, Equation \eqref{suppeqn:shat terms}, will change dramatically, as will the elements of $\bZ_i$. However, the latter change will be relatively innocuous as far as the proof is concerned, because Lemma \ref{supplem:Cramer lp} covers the objects already. But the change to Equation \eqref{suppeqn:shat terms} will result in additional terms, with potentially slower rates, appearing the Edgeworth expansion. See the discussion in Section \ref{supp:terms lp}.

There are certainly many other options for (first-order) valid Studentization. Other population choices include (i) using $\hat{v}(X_i) = ( Y_i - \hat{m}(\x) )^2$; (ii) using local or assuming global heteroskedasticity; (iii) using other nonparametric estimators for $v(X_i)$, relying on new tuning parameters. None of these can be recommended based on our results. As above, some can be accommodated into our proof more or less directly, depending on the implementation details.

%%%%%%%%%%%%%%%%%%%%%%%%%%%%%%%%%%%%%%%%%%%%%%%%%%%%%%%%%%%%%%%%%%%%%%
%%%%%%%%%%%%%%%%%%%%%%%%%%%%%%%%%%%%%%%%%%%%%%%%%%%%%%%%%%%%%%%%%%%%%%
\section{Check Function Loss}
	\label{supp:check function}

In the main text, it was pointed out that coverage error can be measured by the check function loss:
\[
	\sup_{\f \in \F_\S} \mathcal{L}\Big( \P_\f [ \tf \in \i ] - (1-\alpha)  \Big) , \qquad \quad \mathcal{L}(e) =\mathcal{L}_\tau(e)=e\left(\tau - \One\{e<0\}\right)
\]	
Using the check function loss allows the researcher, through their choice of $\tau$, to evaluate inference procedures according to their preferences against over- and under-coverage. Setting $\tau = 1/2$ recovers the above, symmetric measure of coverage error. Guarding more against undercoverage (a preference for conservative intervals) requires choosing a $\tau < 1/2$. For example, setting $\tau = 1/3$ encodes the belief that undercoverage is twice as bad as the same amount of overcoverage. 

Using this loss will affect the constants of the optimal bandwidths and kernels (dependent on how these are optimized, such as for length, coverage error, or trading these off) but the rates will not be impacted. This is due to standard properties of the check function, which, for completeness, we spell out in the following result.

\begin{lemma}
	\label{supplem:check function}
	$\mathcal{L}(e) = e\left(\tau - \One\{e<0\}\right)$ obeys:
	\begin{enumerate}
		\item $\mathcal{L}(ae) = a\mathcal{L}(e)$ for $a > 0$,
		\item $\mathcal{L}(e) \leq (\tau + 1) |e|$, and 
		\item $\mathcal{L}(e_1 + e_2) \leq \mathcal{L}(e_1) + \mathcal{L}(e_2)$ for $a > 0$.
	\end{enumerate}
\end{lemma} 
\begin{proof}
The first property follows because $\mathcal{L}(ae) = (ae)\left(\tau - \One\{(ae)<0\}\right)$ and, as $a > 0$, $\One\{(ae)<0\} = \One\{e<0\}$. The second uses the obvious bounds. The third, the triangle inequality, holds as follows. 
\begin{align*}
	\mathcal{L}(e_1 + e_2) & = (e_1 + e_2) \left(\tau - \One\{(e_1 + e_2)<0\}\right)   			\\
	& = e_1 \left(\tau - \One\{e_1<0\}\right) + e_2 \left(\tau - \One\{e_2<0\}\right)   			\\
	& \quad + e_1 \One\{e_1<0\}  + e_2 \One\{e_2<0\} - (e_1 + e_2)\One\{(e_1 + e_2)<0\}.
\end{align*}
In the second equality, the first line is exactly $\mathcal{L}(e_1) + \mathcal{L}(e_2)$. The second line is nonpositive. To this, consider four cases. (1) If $e_1 \geq 0$ and $e_2 \geq 0$, then all the indicators are zero and the second line is zero. (2) If $e_1 < 0$ and $e_2 < 0$, then all the indicators are one and the second line is $e_1 + e_2 - (e_1 + e_2)$ and is again zero. (3) If $e_1 \geq 0$, $e_2 < 0$, and $e_1 \geq |e_2|$, then $\One\{e_1<0\} = \One\{(e_1 + e_2)<0\} = 0$, and the second line is $e_2 < 0$. (4)  If $e_1 \geq 0$, $e_2 < 0$, and $e_1 < |e_2|$, then $\One\{e_2<0\} = \One\{(e_1 + e_2)<0\} = 1$, and the second line is $e_2 - (e_1 + e_2) = - e_1 < 0$. 
\end{proof}

%%%%%%%%%%%%%%%%%%%%%%%%%%%%%%%%%%%
%%%%%%%%%%%%%%%%%%%%%%%%%%%%%%%%%%%
\section{Simulation Results and Numerical Details}
	\label{supp:numerical}

%%%%%%%%%%%%%%%%%%%%%%%%%%%%%%%%%%%
\subsection{Simulation Study}

In this section we present the complete results from our simulation study addressing the finite-sample performance of the methods described in the main paper. All results are qualitatively consistent with the main theoretical results of our paper.

We study model \eqref{suppeqn:model} with $X_i$ uniformly distributed on $[-1,1]$, $\e$ distributed independently standard normal, and  
\[
	\mu(x) = \frac{ \sin(3\pi x/2 ) }{ 1+18x^2(\sign(x) +1)   },
\]
where $\sign(x)-1,0$ or $-1$ according to $x>0$, $x=0$ or $x<0$, respectively.

We consider $5,000$ simulation replications, where for each replication we generate data as i.i.d. draws of size $n=\{100,250,500,750,1000,2000\}$. The point of evaluation is one of six equally spaced evaluation points $\x\in\{-1,-0.6,-0.2,0.2,0.6,1\}$ using the Epanechnikov and Uniform kernel, setting $p=1$ (for $\v=0$) and $p=2$ (for $\v=1$). 
Finally, we evaluate the performance of the confidence intervals using several bandwidth choices. First, we use $\hat{h}_\RBC$, a data-driven version of the inference-optimal bandwidth $h_\RBC$. We also consider the analogous version for undersmoothing confidence intervals, $\hat{h}_\US$, and the standard choice in practice, $\hat{h}_\MSE$.
In all cases, robust bias correction is implemented using $\rho=\rho^*$.

We report empirical coverage probabilities and average interval length of nominal 95\% confidence interval for $\mu(\x)$ and $\mu^{(1)}(\x)$ based on robust bias correction and undersmoothing.

First, in Figures \ref{suppfig:ec_nu0_epa_hrbc}, \ref{suppfig:ec_nu0_epa_hus}, and \ref{suppfig:ec_nu0_epa_hmse} we present empirical coverage probabilities for $\v=0$ using the Epanechnikov kernel for each evaluation point and choice of bandwidth selector, as a function on the different sample sizes considered.  
Overall, we can see that robust bias correction yields close to accurate coverage, improving over undersmoothing in almost every case. Performance is highly superior at points where the functions present high curvature and also at the boundary. Performance is never worse even when the function is quite linear.
We obtain similar findings when looking at the results for $\v=1$ in Figure \ref{suppfig:ec_nu1_epa_hrbc}, \ref{suppfig:ec_nu1_epa_hus}, and \ref{suppfig:ec_nu1_epa_hmse}, where robust bias correction outperforms undersmoothing even more.

We compare confidence interval performance in terms of length, taking coverage into account by looking at RBC and US confidence intervals implemented with their corresponding coverage error optimal bandwidth choices ($\hat{h}_\RBC$ and $\hat{h}_\US$, respectively), which is when they perform best in terms of coverage. We also include other valid, but non optimal choices $\irbc(\hat{h}_{\MSE})$,  $\irbc(\hat{h}_{\US})$. Figures \ref{suppfig:il_nu0_epa} and \ref{suppfig:il_nu1_epa} present the results for $\v=0$ and $\v=1$, respectively, using the Epanechnikov kernel. We find that, in most cases, RBC confidence intervals are, on average, not larger than US, and sometimes even shorter.
Finally, we report the average (over simulations) of the estimated bandwidths in Figures \ref{suppfig:h_nu0_epa} and \ref{suppfig:h_nu1_epa}. 

All the information used to generate the plots can be found in Tables \ref{supptable:ec_nu0_epa} and \ref{supptable:ec_nu1_epa} (for coverage probabilities), and \ref{supptable:il_nu0_epa} and \ref{supptable:il_nu1_epa} (for average length). 
We find similar results for the performance of RBC and US confidence intervals when using the Uniform kernel, as shown in the remaining figures and tables, corresponding exactly to those for the Epanechnikov kernel.

%%%%%%%%%%%%%%%%%%%%%%%% EC EPA

%%% h_RBC
\clearpage

\begin{figure}[!htb]
	\centering
	\caption{Empirical Coverage for 95\% Confidence Intervals\\ Epanechnikov Kernel, $\hat{h}_{\RBC}$, $\v=0$}
	\label{suppfig:ec_nu0_epa_hrbc}
	\begin{subfigure}[b]{0.5\textwidth}
		\includegraphics[height=0.3\textheight,width=0.95\textwidth]{simuls/output/ce_hrbc_kepa_p1_d0_x1.pdf}
		\subcaption{$\x=-1$}
	\end{subfigure}%
	\begin{subfigure}[b]{0.5\textwidth}
		\includegraphics[height=0.3\textheight,width=0.95\textwidth]{simuls/output/ce_hrbc_kepa_p1_d0_x2.pdf}
		\subcaption{$\x=-0.6$}
	\end{subfigure}	
	\begin{subfigure}[b]{0.5\textwidth}
	\includegraphics[height=0.3\textheight,width=0.95\textwidth]{simuls/output/ce_hrbc_kepa_p1_d0_x3.pdf}
	\subcaption{$\x=-0.2$}
\end{subfigure}%	
	\begin{subfigure}[b]{0.5\textwidth}
		\includegraphics[height=0.3\textheight,width=0.95\textwidth]{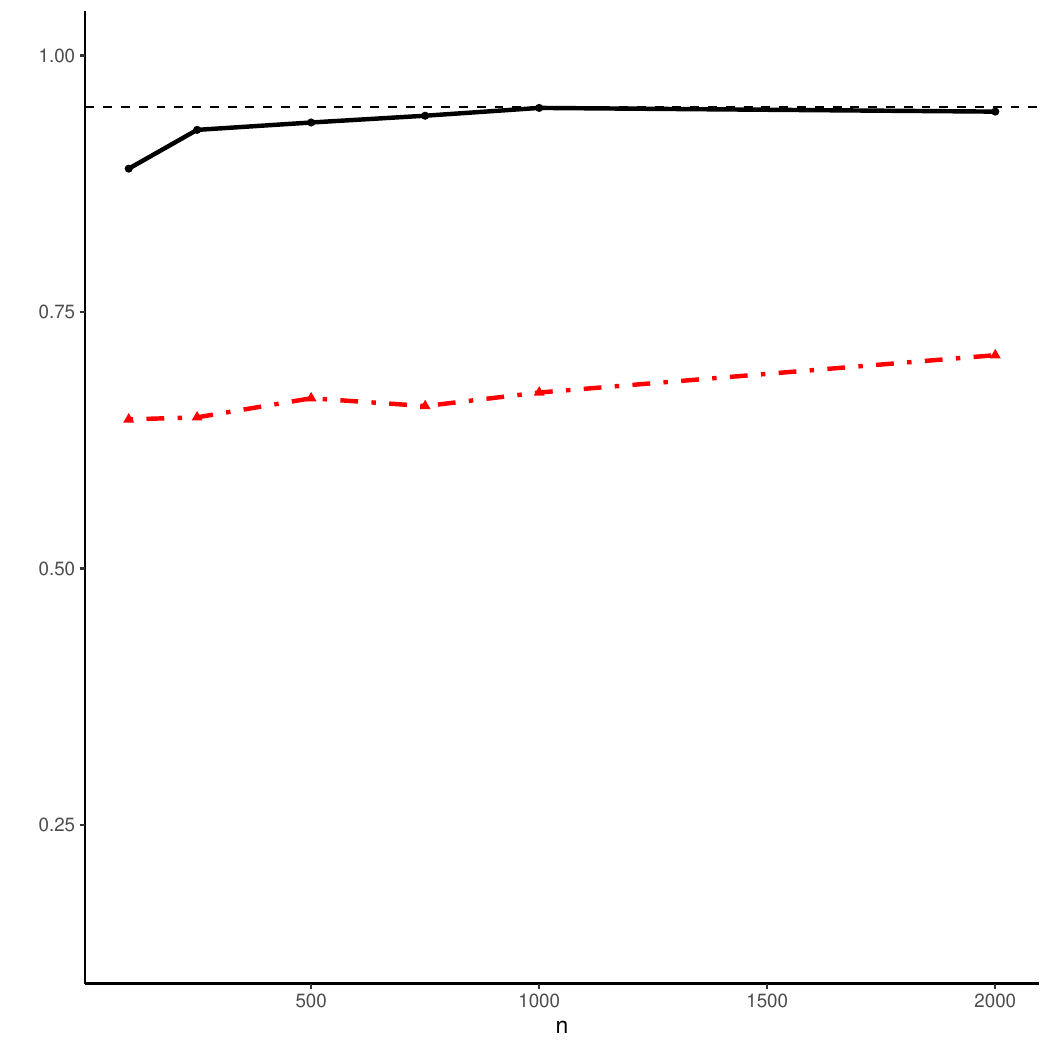}
		\subcaption{$\x=0.2$}
	\end{subfigure}	
	\begin{subfigure}[b]{0.5\textwidth}
		\includegraphics[height=0.3\textheight,width=0.95\textwidth]{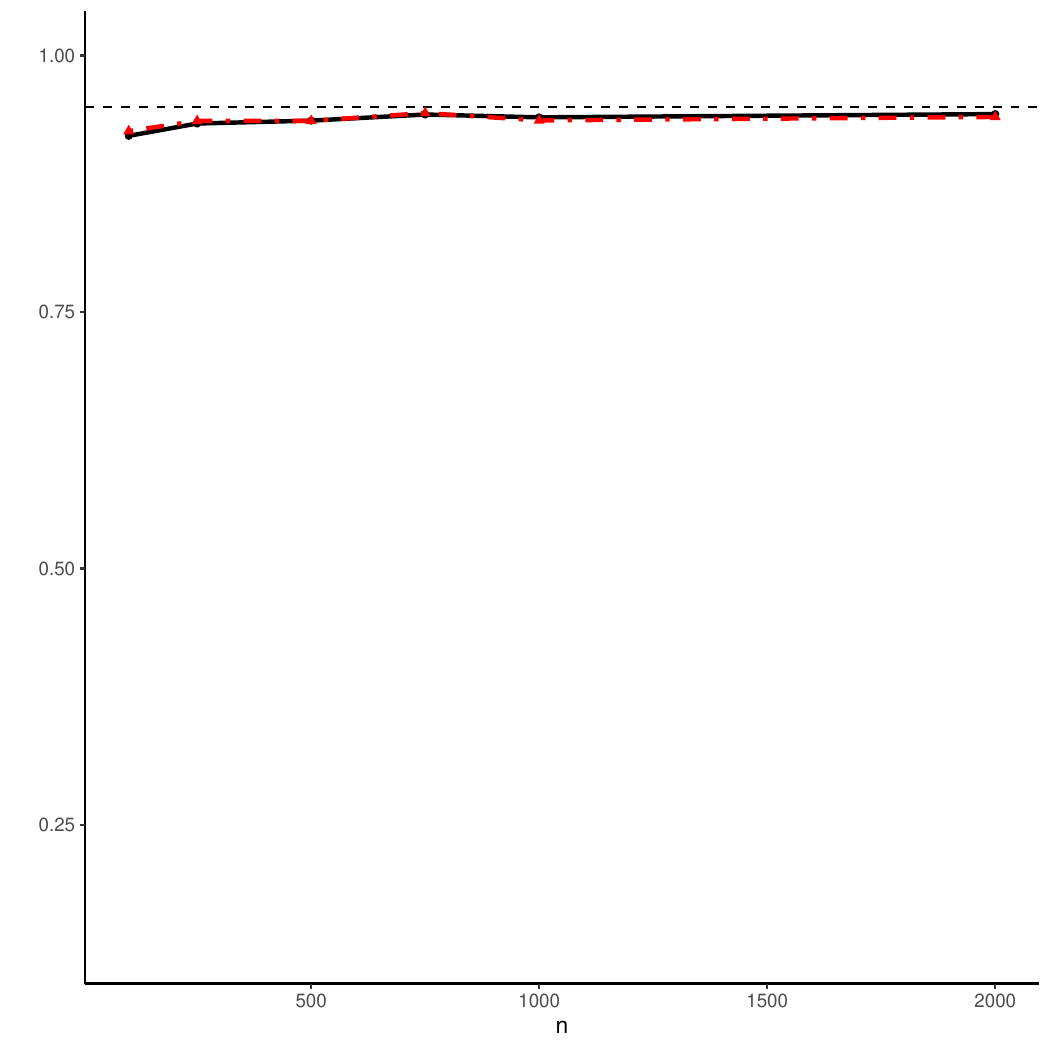}
		\subcaption{$\x=0.6$}
	\end{subfigure}%
	\begin{subfigure}[b]{0.5\textwidth}
	\includegraphics[height=0.3\textheight,width=0.95\textwidth]{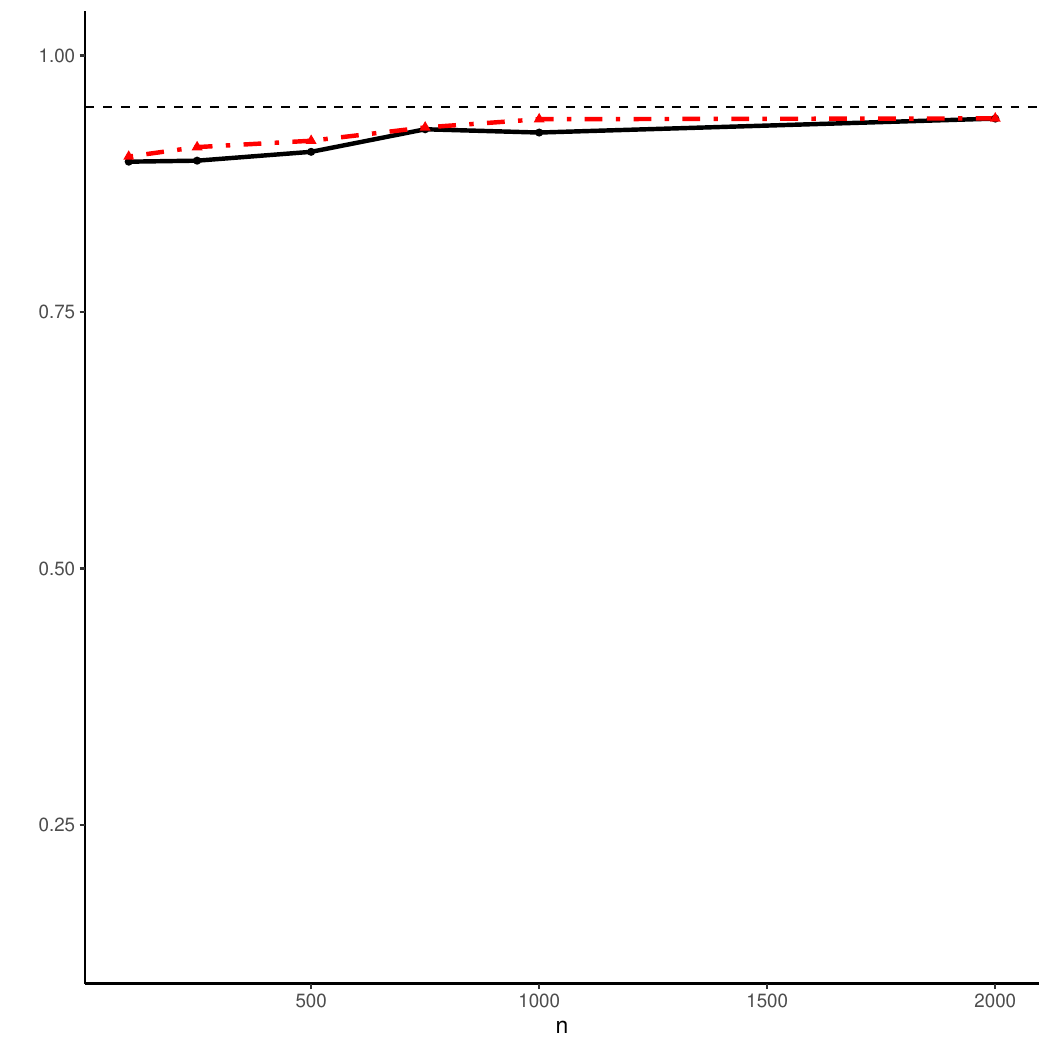}
	\subcaption{$\x=1$}
\end{subfigure}%
	\begin{flushleft}\footnotesize Notes: \blackline Robust Bias Correction, \redline Undersmoothing
	\end{flushleft}
\end{figure}

\clearpage
\begin{figure}[!htb]
	\centering
	\caption{Empirical Coverage for 95\% Confidence Intervals\\
Epanechnikov Kernel, $\hat{h}_{\RBC}$, $\v=1$}
	 \label{suppfig:ec_nu1_epa_hrbc}	
	\begin{subfigure}[b]{0.5\textwidth}
		\includegraphics[height=0.3\textheight,width=0.95\textwidth]{simuls/output/ce_hrbc_kepa_p2_d1_x1.pdf}
		\subcaption{$\x=-1$}
	\end{subfigure}%
	\begin{subfigure}[b]{0.5\textwidth}
		\includegraphics[height=0.3\textheight,width=0.95\textwidth]{simuls/output/ce_hrbc_kepa_p2_d1_x2.pdf}
		\subcaption{$\x=-0.6$}
	\end{subfigure}	
	\begin{subfigure}[b]{0.5\textwidth}
		\includegraphics[height=0.3\textheight,width=0.95\textwidth]{simuls/output/ce_hrbc_kepa_p2_d1_x3.pdf}
		\subcaption{$\x=-0.2$}
	\end{subfigure}%	
	\begin{subfigure}[b]{0.5\textwidth}
		\includegraphics[height=0.3\textheight,width=0.95\textwidth]{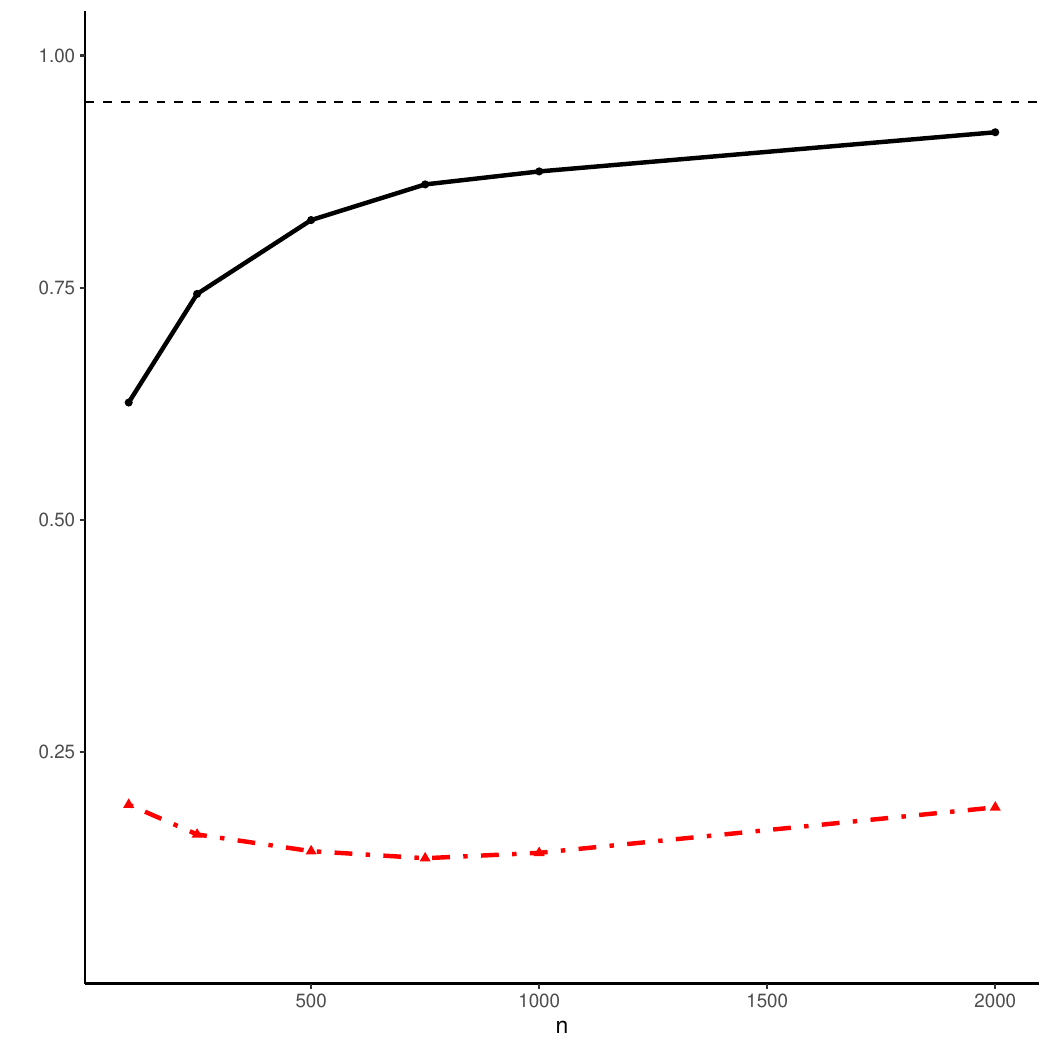}
		\subcaption{$\x=0.2$}
	\end{subfigure}	
	\begin{subfigure}[b]{0.5\textwidth}
		\includegraphics[height=0.3\textheight,width=0.95\textwidth]{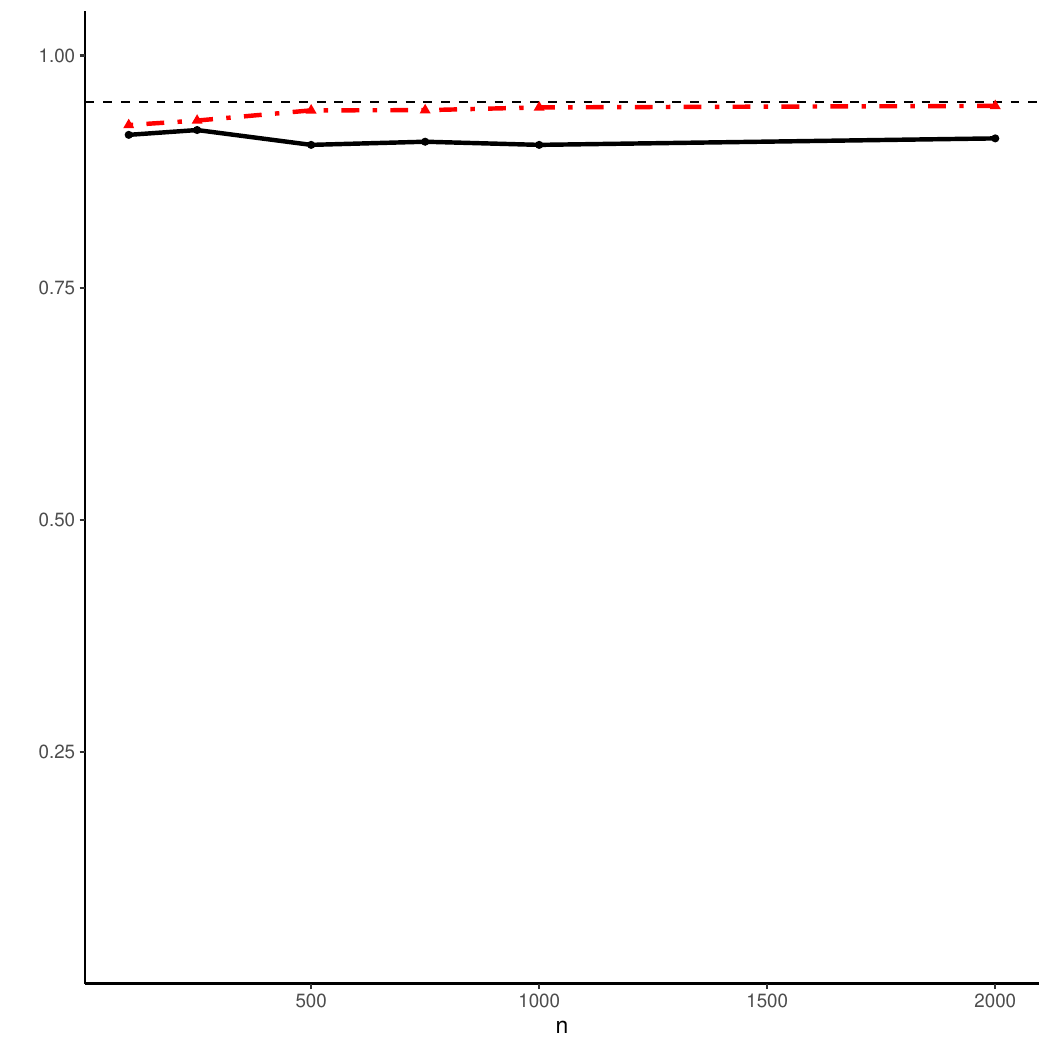}
		\subcaption{$\x=0.6$}
	\end{subfigure}%
	\begin{subfigure}[b]{0.5\textwidth}
		\includegraphics[height=0.3\textheight,width=0.95\textwidth]{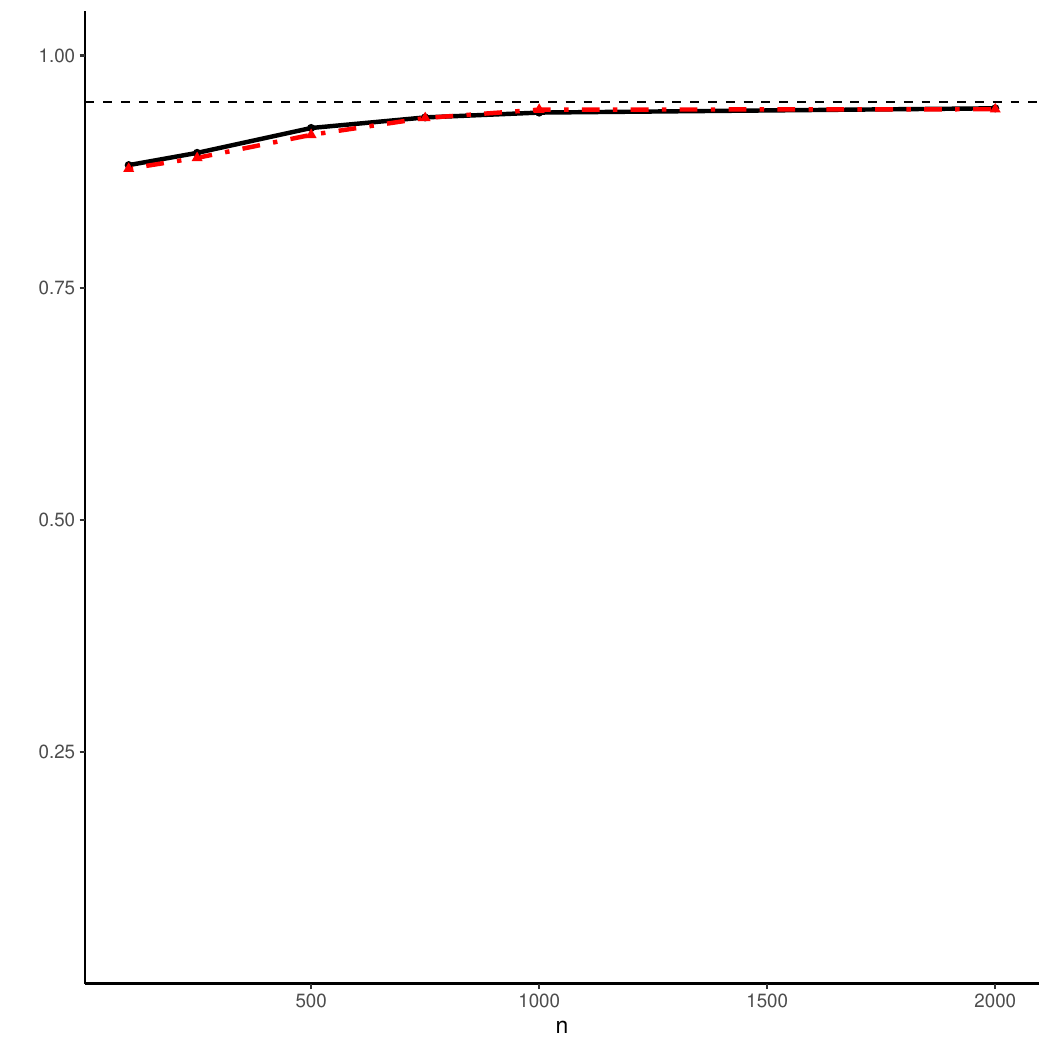}
		\subcaption{$\x=1$}
	\end{subfigure}%
	\begin{flushleft}\footnotesize Notes: \blackline Robust Bias Correction, \redline Undersmoothing
	\end{flushleft}
\end{figure}

%%% h_US
\clearpage
\begin{figure}[!htb]	
	\centering
	\caption{Empirical Coverage for 95\% Confidence Intervals\\
		Epanechnikov Kernel, $\hat{h}_{\US}$, $\v=0$}
\label{suppfig:ec_nu0_epa_hus}	
	\begin{subfigure}[b]{0.5\textwidth}
		\includegraphics[height=0.3\textheight,width=0.95\textwidth]{simuls/output/ce_hus_kepa_p1_d0_x1.pdf}
		\subcaption{$\x=-1$}
	\end{subfigure}%
	\begin{subfigure}[b]{0.5\textwidth}
		\includegraphics[height=0.3\textheight,width=0.95\textwidth]{simuls/output/ce_hus_kepa_p1_d0_x2.pdf}
		\subcaption{$\x=-0.6$}
	\end{subfigure}	
	\begin{subfigure}[b]{0.5\textwidth}
		\includegraphics[height=0.3\textheight,width=0.95\textwidth]{simuls/output/ce_hus_kepa_p1_d0_x3.pdf}
		\subcaption{$\x=-0.2$}
	\end{subfigure}%	
	\begin{subfigure}[b]{0.5\textwidth}
		\includegraphics[height=0.3\textheight,width=0.95\textwidth]{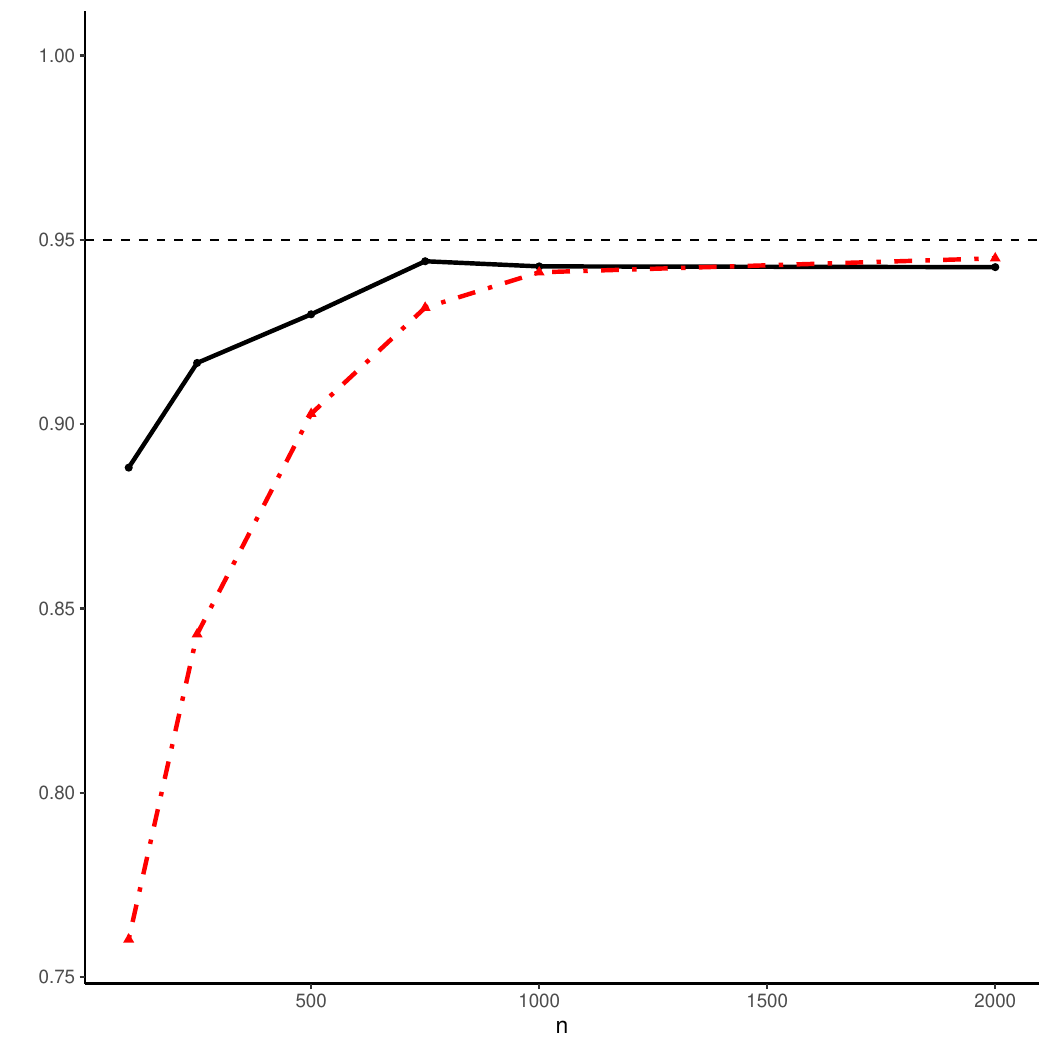}
		\subcaption{$\x=0.2$}
	\end{subfigure}	
	\begin{subfigure}[b]{0.5\textwidth}
		\includegraphics[height=0.3\textheight,width=0.95\textwidth]{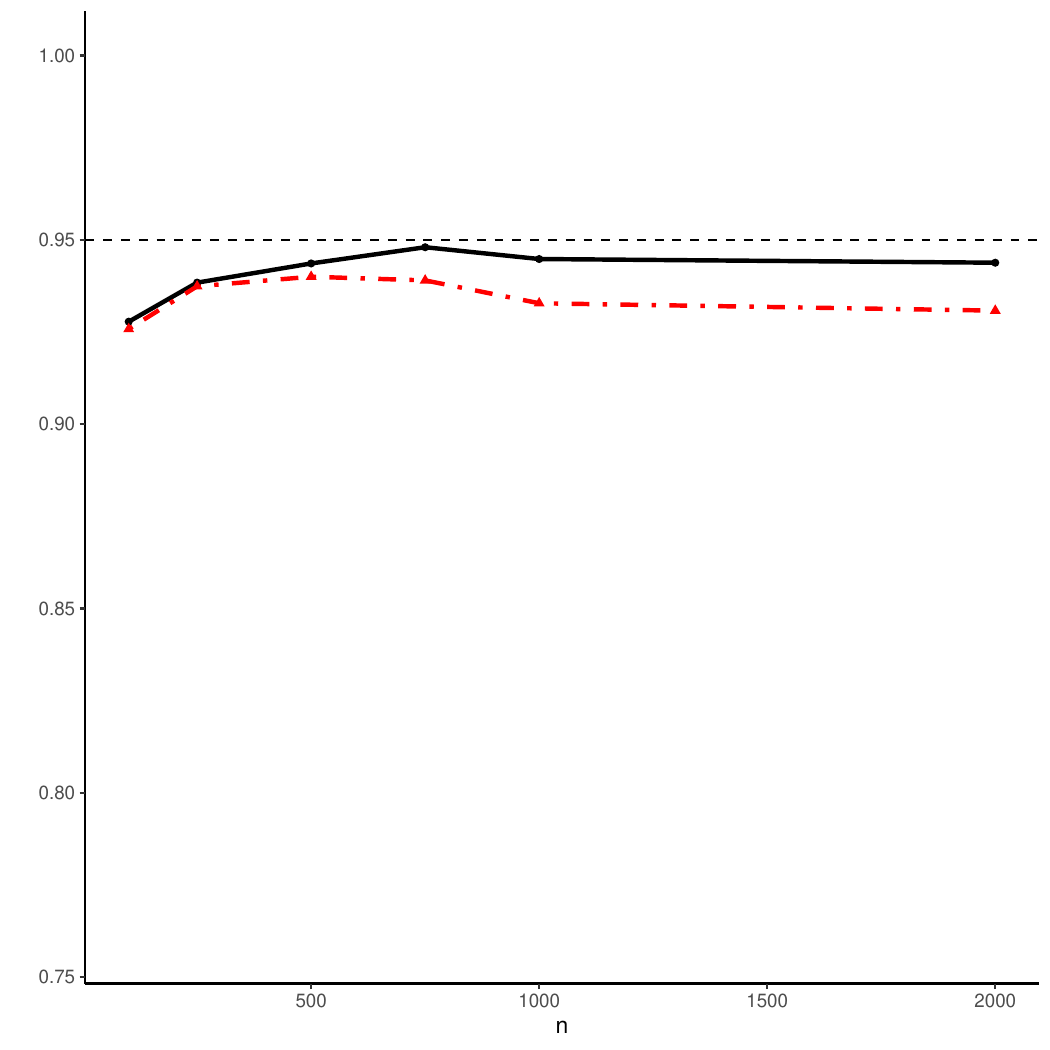}
		\subcaption{$\x=0.6$}
	\end{subfigure}%
	\begin{subfigure}[b]{0.5\textwidth}
		\includegraphics[height=0.3\textheight,width=0.95\textwidth]{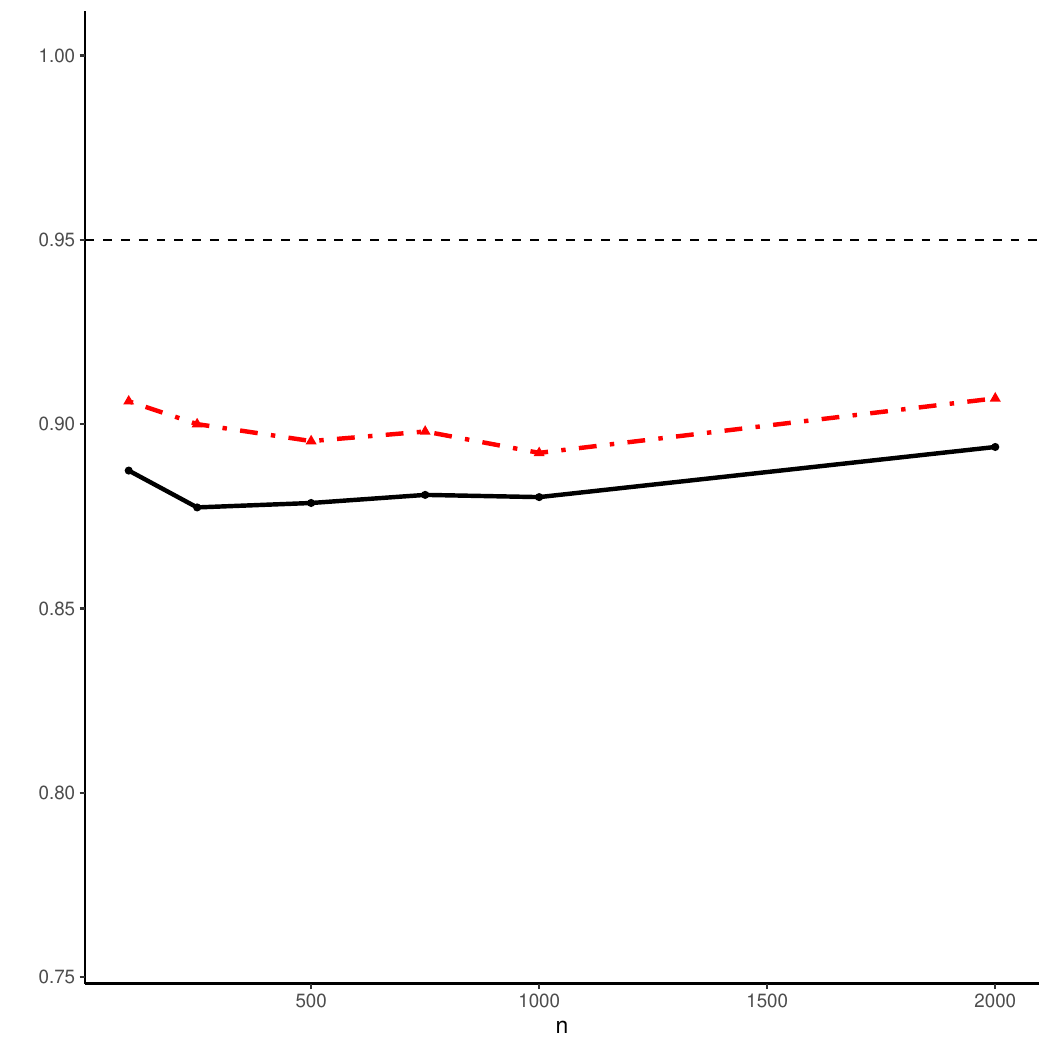}
		\subcaption{$\x=1$}
	\end{subfigure}%
	\begin{flushleft}\footnotesize Notes: \blackline Robust Bias Correction, \redline Undersmoothing
	\end{flushleft}
\end{figure}

\clearpage
\begin{figure}[!htb]
	\centering
	\caption{Empirical Coverage for 95\% Confidence Intervals\\
	Epanechnikov Kernel, $\hat{h}_{\US}$, $\v=1$}
		\label{suppfig:ec_nu1_epa_hus}	
	\begin{subfigure}[b]{0.5\textwidth}
		\includegraphics[height=0.3\textheight,width=0.95\textwidth]{simuls/output/ce_hus_kepa_p2_d1_x1.pdf}
		\subcaption{$\x=-1$}
	\end{subfigure}%
	\begin{subfigure}[b]{0.5\textwidth}
		\includegraphics[height=0.3\textheight,width=0.95\textwidth]{simuls/output/ce_hus_kepa_p2_d1_x2.pdf}
		\subcaption{$\x=-0.6$}
	\end{subfigure}	
	\begin{subfigure}[b]{0.5\textwidth}
		\includegraphics[height=0.3\textheight,width=0.95\textwidth]{simuls/output/ce_hus_kepa_p2_d1_x3.pdf}
		\subcaption{$\x=-0.2$}
	\end{subfigure}%	
	\begin{subfigure}[b]{0.5\textwidth}
		\includegraphics[height=0.3\textheight,width=0.95\textwidth]{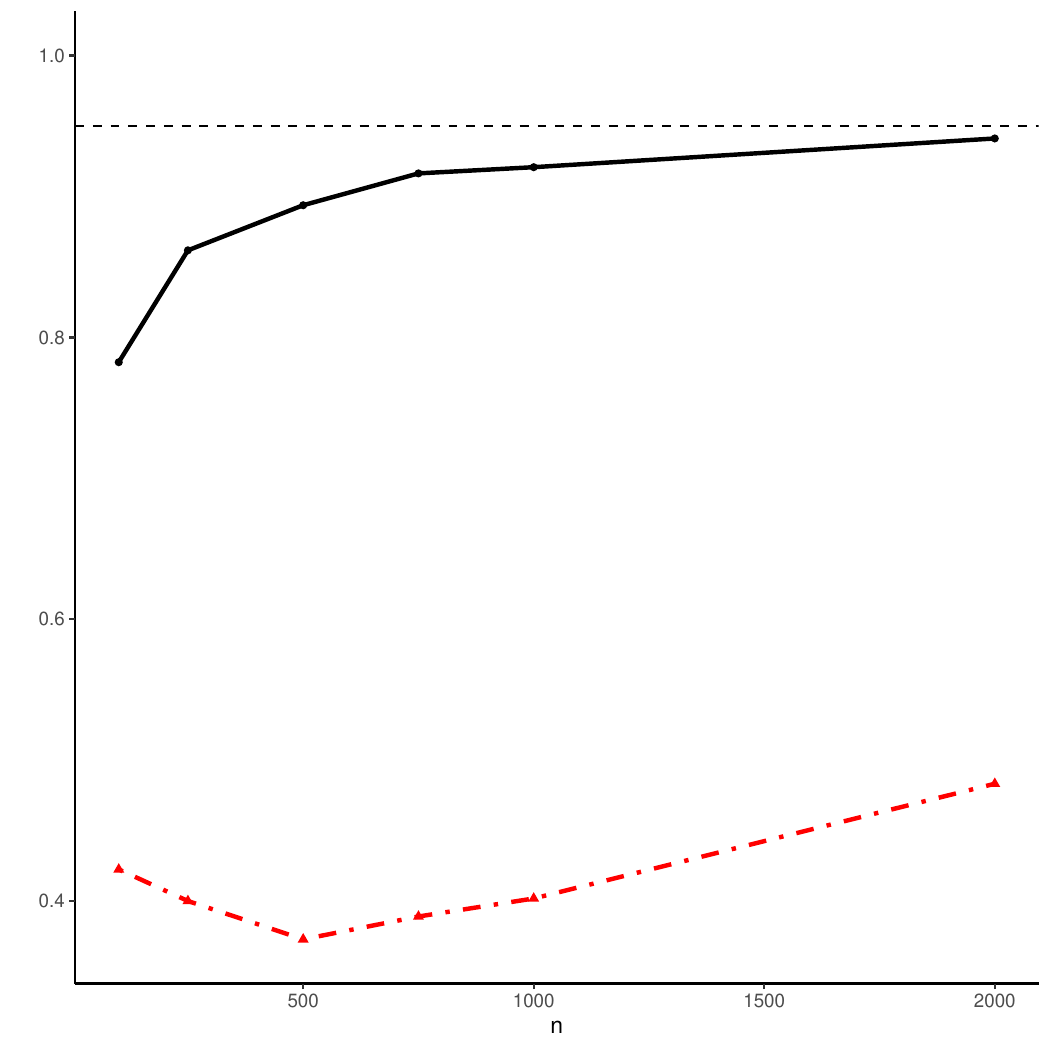}
		\subcaption{$\x=0.2$}
	\end{subfigure}	
	\begin{subfigure}[b]{0.5\textwidth}
		\includegraphics[height=0.3\textheight,width=0.95\textwidth]{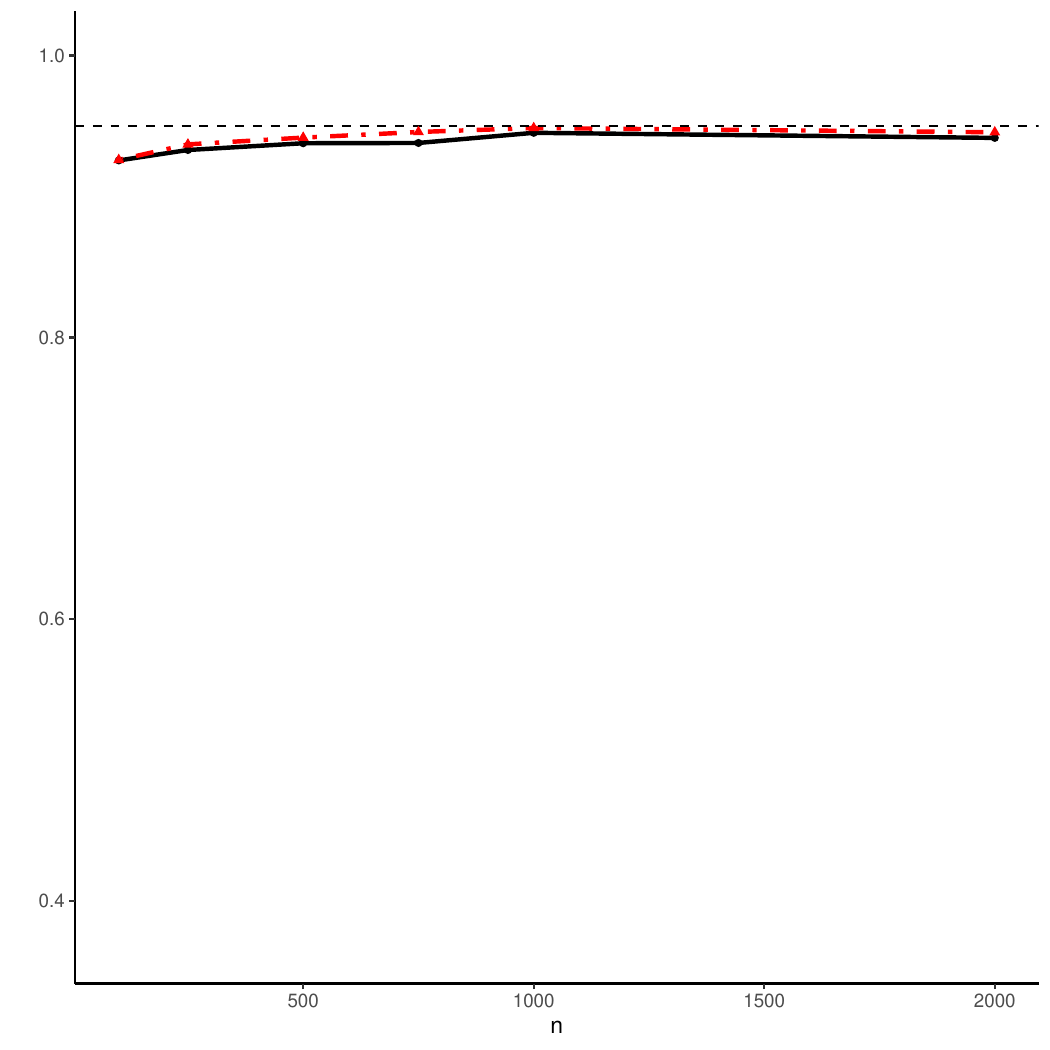}
		\subcaption{$\x=0.6$}
	\end{subfigure}%
	\begin{subfigure}[b]{0.5\textwidth}
		\includegraphics[height=0.3\textheight,width=0.95\textwidth]{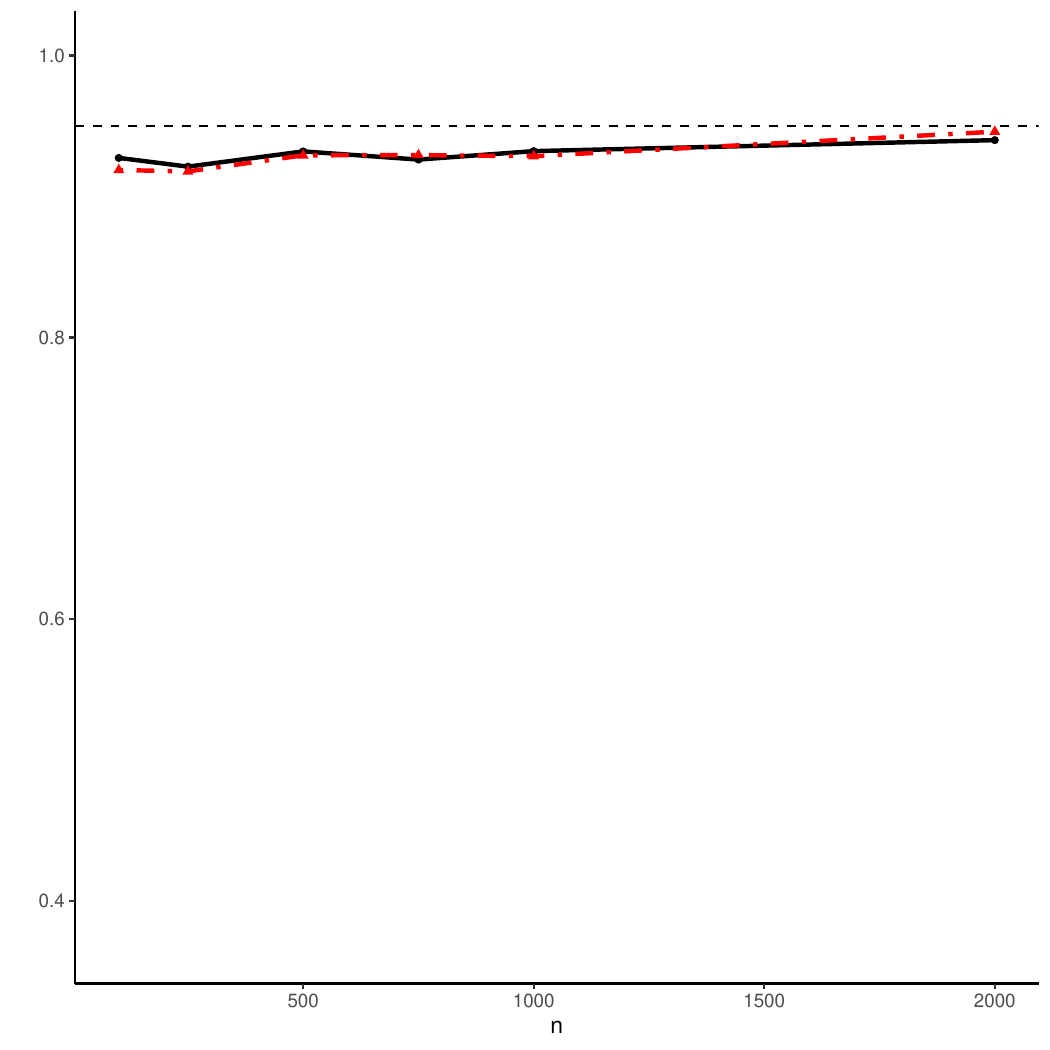}
		\subcaption{$\x=1$}
	\end{subfigure}%
	\begin{flushleft}\footnotesize Notes: \blackline Robust Bias Correction, \redline Undersmoothing
	\end{flushleft}
\end{figure}

%%% h_MSE
\clearpage
\begin{figure}[!htb]	
	\centering
	\caption{Empirical Coverage for 95\% Confidence Intervals\\
	Epanechnikov Kernel, $\hat{h}_{\MSE}$, $\v=0$}
	\label{suppfig:ec_nu0_epa_hmse}	
	\begin{subfigure}[b]{0.5\textwidth}
		\includegraphics[height=0.3\textheight,width=0.95\textwidth]{simuls/output/ce_hmse_kepa_p1_d0_x1.pdf}
		\subcaption{$\x=-1$}
	\end{subfigure}%
	\begin{subfigure}[b]{0.5\textwidth}
		\includegraphics[height=0.3\textheight,width=0.95\textwidth]{simuls/output/ce_hmse_kepa_p1_d0_x2.pdf}
		\subcaption{$\x=-0.6$}
	\end{subfigure}	
	\begin{subfigure}[b]{0.5\textwidth}
		\includegraphics[height=0.3\textheight,width=0.95\textwidth]{simuls/output/ce_hmse_kepa_p1_d0_x3.pdf}
		\subcaption{$\x=-0.2$}
	\end{subfigure}%	
	\begin{subfigure}[b]{0.5\textwidth}
		\includegraphics[height=0.3\textheight,width=0.95\textwidth]{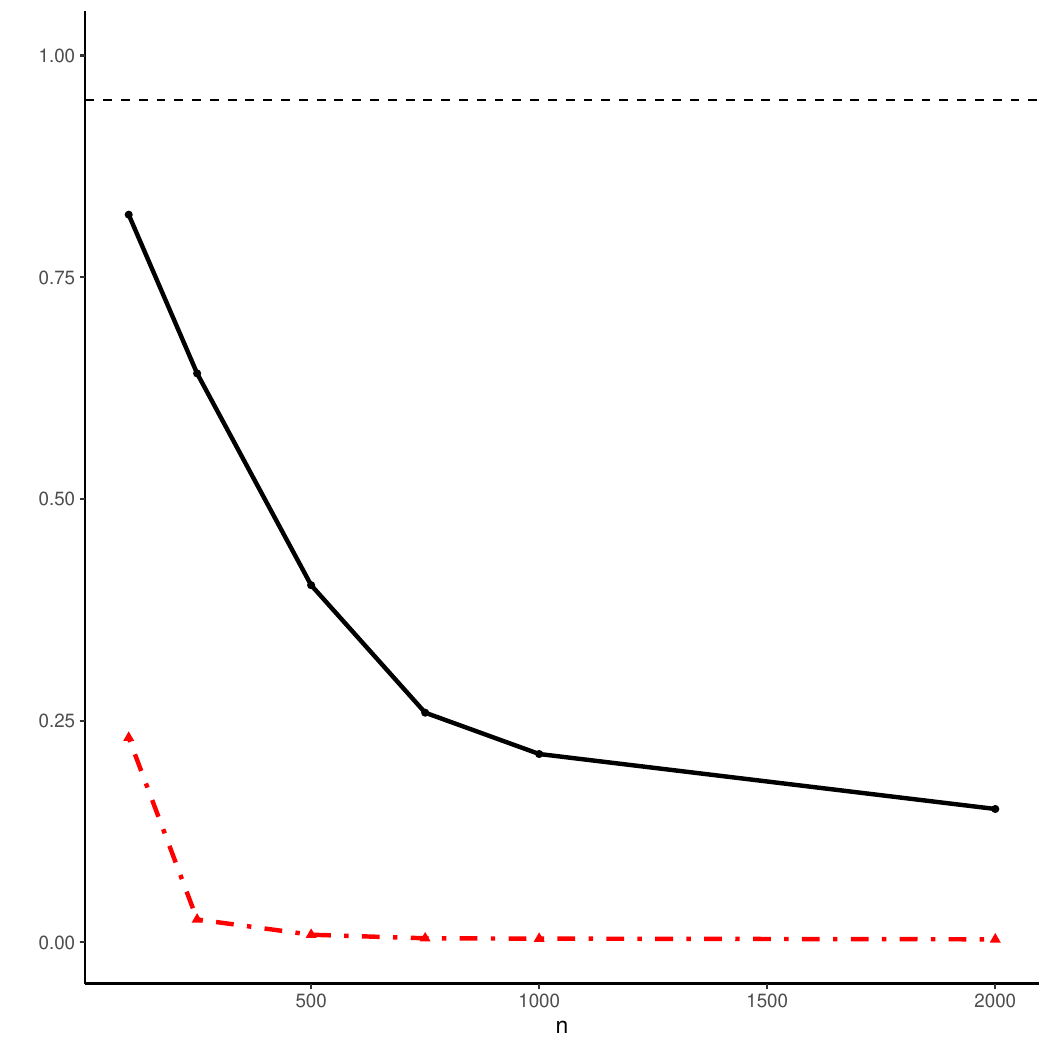}
		\subcaption{$\x=0.2$}
	\end{subfigure}	
	\begin{subfigure}[b]{0.5\textwidth}
		\includegraphics[height=0.3\textheight,width=0.95\textwidth]{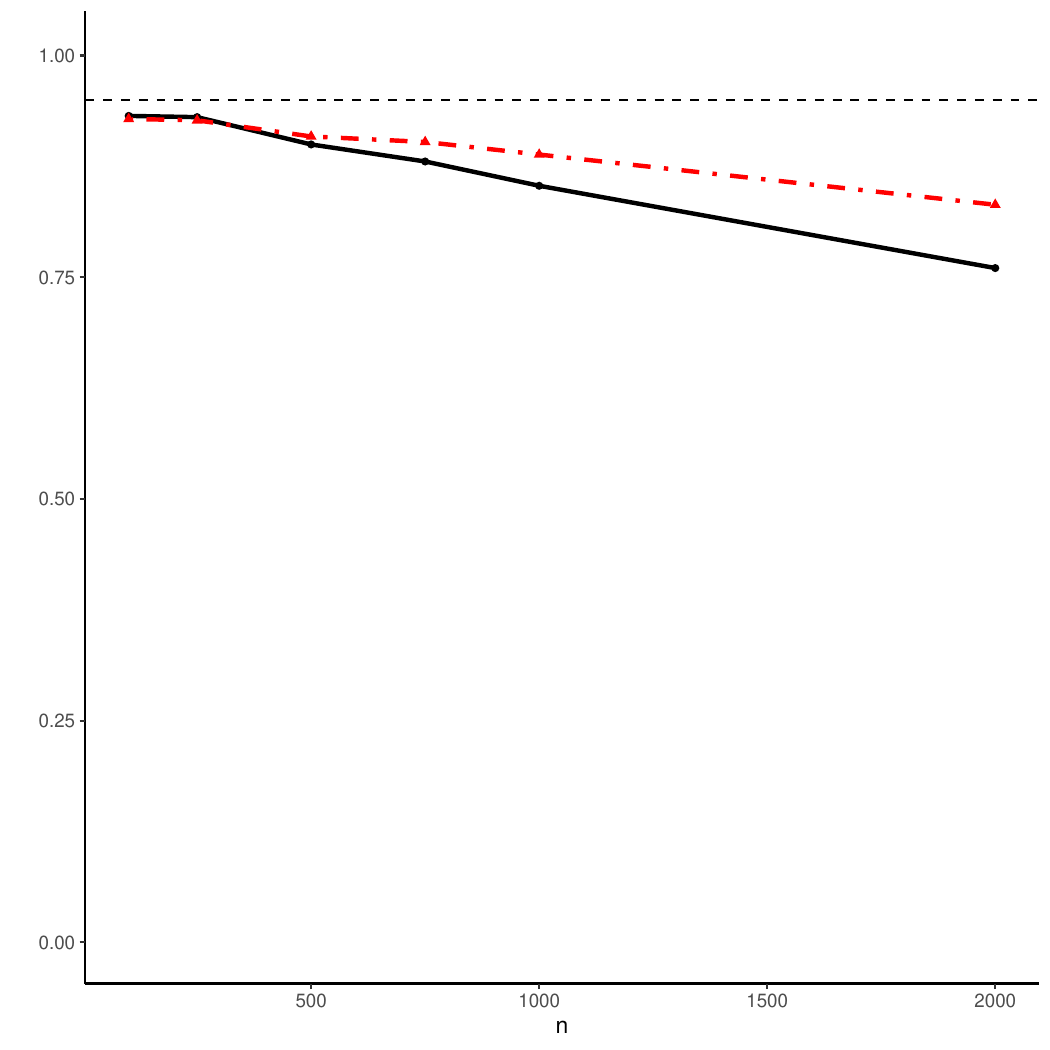}
		\subcaption{$\x=0.6$}
	\end{subfigure}%
	\begin{subfigure}[b]{0.5\textwidth}
		\includegraphics[height=0.3\textheight,width=0.95\textwidth]{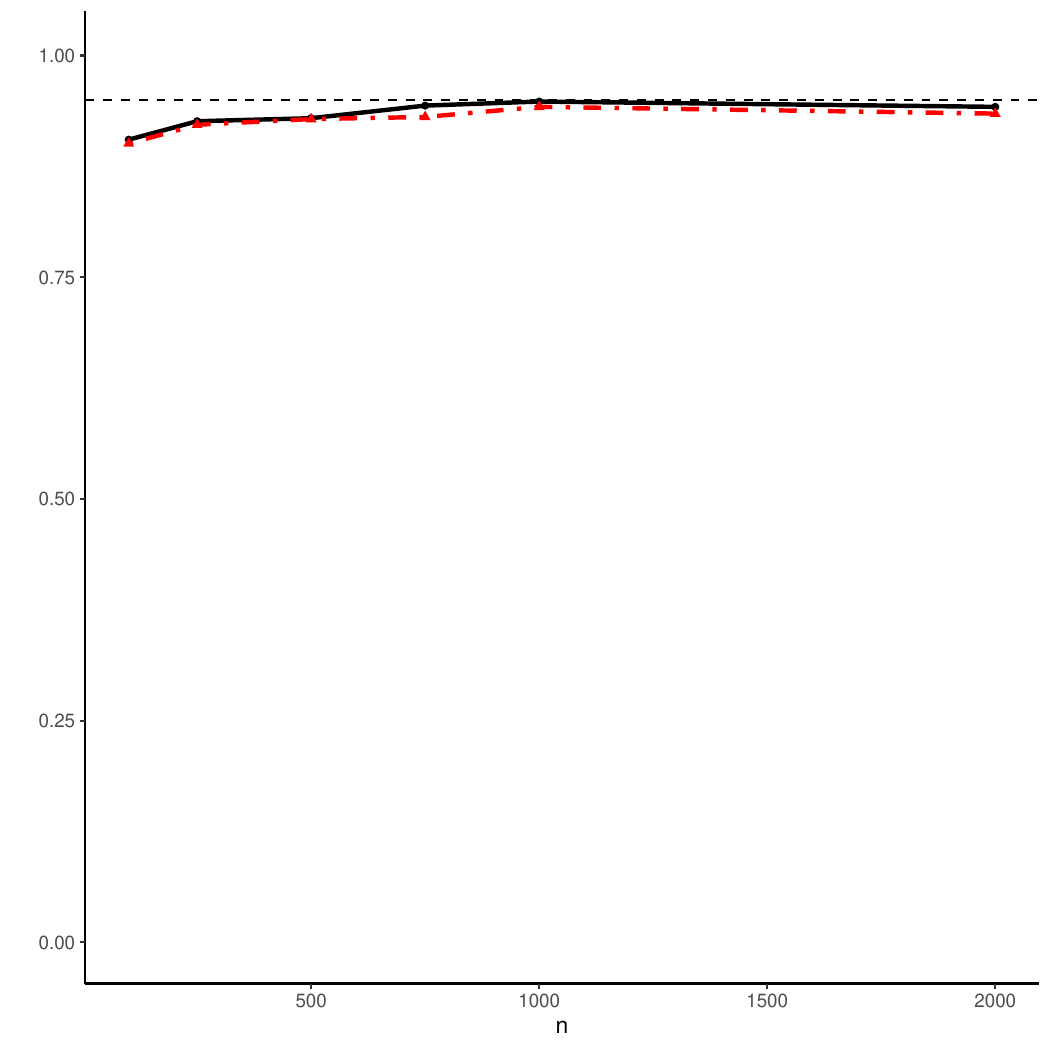}
		\subcaption{$\x=1$}
	\end{subfigure}%
	\begin{flushleft}\footnotesize Notes: \blackline Robust Bias Correction, \redline Undersmoothing
	\end{flushleft}
\end{figure}

\clearpage
\begin{figure}[!htb]
	\centering
	\caption{Empirical Coverage for 95\% Confidence Intervals\\
	Epanechnikov Kernel, $\hat{h}_{\MSE}$, $\v=1$}
	\label{suppfig:ec_nu1_epa_hmse}	
	\begin{subfigure}[b]{0.5\textwidth}
		\includegraphics[height=0.3\textheight,width=0.95\textwidth]{simuls/output/ce_hmse_kepa_p2_d1_x1.pdf}
		\subcaption{$\x=-1$}
	\end{subfigure}%
	\begin{subfigure}[b]{0.5\textwidth}
		\includegraphics[height=0.3\textheight,width=0.95\textwidth]{simuls/output/ce_hmse_kepa_p2_d1_x2.pdf}
		\subcaption{$\x=-0.6$}
	\end{subfigure}	
	\begin{subfigure}[b]{0.5\textwidth}
		\includegraphics[height=0.3\textheight,width=0.95\textwidth]{simuls/output/ce_hmse_kepa_p2_d1_x3.pdf}
		\subcaption{$\x=-0.2$}
	\end{subfigure}%	
	\begin{subfigure}[b]{0.5\textwidth}
		\includegraphics[height=0.3\textheight,width=0.95\textwidth]{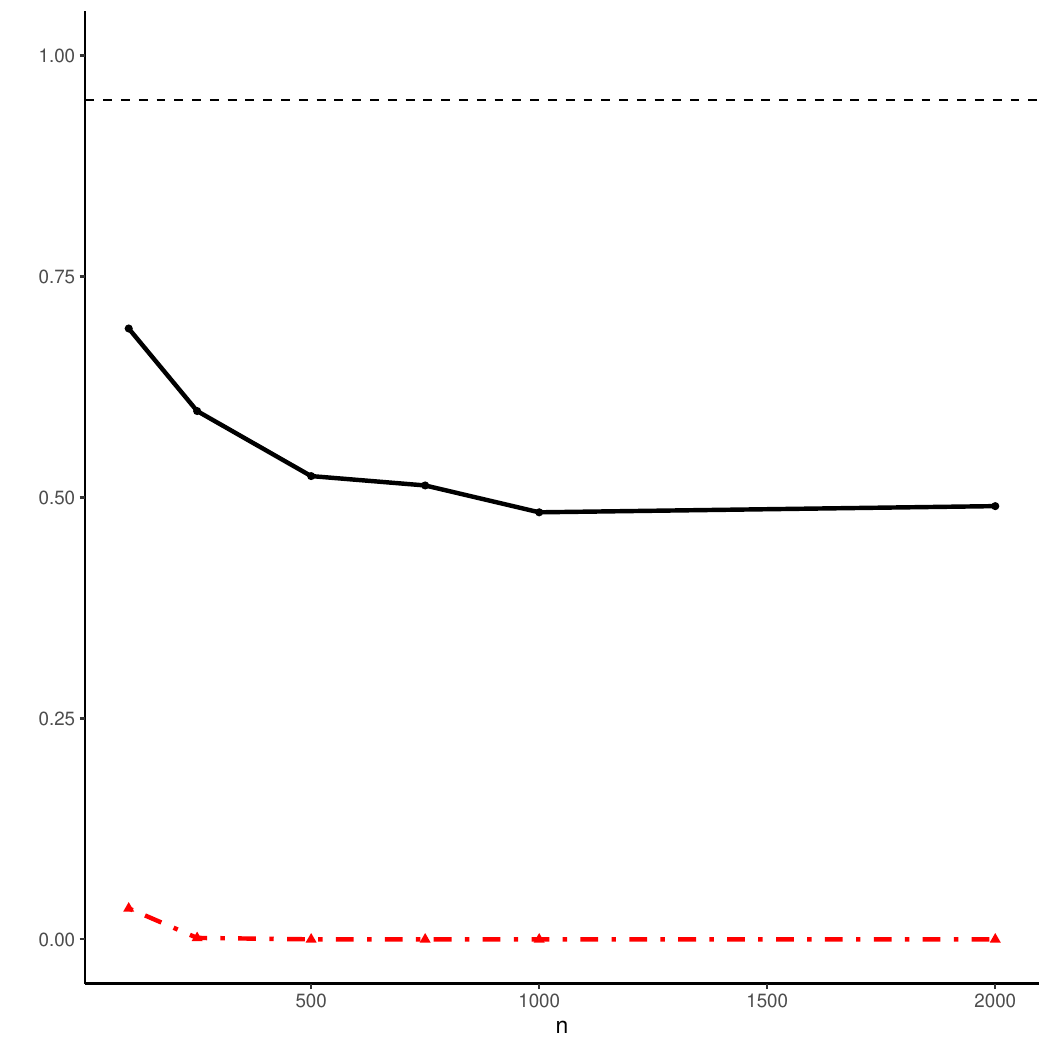}
		\subcaption{$\x=0.2$}
	\end{subfigure}	
	\begin{subfigure}[b]{0.5\textwidth}
		\includegraphics[height=0.3\textheight,width=0.95\textwidth]{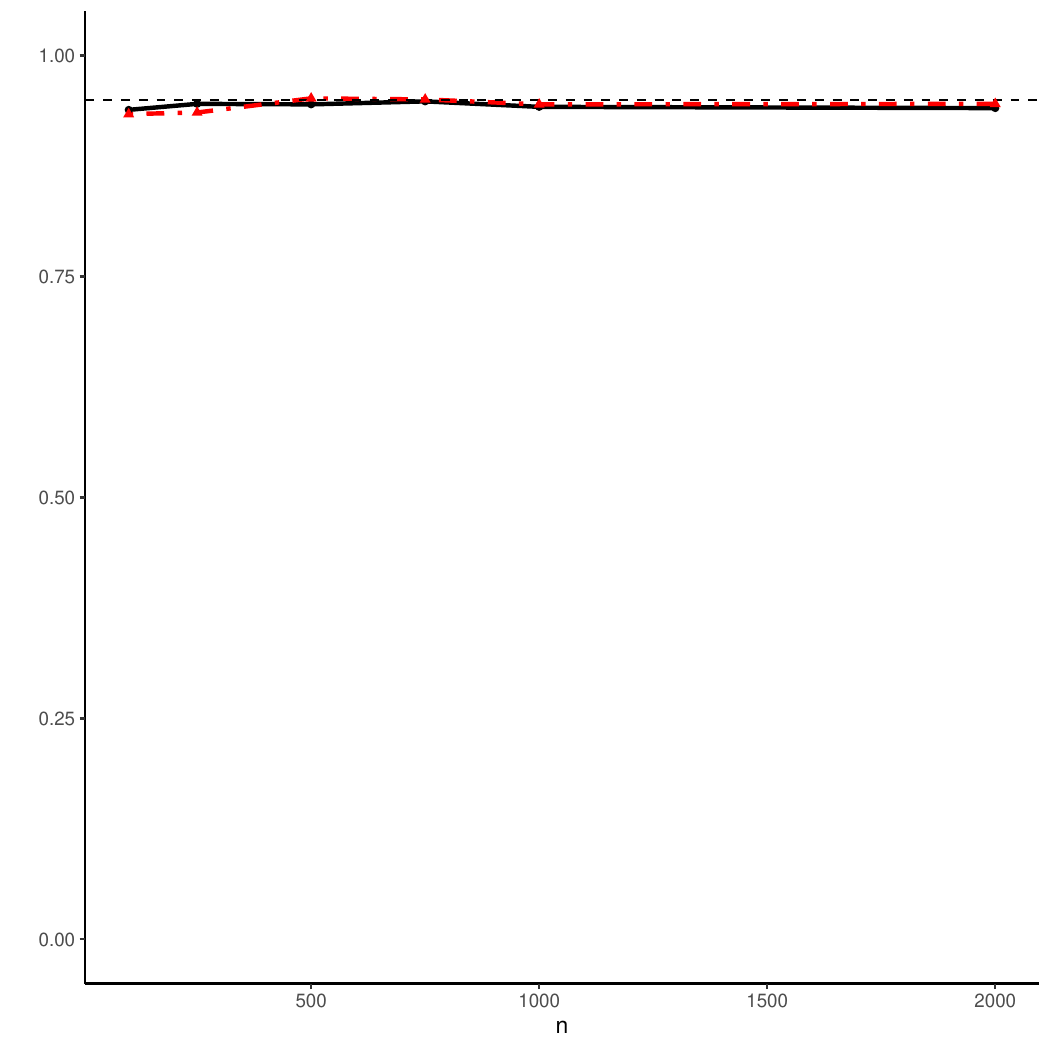}
		\subcaption{$\x=0.6$}
	\end{subfigure}%
	\begin{subfigure}[b]{0.5\textwidth}
		\includegraphics[height=0.3\textheight,width=0.95\textwidth]{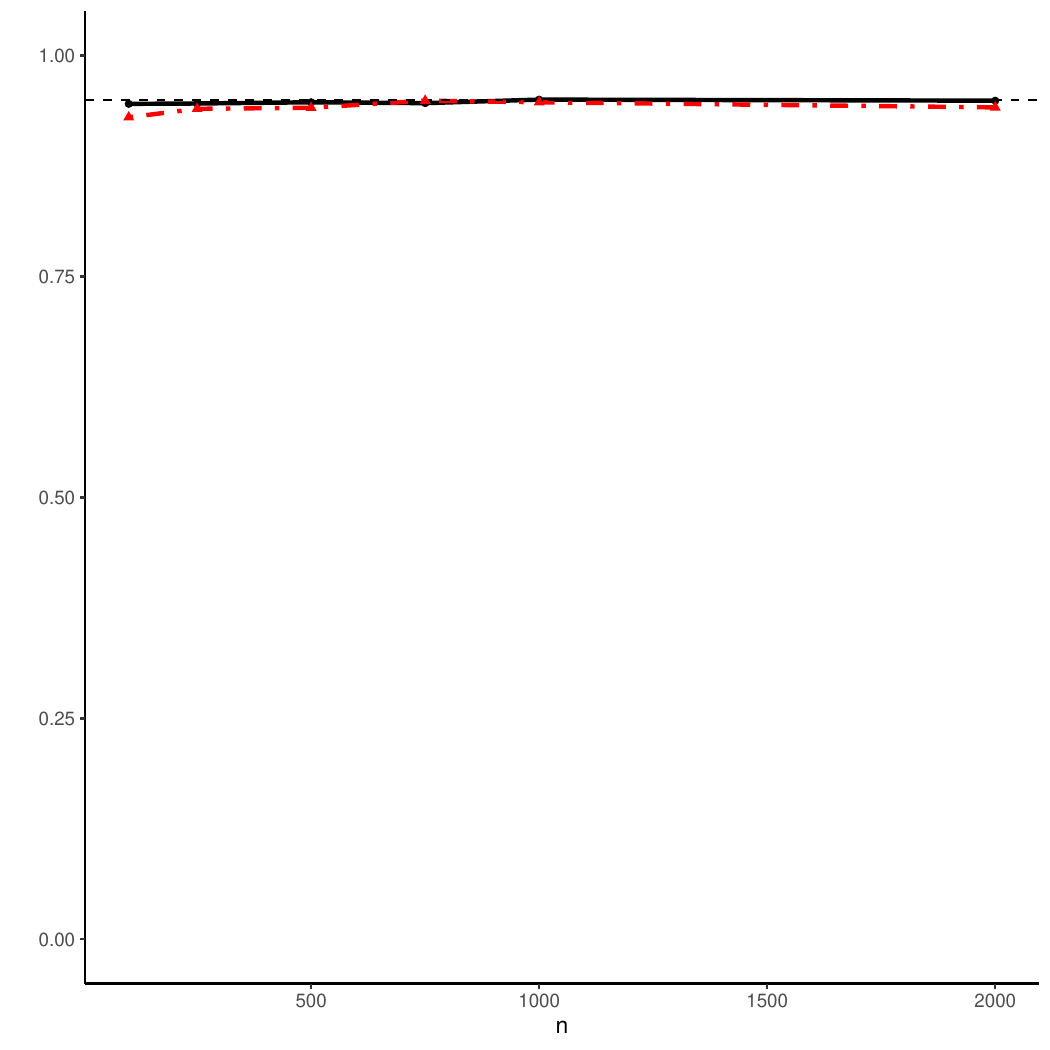}
		\subcaption{$\x=1$}
	\end{subfigure}%
	\begin{flushleft}\footnotesize Notes: \blackline Robust Bias Correction, \redline Undersmoothing
	\end{flushleft}
\end{figure}

%%%%%%%%%%%%%%%%%%%%%%%% EC UNI

%%% h_RBC
\clearpage

\begin{figure}[!htb]
	\centering
	\caption{Empirical Coverage for 95\% Confidence Intervals\\ Uniform Kernel, $\hat{h}_{\RBC}$, $\v=0$}
	\label{suppfig:ec_nu0_uni_hrbc}
	\begin{subfigure}[b]{0.5\textwidth}
		\includegraphics[height=0.3\textheight,width=0.95\textwidth]{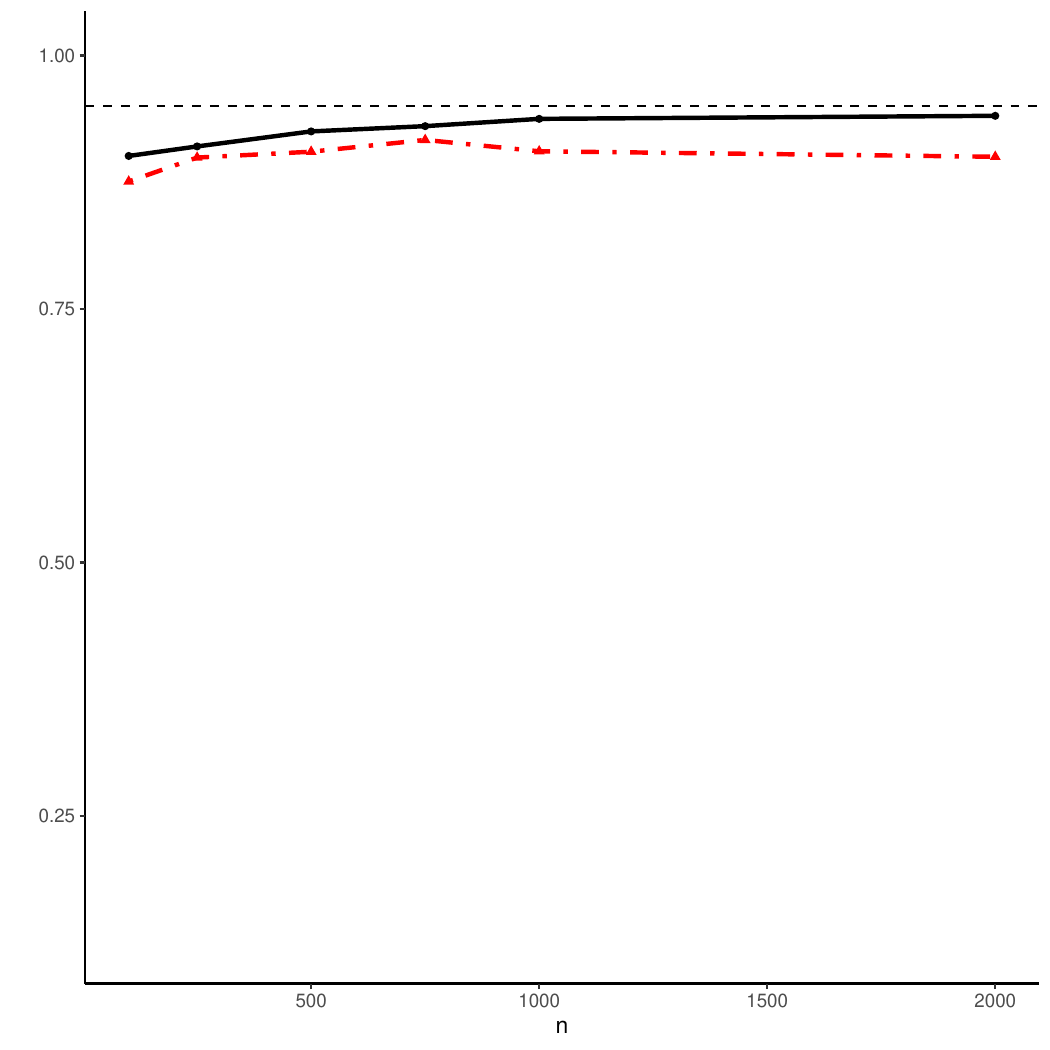}
		\subcaption{$\x=-1$}
	\end{subfigure}%
	\begin{subfigure}[b]{0.5\textwidth}
		\includegraphics[height=0.3\textheight,width=0.95\textwidth]{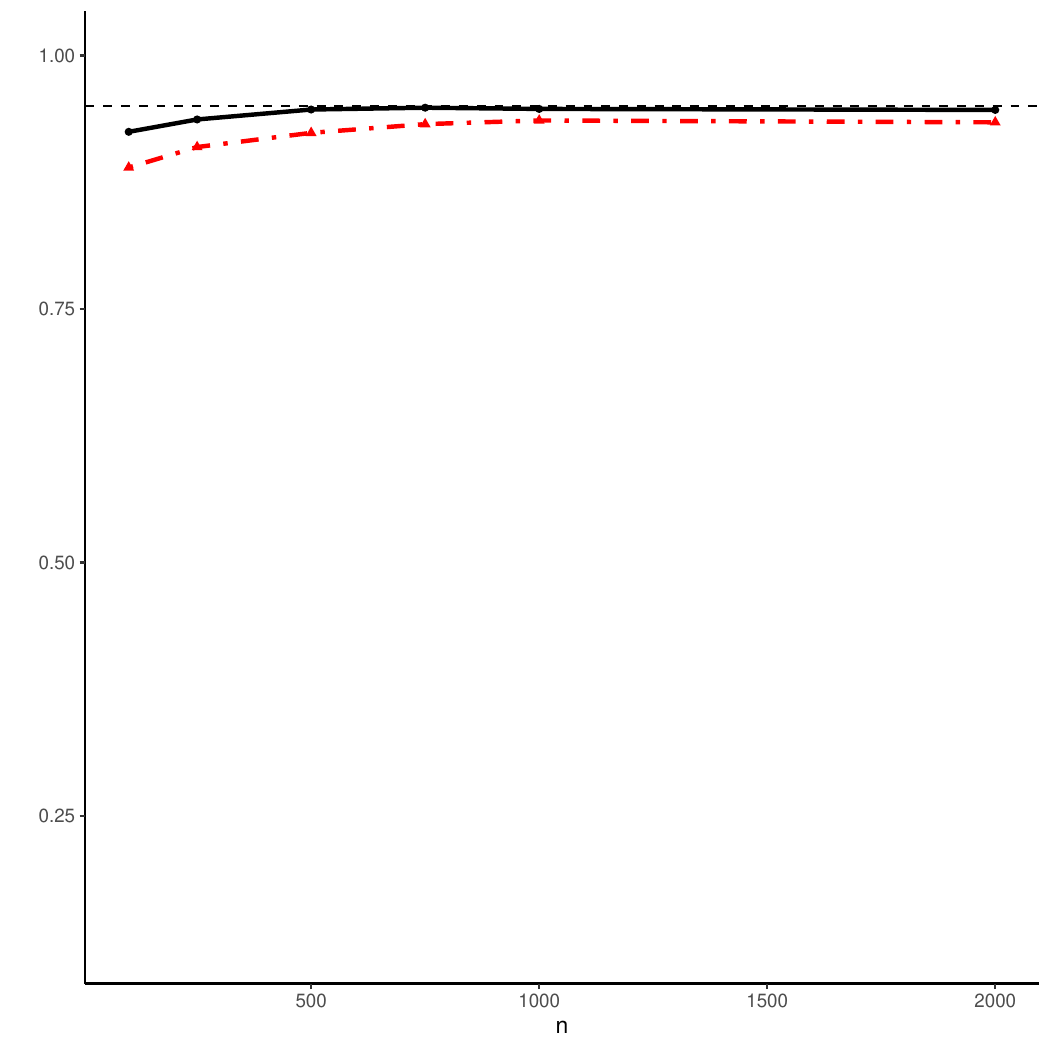}
		\subcaption{$\x=-0.6$}
	\end{subfigure}	
	\begin{subfigure}[b]{0.5\textwidth}
		\includegraphics[height=0.3\textheight,width=0.95\textwidth]{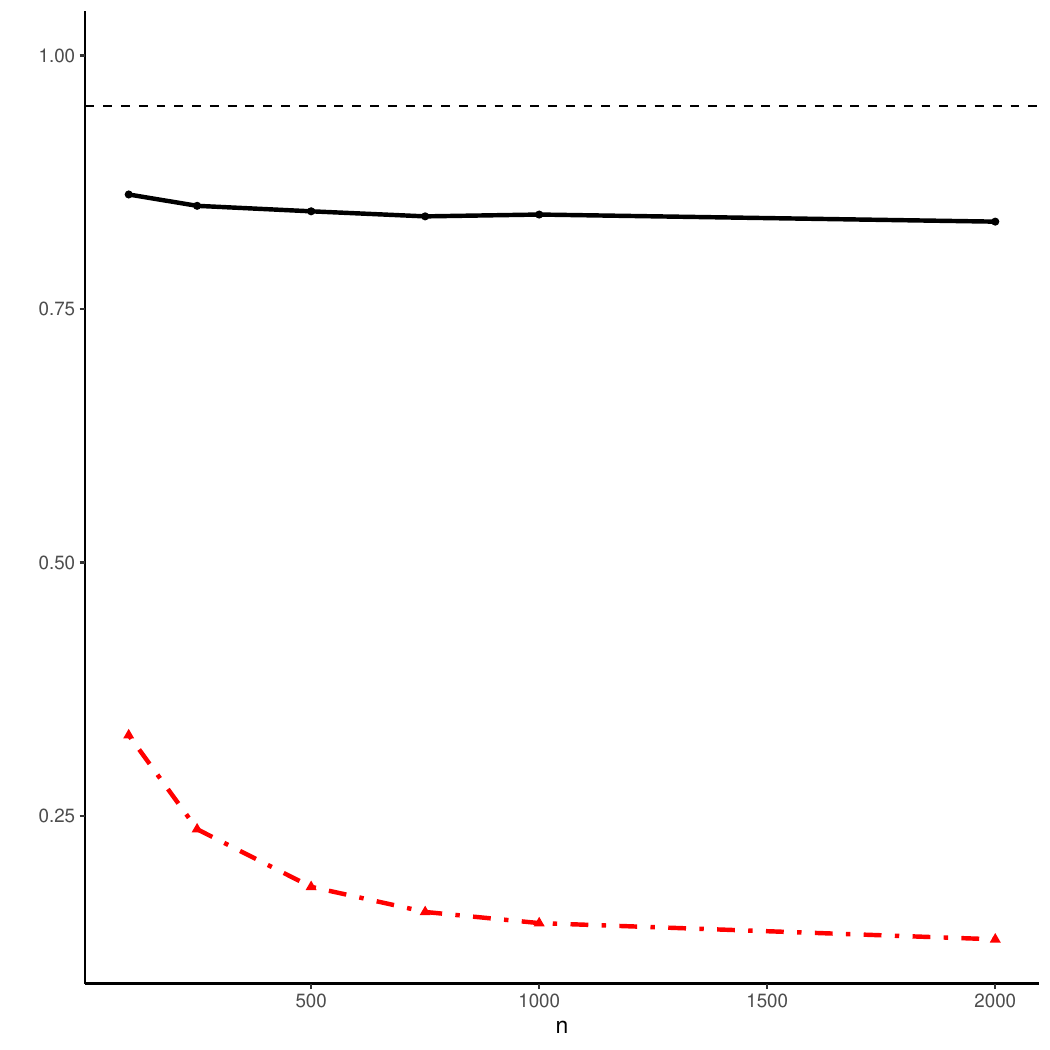}
		\subcaption{$\x=-0.2$}
	\end{subfigure}%	
	\begin{subfigure}[b]{0.5\textwidth}
		\includegraphics[height=0.3\textheight,width=0.95\textwidth]{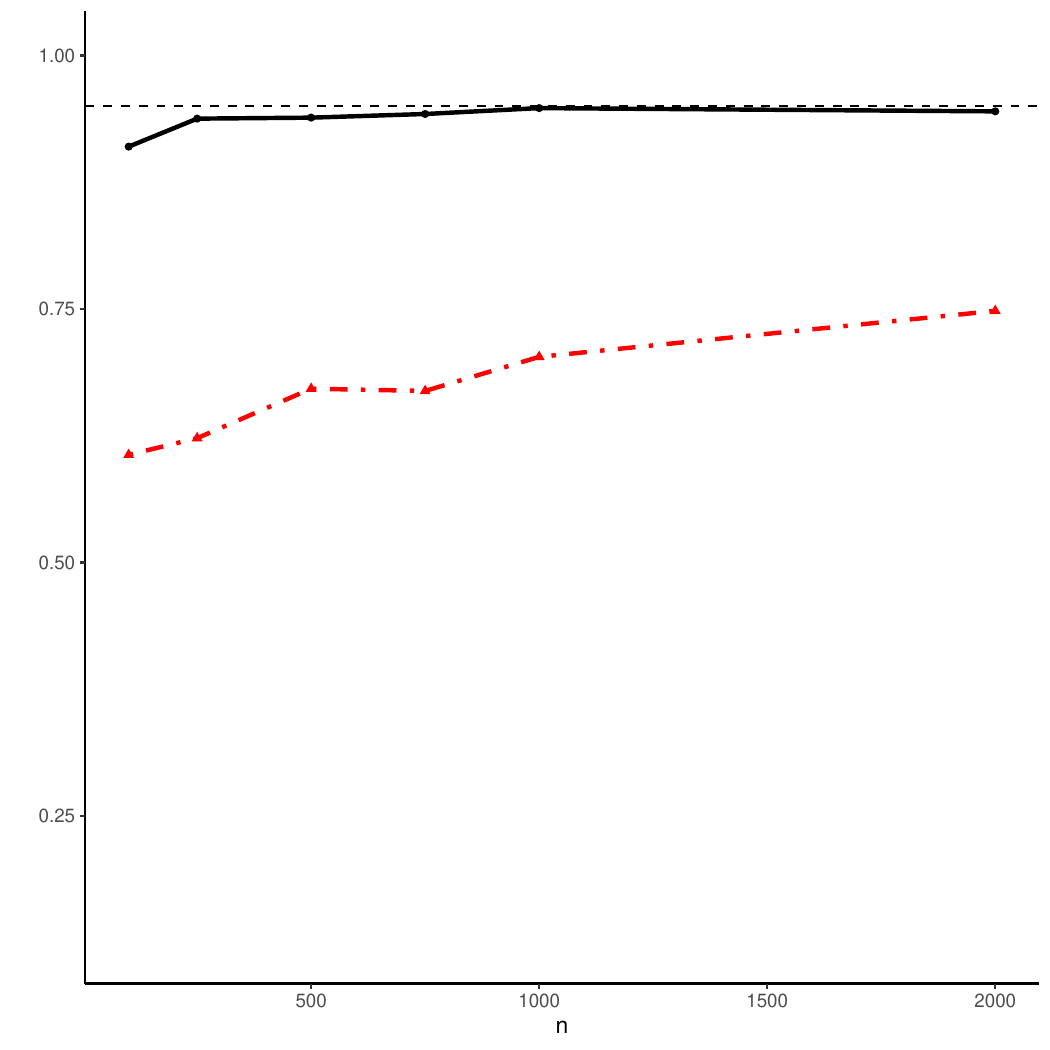}
		\subcaption{$\x=0.2$}
	\end{subfigure}	
	\begin{subfigure}[b]{0.5\textwidth}
		\includegraphics[height=0.3\textheight,width=0.95\textwidth]{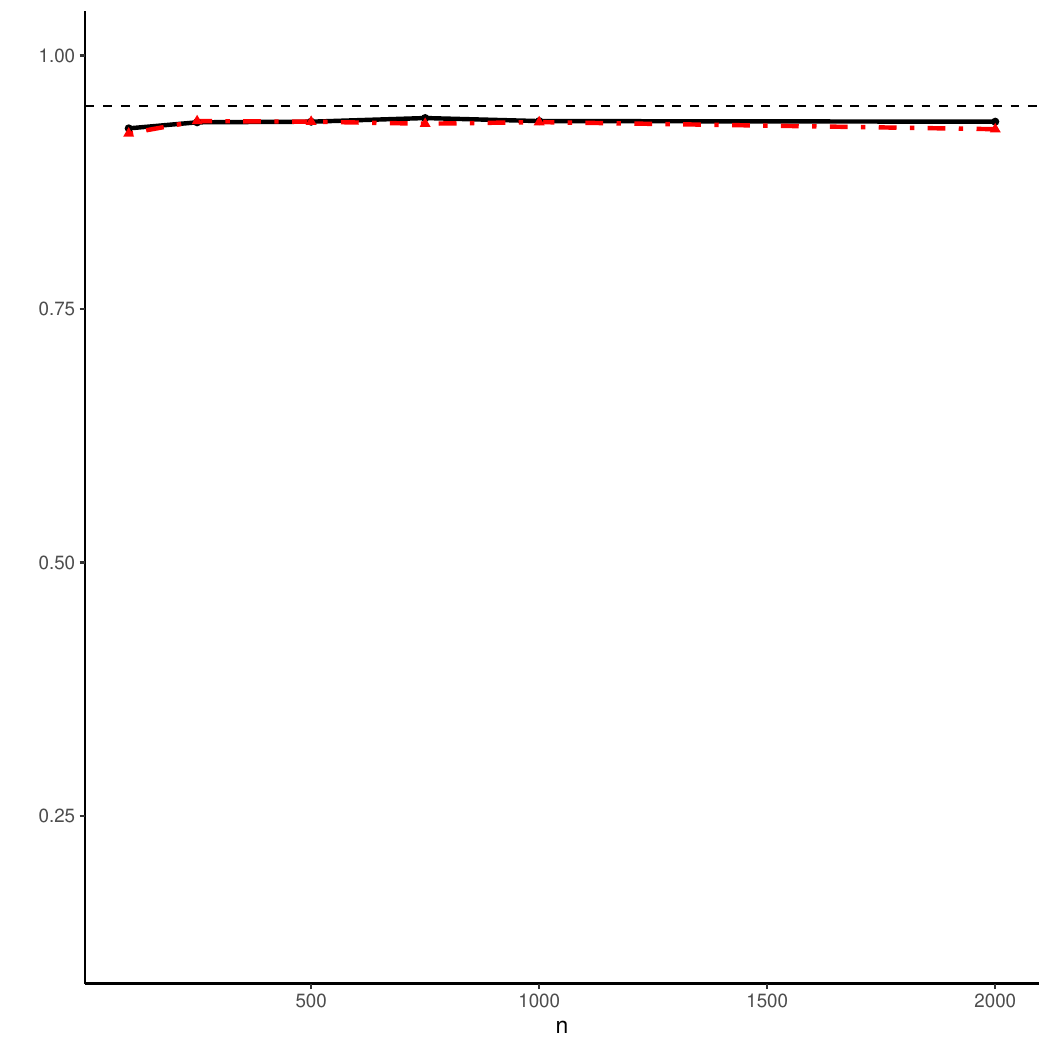}
		\subcaption{$\x=0.6$}
	\end{subfigure}%
	\begin{subfigure}[b]{0.5\textwidth}
		\includegraphics[height=0.3\textheight,width=0.95\textwidth]{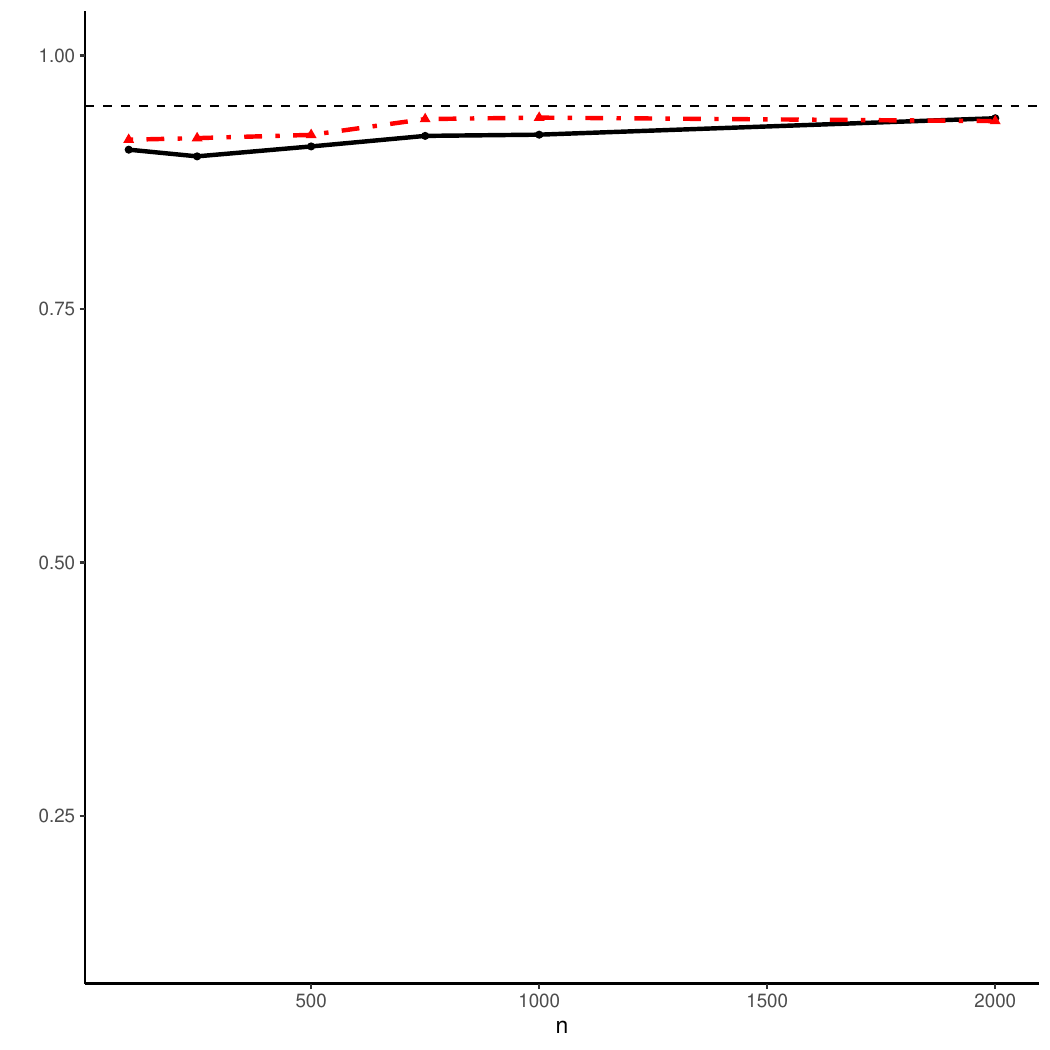}
		\subcaption{$\x=1$}
	\end{subfigure}%
	\begin{flushleft}\footnotesize Notes: \blackline Robust Bias Correction, \redline Undersmoothing
	\end{flushleft}
\end{figure}

\clearpage
\begin{figure}[!htb]
	\centering
	\caption{Empirical Coverage for 95\% Confidence Intervals\\
	Uniform Kernel, $\hat{h}_{\RBC}$, $\v=1$}
	\label{suppfig:ec_nu1_uni_hrbc}	
	\begin{subfigure}[b]{0.5\textwidth}
		\includegraphics[height=0.3\textheight,width=0.95\textwidth]{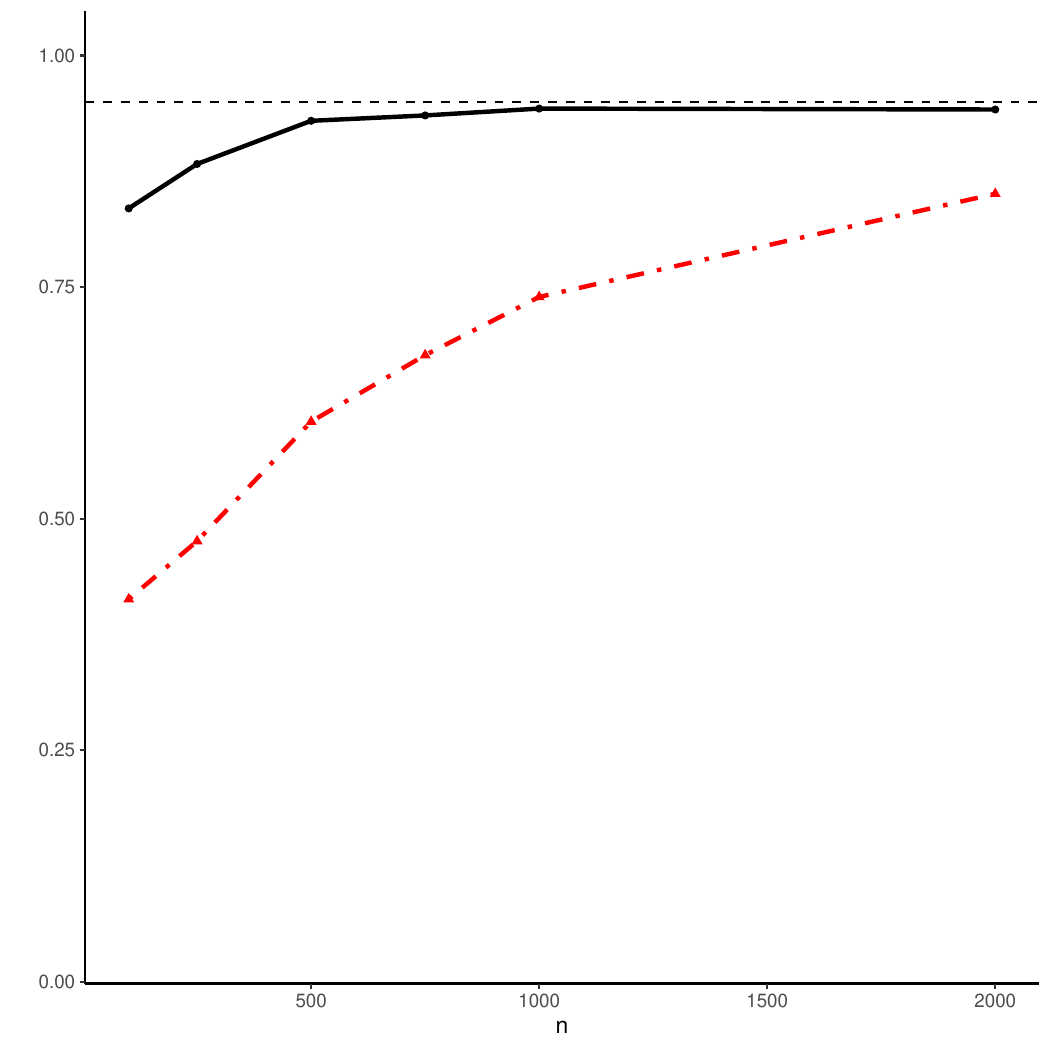}
		\subcaption{$\x=-1$}
	\end{subfigure}%
	\begin{subfigure}[b]{0.5\textwidth}
		\includegraphics[height=0.3\textheight,width=0.95\textwidth]{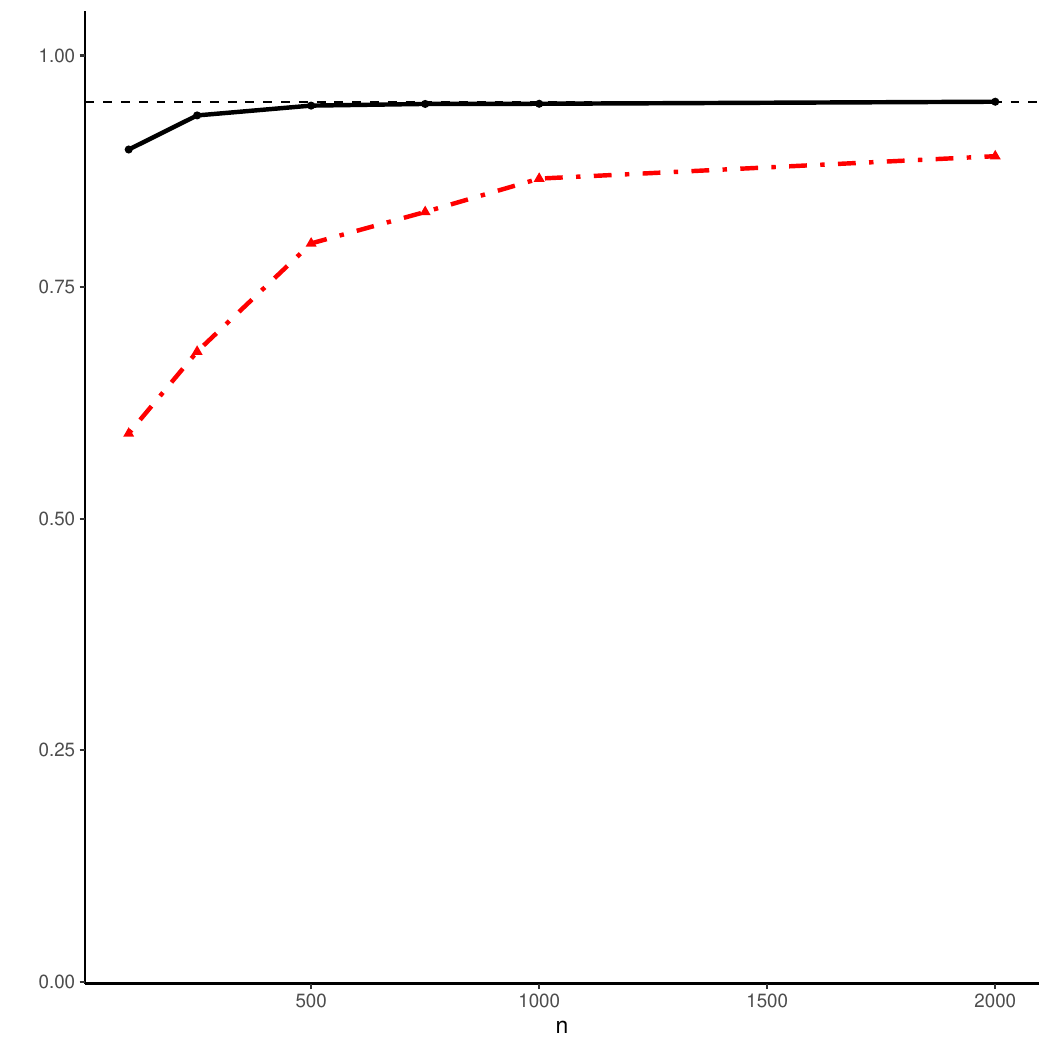}
		\subcaption{$\x=-0.6$}
	\end{subfigure}	
	\begin{subfigure}[b]{0.5\textwidth}
		\includegraphics[height=0.3\textheight,width=0.95\textwidth]{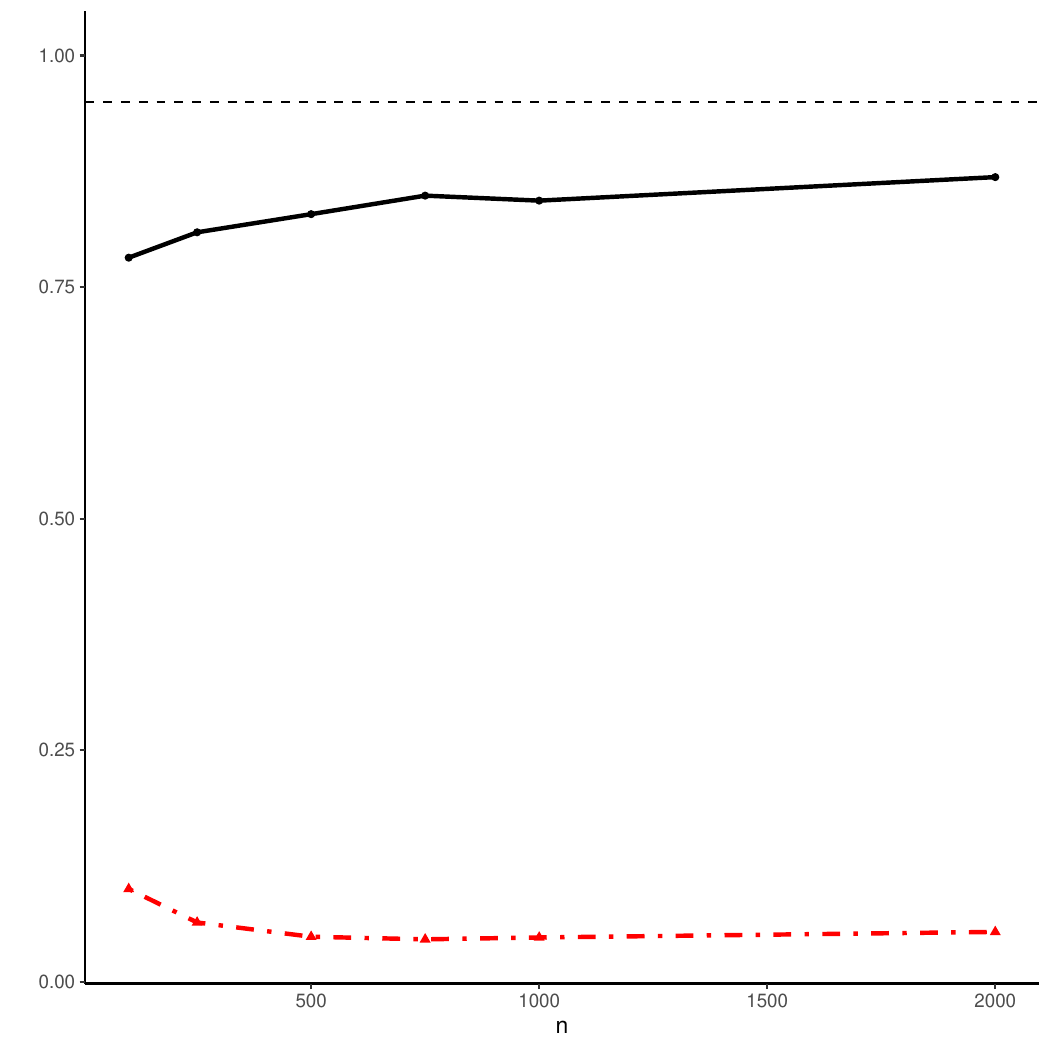}
		\subcaption{$\x=-0.2$}
	\end{subfigure}%	
	\begin{subfigure}[b]{0.5\textwidth}
		\includegraphics[height=0.3\textheight,width=0.95\textwidth]{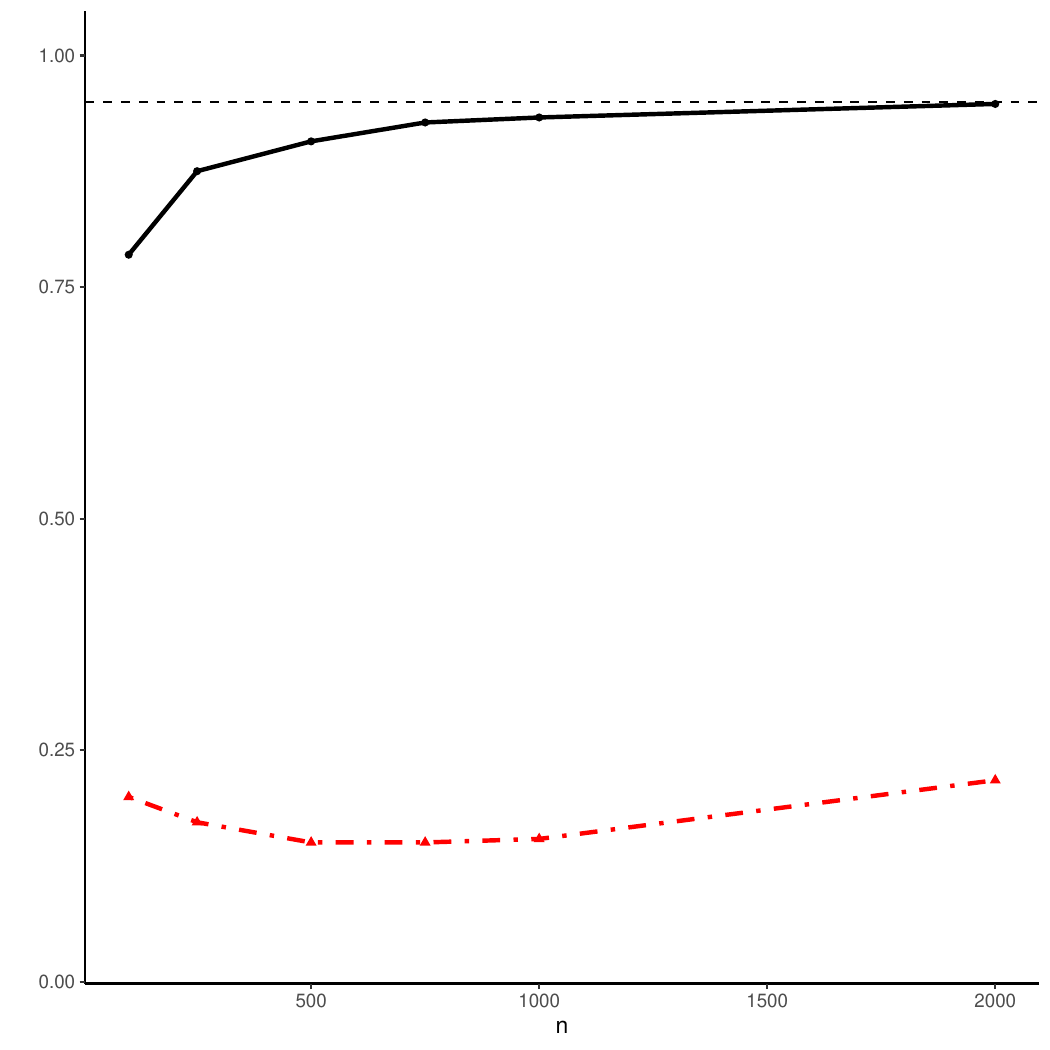}
		\subcaption{$\x=0.2$}
	\end{subfigure}	
	\begin{subfigure}[b]{0.5\textwidth}
		\includegraphics[height=0.3\textheight,width=0.95\textwidth]{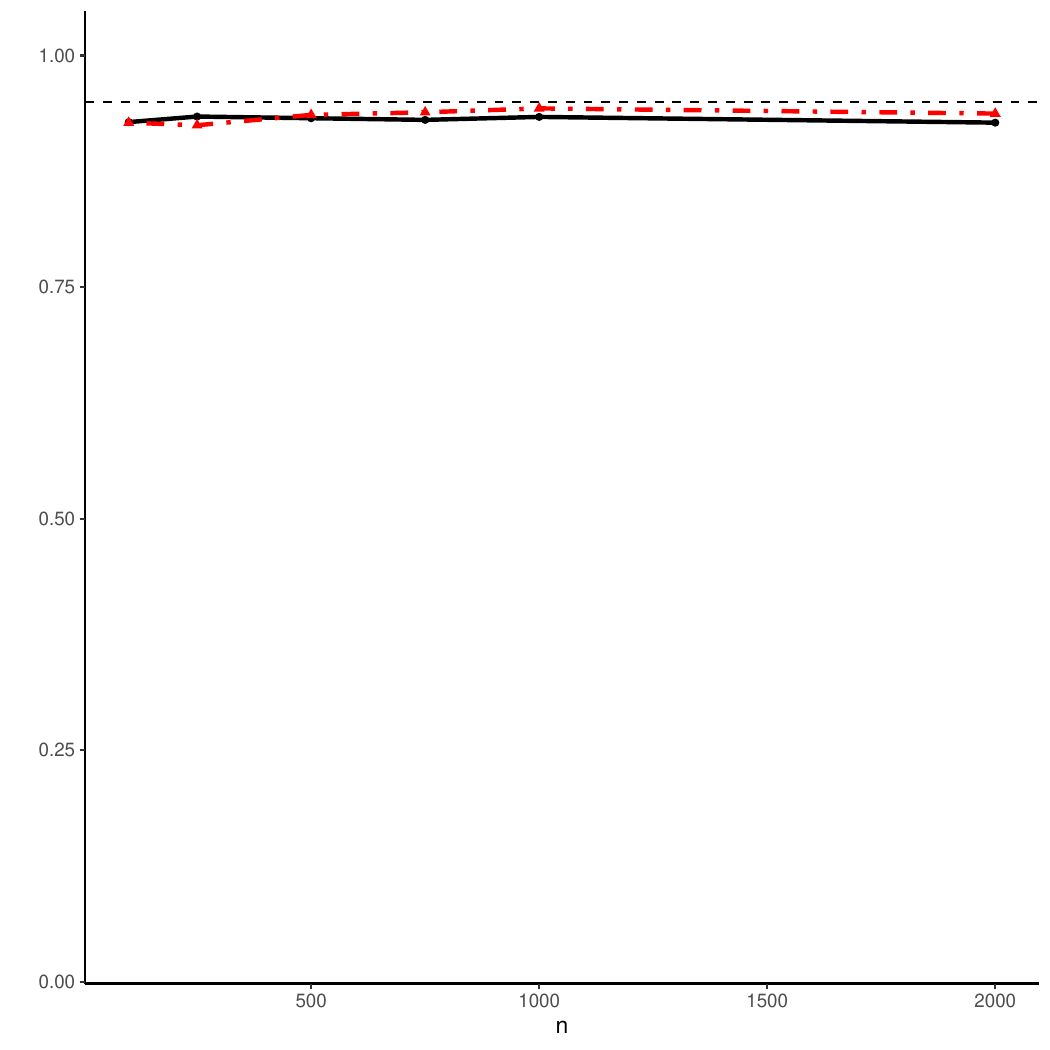}
		\subcaption{$\x=0.6$}
	\end{subfigure}%
	\begin{subfigure}[b]{0.5\textwidth}
		\includegraphics[height=0.3\textheight,width=0.95\textwidth]{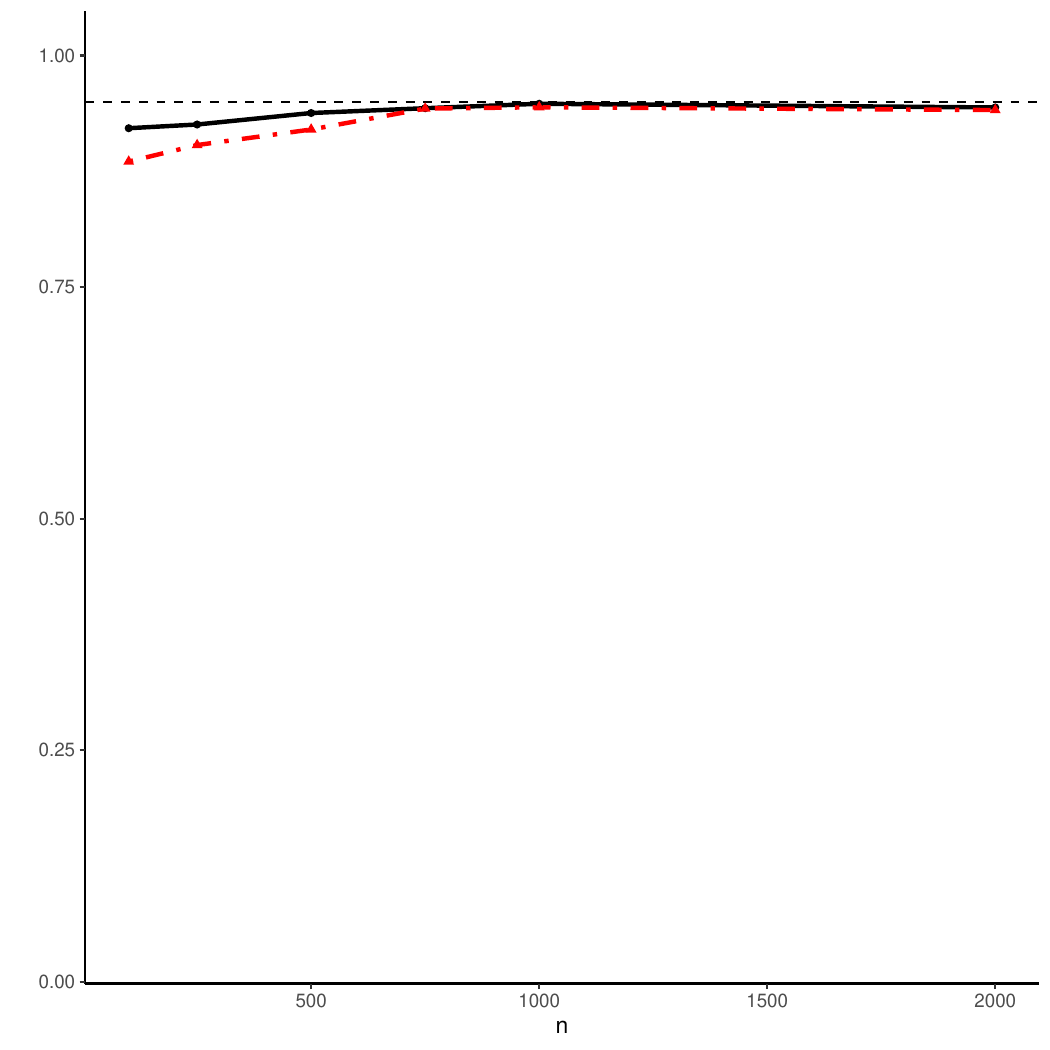}
		\subcaption{$\x=1$}
	\end{subfigure}%
	\begin{flushleft}\footnotesize Notes: \blackline Robust Bias Correction, \redline Undersmoothing
	\end{flushleft}
\end{figure}

%%% h_US
\clearpage
\begin{figure}[!htb]	
	\centering
	\caption{Empirical Coverage for 95\% Confidence Intervals\\
	Uniform Kernel, $\hat{h}_{\US}$, $\v=0$}
	\label{suppfig:ec_nu0_uni_hus}	
	\begin{subfigure}[b]{0.5\textwidth}
		\includegraphics[height=0.3\textheight,width=0.95\textwidth]{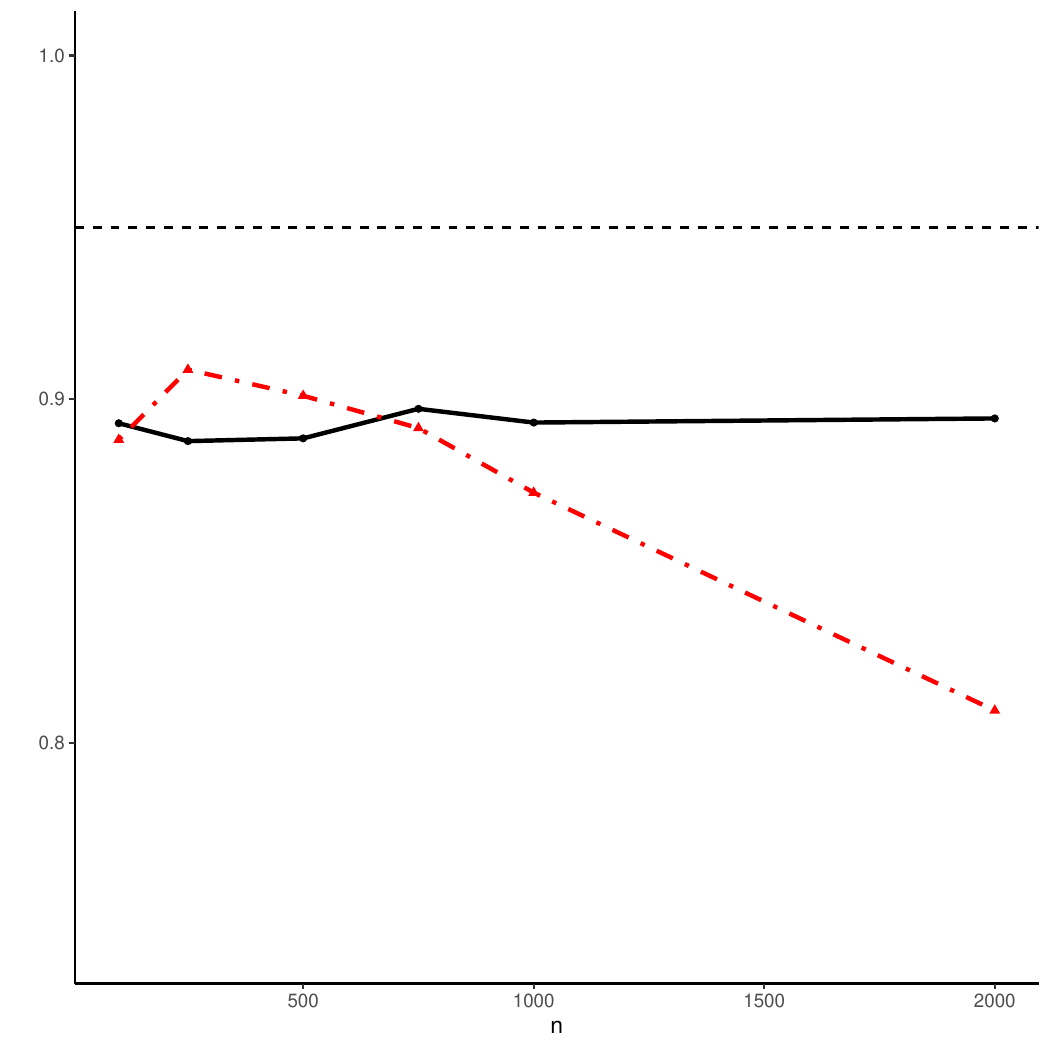}
		\subcaption{$\x=-1$}
	\end{subfigure}%
	\begin{subfigure}[b]{0.5\textwidth}
		\includegraphics[height=0.3\textheight,width=0.95\textwidth]{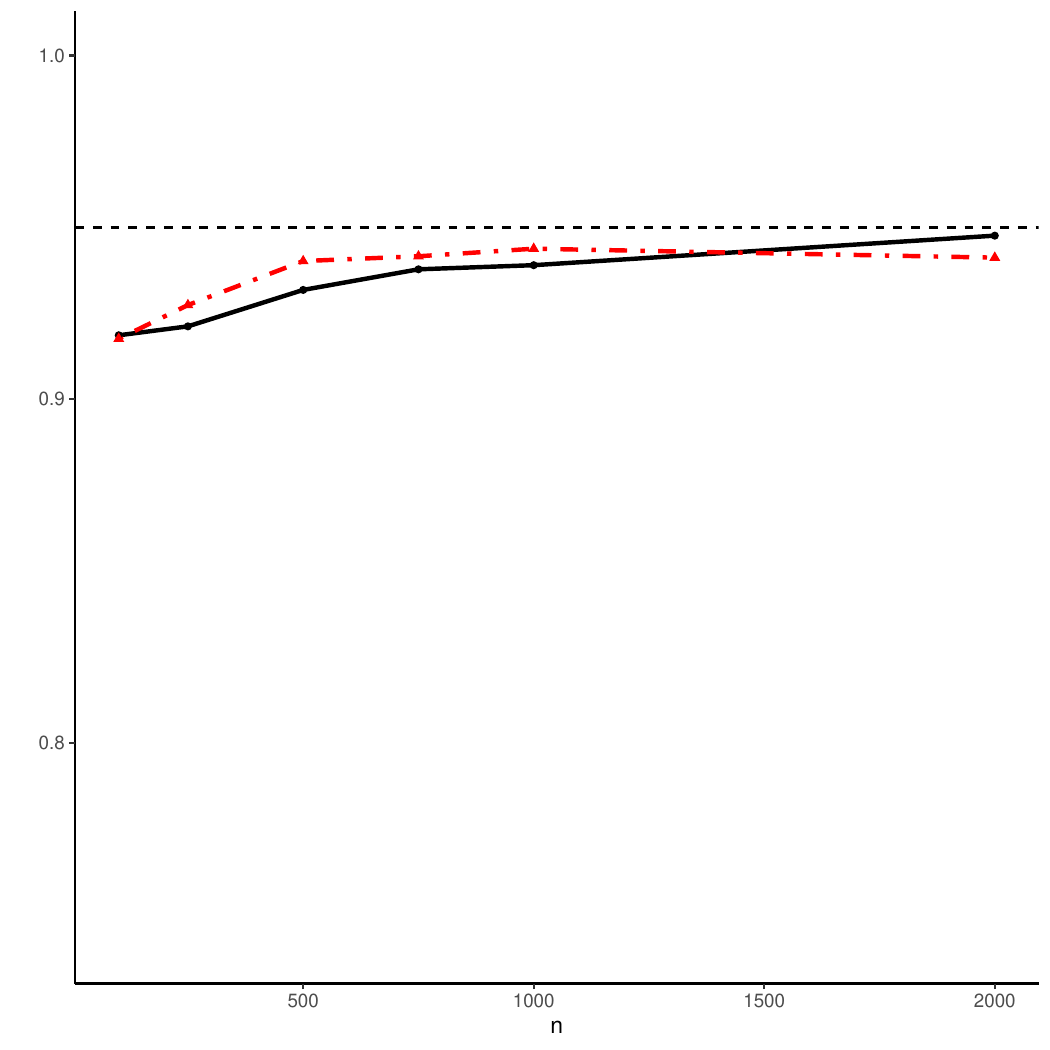}
		\subcaption{$\x=-0.6$}
	\end{subfigure}	
	\begin{subfigure}[b]{0.5\textwidth}
		\includegraphics[height=0.3\textheight,width=0.95\textwidth]{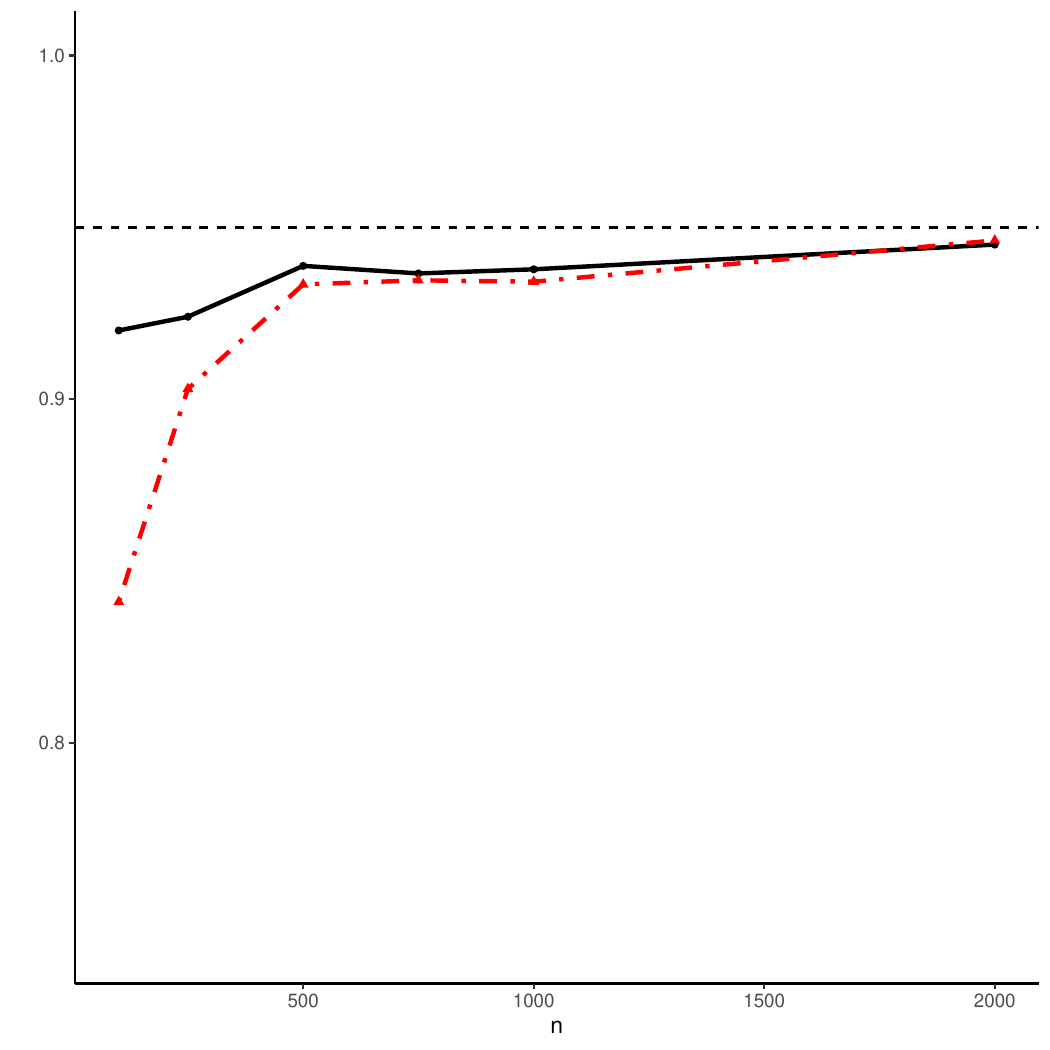}
		\subcaption{$\x=-0.2$}
	\end{subfigure}%	
	\begin{subfigure}[b]{0.5\textwidth}
		\includegraphics[height=0.3\textheight,width=0.95\textwidth]{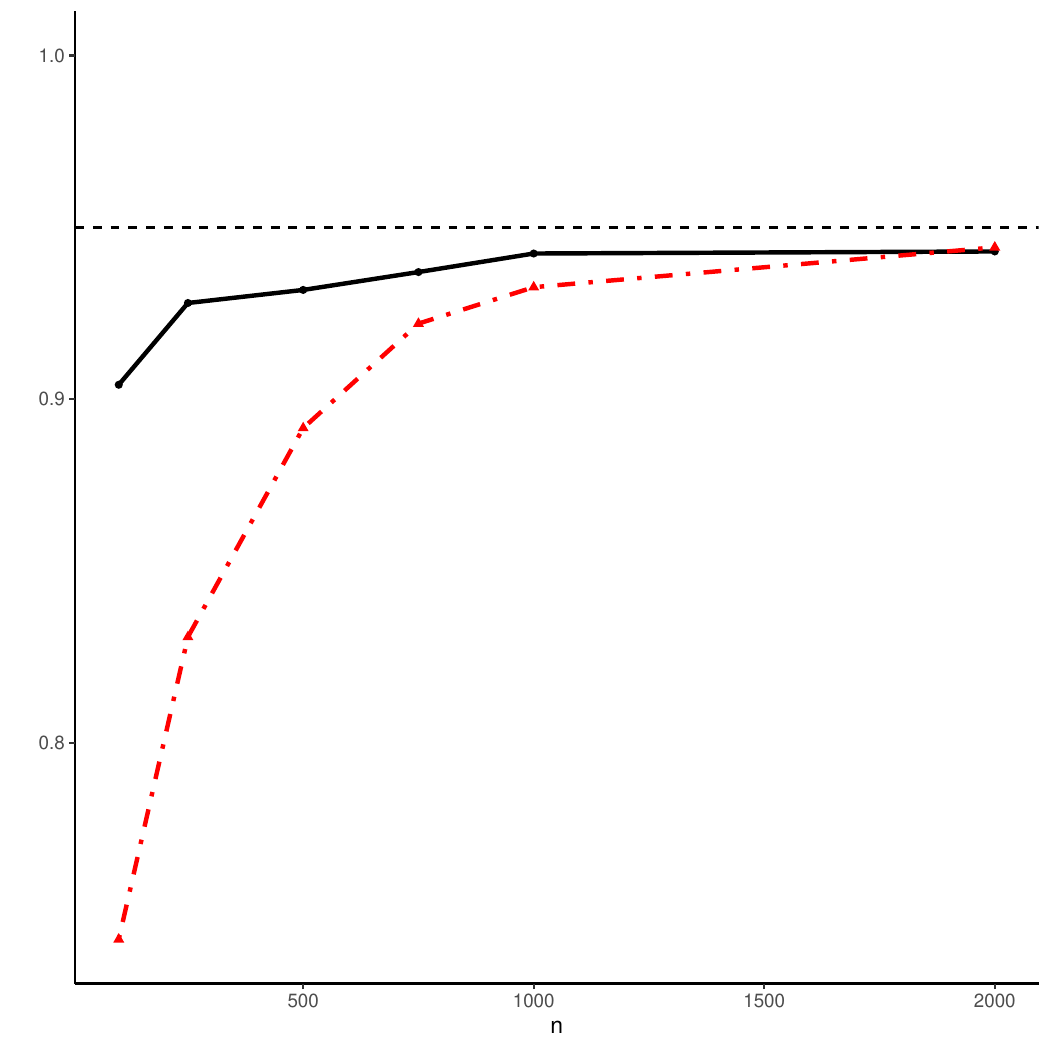}
		\subcaption{$\x=0.2$}
	\end{subfigure}	
	\begin{subfigure}[b]{0.5\textwidth}
		\includegraphics[height=0.3\textheight,width=0.95\textwidth]{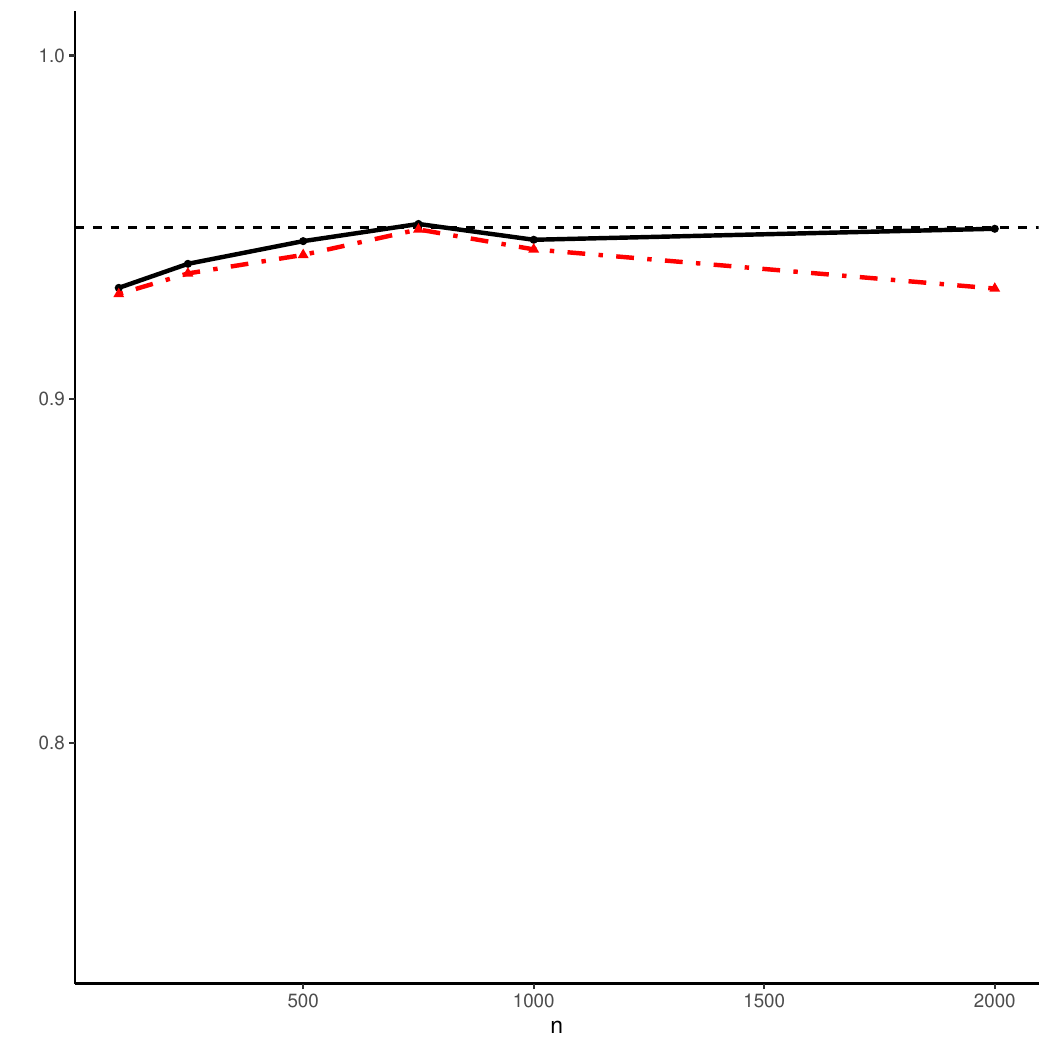}
		\subcaption{$\x=0.6$}
	\end{subfigure}%
	\begin{subfigure}[b]{0.5\textwidth}
		\includegraphics[height=0.3\textheight,width=0.95\textwidth]{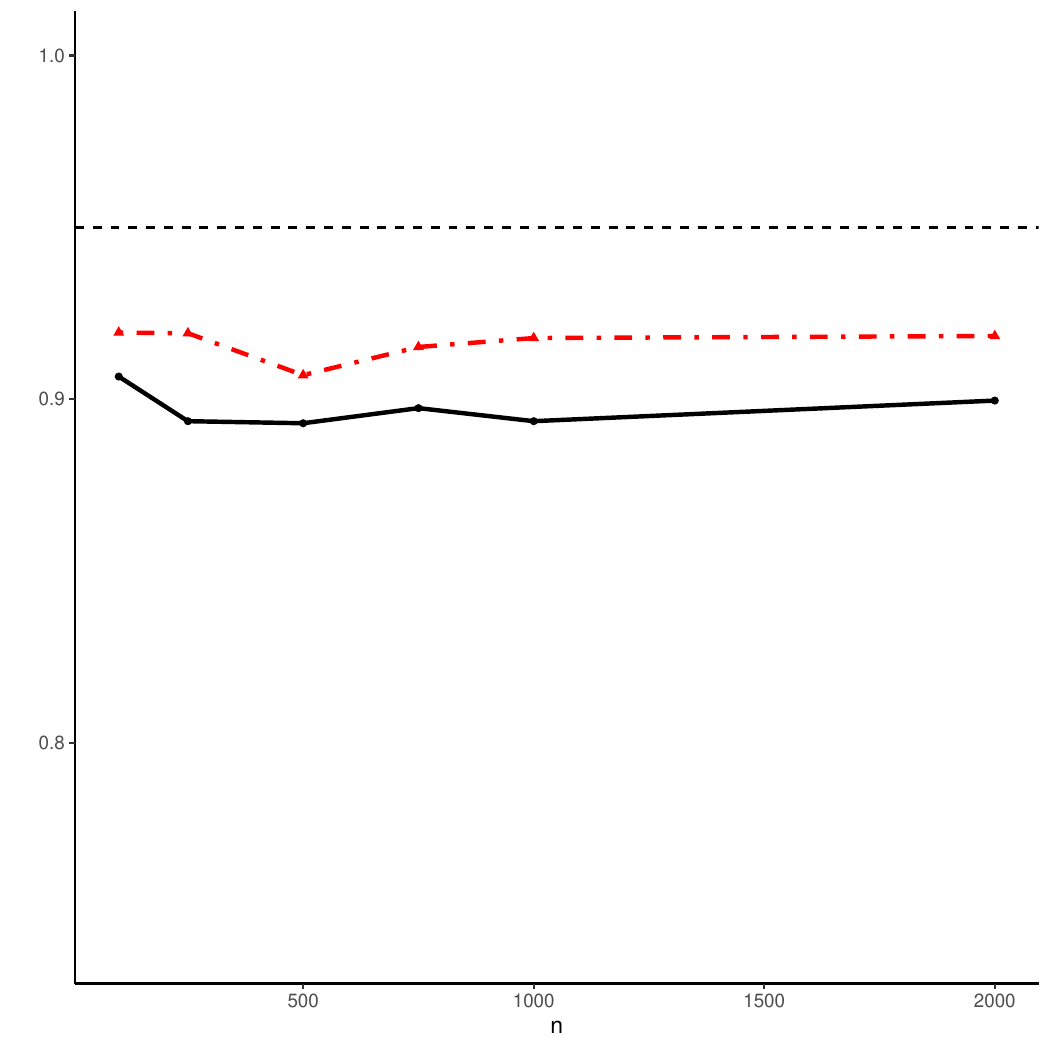}
		\subcaption{$\x=1$}
	\end{subfigure}%
	\begin{flushleft}\footnotesize Notes: \blackline Robust Bias Correction, \redline Undersmoothing
	\end{flushleft}
\end{figure}

\clearpage
\begin{figure}[!htb]
	\centering
	\caption{Empirical Coverage for 95\% Confidence Intervals\\
	Uniform Kernel, $\hat{h}_{\US}$, $\v=1$}
	\label{suppfig:ec_nu1_uni_hus}	
	\begin{subfigure}[b]{0.5\textwidth}
		\includegraphics[height=0.3\textheight,width=0.95\textwidth]{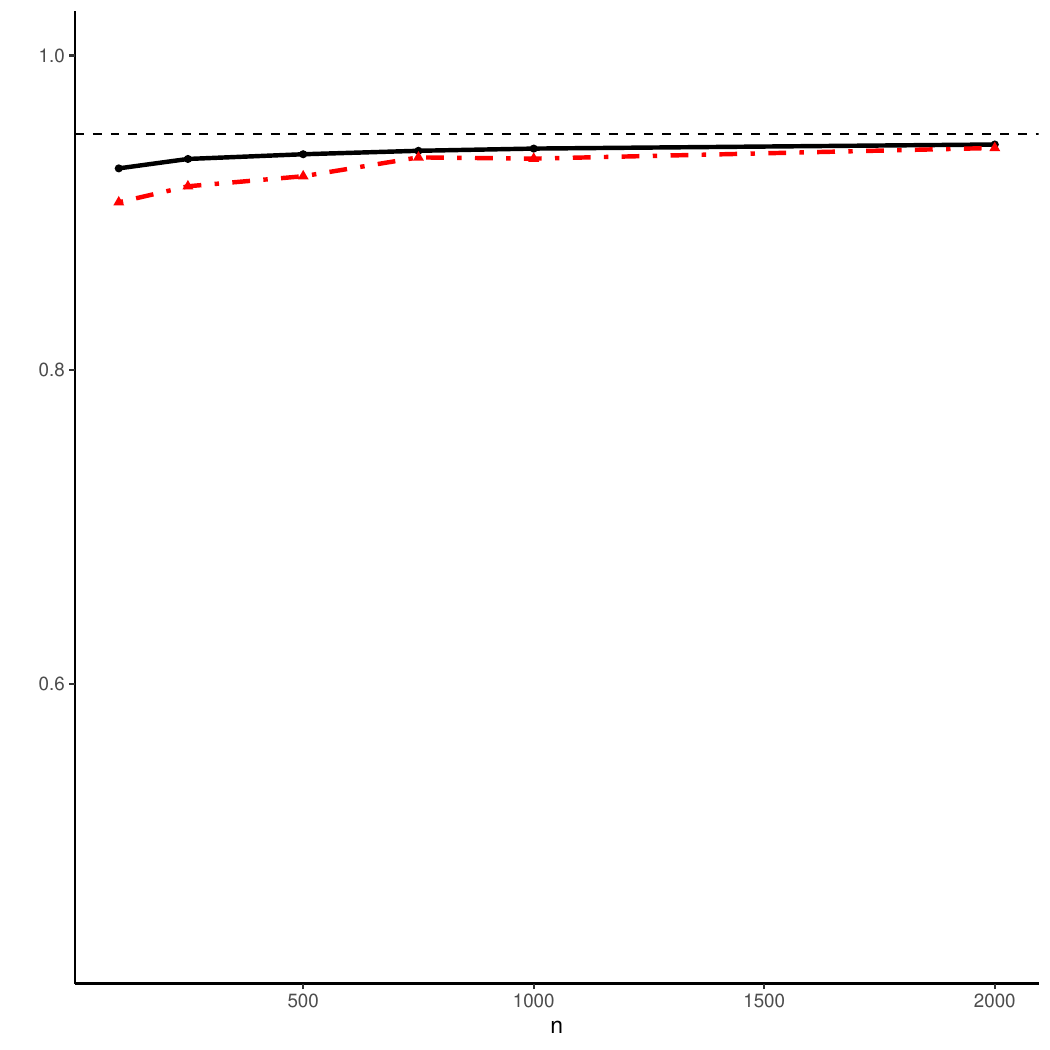}
		\subcaption{$\x=-1$}
	\end{subfigure}%
	\begin{subfigure}[b]{0.5\textwidth}
		\includegraphics[height=0.3\textheight,width=0.95\textwidth]{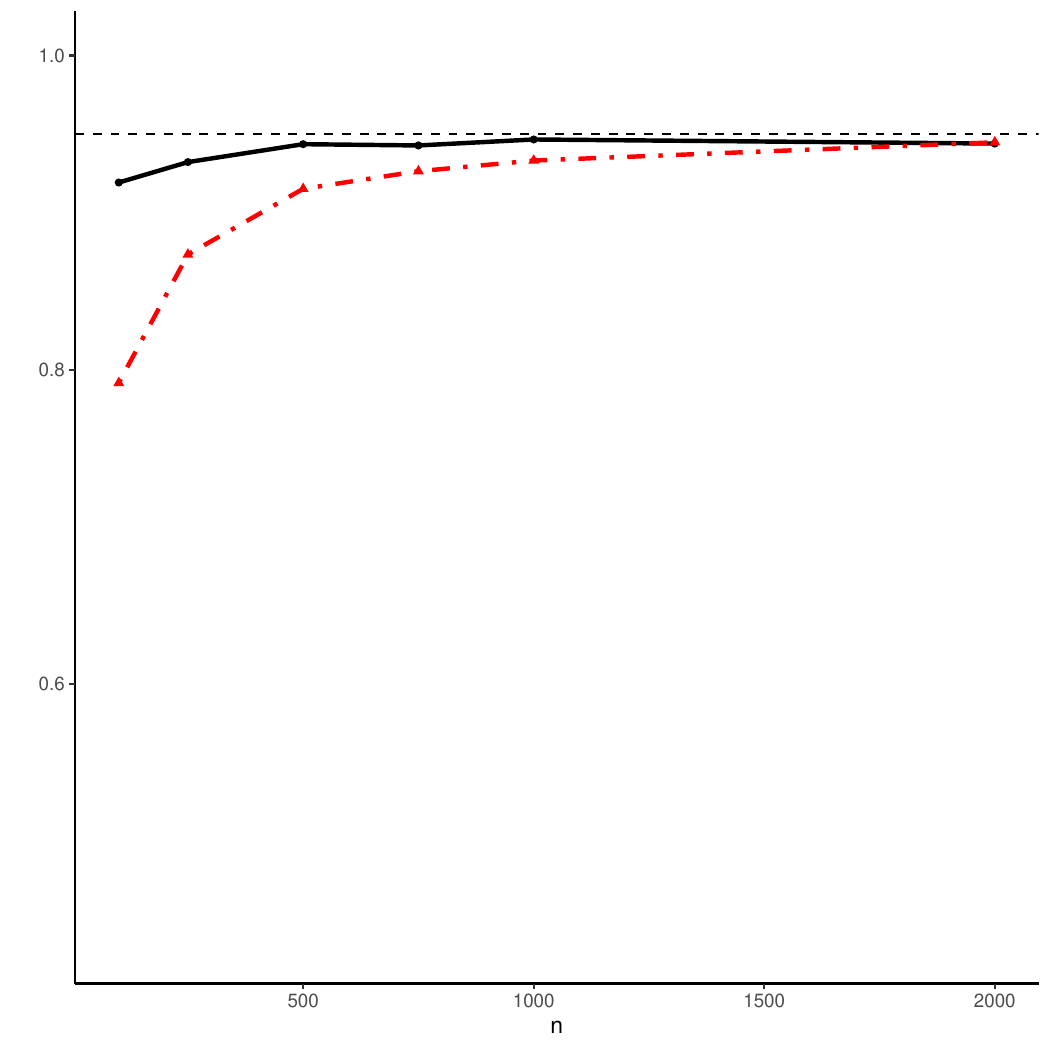}
		\subcaption{$\x=-0.6$}
	\end{subfigure}	
	\begin{subfigure}[b]{0.5\textwidth}
		\includegraphics[height=0.3\textheight,width=0.95\textwidth]{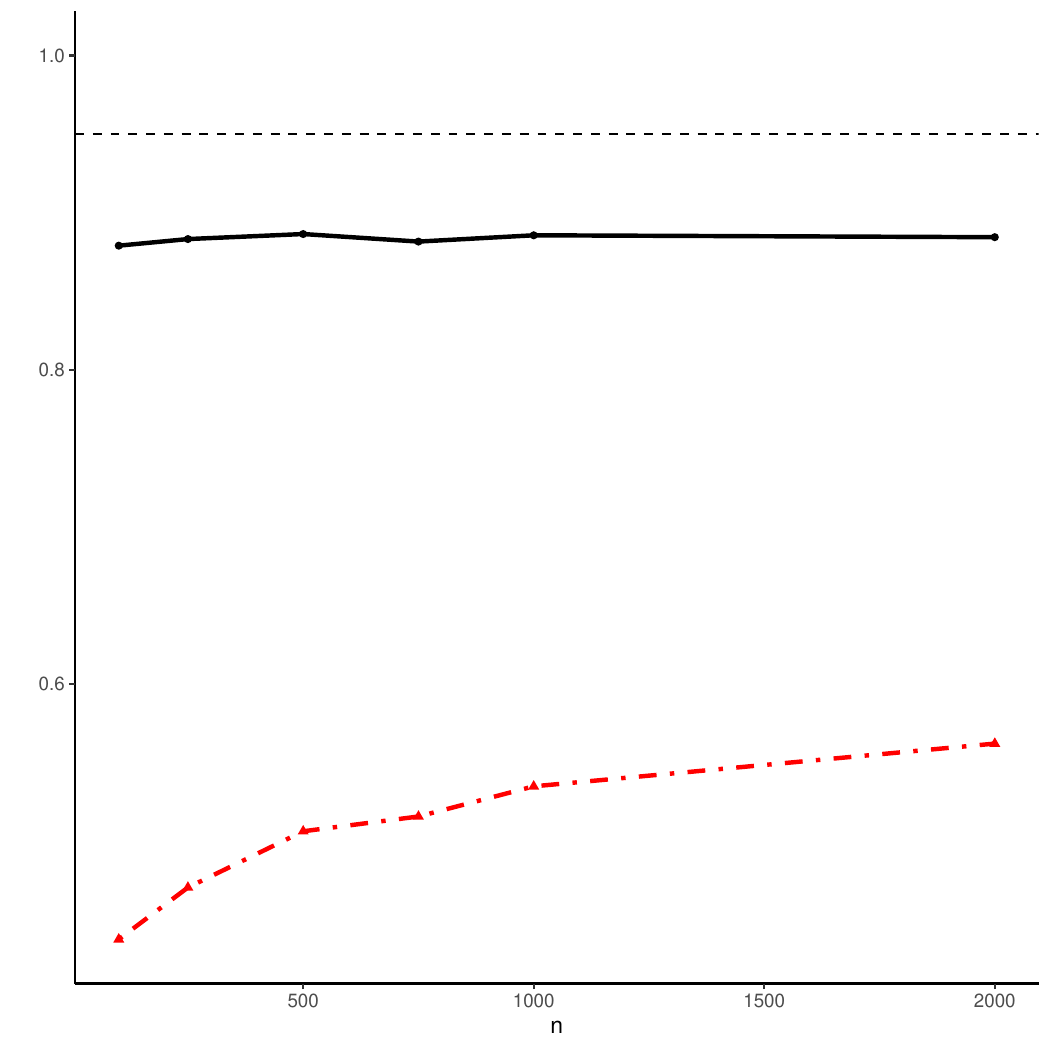}
		\subcaption{$\x=-0.2$}
	\end{subfigure}%	
	\begin{subfigure}[b]{0.5\textwidth}
		\includegraphics[height=0.3\textheight,width=0.95\textwidth]{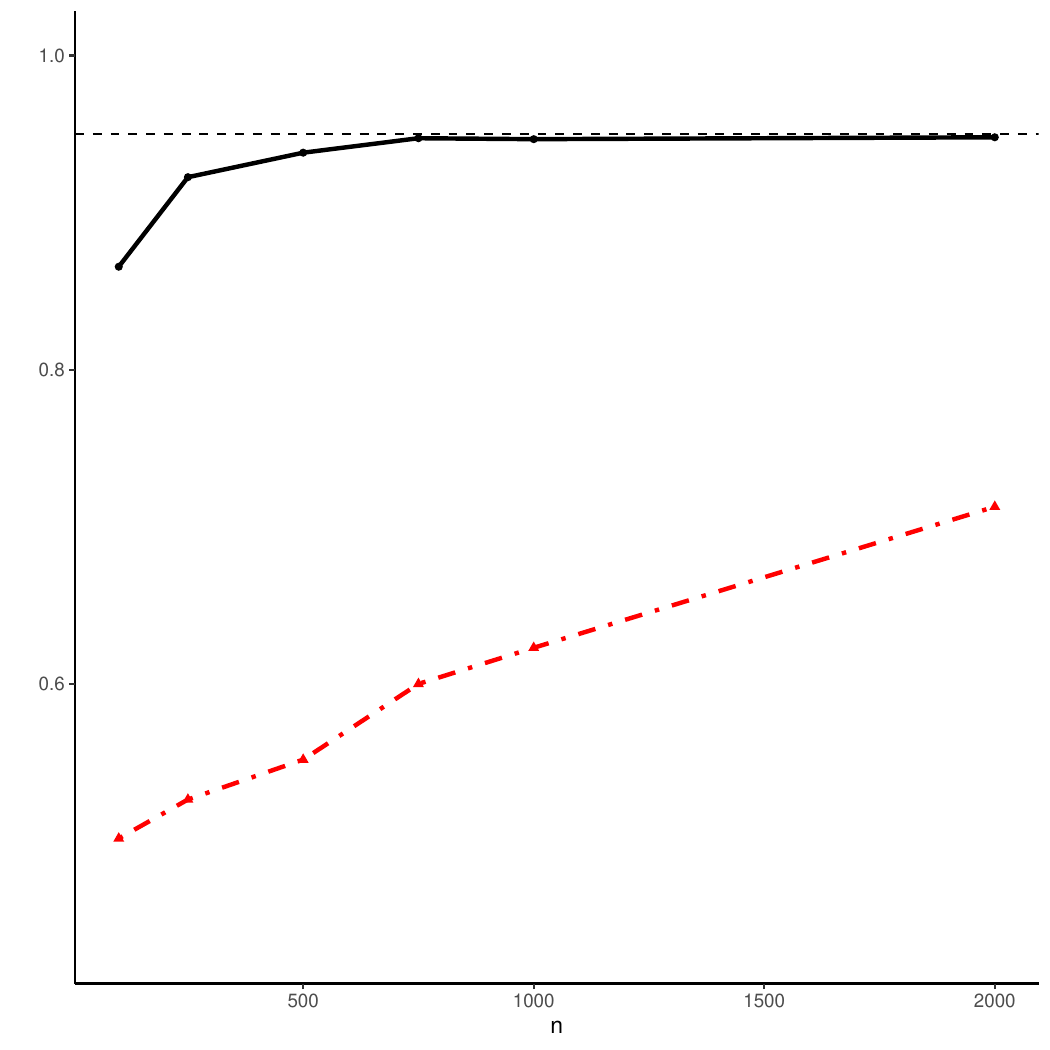}
		\subcaption{$\x=0.2$}
	\end{subfigure}	
	\begin{subfigure}[b]{0.5\textwidth}
		\includegraphics[height=0.3\textheight,width=0.95\textwidth]{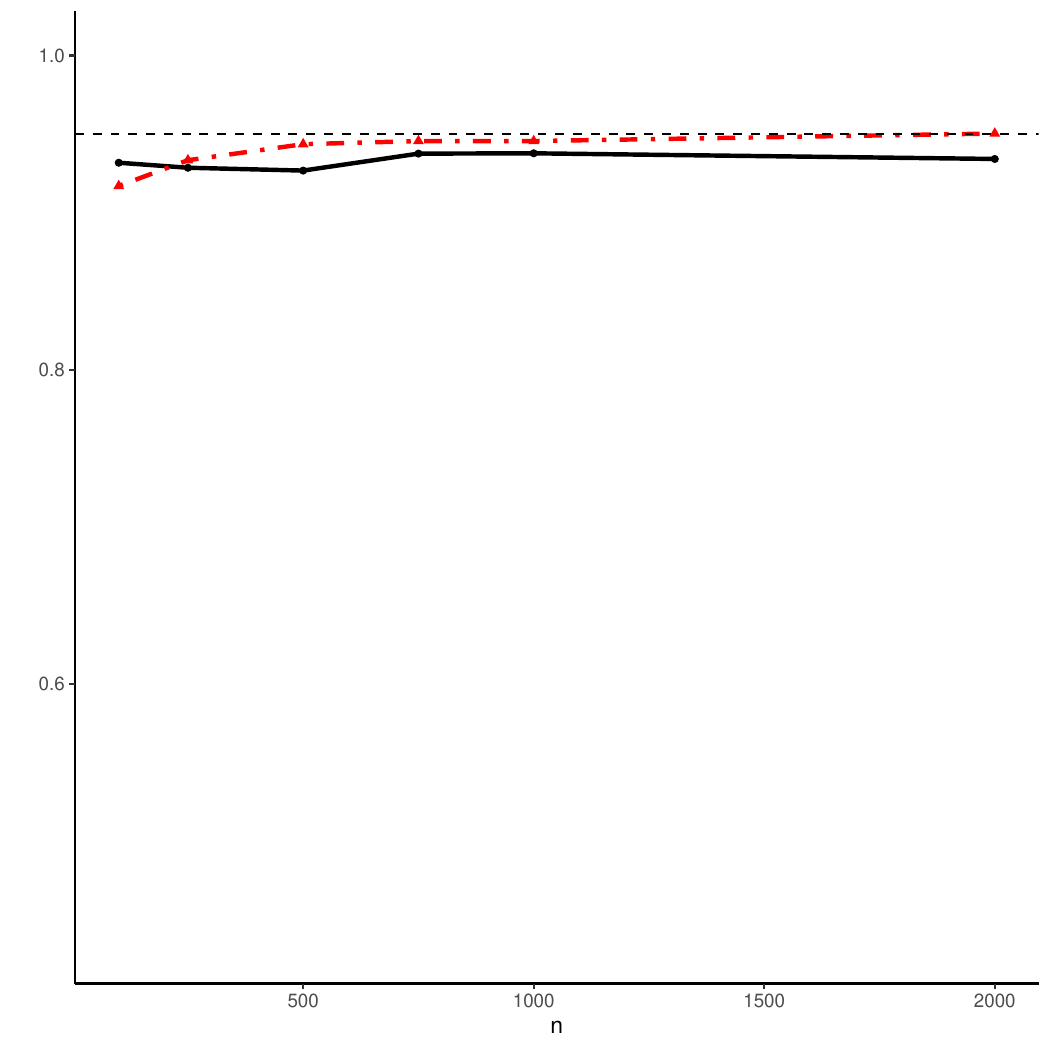}
		\subcaption{$\x=0.6$}
	\end{subfigure}%
	\begin{subfigure}[b]{0.5\textwidth}
		\includegraphics[height=0.3\textheight,width=0.95\textwidth]{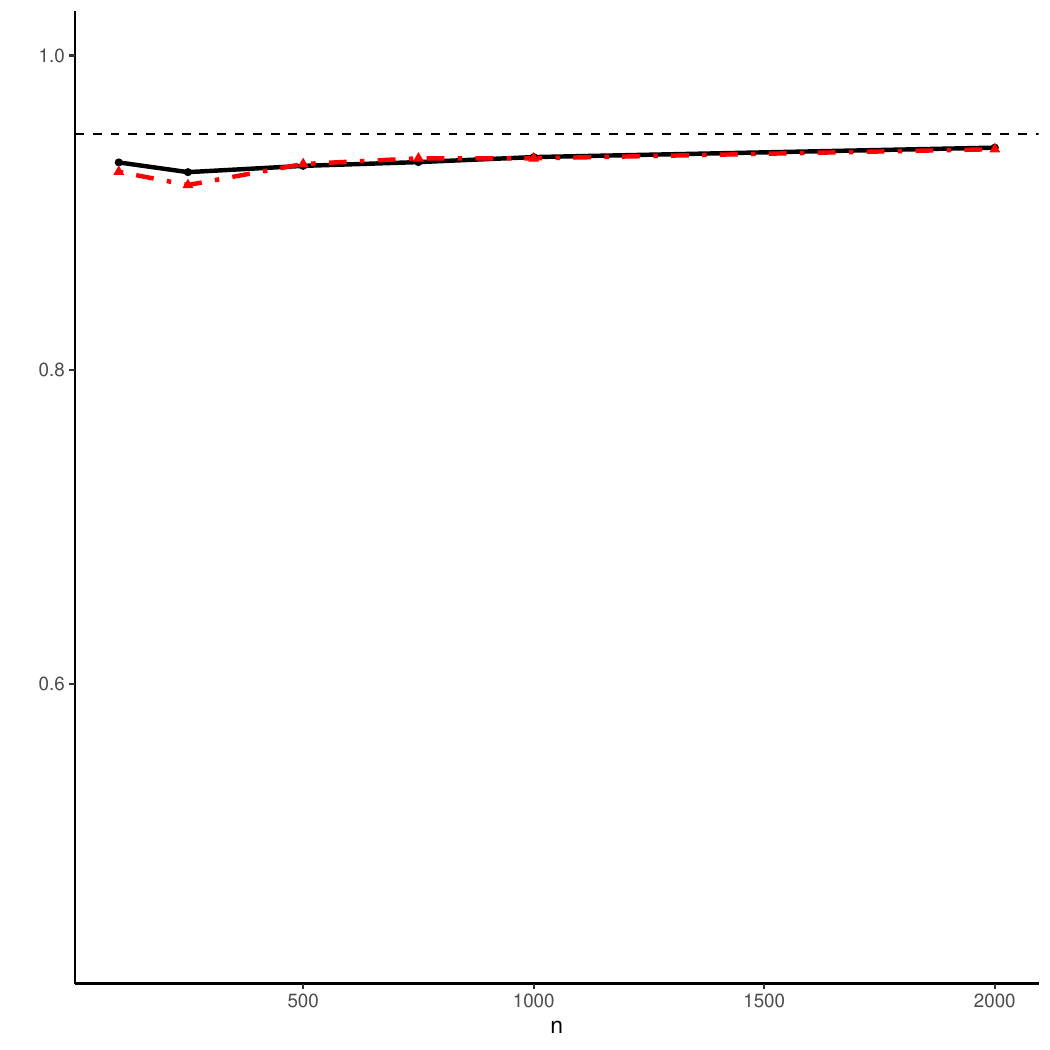}
		\subcaption{$\x=1$}
	\end{subfigure}%
	\begin{flushleft}\footnotesize Notes: \blackline Robust Bias Correction, \redline Undersmoothing
	\end{flushleft}
\end{figure}

%%% h_MSE
\clearpage
\begin{figure}[!htb]	
	\centering
	\caption{Empirical Coverage for 95\% Confidence Intervals\\
	Uniform Kernel, $\hat{h}_{\MSE}$, $\v=0$}
	\label{suppfig:ec_nu0_uni_hmse}	
	\begin{subfigure}[b]{0.5\textwidth}
		\includegraphics[height=0.3\textheight,width=0.95\textwidth]{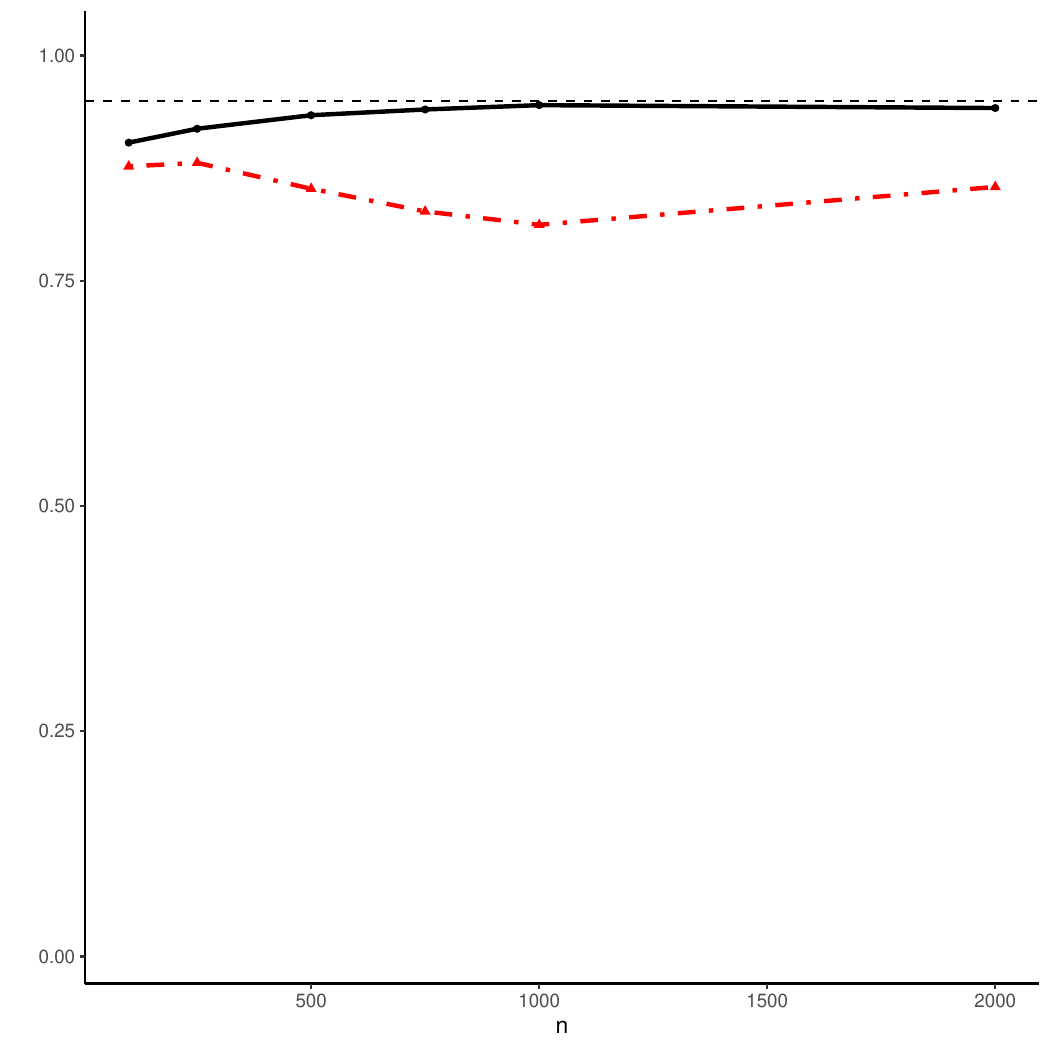}
		\subcaption{$\x=-1$}
	\end{subfigure}%
	\begin{subfigure}[b]{0.5\textwidth}
		\includegraphics[height=0.3\textheight,width=0.95\textwidth]{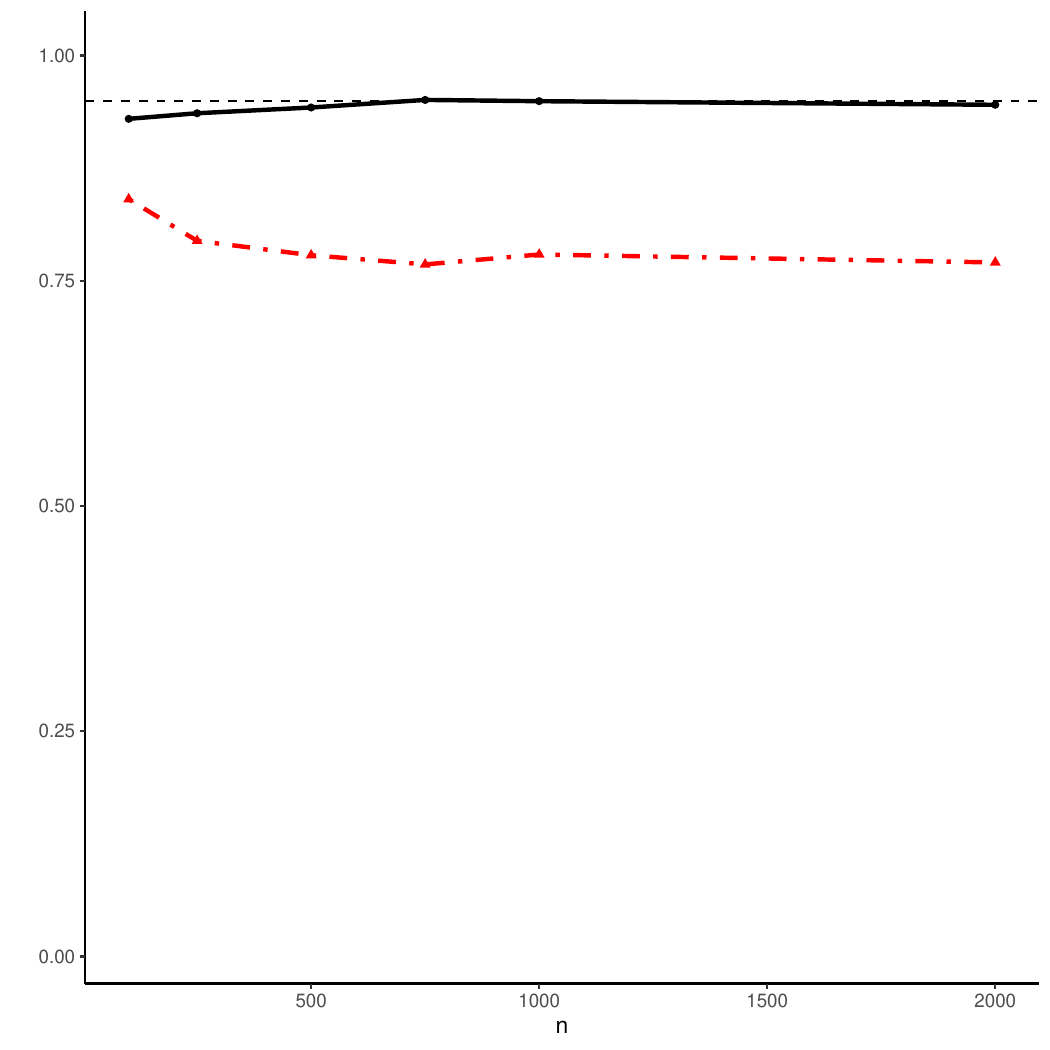}
		\subcaption{$\x=-0.6$}
	\end{subfigure}	
	\begin{subfigure}[b]{0.5\textwidth}
		\includegraphics[height=0.3\textheight,width=0.95\textwidth]{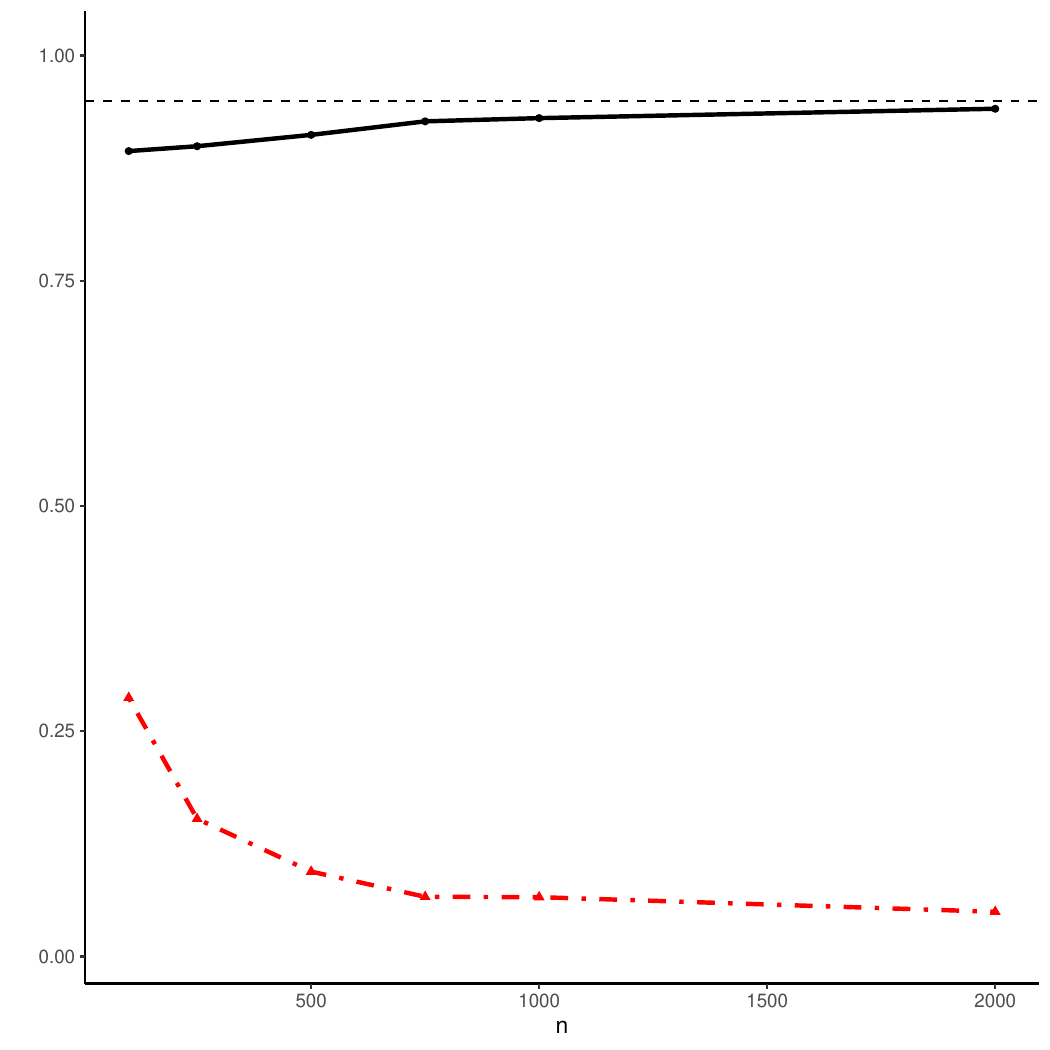}
		\subcaption{$\x=-0.2$}
	\end{subfigure}%	
	\begin{subfigure}[b]{0.5\textwidth}
		\includegraphics[height=0.3\textheight,width=0.95\textwidth]{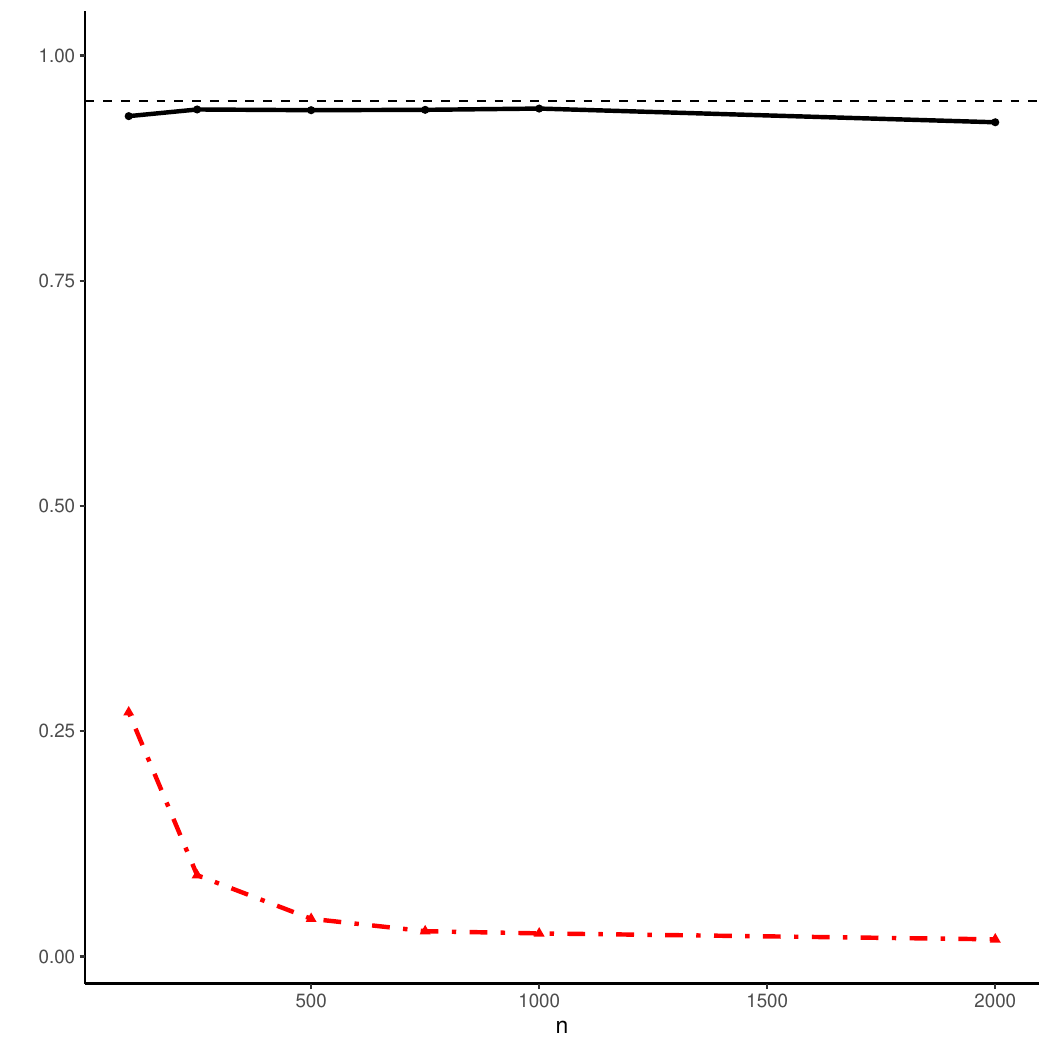}
		\subcaption{$\x=0.2$}
	\end{subfigure}	
	\begin{subfigure}[b]{0.5\textwidth}
		\includegraphics[height=0.3\textheight,width=0.95\textwidth]{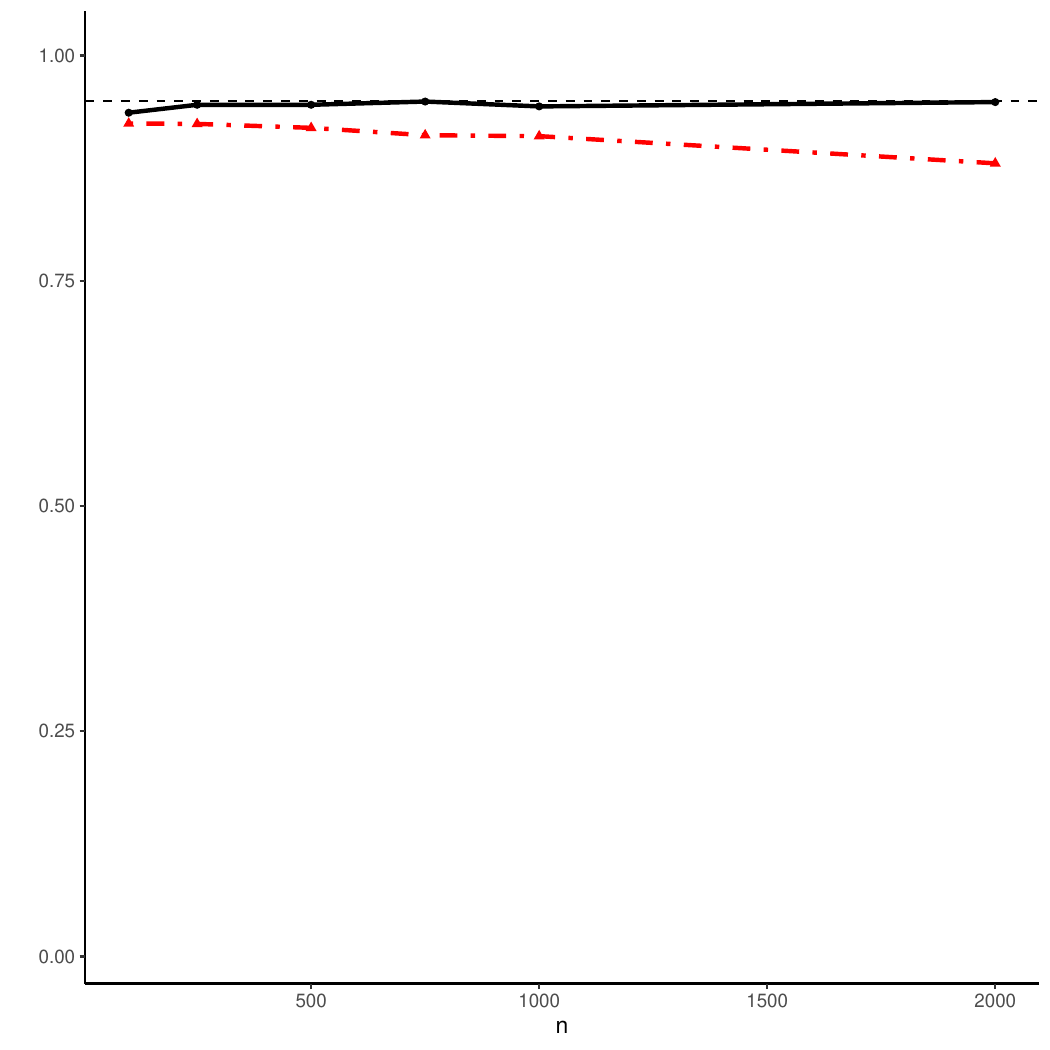}
		\subcaption{$\x=0.6$}
	\end{subfigure}%
	\begin{subfigure}[b]{0.5\textwidth}
		\includegraphics[height=0.3\textheight,width=0.95\textwidth]{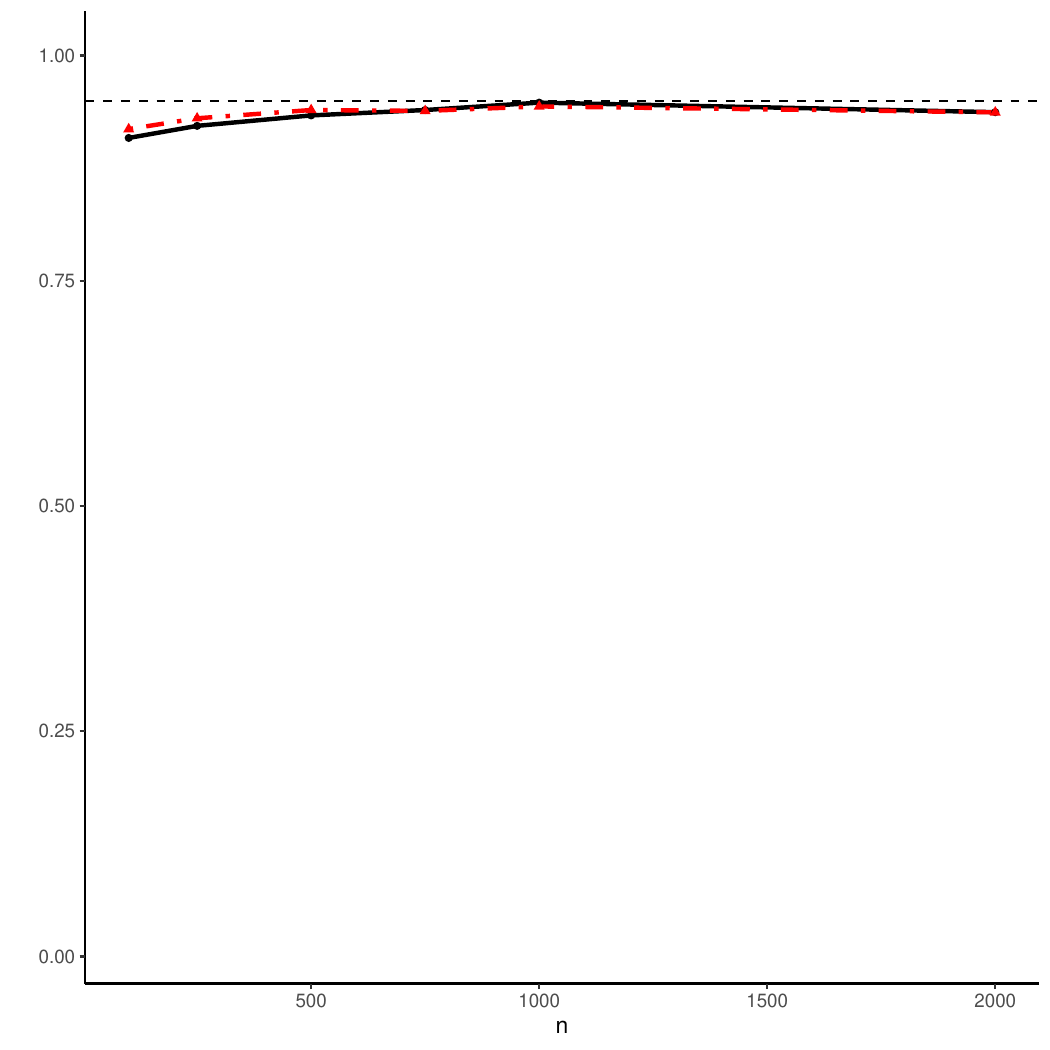}
		\subcaption{$\x=1$}
	\end{subfigure}%
	\begin{flushleft}\footnotesize Notes: \blackline Robust Bias Correction, \redline Undersmoothing
	\end{flushleft}
\end{figure}

\clearpage
\begin{figure}[!htb]
	\centering
	\caption{Empirical Coverage for 95\% Confidence Intervals\\
	Uniform Kernel, $\hat{h}_{\MSE}$, $\v=1$}
	\label{suppfig:ec_nu1_uni_hmse}	
	\begin{subfigure}[b]{0.5\textwidth}
		\includegraphics[height=0.3\textheight,width=0.95\textwidth]{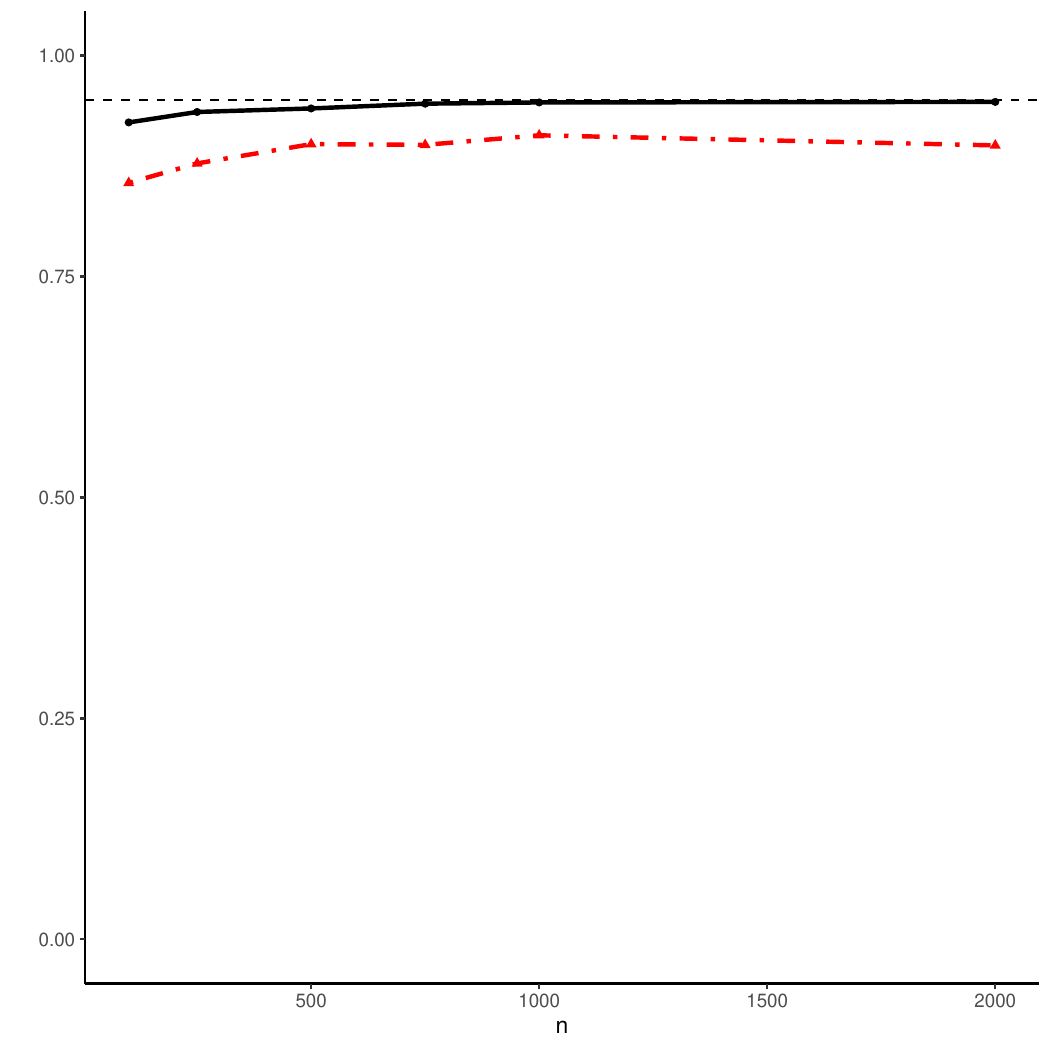}
		\subcaption{$\x=-1$}
	\end{subfigure}%
	\begin{subfigure}[b]{0.5\textwidth}
		\includegraphics[height=0.3\textheight,width=0.95\textwidth]{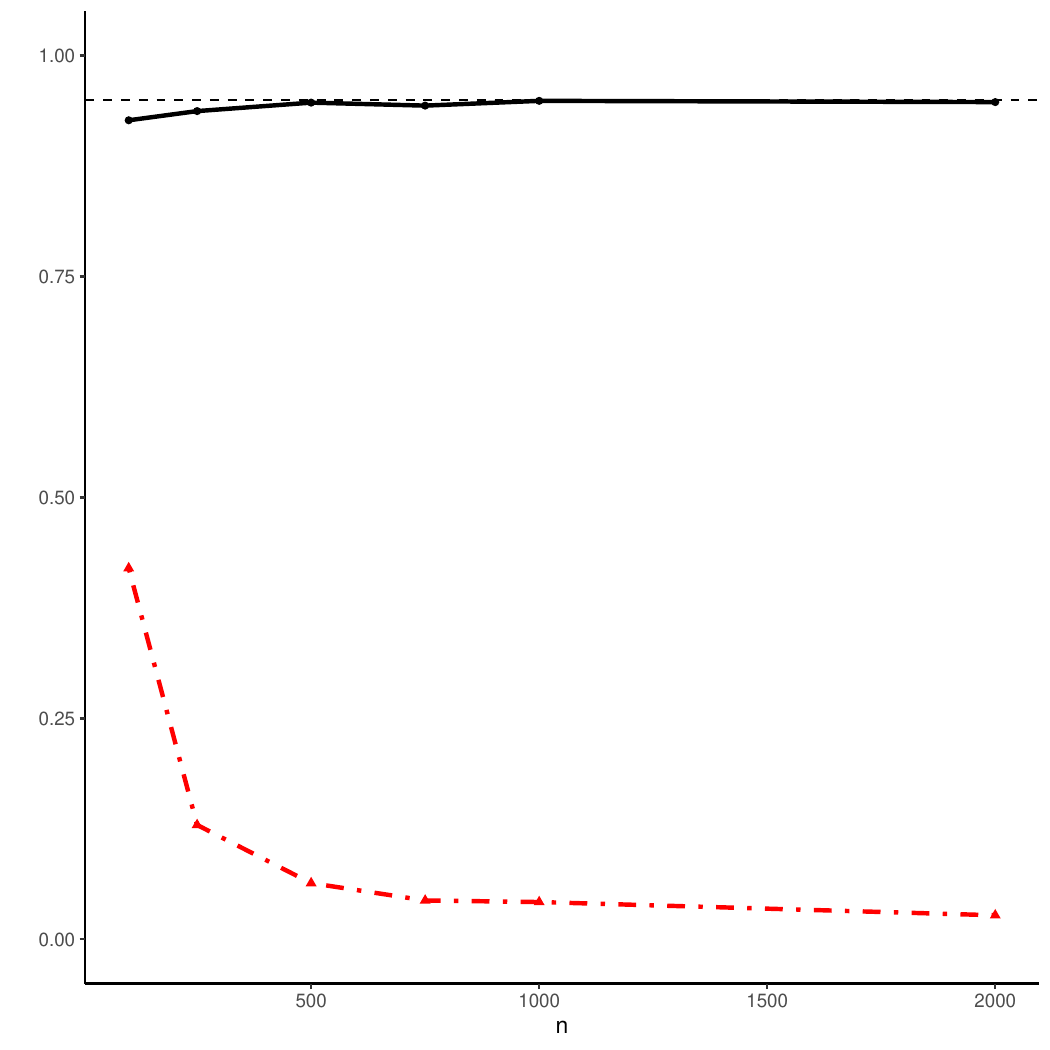}
		\subcaption{$\x=-0.6$}
	\end{subfigure}	
	\begin{subfigure}[b]{0.5\textwidth}
		\includegraphics[height=0.3\textheight,width=0.95\textwidth]{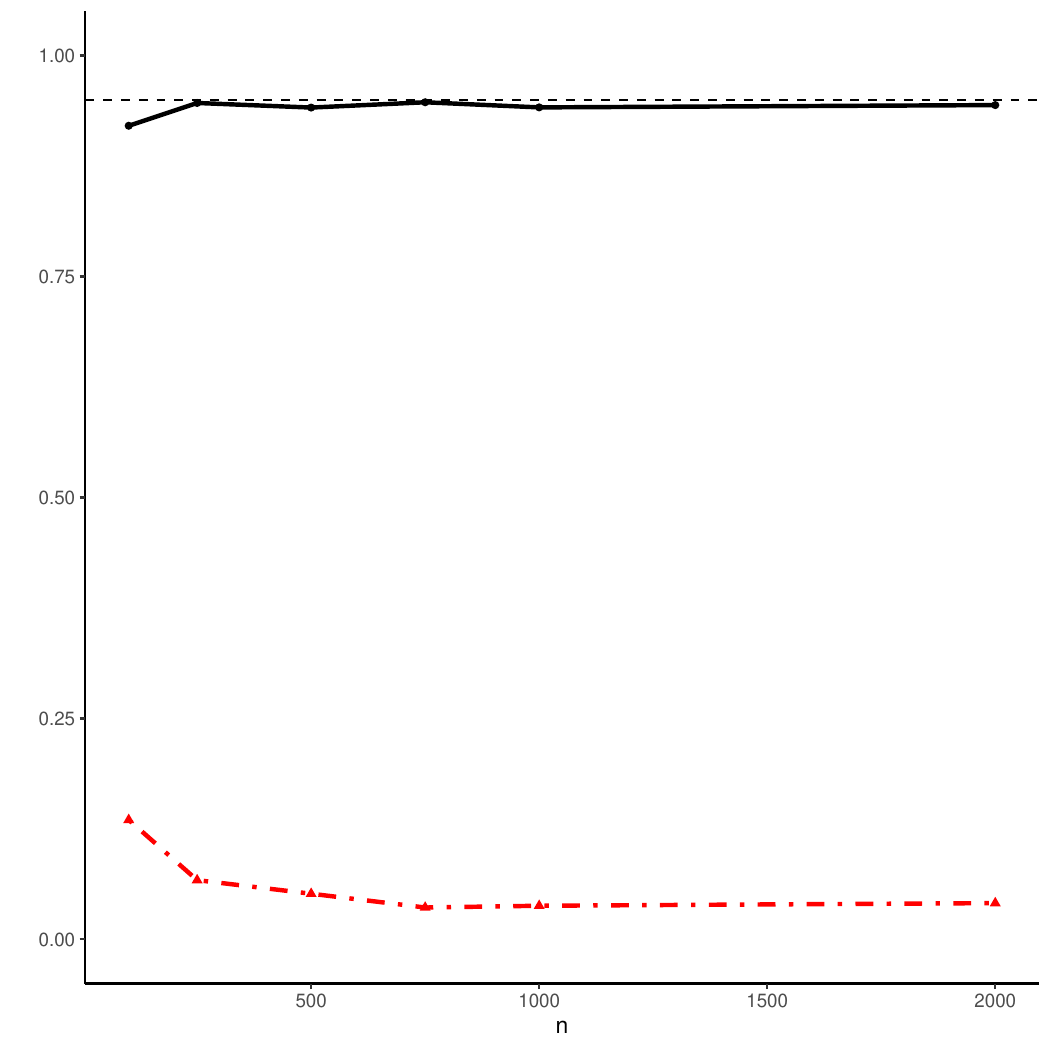}
		\subcaption{$\x=-0.2$}
	\end{subfigure}%	
	\begin{subfigure}[b]{0.5\textwidth}
		\includegraphics[height=0.3\textheight,width=0.95\textwidth]{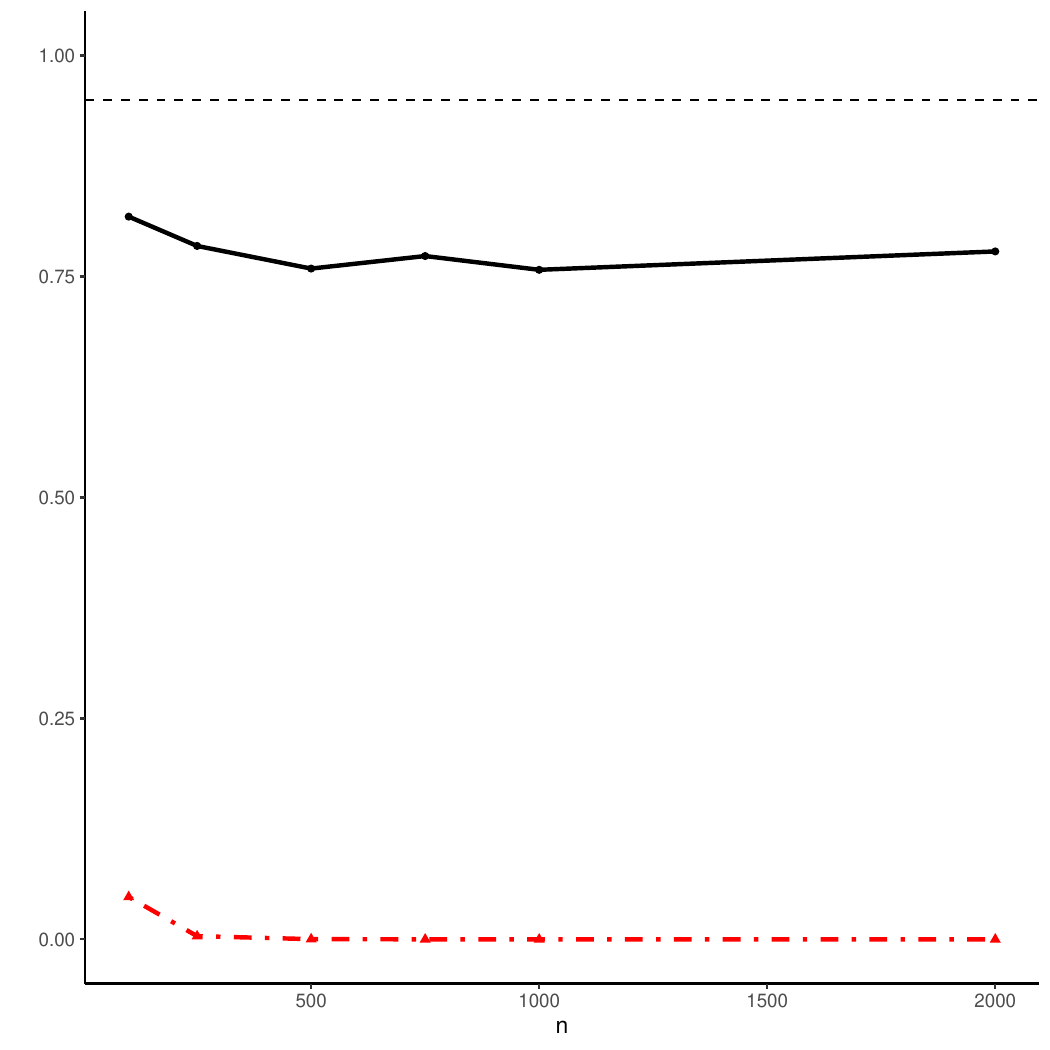}
		\subcaption{$\x=0.2$}
	\end{subfigure}	
	\begin{subfigure}[b]{0.5\textwidth}
		\includegraphics[height=0.3\textheight,width=0.95\textwidth]{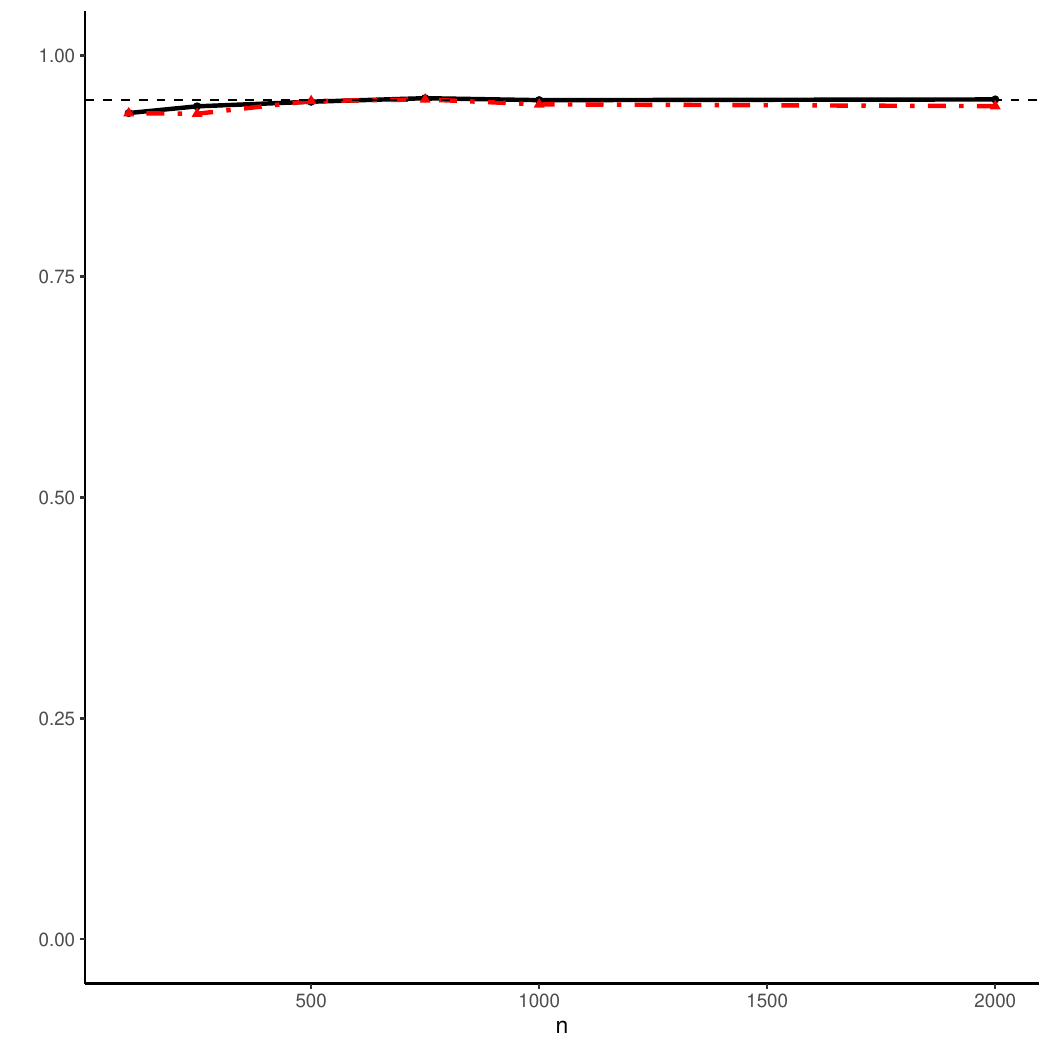}
		\subcaption{$\x=0.6$}
	\end{subfigure}%
	\begin{subfigure}[b]{0.5\textwidth}
		\includegraphics[height=0.3\textheight,width=0.95\textwidth]{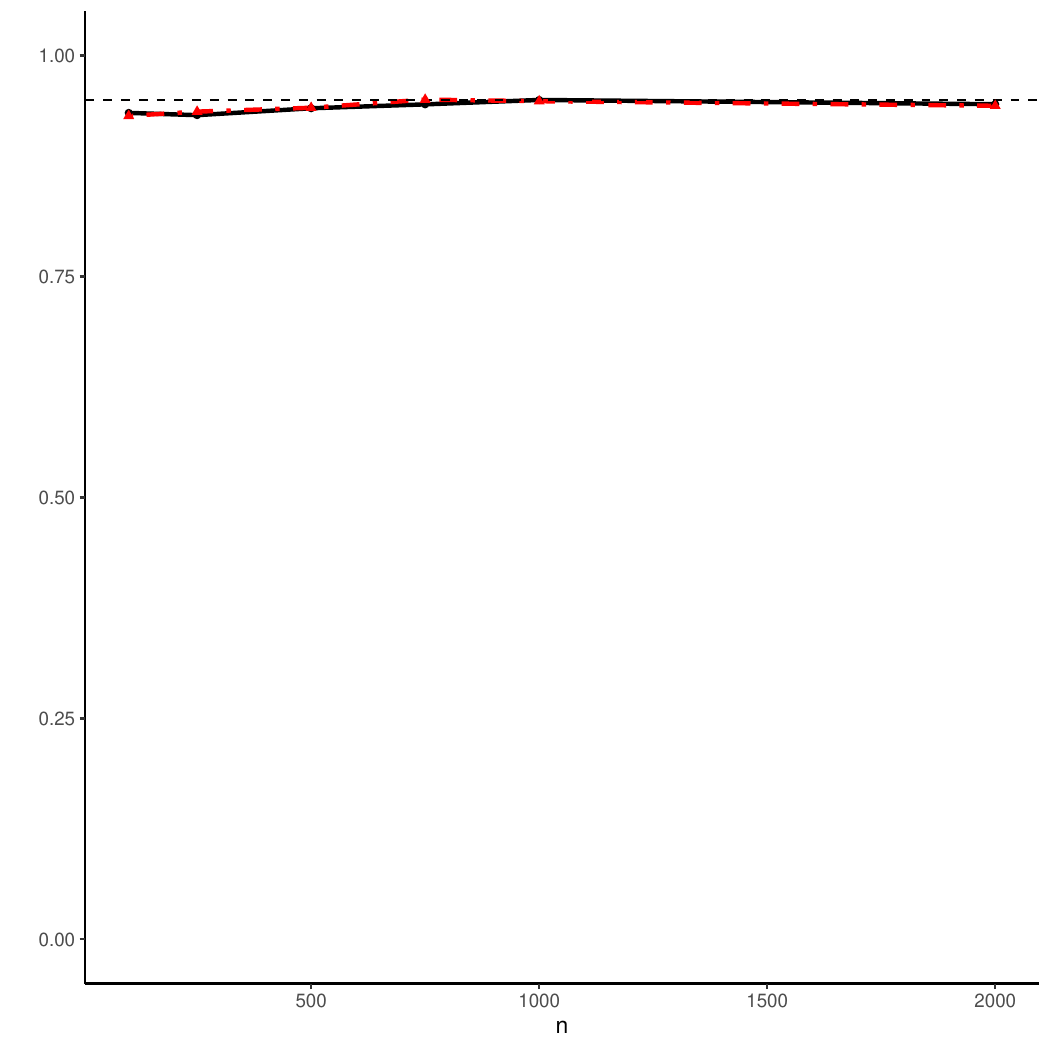}
		\subcaption{$\x=1$}
	\end{subfigure}%
	\begin{flushleft}\footnotesize Notes: \blackline Robust Bias Correction, \redline Undersmoothing
	\end{flushleft}
\end{figure}

%%%%%%%%%%%%%%%%%%%%%%%%% IL EPA
\clearpage

\begin{figure}[!htb]
	\centering
	\caption{Average Interval Length for 95\% Confidence Intervals\\
	Epanechnikov Kernel, $\v=0$}
	\label{suppfig:il_nu0_epa}
	\begin{subfigure}[b]{0.5\textwidth}
		\includegraphics[height=0.3\textheight,width=0.95\textwidth]{simuls/output/il_kepa_p1_d0_x1.pdf}
		\subcaption{$\x=-1$}
	\end{subfigure}%
	\begin{subfigure}[b]{0.5\textwidth}
		\includegraphics[height=0.3\textheight,width=0.95\textwidth]{simuls/output/il_kepa_p1_d0_x2.pdf}
		\subcaption{$\x=-0.6$}
	\end{subfigure}	
	\begin{subfigure}[b]{0.5\textwidth}
		\includegraphics[height=0.3\textheight,width=0.95\textwidth]{simuls/output/il_kepa_p1_d0_x3.pdf}
		\subcaption{$\x=-0.2$}
	\end{subfigure}%	
	\begin{subfigure}[b]{0.5\textwidth}
		\includegraphics[height=0.3\textheight,width=0.95\textwidth]{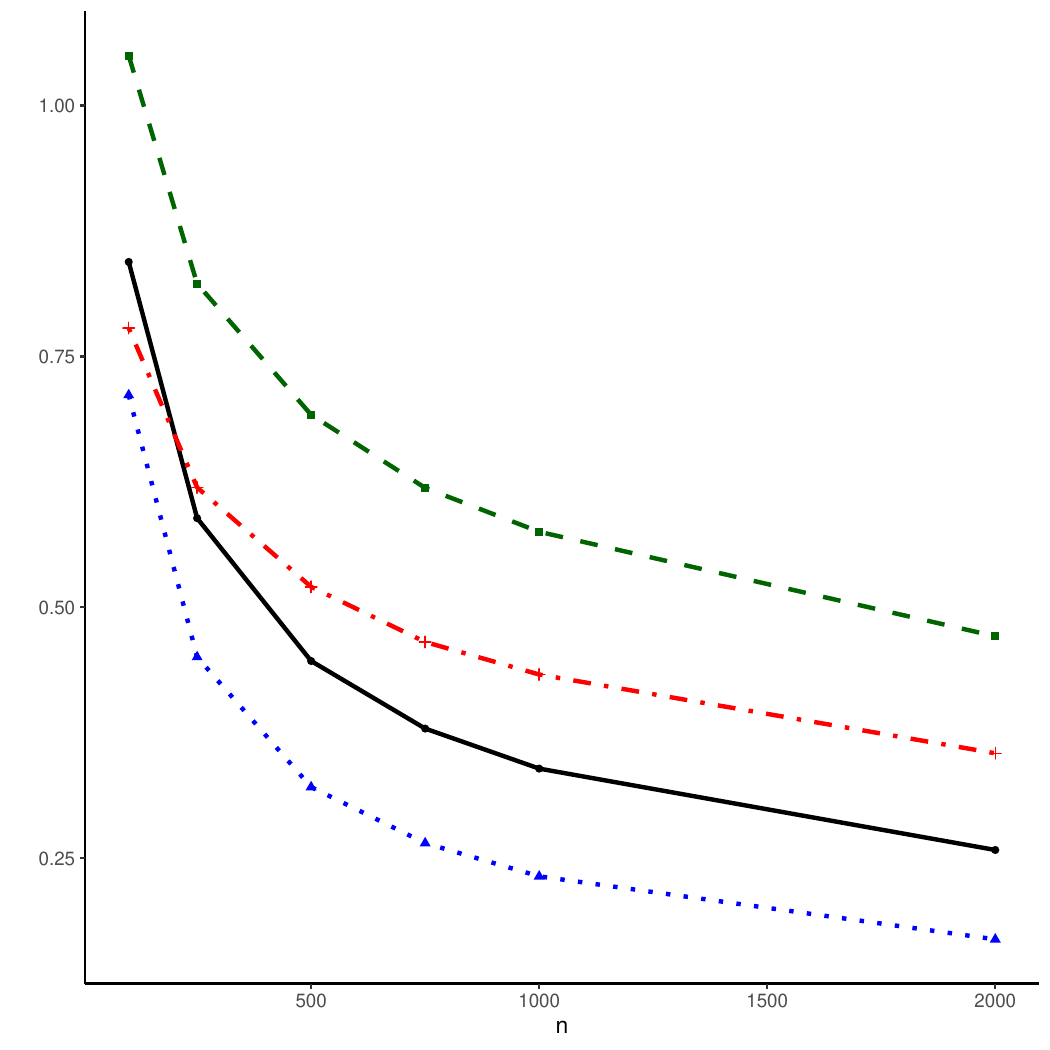}
		\subcaption{$\x=0.2$}
	\end{subfigure}	
	\begin{subfigure}[b]{0.5\textwidth}
		\includegraphics[height=0.3\textheight,width=0.95\textwidth]{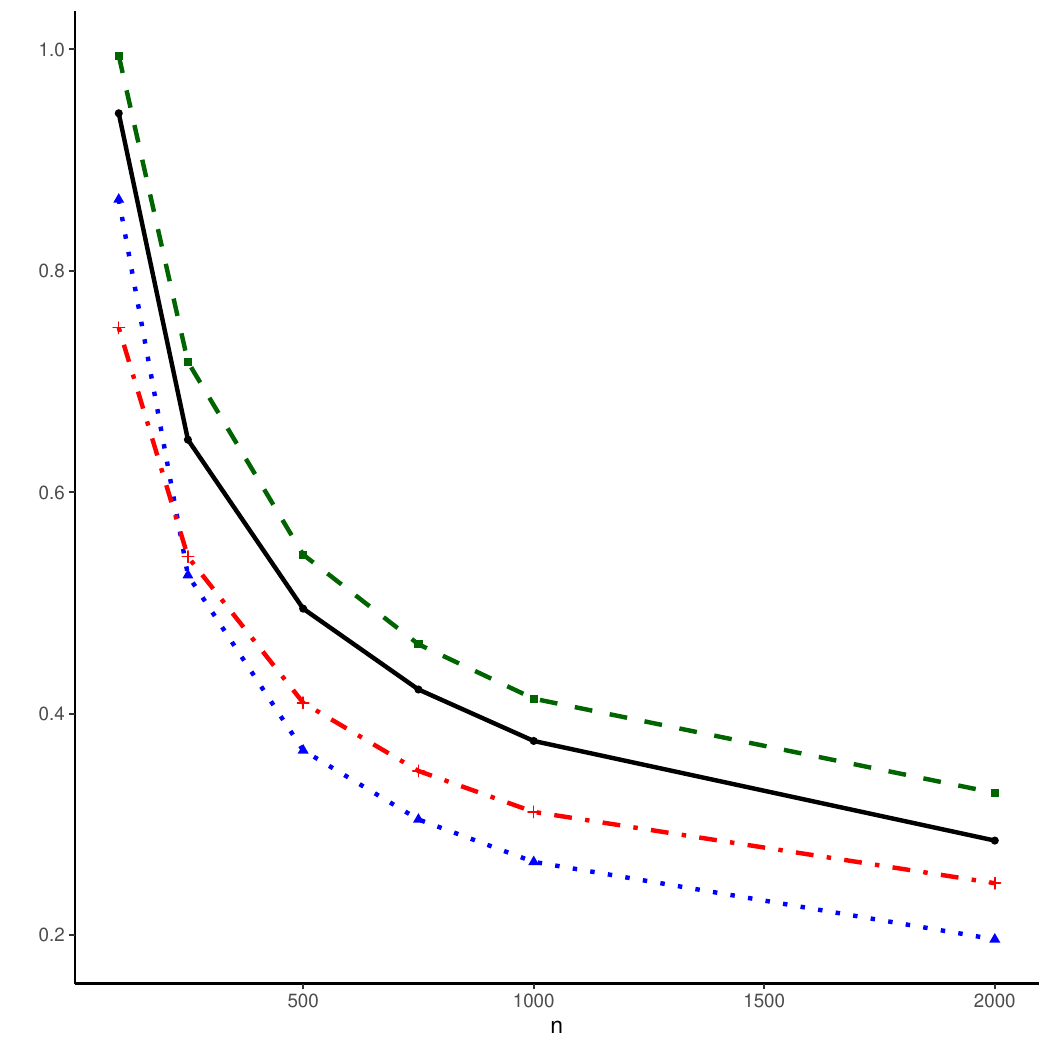}
		\subcaption{$\x=0.6$}
	\end{subfigure}%
	\begin{subfigure}[b]{0.5\textwidth}
		\includegraphics[height=0.3\textheight,width=0.95\textwidth]{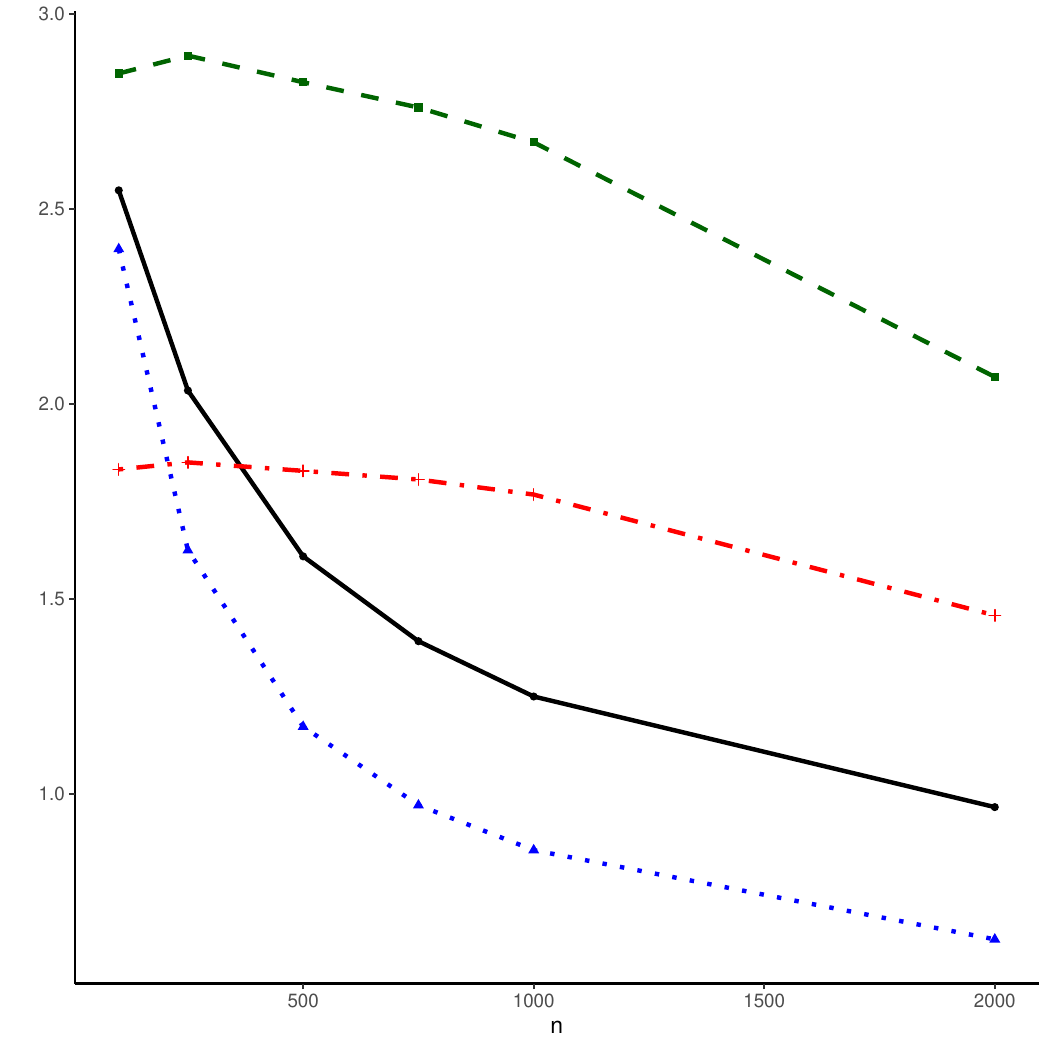}
		\subcaption{$\x=1$}
	\end{subfigure}%
	\begin{flushleft}
\footnotesize Notes: \blackline $\irbc(\hat{h}_{\RBC})$, \blueline $\irbc(\hat{h}_{\MSE})$, \greenline   $\irbc(\hat{h}_{\US})$,  \redline$\ius(\hat{h}_{\US})$
\end{flushleft}
\end{figure}

\clearpage
\begin{figure}[!htb]
	\centering
	\caption{Average Interval Length for 95\% Confidence Intervals\\
			Epanechnikov Kernel, $\v=1$}
\label{suppfig:il_nu1_epa}	
	\begin{subfigure}[b]{0.5\textwidth}
		\includegraphics[height=0.3\textheight,width=0.95\textwidth]{simuls/output/il_kepa_p2_d1_x1.pdf}
		\subcaption{$\x=-1$}
	\end{subfigure}%
	\begin{subfigure}[b]{0.5\textwidth}
		\includegraphics[height=0.3\textheight,width=0.95\textwidth]{simuls/output/il_kepa_p2_d1_x2.pdf}
		\subcaption{$\x=-0.6$}
	\end{subfigure}	
	\begin{subfigure}[b]{0.5\textwidth}
		\includegraphics[height=0.3\textheight,width=0.95\textwidth]{simuls/output/il_kepa_p2_d1_x3.pdf}
		\subcaption{$\x=-0.2$}
	\end{subfigure}%	
	\begin{subfigure}[b]{0.5\textwidth}
		\includegraphics[height=0.3\textheight,width=0.95\textwidth]{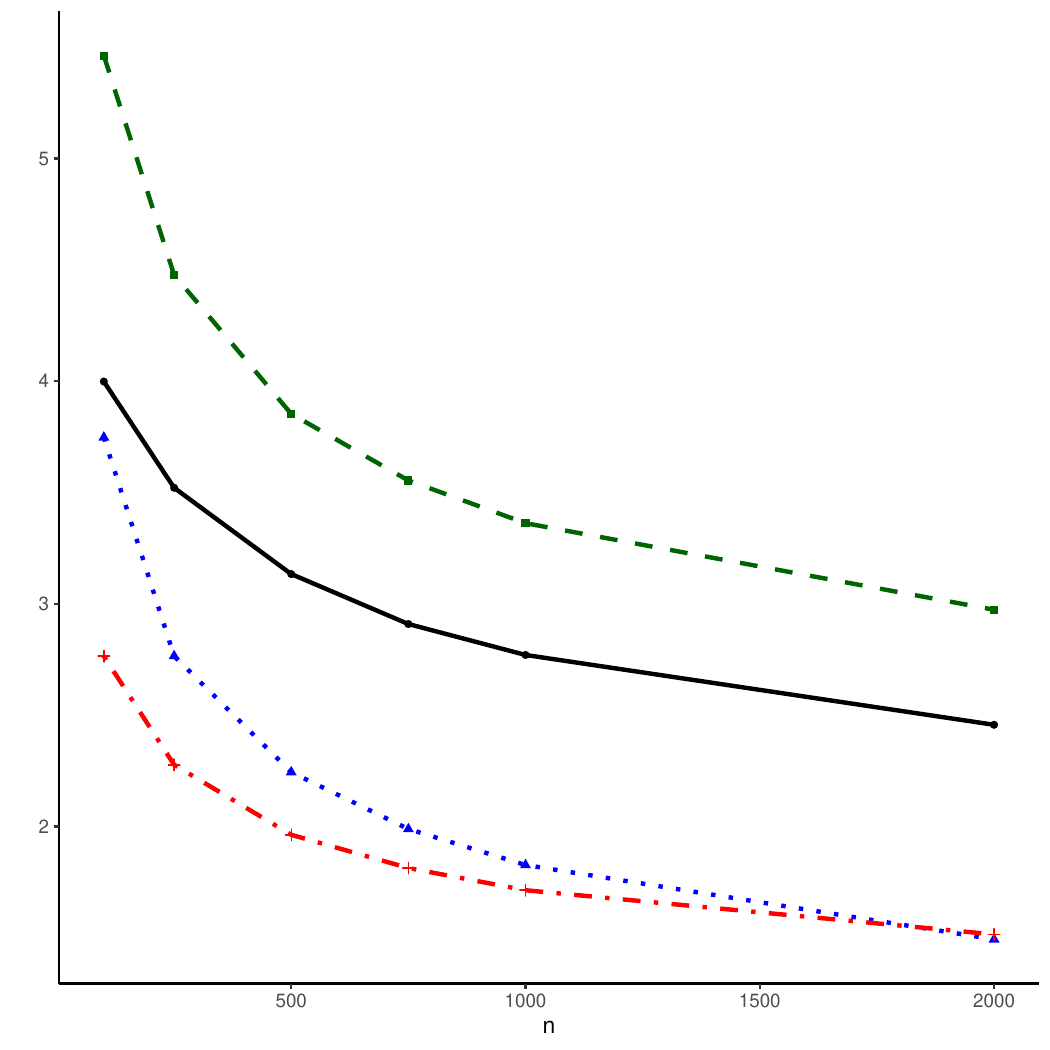}
		\subcaption{$\x=0.2$}
	\end{subfigure}	
	\begin{subfigure}[b]{0.5\textwidth}
		\includegraphics[height=0.3\textheight,width=0.95\textwidth]{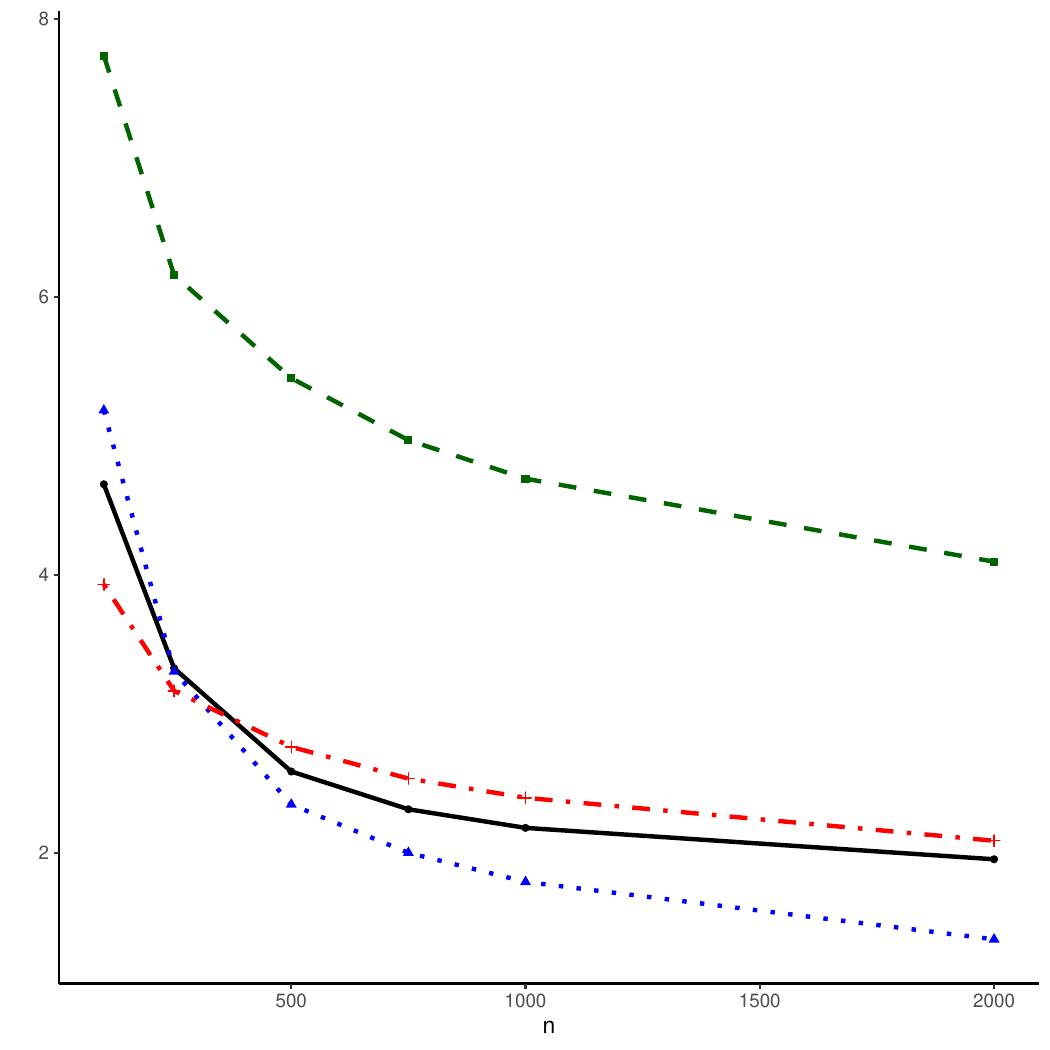}
		\subcaption{$\x=0.6$}
	\end{subfigure}%
	\begin{subfigure}[b]{0.5\textwidth}
		\includegraphics[height=0.3\textheight,width=0.95\textwidth]{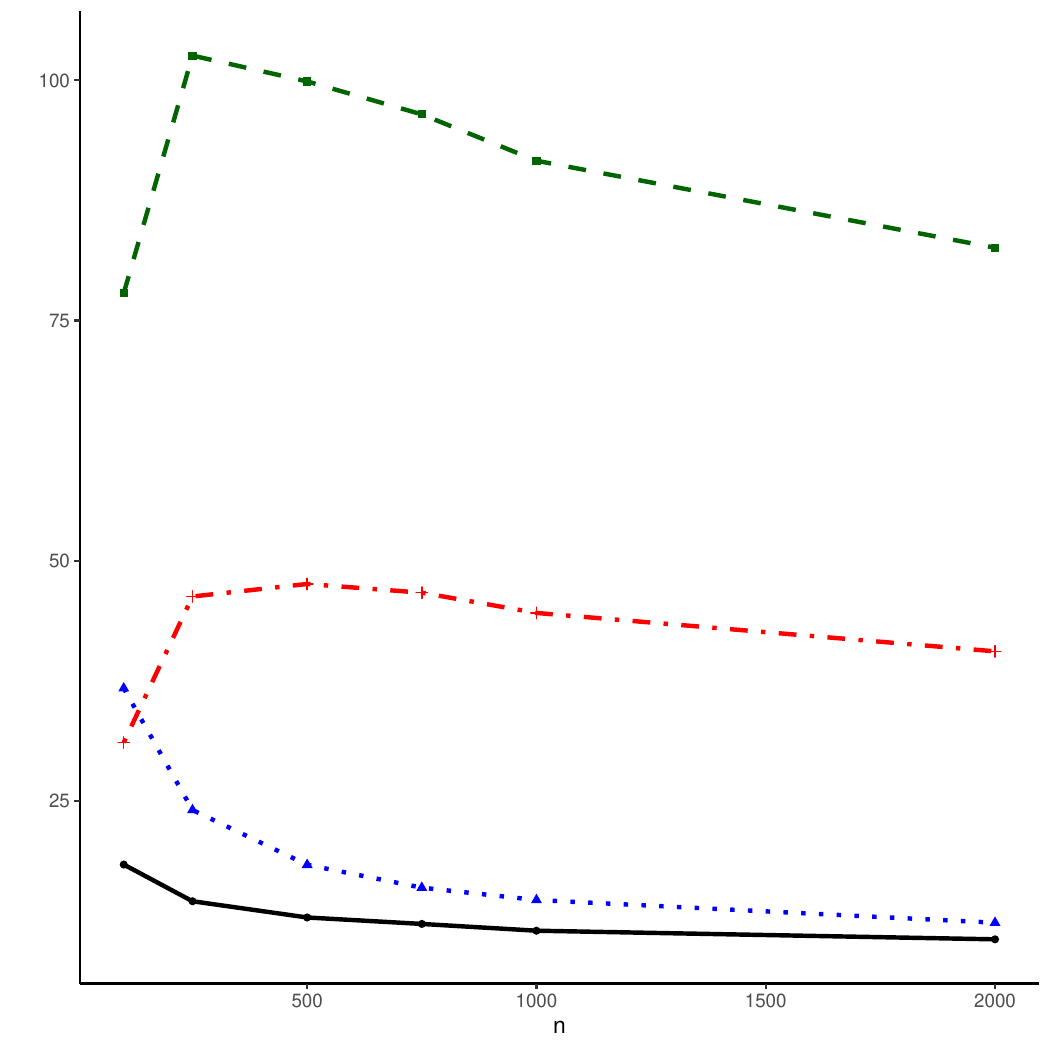}
		\subcaption{$\x=1$}
	\end{subfigure}%
	\begin{flushleft}
\footnotesize Notes: \blackline $\irbc(\hat{h}_{\RBC})$, \blueline $\irbc(\hat{h}_{\MSE})$, \greenline   $\irbc(\hat{h}_{\US})$,  \redline$\ius(\hat{h}_{\US})$
\end{flushleft}
\end{figure}

%%%%%%%%%%%%%%%%%%%%%%%%% IL UNI
\clearpage

\begin{figure}[!htb]
	\centering
	\caption{Average Interval Length for 95\% Confidence Intervals\\
	Uniform Kernel, $\v=0$}
	\label{suppfig:il_nu0_uni}
	\begin{subfigure}[b]{0.5\textwidth}
		\includegraphics[height=0.3\textheight,width=0.95\textwidth]{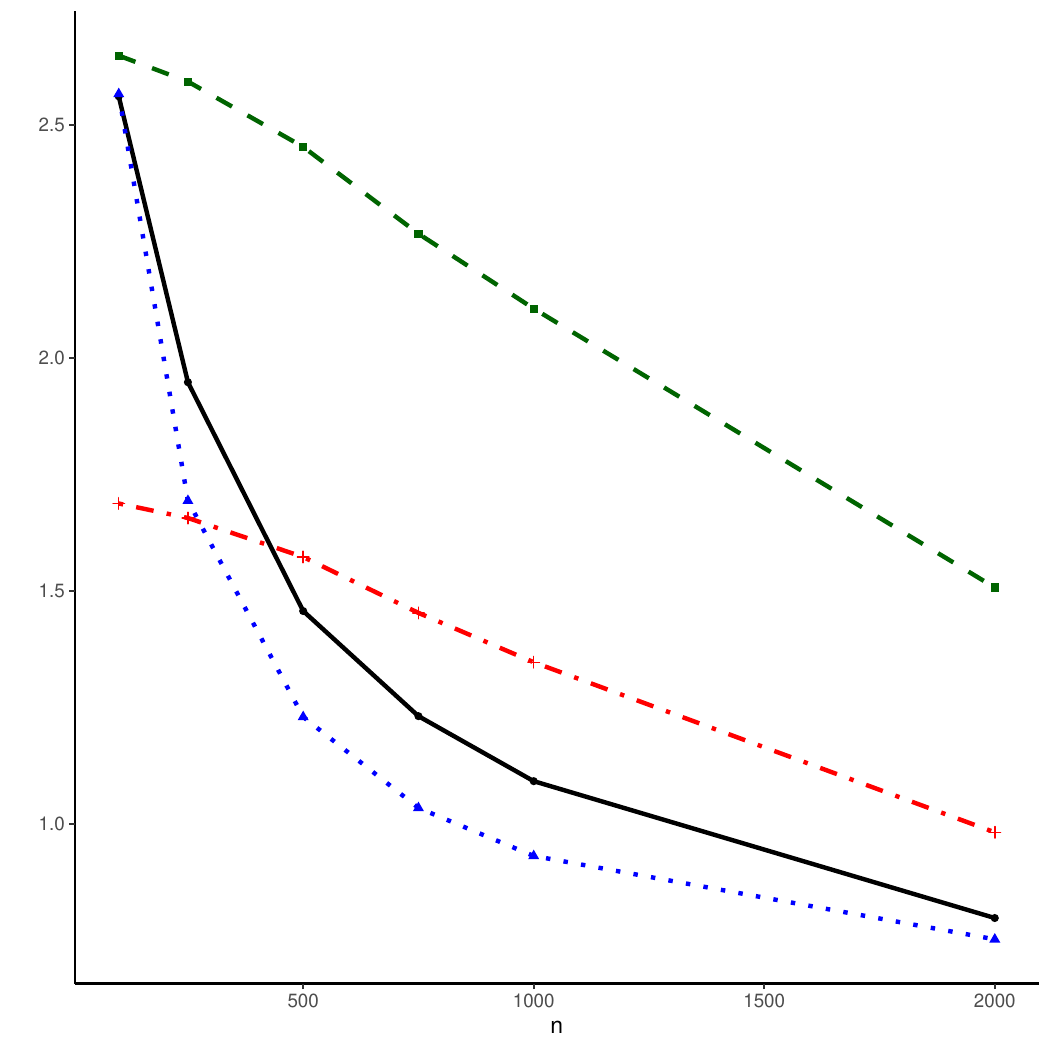}
		\subcaption{$\x=-1$}
	\end{subfigure}%
	\begin{subfigure}[b]{0.5\textwidth}
		\includegraphics[height=0.3\textheight,width=0.95\textwidth]{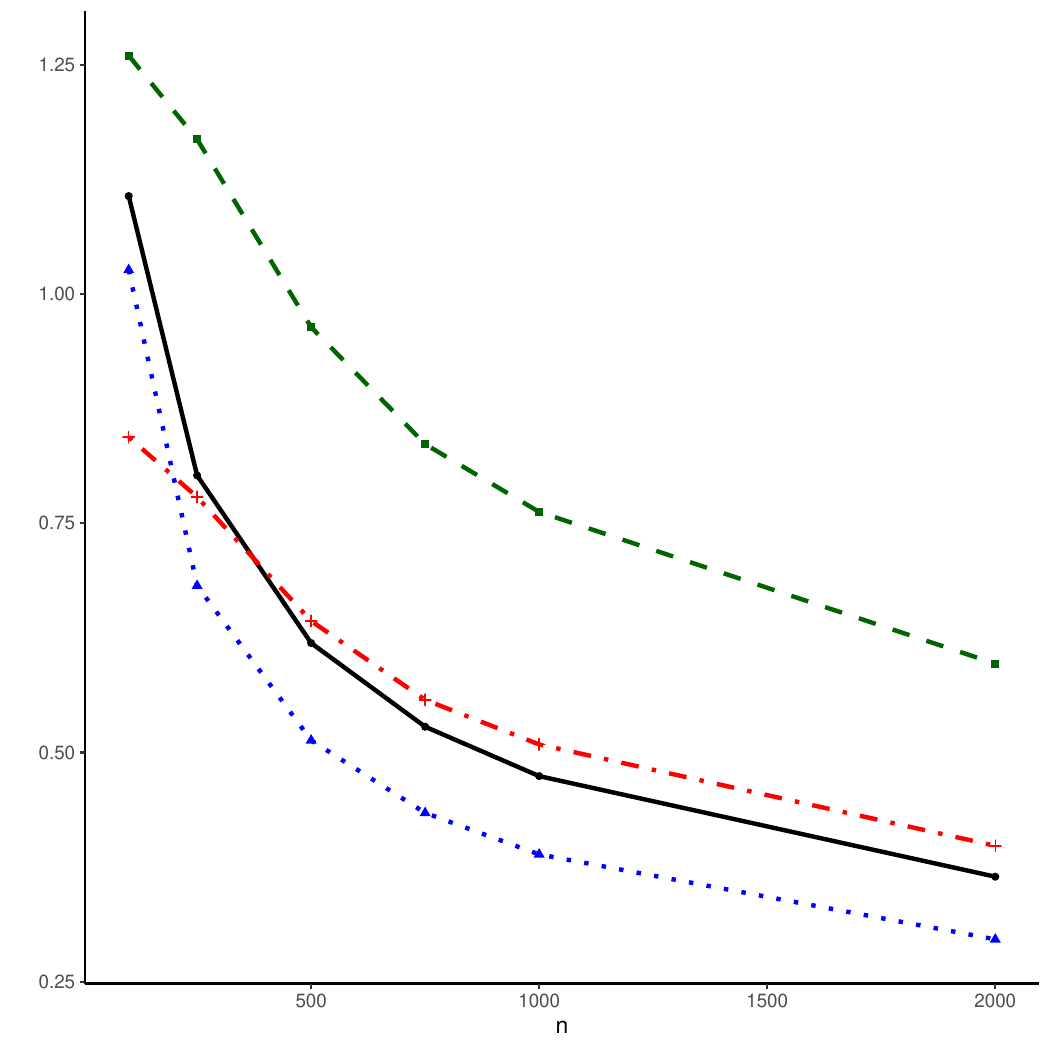}
		\subcaption{$\x=-0.6$}
	\end{subfigure}	
	\begin{subfigure}[b]{0.5\textwidth}
		\includegraphics[height=0.3\textheight,width=0.95\textwidth]{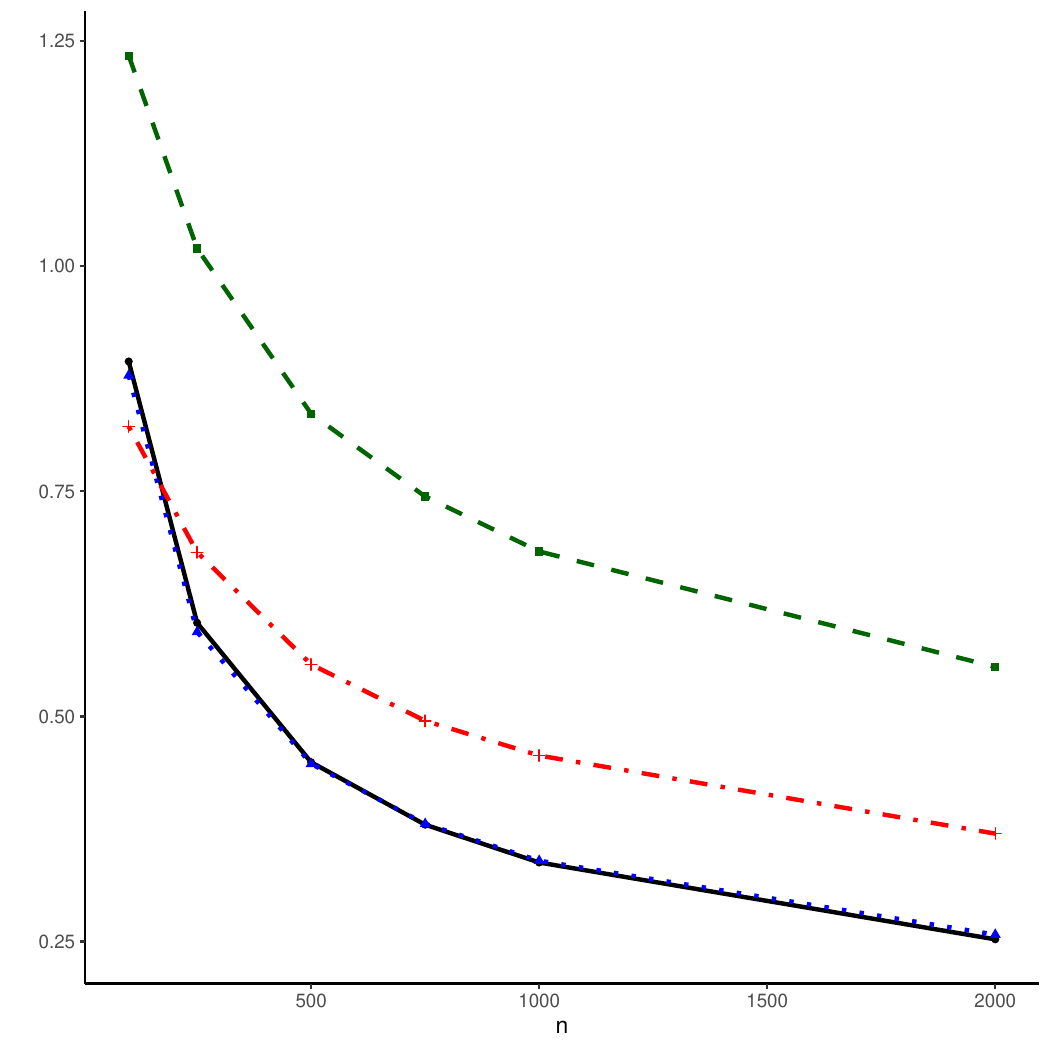}
		\subcaption{$\x=-0.2$}
	\end{subfigure}%	
	\begin{subfigure}[b]{0.5\textwidth}
		\includegraphics[height=0.3\textheight,width=0.95\textwidth]{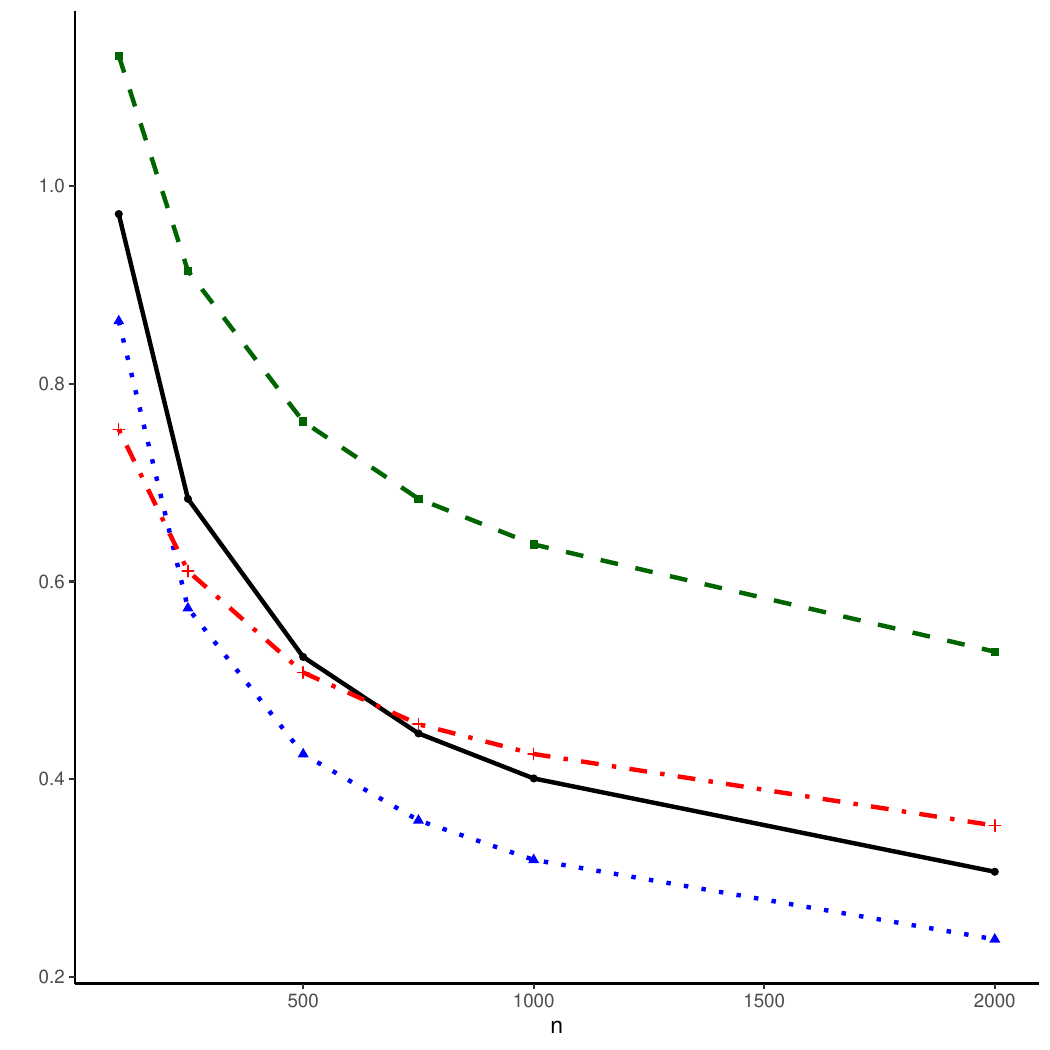}
		\subcaption{$\x=0.2$}
	\end{subfigure}	
	\begin{subfigure}[b]{0.5\textwidth}
		\includegraphics[height=0.3\textheight,width=0.95\textwidth]{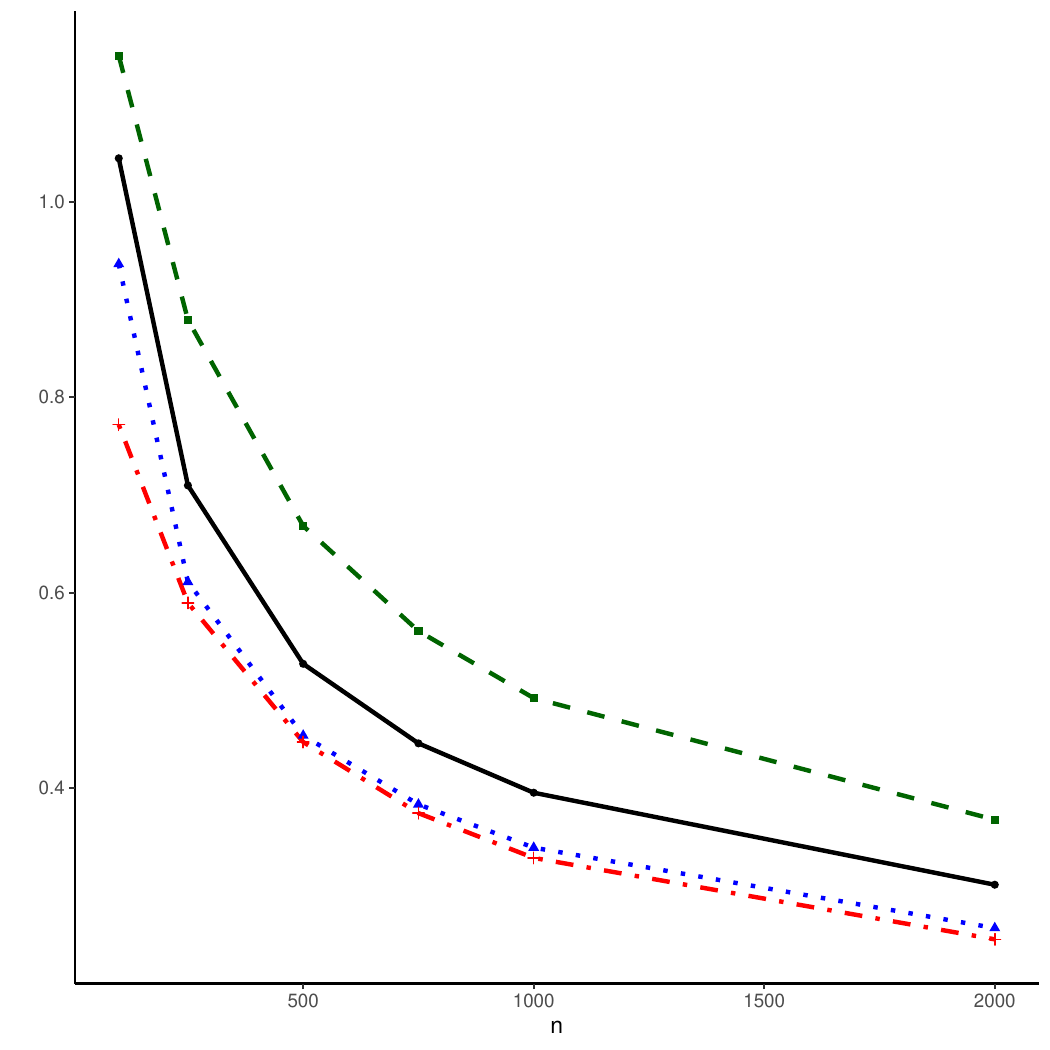}
		\subcaption{$\x=0.6$}
	\end{subfigure}%
	\begin{subfigure}[b]{0.5\textwidth}
		\includegraphics[height=0.3\textheight,width=0.95\textwidth]{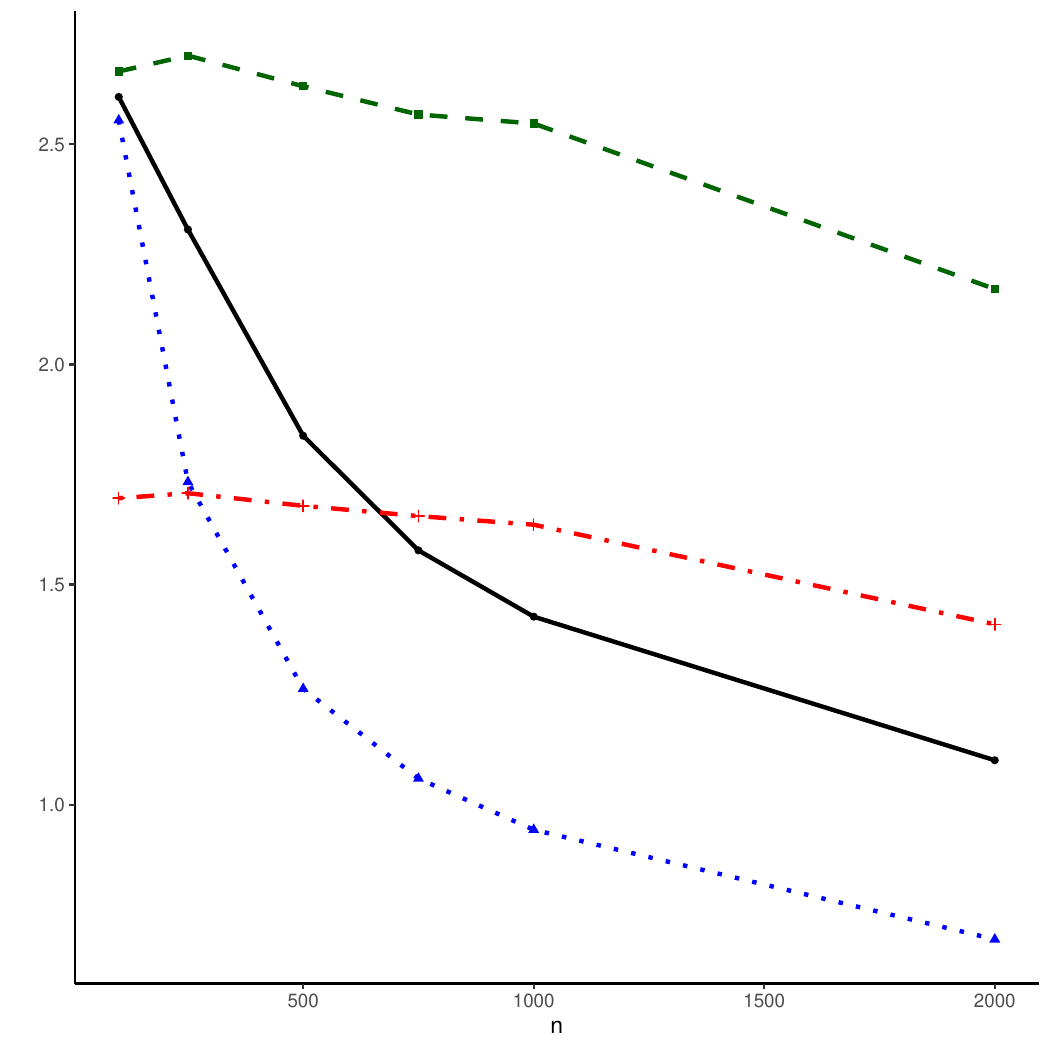}
		\subcaption{$\x=1$}
	\end{subfigure}%
	\begin{flushleft}
\footnotesize Notes: \blackline $\irbc(\hat{h}_{\RBC})$, \blueline $\irbc(\hat{h}_{\MSE})$, \greenline   $\irbc(\hat{h}_{\US})$,  \redline$\ius(\hat{h}_{\US})$
\end{flushleft}
\end{figure}

\clearpage
\begin{figure}[!htb]
	\centering
	\caption{Average Interval Length for 95\% Confidence Intervals\\
		Uniform Kernel, $\v=1$}
	\label{suppfig:il_nu1_uni}	
	\begin{subfigure}[b]{0.5\textwidth}
		\includegraphics[height=0.3\textheight,width=0.95\textwidth]{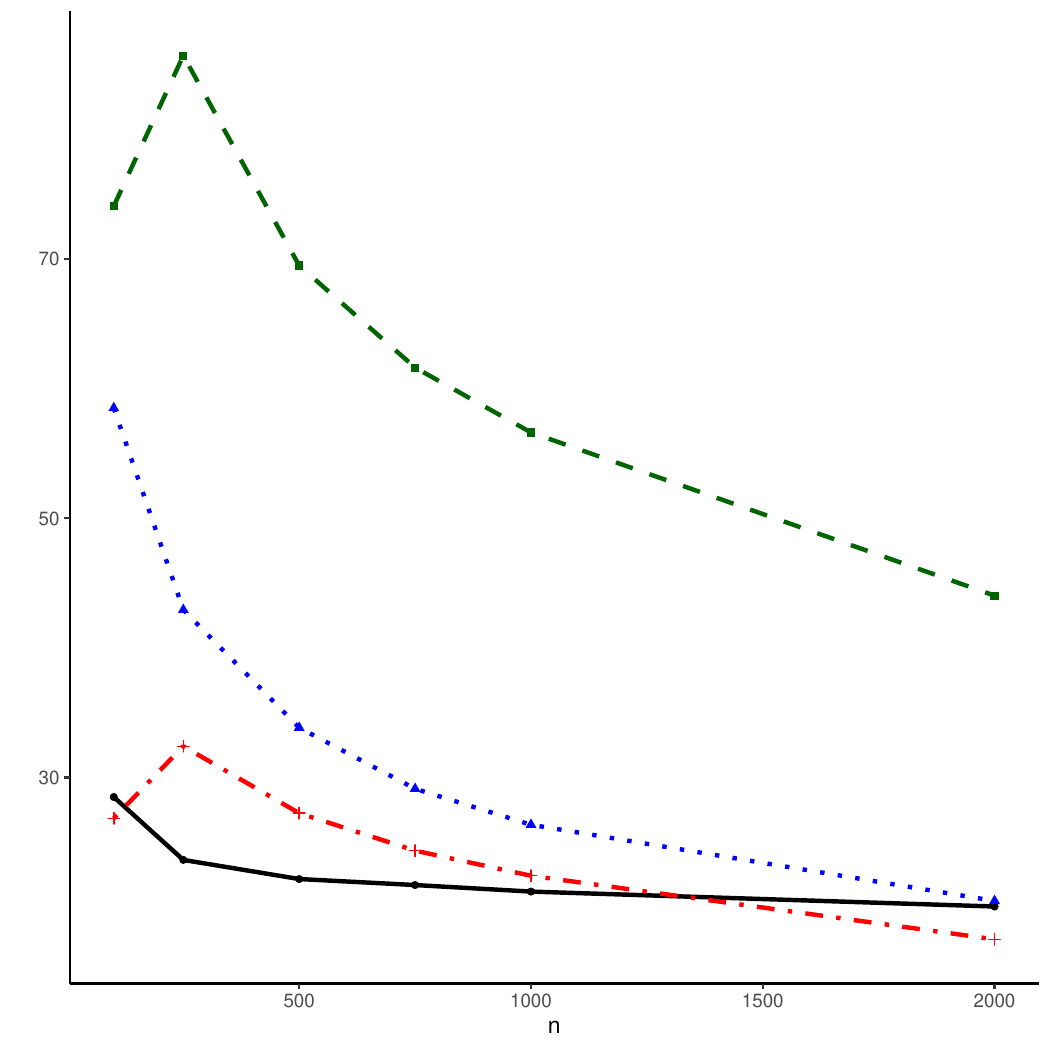}
		\subcaption{$\x=-1$}
	\end{subfigure}%
	\begin{subfigure}[b]{0.5\textwidth}
		\includegraphics[height=0.3\textheight,width=0.95\textwidth]{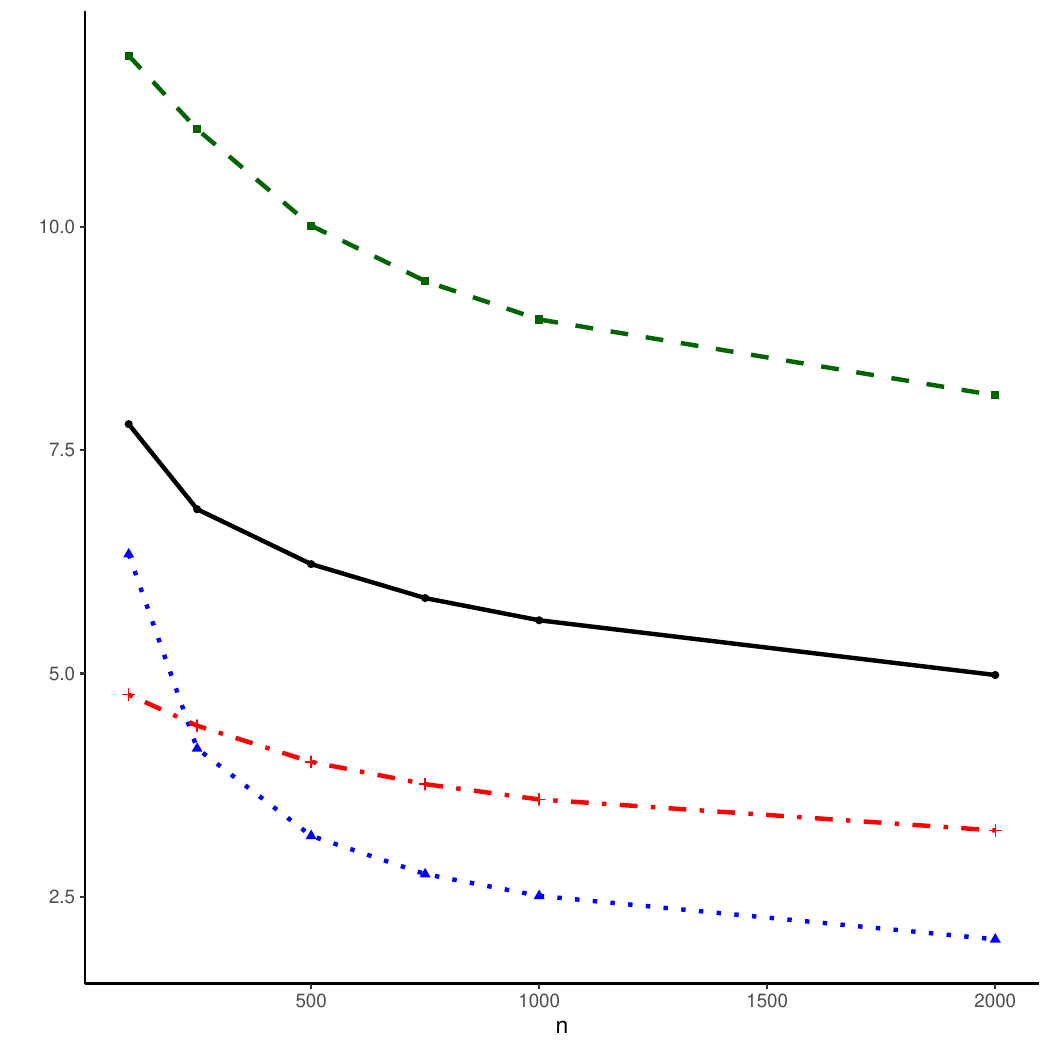}
		\subcaption{$\x=-0.6$}
	\end{subfigure}	
	\begin{subfigure}[b]{0.5\textwidth}
		\includegraphics[height=0.3\textheight,width=0.95\textwidth]{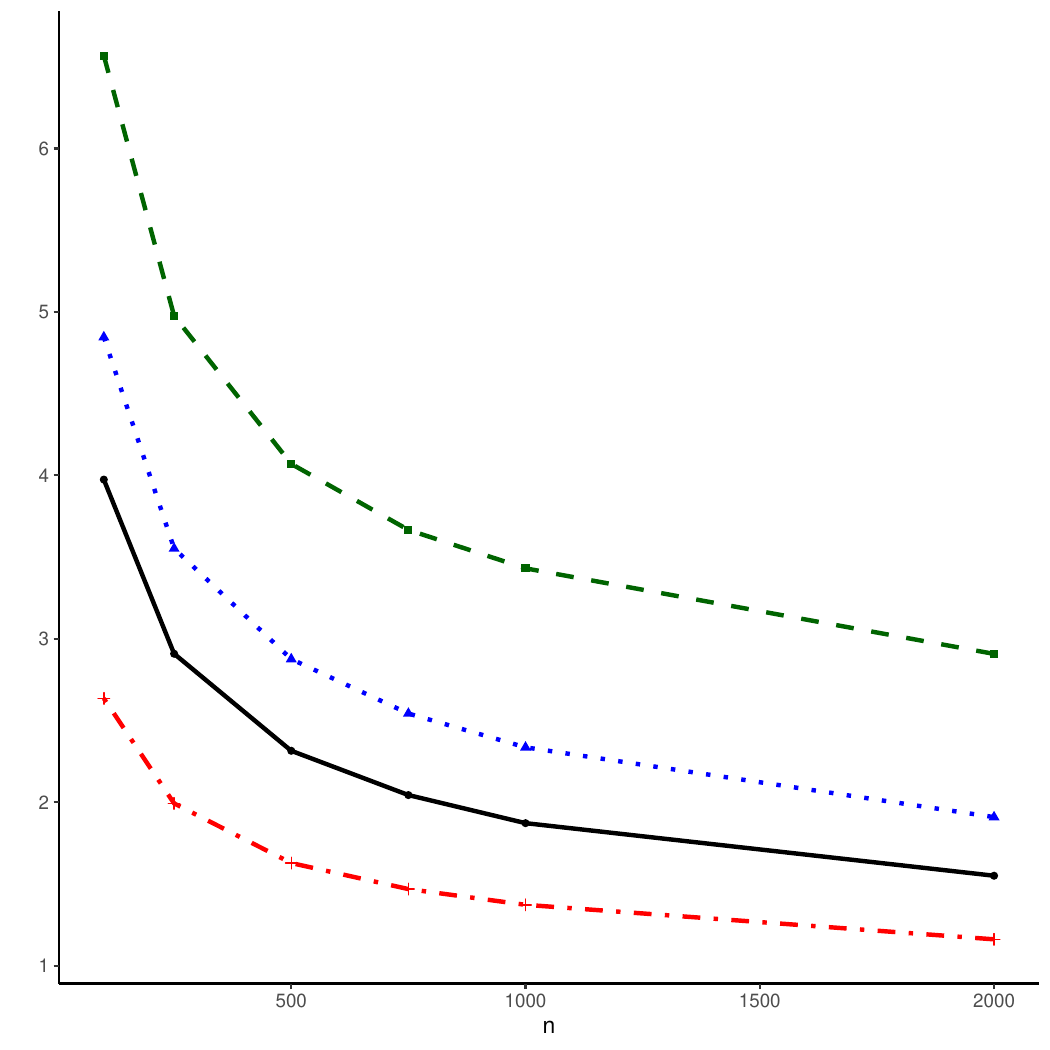}
		\subcaption{$\x=-0.2$}
	\end{subfigure}%	
	\begin{subfigure}[b]{0.5\textwidth}
		\includegraphics[height=0.3\textheight,width=0.95\textwidth]{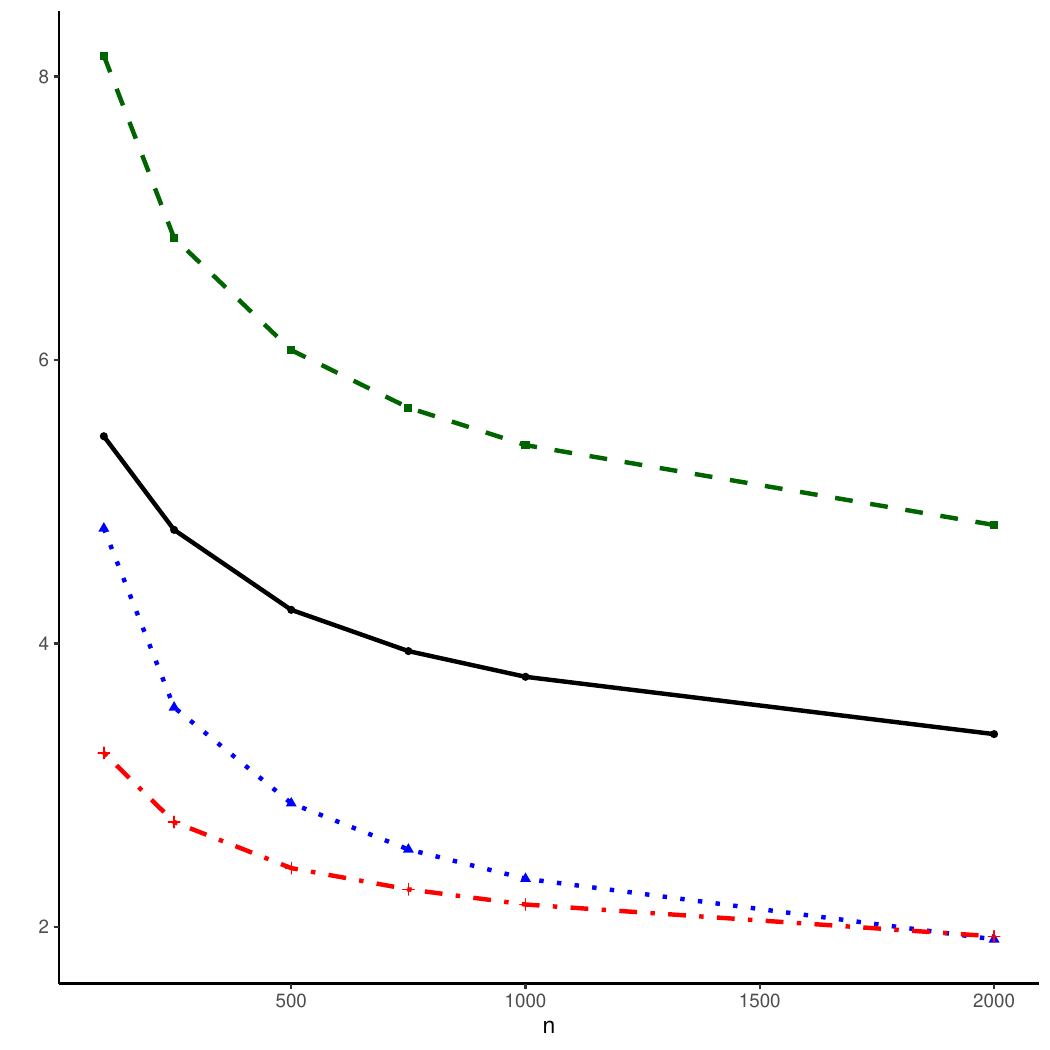}
		\subcaption{$\x=0.2$}
	\end{subfigure}	
	\begin{subfigure}[b]{0.5\textwidth}
		\includegraphics[height=0.3\textheight,width=0.95\textwidth]{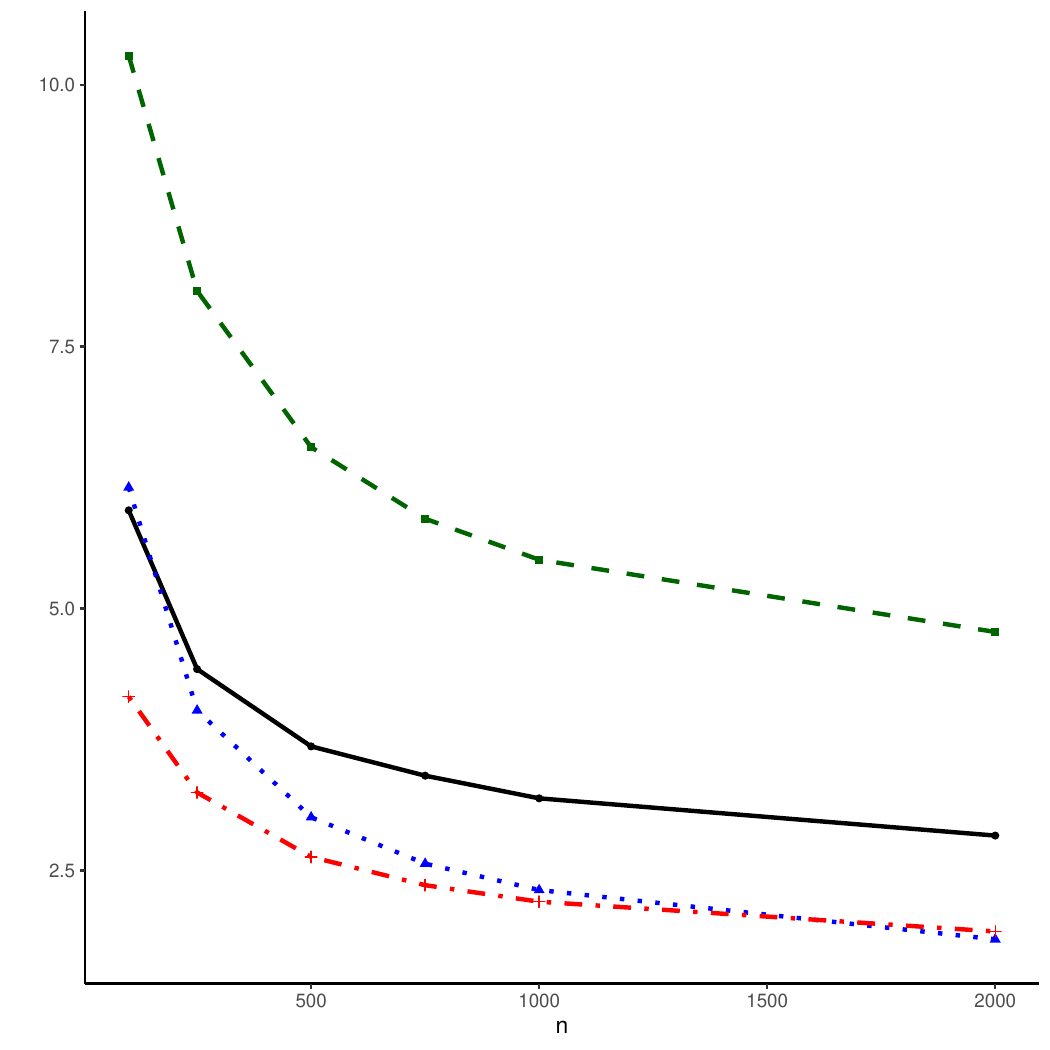}
		\subcaption{$\x=0.6$}
	\end{subfigure}%
	\begin{subfigure}[b]{0.5\textwidth}
		\includegraphics[height=0.3\textheight,width=0.95\textwidth]{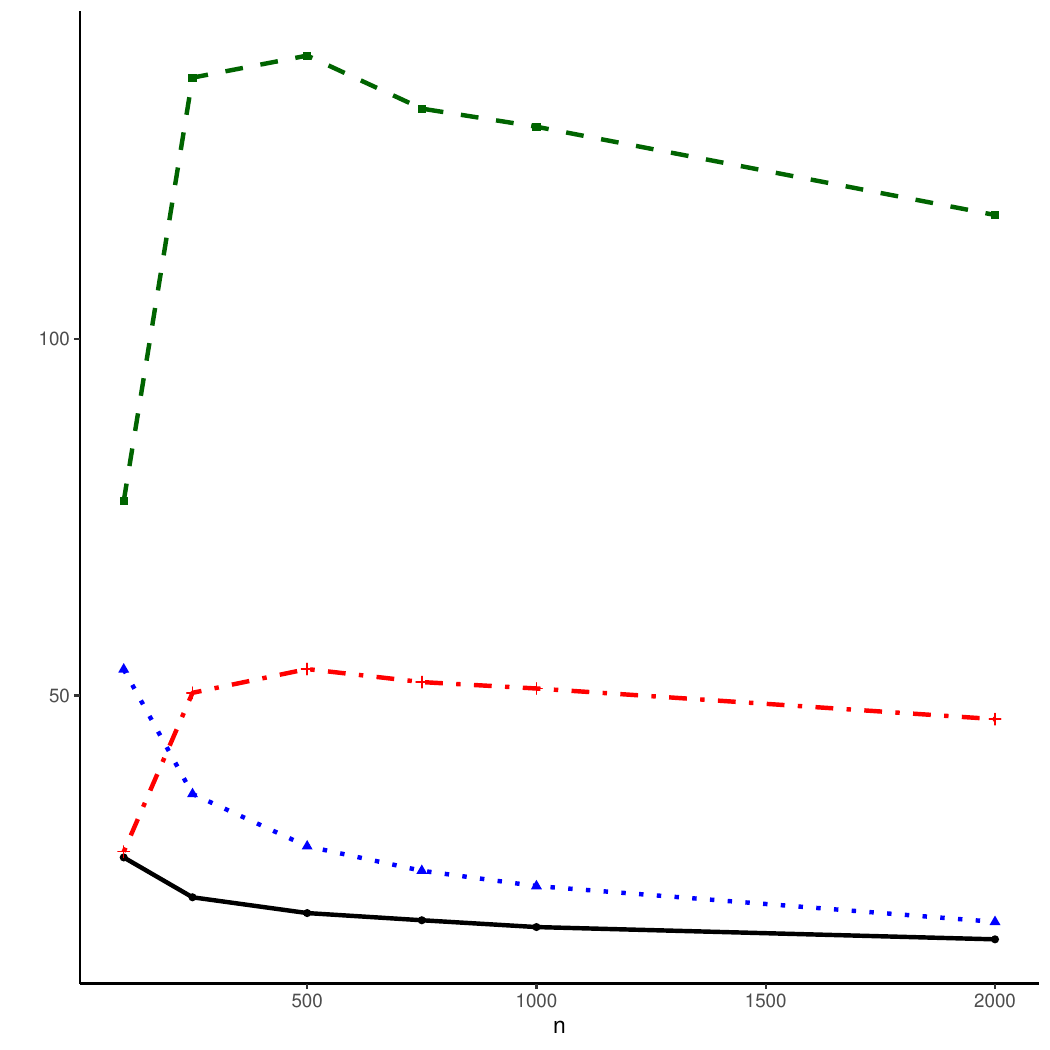}
		\subcaption{$\x=1$}
	\end{subfigure}%
	\begin{flushleft}
\footnotesize Notes: \blackline $\irbc(\hat{h}_{\RBC})$, \blueline $\irbc(\hat{h}_{\MSE})$, \greenline   $\irbc(\hat{h}_{\US})$,  \redline$\ius(\hat{h}_{\US})$
\end{flushleft}
\end{figure}

%%%%%%%%%%%%%%%%%%%%%% BWS EPA

\clearpage
\begin{figure}[!htb]
	\centering
	\caption{Average Estimated Bandwidths, Epanechnikov Kernel,  $\v=0$}
	\label{suppfig:h_nu0_epa}	
	\begin{subfigure}[b]{0.5\textwidth}
		\includegraphics[height=0.3\textheight,width=0.95\textwidth]{simuls/output/h_kepa_p1_d0_x1.pdf}
		\subcaption{$\x=-1$}
	\end{subfigure}%
	\begin{subfigure}[b]{0.5\textwidth}
		\includegraphics[height=0.3\textheight,width=0.95\textwidth]{simuls/output/h_kepa_p1_d0_x2.pdf}
		\subcaption{$\x=-0.6$}
	\end{subfigure}	
	\begin{subfigure}[b]{0.5\textwidth}
		\includegraphics[height=0.3\textheight,width=0.95\textwidth]{simuls/output/h_kepa_p1_d0_x3.pdf}
		\subcaption{$\x=-0.2$}
	\end{subfigure}%	
	\begin{subfigure}[b]{0.5\textwidth}
		\includegraphics[height=0.3\textheight,width=0.95\textwidth]{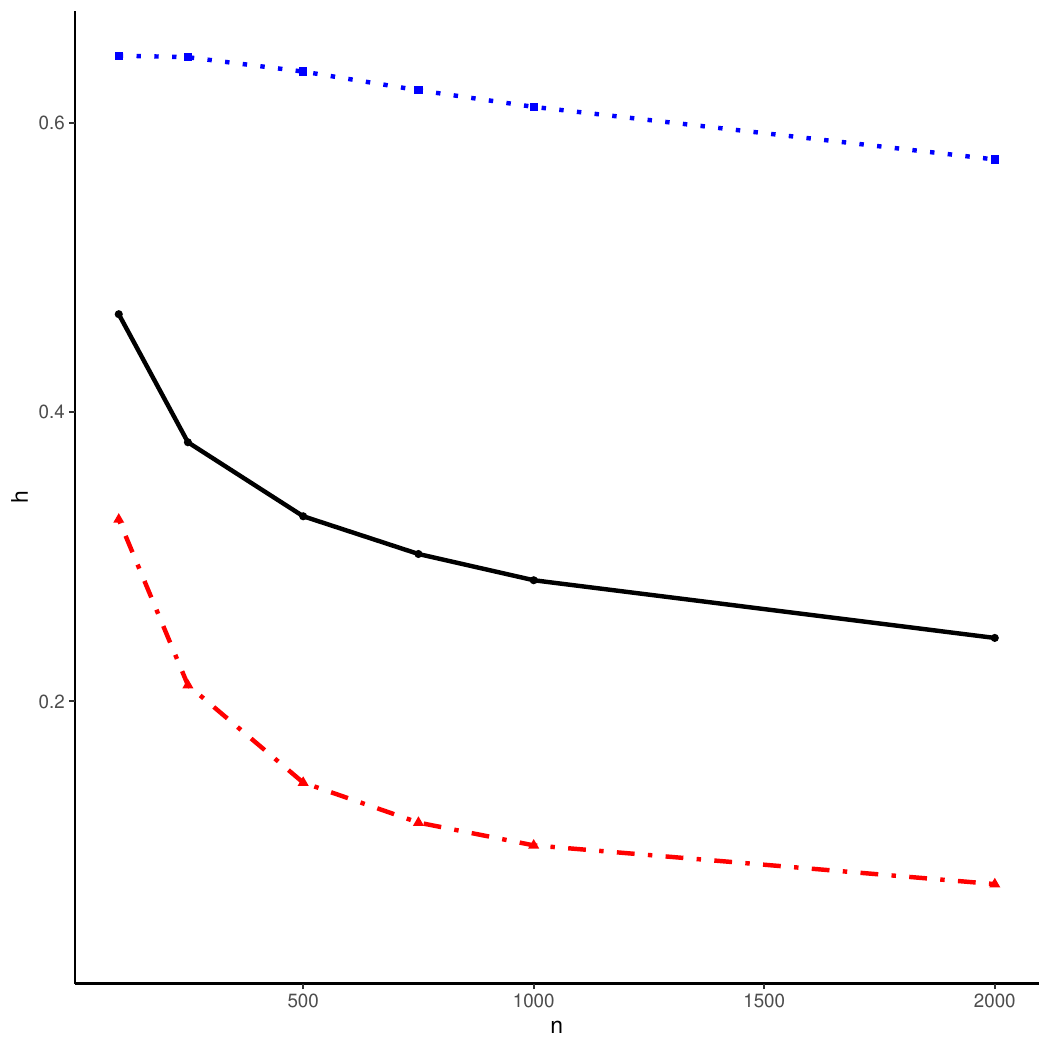}
		\subcaption{$\x=0.2$}
	\end{subfigure}	
	\begin{subfigure}[b]{0.5\textwidth}
		\includegraphics[height=0.3\textheight,width=0.95\textwidth]{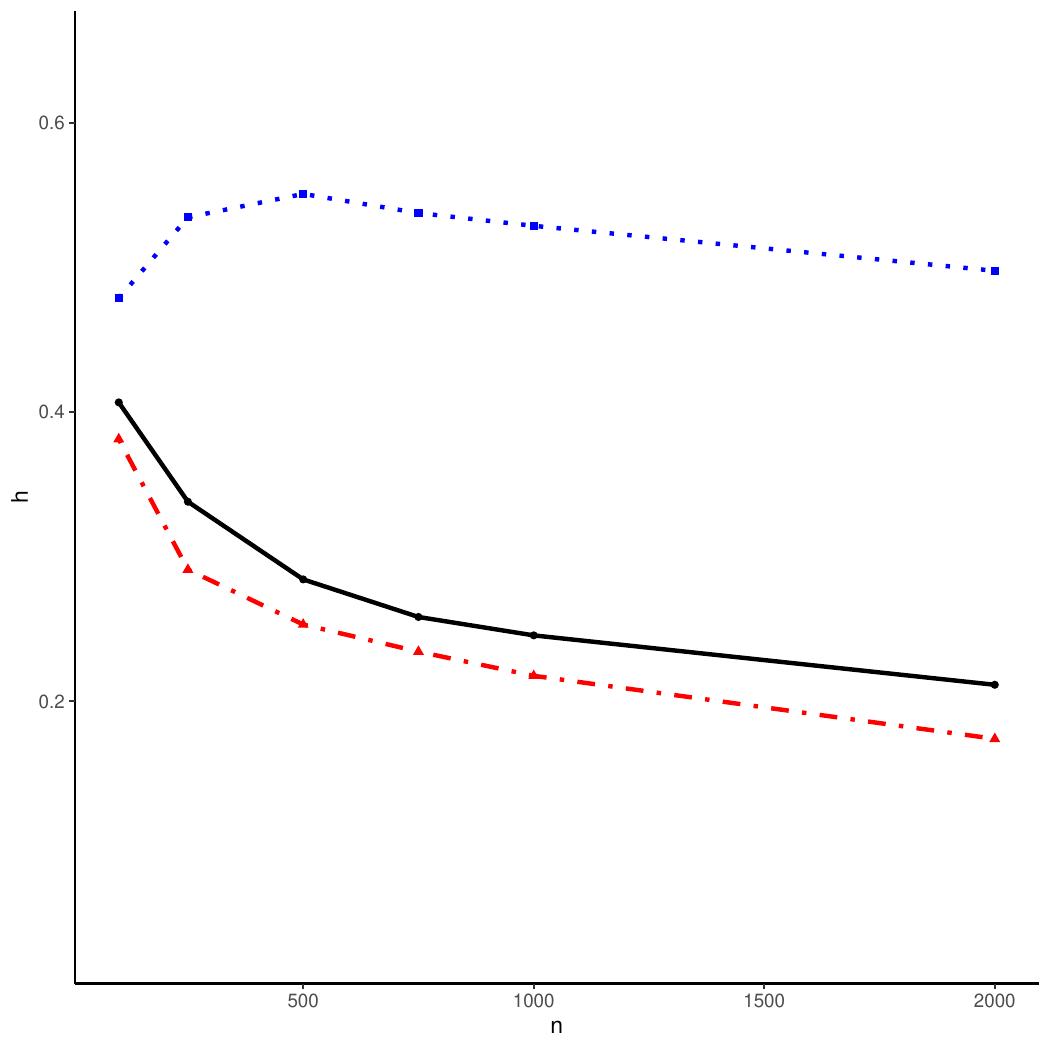}
		\subcaption{$\x=0.6$}
	\end{subfigure}%
	\begin{subfigure}[b]{0.5\textwidth}
		\includegraphics[height=0.3\textheight,width=0.95\textwidth]{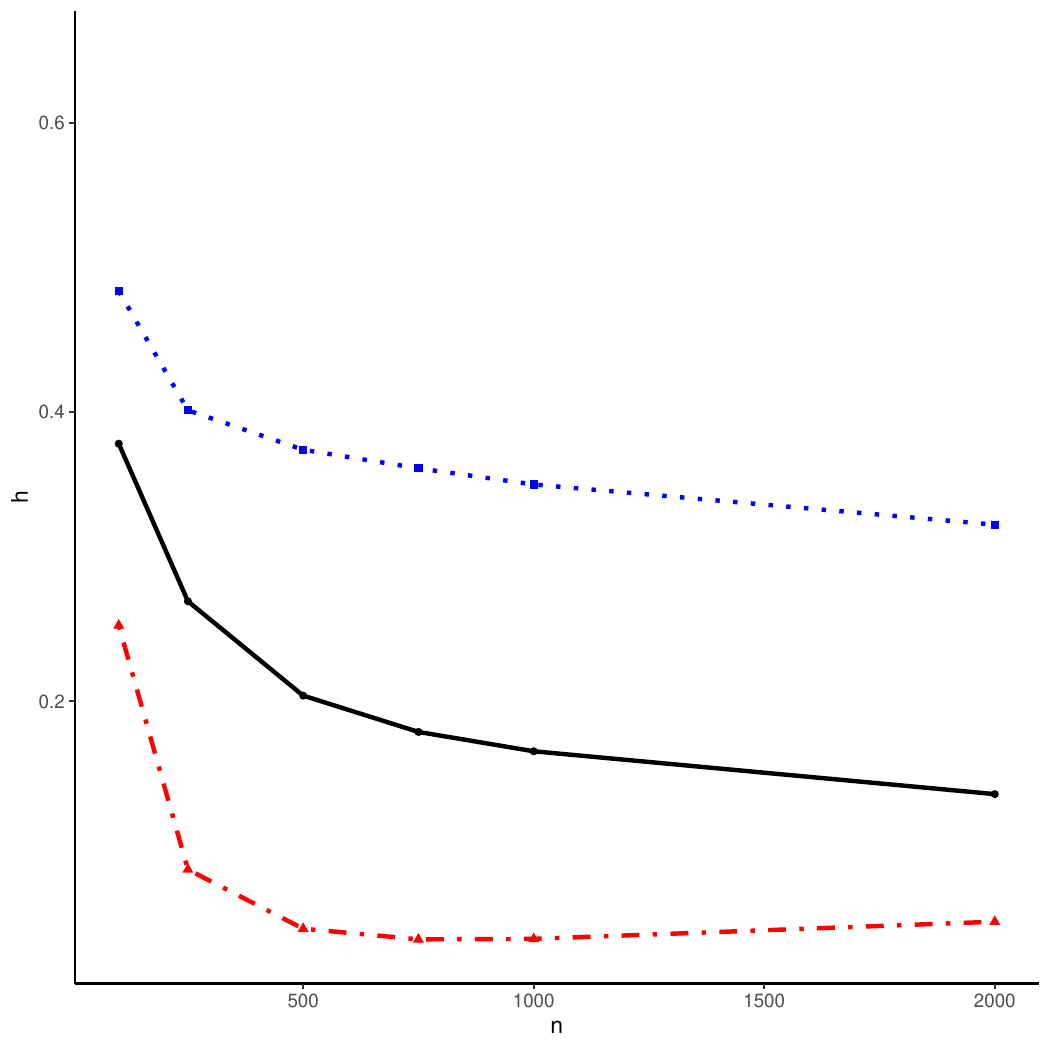}
		\subcaption{$\x=1$}
	\end{subfigure}%
	\begin{flushleft}\footnotesize Notes: \blackline $\hat{h}_{\RBC}$, \redline $\hat{h}_{\US}$, \blueline $\hat{h}_{\MSE}$
\end{flushleft}
\end{figure}

\clearpage
\begin{figure}[!htb]
	\centering
	\caption{Average Estimated Bandwidths, Epanechnikov Kernel,  $\v=1$}
\label{suppfig:h_nu1_epa}	
	\begin{subfigure}[b]{0.5\textwidth}
		\includegraphics[height=0.3\textheight,width=0.95\textwidth]{simuls/output/h_kepa_p2_d1_x1.pdf}
		\subcaption{$\x=-1$}
	\end{subfigure}%
	\begin{subfigure}[b]{0.5\textwidth}
		\includegraphics[height=0.3\textheight,width=0.95\textwidth]{simuls/output/h_kepa_p2_d1_x2.pdf}
		\subcaption{$\x=-0.6$}
	\end{subfigure}	
	\begin{subfigure}[b]{0.5\textwidth}
		\includegraphics[height=0.3\textheight,width=0.95\textwidth]{simuls/output/h_kepa_p2_d1_x3.pdf}
		\subcaption{$\x=-0.2$}
	\end{subfigure}%	
	\begin{subfigure}[b]{0.5\textwidth}
		\includegraphics[height=0.3\textheight,width=0.95\textwidth]{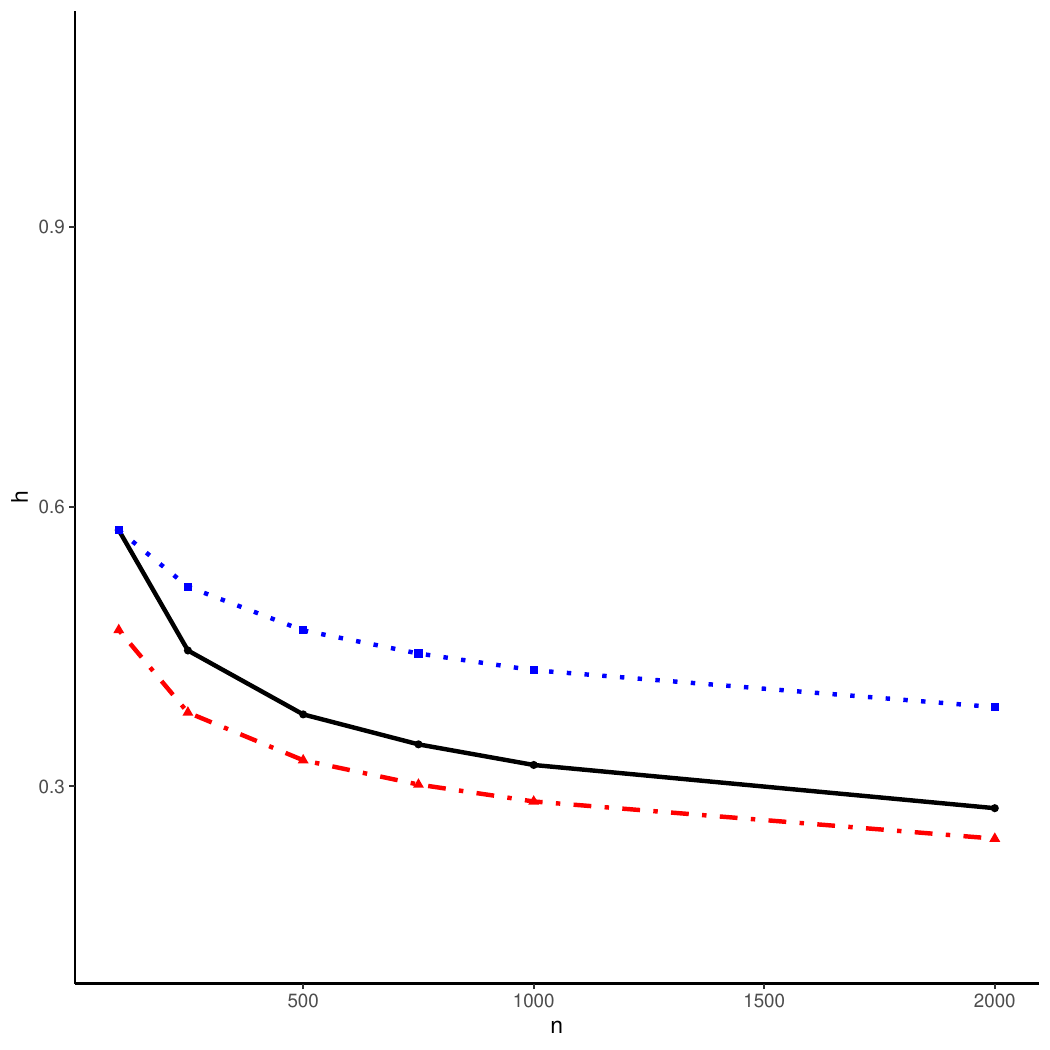}
		\subcaption{$\x=0.2$}
	\end{subfigure}	
	\begin{subfigure}[b]{0.5\textwidth}
		\includegraphics[height=0.3\textheight,width=0.95\textwidth]{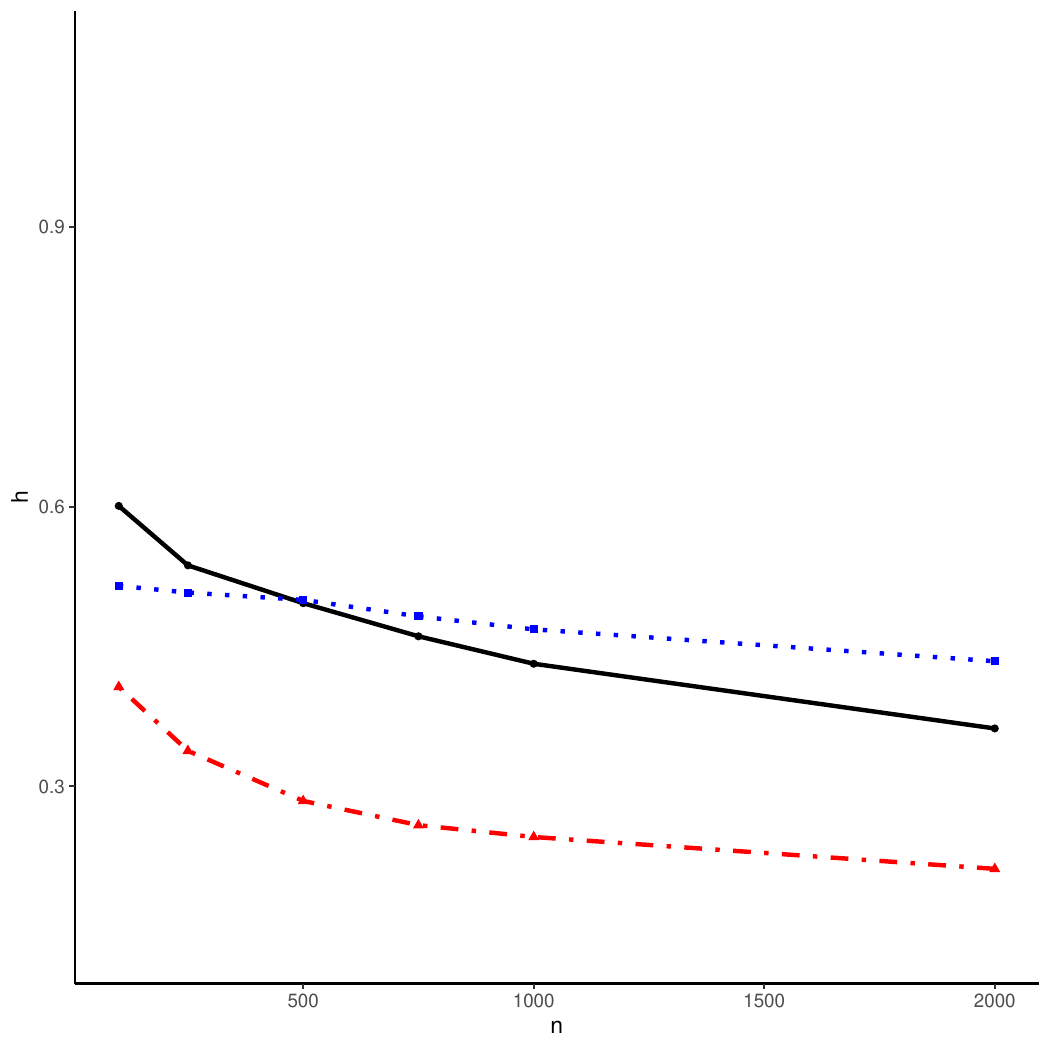}
		\subcaption{$\x=0.6$}
	\end{subfigure}%
	\begin{subfigure}[b]{0.5\textwidth}
		\includegraphics[height=0.3\textheight,width=0.95\textwidth]{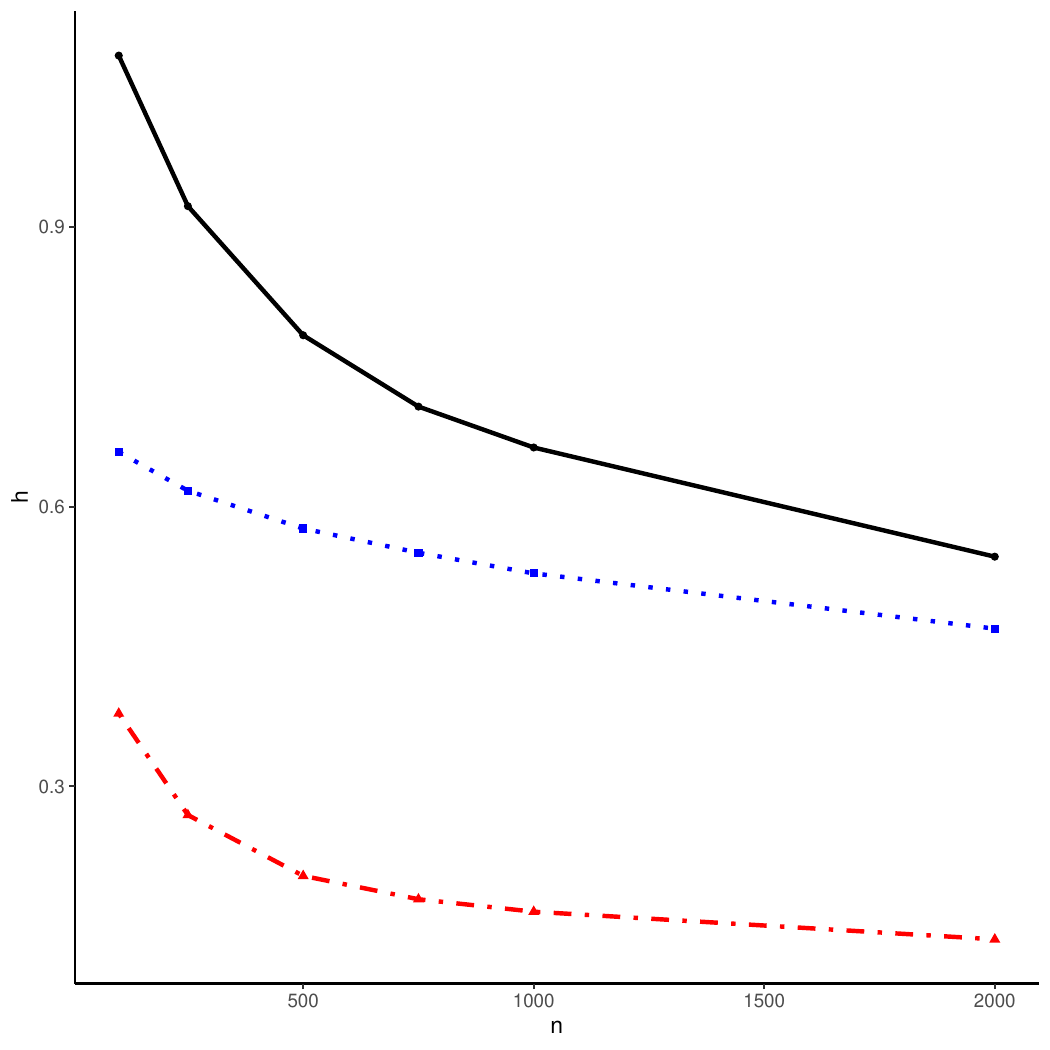}
		\subcaption{$\x=1$}
	\end{subfigure}%
	\begin{flushleft}\footnotesize Notes: \blackline $\hat{h}_{\RBC}$, \redline $\hat{h}_{\US}$, \blueline $\hat{h}_{\MSE}$
		\end{flushleft}
\end{figure}

%%%%%%%%%%%%%%%%%%%%%% BWS UNI

\clearpage
\begin{figure}[!htb]
	\centering
	\caption{Average Estimated Bandwidths, Uniform Kernel,  $\v=0$}
	\label{suppfig:h_nu0_uni}	
	\begin{subfigure}[b]{0.5\textwidth}
		\includegraphics[height=0.3\textheight,width=0.95\textwidth]{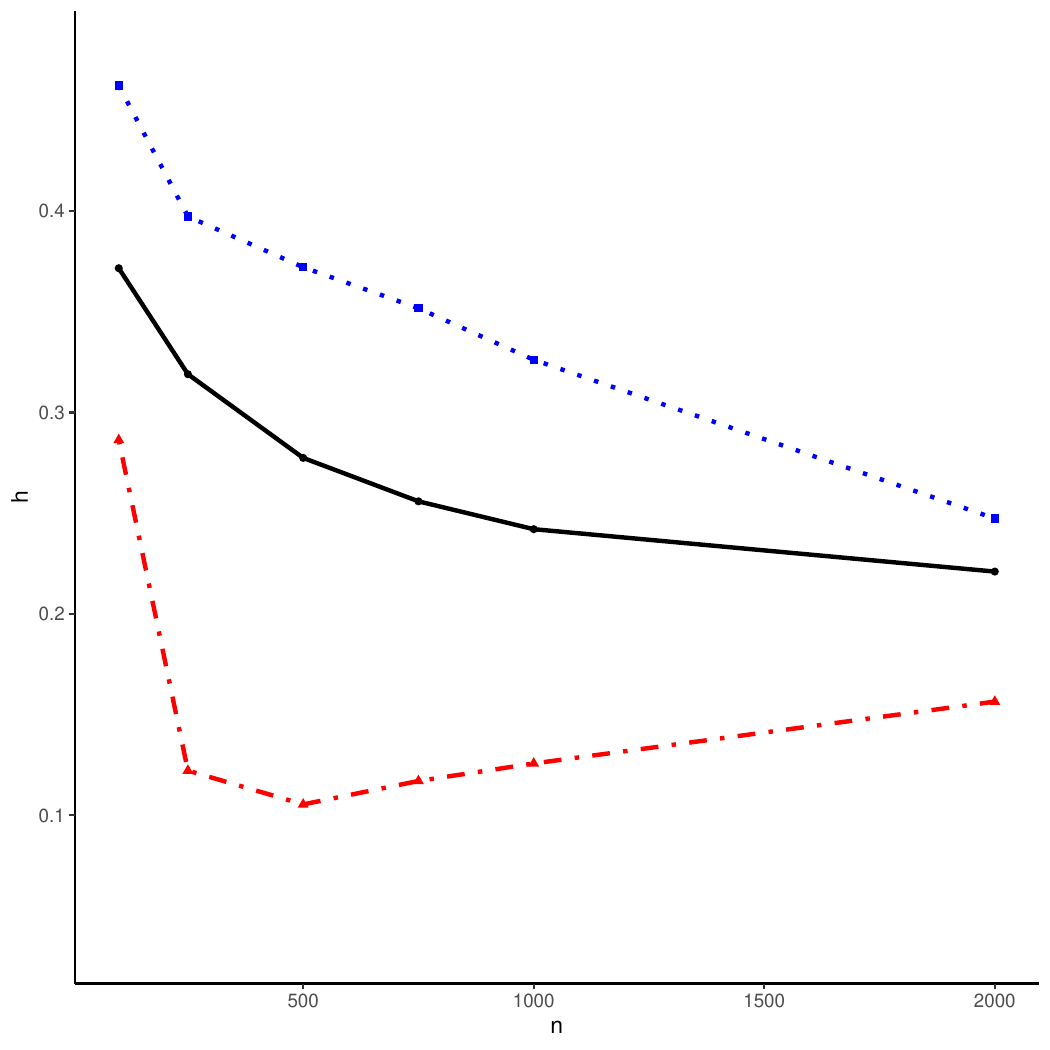}
		\subcaption{$\x=-1$}
	\end{subfigure}%
	\begin{subfigure}[b]{0.5\textwidth}
		\includegraphics[height=0.3\textheight,width=0.95\textwidth]{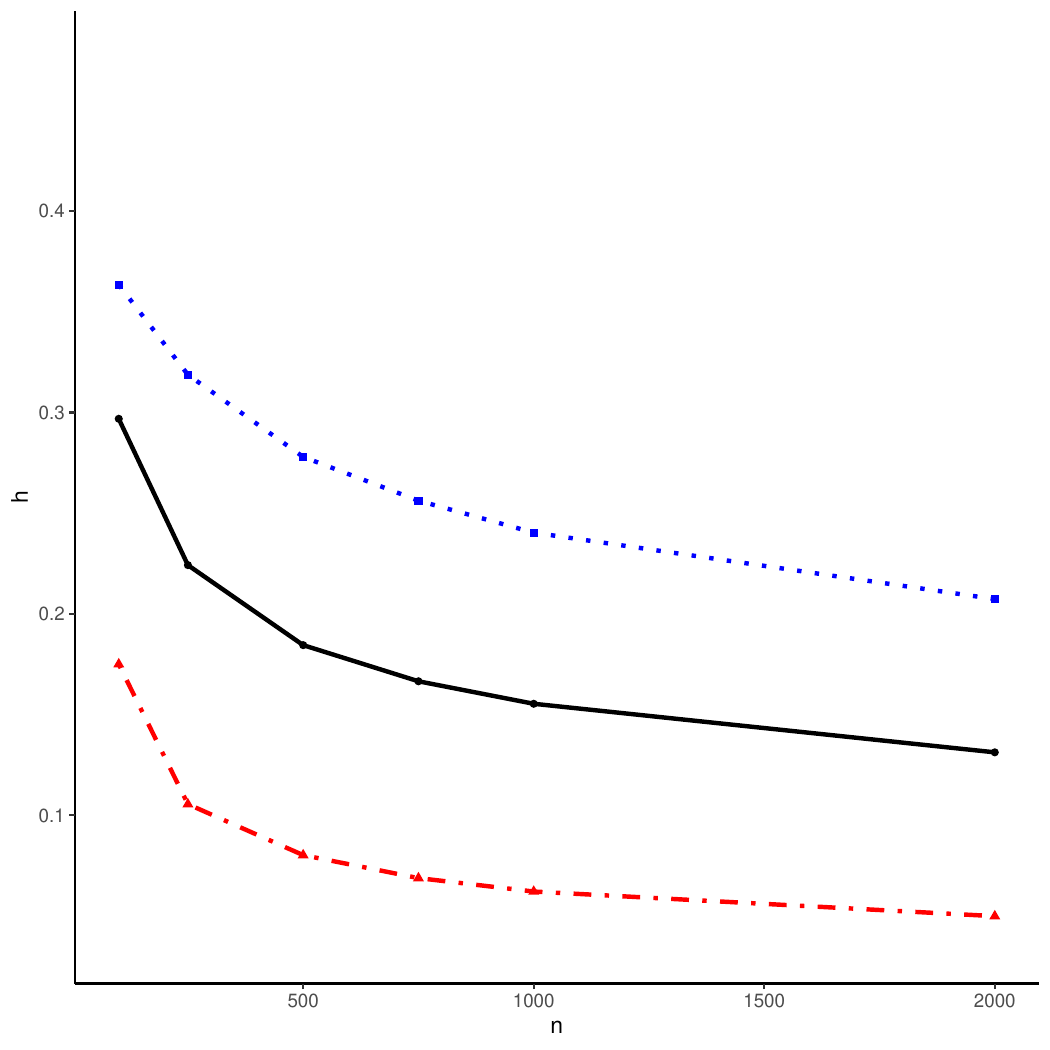}
		\subcaption{$\x=-0.6$}
	\end{subfigure}	
	\begin{subfigure}[b]{0.5\textwidth}
		\includegraphics[height=0.3\textheight,width=0.95\textwidth]{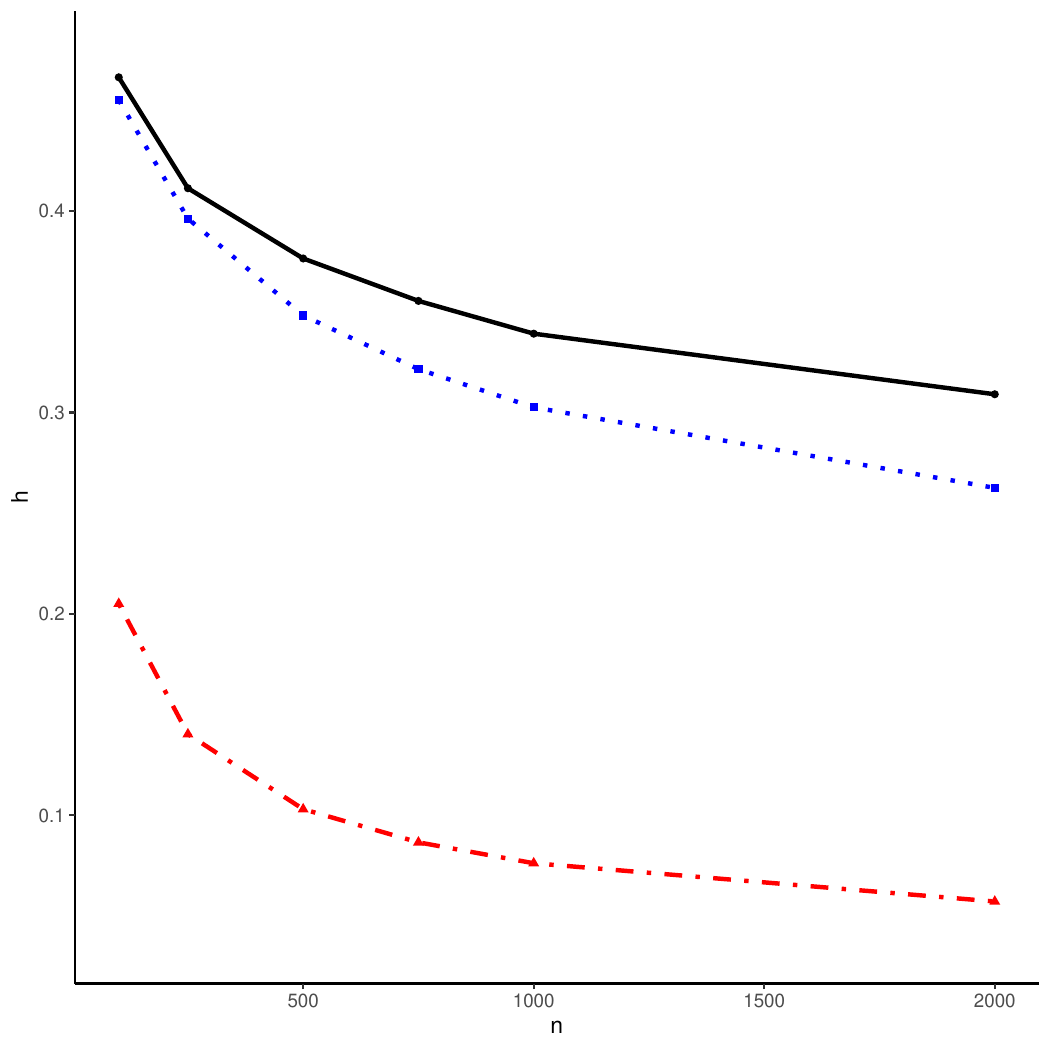}
		\subcaption{$\x=-0.2$}
	\end{subfigure}%	
	\begin{subfigure}[b]{0.5\textwidth}
		\includegraphics[height=0.3\textheight,width=0.95\textwidth]{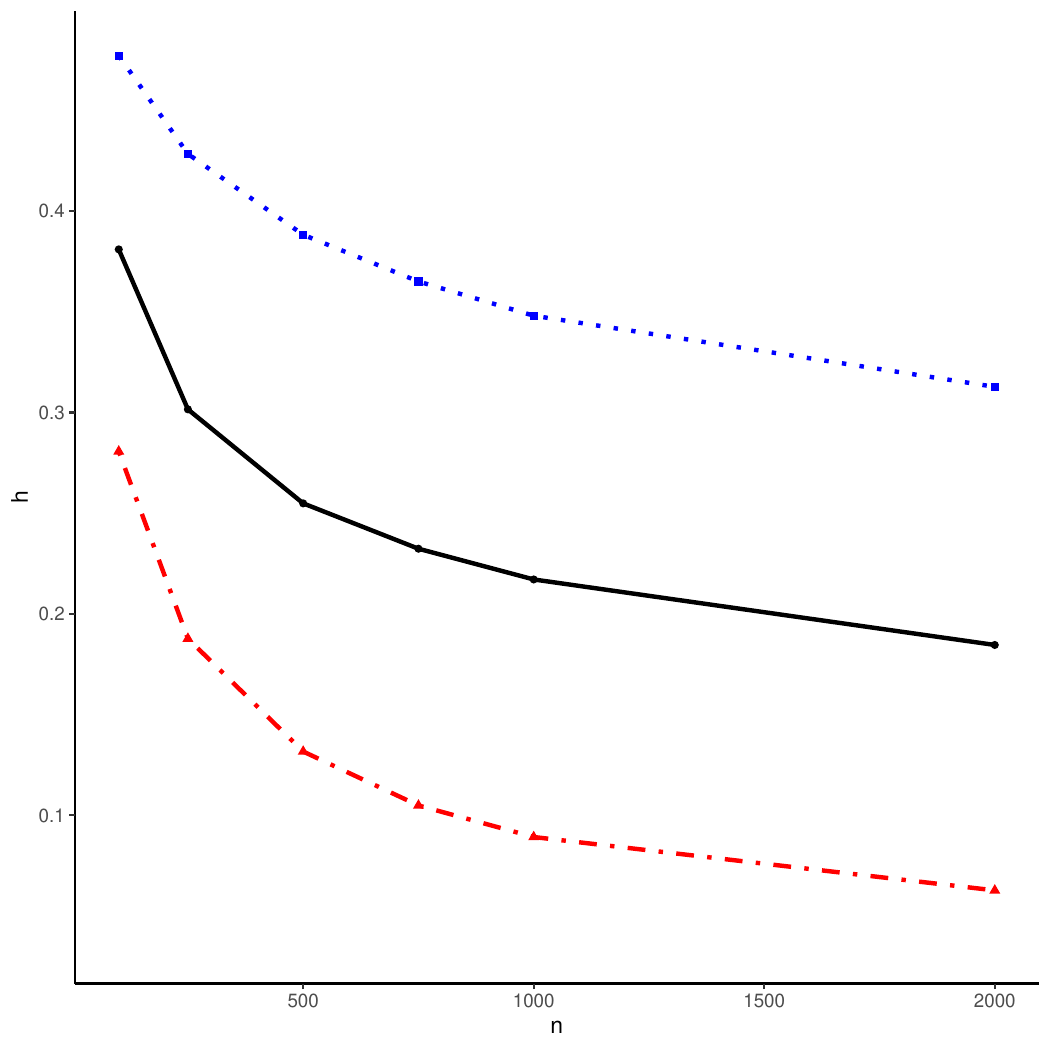}
		\subcaption{$\x=0.2$}
	\end{subfigure}	
	\begin{subfigure}[b]{0.5\textwidth}
		\includegraphics[height=0.3\textheight,width=0.95\textwidth]{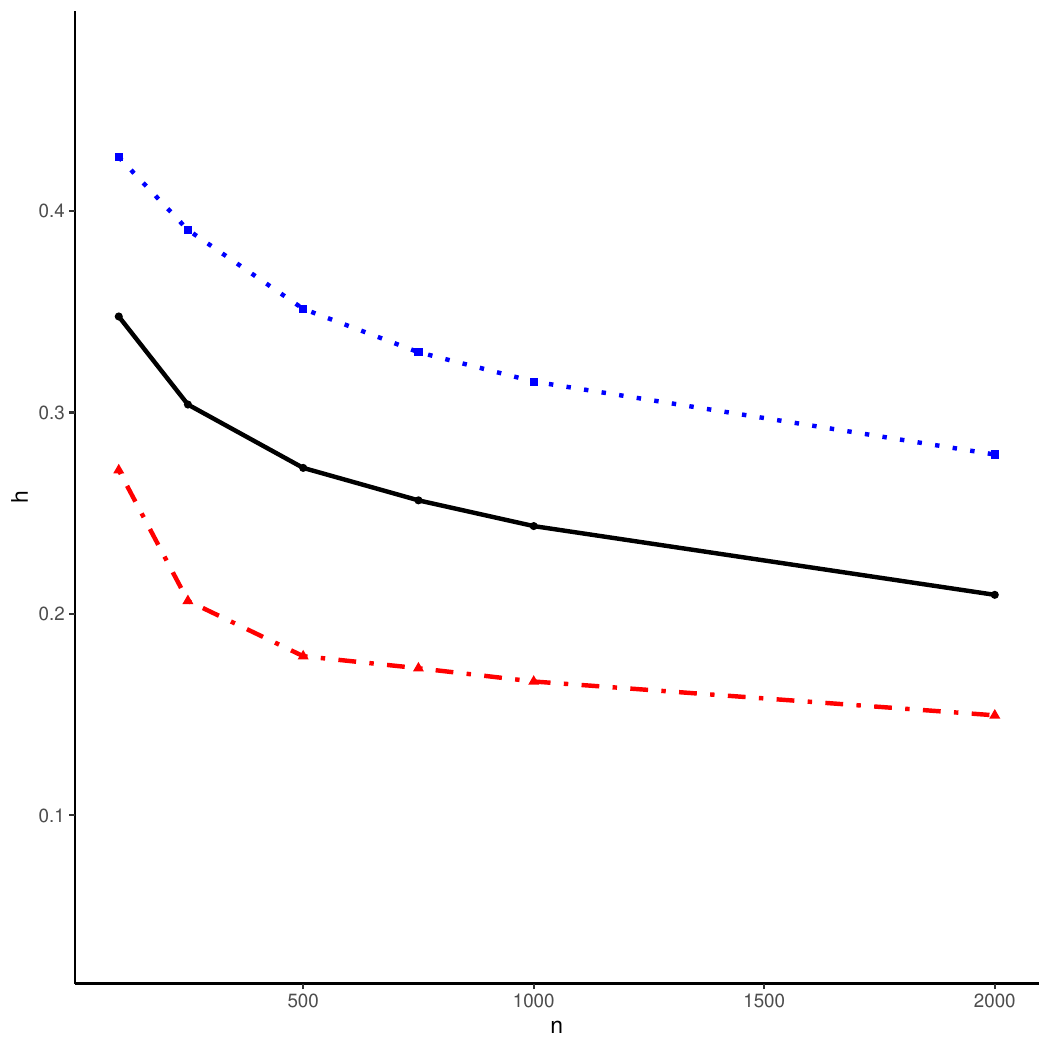}
		\subcaption{$\x=0.6$}
	\end{subfigure}%
	\begin{subfigure}[b]{0.5\textwidth}
		\includegraphics[height=0.3\textheight,width=0.95\textwidth]{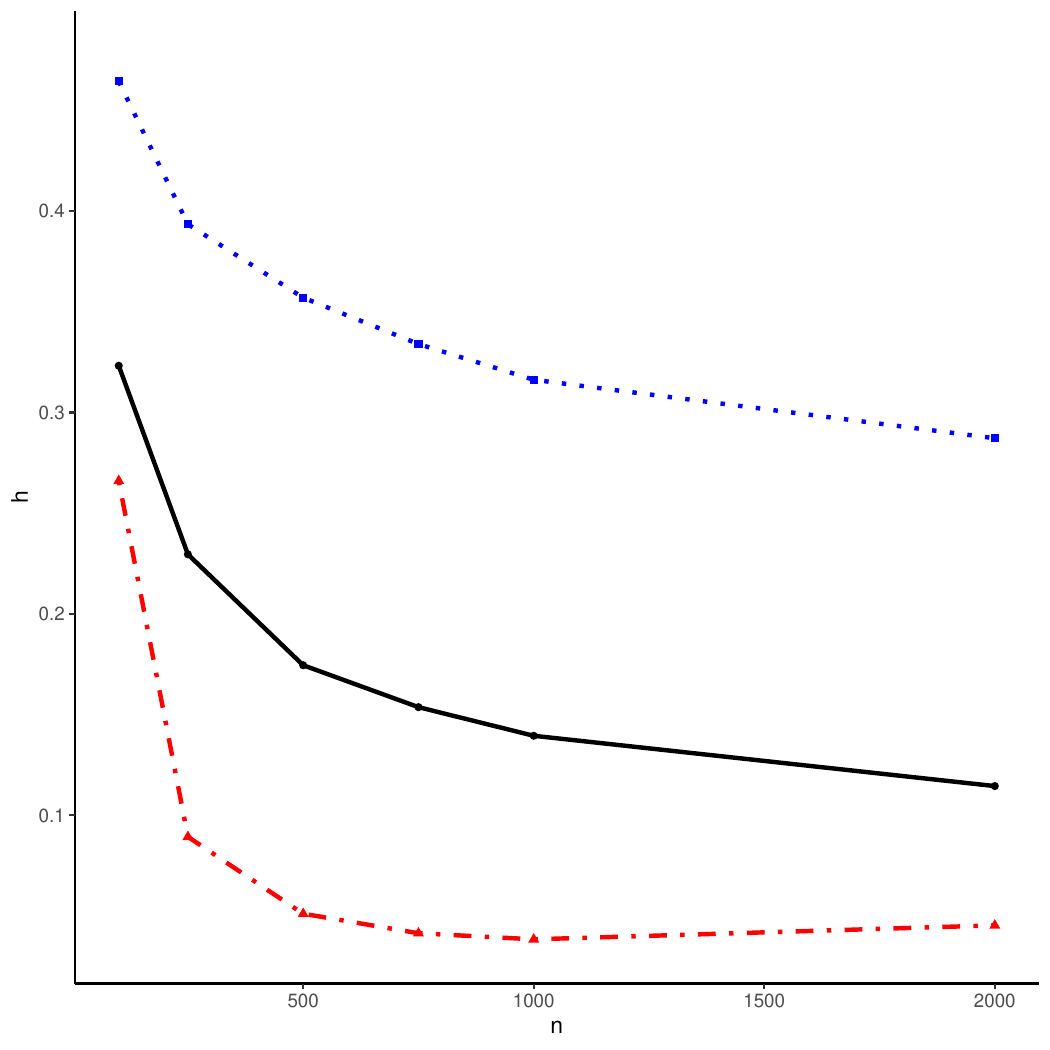}
		\subcaption{$\x=1$}
	\end{subfigure}%
	\begin{flushleft}\footnotesize Notes: \blackline $\hat{h}_{\RBC}$, \redline $\hat{h}_{\US}$, \blueline $\hat{h}_{\MSE}$
	\end{flushleft}
\end{figure}

\clearpage
\begin{figure}[!htb]
	\centering
	\caption{Average Estimated Bandwidths, Uniform Kernel,  $\v=1$}
	\label{suppfig:h_nu1_uni}	
	\begin{subfigure}[b]{0.5\textwidth}
		\includegraphics[height=0.3\textheight,width=0.95\textwidth]{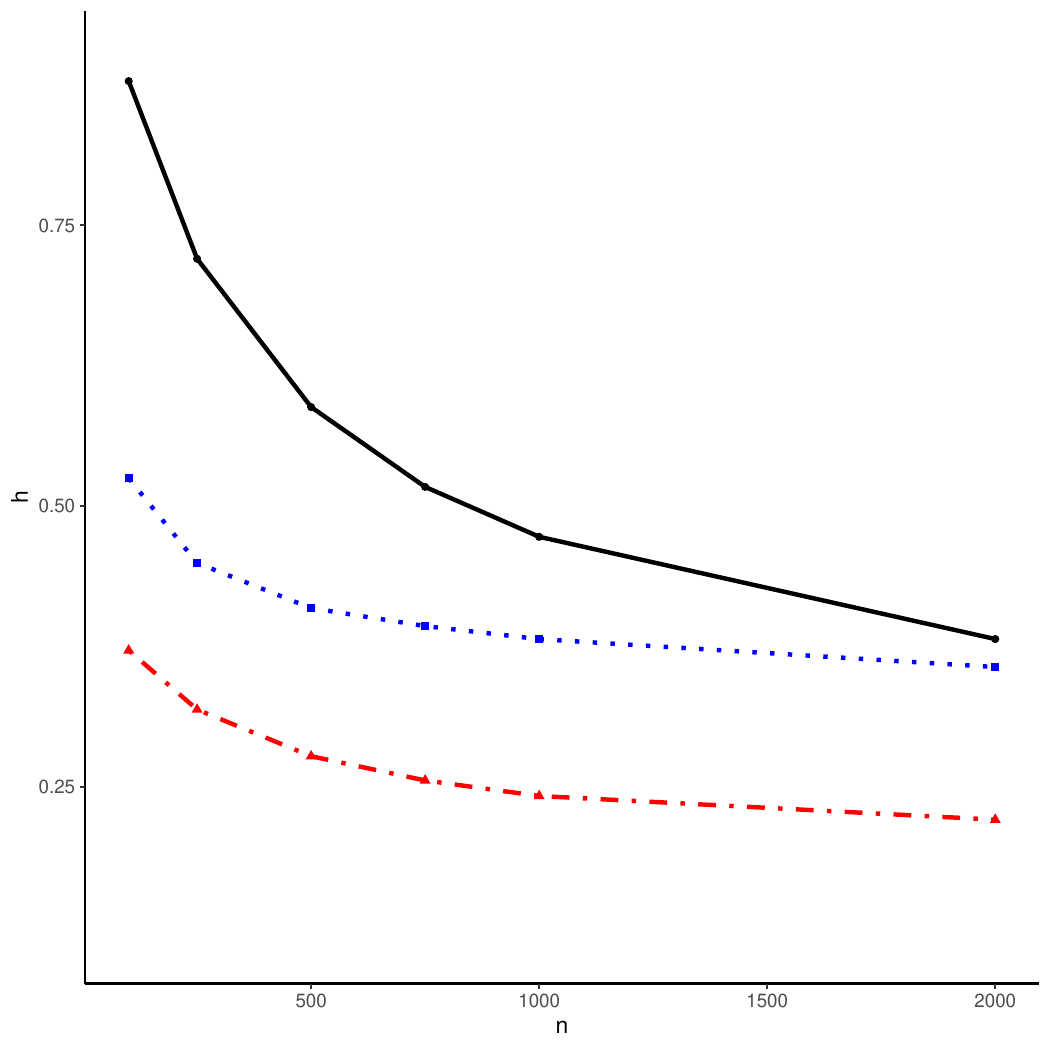}
		\subcaption{$\x=-1$}
	\end{subfigure}%
	\begin{subfigure}[b]{0.5\textwidth}
		\includegraphics[height=0.3\textheight,width=0.95\textwidth]{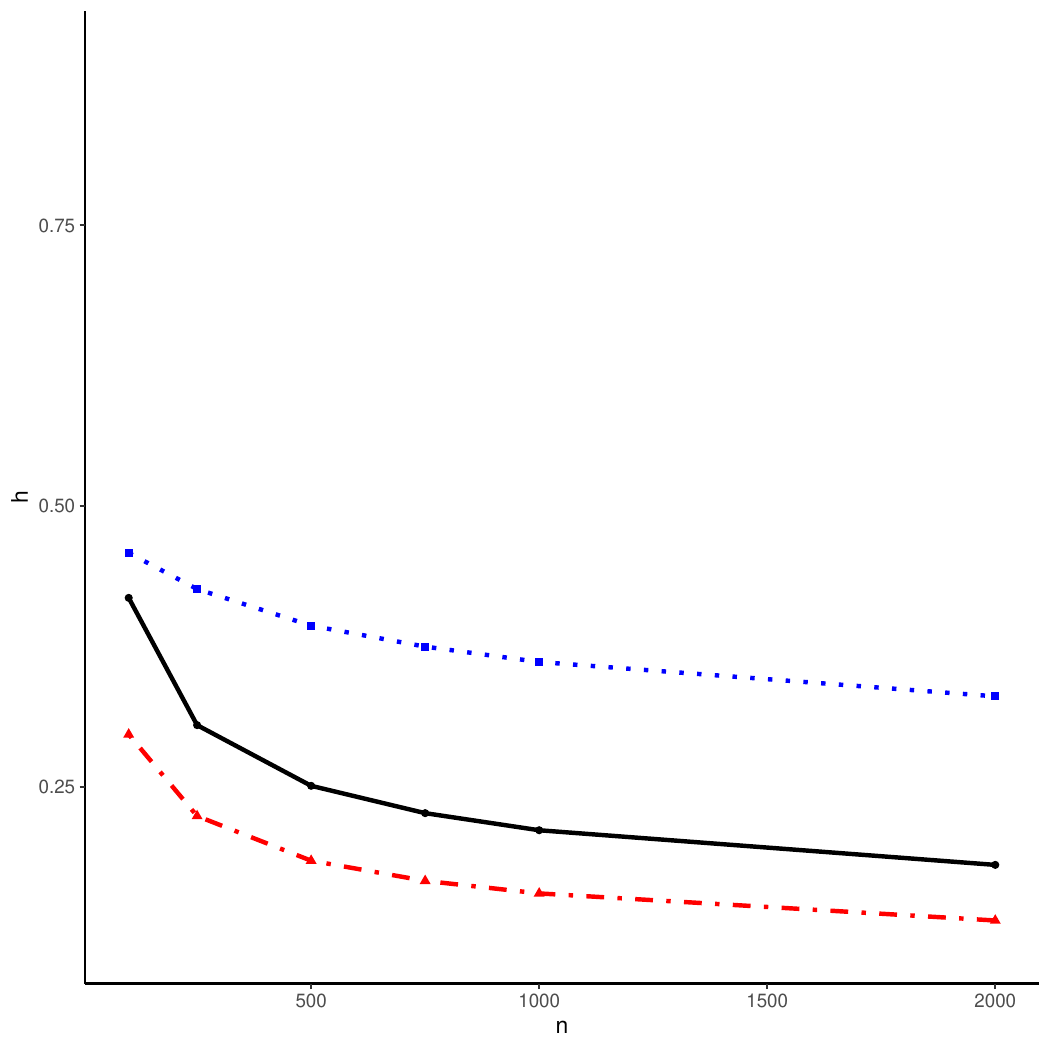}
		\subcaption{$\x=-0.6$}
	\end{subfigure}	
	\begin{subfigure}[b]{0.5\textwidth}
		\includegraphics[height=0.3\textheight,width=0.95\textwidth]{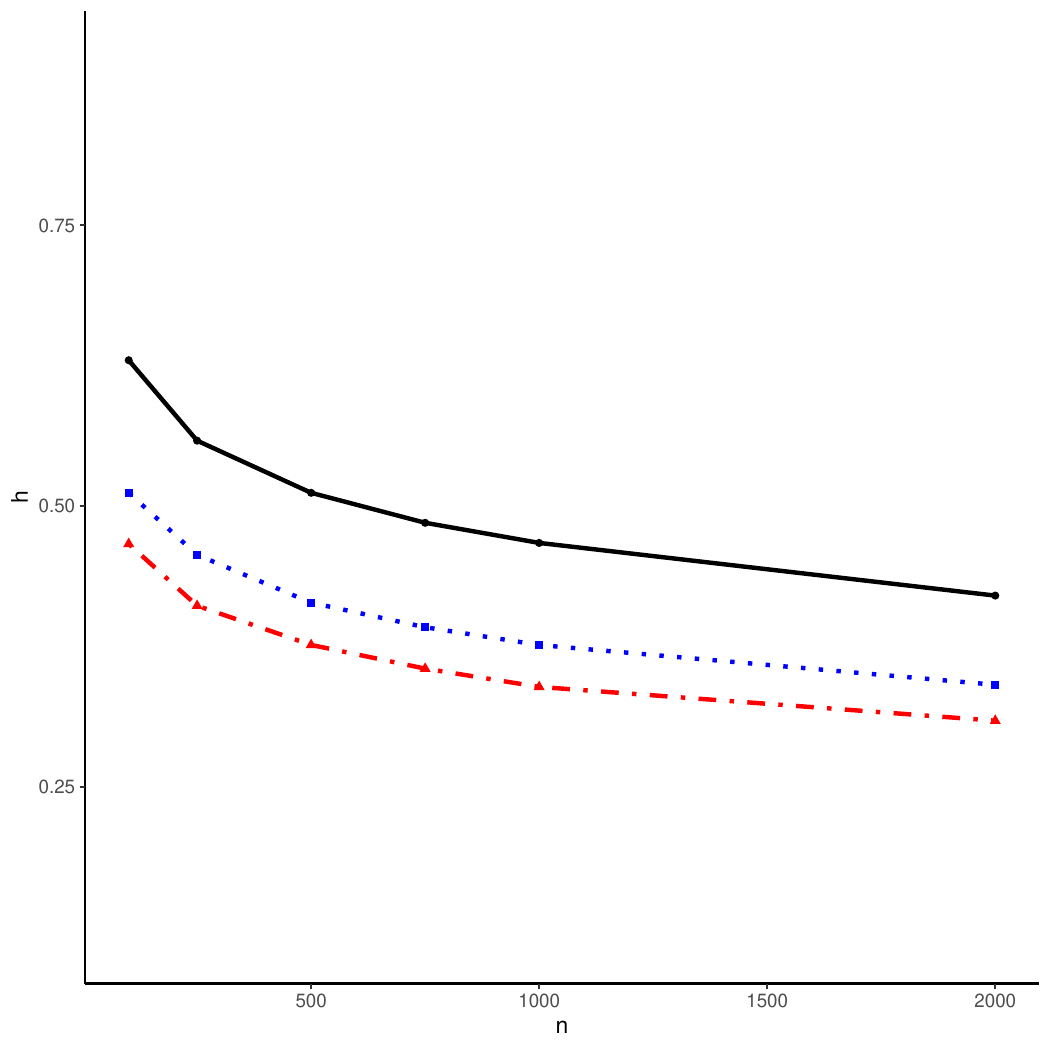}
		\subcaption{$\x=-0.2$}
	\end{subfigure}%	
	\begin{subfigure}[b]{0.5\textwidth}
		\includegraphics[height=0.3\textheight,width=0.95\textwidth]{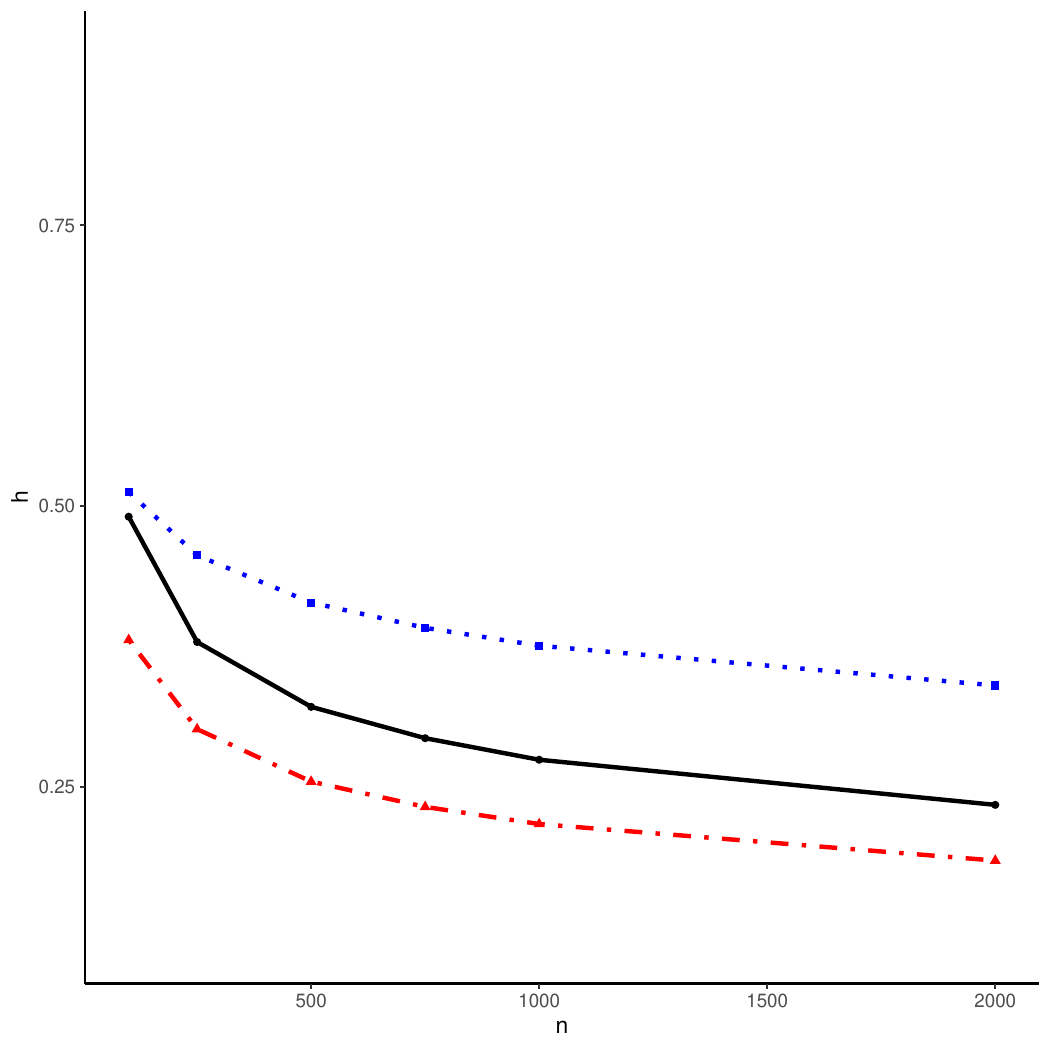}
		\subcaption{$\x=0.2$}
	\end{subfigure}	
	\begin{subfigure}[b]{0.5\textwidth}
		\includegraphics[height=0.3\textheight,width=0.95\textwidth]{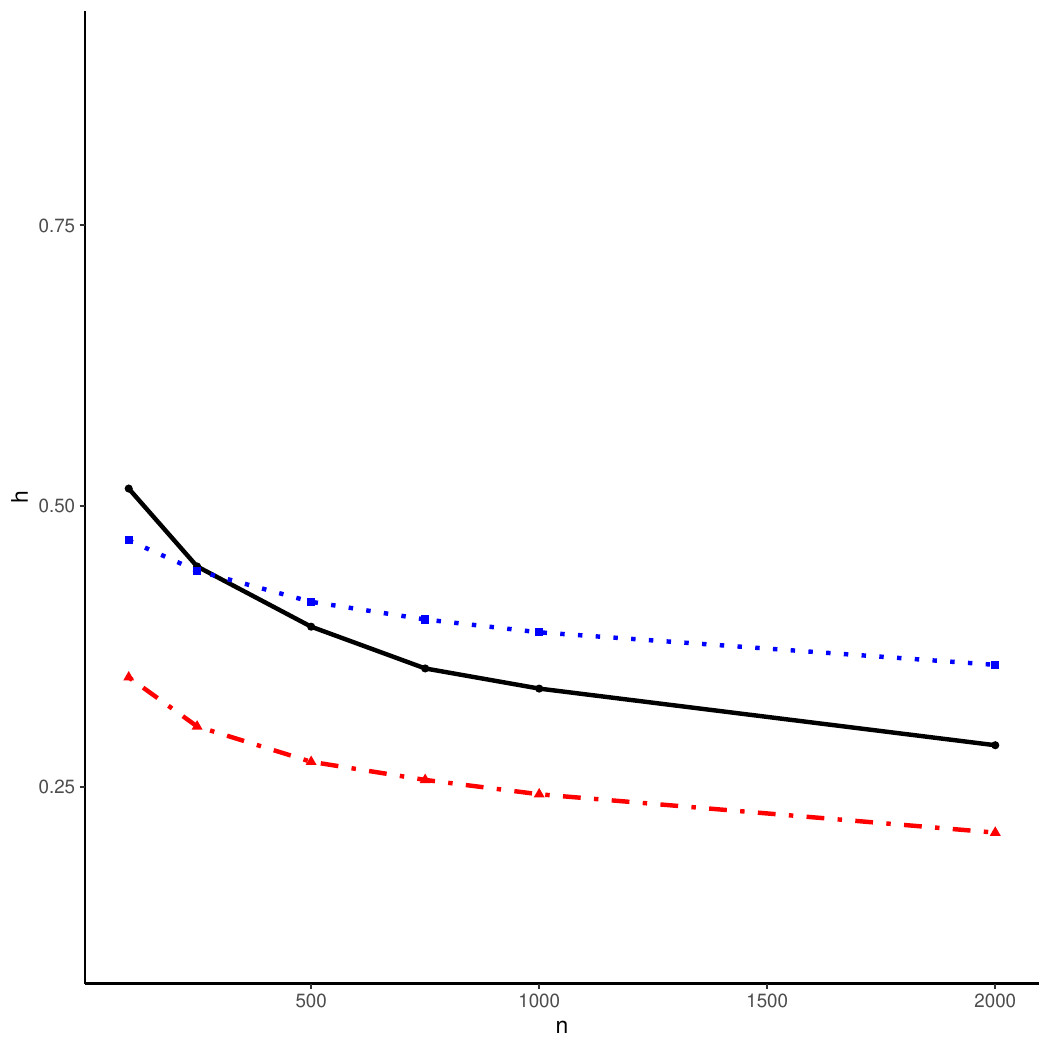}
		\subcaption{$\x=0.6$}
	\end{subfigure}%
	\begin{subfigure}[b]{0.5\textwidth}
		\includegraphics[height=0.3\textheight,width=0.95\textwidth]{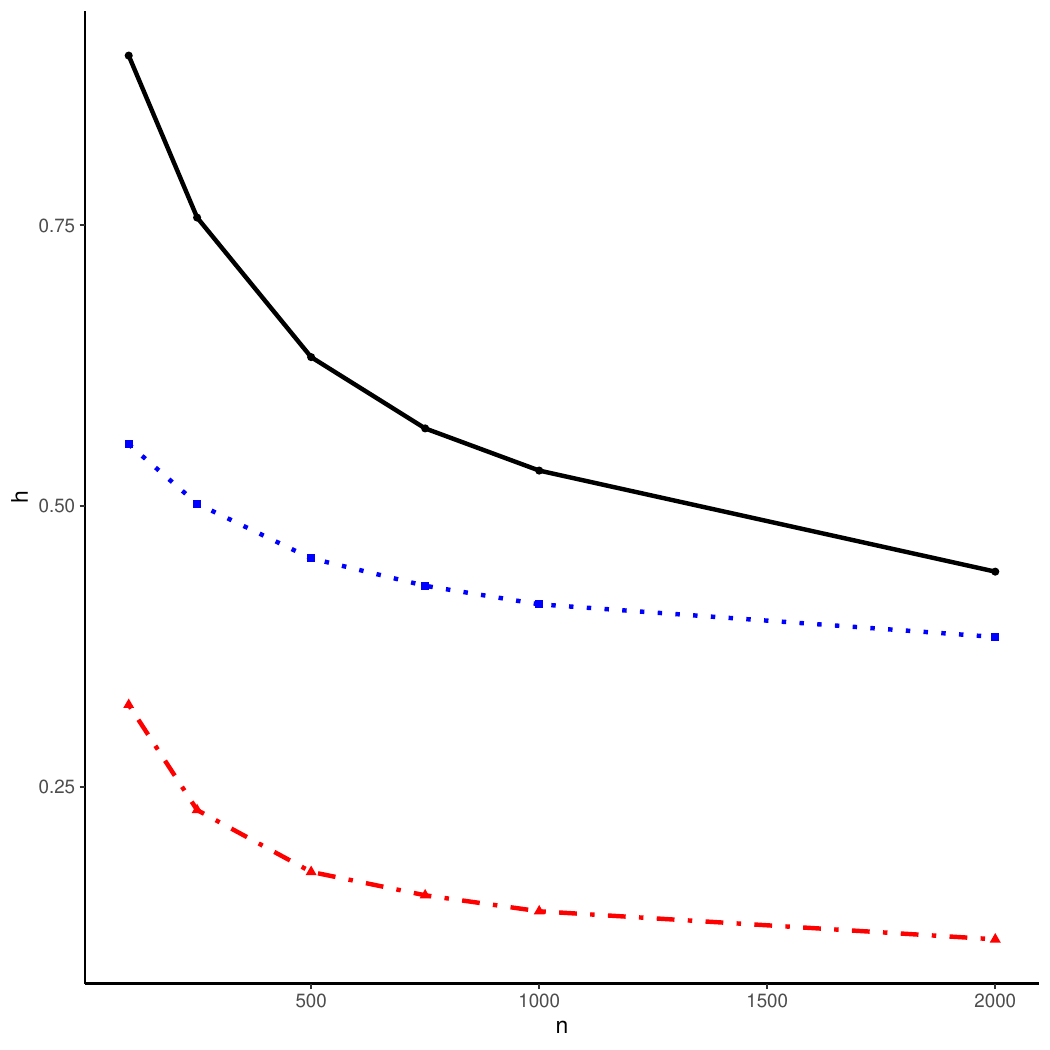}
		\subcaption{$\x=1$}
	\end{subfigure}%
	\begin{flushleft}\footnotesize Notes: \blackline $\hat{h}_{\RBC}$, \redline $\hat{h}_{\US}$, \blueline $\hat{h}_{\MSE}$
	\end{flushleft}
\end{figure}

%%%%%%%%%%%%%%%%%%%%%% TABLES EPA

\clearpage
\begin{table}
	\begin{center}\renewcommand{\arraystretch}{1}
		\caption{Empirical Coverage Probabilities, 95\% Confidence Intervals, $\v=0$, Epanechnikov Kernel}
		\label{supptable:ec_nu0_epa}
		\vspace{-.1in}
		\resizebox{\columnwidth}{!}{\input{simuls/output/table_ec_kepa_p1_d0.txt}}
	\end{center}
\end{table}

\clearpage
\begin{table}
	\begin{center}\renewcommand{\arraystretch}{1}
	\caption{Empirical Coverage Probabilities, 95\% Confidence Intervals, $\v=1$, Epanechnikov Kernel}
	\label{supptable:ec_nu1_epa}
	\vspace{-.1in}
	\resizebox{\columnwidth}{!}{\input{simuls/output/table_ec_kepa_p2_d1.txt}}
	\end{center}
\end{table}

\clearpage
\begin{table}
	\begin{center}\renewcommand{\arraystretch}{1}
		\caption{Average Interval Length, 95\% Confidence Intervals, $\v=0$, Epanechnikov Kernel}
		\label{supptable:il_nu0_epa}
		\vspace{-.1in}
		\resizebox{\columnwidth}{!}{\input{simuls/output/table_il_kepa_p1_d0.txt}}
	\end{center}
\end{table}

\clearpage
\begin{table}
	\begin{center}\renewcommand{\arraystretch}{1}
		\caption{Average Interval Length, 95\% Confidence Intervals, $\v=1$, Epanechnikov Kernel}
		\label{supptable:il_nu1_epa}
		\vspace{-.1in}
		\resizebox{\columnwidth}{!}{\input{simuls/output/table_il_kepa_p2_d1.txt}}
	\end{center}
\end{table}

%%%%%%%%%%%%%%%%%%%%%% TABLES UNI

\clearpage
\begin{table}
	\begin{center}\renewcommand{\arraystretch}{1}
		\caption{Empirical Coverage Probabilities, 95\% Confidence Intervals, $\v=0$, Uniform Kernel}
		\label{supptable:ec_nu0_uni}
		\vspace{-.1in}
		\resizebox{\columnwidth}{!}{\input{simuls/output/table_ec_kuni_p1_d0.txt}}
	\end{center}
\end{table}

\clearpage
\begin{table}
	\begin{center}\renewcommand{\arraystretch}{1}
		\caption{Empirical Coverage Probabilities, 95\% Confidence Intervals, $\v=1$, Uniform Kernel}
		\label{supptable:ec_nu1_uni}
		\vspace{-.1in}
		\resizebox{\columnwidth}{!}{\input{simuls/output/table_ec_kuni_p2_d1.txt}}
	\end{center}
\end{table}

\clearpage
\begin{table}
	\begin{center}\renewcommand{\arraystretch}{1}
		\caption{Average Interval Length, 95\% Confidence Intervals, $\v=0$, Uniform Kernel}
		\label{supptable:il_nu0_uni}
		\vspace{-.1in}
		\resizebox{\columnwidth}{!}{\input{simuls/output/table_il_kuni_p1_d0.txt}}
	\end{center}
\end{table}

\clearpage
\begin{table}
	\begin{center}\renewcommand{\arraystretch}{1}
		\caption{Average Interval Length, 95\% Confidence Intervals, $\v=1$, Uniform Kernel}
		\label{supptable:il_nu1_uni}
		\vspace{-.1in}
		\resizebox{\columnwidth}{!}{\input{simuls/output/table_il_kuni_p2_d1.txt}}
	\end{center}
\end{table}

\clearpage
%%%%%%%%%%%%%%%%%%%%%%%%%%%%%%%%%%%
\subsection{Numerical Computations}

In the main text we discussed the optimization of $\rho$ by minimizing the $L_2$ distance to the known optimal kernel shape in various contexts. These optimal kernel shapes are shown in the figures below for both the Triangular and Epanechnikov kernels, at interior and boundary points, for levels and derivatives. In each case the black line shows $\mathcal{K}^*_{p+1}(u)$ while the dash-dotted red line is $\mathcal{K}_\RBC(u; K, \rho^*, \v)$.

\begin{figure}[ht]
	\centering
	\caption{$\mathcal{K}^*_{p+1}(u)$ vs. $\mathcal{K}_\RBC(u; K, \rho^*, \v)$, $\v=0$}\label{suppfig:k1}
	\subcaption{Triangular Kernel, Boundary Point}
	\begin{subfigure}[b]{0.5\textwidth}
		\includegraphics[height=0.25\textheight,width=0.85\textwidth]{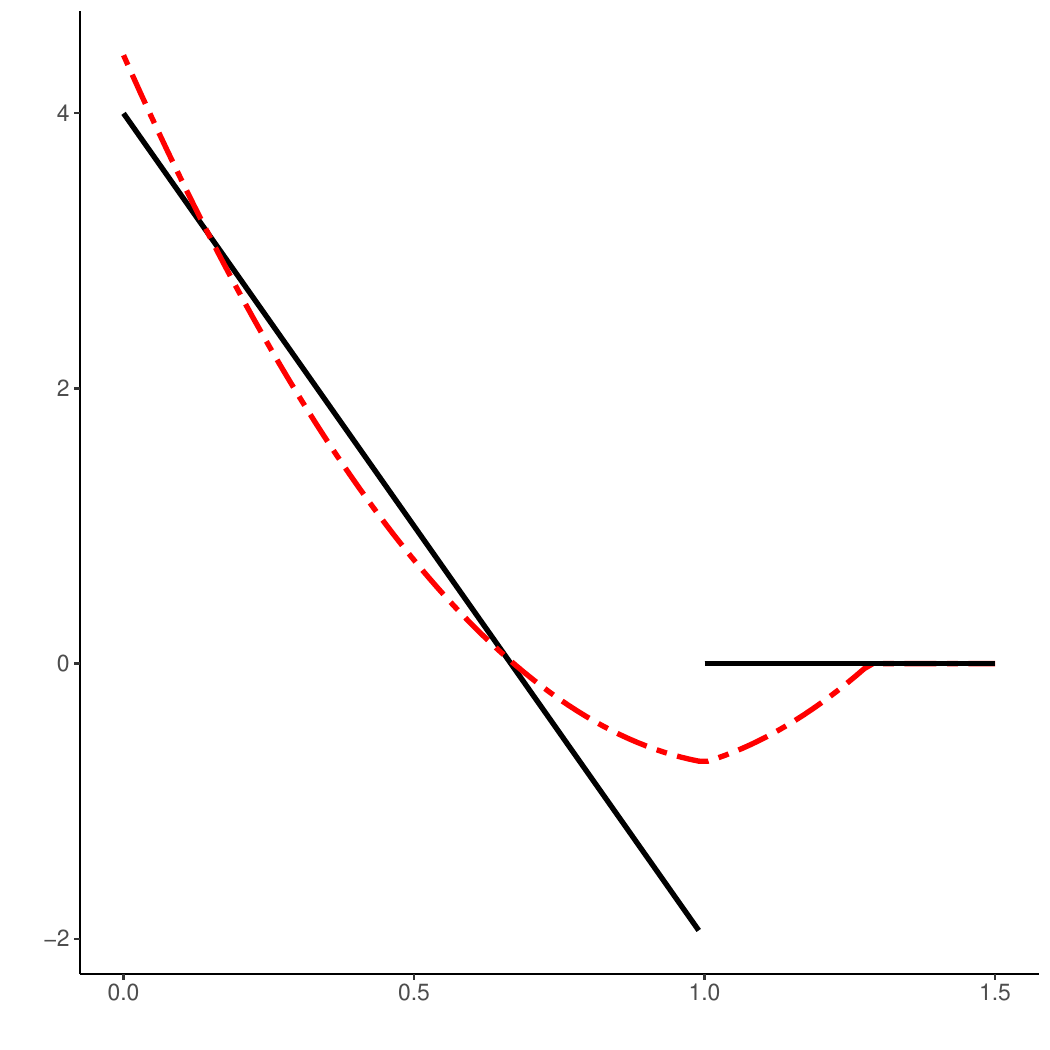}
		\caption{$p=0$}
	\end{subfigure}%
	\begin{subfigure}[b]{0.5\textwidth}
		\includegraphics[height=0.25\textheight,width=0.85\textwidth]{simuls/output/ktria_p1_v0_bnd_L2.pdf}
		\caption{$p=1$}
	\end{subfigure}
	
	\begin{subfigure}[b]{0.5\textwidth}
		\includegraphics[height=0.25\textheight,width=0.85\textwidth]{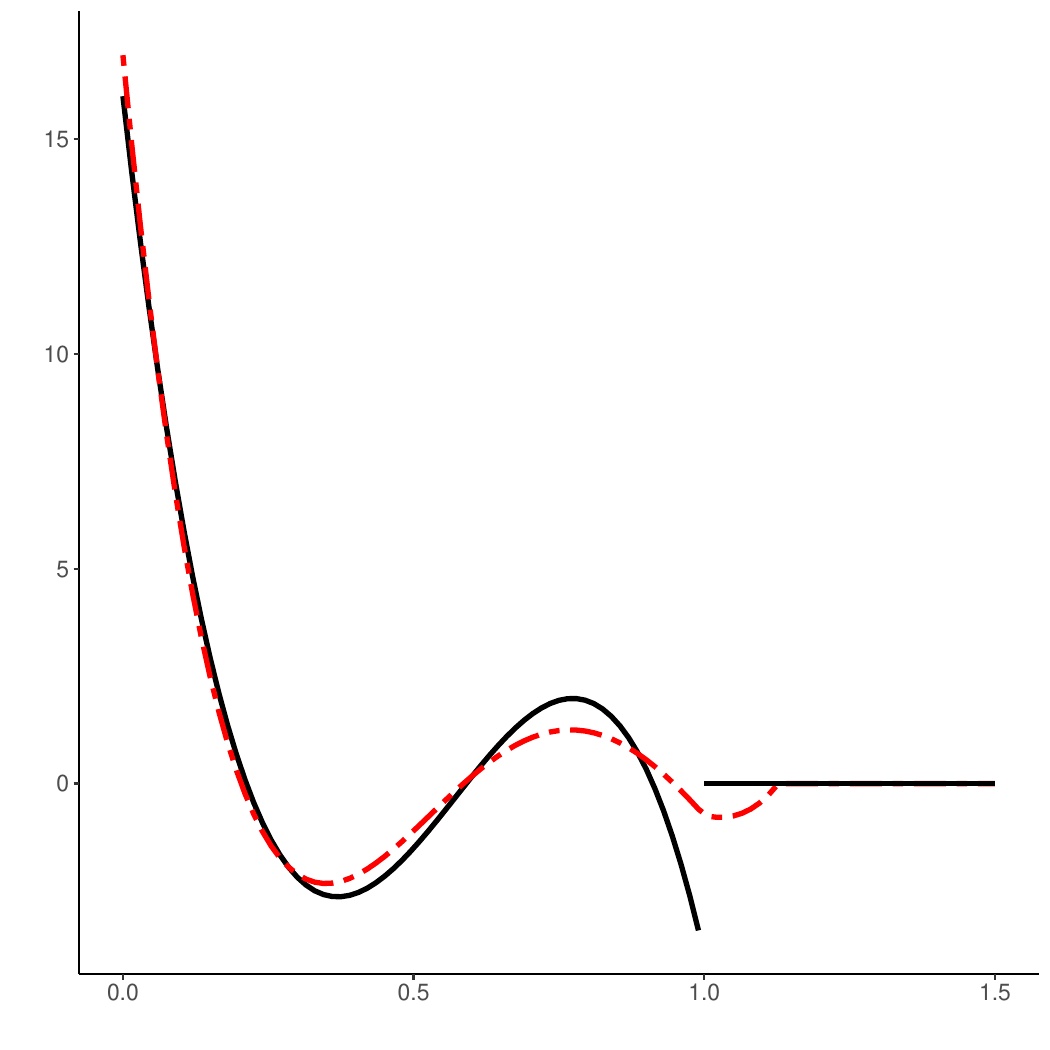}
		\caption{$p=2$}
	\end{subfigure}%
	\begin{subfigure}[b]{0.5\textwidth}
		\includegraphics[height=0.25\textheight,width=0.85\textwidth]{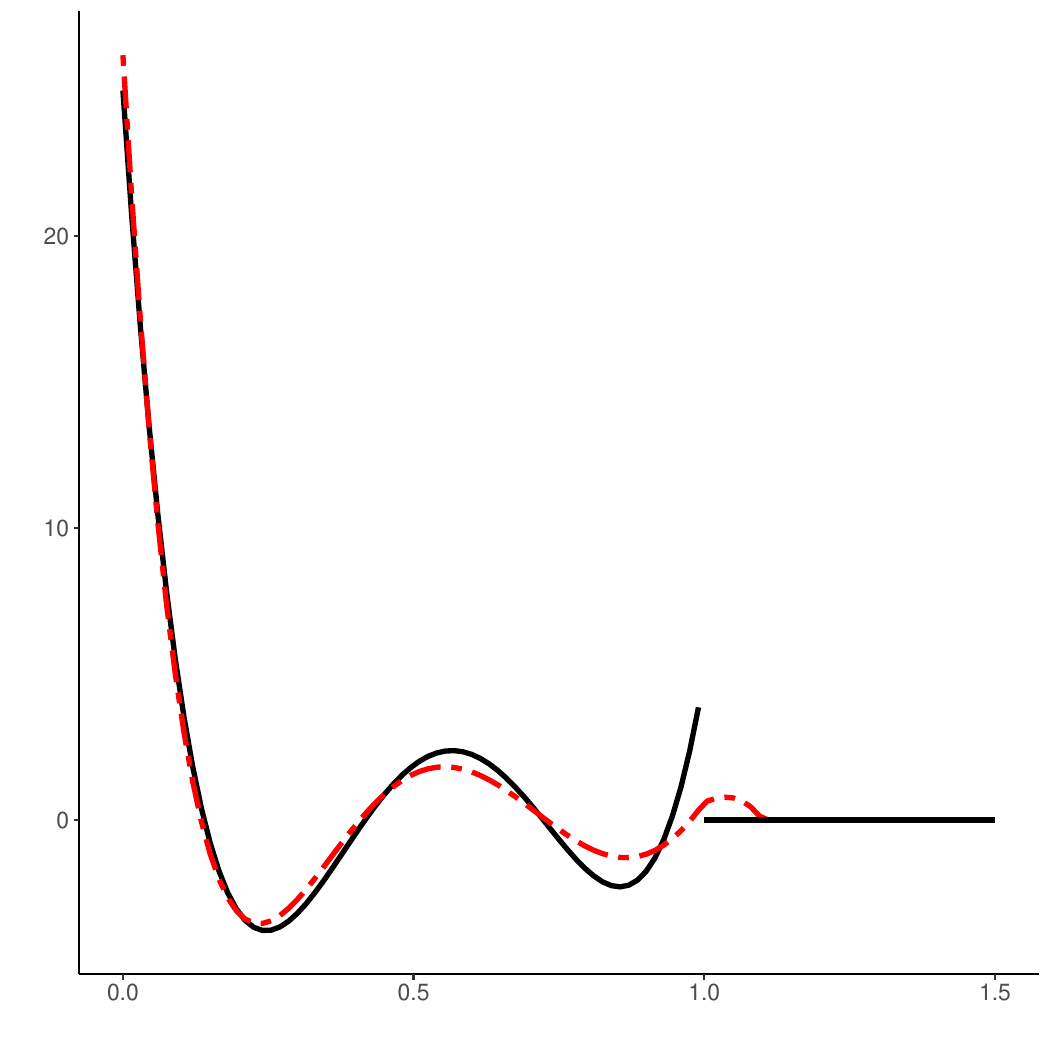}
		\caption{$p=3$}
	\end{subfigure}	
	
	\subcaption{Epanechnikov Kernel, Interior Point}
	
	\begin{subfigure}[b]{0.5\textwidth}
		\includegraphics[height=0.25\textheight,width=0.85\textwidth]{simuls/output/kepan_p1_v0_int_L2.pdf}
		\caption{$p=1$}
	\end{subfigure}%
	\begin{subfigure}[b]{0.5\textwidth}
		\includegraphics[height=0.25\textheight,width=0.85\textwidth]{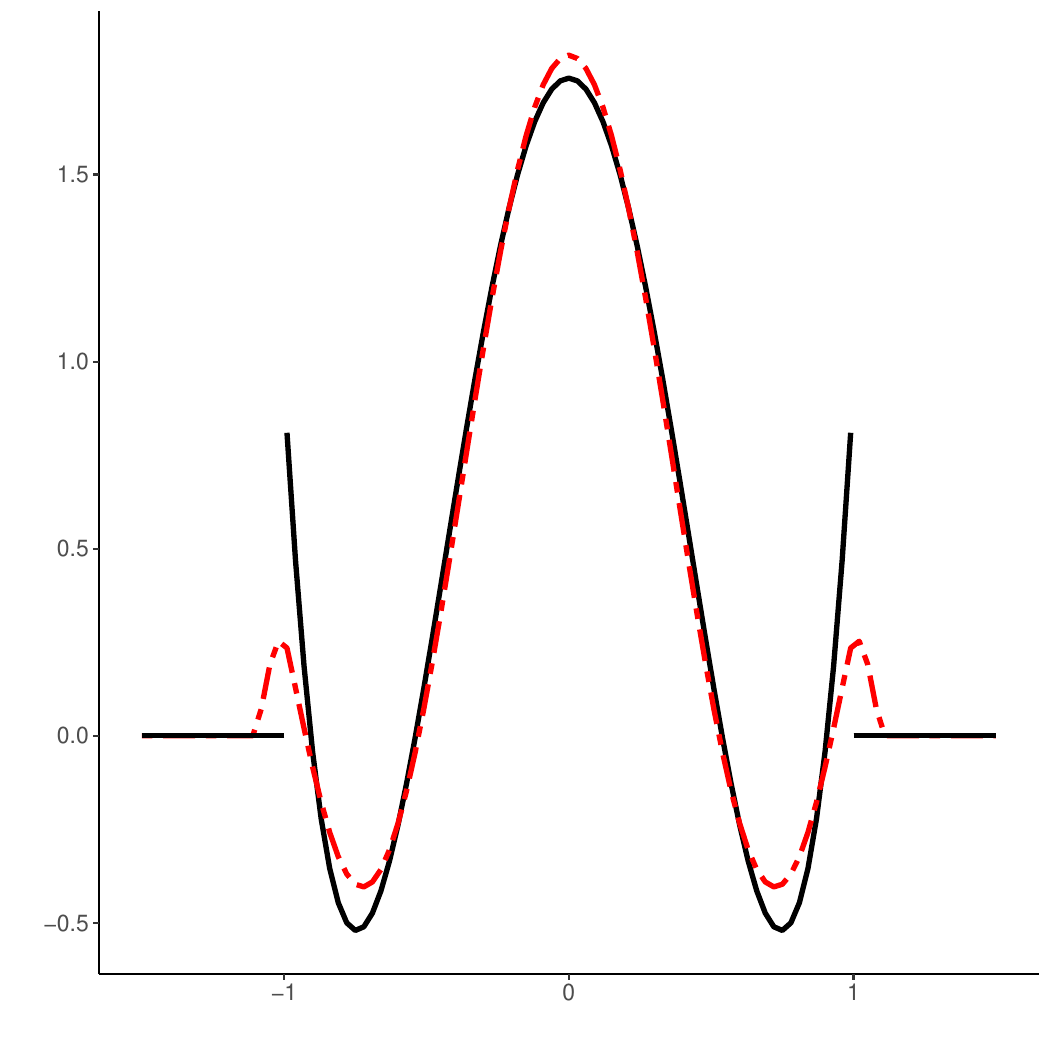}
		\caption{$p=3$}
	\end{subfigure}
	
	%	\begin{flushleft}\footnotesize Notes: \blackline $\mathcal{K}^*_{p+1}(u)$, \redline $\mathcal{K}_\RBC(u; K, \rho^*, \v)$\end{flushleft}
\end{figure}

\clearpage
\begin{figure}[ht]
	\centering
	\caption{$\mathcal{K}^*_{p+1}(u)$ vs. $\mathcal{K}_\RBC(u; K, \rho^*, \v)$, $\v=1$}\label{suppfig:k2}
	\subcaption{Triangular Kernel, Boundary Point}
	\begin{subfigure}[b]{0.5\textwidth}
		\includegraphics[height=0.25\textheight,width=0.85\textwidth]{simuls/output/ktria_p1_v1_bnd_L2.pdf}
		\caption{$p=1$}
	\end{subfigure}%
	\begin{subfigure}[b]{0.5\textwidth}
		\includegraphics[height=0.25\textheight,width=0.85\textwidth]{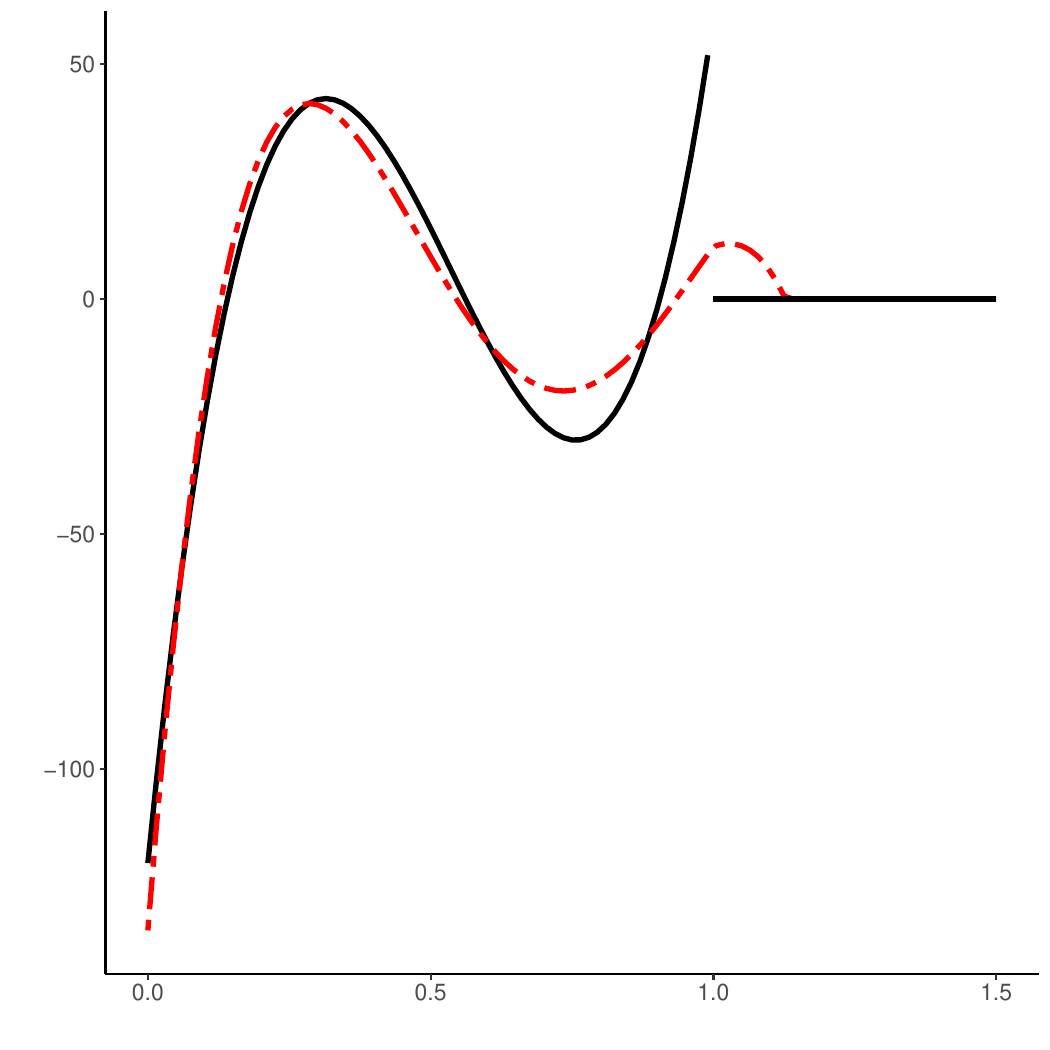}
		\caption{$p=2$}
	\end{subfigure}
	
	\begin{subfigure}[b]{0.5\textwidth}
		\includegraphics[height=0.25\textheight,width=0.85\textwidth]{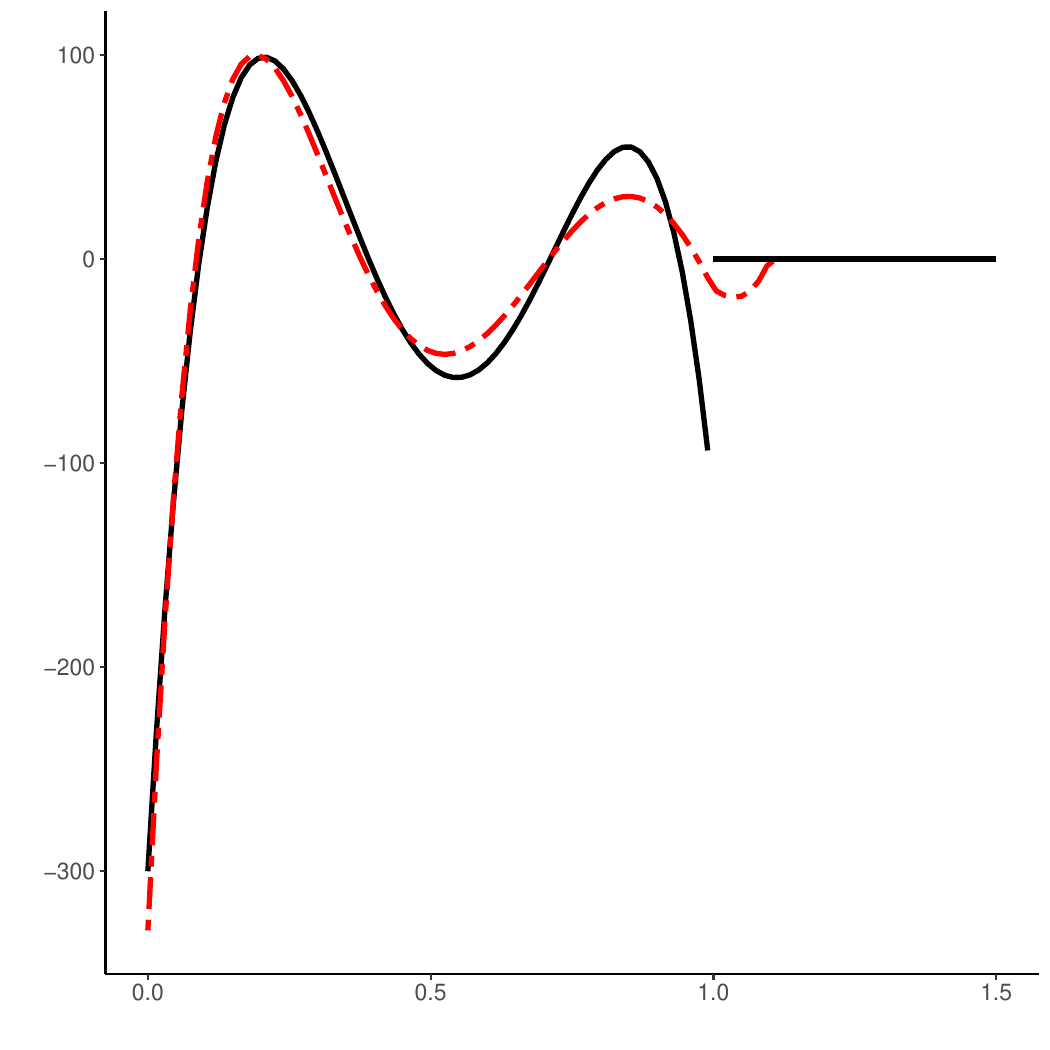}
		\caption{$p=3$}
	\end{subfigure}%
	\begin{subfigure}[b]{0.5\textwidth}
		\includegraphics[height=0.25\textheight,width=0.85\textwidth]{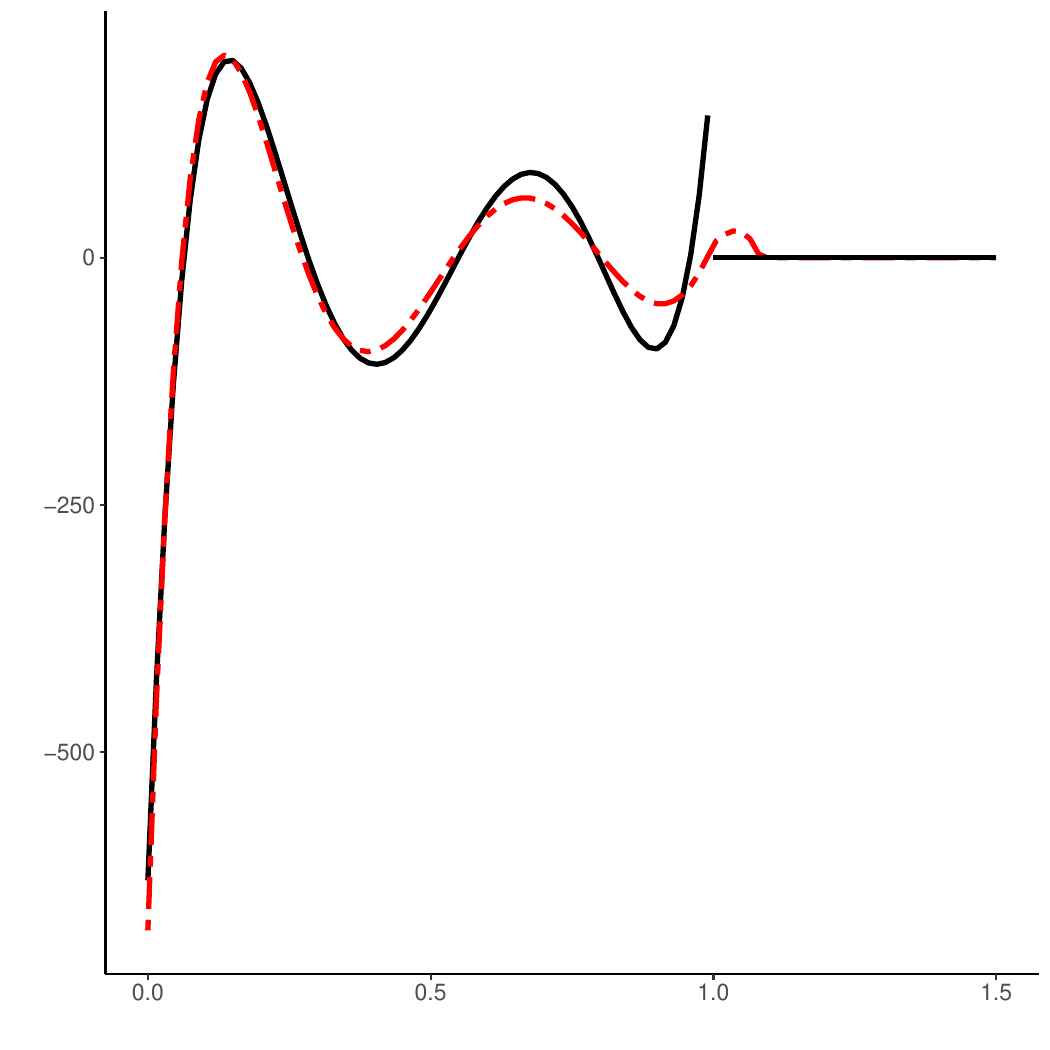}
		\caption{$p=4$}
	\end{subfigure}
	
	\subcaption{Epanechnikov Kernel, Interior Point}
	
	\begin{subfigure}[b]{0.5\textwidth}
		\includegraphics[height=0.25\textheight,width=0.85\textwidth]{simuls/output/kepan_p2_v1_int_L2.pdf}
		\caption{$p=2$}
	\end{subfigure}%
	\begin{subfigure}[b]{0.5\textwidth}
		\includegraphics[height=0.25\textheight,width=0.85\textwidth]{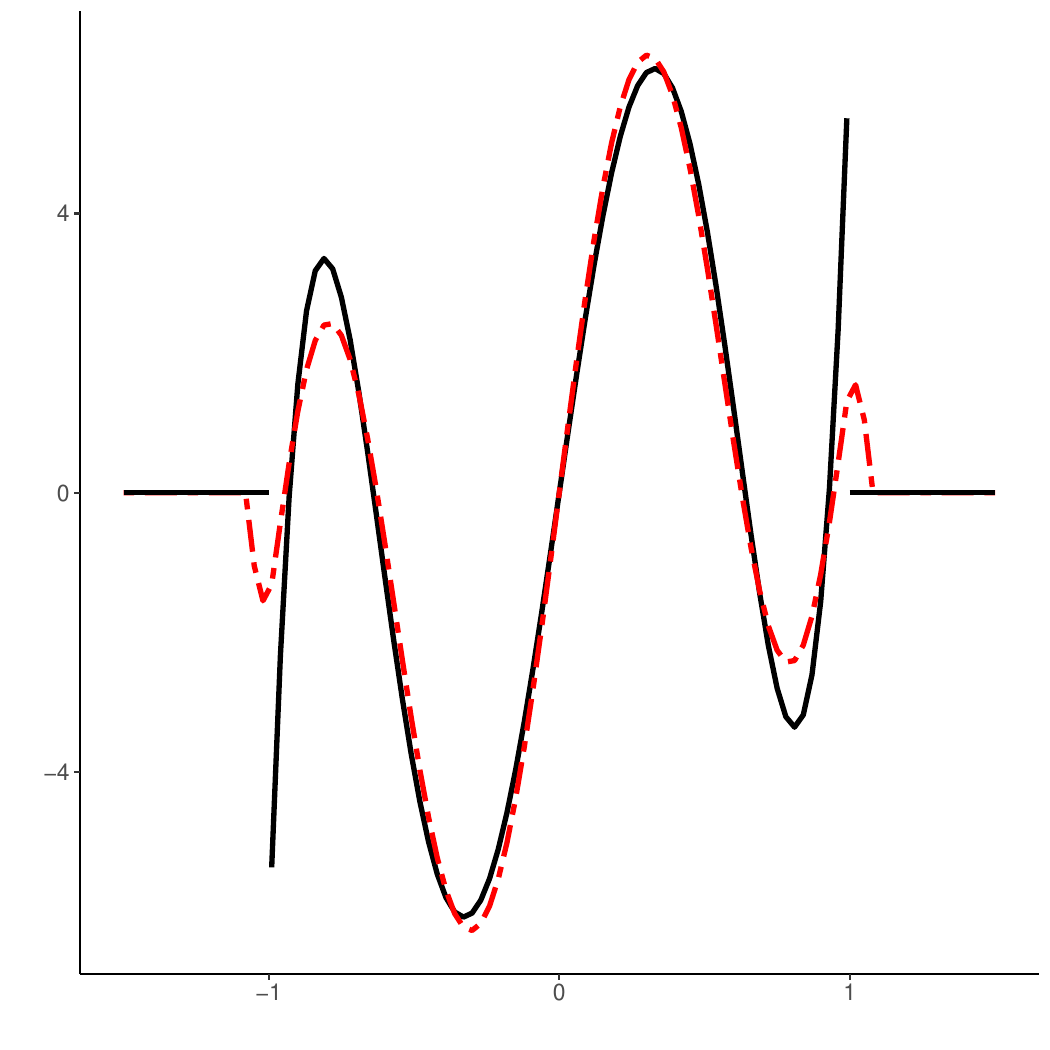}
		\caption{$p=4$}
	\end{subfigure}

	%	\begin{flushleft}\footnotesize Notes: \blackline $\mathcal{K}^*_{p+1}(u)$, \redline $\mathcal{K}_\RBC(u; K, \rho^*, \v)$\end{flushleft}
\end{figure}

\clearpage

%%%%%%%%%%%%%%%%%%%%%%%%%%%%%%%%%%%%%%%%%%%%%%%%%%%%%%%%%%%%%%%%%%%%%%
%%%%%%%%%%%%%%%%%%%%%%%%%%%%%%%%%%%%%%%%%%%%%%%%%%%%%%%%%%%%%%%%%%%%%%
\section{List of Notation}
	\label{supp:notation lp}

Below is a (hopefully) complete list of the notation used in this Part, group by Section, roughly in order of introduction. This is intended only as a reference. Each object is redefined below when it is needed.

Asymptotic orders and their in-probability versions hold uniformly in $\F_\S$, as required by our framework; e.g., $A_n = o_\P(a_n)$ means $\sup_{\f \in \F_\S} \P_\f [\vert A_n/a_n \vert > \epsilon] = o(1)$ for every $\epsilon > 0$.

%%%%%%%%%%%%%%%%%%%%%%%%%%%%%%%%%%%%%%%%%%%%%%%%%%%%%%%%%%%%%%%%%%%%%%
\subsection*{Local Polynomial Regression, $t$-Statistics, and Confidence Intervals}

\begin{itemize}
	
	\item $\{(Y_1, X_1), \ldots, (Y_n, X_n)\}$ is a random sample distributed according to $\f$, the data-generating process. $\f$ is assumed to belong to a class $\F_\S$
	
	\item $\tf = \mu^{(\v)}_\f(\x) := \left.  \frac{\partial^\v}{\partial x^\v} \E_\f \left[ Y \mid X \!=\! x \right] \right\vert_{x=\x}$, where $\v \leq \S$, where $\mu(\cdot)$ possess at least $S$ derivatives.
	
	\item $\mu_\f(\x) = \mu^{(0)}_\f(\x) = \E_\f [ Y \mid X \!=\! \x ] $
	
	\item Where it causes no confusion the point of evaluation $\x$ will be omitted as an argument, so that for a function $g(\cdot)$ we will write $g\defsym g(\x)$
	
	\item $\mhat^{(\v)} = \v! \be_\v' \bhat_p = \frac{1}{n \h^\v} \v! \be_\v'\Gp^{-1} \Op  \bY$
	
	\item $\bhat_p = \argmin_{\bbeta \in \mathbb{R}^{p+1}} \sumi ( Y_i - \br_p(X_i - \x)'\bbeta)^2  K \left( \Xhi \right)$
	
	\item $\bhat_{p+1} = \argmin_{\bbeta \in \mathbb{R}^{(p+1)+1}} \sumi ( Y_i - \br_{p+1}(X_i - \x)'\bbeta)^2  K \left( \Xbi \right)$
		
	\item $\be_k$ is a conformable zero vector with a one in the $(k+1)$ position, for example $\be_\v$ is the $(p+1)$-vector with a one in the $\v^{\text{th}}$ position and zeros in the rest
	
	\item $\h$ is a bandwidth sequence that vanishes as $n$ diverges
	
	\item $p$ is an integer greater than $\v$, with $p - \v$ odd
	
	\item $\br_p(u) = (1, u, u^2, \ldots, u^p)'$
	
	\item $\Xhi = (X_i - \x)/\h$, for a bandwidth $\h$ and point of interest $\x$
	
	\item to save space, products of functions will often be written together, with only one argument, for example
	\[ (K \br_p \br_p')(\Xhi) \defsym K(\Xhi) r_p(\Xhi) r_p(\Xhi)' = K\left(\frac{X_i - \x}{\h}\right) \br_p \left(\frac{X_i - \x}{\h}\right) \br_p \left(\frac{X_i - \x}{\h}\right)' ,  \]
	
	\item $\Gp = \frac{1}{n\h} \sumi (K \br_p \br_p')(\Xhi) = (\bRc' \bW \bRc)/n$
	
	\item $\Op = [ (K \br_p)(X_{\h,1}),  (K \br_p)(X_{\h,2}), \ldots,  (K \br_p)(X_{\h,n})] = \bRc' \bW$
	
	\item $\bY = (Y_1, \ldots, Y_n)'$
	
	\item $\bR = \left[ \br_p( X_1 - \x), \cdots, \br_p( X_n - \x ) \right]'$
	
	\item $\bW = \diag\left(\h^{-1} K(\Xhi): i = 1, \ldots, n\right)$
	
	\item $\H = \diag\left(1, \h, \h^2, \ldots, \h^p \right)$
	
	\item $\bRc = \bR \H^{-1} = \left[ \br_p( X_{\h,1}), \cdots, \br_p( X_{\h,n} ) \right]'$
	
	\item $\diag(a_i:i = 1, \ldots, k)$ denote the $k \times k$ diagonal matrix constructed using the elements $a_1, a_2, \cdots, a_k$
	
	\item $\Lp_k = \Op \left[ X_{\h,1}^{p+k}, \ldots, X_{\h,n}^{p+k} \right]'/n$, where, in particular $\Lp_1$ was denoted $\Lp$ in the main text
	
	\item $\b$ is a bandwidth sequence that vanishes as $n$ diverges
	
	\item $\Xbi = (X_i - \x)/\b$, for a bandwidth $\b$ and point of interest $\x$, exactly like $\Xhi$ but with $\b$ in place of $\h$
	
	\item $\Oq = [ (K \br_{p+1})(X_{\b,1}),  (K \br_{p+1})(X_{\b,2}), \ldots,  (K \br_{p+1})(X_{\b,n})]$, exactly like $\Op$ but with $\b$ in place of $\h$ and $p+1$ in place of $p$
	
	\item $\Gq = \frac{1}{n\b} \sumi (K \br_{p+1} \br_{p+1}')(\Xbi)$, exactly like $\Gp$ but with $\b$ in place of $\h$ and $p+1$ in place of $p$, and 
	
	\item $\Lq_k = \Oq \left[ X_{\b,1}^{p+1+k}, \ldots, X_{\b,n}^{p+1+k} \right]'/n$, exactly like $\Lp_k$  but with $\b$ in place of $\h$ and $p+1$ in place of $p$ (implying $\Oq$ in place of $\Op$)
	
	\item $\mhat^{(\v)} = \frac{1}{n \h^\v} \v! \be_\v'\Gp^{-1} \Op  \bY$ 
	
	$\thatrbc = \mhat^{(\v)}  - \h^{p+1 - \v}  \v! \be_\v'\Gp^{-1} \Lp_1 \frac{\mhat^{(p+1)}}{(p+1)!} = \frac{1}{n \h^\v} \v! \be_\v'\Gp^{-1} \Orbc  \bY$
	
	\item $\Orbc = \Op - \rho^{p+1}  \Lp_1 \be_{p+1}' \Gq^{-1} \Oq $ 
	
	\item $\rho = \h / \b$, the ratio of the two bandwidth sequences
	
	\item $\Sigma = \diag(v(X_i): i = 1,\ldots, n)$, with $v(x) = \V[Y \vert X = x]$
	
	\item $\sp^2  =  \v!^2 \be_\v'\Gp^{-1} (\h \Op \bS \Op' /n) \Gp^{-1} \be_\v$
	
	$\srbc^2  =  \v!^2 \be_\v'\Gp^{-1} (\h \O_\RBC \bS \O_\RBC' /n) \Gp^{-1} \be_\v$
	
	\item $\shatp^2 =  \v!^2 \be_\v'\Gp^{-1} (\h \Op \Shatp \Op' /n) \Gp^{-1} \be_\v$
	
	$\shatrbc^2 =  \v!^2 \be_\v'\Gp^{-1} (\h \O_\RBC \Shatrbc \O_\RBC' /n) \Gp^{-1} \be_\v$
	
	\item $\Shatp = \diag(\hat{v}(X_i): i = 1,\ldots, n)$, with $\hat{v}(X_i) = ( Y_i - \br_p(X_i - \x)'\bhat_p )^2$ for $\bhat_p$ defined in  Equation \eqref{suppeqn:lp}, and 
	
	\item $\Shatrbc = \diag(\hat{v}(X_i): i = 1,\ldots, n)$, with $\hat{v}(X_i) = ( Y_i - \br_{p+1}(X_i - \x)'\bhat_{p+1} )^2$ for $\bhat_{p+1}$ defined exactly as in Equation \eqref{suppeqn:lp} but with $p+1$ in place of $p$ and $\b$ in place of $\h$.
	
	\item $\displaystyle \tp = \frac{\sqrt{n\h^{1 + 2\v}}( \mhat_p^{(\v)} - \tf)}{\shatp}$ 
	
	$\displaystyle \trbc = \frac{( \thatrbc - \tf)}{\vhatrbc} = \frac{\sqrt{n\h^{1 + 2\v}}( \thatrbc - \tf)}{\shatrbc}$

	\item $\ip  = \left[ \mhat_p^{(\v)} - z_u \shatp \big/ \sqrt{n\h^{1 + 2\v}} \; , \;  \mhat_p^{(\v)} - z_l \shatp \big/ \sqrt{n\h^{1 + 2\v}} \right]$
	
	$\irbc = \left[ \thatrbc - z_u \vhatrbc \; , \;  \thatrbc - z_l \vhatrbc \right] = \left[ \thatrbc - z_u \shat_\RBC \big/ \sqrt{n\h^{1 + 2\v}} \; , \;  \thatrbc - z_l \shat_\RBC \big/ \sqrt{n\h^{1 + 2\v}} \right]$
	
\end{itemize}

%%%%%%%%%%%%%%%%%%%%%%%%%%%%%%%%%%%%%%%%%%%%%%%%%%%%%%%%%%%%%%%%%%%%%%
\subsection*{Main Results and Proofs}

\begin{itemize}
	
	\item See Section \ref{supp:terms lp} for definitions of all terms in the Edgeworth expansion.
	
	\item $\Phi(z)$ is the Normal distribution function.
	
	\item $C$ shall be a generic conformable constant that may take different values in different places. Note that $C$ may be a vector or matrix but will generally not be denoted by a bold symbol. If more than one constant is needed, $C_1$, $C_2$, \ldots, will be used.
	
	\item {\bf Norms.} Unless explicitly noted otherwise, $|\cdot|$ will be the Euclidean/Frobenius norm: for a scalar $c\in \R^1$, $|c|$ is the absolute value; for a vector $\bm{c}$, $|\bm{c}| = \sqrt{\bm{c}'\bm{c}}$; for a matrix $\bm{C}$, $|\bm{C}| = \sqrt{\trace(\bm{C}'\bm{C})}$.
	
	\item $\tO = \sqrt{n\h}$. 
	
	\item $r_{\ti,\f} = \max\{\tO^{-2}, \bti^2, \tO^{-1} \bti \}$, i.e.\ the slowest vanishing of the rates, and 
	
	\item $r_n$ as a generic sequence that obeys $r_n = o(r_{\ti,\f})$.
\end{itemize}

%%%%%%%%%%%%%%%%%%%%%%%%%%%%%%%%%%%%%%%%%%%%%%%%%%%%%%%%%%%%%%%%%%%%%%
\subsection*{Bias and the Role of Smoothness}

\begin{itemize}
	
	\item $\bbeta_k$ (usually $k=p$ or $k=p+1$) as the $k+1$ vector with $(j+1)$ element equal to $\mu^{(j)}(\x)/j!$ for $j = 0, 1, \ldots, k$ as long as $j \leq S$, and zero otherwise
	
	\item $\preTaylor_k$ as the $n$-vector with $i^{\text{th}}$ entry $[\mu(X_i) - \br_{k}(X_i - \x)'\bbeta_{k}]$
	
	\item $\bM = [\mu(X_1), \ldots, \mu(X_n)]'$
	
	\item $\rho = \h / \b$, the ratio of the two bandwidth sequences
	
	\item $\Gpt = \E[\Gp]$, $\Gqt = \E[\Gq]$, $\Lpt_k = \E[\Lp_k]$, $\Lqt_k = \E[\Lq_k]$, and so forth. A tilde always denotes a fixed-$n$ expectation, and all expectations are fixed-$n$ calculations unless explicitly denoted otherwise. The dependence on $\F_\S$ is suppressed notationally.
	
	\item $\bti = \bi$, the fixed-$n$ bias for interval $\i$ or $t$-statistic $\ti$. They are identical for all $\i$ and $\f$, e.g., $\btrbc = \birbc = \biasletter_{\trbc,\f}$. See Equation \eqref{suppeqn:eta lp}
	
	\item $\bct = \bci$, the constant portion of the fixed-$n$ bias for interval $\i$ or $t$-statistic $\ti$. They are identical for all $\i$ and $\f$, e.g., $\bctrbc = \bcirbc = \biasConstant_{\trbc,\f}$. See Tables \ref{supptable:us bias list} and \ref{supptable:rbc bias list}

\end{itemize}

\end{document}

%% file: simuls/output/table_ec_kepa_p1_d0.txt
%latex.default(table1, file = paste("output/table_ec_k", kernel,     "_p", p, "_d", deriv, ".txt", sep = ""), landscape = FALSE,     outer.size = "scriptsize", col.just = rep("c", ncol(table1)),     center = "none", title = "", table.env = FALSE, n.cgroup = c(3,         3, 3), cgroup = c("$h_{\\texttt{RBC}}$", "$h_{\\texttt{US}}$",         "$h_{\\texttt{MSE}}$"), n.rgroup = c(6, 6, 6, 6, 6, 6),     rgroup = c(paste("$x=$", eval, sep = "")))%
\begin{tabular}{lccccccccccc}
\hline\hline
\multicolumn{1}{l}{\bfseries }&\multicolumn{3}{c}{\bfseries $h_{\texttt{RBC}}$}&\multicolumn{1}{c}{\bfseries }&\multicolumn{3}{c}{\bfseries $h_{\texttt{US}}$}&\multicolumn{1}{c}{\bfseries }&\multicolumn{3}{c}{\bfseries $h_{\texttt{MSE}}$}\tabularnewline
\cline{2-4} \cline{6-8} \cline{10-12}
\multicolumn{1}{l}{}&\multicolumn{1}{c}{$h$}&\multicolumn{1}{c}{RBC}&\multicolumn{1}{c}{US}&\multicolumn{1}{c}{}&\multicolumn{1}{c}{$h$}&\multicolumn{1}{c}{RBC}&\multicolumn{1}{c}{US}&\multicolumn{1}{c}{}&\multicolumn{1}{c}{$h$}&\multicolumn{1}{c}{RBC}&\multicolumn{1}{c}{US}\tabularnewline
\hline
{\bfseries $x=$-1}&&&&&&&&&&&\tabularnewline
~~100&0.436&0.881&0.877&&0.320&0.873&0.889&&0.507&0.899&0.875\tabularnewline
~~250&0.368&0.906&0.892&&0.115&0.879&0.893&&0.462&0.912&0.862\tabularnewline
~~500&0.321&0.925&0.902&&0.116&0.879&0.893&&0.438&0.930&0.828\tabularnewline
~~750&0.295&0.935&0.915&&0.168&0.881&0.880&&0.420&0.934&0.797\tabularnewline
~~1000&0.280&0.941&0.908&&0.205&0.887&0.860&&0.404&0.930&0.769\tabularnewline
~~2000&0.255&0.941&0.902&&0.143&0.920&0.902&&0.356&0.924&0.696\tabularnewline
\hline
{\bfseries $x=$-0.6}&&&&&&&&&&&\tabularnewline
~~100&0.335&0.922&0.898&&0.255&0.919&0.909&&0.356&0.929&0.897\tabularnewline
~~250&0.262&0.935&0.922&&0.145&0.927&0.931&&0.342&0.940&0.874\tabularnewline
~~500&0.221&0.941&0.927&&0.104&0.942&0.941&&0.316&0.944&0.869\tabularnewline
~~750&0.200&0.948&0.941&&0.090&0.938&0.942&&0.291&0.947&0.867\tabularnewline
~~1000&0.186&0.949&0.942&&0.081&0.946&0.950&&0.274&0.950&0.870\tabularnewline
~~2000&0.158&0.947&0.936&&0.063&0.946&0.941&&0.235&0.944&0.868\tabularnewline
\hline
{\bfseries $x=$-0.2}&&&&&&&&&&&\tabularnewline
~~100&0.564&0.800&0.388&&0.242&0.910&0.873&&0.512&0.858&0.446\tabularnewline
~~250&0.490&0.794&0.286&&0.169&0.924&0.911&&0.441&0.874&0.316\tabularnewline
~~500&0.446&0.791&0.220&&0.127&0.941&0.935&&0.386&0.890&0.234\tabularnewline
~~750&0.423&0.786&0.182&&0.107&0.936&0.932&&0.357&0.905&0.207\tabularnewline
~~1000&0.402&0.785&0.164&&0.095&0.936&0.934&&0.337&0.908&0.189\tabularnewline
~~2000&0.368&0.785&0.139&&0.071&0.948&0.945&&0.293&0.933&0.153\tabularnewline
\hline
{\bfseries $x=$0.2}&&&&&&&&&&&\tabularnewline
~~100&0.468&0.890&0.645&&0.326&0.888&0.760&&0.647&0.821&0.231\tabularnewline
~~250&0.379&0.928&0.647&&0.211&0.917&0.843&&0.645&0.642&0.026\tabularnewline
~~500&0.328&0.935&0.666&&0.144&0.930&0.903&&0.635&0.403&0.009\tabularnewline
~~750&0.302&0.941&0.658&&0.116&0.944&0.932&&0.623&0.259&0.005\tabularnewline
~~1000&0.284&0.949&0.672&&0.100&0.943&0.941&&0.611&0.212&0.004\tabularnewline
~~2000&0.244&0.945&0.708&&0.074&0.943&0.945&&0.575&0.150&0.003\tabularnewline
\hline
{\bfseries $x=$0.6}&&&&&&&&&&&\tabularnewline
~~100&0.407&0.922&0.926&&0.381&0.928&0.926&&0.479&0.932&0.929\tabularnewline
~~250&0.338&0.934&0.936&&0.291&0.938&0.937&&0.535&0.931&0.927\tabularnewline
~~500&0.284&0.937&0.936&&0.253&0.944&0.940&&0.551&0.900&0.909\tabularnewline
~~750&0.258&0.943&0.944&&0.234&0.948&0.939&&0.538&0.881&0.903\tabularnewline
~~1000&0.246&0.940&0.937&&0.218&0.945&0.933&&0.529&0.853&0.888\tabularnewline
~~2000&0.211&0.943&0.940&&0.174&0.944&0.931&&0.498&0.760&0.832\tabularnewline
\hline
{\bfseries $x=$1}&&&&&&&&&&&\tabularnewline
~~100&0.378&0.897&0.902&&0.253&0.887&0.906&&0.484&0.905&0.901\tabularnewline
~~250&0.269&0.898&0.911&&0.084&0.877&0.900&&0.401&0.926&0.922\tabularnewline
~~500&0.204&0.906&0.917&&0.043&0.879&0.895&&0.374&0.929&0.928\tabularnewline
~~750&0.179&0.928&0.930&&0.035&0.881&0.898&&0.361&0.944&0.931\tabularnewline
~~1000&0.165&0.925&0.938&&0.036&0.880&0.892&&0.350&0.948&0.942\tabularnewline
~~2000&0.136&0.939&0.939&&0.048&0.894&0.907&&0.322&0.942&0.935\tabularnewline
\hline
\end{tabular}

%% file: simuls/output/table_ec_kepa_p2_d1.txt
%latex.default(table1, file = paste("output/table_ec_k", kernel,     "_p", p, "_d", deriv, ".txt", sep = ""), landscape = FALSE,     outer.size = "scriptsize", col.just = rep("c", ncol(table1)),     center = "none", title = "", table.env = FALSE, n.cgroup = c(3,         3, 3), cgroup = c("$h_{\\texttt{RBC}}$", "$h_{\\texttt{US}}$",         "$h_{\\texttt{MSE}}$"), n.rgroup = c(6, 6, 6, 6, 6, 6),     rgroup = c(paste("$x=$", eval, sep = "")))%
\begin{tabular}{lccccccccccc}
\hline\hline
\multicolumn{1}{l}{\bfseries }&\multicolumn{3}{c}{\bfseries $h_{\texttt{RBC}}$}&\multicolumn{1}{c}{\bfseries }&\multicolumn{3}{c}{\bfseries $h_{\texttt{US}}$}&\multicolumn{1}{c}{\bfseries }&\multicolumn{3}{c}{\bfseries $h_{\texttt{MSE}}$}\tabularnewline
\cline{2-4} \cline{6-8} \cline{10-12}
\multicolumn{1}{l}{}&\multicolumn{1}{c}{$h$}&\multicolumn{1}{c}{RBC}&\multicolumn{1}{c}{US}&\multicolumn{1}{c}{}&\multicolumn{1}{c}{$h$}&\multicolumn{1}{c}{RBC}&\multicolumn{1}{c}{US}&\multicolumn{1}{c}{}&\multicolumn{1}{c}{$h$}&\multicolumn{1}{c}{RBC}&\multicolumn{1}{c}{US}\tabularnewline
\hline
{\bfseries $x=$-1}&&&&&&&&&&&\tabularnewline
~~100&1.041&0.689&0.337&&0.436&0.926&0.891&&0.556&0.923&0.864\tabularnewline
~~250&0.861&0.774&0.399&&0.368&0.924&0.910&&0.484&0.937&0.890\tabularnewline
~~500&0.718&0.872&0.488&&0.321&0.930&0.920&&0.451&0.949&0.901\tabularnewline
~~750&0.635&0.915&0.561&&0.295&0.935&0.929&&0.434&0.950&0.900\tabularnewline
~~1000&0.582&0.931&0.625&&0.280&0.937&0.932&&0.420&0.949&0.908\tabularnewline
~~2000&0.475&0.942&0.760&&0.255&0.938&0.941&&0.392&0.955&0.904\tabularnewline
\hline
{\bfseries $x=$-0.6}&&&&&&&&&&&\tabularnewline
~~100&0.482&0.874&0.621&&0.335&0.911&0.809&&0.486&0.918&0.571\tabularnewline
~~250&0.360&0.936&0.654&&0.262&0.940&0.876&&0.475&0.933&0.229\tabularnewline
~~500&0.298&0.943&0.774&&0.221&0.942&0.910&&0.446&0.943&0.072\tabularnewline
~~750&0.268&0.948&0.815&&0.200&0.943&0.925&&0.424&0.942&0.043\tabularnewline
~~1000&0.250&0.944&0.850&&0.186&0.949&0.933&&0.408&0.942&0.041\tabularnewline
~~2000&0.214&0.948&0.884&&0.158&0.947&0.937&&0.373&0.943&0.024\tabularnewline
\hline
{\bfseries $x=$-0.2}&&&&&&&&&&&\tabularnewline
~~100&0.732&0.660&0.097&&0.564&0.819&0.388&&0.548&0.906&0.277\tabularnewline
~~250&0.642&0.677&0.072&&0.490&0.850&0.400&&0.488&0.937&0.204\tabularnewline
~~500&0.590&0.713&0.056&&0.446&0.858&0.431&&0.443&0.939&0.177\tabularnewline
~~750&0.563&0.727&0.048&&0.423&0.851&0.441&&0.420&0.942&0.163\tabularnewline
~~1000&0.539&0.743&0.049&&0.402&0.864&0.459&&0.403&0.942&0.167\tabularnewline
~~2000&0.488&0.783&0.050&&0.368&0.858&0.493&&0.365&0.934&0.175\tabularnewline
\hline
{\bfseries $x=$0.2}&&&&&&&&&&&\tabularnewline
~~100&0.575&0.626&0.193&&0.468&0.782&0.422&&0.575&0.691&0.035\tabularnewline
~~250&0.446&0.743&0.161&&0.379&0.862&0.400&&0.514&0.598&0.002\tabularnewline
~~500&0.377&0.823&0.143&&0.328&0.894&0.373&&0.467&0.524&0.000\tabularnewline
~~750&0.345&0.861&0.135&&0.302&0.916&0.389&&0.442&0.514&0.000\tabularnewline
~~1000&0.323&0.875&0.141&&0.284&0.921&0.402&&0.425&0.483&0.000\tabularnewline
~~2000&0.277&0.917&0.190&&0.244&0.941&0.483&&0.385&0.490&0.000\tabularnewline
\hline
{\bfseries $x=$0.6}&&&&&&&&&&&\tabularnewline
~~100&0.601&0.915&0.925&&0.407&0.926&0.926&&0.515&0.939&0.934\tabularnewline
~~250&0.537&0.920&0.930&&0.338&0.933&0.937&&0.508&0.945&0.936\tabularnewline
~~500&0.496&0.904&0.941&&0.284&0.938&0.942&&0.500&0.945&0.952\tabularnewline
~~750&0.461&0.907&0.941&&0.258&0.938&0.946&&0.482&0.948&0.950\tabularnewline
~~1000&0.431&0.904&0.944&&0.246&0.945&0.949&&0.468&0.942&0.945\tabularnewline
~~2000&0.362&0.911&0.946&&0.211&0.942&0.946&&0.434&0.941&0.945\tabularnewline
\hline
{\bfseries $x=$1}&&&&&&&&&&&\tabularnewline
~~100&1.084&0.882&0.878&&0.378&0.927&0.919&&0.659&0.945&0.930\tabularnewline
~~250&0.922&0.895&0.890&&0.269&0.921&0.918&&0.617&0.946&0.940\tabularnewline
~~500&0.784&0.922&0.915&&0.204&0.932&0.929&&0.577&0.947&0.941\tabularnewline
~~750&0.707&0.933&0.933&&0.179&0.926&0.929&&0.551&0.946&0.949\tabularnewline
~~1000&0.663&0.939&0.942&&0.165&0.932&0.929&&0.528&0.950&0.947\tabularnewline
~~2000&0.546&0.943&0.942&&0.136&0.940&0.946&&0.469&0.949&0.941\tabularnewline
\hline
\end{tabular}

%% file: simuls/output/table_il_kepa_p1_d0.txt
%latex.default(table2, file = paste("output/table_il_k", kernel,     "_p", p, "_d", deriv, ".txt", sep = ""), landscape = FALSE,     outer.size = "scriptsize", col.just = rep("c", ncol(table1)),     center = "none", title = "", table.env = FALSE, n.cgroup = c(3,         3, 3), cgroup = c("$h_{\\texttt{RBC}}$", "$h_{\\texttt{US}}$",         "$h_{\\texttt{MSE}}$"), n.rgroup = c(6, 6, 6, 6, 6, 6),     rgroup = c(paste("$x=$", eval, sep = "")))%
\begin{tabular}{lccccccccccc}
\hline\hline
\multicolumn{1}{l}{\bfseries }&\multicolumn{3}{c}{\bfseries $h_{\texttt{RBC}}$}&\multicolumn{1}{c}{\bfseries }&\multicolumn{3}{c}{\bfseries $h_{\texttt{US}}$}&\multicolumn{1}{c}{\bfseries }&\multicolumn{3}{c}{\bfseries $h_{\texttt{MSE}}$}\tabularnewline
\cline{2-4} \cline{6-8} \cline{10-12}
\multicolumn{1}{l}{}&\multicolumn{1}{c}{$h$}&\multicolumn{1}{c}{RBC}&\multicolumn{1}{c}{US}&\multicolumn{1}{c}{}&\multicolumn{1}{c}{$h$}&\multicolumn{1}{c}{RBC}&\multicolumn{1}{c}{US}&\multicolumn{1}{c}{}&\multicolumn{1}{c}{$h$}&\multicolumn{1}{c}{RBC}&\multicolumn{1}{c}{US}\tabularnewline
\hline
{\bfseries $x=$-1}&&&&&&&&&&&\tabularnewline
~~100&0.436&2.442&1.674&&0.320&2.761&1.793&&0.507&2.330&1.656\tabularnewline
~~250&0.368&1.713&1.239&&0.115&2.781&1.795&&0.462&1.492&1.089\tabularnewline
~~500&0.321&1.282&0.940&&0.116&2.585&1.697&&0.438&1.084&0.796\tabularnewline
~~750&0.295&1.090&0.801&&0.168&2.210&1.474&&0.420&0.907&0.667\tabularnewline
~~1000&0.280&0.966&0.710&&0.205&1.799&1.240&&0.404&0.801&0.589\tabularnewline
~~2000&0.255&0.711&0.524&&0.143&1.222&0.891&&0.356&0.605&0.446\tabularnewline
\hline
{\bfseries $x=$-0.6}&&&&&&&&&&&\tabularnewline
~~100&0.335&1.020&0.762&&0.255&1.173&0.863&&0.356&0.983&0.734\tabularnewline
~~250&0.262&0.715&0.537&&0.145&0.985&0.739&&0.342&0.633&0.473\tabularnewline
~~500&0.221&0.547&0.411&&0.104&0.812&0.611&&0.316&0.461&0.346\tabularnewline
~~750&0.200&0.467&0.352&&0.090&0.706&0.531&&0.291&0.390&0.293\tabularnewline
~~1000&0.186&0.419&0.315&&0.081&0.646&0.487&&0.274&0.349&0.262\tabularnewline
~~2000&0.158&0.322&0.242&&0.063&0.513&0.386&&0.235&0.266&0.200\tabularnewline
\hline
{\bfseries $x=$-0.2}&&&&&&&&&&&\tabularnewline
~~100&0.564&0.786&0.592&&0.242&1.171&0.864&&0.512&0.799&0.601\tabularnewline
~~250&0.490&0.533&0.401&&0.169&0.890&0.671&&0.441&0.543&0.409\tabularnewline
~~500&0.446&0.396&0.298&&0.127&0.724&0.545&&0.386&0.411&0.309\tabularnewline
~~750&0.423&0.334&0.251&&0.107&0.644&0.485&&0.357&0.349&0.263\tabularnewline
~~1000&0.402&0.297&0.223&&0.095&0.592&0.446&&0.337&0.312&0.234\tabularnewline
~~2000&0.368&0.221&0.166&&0.071&0.481&0.362&&0.293&0.236&0.178\tabularnewline
\hline
{\bfseries $x=$0.2}&&&&&&&&&&&\tabularnewline
~~100&0.468&0.844&0.632&&0.326&1.050&0.778&&0.647&0.711&0.532\tabularnewline
~~250&0.379&0.589&0.443&&0.211&0.822&0.619&&0.645&0.451&0.338\tabularnewline
~~500&0.328&0.447&0.336&&0.144&0.692&0.520&&0.635&0.321&0.241\tabularnewline
~~750&0.302&0.379&0.286&&0.116&0.619&0.466&&0.623&0.265&0.200\tabularnewline
~~1000&0.284&0.339&0.255&&0.100&0.575&0.433&&0.611&0.232&0.174\tabularnewline
~~2000&0.244&0.258&0.194&&0.074&0.471&0.355&&0.575&0.169&0.127\tabularnewline
\hline
{\bfseries $x=$0.6}&&&&&&&&&&&\tabularnewline
~~100&0.407&0.942&0.706&&0.381&0.994&0.749&&0.479&0.864&0.643\tabularnewline
~~250&0.338&0.647&0.489&&0.291&0.718&0.542&&0.535&0.525&0.399\tabularnewline
~~500&0.284&0.495&0.373&&0.253&0.544&0.410&&0.551&0.367&0.282\tabularnewline
~~750&0.258&0.422&0.318&&0.234&0.463&0.348&&0.538&0.304&0.234\tabularnewline
~~1000&0.246&0.375&0.283&&0.218&0.414&0.311&&0.529&0.266&0.205\tabularnewline
~~2000&0.211&0.285&0.215&&0.174&0.328&0.247&&0.498&0.196&0.150\tabularnewline
\hline
{\bfseries $x=$1}&&&&&&&&&&&\tabularnewline
~~100&0.378&2.547&1.725&&0.253&2.847&1.832&&0.484&2.397&1.704\tabularnewline
~~250&0.269&2.034&1.440&&0.084&2.893&1.850&&0.401&1.625&1.182\tabularnewline
~~500&0.204&1.609&1.169&&0.043&2.825&1.828&&0.374&1.172&0.861\tabularnewline
~~750&0.179&1.391&1.018&&0.035&2.760&1.806&&0.361&0.971&0.715\tabularnewline
~~1000&0.165&1.250&0.917&&0.036&2.670&1.767&&0.350&0.856&0.630\tabularnewline
~~2000&0.136&0.966&0.710&&0.048&2.069&1.458&&0.322&0.627&0.462\tabularnewline
\hline
\end{tabular}

%% file: simuls/output/table_il_kepa_p2_d1.txt
%latex.default(table2, file = paste("output/table_il_k", kernel,     "_p", p, "_d", deriv, ".txt", sep = ""), landscape = FALSE,     outer.size = "scriptsize", col.just = rep("c", ncol(table1)),     center = "none", title = "", table.env = FALSE, n.cgroup = c(3,         3, 3), cgroup = c("$h_{\\texttt{RBC}}$", "$h_{\\texttt{US}}$",         "$h_{\\texttt{MSE}}$"), n.rgroup = c(6, 6, 6, 6, 6, 6),     rgroup = c(paste("$x=$", eval, sep = "")))%
\begin{tabular}{lccccccccccc}
\hline\hline
\multicolumn{1}{l}{\bfseries }&\multicolumn{3}{c}{\bfseries $h_{\texttt{RBC}}$}&\multicolumn{1}{c}{\bfseries }&\multicolumn{3}{c}{\bfseries $h_{\texttt{US}}$}&\multicolumn{1}{c}{\bfseries }&\multicolumn{3}{c}{\bfseries $h_{\texttt{MSE}}$}\tabularnewline
\cline{2-4} \cline{6-8} \cline{10-12}
\multicolumn{1}{l}{}&\multicolumn{1}{c}{$h$}&\multicolumn{1}{c}{RBC}&\multicolumn{1}{c}{US}&\multicolumn{1}{c}{}&\multicolumn{1}{c}{$h$}&\multicolumn{1}{c}{RBC}&\multicolumn{1}{c}{US}&\multicolumn{1}{c}{}&\multicolumn{1}{c}{$h$}&\multicolumn{1}{c}{RBC}&\multicolumn{1}{c}{US}\tabularnewline
\hline
{\bfseries $x=$-1}&&&&&&&&&&&\tabularnewline
~~100&1.041&19.701&9.445&&0.436&69.744&28.528&&0.556&48.016&22.617\tabularnewline
~~250&0.861&16.338&8.028&&0.368&61.142&28.904&&0.484&34.436&16.763\tabularnewline
~~500&0.718&14.947&7.391&&0.321&49.275&24.059&&0.451&26.558&13.049\tabularnewline
~~750&0.635&14.553&7.190&&0.295&44.292&21.675&&0.434&22.961&11.300\tabularnewline
~~1000&0.582&14.242&7.042&&0.280&40.630&19.944&&0.420&20.849&10.271\tabularnewline
~~2000&0.475&13.336&6.587&&0.255&32.206&15.851&&0.392&16.238&8.016\tabularnewline
\hline
{\bfseries $x=$-0.6}&&&&&&&&&&&\tabularnewline
~~100&0.482&6.098&3.158&&0.335&9.635&4.835&&0.486&5.571&2.846\tabularnewline
~~250&0.360&5.109&2.557&&0.262&8.041&4.096&&0.475&3.562&1.750\tabularnewline
~~500&0.298&4.512&2.289&&0.221&7.088&3.612&&0.446&2.666&1.273\tabularnewline
~~750&0.268&4.255&2.165&&0.200&6.681&3.405&&0.424&2.287&1.096\tabularnewline
~~1000&0.250&4.084&2.079&&0.186&6.358&3.243&&0.408&2.055&0.997\tabularnewline
~~2000&0.214&3.627&1.846&&0.158&5.737&2.921&&0.373&1.597&0.804\tabularnewline
\hline
{\bfseries $x=$-0.2}&&&&&&&&&&&\tabularnewline
~~100&0.732&3.023&1.531&&0.564&4.600&2.343&&0.548&4.033&2.033\tabularnewline
~~250&0.642&2.216&1.127&&0.490&3.476&1.775&&0.488&2.994&1.513\tabularnewline
~~500&0.590&1.757&0.896&&0.446&2.832&1.443&&0.443&2.428&1.231\tabularnewline
~~750&0.563&1.535&0.781&&0.423&2.541&1.294&&0.420&2.153&1.091\tabularnewline
~~1000&0.539&1.407&0.716&&0.402&2.364&1.204&&0.403&1.977&1.004\tabularnewline
~~2000&0.488&1.156&0.588&&0.368&1.971&1.004&&0.365&1.617&0.822\tabularnewline
\hline
{\bfseries $x=$0.2}&&&&&&&&&&&\tabularnewline
~~100&0.575&3.998&2.031&&0.468&5.462&2.765&&0.575&3.747&1.885\tabularnewline
~~250&0.446&3.521&1.783&&0.379&4.477&2.277&&0.514&2.765&1.396\tabularnewline
~~500&0.377&3.133&1.591&&0.328&3.853&1.962&&0.467&2.244&1.137\tabularnewline
~~750&0.345&2.909&1.478&&0.302&3.554&1.813&&0.442&1.988&1.009\tabularnewline
~~1000&0.323&2.770&1.408&&0.284&3.362&1.713&&0.425&1.827&0.929\tabularnewline
~~2000&0.277&2.456&1.250&&0.244&2.973&1.515&&0.385&1.493&0.760\tabularnewline
\hline
{\bfseries $x=$0.6}&&&&&&&&&&&\tabularnewline
~~100&0.601&4.652&2.707&&0.407&7.736&3.930&&0.515&5.184&2.714\tabularnewline
~~250&0.537&3.328&1.792&&0.338&6.157&3.165&&0.508&3.305&1.709\tabularnewline
~~500&0.496&2.585&1.354&&0.284&5.415&2.762&&0.500&2.348&1.187\tabularnewline
~~750&0.461&2.313&1.196&&0.258&4.969&2.534&&0.482&2.000&0.994\tabularnewline
~~1000&0.431&2.179&1.124&&0.246&4.692&2.396&&0.468&1.791&0.883\tabularnewline
~~2000&0.362&1.953&1.001&&0.211&4.095&2.087&&0.434&1.377&0.680\tabularnewline
\hline
{\bfseries $x=$1}&&&&&&&&&&&\tabularnewline
~~100&1.084&18.377&8.838&&0.378&77.836&31.082&&0.659&36.723&17.371\tabularnewline
~~250&0.922&14.556&7.134&&0.269&102.575&46.277&&0.617&24.053&11.744\tabularnewline
~~500&0.784&12.866&6.341&&0.204&99.887&47.579&&0.577&18.347&9.047\tabularnewline
~~750&0.707&12.195&6.022&&0.179&96.458&46.670&&0.551&15.971&7.885\tabularnewline
~~1000&0.663&11.490&5.680&&0.165&91.635&44.564&&0.528&14.687&7.261\tabularnewline
~~2000&0.546&10.592&5.232&&0.136&82.575&40.567&&0.469&12.323&6.085\tabularnewline
\hline
\end{tabular}

%% file: simuls/output/table_ec_kuni_p1_d0.txt
%latex.default(table1, file = paste("output/table_ec_k", kernel,     "_p", p, "_d", deriv, ".txt", sep = ""), landscape = FALSE,     outer.size = "scriptsize", col.just = rep("c", ncol(table1)),     center = "none", title = "", table.env = FALSE, n.cgroup = c(3,         3, 3), cgroup = c("$h_{\\texttt{RBC}}$", "$h_{\\texttt{US}}$",         "$h_{\\texttt{MSE}}$"), n.rgroup = c(6, 6, 6, 6, 6, 6),     rgroup = c(paste("$x=$", eval, sep = "")))%
\begin{tabular}{lccccccccccc}
\hline\hline
\multicolumn{1}{l}{\bfseries }&\multicolumn{3}{c}{\bfseries $h_{\texttt{RBC}}$}&\multicolumn{1}{c}{\bfseries }&\multicolumn{3}{c}{\bfseries $h_{\texttt{US}}$}&\multicolumn{1}{c}{\bfseries }&\multicolumn{3}{c}{\bfseries $h_{\texttt{MSE}}$}\tabularnewline
\cline{2-4} \cline{6-8} \cline{10-12}
\multicolumn{1}{l}{}&\multicolumn{1}{c}{$h$}&\multicolumn{1}{c}{RBC}&\multicolumn{1}{c}{US}&\multicolumn{1}{c}{}&\multicolumn{1}{c}{$h$}&\multicolumn{1}{c}{RBC}&\multicolumn{1}{c}{US}&\multicolumn{1}{c}{}&\multicolumn{1}{c}{$h$}&\multicolumn{1}{c}{RBC}&\multicolumn{1}{c}{US}\tabularnewline
\hline
{\bfseries $x=$-1}&&&&&&&&&&&\tabularnewline
~~100&0.372&0.901&0.876&&0.286&0.893&0.888&&0.462&0.903&0.877\tabularnewline
~~250&0.319&0.910&0.899&&0.122&0.888&0.909&&0.397&0.919&0.881\tabularnewline
~~500&0.277&0.925&0.905&&0.105&0.889&0.901&&0.372&0.934&0.852\tabularnewline
~~750&0.256&0.930&0.917&&0.117&0.897&0.892&&0.352&0.940&0.827\tabularnewline
~~1000&0.242&0.938&0.906&&0.126&0.893&0.873&&0.326&0.945&0.812\tabularnewline
~~2000&0.221&0.941&0.900&&0.156&0.894&0.809&&0.247&0.942&0.854\tabularnewline
\hline
{\bfseries $x=$-0.6}&&&&&&&&&&&\tabularnewline
~~100&0.297&0.925&0.890&&0.175&0.919&0.918&&0.364&0.930&0.841\tabularnewline
~~250&0.224&0.937&0.910&&0.106&0.921&0.927&&0.319&0.936&0.794\tabularnewline
~~500&0.184&0.947&0.924&&0.080&0.932&0.940&&0.278&0.942&0.778\tabularnewline
~~750&0.167&0.949&0.932&&0.069&0.938&0.942&&0.256&0.951&0.768\tabularnewline
~~1000&0.155&0.947&0.936&&0.062&0.939&0.944&&0.240&0.949&0.779\tabularnewline
~~2000&0.131&0.946&0.934&&0.050&0.948&0.941&&0.207&0.945&0.770\tabularnewline
\hline
{\bfseries $x=$-0.2}&&&&&&&&&&&\tabularnewline
~~100&0.466&0.863&0.330&&0.205&0.920&0.841&&0.455&0.894&0.287\tabularnewline
~~250&0.411&0.852&0.237&&0.140&0.924&0.903&&0.396&0.899&0.153\tabularnewline
~~500&0.377&0.846&0.180&&0.103&0.939&0.933&&0.348&0.912&0.094\tabularnewline
~~750&0.355&0.841&0.155&&0.087&0.937&0.935&&0.322&0.927&0.066\tabularnewline
~~1000&0.339&0.843&0.144&&0.076&0.938&0.934&&0.303&0.931&0.066\tabularnewline
~~2000&0.309&0.836&0.128&&0.057&0.945&0.946&&0.263&0.941&0.050\tabularnewline
\hline
{\bfseries $x=$0.2}&&&&&&&&&&&\tabularnewline
~~100&0.381&0.910&0.606&&0.281&0.904&0.743&&0.477&0.933&0.271\tabularnewline
~~250&0.302&0.938&0.623&&0.188&0.928&0.831&&0.428&0.940&0.090\tabularnewline
~~500&0.255&0.939&0.671&&0.132&0.932&0.892&&0.388&0.939&0.042\tabularnewline
~~750&0.232&0.942&0.669&&0.105&0.937&0.922&&0.365&0.940&0.028\tabularnewline
~~1000&0.217&0.948&0.703&&0.089&0.942&0.933&&0.348&0.941&0.026\tabularnewline
~~2000&0.185&0.945&0.748&&0.063&0.943&0.944&&0.313&0.926&0.019\tabularnewline
\hline
{\bfseries $x=$0.6}&&&&&&&&&&&\tabularnewline
~~100&0.348&0.928&0.923&&0.271&0.932&0.931&&0.427&0.937&0.925\tabularnewline
~~250&0.304&0.934&0.935&&0.206&0.939&0.937&&0.391&0.945&0.924\tabularnewline
~~500&0.273&0.935&0.935&&0.179&0.946&0.942&&0.352&0.945&0.920\tabularnewline
~~750&0.256&0.938&0.933&&0.173&0.951&0.949&&0.330&0.949&0.912\tabularnewline
~~1000&0.244&0.935&0.934&&0.166&0.946&0.944&&0.315&0.944&0.911\tabularnewline
~~2000&0.209&0.935&0.927&&0.150&0.950&0.932&&0.279&0.948&0.880\tabularnewline
\hline
{\bfseries $x=$1}&&&&&&&&&&&\tabularnewline
~~100&0.323&0.907&0.917&&0.266&0.907&0.919&&0.464&0.909&0.918\tabularnewline
~~250&0.230&0.901&0.919&&0.089&0.894&0.919&&0.393&0.922&0.930\tabularnewline
~~500&0.175&0.910&0.922&&0.051&0.893&0.907&&0.357&0.934&0.940\tabularnewline
~~750&0.154&0.921&0.937&&0.041&0.897&0.915&&0.334&0.940&0.939\tabularnewline
~~1000&0.139&0.922&0.939&&0.038&0.894&0.918&&0.316&0.948&0.944\tabularnewline
~~2000&0.115&0.938&0.935&&0.045&0.900&0.918&&0.287&0.937&0.937\tabularnewline
\hline
\end{tabular}

%% file: simuls/output/table_ec_kuni_p2_d1.txt
%latex.default(table1, file = paste("output/table_ec_k", kernel,     "_p", p, "_d", deriv, ".txt", sep = ""), landscape = FALSE,     outer.size = "scriptsize", col.just = rep("c", ncol(table1)),     center = "none", title = "", table.env = FALSE, n.cgroup = c(3,         3, 3), cgroup = c("$h_{\\texttt{RBC}}$", "$h_{\\texttt{US}}$",         "$h_{\\texttt{MSE}}$"), n.rgroup = c(6, 6, 6, 6, 6, 6),     rgroup = c(paste("$x=$", eval, sep = "")))%
\begin{tabular}{lccccccccccc}
\hline\hline
\multicolumn{1}{l}{\bfseries }&\multicolumn{3}{c}{\bfseries $h_{\texttt{RBC}}$}&\multicolumn{1}{c}{\bfseries }&\multicolumn{3}{c}{\bfseries $h_{\texttt{US}}$}&\multicolumn{1}{c}{\bfseries }&\multicolumn{3}{c}{\bfseries $h_{\texttt{MSE}}$}\tabularnewline
\cline{2-4} \cline{6-8} \cline{10-12}
\multicolumn{1}{l}{}&\multicolumn{1}{c}{$h$}&\multicolumn{1}{c}{RBC}&\multicolumn{1}{c}{US}&\multicolumn{1}{c}{}&\multicolumn{1}{c}{$h$}&\multicolumn{1}{c}{RBC}&\multicolumn{1}{c}{US}&\multicolumn{1}{c}{}&\multicolumn{1}{c}{$h$}&\multicolumn{1}{c}{RBC}&\multicolumn{1}{c}{US}\tabularnewline
\hline
{\bfseries $x=$-1}&&&&&&&&&&&\tabularnewline
~~100&0.878&0.835&0.413&&0.372&0.928&0.907&&0.525&0.924&0.856\tabularnewline
~~250&0.720&0.883&0.476&&0.319&0.934&0.917&&0.449&0.936&0.878\tabularnewline
~~500&0.588&0.930&0.605&&0.277&0.937&0.923&&0.409&0.940&0.900\tabularnewline
~~750&0.517&0.935&0.677&&0.256&0.939&0.935&&0.393&0.946&0.899\tabularnewline
~~1000&0.473&0.943&0.739&&0.242&0.941&0.934&&0.382&0.947&0.910\tabularnewline
~~2000&0.382&0.942&0.851&&0.221&0.943&0.941&&0.357&0.948&0.898\tabularnewline
\hline
{\bfseries $x=$-0.6}&&&&&&&&&&&\tabularnewline
~~100&0.418&0.899&0.592&&0.297&0.919&0.792&&0.458&0.927&0.420\tabularnewline
~~250&0.305&0.935&0.680&&0.224&0.932&0.873&&0.426&0.937&0.129\tabularnewline
~~500&0.251&0.946&0.797&&0.184&0.944&0.915&&0.394&0.947&0.064\tabularnewline
~~750&0.227&0.948&0.831&&0.167&0.943&0.926&&0.375&0.943&0.044\tabularnewline
~~1000&0.212&0.948&0.867&&0.155&0.947&0.933&&0.361&0.949&0.042\tabularnewline
~~2000&0.181&0.950&0.891&&0.131&0.944&0.945&&0.331&0.947&0.027\tabularnewline
\hline
{\bfseries $x=$-0.2}&&&&&&&&&&&\tabularnewline
~~100&0.630&0.782&0.100&&0.466&0.879&0.437&&0.512&0.921&0.135\tabularnewline
~~250&0.558&0.809&0.064&&0.411&0.883&0.470&&0.457&0.946&0.067\tabularnewline
~~500&0.512&0.829&0.049&&0.377&0.886&0.506&&0.414&0.941&0.052\tabularnewline
~~750&0.485&0.849&0.046&&0.355&0.882&0.516&&0.392&0.947&0.036\tabularnewline
~~1000&0.467&0.843&0.048&&0.339&0.886&0.535&&0.376&0.941&0.038\tabularnewline
~~2000&0.420&0.869&0.054&&0.309&0.884&0.562&&0.341&0.944&0.041\tabularnewline
\hline
{\bfseries $x=$0.2}&&&&&&&&&&&\tabularnewline
~~100&0.491&0.785&0.200&&0.381&0.866&0.502&&0.513&0.818&0.048\tabularnewline
~~250&0.379&0.875&0.172&&0.302&0.923&0.526&&0.456&0.785&0.004\tabularnewline
~~500&0.321&0.907&0.150&&0.255&0.938&0.552&&0.414&0.759&0.000\tabularnewline
~~750&0.294&0.928&0.150&&0.232&0.947&0.600&&0.392&0.773&0.000\tabularnewline
~~1000&0.274&0.933&0.154&&0.217&0.947&0.623&&0.376&0.758&0.000\tabularnewline
~~2000&0.234&0.948&0.218&&0.185&0.948&0.713&&0.340&0.778&0.000\tabularnewline
\hline
{\bfseries $x=$0.6}&&&&&&&&&&&\tabularnewline
~~100&0.516&0.928&0.928&&0.348&0.932&0.917&&0.470&0.935&0.935\tabularnewline
~~250&0.446&0.934&0.925&&0.304&0.929&0.933&&0.442&0.943&0.934\tabularnewline
~~500&0.393&0.932&0.936&&0.273&0.927&0.944&&0.415&0.948&0.949\tabularnewline
~~750&0.356&0.931&0.939&&0.256&0.938&0.946&&0.399&0.952&0.950\tabularnewline
~~1000&0.338&0.934&0.943&&0.244&0.938&0.946&&0.388&0.950&0.945\tabularnewline
~~2000&0.287&0.928&0.937&&0.209&0.934&0.950&&0.359&0.950&0.943\tabularnewline
\hline
{\bfseries $x=$1}&&&&&&&&&&&\tabularnewline
~~100&0.901&0.922&0.886&&0.323&0.932&0.926&&0.555&0.935&0.932\tabularnewline
~~250&0.757&0.926&0.903&&0.230&0.926&0.918&&0.502&0.933&0.937\tabularnewline
~~500&0.633&0.938&0.920&&0.175&0.930&0.931&&0.454&0.940&0.941\tabularnewline
~~750&0.569&0.943&0.943&&0.154&0.932&0.935&&0.429&0.945&0.950\tabularnewline
~~1000&0.532&0.948&0.944&&0.139&0.935&0.935&&0.413&0.950&0.948\tabularnewline
~~2000&0.442&0.944&0.941&&0.115&0.941&0.940&&0.384&0.945&0.943\tabularnewline
\hline
\end{tabular}

%% file: simuls/output/table_il_kuni_p1_d0.txt
%latex.default(table2, file = paste("output/table_il_k", kernel,     "_p", p, "_d", deriv, ".txt", sep = ""), landscape = FALSE,     outer.size = "scriptsize", col.just = rep("c", ncol(table1)),     center = "none", title = "", table.env = FALSE, n.cgroup = c(3,         3, 3), cgroup = c("$h_{\\texttt{RBC}}$", "$h_{\\texttt{US}}$",         "$h_{\\texttt{MSE}}$"), n.rgroup = c(6, 6, 6, 6, 6, 6),     rgroup = c(paste("$x=$", eval, sep = "")))%
\begin{tabular}{lccccccccccc}
\hline\hline
\multicolumn{1}{l}{\bfseries }&\multicolumn{3}{c}{\bfseries $h_{\texttt{RBC}}$}&\multicolumn{1}{c}{\bfseries }&\multicolumn{3}{c}{\bfseries $h_{\texttt{US}}$}&\multicolumn{1}{c}{\bfseries }&\multicolumn{3}{c}{\bfseries $h_{\texttt{MSE}}$}\tabularnewline
\cline{2-4} \cline{6-8} \cline{10-12}
\multicolumn{1}{l}{}&\multicolumn{1}{c}{$h$}&\multicolumn{1}{c}{RBC}&\multicolumn{1}{c}{US}&\multicolumn{1}{c}{}&\multicolumn{1}{c}{$h$}&\multicolumn{1}{c}{RBC}&\multicolumn{1}{c}{US}&\multicolumn{1}{c}{}&\multicolumn{1}{c}{$h$}&\multicolumn{1}{c}{RBC}&\multicolumn{1}{c}{US}\tabularnewline
\hline
{\bfseries $x=$-1}&&&&&&&&&&&\tabularnewline
~~100&0.372&2.561&1.638&&0.286&2.649&1.688&&0.462&2.566&1.636\tabularnewline
~~250&0.319&1.948&1.274&&0.122&2.593&1.656&&0.397&1.694&1.116\tabularnewline
~~500&0.277&1.457&0.965&&0.105&2.453&1.573&&0.372&1.230&0.816\tabularnewline
~~750&0.256&1.231&0.818&&0.117&2.266&1.453&&0.352&1.034&0.688\tabularnewline
~~1000&0.242&1.092&0.725&&0.126&2.106&1.347&&0.326&0.931&0.620\tabularnewline
~~2000&0.221&0.798&0.532&&0.156&1.507&0.982&&0.247&0.753&0.502\tabularnewline
\hline
{\bfseries $x=$-0.6}&&&&&&&&&&&\tabularnewline
~~100&0.297&1.107&0.741&&0.175&1.260&0.844&&0.364&1.026&0.690\tabularnewline
~~250&0.224&0.802&0.535&&0.106&1.169&0.779&&0.319&0.682&0.455\tabularnewline
~~500&0.184&0.619&0.413&&0.080&0.964&0.643&&0.278&0.513&0.342\tabularnewline
~~750&0.167&0.528&0.352&&0.069&0.836&0.557&&0.256&0.434&0.290\tabularnewline
~~1000&0.155&0.474&0.316&&0.062&0.762&0.509&&0.240&0.389&0.259\tabularnewline
~~2000&0.131&0.364&0.243&&0.050&0.597&0.398&&0.207&0.296&0.197\tabularnewline
\hline
{\bfseries $x=$-0.2}&&&&&&&&&&&\tabularnewline
~~100&0.466&0.894&0.598&&0.205&1.234&0.822&&0.455&0.879&0.585\tabularnewline
~~250&0.411&0.604&0.404&&0.140&1.020&0.682&&0.396&0.594&0.395\tabularnewline
~~500&0.377&0.449&0.300&&0.103&0.836&0.558&&0.348&0.447&0.298\tabularnewline
~~750&0.355&0.380&0.254&&0.087&0.744&0.495&&0.322&0.381&0.254\tabularnewline
~~1000&0.339&0.338&0.225&&0.076&0.683&0.457&&0.303&0.340&0.226\tabularnewline
~~2000&0.309&0.252&0.168&&0.057&0.555&0.370&&0.263&0.258&0.172\tabularnewline
\hline
{\bfseries $x=$0.2}&&&&&&&&&&&\tabularnewline
~~100&0.381&0.972&0.646&&0.281&1.132&0.754&&0.477&0.863&0.572\tabularnewline
~~250&0.302&0.684&0.457&&0.188&0.914&0.611&&0.428&0.573&0.381\tabularnewline
~~500&0.255&0.524&0.349&&0.132&0.762&0.508&&0.388&0.425&0.283\tabularnewline
~~750&0.232&0.446&0.298&&0.105&0.684&0.456&&0.365&0.358&0.239\tabularnewline
~~1000&0.217&0.401&0.267&&0.089&0.637&0.426&&0.348&0.318&0.212\tabularnewline
~~2000&0.185&0.306&0.204&&0.063&0.529&0.353&&0.313&0.238&0.159\tabularnewline
\hline
{\bfseries $x=$0.6}&&&&&&&&&&&\tabularnewline
~~100&0.348&1.045&0.702&&0.271&1.150&0.772&&0.427&0.936&0.635\tabularnewline
~~250&0.304&0.710&0.477&&0.206&0.879&0.590&&0.391&0.611&0.411\tabularnewline
~~500&0.273&0.527&0.354&&0.179&0.668&0.447&&0.352&0.454&0.304\tabularnewline
~~750&0.256&0.446&0.299&&0.173&0.561&0.374&&0.330&0.383&0.256\tabularnewline
~~1000&0.244&0.395&0.265&&0.166&0.492&0.329&&0.315&0.339&0.227\tabularnewline
~~2000&0.209&0.301&0.201&&0.150&0.368&0.245&&0.279&0.257&0.171\tabularnewline
\hline
{\bfseries $x=$1}&&&&&&&&&&&\tabularnewline
~~100&0.323&2.608&1.664&&0.266&2.666&1.696&&0.464&2.555&1.631\tabularnewline
~~250&0.230&2.307&1.477&&0.089&2.702&1.707&&0.393&1.733&1.136\tabularnewline
~~500&0.175&1.838&1.206&&0.051&2.632&1.679&&0.357&1.263&0.838\tabularnewline
~~750&0.154&1.578&1.044&&0.041&2.568&1.656&&0.334&1.059&0.705\tabularnewline
~~1000&0.139&1.427&0.946&&0.038&2.547&1.636&&0.316&0.943&0.628\tabularnewline
~~2000&0.115&1.101&0.733&&0.045&2.171&1.410&&0.287&0.694&0.463\tabularnewline
\hline
\end{tabular}

%% file: simuls/output/table_il_kuni_p2_d1.txt
%latex.default(table2, file = paste("output/table_il_k", kernel,     "_p", p, "_d", deriv, ".txt", sep = ""), landscape = FALSE,     outer.size = "scriptsize", col.just = rep("c", ncol(table1)),     center = "none", title = "", table.env = FALSE, n.cgroup = c(3,         3, 3), cgroup = c("$h_{\\texttt{RBC}}$", "$h_{\\texttt{US}}$",         "$h_{\\texttt{MSE}}$"), n.rgroup = c(6, 6, 6, 6, 6, 6),     rgroup = c(paste("$x=$", eval, sep = "")))%
\begin{tabular}{lccccccccccc}
\hline\hline
\multicolumn{1}{l}{\bfseries }&\multicolumn{3}{c}{\bfseries $h_{\texttt{RBC}}$}&\multicolumn{1}{c}{\bfseries }&\multicolumn{3}{c}{\bfseries $h_{\texttt{US}}$}&\multicolumn{1}{c}{\bfseries }&\multicolumn{3}{c}{\bfseries $h_{\texttt{MSE}}$}\tabularnewline
\cline{2-4} \cline{6-8} \cline{10-12}
\multicolumn{1}{l}{}&\multicolumn{1}{c}{$h$}&\multicolumn{1}{c}{RBC}&\multicolumn{1}{c}{US}&\multicolumn{1}{c}{}&\multicolumn{1}{c}{$h$}&\multicolumn{1}{c}{RBC}&\multicolumn{1}{c}{US}&\multicolumn{1}{c}{}&\multicolumn{1}{c}{$h$}&\multicolumn{1}{c}{RBC}&\multicolumn{1}{c}{US}\tabularnewline
\hline
{\bfseries $x=$-1}&&&&&&&&&&&\tabularnewline
~~100&0.878&28.499&11.007&&0.372&74.089&26.860&&0.525&58.507&21.792\tabularnewline
~~250&0.720&23.646&9.391&&0.319&85.698&32.415&&0.449&42.941&16.854\tabularnewline
~~500&0.588&22.170&8.862&&0.277&69.500&27.253&&0.409&33.847&13.454\tabularnewline
~~750&0.517&21.707&8.686&&0.256&61.629&24.361&&0.393&29.132&11.609\tabularnewline
~~1000&0.473&21.202&8.489&&0.242&56.624&22.424&&0.382&26.357&10.515\tabularnewline
~~2000&0.382&20.050&8.033&&0.221&44.037&17.519&&0.357&20.483&8.194\tabularnewline
\hline
{\bfseries $x=$-0.6}&&&&&&&&&&&\tabularnewline
~~100&0.418&7.793&3.196&&0.297&11.921&4.762&&0.458&6.335&2.673\tabularnewline
~~250&0.305&6.840&2.707&&0.224&11.095&4.415&&0.426&4.159&1.683\tabularnewline
~~500&0.251&6.225&2.472&&0.184&10.016&4.010&&0.394&3.181&1.265\tabularnewline
~~750&0.227&5.844&2.334&&0.167&9.395&3.760&&0.375&2.752&1.095\tabularnewline
~~1000&0.212&5.596&2.236&&0.155&8.965&3.589&&0.361&2.509&0.998\tabularnewline
~~2000&0.181&4.983&1.991&&0.131&8.115&3.242&&0.331&2.022&0.807\tabularnewline
\hline
{\bfseries $x=$-0.2}&&&&&&&&&&&\tabularnewline
~~100&0.630&3.974&1.598&&0.466&6.568&2.636&&0.512&4.846&1.878\tabularnewline
~~250&0.558&2.909&1.166&&0.411&4.975&1.994&&0.457&3.552&1.399\tabularnewline
~~500&0.512&2.315&0.930&&0.377&4.070&1.629&&0.414&2.876&1.144\tabularnewline
~~750&0.485&2.045&0.818&&0.355&3.667&1.469&&0.392&2.543&1.011\tabularnewline
~~1000&0.467&1.873&0.750&&0.339&3.431&1.373&&0.376&2.336&0.931\tabularnewline
~~2000&0.420&1.551&0.621&&0.309&2.908&1.162&&0.341&1.910&0.763\tabularnewline
\hline
{\bfseries $x=$0.2}&&&&&&&&&&&\tabularnewline
~~100&0.491&5.462&2.151&&0.381&8.147&3.229&&0.513&4.813&1.859\tabularnewline
~~250&0.379&4.802&1.897&&0.302&6.861&2.742&&0.456&3.550&1.395\tabularnewline
~~500&0.321&4.238&1.686&&0.255&6.070&2.417&&0.414&2.874&1.141\tabularnewline
~~750&0.294&3.947&1.572&&0.232&5.664&2.267&&0.392&2.548&1.014\tabularnewline
~~1000&0.274&3.766&1.502&&0.217&5.400&2.160&&0.376&2.342&0.934\tabularnewline
~~2000&0.234&3.362&1.342&&0.185&4.837&1.935&&0.340&1.914&0.766\tabularnewline
\hline
{\bfseries $x=$0.6}&&&&&&&&&&&\tabularnewline
~~100&0.516&5.937&2.723&&0.348&10.278&4.159&&0.470&6.156&2.648\tabularnewline
~~250&0.446&4.422&1.876&&0.304&8.031&3.244&&0.442&4.028&1.669\tabularnewline
~~500&0.393&3.684&1.516&&0.273&6.540&2.629&&0.415&3.008&1.218\tabularnewline
~~750&0.356&3.405&1.382&&0.256&5.857&2.359&&0.399&2.565&1.034\tabularnewline
~~1000&0.338&3.188&1.292&&0.244&5.466&2.202&&0.388&2.313&0.933\tabularnewline
~~2000&0.287&2.833&1.137&&0.209&4.776&1.918&&0.359&1.842&0.740\tabularnewline
\hline
{\bfseries $x=$1}&&&&&&&&&&&\tabularnewline
~~100&0.901&27.285&10.546&&0.323&77.275&28.086&&0.555&53.679&20.014\tabularnewline
~~250&0.757&21.694&8.600&&0.230&136.702&50.383&&0.502&36.186&14.246\tabularnewline
~~500&0.633&19.479&7.765&&0.175&139.797&53.718&&0.454&28.845&11.503\tabularnewline
~~750&0.569&18.470&7.391&&0.154&132.358&51.889&&0.429&25.410&10.147\tabularnewline
~~1000&0.532&17.516&7.020&&0.139&129.817&50.990&&0.413&23.265&9.305\tabularnewline
~~2000&0.442&15.794&6.319&&0.115&117.431&46.691&&0.384&18.243&7.290\tabularnewline
\hline
\end{tabular}